\pdfoutput=1

\def\ZZinstitution{Purdue University}

\def\ZZcampus{West\space Lafayette}

\def\ZZprogram{Mathematics}

\def\ZZdegree{Doctor of Philosophy}

\def\ZZauthor{Wei Deng}

\def\ZZdocument{A Dissertation}

\def\ZZgraduation{December 2021}

\def\ZZtitle{Non-convex Bayesian Learning via Stochastic Gradient Markov Chain Monte Carlo}

\def\ZZshowdiagonalline{false}

\def\ZZshowgridlines{false}

\def\ZZshowmarginlines{false}

\def\ZZshowtimestamp{false}

\def\ZZtodonotes{false}

\documentclass{PurdueThesis}

\def\ZZatinformation{}

\graphicspath{{graphics/}}

\makeatletter
    \def\input@path{{packages/}}
\makeatother

 \newcommand{\bibprocessor}{biblatex}
 \usepackage[backend=biber, citestyle=alphabetic]{biblatex}
\ifthen{\equal{biblatex}{\bibprocessor}}
{
  \DeclareSourcemap{
    \maps[datatype=bibtex]{
      \map{
        \step[fieldset=urldate, null]
      }
    }
  }
}

\pdfoutput=1

\usepackage{bold-extra}

\usepackage{chemfig}

\usepackage
[
  n,            
  advantage,
  operators,
  sets,
  adversary,
  landau,
  probability,
  notions,
  logic,
  ff,
  mm,
  primitives,
  events,
  complexity,
  oracles,
  asymptotics,
  keys,
]
{cryptocode}

\usepackage{fancyvrb}
  \DefineShortVerb{\|}  

\usepackage{url}
\usepackage{listings}
\usepackage{bm}
\usepackage{physics}
\usepackage{bbm}
\usepackage{mathtools}
\usepackage{amssymb,amsmath,amsthm}
\usepackage{algorithm}
\usepackage{algorithmic}
\usepackage{multirow}

\newcommand{\btheta}{{\boldsymbol \theta}}
\newcommand{\bx}{{\boldsymbol x}}
\newcommand{\bbeta}{{\boldsymbol \beta}}
\newcommand{\bxi}{{\boldsymbol \xi}}
\newcommand{\bW}{{\boldsymbol W}}

\def\hE{\mathbb{E}}
\newcommand{\bphi}{{\boldsymbol \phi}}
\def\cE{{\cal E}}
\usepackage{apptools}

\usepackage{xr}
\usepackage{subfloat}
\usepackage{wrapfig}

\def\b{\beta}
\def\g{\gamma}
\def\d{\delta}

\def\l{\lambda}
\def\m{\mu}
\def\n{\nu}
\def\si{\sigma}
\def\t{\tau}

\def\h{\widehat}
\def\L{\Lambda}
\def\Si{\Sigma}
\def\F{\Phi}
\def\cC{{\cal C}}

\def\cE{{\cal E}}

\def\cI{{\cal I}}
\def\cJ{{\cal J}}

\def\cL{{\cal L}}

\def\hE{\mathbb{E}}

\def\hL{\mathbb{L}}

\def\hN{\mathbb{N}}

\def\hP{\mathbb{P}}

\def\hR{\mathbb{R}}

\def\ms{\medskip}

\def\q{\quad}
\def\qq{\qquad}

\def\cd{\cdot}

\def\ran{{\rangle}}

\def\bx{{\bf x}}
\def\by{{\bf y}}

\def\pop{\bigotimes}
\newcommand{\bR}{{\boldsymbol R}}
\newcommand{\MX}{{\mathcal{X}}}
\newcommand{\bchi}{{\boldsymbol{\chi}}}
\newcommand{\xeta}{x}

\newcommand{\bw}{{\boldsymbol w}}

\newcommand{\bL}{{\boldsymbol L}}

\newcommand{\bU}{{\boldsymbol U}}
\newcommand{\bH}{{\boldsymbol H}}
\newcommand{\bI}{{\boldsymbol I}}

\newcommand{\MK}{{\mathcal{K}}}
\newcommand{\cov}{{\mbox{Cov}}}

\newcommand{\bTheta}{{\boldsymbol \Theta}}

\newcommand{\bSigma}{{\boldsymbol \Sigma}}

\def\hE{\mathbb{E}}
\def\cE{{\cal E}}

\newcommand{\bT}{{\boldsymbol T}}

\newcommand{\la}{\langle}
\newcommand{\ra}{\rangle}

\newcommand{\ba}{\begin{array}}
\newcommand{\ea}{\end{array}}
\newcommand{\bea}{\begin{eqnarray}}
\newcommand{\eea}{\end{eqnarray}}
\newcommand{\beaa}{\begin{eqnarray*}}
\newcommand{\eeaa}{\end{eqnarray*}}

\def\b{\beta}
\def\g{\gamma}
\def\d{\delta}

\def\l{\lambda}
\def\m{\mu}
\def\n{\nu}
\def\si{\sigma}
\def\t{\tau}

\def\h{\widehat}
\def\L{\Lambda}
\def\Si{\Sigma}
\def\F{\Phi}
\def\cC{{\cal C}}

\def\cE{{\cal E}}

\def\cI{{\cal I}}
\def\cJ{{\cal J}}

\def\cL{{\cal L}}

\def\hE{\mathbb{E}}

\def\hL{\mathbb{L}}

\def\hN{\mathbb{N}}

\def\hP{\mathbb{P}}

\def\hR{\mathbb{R}}

\def\ms{\medskip}

\def\q{\quad}
\def\qq{\qquad}

\def\cd{\cdot}

\def\ran{{\rangle}}

\def\bx{{\bf x}}
\def\by{{\bf y}}

\def\tr{\hbox{\rm tr}}

\def\qed{ \hfill \vrule width.25cm height.25cm depth0cm\smallskip}

\newcommand{\basa}{\begin{assumption}}
\newcommand{\easa}{\end{assumption}}

\newcommand{\bas}{\begin{assum}}
\newcommand{\eas}{\end{assum}}

\def\ran{\mathop{\rangle}}

\def\limP2{\,\mathop{\buildrel \Pi_2\over\longrightarrow\,}}

\def\h{\widehat}

 \def\cd{\cdot}

\def\as{\hbox{\rm -a.s.{ }}}
\def\sgn{\hbox{\rm sgn$\,$}}

\def\tr{\hbox{\rm tr$\,$}}

\def\Re{\hbox{\rm Re$\,$}}

\def\bx{{\bf x}}

\def\1{{\bf 1}}
\def\by{{\bf y}}

\def\:{\!:\!}

\def \proof{{\noindent \bf Proof\quad}}

\usepackage{amssymb}

\newtheorem{corollary}{Corollary}
\newtheorem{lemma}{Lemma}
\newtheorem{remark}{Remark}
\newtheorem{assumption}{Assumption}
\newtheorem{assump}{Assumption}

\let\origtheassumption\theassumption

\edef\oldassumption{\the\numexpr\value{assumption}+1}

\usepackage{microtype}
\usepackage{graphicx}

\usepackage[version=4]{mhchem}
  \newcommand{\nitrate}{\ce{NO3{-}}}
  
\usepackage{cancel}

\usepackage{etoc}

\usepackage{makeidx}
  \makeindex

\usepackage{mathtools}

\usepackage{media9}

\usepackage{multicol}

\usepackage{pa-ditto}

\usepackage{pa-figure-dash}

\usepackage{pa-logos}

\usepackage{pa-mismath}
  \enumber
  \inumber
  \jnumber
  \pinumber

\usepackage{pa-repeat}

\usepackage{placeins}

\usepackage{textcomp}
  
\usepackage{tikz}
  \usepackage{circuitikz}
  \usepackage{menukeys}
  \usetikzlibrary{calc,shapes.symbols,shadows}
  \usetikzlibrary{decorations.pathmorphing} 
  \usetikzlibrary{fit}
  \usetikzlibrary{backgrounds}	

\usepackage[compat=1.1.0]{tikz-feynman}

\newcommand{\tabularspace}{\noalign{\vspace*{2pt}}}

\usepackage{xfrac}

\newcommand{\I}[1]{\MyRepeat{\indent}{#1}}

\long\def\MyI#1%
{%
  {%
    \fontsize{8}{10}\tt
    \VerbatimInput[
      firstnumber = 1,
      numbers     = left,
      xleftmargin = 0.33in,
    ]{#1}
  }%
}  

\newcommand{\MyIO}
{%
  \input{z.out}

  {%
    \fontsize{8}{10}\tt
    \VerbatimInput[
      firstnumber = 1,
      numbers     = left,
      xleftmargin = 0.33in,
    ]{z.out}
  }
  \FloatBarrier
}  

\usepackage{pa-typographic-conventions}

\begin{document}

\maketitle

\ProvidesFile{ch-0.front.tex}[2021-08-23 front matter chapter]
%
%
%
%
%
%

%
%
%
\begin{statement}
  \entry{Dr.~Guang Lin, Co-Chair}{Department of Mathematics and School of Mechanical Engineering}
  \entry{Dr.~Faming Liang, Co-Chair}{Department of Statistics}
  \entry{Dr.~Samy Tindel}{Department of Mathematics}
  \entry{Dr.~Michael Zhu}{Department of Statistics}
  \approvedby{Dr.~Plamen D. Stefanov}
\end{statement}

\begin{dedication}
  Dedicated to my parents and my little sister.\\
  
  $\newline$
  In loving memory of my grandparents. 
 \end{dedication}

\begin{acknowledgments}
  
First and foremost, I would like to express my sincere gratitude to my advisors: Dr. Guang Lin and Dr. Faming Liang. Dr. Lin proposed a clear direction for my Ph.D. study and suggested how to get involved in research based on my own strength. He not only granted me the freedom to pursue my research interests but also provided endless academic and financial support. He initiated numerous inspiring collaborations and discussions with experienced scholars, such as Dr. Tuo Zhao, Yian Ma, Zhao Song, and Qiang Ye, which greatly enlightened me on my research and collaboration abilities. In addition to the academic advice, I am very thankful for his generous research funding \footnote[2]{I gratefully acknowledge the support of the National Science Foundation (DMS-1555072, DMS-1736364, DMS-2053746, and DMS-2134209), John Deere Smart Gear Boxes Phase I and II contracts, and Brookhaven National Laboratory Subcontract 382247, ARO/MURI grant W911NF-15-1-0562, and U.S. Department of Energy (DOE) Office of Science Advanced Scientific Computing Research program DE-SC0021142.} and rich GPU computing resources. Furthermore, I am deeply fortunate to obtain tremendous academic guidance from Dr. Liang, who introduced me to the fascinating world of Monte Carlo methods and showed me how to think critically and act independently. I always remember these exciting and fruitful discussions in which his eyes lit up. I am very grateful for his patient assistance in checking the details of proofs, the amazing group meetings and discussions, and all kinds of letters of recommendation. I also acknowledge the Department of Mathematics for providing such a free academic environment and offering me the Bilsland dissertation fellowship.

I would like to thank Dr. Samy Tindel and Michael Zhu for serving on my committee. I have learned quite a lot from the inspiring discussions with them. I am indebted to Dr. Samy Tindel and Nung Kwan Yip, who offered brilliant courses in probability theory, stochastic differential equations, dynamical systems, functional inequalities, and Sobolev space. It is those elegant, systematic, and supreme results that make me realize what is worth a lifetime of chasing. I deeply thank Dr. Vinayak Rao for the exceptional courses in computational statistics and Bayesian data analysis (joint with Dr. Arman Sabbaghi), which motivated me to embark on the journey of Bayesian learning; I am also grateful for the research discussions with Dr. Rao for publishing my first NeurIPS paper. Another big thank you goes to Dr. Michael Zhu, who instructed me in publishing my first academic paper and explained to me the promising future of deep learning in 2016. I am thankful for Dr. Kiseop Lee, who taught me the interesting financial applications of stochastic differential equations.

I am super fortunate to work with my talented collaborators, including Qi Feng, Xiao Zhang, Liyao Gao, Botao Hao, Georgios Karagiannis, Yian Ma, Zhao Song, Siqi Liang, Yating Wang, Junwei Pan, Deguang Kong, and Qian Zhang, among others. My research journey would be full of thorns and thistles without your help. I am very lucky to have my five-year roommates, Gaoping Huang and Jiexin Duan, to share my pleasure and sadness. During my spare time, I would like to say thank you to my dear friends, Xiaokai Yuan and Hengrong Du, for the countless dinnertime conversations. I also enjoy the pleasant time chatting with Heng Du, Jingshuang Chen, Xiaodong Huang, Warren Katz, Yan Sun, Mao Ye, Peiyi Zhang, Cheng Li, Xi Tan, Run Chen, Bao Wang, Zhanyu Wang, Chi-Hua Wang, Brian Hogle, Emmanuel Gil Torres, Frankie Chan, Joan Ponce,  among many others at Purdue University. I miss those funny tennis games with my friends, Nan Jiang, Chuan Zuo, Ke Wu, and Fukeng Huang.

Before I come to the United States, I deeply appreciate my Baidu colleagues and friends, Weimin Lu, Zhixian Yang, Min Zhang, Junying Lin, Huajie Sheng, Qiuyang Qu, Peng Zheng, Wei Ma, Jun Zhang, Yuan Gao, and Yuan Yang, for your consistent support and encouragement. I am also thankful to Dr. Yongguang Yu, Yanxun Chang, Houjin Chen, Naihua Xiu, and Guochen Feng at Beijing Jiaotong Univeresity for the career suggestions.


Finally, I would like to thank my beloved parents, who made so many sacrifices for their son. They moved around a lot when I was a child to help me obtain a better education. My mom gave up a lot of entertainment to help me focus on my studies, which taught me self-discipline and persistence in my studies. My dad gave me the great curiosity and courage to pursue anything mysterious and elegant, especially mathematics. I miss my little sister and the time I spent with her while we were growing up. I miss the memorable moments I lie down beside my grandparents on the roof of a house in the countryside, and they told me fairy tales while we were looking at the starry sky. I feel extremely indebted and even guilty to the considerate caring and great sacrifice from my whole family, including my parents, sister, grandpa, aunts, cousins, especially during the three-month brain disease in 2011. Now ten years have passed, wish you all safe and healthy during this COVID-19 pandemic.
\end{acknowledgments}


\pdfbookmark{TABLE OF CONTENTS}{Contents}
\tableofcontents

\listoftables

\listoffigures




\begin{abbreviations}
  AI & Artificial intelligence \cr
  AWSGLD & Adaptively weighted stochastic gradient Langevin dynamics \cr
  BMA & Bayesian model averaging \cr
  CNN & Convolutional neural network \cr
  CSGLD & Contour stochastic gradient Langevin dynamics \cr
  CV & Computer vision \cr
  DEO & Deterministic even odd scheme \cr
  DNN & Deep neural network \cr
  HMC & Hamiltonian Monte Carlo \cr
  ICSGLD & Interacting contour stochastic gradient Langevin dynamics \cr
  LSI & Log-Sobolev  inequality \cr
  MCMC & Markov chain Monte Carlo \cr
  MH & Metropolis Hasting \cr
  mSGD & Momentum stochastic gradient descent \cr
  NR & Non-reversibility \cr
  PT & Parallel tempering (also known replica exchange) \cr
  reLD & Replica exchange Langevin diffusion \cr
  reSGLD & Replica exchange stochastic gradient Langevin dynamics \cr
  SA & Stochastic approximation \cr
  SAMC & Stochastic approximation Monte Carlo \cr
  SDE & Stochastic differential equation \cr
  SGD & Stochastic gradient descent \cr
  SGHMC & Stochastic gradient Hamiltonian Monte Carlo \cr
  SGLD & Stochastic gradient Langevin dynamics \cr
  SGMCMC & Stochastic gradient Markov chain Monte Carlo \cr
  SVRG & Stochastic variance-reduced gradient \cr
  TS  & Thompson sampling \cr
  $W_2$ & 2-Wasserstein distance \cr
  
\end{abbreviations}

\begin{abstract}%

The rise of artificial intelligence (AI) hinges on the efficient training of modern deep neural networks (DNNs) for non-convex optimization and uncertainty quantification, which boils down to a non-convex Bayesian learning problem. A standard tool to handle the problem is Langevin Monte Carlo, which proposes to approximate the posterior distribution with theoretical guarantees. However, non-convex Bayesian learning in real big data applications can be arbitrarily slow and often fails to capture the uncertainty or informative modes given a limited time. As a result, advanced techniques are still required.

In this thesis, we start with the replica exchange Langevin Monte Carlo (also known as parallel tempering), which is a Markov jump process that proposes appropriate swaps between exploration and exploitation to achieve accelerations. However, the na\"ive extension of swaps to big data problems leads to a large bias, and the bias-corrected swaps are required. Such a mechanism leads to few effective swaps and insignificant accelerations. To alleviate this issue, we first propose a control variates method to reduce the variance of noisy energy estimators and show a potential to accelerate the exponential convergence. We also present the population-chain replica exchange and propose a generalized deterministic even-odd scheme to track the non-reversibility and obtain an optimal round trip rate. Further approximations are conducted based on stochastic gradient descents, which yield a user-friendly nature for large-scale uncertainty approximation tasks without much tuning costs. 

In the second part of the thesis, we study scalable dynamic importance sampling algorithms based on stochastic approximation. Traditional dynamic importance sampling algorithms have achieved successes in bioinformatics and statistical physics, however, the lack of scalability has greatly limited their extensions to big data applications. To handle this scalability issue, we resolve the vanishing gradient problem and propose two dynamic importance sampling algorithms based on stochastic gradient Langevin dynamics. Theoretically, we establish the stability condition for the underlying ordinary differential equation (ODE) system and guarantee the asymptotic convergence of the latent variable to the desired fixed point. Interestingly, such a result still holds given non-convex energy landscapes. In addition, we also propose a pleasingly parallel version of such algorithms with interacting latent variables. We show that the interacting algorithm can be theoretically more efficient than the single-chain alternative with an equivalent computational budget.

\end{abstract}

%
%

\chapter{INTRODUCTION}

AI safety has long been an important issue in the deep learning community. A promising solution to the problem is  Markov chain Monte Carlo (MCMC), which leads to asymptotically correct uncertainty quantification for deep neural network (DNN) models. However, traditional MCMC algorithms \cite{Metropolis1953,Hastings1970} are not scalable to big datasets that deep learning models rely on, although they have achieved significant successes in many scientific areas such as statistical physics and bioinformatics. It was not until the study of stochastic gradient Langevin dynamics (SGLD) \cite{Welling11} that resolves the scalability issue encountered in Monte Carlo computing for big data problems. Ever since a variety of scalable stochastic gradient Markov chain Monte Carlo (SGMCMC) algorithms have been developed based on strategies such as momentum augmentation \cite{Chen14, yian2015, Ding14}, Hessian approximation \cite{Ahn12, Li16, Simsekli2016}, and higher-order numerical schemes \cite{Chen15, Li19}. Despite their theoretical guarantees in statistical inference \cite{Chen15, Teh16, VollmerZW2016} and non-convex optimization \cite{Yuchen17, Maxim17, Xu18}, these algorithms often converge slowly, which makes them hard to be used for efficient uncertainty quantification for many AI safety problems.

To develop more efficient SGMCMC algorithms, we seek inspirations from traditional MCMC algorithms, such as simulated annealing \cite{Kirkpatrick83optimizationby}, simulated tempering \cite{ST}, parallel tempering (also known as replica exchange) \cite{PhysRevLett86, Geyer91}, and histogram algorithms \cite{Berg1991Multicanonical, Hesselbo1995MonteCS, WangLandau2001}. In particular, simulated annealing proposes to decay temperatures to increase the hitting probability to the global optima \cite{Mangoubi18}, which, however, often gets stuck into a local optimum with a fast cooling schedule. Simulated tempering requires a lot on the approximation of the normalizing constant, and is often not adopted in big data. 

In the first part of the thesis, we start with the replica exchange Langevin Monte Carlo (also known as parallel tempering), which includes a Poisson jump process to accelerate the computations and is suitable for implementation and parallelism. Specifically, the replica exchange Langevin diffusion utilizes multiple diffusion processes with different temperatures to balance between exploration and exploitation and proposes to swap the processes during the training. Intuitively, the high-temperature process acts as a bridge to connect the various modes. As such, the acceleration effect can be theoretically quantified \cite{Paul12, chen2018accelerating}. However, despite these advantages, a proper replica exchange SGMCMC (reSGMCMC) has long been missing in the deep learning community. 

In chapter \ref{icml20}, we figure out why the popular reSGMCMC is not widely adopted. In particular, the noisy energy estimator in big data leads to a large bias in the swaps, and a correction term depending on the variance of the energy estimator is required to maintain the detailed balance in a stochastic sense \cite{Andrieu09, Matias19}. We further establish the convergence of the algorithm in 2-Wasserstein distance ($W_2$), which shows the potential of using biased corrections and a large batch size to obtain better performance. However, the corrections in the bias-corrected swaps are too large and often yield little accelerations. Therefore, how to reduce the variance of the noisy energy estimators becomes quite important for accelerating the computations.

In chapter \ref{vr_resgld_iclr}, we study the control variates method to reduce the variance of the energy estimator. The crucial goal is to build correlated control variates to counteract the noise of the energy estimators. Motivated by the stochastic variance-reduced gradient (SVRG) in optimization \cite{SVRG}, we propose to update the control variate periodically to build correlated energy estimators. Interestingly, the variance-reduced energy estimators yield the potential of exponential accelerations instead of solely reducing the discretization error as in the gradient-based variance reduction \cite{Dubey16,Xu18}.

In chapter \ref{NRPT_uncertainty}, we propose an alternative approach to obtain more effective swaps by increasing the number of chains. A standard scheme to conduct the population-chain replica exchange is the deterministic even-odd scheme (DEO) \cite{DEO}, which exploits the non-reversibility property and has successfully reduced the communication cost from $O(P^2)$ to $O(P)$ given sufficient many $P$ chains \cite{Syed_jrssb}. However, the noisy energy estimators in big data have greatly reduced the swap rate. As a result, it is quite challenging to adopt sufficient many chains due to the limited budget in practice. To handle this issue, we generalize the DEO scheme to promote the non-reversibility and obtain an optimal communication cost $O(P \log P)$ in non-asymptotic scenarios. In addition, we also propose the approximate exploration kernels based on stochastic gradient descents with different learning rates. Such a user-friendly nature greatly reduces the cost in hyperparameter tuning and makes the algorithm much more scalable to big datasets.

In the second part of the thesis, we adapt the histogram algorithms in deep learning. The histogram algorithms, such as the multicanonical \cite{Berg1991Multicanonical}, Wang-Landau \cite{WangLandau2001}, and 1/k-ensemble sampler \cite{Hesselbo1995MonteCS} algorithms, were first proposed to sample discrete states of Ising models by yielding a flat or other desired histograms in the energy space and then extended as a general dynamic importance sampling algorithm, the so-called stochastic approximation Monte Carlo (SAMC) algorithm \cite{Liang05, Liang07, LiangPL2009, Liang05}. Theoretical studies \cite{leli2008, Liang10, Fort15} support the efficiency of the histogram algorithms in Monte Carlo computing for small data problems. However, it is still unclear how to adapt the histogram idea to accelerate the convergence of SGMCMC, ensuring efficient uncertainty quantification for AI safety problems. 

In chapter \ref{CSGLD_multi_distribution}, we propose an adaptively weighted stochastic gradient Langevin dynamics algorithm (SGLD), so-called contour stochastic gradient Langevin dynamics (CSGLD), for non-convex Bayesian learning in big data statistics. The proposed algorithm is essentially a scalable dynamic importance sampler, which automatically flattens the target distribution such that the simulation for a multi-modal distribution can be greatly facilitated. Theoretically, we prove a stability condition and establish the asymptotic convergence of the self-adapting parameter to a unique fixed-point, regardless of the non-convexity of the original energy function; we also present an error analysis for the weighted averaging estimators. Empirically, the CSGLD algorithm is tested on multiple benchmark datasets including CIFAR10 and CIFAR100. The numerical results indicate its superiority over the existing state-of-the-art algorithms in training deep neural networks.

In chapter \ref{ICSGLD}, we propose an interacting contour stochastic gradient Langevin dynamics (ICSGLD) sampler, an embarrassingly parallel multiple-chain contour stochastic gradient Langevin dynamics (CSGLD) sampler with efficient interactions.  We show that ICSGLD can be theoretically more efficient than a single-chain CSGLD with an equivalent computational budget. We also present a novel random-field function, which is proven to facilitate the estimation of self-adapting parameters for ICSGLD in big data problems. Empirically, we compare the proposed algorithm with various well-established benchmark methods for posterior sampling. The numerical results show a great potential of ICSGLD for uncertainty estimation in both supervised learning and reinforcement learning.

In chapter \ref{awsgld}, we propose an alternatively adaptively weighted stochastic gradient Langevin dynamics (AWSGLD) algorithm for Bayesian learning of big data problems. The proposed algorithm is scalable and possesses a self-adjusting mechanism: It adaptively flattens the high energy region and protrudes the low energy region during simulations such that both Monte Carlo simulation and global optimization tasks can be greatly facilitated in a single run. The self-adjusting mechanism enables the proposed algorithm to be immune to local traps.

\part{REPLICA EXCHANGE STOCHASTIC LANGEVIN MONTE CARLO}

\chapter{NON-CONVEX LEARNING VIA REPLICA EXCHANGE STOCHASTIC GRADIENT LANGEVIN DYNAMICS}
\label{icml20}

\let\origtheassumption\theassumption

\edef\oldassumption{\the\numexpr\value{assumption}+1}

\section{Introduction}

Replica exchange method, also known as parallel tempering, has become the go-to workhorse for simulations of complex multi-modal distributions. However, a proper replica exchange SGMCMC (reSGMCMC) has long been missing in the deep learning community.

A bottleneck that hinders the development of reSGMCMC is the na\"ive extension of the acceptance-rejection criterion that fails in mini-batch settings. Various attempts \cite{talldata17, Anoop14} were proposed to solve this issue. However, they introduce biases even with the ideal normality assumption on the noise. Some unbiased estimators \cite{PhysRevLett85, Alexandros06} have ever been presented, but the large variance leads to inefficient inference. To remove the bias while maintaining efficiency, \cite{penalty_swap99} proposed a corrected criterion under normality assumptions, and \cite{Daniel17, Matias19} further analyzed the model errors with the asymptotic normality assumptions. However, the above algorithms fail when the required corrections are time-varying and much larger than the energies as shown in Fig.\ref{cifar_biases}(a-b). Consequently, an effective algorithm with the potential to adaptively estimate the corrections and balance between acceleration and accuracy is in great demand.

In this chapter, we propose an adaptive replica exchange SGMCMC algorithm via stochastic approximation (SA) \cite{RobbinsM1951, Liang07, deng2019}, a standard method in adaptive sampling to estimate the latent variable: the unknown correction. The adaptive algorithm not only shows the asymptotic convergence in standard scenarios but also gives a good estimate when the corrections are time-varying and excessively large. We theoretically analyze the discretization error for reSGMCMC in mini-batch settings and show the accelerated convergence in 2-Wasserstein distance. Such analysis sheds light on the use of biased estimates of unknown corrections to obtain a trade-off between acceleration and accuracy. In summary, this algorithm has three main contributions:

\begin{itemize}
    \item We propose a novel reSGMCMC to speed up the computations of SGMCMC in DNNs with theoretical guarantees. The theory shows the potential of using biased corrections and a large batch size to obtain better performance.
    \item We identify the problem of time-varying corrections in DNNs and propose to adaptively estimate the time-varying corrections, with potential extension to a variety of time-series prediction techniques.
    \item We test the algorithm through extensive experiments using various models. It achieves the state-of-the-art results in both supervised learning and semi-supervised learning tasks on CIFAR10, CIFAR100, and SVHN datasets.
\end{itemize}

\section{Preliminaries}

A standard sampling algorithm is the Langevin diffusion, which is a stochastic differential equation (SDE) as follows:
\begin{equation}
\label{sde_1}
    d \bbeta_t  = - \nabla U(\bbeta_t ) dt+\sqrt{2\tau^{(1)}} d\bW_t,
\end{equation}
where $\bbeta_t\in\mathbb{R}^d$, $U(\cdot)$ is the energy function, $\bW_t\in\mathbb{R}^d$ is the Brownian motion, and $\tau^{(1)}>0$ is the temperature. 

Under mild growth conditions on $U$, the Langevin diffusion $\{\bbeta_t^{(1)}\}_{t\geq 0}$ converges to the unique invariant Gibbs distribution $\pi_{\tau^{(1)}}(\bbeta)\propto \exp(-\frac{U(\bbeta)}{\tau^{(1)}})$, where the temperature $\tau^{(1)}$ is crucial for both optimization and sampling of the non-convex energy function $U$. On the one hand, a high-temperature $\tau^{(1)}$ achieves the \emph{exploration} effect: the convergence to the flattened Gibbs distribution of the whole domain is greatly facilitated. However, the flattened distribution is less concentrated around the global optima \cite{Maxim17}, and the geometric connection to the global minimum is affected \cite{Yuchen17}.  
On the other hand, a low-temperature $\tau^{(1)}$ leads to the \emph{exploitation} effect: the solutions explore the local geometry rapidly, but they are more likely to get trapped in local optima, leading to a slow convergence in both optimization and sampling. Therefore, a fixed temperature is quite limited in the computations.


A powerful algorithm called replica exchange  Langevin diffusion (reLD), also known \begin{wrapfigure}{r}{0.25\textwidth}
   \begin{center}
   \vskip -0.45in
     \includegraphics[width=0.25\textwidth]{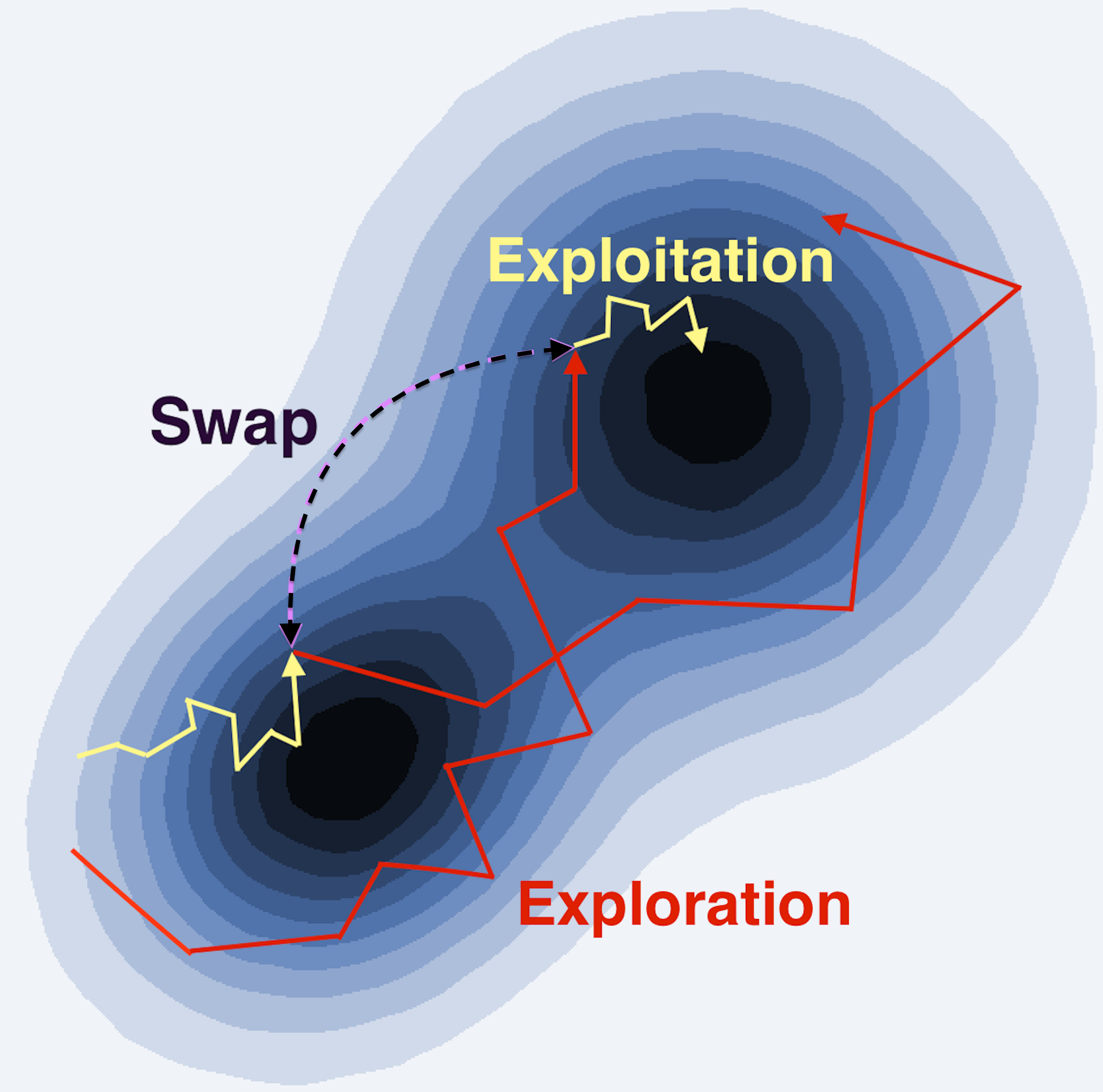}
   \end{center}
   \vskip -0.25in
  \caption{{Paths.}}
  \label{path_demo}
\end{wrapfigure} as parallel tempering Langevin diffusion, has been proposed to accelerate the convergence of the SDE as shown in Figure \ref{path_demo}. reLD proposes to simulate a high-temperature particle $\bbeta_t^{(2)}$ for \emph{exploration} and a low-temperature particle $\bbeta_t^{(1)}$ for \emph{exploitation} and allows them to swap simultaneously. Now consider the following coupled processes with a higher temperature $\tau^{(2)}>\tau^{(1)}$ and $\bW^{(2)}$ independent of $\bW^{(1)}$:
\begin{equation}
\label{sde_2_couple_icml20}
\begin{split}
    d \bbeta_t^{(2)} &= - \nabla U(\bbeta_t^{(2)}) dt+\sqrt{2\tau^{(2)}} d\bW_t^{(2)}.\\
    d \bbeta_t^{(1)} &= - \nabla U(\bbeta_t^{(1)}) dt+\sqrt{2\tau^{(1)}} d\bW_t^{(1)}\\
\end{split}
\end{equation}
Eq.(\ref{sde_2_couple_icml20}) converges to the invariant distribution with density 
\begin{equation}
\label{pt_density_main}
\begin{split}
    \pi(\bbeta^{(1)}, \bbeta^{(2)})\propto e^{-\frac{U(\bbeta^{(1)})}{\tau^{(1)}}-\frac{U(\bbeta^{(2)})}{\tau^{(2)}}}.
\end{split}
\end{equation}

By allowing the two particles to swap, the particles at the position $(\beta^{(1)}, \beta^{(2)})$ are likely to change from $(\beta^{(1)},\beta^{(2)})$ to $(\beta^{(2)},\beta^{(1)})$ with a swapping rate $a (1\wedge S(\beta^{(1)},\beta^{(2)}))dt$, where the constant $a\geq 0$ is the swapping intensity, and $S(\cdot, \cdot)$ satisfies
\begin{equation}
\begin{split}
    S(\beta^{(1)}, \beta^{(2)}):=e^{ \left(\frac{1}{\tau^{(1)}}-\frac{1}{\tau^{(2)}}\right)\left(U(\beta^{(1)})-U(\beta^{(2)})\right)}.
\end{split}
\end{equation}
In such a case, reLD is a Markov jump process, which is reversible \cite{chen2018accelerating} and leads to the same invariant distribution ($\ref{pt_density_main}$).

\section{Replica Exchange Stochastic Gradient Langevin Dynamics}

The wide adoption of the replica exchange Monte Carlo in traditional MCMC algorithms motivates us to design replica exchange stochastic gradient Langevin dynamics for DNNs, but the straightforward extension of reLD to replica exchange stochastic gradient Langevin dynamics is highly non-trivial \cite{Chen14, yian2015, Simsekli2016}. In this section, we will first show that na\"ive extensions of replica exchange Monte Carlo to SGLD (na\"ive reSGLD) lead to large biases. Afterward, we will present an adaptive replica exchange stochastic gradient Langevin dynamics (reSGLD) that will automatically adjust the bias and yield a good approximation to the correct distribution.

\subsection{Na\"ive reSGLD}

We denote the entire data by $\mathcal{D}=\{\bm{\mathrm{d}}_i\}_{i=1}^N$, where $\bm{\mathrm{d}}_i$ is a data point. Given the model parameter $\widetilde \bbeta$, we consider the following energy function (negative log-posterior)
\begin{equation}
L(\widetilde \bbeta)= -\log p(\widetilde \bbeta) - \sum_{i=1}^N \log P(\bm{\mathrm{d}}_i\mid\widetilde \bbeta).
\end{equation}
where $p(\widetilde \bbeta)$ is a proper prior and $\sum_{i=1}^N \log P(\bm{\mathrm{d}}_i\mid\widetilde \bbeta)$ is the complete data log-likelihood. When the number of data points $N$ is large, it is expensive to evaluate $L(\widetilde \bbeta)$ directly. Instead, we propose to approximate the energy function with a mini-batch of data $\mathcal{B}=\{\bm{\mathrm{d}}_{s_i}\}_{i=1}^n$, where $s_i\in \{1, 2, ..., N\}$. We can easily check that if $\mathcal{B}$ is sampled randomly with or without replacement, we obtain the following unbiased estimator of the energy function
\begin{equation}
\widetilde L(\widetilde \bbeta) = -\log p(\widetilde \bbeta)-\frac{N}{n}\sum_{i=1}^n \log P(\bm{\mathrm{d}}_{s_i}\mid\widetilde \bbeta).
\end{equation}
Let $\widetilde \bbeta_k$ denote the estimate of $\widetilde \bbeta$ at the $k$-th iteration. Next, SGLD proposes the following iterations:
\begin{equation}
\begin{split}
&\widetilde \bbeta_{k+1}=\widetilde \bbeta_k - \eta_k \nabla \widetilde{L}(\widetilde \bbeta_k)+\sqrt{2\eta_k\tau^{(1)}}\bxi_k,\\
\end{split}
\end{equation}
where $\eta_k$ is the learning rate, the stochastic gradient $\nabla \widetilde{L}(\widetilde \bbeta_k)$ is the unbiased estimator of the exact gradient $\nabla L(\widetilde \bbeta_k)$, $\bxi$ is a standard $d$-dimensional Gaussian vector with mean $\bm{0}$ and identity covariance matrix. It is known that SGLD asymptotically converges to a unique invariant distribution $\pi(\widetilde \bbeta)\propto \exp\left(-L(\widetilde \bbeta)/\tau^{(1)}\right)$ \cite{Teh16} as $\eta_k\rightarrow 0$. If we simply replace gradients with stochastic gradients in the replica exchange gradient Langevin dynamics, we have
\begin{equation}
\label{naive_resgld}
\begin{split}
    \widetilde \bbeta^{(2)}_{k+1} &= \widetilde \bbeta^{(2)}_{k} - \eta_k\nabla \widetilde L(\widetilde \bbeta^{(2)}_k)+\sqrt{2\eta_k\tau^{(2)}} \bxi_k^{(2)}.\\
    \widetilde \bbeta^{(1)}_{k+1} &= \widetilde \bbeta^{(1)}_{k}- \eta_k \nabla \widetilde L(\widetilde \bbeta^{(1)}_k)+\sqrt{2\eta_k\tau^{(1)}} \bxi_k^{(1)}\\
\end{split}
\end{equation}
Furthermore, we swap the Markov chains in (\ref{naive_resgld}) with the na\"ive stochastic swapping rate $a (1\wedge\mathbb{S}(\widetilde \bbeta_{k+1}^{(1)}, \widetilde \bbeta_{k+1}^{(2)}))\eta_k$\footnote[4]{In the implementations, we fix $a\eta_k=1$ by default.}:
\begin{equation}
\label{swap_naive}
\begin{split}
    &\mathbb{S}(\widetilde \bbeta_{k+1}^{(1)}, \widetilde \bbeta_{k+1}^{(2)})=e^{ \left(\frac{1}{\tau^{(1)}}-\frac{1}{\tau^{(2)}}\right)\left(\widetilde L(\widetilde \bbeta_{k+1}^{(1)})-\widetilde L(\widetilde \bbeta_{k+1}^{(2)})\right)}.\\
\end{split}
\end{equation}
Apparently, the unbiased estimators $\widetilde L(\widetilde \bbeta_{k+1}^{(1)})$ and $\widetilde L(\widetilde \bbeta_{k+1}^{(2)})$ in $\mathbb{S}(\widetilde \bbeta_{k+1}^{(1)}, \widetilde \bbeta_{k+1}^{(2)})$ do not provide an unbiased estimator of $S(\widetilde \bbeta_{k+1}^{(1)}, \widetilde \bbeta_{k+1}^{(2)})$ after a non-linear transformation as shown in (\ref{swap_naive}), which leads to a large bias.

\subsection{Replica Exchange SGLD with Correction}
A viable MCMC algorithm requires the approximately unbiased estimators of the swapping rates to ``satisfy'' the detailed balance property \cite{penalty_swap99, Andrieu09, Geoff12} and the weak solution of a Markov jump process with unbiased stochastic coefficients has also been studied in \cite{Gyongy86, Amel}. When we make normality assumption on the stochastic energy $\widetilde L(\bbeta)\sim\mathcal{N}( L(\bbeta), \sigma^2)$, it follows 
\begin{equation}
\label{diff_L}
    \widetilde L(\widetilde \bbeta^{(1)})-\widetilde L(\widetilde \bbeta^{(2)})=
    L(\widetilde \bbeta^{(1)})-L(\widetilde \bbeta^{(2)}) + \sqrt{2}\sigma W_1, \\
\end{equation}
where $W_1$ follows the standard normal distribution and can be viewed as a Brownian motion at $t=1$. Consider the evolution of the stochastic swapping rate $\{\widetilde S_t\}_{t\in[0,1]}$ in each swap as a geometric Brownian motion:
\begin{equation}
\begin{split}
\label{analytic}
    \widetilde S_t&=e^{\left(\frac{1}{\tau^{(1)}}-\frac{1}{\tau^{(2)}}\right)\left(\widetilde L(\widetilde \bbeta^{(1)})-\widetilde L(\widetilde \bbeta^{(2)})-\left(\frac{1}{\tau^{(1)}}-\frac{1}{\tau^{(2)}}\right)\sigma^2 t\right)}\\
    &=e^{\left(\frac{1}{\tau^{(1)}}-\frac{1}{\tau^{(2)}}\right)\left(L(\widetilde \bbeta^{(1)})-L(\widetilde \bbeta^{(2)})-\left(\frac{1}{\tau^{(1)}}-\frac{1}{\tau^{(2)}}\right)\sigma^2 t+\sqrt{2}\sigma W_t\right)}.\\
\end{split}
\end{equation}
Set $\tau_{\delta}=\frac{1}{\tau^{(1)}}-\frac{1}{\tau^{(2)}}$ and take the partial derivatives of $\widetilde S_t$
\begin{equation*}
\begin{split}
     &\frac{d \widetilde S_t}{d t}=-\tau_{\delta}^2\sigma^2\widetilde S_t,\  \frac{d \widetilde S_t}{dW_t}=\sqrt{2}\tau_{\delta}\sigma \widetilde S_t,\frac{d^2 \widetilde S_t}{d W_t^2}=2\tau_{\delta}^2\sigma^2 \widetilde S_t.
\end{split}
\end{equation*}
It\^{o}'s lemma shows that
\begin{equation*}
\begin{split}
    d \widetilde S_t&=\left(\frac{d\widetilde S_t}{dt}+\frac{1}{2}\frac{d^2 \widetilde S_t}{dW_t^2}\right)dt + \frac{d\widetilde S_t}{dW_t}dW_t= \sqrt{2}\tau_{\delta}\sigma \widetilde S_t d W_t.
\end{split}
\end{equation*}
Notice that $\{\widetilde S_t\}_{t\in[0,1]}$ is a Martingale and yields the same expectation for $\forall t\in [0, 1]$. By fixing $t=1$ in (\ref{analytic}), we have
\begin{equation}
\begin{split}
\label{sto_acceptance}
    \widetilde S_1&=e^{ \left(\frac{1}{\tau^{(1)}}-\frac{1}{\tau^{(2)}}\right)\left( \widetilde L(\widetilde \bbeta^{(1)})- \widetilde L(\widetilde \bbeta^{(2)})-\left(\frac{1}{\tau^{(1)}}-\frac{1}{\tau^{(2)}}\right)\sigma^2\right)},\\
\end{split}
\end{equation}
where the stochastic swapping rate $\widetilde S_1$ is an unbiased estimator of $\widetilde S_0=e^{\left(\frac{1}{\tau^{(1)}}-\frac{1}{\tau^{(2)}}\right)\left(L(\widetilde \bbeta^{(1)})-L(\widetilde \bbeta^{(2)})\right)}$, and the correction term $\left(\frac{1}{\tau^{(1)}}-\frac{1}{\tau^{(2)}}\right)\sigma^2$ aims to remove the bias from the swaps. 

An advantage of interpreting the correction term from the perspective of geometric Brownian motion is that we may extend it to geometric L\'{e}vy process \cite{David04}, which is more suitable for the heavy-tailed energy noise \cite{Simsekli2019b}. Admittedly, the estimation of the tail-index of extreme-value distributions and the correction under L\'{e}vy process go beyond the scope of this chapter, so we leave it for future works.

\subsection{Adaptive Replica Exchange SGLD}

In reality, the exact variance $\sigma^2$ is hardly known and subject to estimation. The normality assumption may be violated and even no longer time-independent.

\subsubsection{Fixed Variance $\sigma^2$}
We use stochastic approximation (SA) to adaptively estimate the unknown variance while sampling from the posterior. In each SA step, we obtain an unbiased sample variance $\tilde \sigma^2$ for the true $\sigma^2$ and update the adaptive estimate $\widehat \sigma_{m+1}^2$ through
\begin{equation}
    \widehat \sigma^2_{m+1} = (1-\gamma_m)\widehat \sigma^2_{m}+\gamma_m \tilde \sigma^2_{m+1},
\end{equation}
where $\gamma_m$ is the smoothing factor at the $m$-th SA step. The SA step is updated less frequently than the standard sampling to reduce the computational cost. When the normality assumption holds, we notice that $\widehat \sigma_{m}^2=\sum_{i=1}^m \tilde \sigma_{i}^2/m$ when $\gamma_m=\frac{1}{m}$. Following central limit theorem (CLT), we have that $\widehat \sigma_m^2-\sigma^2 =\mathcal{O}(\frac{1}{m})$. 
Inspired by theorem 2 from \cite{Chen15}, we expect that the weak convergence of the adaptive sampling algorithm holds since the bias decreases sufficiently fast ($\frac{1}{m}\sum_{l=1}^{m} \mathcal{O}(\frac{1}{l})\rightarrow 0$ as $m\rightarrow \infty$).

In practice, the normality assumption is likely to be violated when we use a small batch size $n$, but the unknown distribution asymptotically approximates the normal distribution as $n\rightarrow \infty$ and yield a bias $\mathcal{O}(\frac{1}{n})$ in each SA step. Besides, the mini-batch setting usually introduces a very large noise on the estimator of the energy function, which requires a large correction term and leads to \emph{almost-zero swapping rates}. 

To handle this issue, we introduce a correction factor $F$ to reduce the correction term from $\left(\frac{1}{\tau^{(1)}}-\frac{1}{\tau^{(2)}}\right)\widehat \sigma^2$ to $\frac{\left(\frac{1}{\tau^{(1)}}-\frac{1}{\tau^{(2)}}\right)\widehat \sigma^2}{F}$. We note that a large $F>1$ introduces some bias, but may significantly increase the acceleration effect, giving rise to an acceleration-accuracy trade-off in finite time. Now, we show the algorithm in Alg.\ref{alg_icml20}. In addition to simulations of multi-modal distributions, our algorithm can be also combined with simulated annealing \cite{PT_SA, OPT_SA}
to accelerate the non-convex optimization and increase the hitting probability to the global optima \cite{Mangoubi18}.

\begin{algorithm}[tb]
  \caption{Adaptive Replica Exchange Stochastic Gradient Langevin Dynamics Algorithm. For \emph{sampling purposes}, we fix the temperatures $\tau^{(1)}$ and $\tau^{(2)}$; for \emph{optimization purposes}, we keep annealing $\tau^{(1)}$ and $\tau^{(2)}$ during each epoch. Empirically, a larger $\gamma_m$ tracks the dynamics better but is less robust. The intensity $r$ and $\eta$ are omitted in the corrected swaps.}
  \label{alg_icml20}
\begin{algorithmic}
\REPEAT
  \STATE{\textbf{Sampling Step}}
  \begin{equation*}
  \small
      \begin{split}
        \widetilde \bbeta^{(2)}_{k+1} &= \widetilde \bbeta^{(2)}_{k} - \eta_k^{(2)}\nabla \widetilde L(\widetilde \bbeta^{(2)}_k)+\sqrt{2\eta_k^{(2)}\tau^{(2)}} \bxi_k^{(2)},\\
        \widetilde \bbeta^{(1)}_{k+1} &= \widetilde \bbeta^{(1)}_{k}- \eta_k^{(1)} \nabla \widetilde L(\widetilde \bbeta^{(1)}_k)+\sqrt{2\eta_k^{(1)}\tau^{(1)}} \bxi_k^{(1)}\\
      \end{split}
  \end{equation*}
  \STATE{\textbf{SA Step}}
  \STATE{Obtain an unbiased estimate $\tilde \sigma_{m+1}^2$ for $\sigma^2$.}
  \begin{equation*}
  \begin{split}
      &\widehat \sigma^2_{m+1} = (1-\gamma_m)\widehat \sigma^2_{m}+\gamma_m \tilde \sigma^2_{m+1},\\
  \end{split}
  \end{equation*}
  \STATE{\textbf{Swapping Step}}
  \STATE{Generate a uniform random number $u\in [0,1]$.}
  \begin{equation*}
      \textstyle \hat S_1=e^{ \left(\frac{1}{\tau^{(1)}}-\frac{1}{\tau^{(2)}}\right)\left( \widetilde L(\widetilde \bbeta_{k+1}^{(1)})- \widetilde L(\widetilde \bbeta_{k+1}^{(2)})-\frac{\left(\frac{1}{\tau^{(1)}}-\frac{1}{\tau^{(2)}}\right)\hat \sigma^2_{m+1}}{F}\right)}.
  \end{equation*}
  \IF{$u<\hat S_1$}
  \STATE Swap $\widetilde \bbeta_{k+1}^{(1)}$ and $\widetilde \bbeta_{k+1}^{(2)}$.
  \ENDIF
  \UNTIL{$k=k_{\max}$}
    \vskip -1 in
\end{algorithmic}
\end{algorithm}

\subsubsection{Time-varying Variance $\sigma^2$}
In practice, the variance $\sigma^2$ usually varies with time, resulting in time-varying corrections. For example, in the optimization of residual networks on CIFAR10 and CIFAR100 datasets, we notice from Fig.\ref{cifar_biases}(a-b) that the corrections are time-varying. As such, we cannot use a fixed correction anymore to deal with the bias. The treatment for the time-varying corrections includes standard methods for time-series data, and a complete recipe for modeling the data goes beyond our scope. We still adopt the method of stochastic approximation and choose a fixed smoothing factor $\gamma$ so that
\begin{equation}
\label{smoothing}
    \widehat \sigma^2_{m+1} = (1-\gamma)\widehat \sigma^2_{m}+\gamma \tilde \sigma^2_{m+1}.
\end{equation}

Such a method resembles the simple exponential smoothing and acts as robust filters to remove high-frequency noise. It can be viewed as a special case of autoregressive integrated moving average (ARIMA) (0,1,1) model but often outperforms the ARIMA equivalents because it is less sensitive to the model selection error \cite{john66}. From the regression perspective, this method can be viewed as a zero-degree
local polynomial kernel model \cite{Gijbels99}, which is robust to distributional
assumptions.

\section{Convergence Analysis}

We theoretically analyze the acceleration effect and the accuracy of reSGLD in terms of 2-Wasserstein distance between the Borel probability measures $\mu$ and $\nu$ on $\mathbb{R}^d$ 
\begin{equation*}
    W_2(\mu, \nu):=\inf_{\Gamma\in \text{Couplings}(\mu, \nu)}\sqrt{\int\|\bbeta_{\mu}-\bbeta_{\nu}\|^2 d \Gamma(\bbeta_{\mu}, \bbeta_{\nu})},
\end{equation*}
where $\|\cdot\|$ is the Euclidean norm, and the infimum is taken over all joint distributions $\Gamma(\bbeta_{\mu}, \bbeta_{\nu})$ with $\mu$ and $\nu$ being its marginal distributions.

Our analysis begins with the fact that reSGLD in Algorithm.\ref{alg_icml20} tracks the replica exchange Langevin diffusion (\ref{sde_2_couple_icml20}). For ease of analysis, we consider a fixed learning rate $\eta$ for both chains. reSGLD can be viewed as a special discretization of the continuous-time Markov jump process. In particular, it differs from the standard discretization of the continuous-time Langevin algorithms \cite{chen2018accelerating, yin_zhu_10, Maxim17,Sato2014ApproximationAO} in that we need to consider the discretization of the Markov jump process in a stochastic environment. To handle this issue, we follow \cite{Paul12} and view the swaps of positions as swaps of the temperatures, which have been proven equivalent in distribution. 
\begin{lemma}[Discretization Error, informal version of Lemma \ref{discretization_appendix_reSGLD}]
\label{discretization_main_reSGLD}
Given the smoothness and dissipativity assumptions \eqref{assump: lip and alpha beta} and \eqref{assump: dissipitive}, and a small learning rate $\eta$, we have
\begin{equation*}
\begin{split}
\hE[\sup_{0\le t\le T}\|\bbeta_t-\widetilde \bbeta^{\eta}_t\|^2] \le \mathcal{\tilde O}(\eta+\max_i\hE[\|\bphi_i\|^2]+ \max_{i}\sqrt{\hE\left[\mid\psi_{i}\mid^2\right]}),\\
\end{split}
\end{equation*}
where $\widetilde \bbeta^{\eta}_t$ is the continuous-time interpolation for reSGLD, $\bphi:=\nabla \widetilde U-\nabla U$ is the noise in the stochastic gradient, and $\psi:=\widetilde S-S$ is the noise in the stochastic swapping rate. 
\end{lemma}

Then we quantify the evolution of the 2-Wasserstein distance between $\nu_{t}$ and the invariant distribution $\pi$, where $\nu_t$ is the probability measure associated with reLD at time $t$. The key tool is the exponential decay of entropy when $\pi$ satisfies the log-Sobolev inequality (LSI) \cite{Bakry2014}. To justify LSI, we first verify LSI for reLD without swaps, which is a direct result \cite{Cattiaux2010} given the Lyapunov function criterion and the Poincar\'{e} inequality \cite{chen2018accelerating}. Then we verify LSI for reLD with swaps by analyzing the corresponding Dirichlet form, which is strictly larger than the Dirichlet form associated with reLD without swaps. Finally, the exponential decay of the 2-Wasserstein distance follows from the Otto-Villani theorem \cite{Bakry2014} by connecting 2-Wasserstein distance with the relative entropy.

\begin{lemma}[Accelerated exponential decay of $W_2$, informal version of Lemma \ref{exponential decay_appendix}]\label{exponential decay_main}
Under the smoothness and dissipativity assumptions \eqref{assump: lip and alpha beta} and \eqref{assump: dissipitive}, we have that the replica exchange Langevin diffusion converges exponentially fast to the invariant distribution $\pi$:
\begin{equation}
    W_2(\nu_t,\pi) \leq  D_0 e^{-k\eta(1+\delta_S)/c_{\text{LS}}},
\end{equation}
where $\delta_{S}:=\inf_{t>0}\frac{\cE_S(\sqrt{\frac{d\nu_t}{d\pi}})}{\cE(\sqrt{\frac{d\nu_t}{d\pi}})}-1$ is the acceleration effect depending on the swapping rate $S$, $\cE$ and $\cE_S$ are the Dirichlet forms defined in the appendix, $c_{\text{LS}}$ is the constant in the log-Sobolev inequality, $D_0=\sqrt{2c_{\text{LS}}D(\nu_0\|\pi)}$.
\end{lemma}{}

Finally, combining the definition of Wasserstein distance and the triangle inequality, we have that 

\begin{theorem}[Convergence of reSGLD, informal version of Theorem \ref{convergence_reSGLD_appendix}]
\label{convergence_reSGLD_main}
Let the smoothness and dissipativity assumptions \eqref{assump: lip and alpha beta} and \eqref{assump: dissipitive} hold. For the distribution $\{\mu_{k}\}_{k\ge 0}$ associated with the discrete dynamics $\{\widetilde \bbeta_k\}_{k\ge 1}$, we have the following estimates for $k\in \mathbb N^{+}$:
\begin{equation*}
\begin{split}
    W_2(\mu_{k}, &\pi) \le  D_0 e^{-k\eta(1+\delta_S)/c_{\text{LS}}}+\footnotesize{\mathcal{\tilde O}(\eta^{\frac{1}{2}}+\max_i(\hE[\|\bphi_i\|^2])^{\frac{1}{2}}+ \max_{i}(\hE\left[\mid\psi_{i}\mid^2\right]})^{\frac{1}{4}}),\\
\end{split}
\end{equation*}
where $D_0=\sqrt{2c_{\text{LS}}D(\mu_0\|\pi)}$, $\delta_{S}:=\min_{i}\frac{\cE_S(\sqrt{\frac{d\mu_i}{d\pi}})}{\cE(\sqrt{\frac{d\mu_i}{d\pi}})}-1$.
\end{theorem}{}
Ideally, we want to boost the acceleration effect $\delta_S$ by using a larger swapping rate $S$ and increase the accuracy by reducing the mean squared errors $\hE[\|\bphi_i\|^2]$ and $\hE[\mid\psi_i\mid^2]$. One possible way is to apply a \emph{large enough batch size}, which may be yet inefficient given a large dataset. Another way is to \emph{balance} between acceleration and accuracy by tuning the correction factor $F$. In practice, a larger $F$ leads to a larger acceleration effect and also injects more biases.

\section{Experiments}

\subsection{Simulations of Gaussian Mixture Distributions}

In this group of experiments, we evaluate the acceleration effects and the biases for reSGLD on multi-modal distributions based on different assumptions on the estimators for the energy function. As a comparison, we choose SGLD and the na\"ive reSGLD without corrections as baselines. The learning rates $\eta^{(1)}$ and $\eta^{(2)}$ are both set to 0.03, and the temperatures $\tau^{(1)}$ and $\tau^{(2)}$ are set to 1 and 10, respectively. In particular, SGLD uses the learning rate $\eta^{(1)}$ and the temperature $\tau^{(1)}$. We simulate 100,000 samples from each distribution and propose to estimate the correction every 100 iterations. The correction estimator is calculated based on the variance of 10 samples of $\widetilde U_1(x)$. The initial correction is set to 100 and the step size $\gamma_m$ for stochastic approximation is chosen as $\frac{1}{m}$. The correction factor $F$ is 1 in the first two examples.

\begin{figure*}[!ht]
  \centering
  \subfloat[$\scriptstyle \widetilde U_1(x)\sim \mathcal{N}(U_1(x), 2^2)$.]{\includegraphics[scale=0.21]{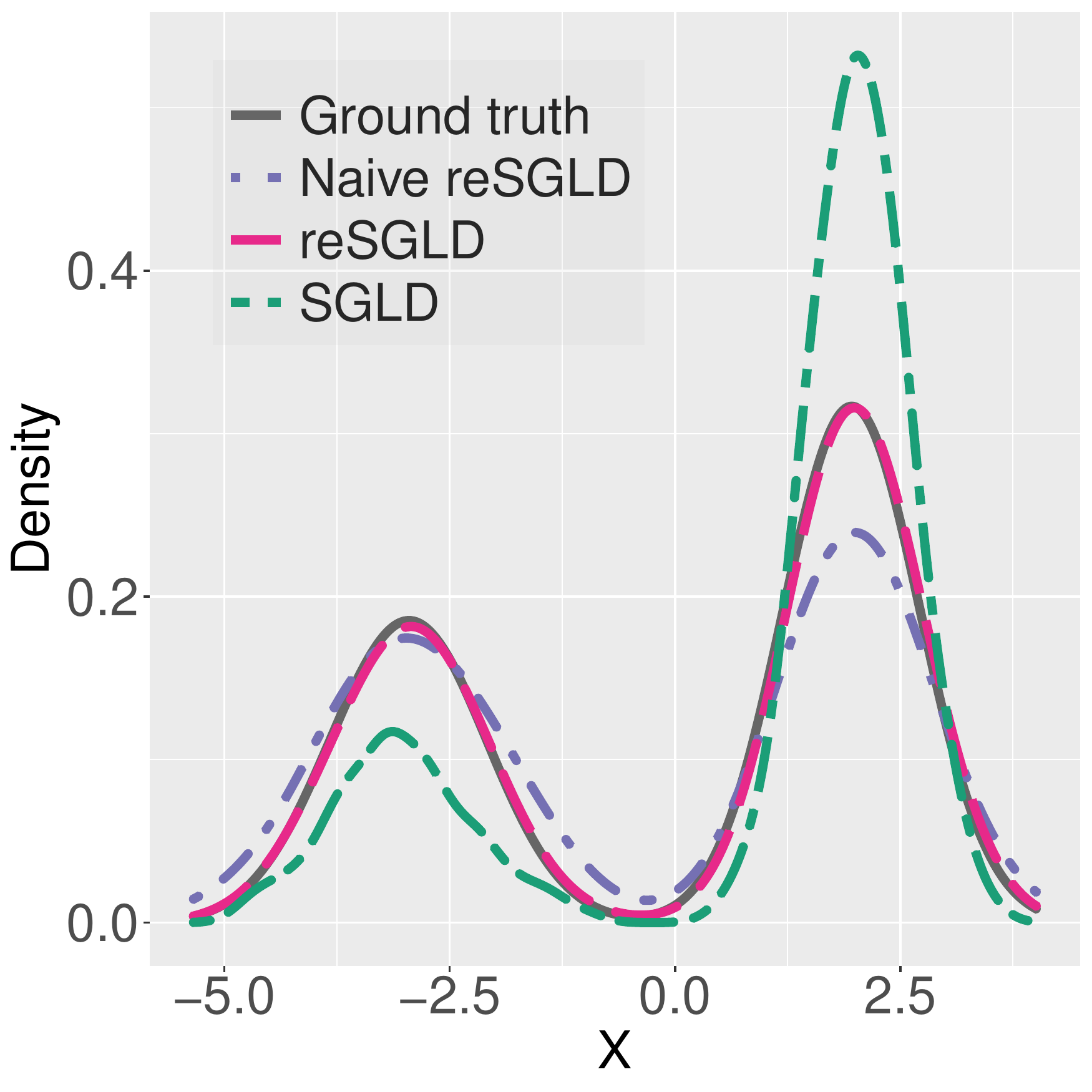}}\label{fig: 22a}\quad
  \subfloat[Convergence.]{\includegraphics[scale=0.21]{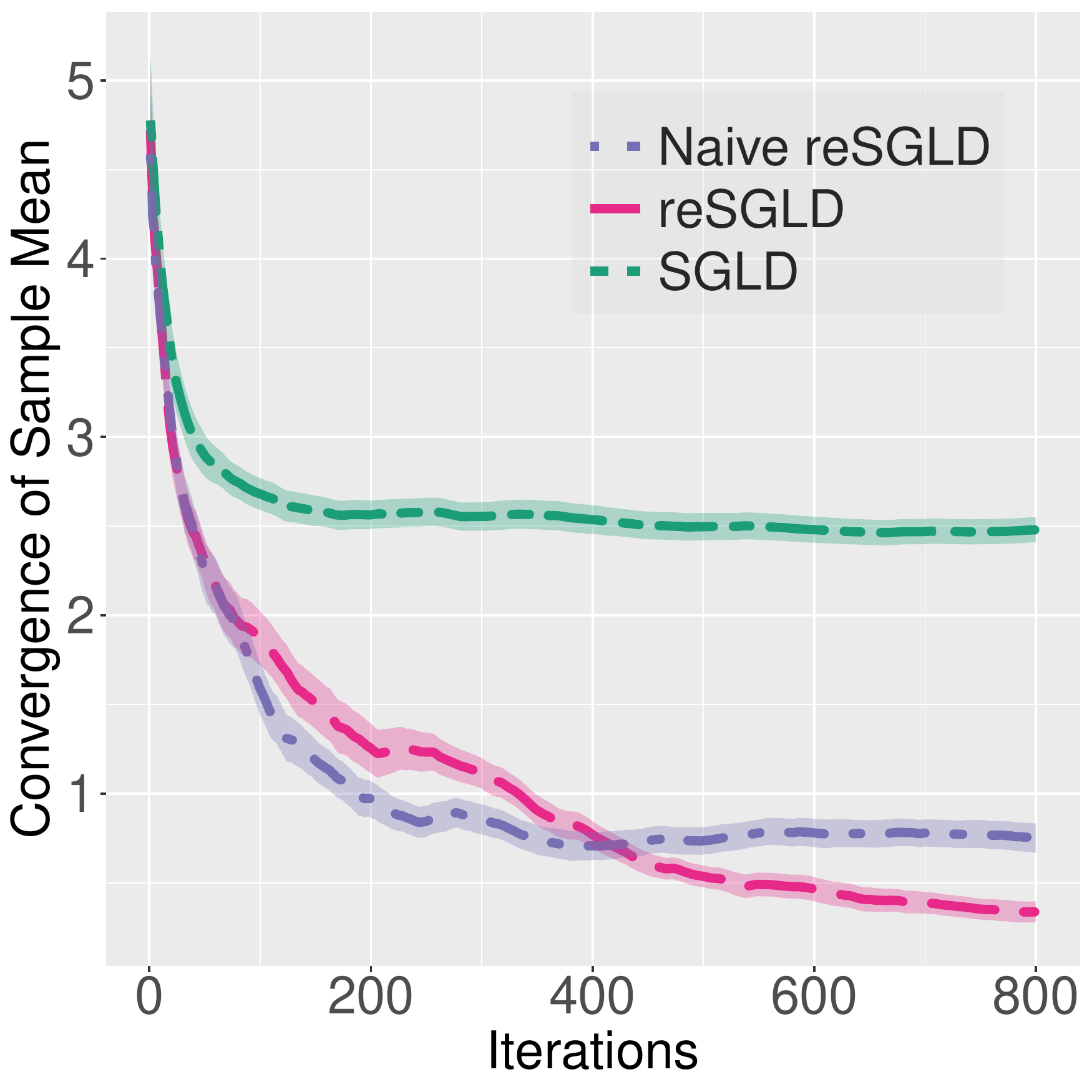}}\label{fig: 22b}\quad
  \subfloat[$\scriptstyle \widetilde U_2(x)\sim U_2(x)+ t(\nu=5)$.]{\includegraphics[scale=0.21]{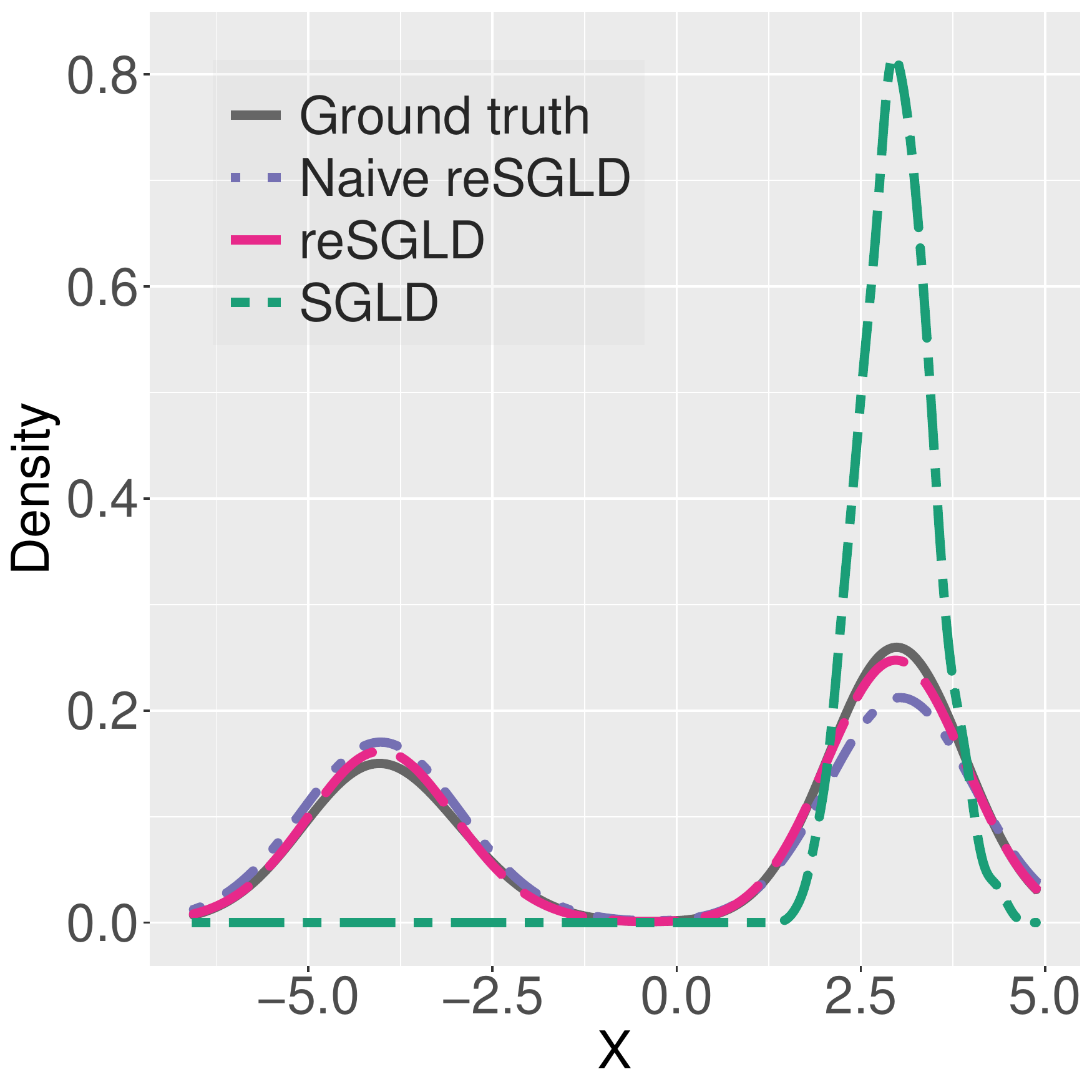}}\label{fig: 22c}\quad
  \subfloat[$\scriptstyle \widetilde U_3(x)\sim U_3(x)+7 t(\nu=10)$.]{\includegraphics[scale=0.22]{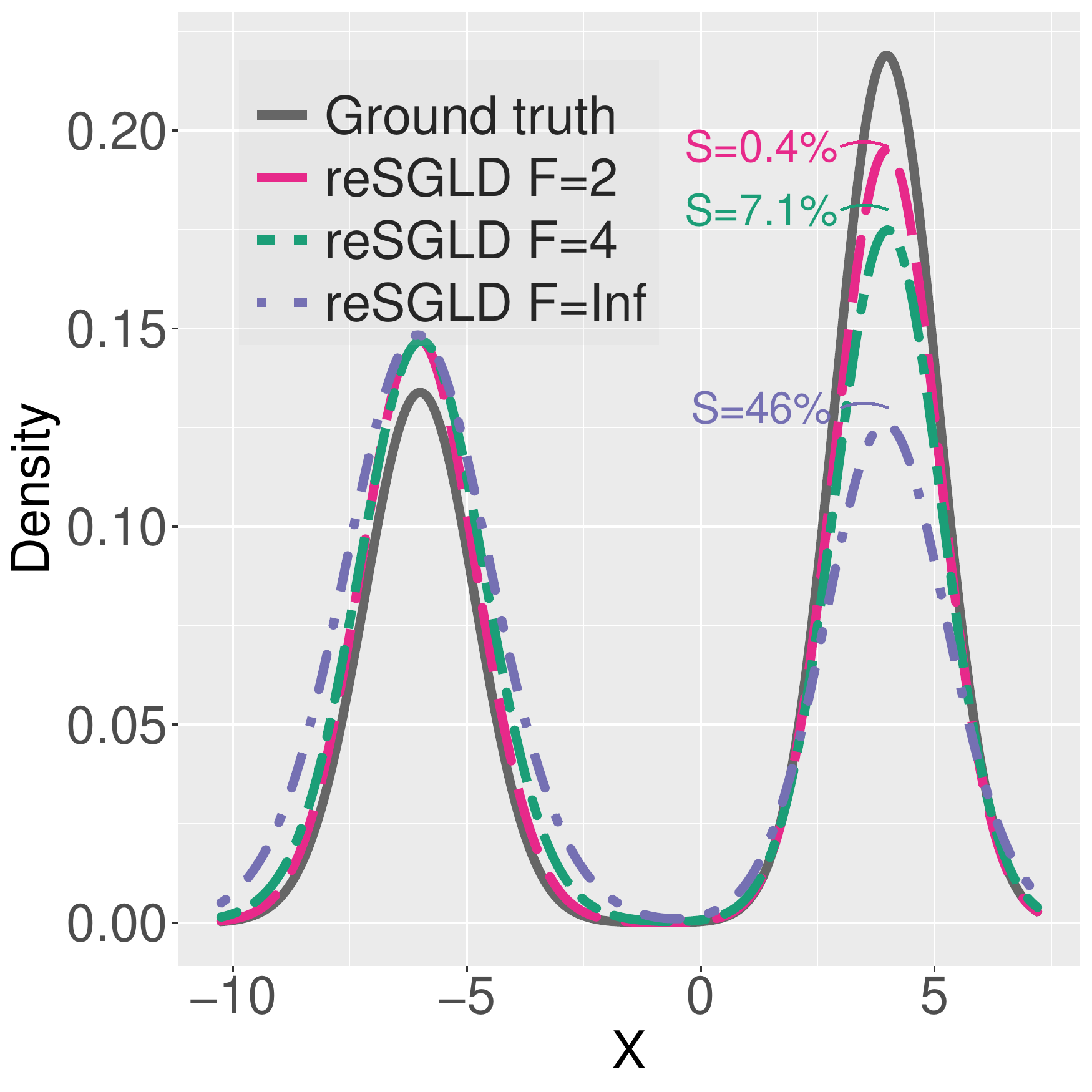}}\label{fig: 22d}
  \caption{Evaluation of reSGLD on Gaussian mixture distributions, where the na\"ive reSGLD doesn't make any corrections and reSGLD proposes to adaptively estimate the unknown corrections. In Fig.1(d), we omit SGLD because it gets stuck in a single mode.}\label{mixture}
\end{figure*}

We first demonstrate reSGLD on a simple Gaussian mixture distribution $e^{-U_1(x)}\sim 0.4 \mathcal{N}(-3, 0.7^2)+0.6\mathcal{N}(2, 0.5^2)$, where $U_1(x)$ is the energy function. We assume we can only obtain the unbiased energy estimator $\widetilde U_1(x)\sim \mathcal{N}(U_1(x), 2^2)$ and the corresponding stochastic gradient at each iteration.  From Fig.\ref{mixture}(a,b), we see that SGLD suffers from the local trap problem and takes a long time to converge. The na\"ive reSGLD algorithm alleviated the local trap problem, but is still far away from the ground truth distribution without a proper correction. The na\"ive reSGLD converges faster than reSGLD in the early phase due to a higher swapping rate, but ends up with a large bias when the training continues. By contrast, reSGLD successfully identifies the right correction through adaptive estimates and yields a close approximation to the ground truth distribution. The high-temperature chain serves as a bridge to facilitate the movement, and the local trap problem is greatly reduced.

In the second example, we relax the normality assumption to a heavy-tail distribution. Given a Gaussian mixture distribution $e^{-U_2(x)}\sim 0.4 \mathcal{N}(-4, 0.7^2)+0.6\mathcal{N}(3, 0.5^2)$, we assume that we can obtain the stochastic energy estimator $\widetilde U_2(x)\sim U_2(x)+t(\nu=5)$, where $t(\nu=5)$ denotes the Student's t-distribution with degree of freedom 5. We see from Fig.\ref{mixture}(c) that reSGLD still gives a good approximation to the ground true distribution while the others don't.

In the third example, we show a case when the correction factor $F$ is useful. We sample from a Gaussian mixture distribution $e^{-U_3(x)}\sim 0.4 \mathcal{N}(-6, 0.7^2)+0.6\mathcal{N}(4, 0.5^2)$. We note that the two modes are far away from each other and the distribution is more difficult to simulate. More interestingly, we assume $\widetilde U_3(x)\sim U_3(x)+7 t(\nu=10)$, which requires a large correction term and ends up with no swaps in the finite 100,000 iterations at $F=1$. In such a case, the unbiased algorithm behaves like the ordinary SGLD as in Fig.\ref{mixture}(c) and still suffers from the local trap problems. To achieve larger acceleration effects to avoid local traps and maintain accuracy, we try $F$ at $2$, $4$ and $\infty$ ($\text{Inf}$), where the latter is equivalent to the na\"ive reSGLD. We see from Fig.\ref{mixture}(d) that $F=2$ shows the best approximations, despite that the swapping rate $S$ is only $0.4\%$. Further increases on the acceleration effect via larger correction factors $F$ give larger swapping rates (7.1\% and 46\%) and potentially accelerate the convergence in the beginning. However, the biases become more significant as we increase $F$ and lead to larger errors in the end.

\subsection{Optimization of Image Data}
\label{SL}
We evaluate the adaptive replica exchange Monte Carlo on CIFAR10 and CIFAR100, which consist of 50,000 $32\times32$ RGB images for training and 10,000 images for testing. CIFAR10 and CIFAR100 have 10 classes and 100 classes, respectively. We adopt the well-known residual networks (ResNet) and wide ResNet (WRN)  as model architectures. We use the 20, 32, 56-layer ResNet (denoted as ResNet-20, et al.), WRN-16-8 and WRN-28-10, where, for example, WRN-16-8 denotes a ResNet that has 16 layers and is 8 times wider than the original. Inspired by the popular momentum stochastic gradient descent, we use stochastic gradient Hamiltonian Monte Carlo (SGHMC) as the baseline sampling algorithm and use the numerical method proposed by \cite{Saatci17} to reduce the tuning cost. We refer to the momentum stochastic gradient descent algorithm as mSGD and the adaptive replica exchange SGHMC algorithm as reSGHMC.

\begin{figure*}[!ht]
  \centering
  \subfloat[Corrections v.s. losses on CIFAR10]{\includegraphics[width=3.85cm, height=3.3cm]{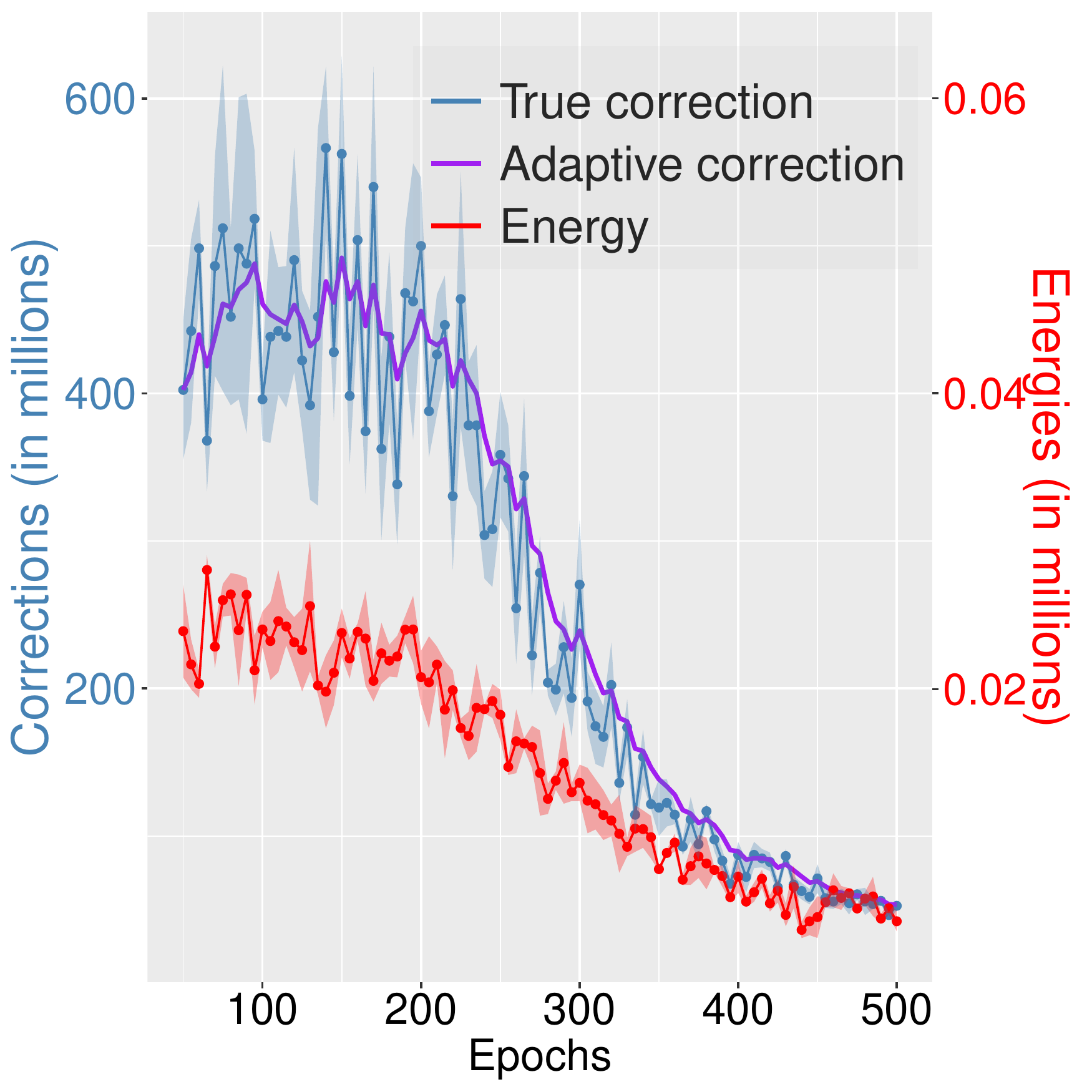}}\quad
  \subfloat[Corrections v.s. losses on CIFAR100]{\includegraphics[width=3.85cm, height=3.3cm]{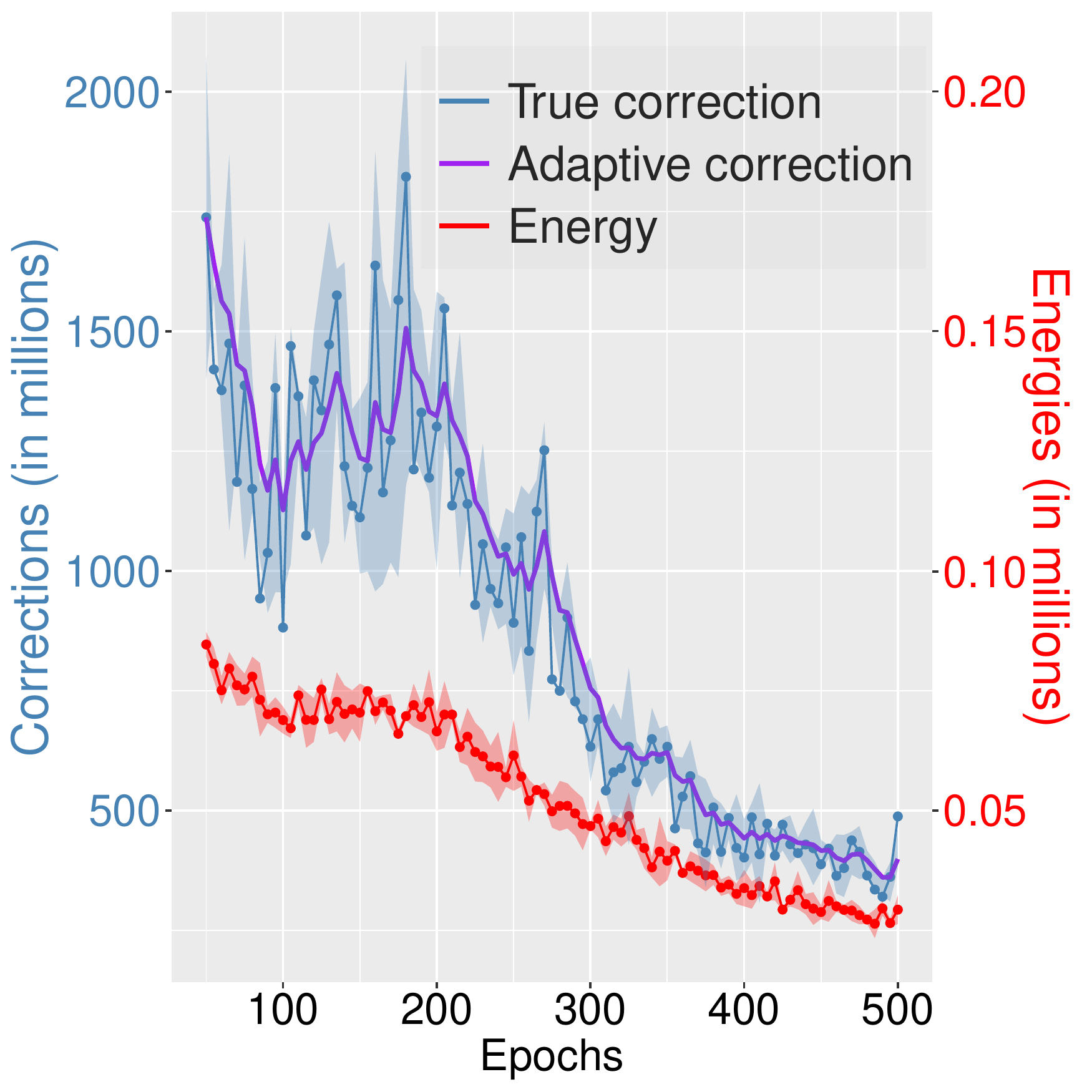}}\quad
  \subfloat[Effects of correction factors on CIFAR10]{\includegraphics[width=3.5cm, height=3.3cm]{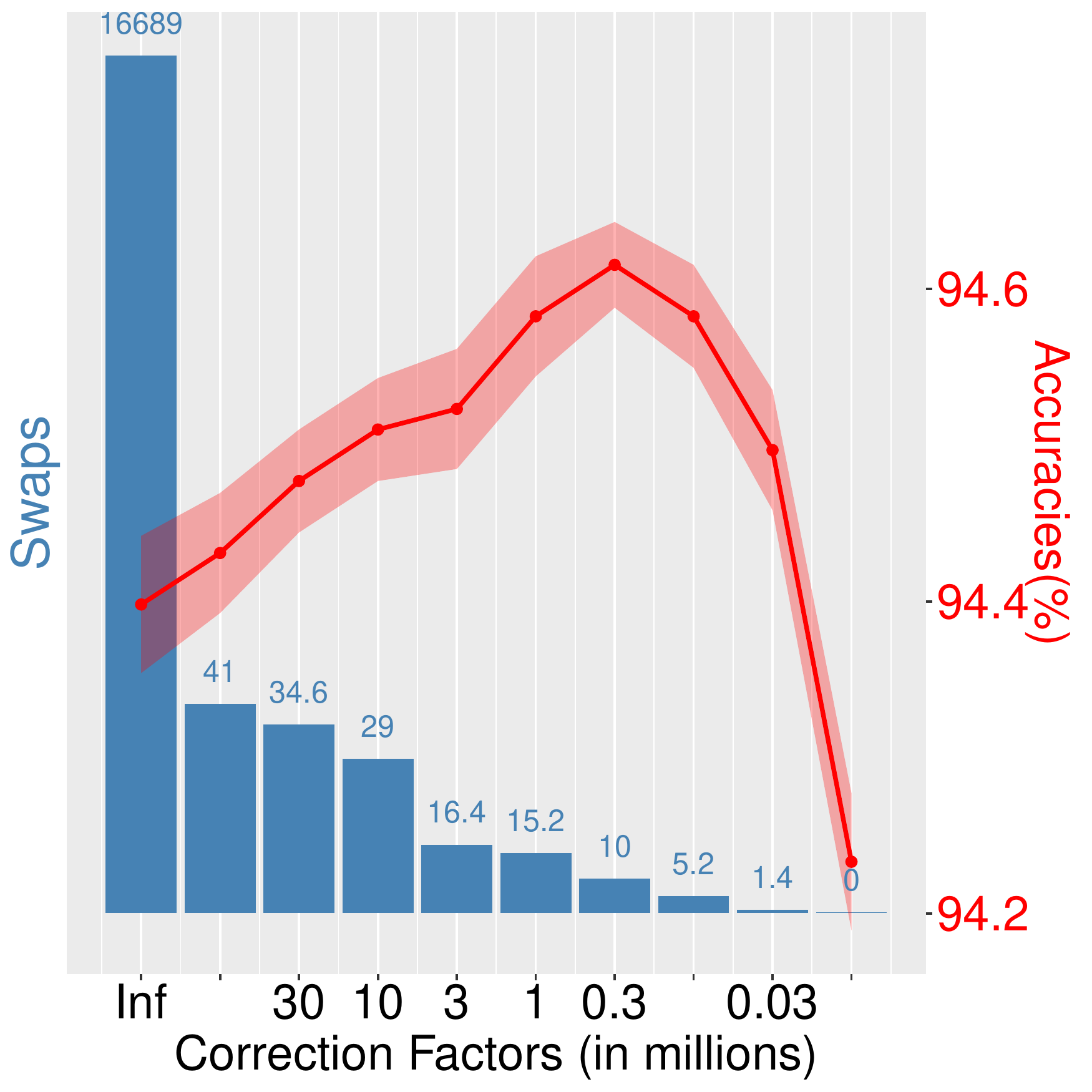}}\quad
  \subfloat[Effects of correction factors on CIFAR100]{\includegraphics[width=3.5cm, height=3.3cm]{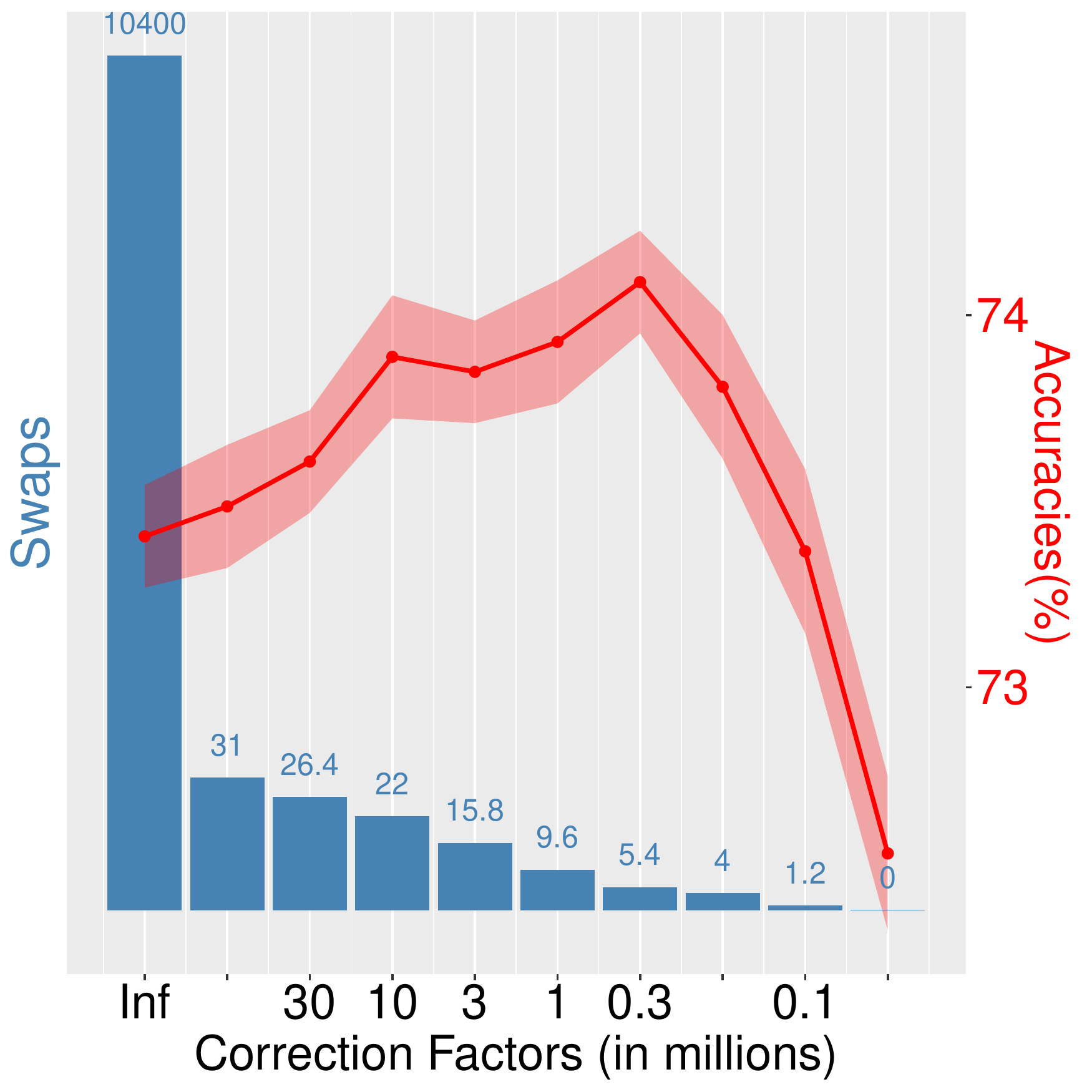}}
  \caption{Time-varying variances of the stochastic energy based on batch-size 256 on CIFAR10 \& CIFAR100 datasets.}
  \label{cifar_biases}
\end{figure*}

We first run several experiments to study the ideal corrections for the optimization of deep neural networks based on the fixed temperatures $\tau^{(1)}=0.01$ and $\tau^{(2)}=0.05$ \footnote[4]{Since data augmentation greatly increases the number of training points and leads to a much more concentrated posterior, the cold temperature effect is inevitable \cite{Florian2020, Aitchison2021}.}. We observe from Fig.(\ref{cifar_biases})(a, b) that the corrections are thousands of times larger than the energy losses, which implies that an exact correction leads to no swaps in practice and no acceleration can be achieved. The desire to obtain more acceleration effects drives us to manually shrink the corrections by increasing $F$ to increase the swapping rates, although we have to suffer from some model bias. 

\begin{table*}[ht]
\begin{sc}
\footnotesize
\caption{Prediction accuracies (\%) with different architectures on CIFAR10 and CIFAR100.}\label{nonconvex_funcs_icml20}
\vskip 0.15in
\begin{center} 
\begin{tabular}{c|ccc|ccc}
\hline
\multirow{2}{*}{Model} & \multicolumn{3}{c|}{CIFAR10} & \multicolumn{3}{c}{CIFAR100} \\
\cline{2-7}
 & \upshape{m}SGD & SGHMC & \lowercase{re}{SGHMC} & \upshape{m}SGD & SGHMC & reSGHMC  \\
\hline
\hline
ResNet-20 & 94.21$\pm$0.16 & 94.22$\pm$0.12 & \textbf{94.62$\pm$0.18} & 72.45$\pm$0.20 & 72.49$\pm$0.18 & \textbf{74.14$\pm$0.22}\\ 
ResNet-32 & 95.15$\pm$0.08 & 95.18$\pm$0.06 & \textbf{95.35$\pm$0.08} & 75.01$\pm$0.22 & 75.14$\pm$0.28 & \textbf{76.55$\pm$0.30} \\ 
ResNet-56 & 96.01$\pm$0.08 & 95.95$\pm$0.10 & \textbf{96.12$\pm$0.06} & 78.96$\pm$0.32 & 79.04$\pm$0.30 & \textbf{80.14$\pm$0.34}    \\
\hline
\hline
WRN-16-8 &  96.71$\pm$0.06 & 96.73$\pm$0.08 & \textbf{96.87$\pm$0.06} & 81.70$\pm$0.26 & 82.07$\pm$0.22 & \textbf{82.95$\pm$0.30}\\
WRN-28-10 & 97.33$\pm$0.08 & 97.32$\pm$0.06 & \textbf{97.42$\pm$0.06} & 83.79$\pm$0.18 & 83.76$\pm$0.14 & \textbf{84.38$\pm$0.18} \\
\hline
\end{tabular}
\end{center} 
\end{sc}
\vspace{-1em}
\end{table*}

We study the model performance by applying different correction factors $F$. We choose batch-size 256 and run the experiments within 500 epochs. We first tune the optimal hyperparameters for mSGD, SGHMC and the low-temperature chain of reSGHMC: we set the learning rate $\eta_k^{(1)}$ to 2e-6 in the first 200 epochs and decay it afterward by a factor of 0.984 every epoch; the low temperature follows an annealing schedule $\tau^{(1)}=\frac{0.01}{1.02^k}$ to accelerate the optimization; the weight decay is set to 25. Then, as to the high-temperature chain of reSGHMC, we use a larger learning rate $\eta^{(2)}_k=1.5 \eta^{(1)}_k$ and a higher temperature $\tau^{(2)}=5\tau^{(1)}$. We set $F$ as $F_0$ in the beginning and then adapt the value to counteract the annealing effect of the temperatures.
The variance estimator is updated each epoch based on the variance of 10 samples of the stochastic energies and the smoothing factor is set to $\gamma=0.3$ in (\ref{smoothing}). Consequently, the computations only increase by less than 5\%. In addition, we use a thinning factor $200$ and report all the results based on Bayesian model averaging. We repeat every experiment five times to obtain the mean and 2 standard deviations.

We see from Fig.\ref{cifar_biases}(c,d) that both datasets rely on a very large initial correction factor $F_0$ to yield good performance and the optimal initial correction factor $\widehat F_0$ is achieved at 3e5 \footnote[4]{The adoption of data augmentation leads to a cold posterior \cite{Florian2020, Aitchison2021} and the accompanying correction factor $F$ becomes larger.}. Empirically, we notice that the first five swaps provide the \emph{largest marginal improvement} in acceleration. A larger $F_0$ than $\widehat F_0$ leads to a larger swapping rate with more swaps and thus a larger acceleration effect, however, the performance still decreases as we increase $F_0$, implying that the biases start to dominate the error and the diminishing marginal improvement on the acceleration effect is no longer significant. We note that there is only one extra hyper-parameter, namely, the correction factor $F$, required to tune, and it is independent of the standard SGHMC. This shows that the tuning cost is acceptable. 

\begin{table}[ht]
\begin{sc}
\vskip -0.1in
\caption{Prediction accuracies (\%) with different batch sizes on CIFAR10 \& CIFAR100 using ResNet-20.}
\label{batch_effects}
\vskip 0.15in
\begin{center} 
\begin{tabular}{c|ccc}
\hline
Batch & \upshape{m}SGD & SGHMC & reSGHMC  \\
\hline
\hline
 &\multicolumn{3}{c}{CIFAR10} \\
256 & 94.21$\pm$0.16 & 94.22$\pm$0.12& 94.62$\pm$0.18 \\ 
1024 & 94.49$\pm$0.12 & 94.57$\pm$0.14& \textbf{95.01$\pm$0.16} \\ 
\hline
\hline
 &\multicolumn{3}{c}{CIFAR100} \\
256 & 72.45$\pm$0.20 & 72.49$\pm$0.18 & 74.14$\pm$0.21 \\ 
1024 & 73.31$\pm$0.18 & 73.23$\pm$0.20& \textbf{75.11$\pm$0.26} \\
\hline
\end{tabular}
\end{center} 
\end{sc}
\vspace{-1em}
\end{table}

To obtain a comprehensive evaluation of reSGHMC, we use the optimal correction factor for reSGHMC and test it on ResNet20, 32, 56, WRN-16-8 and WRN-28-10. From Table.\ref{nonconvex_funcs_icml20}, we see that reSGHMC consistently outperforms SGHMC and mSGD on both datasets, showing the robustness of reSGHMC to various model architectures. For CIFAR10, our method works better with ResNet-20 and ResNet-32 and improves the prediction accuracy by 0.4\% and 0.2\%, respectively. Regarding the other model architectures, it still slightly outperforms the baselines by roughly 0.1\%-0.2\%, although this dataset is highly optimized. Specifically, reSGHMC achieves the state-of-the-art 97.42\% accuracy with WRN-28-10 model. For CIFAR100, reSGHMC works particularly well based on various model architectures. It outperforms the baseline by as high as 1.5\% using ResNet-20 and ResNet-32, and around 1\% based on the other architectures. It also achieves the state-of-the-art 84.38\% based on WRN-28-10 on CIFAR100. 

We also conduct the experiments using larger batch sizes with ResNet-20 and report the results in Table.\ref{batch_effects}. We run the same iterations and keep the other setups the same. We find that a larger batch size significantly boosts the performance of reSGHMC by as much as 0.4\% accuracies on CIFAR10 and 1\% on CIFAR100, which shows the \emph{potential of using a large batch size} in the future.

\section{Conclusion and Future Work}

We propose the adaptive replica exchange SGMCMC algorithm and prove the accelerated convergence in terms of 2-Wasserstein distance. The theory implies an accuracy-acceleration trade-off and guides us to tune the correction factor $F$ to obtain the optimal performance. We support our theory with extensive experiments and obtain significant improvements over the vanilla SGMCMC algorithms on CIFAR10, CIFAR100.

For future works, it is promising to conduct variance reduction \cite{Xu18} of the stochastic noise to obtain a larger acceleration effect. From the computational perspective, it is also interesting to study parallel multi-chain reSGMCMC in larger machine learning tasks. 

\chapter{ACCELERATING CONVERGENCE OF REPLICA EXCHANGE STOCHASTIC GRADIENT MCMC VIA VARIANCE REDUCTION}

\label{vr_resgld_iclr}
\section{Introduction}

reSGLD \cite{deng2020} uses multiple processes based on stochastic gradient Langevin dynamics (SGLD) where interactions between different SGLD chains are conducted in a manner that encourages large jumps. In addition to the ideal utilization of parallel computation, the resulting process is able to jump to more informative modes for more robust uncertainty quantification. However, the noisy energy estimators in mini-batch settings lead to a large bias in the na\"{i}ve swaps, and a large correction is required to reduce the bias, which yields few effective swaps and insignificant accelerations. Therefore, how to reduce the variance of noisy energy estimators becomes essential in speeding up the convergence.

A long standing technique for variance reduction is the control variates method. The key to reducing the variance is to properly design correlated control variates so as to counteract some noise. Towards this direction, \cite{Dubey16, Xu18} proposed to update the control variate periodically for the stochastic gradient estimators and \cite{baker17} studied the construction of control variates using local modes. Despite the advantages in near-convex problems, a natural discrepancy between theory \cite{Niladri18, Xu18, SVRG_HMC_Zou} and practice \cite{kaiming15, bert} is \emph{whether we should avoid the gradient noise in non-convex problems}. To fill in the gap, we only focus on the variance reduction of noisy energy estimators to exploit the theoretical accelerations but no longer consider the variance reduction of the noisy gradients so that the empirical experience from stochastic gradient descents with momentum (\upshape{m}SGD) can be naturally imported.

In this chapter we propose the variance-reduced replica exchange stochastic gradient Langevin dynamics (VR-reSGLD) algorithm to accelerate convergence by reducing the variance of the noisy energy estimators. This algorithm not only \emph{shows the potential of exponential acceleration} via much more effective swaps in the non-asymptotic analysis but also \emph{demonstrates remarkable performance in practical tasks} where a limited time is required; while others \cite{Xu18, SAGA_LD_Zou} may only work well when the dynamics is sufficiently mixed and the discretization error becomes a major component. Empirically, we test the algorithm through extensive experiments and achieve state-of-the-art performance in both optimization and uncertainty estimates.

\begin{figure*}[!ht]
  \centering
  \vskip -0.1in
  \subfloat[Gibbs measures at three temperatures $\tau$.]{\includegraphics[scale=0.22]{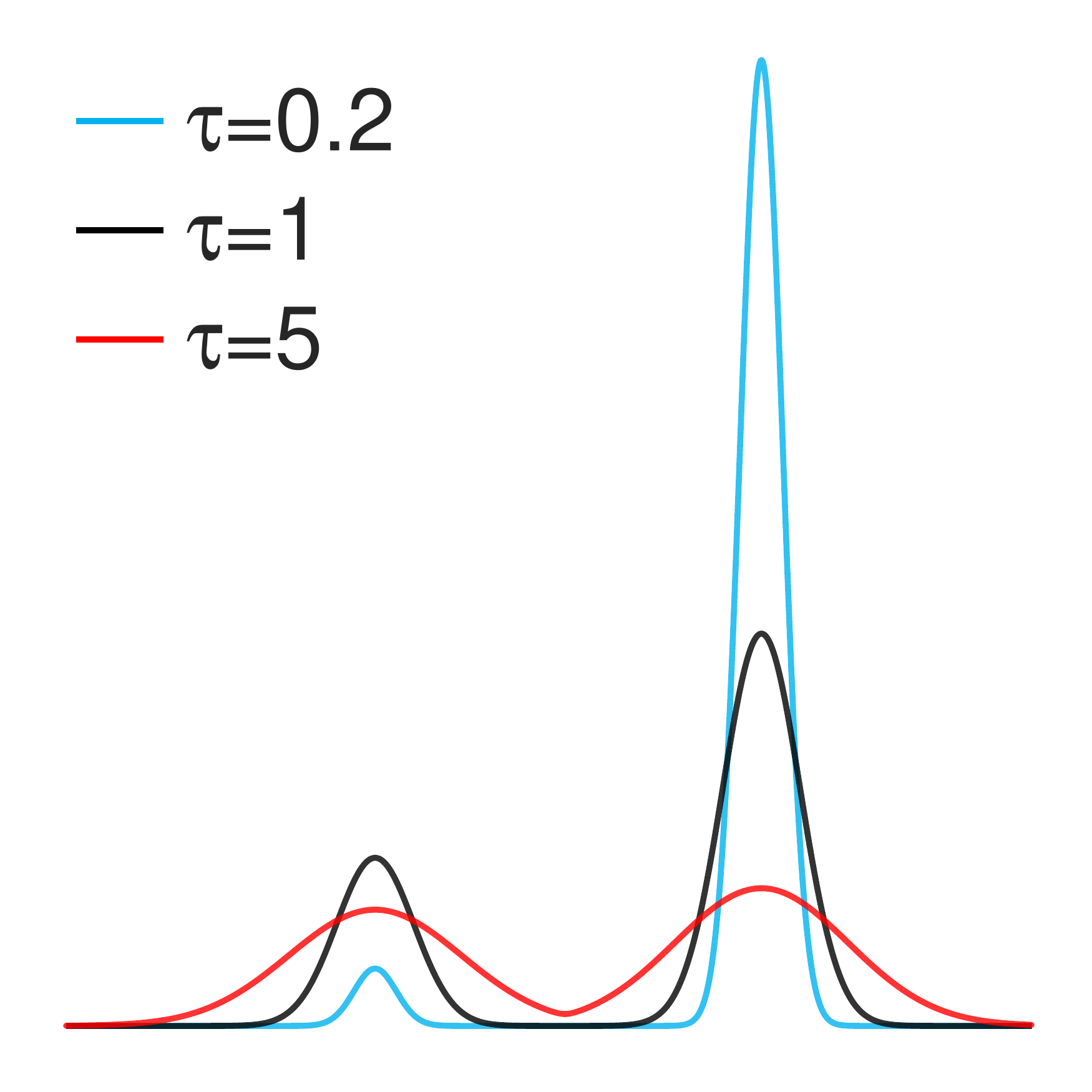}}\quad\quad
  \hspace{0.3cm}
  \subfloat[Sample trajectories on a energy landscape.]{\includegraphics[scale=0.16]{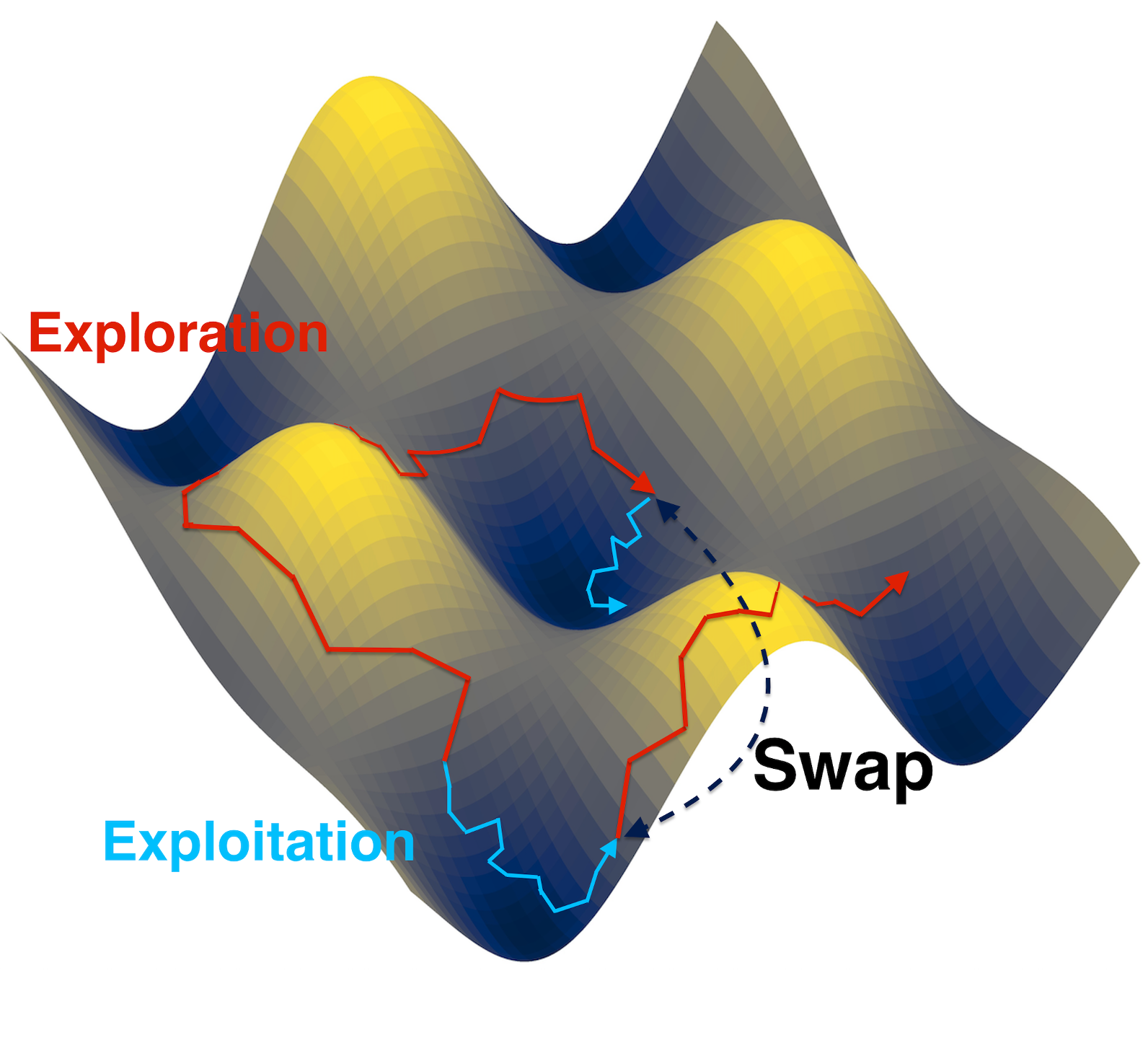}}\quad\quad
  \subfloat[Faster exponential convergence in $W_2$]{\includegraphics[scale=0.22]{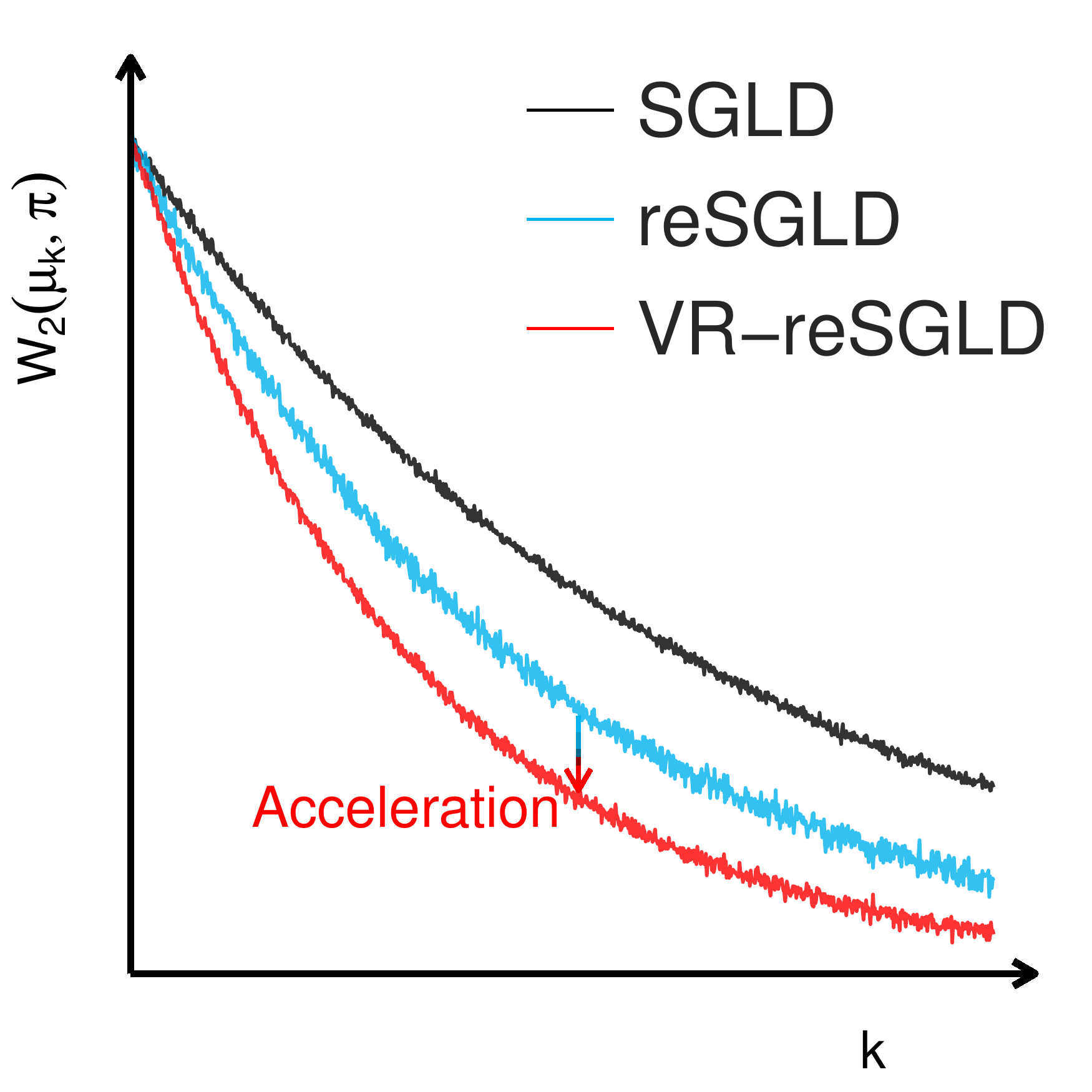}}
  \vspace{-0.5em}
  \caption{An illustration of replica exchange Monte Carlo algorithms for non-convex learning.}
  \label{demo_reld}
  \vspace{-1em}
\end{figure*}

\section{Preliminaries}
\label{reSGLD_section}

A common problem, in Bayesian inference,  is the simulation from a posterior $\mathrm{P}(\bbeta\mid\bm{X})\propto \mathrm{P}(\bbeta)\prod_{i=1}^N \mathrm{P}(\bx_i\mid\bbeta)$, where $\mathrm{P}(\bbeta)$ is a proper prior, $\prod_{i=1}^N \mathrm{P}(\bx_i\mid\bbeta)$ is the likelihood function and $N$ is the number of data points. When $N$ is large, the standard Langevin dynamics is too costly in evaluating the gradients. To tackle this issue, stochastic gradient Langevin dynamics (SGLD) \cite{Welling11} was proposed to make the algorithm scalable by approximating the gradient through a mini-batch data $B$ of size $n$ such that
\begin{equation}
     \bbeta_{k}=  \bbeta_{k-1}- \eta_k \frac{N}{n}\sum_{i\in B_{k}}\nabla  L( \bx_i\mid\bbeta_{k-1})+\sqrt{2\eta_k\tau} \bxi_{k},
\end{equation}
where $\bbeta_{k}\in\mathbb{R}^d$, $\tau$ denotes the temperature, $\eta_k$ is the learning rate at iteration $k$, $\bxi_{k}$ is a standard Gaussian vector, and $L(\cdot):=-\log \mathrm{P}(\bbeta\mid\bm{X})$ is the energy function. SGLD is known to converge weakly to a stationary Gibbs measure $\pi_{\tau}(\bbeta)\propto \exp\left(-L(\bbeta)/\tau\right)$ as $\eta_k$ decays to $0$ \cite{Teh16}.

The temperature $\tau$ is the key to accelerating the computations in multi-modal distributions. On the one hand, a high temperature flattens the Gibbs distribution $\exp\left(-{L(\bbeta)}/{\tau}\right)$ (see the red curve in Fig.\ref{demo_reld}(a)) and accelerates mixing by facilitating exploration of the whole domain, but the resulting distribution becomes much less concentrated around the global optima. On the other hand, a low temperature exploits the local region rapidly; however, it may cause the particles to stick in a local region for an exponentially long time, as shown in the blue curve in Fig.\ref{demo_reld}(a,b). To  bridge the gap between global exploration and local exploitation, \cite{deng2020} proposed the replica exchange SGLD algorithm (reSGLD), which consists of a low-temperature SGLD to encourage exploitation and a high-temperature SGLD to support exploration 
\begin{equation*}
\begin{split}
    \bbeta^{(1)}_{k} &=  \bbeta^{(1)}_{k-1}- \eta_k \frac{N}{n}\sum_{i\in B_{k}}\nabla  L( \bx_i\mid\bbeta^{(1)}_{k-1})+\sqrt{2\eta_k\tau^{(1)}} \bxi_{k}^{(1)} \\
    \bbeta^{(2)}_{k} &=  \bbeta^{(2)}_{k-1}- \eta_k \frac{N}{n}\sum_{i\in B_{k}}\nabla  L( \bx_i\mid\bbeta^{(2)}_{k-1})+\sqrt{2\eta_k\tau^{(2)}} \bxi_{k}^{(2)},
\end{split}
\end{equation*}
where the invariant measure is known to be $\pi(\bbeta^{(1)}, \bbeta^{(2)})\propto\exp\left(-\frac{L(\bbeta^{(1)})}{\tau^{(1)}}-\frac{L(\bbeta^{(2)})}{\tau^{(2)}}\right)$ as $\eta_k\rightarrow 0$ and $\tau^{(1)}< \tau^{(2)}$. Moreover, the two processes may swap the positions to allow tunneling between different modes. To avoid inducing a large bias in mini-batch settings, a corrected swapping rate $\widehat S$ is developed such that
\begin{equation*}
\label{vanilla_S}
\small
\widehat S=\exp\Big\{ \left(\frac{1}{\tau^{(1)}}-\frac{1}{\tau^{(2)}}\right)\Big(  \frac{N}{n}\sum_{i\in B_k} L(\bx_i\mid\bbeta^{(1)}_k)-  \frac{N}{n}\sum_{i\in B_k} L(\bx_i\mid\bbeta^{(2)}_k)-\frac{\left(\frac{1}{\tau^{(1)}}-\frac{1}{\tau^{(2)}}\right)\widehat \sigma^2}{F}\Big)\Big\},
\end{equation*}
where $2\widehat \sigma^2$  is an estimator of the variance of $\frac{N}{n}\sum_{i\in B_k} L(\bx_i\mid\bbeta^{(1)}_k)-  \frac{N}{n}\sum_{i\in B_k} L(\bx_i\mid\bbeta^{(2)}_k)$ and $F$ is the correction factor to balance between acceleration and bias. In other words, the parameters switch the positions from $(\bbeta_k^{(1)}, \bbeta_k^{(2)})$ to $(\bbeta_k^{(2)}, \bbeta_k^{(1)})$ with a probability $a(1\wedge \widehat S)\eta_k$, where the constant $a$ is the swapping intensity and can set to $\frac{1}{\eta_k}$ for simplicity.

From a probabilistic point of view, reSGLD is a discretization scheme of replica exchange Langevin diffusion (reLD) in mini-batch settings. Given a smooth test function $f$ and a (truncated) swapping-rate function $S$, the infinitesimal generator $\cL_{S}$ associated with the continuous-time reLD follows
\begin{equation*}
\begin{split}
    &\cL_{S}f(\bbeta^{(1)}, \bbeta^{(2)})=-\langle\nabla_{\bbeta^{(1)}}f(\bbeta^{(1)},\bbeta^{(2)}),\nabla L(\bbeta^{(1)})\rangle-\langle \nabla_{\bbeta^{(2)}}f(\bbeta^{(1)},\bbeta^{(2)}),\nabla L(\bbeta^{(2)})\rangle\\
    &\ \ \ +\tau^{(1)}\Delta_{\bbeta^{(1)}}f(\bbeta^{(1)},\bbeta^{(2)})+\tau^{(2)}\Delta_{\bbeta^{(2)}}f(\bbeta^{(1)},\bbeta^{(2)})+ aS(\bbeta^{(1)},\bbeta^{(2)})\cd (f(\bbeta^{(2)},\bbeta^{(1)})-f(\bbeta^{(1)},\bbeta^{(2)})),
\end{split}
\end{equation*}
where the last term arises from swaps and $\Delta_{\bbeta^{(\cdot)}}$ is the the Laplace operator with respect to $\bbeta^{(\cdot)}$. Note that the infinitesimal generator is closely related to Dirichlet forms in characterizing the evolution of a stochastic process. By standard calculations in Markov semigroups \cite{chen2018accelerating}, the Dirichlet form $\cE_{S}$ associated with the infinitesimal generator $\cL_{S}$ follows
\begin{equation}
\label{dirichlet_forms_main}
\begin{split}
    \cE_{S}(f)=&\underbrace{\int \Big(\tau^{(1)}\|\nabla_{\bbeta^{(1)}}f(\bbeta^{(1)}, \bbeta^{(2)})\|^2+\tau^{(2)}\|\nabla_{\bbeta^{(2)}}f(\bbeta^{(1)}, \bbeta^{(2)})\|^2 \Big)d\pi(\bbeta^{(1)},\bbeta^{(2)})}_{\text{vanilla term } \cE(f)}\\
    &\ +\underbrace{\frac{a}{2}\int S(\bbeta^{(1)},\bbeta^{(2)})\cd (f(\bbeta^{(2)},\bbeta^{(1)})-f(\bbeta^{(1)},\bbeta^{(2)}))^2d\pi(\bbeta^{(1)},\bbeta^{(2)})}_{\text{acceleration term}},
\end{split}
\end{equation}
which leads to a strictly positive acceleration under mild conditions and is crucial for the exponentially accelerated convergence in the $W_2$ distance (see Fig.\ref{demo_reld}(c)). However, the acceleration depends on the swapping-rate function $S$ and becomes much smaller given a noisy estimate of $\frac{N}{n}\sum_{i\in B} L(\bx_i\mid\bbeta)$ due to the demand of large corrections to reduce the bias.

\section{Variance Reduction in Replica Exchange SGLD}
\label{vrresgld}

The desire to obtain more effective swaps and larger accelerations drives us to design more efficient energy estimators. A na\"{i}ve idea would be to apply a large batch size $n$, which reduces the variance of the noisy energy estimator proportionally. However, this comes with a significantly increased memory overhead and computations and therefore is inappropriate for big data problems.

A natural idea to propose more effective swaps is to reduce the variance of the noisy energy estimator $L(B\mid\bbeta^{(h)})=\frac{N}{n}\sum_{i\in B}L(\bx_i\mid\bbeta^{(h)})$ for $h\in\{1,2\}$. Considering an unbiased estimator $L(B\mid\widehat\bbeta^{(h)})$ for $\sum_{i=1}^N L(\bx_i\mid\widehat\bbeta^{(h)})$ and a constant $c$, we see that a new estimator $\widetilde L(B\mid \bbeta^{(h)})$, which follows
\begin{equation}
    \widetilde L(B\mid\bbeta^{(h)})= L(B\mid\bbeta^{(h)}) +c\left( L(B\mid\widehat\bbeta^{(h)}) -\sum_{i=1}^N L (\bx_i\mid \widehat \bbeta^{(h)})\right),
\end{equation}
is still the unbiased estimator for $\sum_{i=1}^N L(\bx_i\mid \bbeta^{(h)})$. By decomposing the variance, we have
\begin{equation*}
\small
\begin{split}
    &\Var(\widetilde L(B\mid\bbeta^{(h)}))=\Var\left( L(B\mid\bbeta^{(h)})\right)+c^2 \Var\left( L(B\mid\widehat\bbeta^{(h)})\right)+2c\text{Cov}\left( L(B\mid\bbeta^{(h)}),  L(B\mid\widehat\bbeta^{(h)})\right).
\end{split}
\end{equation*}
In such a case, $\Var(\widetilde L(B\mid\bbeta^{(h)}))$ achieves the minimum variance $(1-\rho^2)\Var(L(B\mid\bbeta^{(h)}))$ given $c^{\star}:=-\frac{\text{Cov}( L(B\mid\bbeta^{(h)}),  L(B\mid\widehat\bbeta^{(h)}))}{ \Var(L(B\mid\widehat \bbeta^{(h)}))}$, where $\text{Cov}(\cdot, \cdot)$ denotes the covariance and $\rho$ is the correlation coefficient of $ L(B\mid\bbeta^{(h)})$ and $ L(B\mid\widehat\bbeta^{(h)})$. To propose a correlated control variate, we follow \cite{SVRG} and update $\widehat \bbeta^{(h)}=\bbeta^{(h)}_{m\lfloor \frac{k}{m}\rfloor}$ every $m$ iterations. Moreover, the optimal $c^{\star}$ is often unknown in practice. To handle this issue, a well-known solution \cite{SVRG} is to fix $c=-1$ given a high correlation $\mid\rho\mid$ of the estimators and then we can present the VR-reSGLD algorithm in Algorithm \ref{alg_iclr21}. Since the exact variance for correcting the stochastic swapping rate is unknown and even time-varying, we follow \cite{deng2020} and propose to use stochastic approximation \cite{RobbinsM1951} to adaptively update the unknown variance.

\begin{algorithm}[tb]
  \caption{Variance-reduced replica exchange stochastic gradient Langevin dynamics (VR-reSGLD). The learning rate and temperature can be set to dynamic to speed up the computations. A larger smoothing factor $\gamma$ captures the trend better but becomes less robust. }
  \label{alg_iclr21}
\begin{algorithmic}
\STATE{\textbf{Input } The initial parameters $\bbeta_0^{(1)}$ and $\bbeta_0^{(2)}$, learning rate $\eta$, temperatures $\tau^{(1)}$ and $\tau^{(2)}$, correction factor $F$ and smoothing factor $\gamma$. }
\REPEAT
  \STATE{\textbf{Parallel sampling} \text{Randomly pick a mini-batch set $B_{k}$ of size $n$.}}
  \begin{equation}
      \bbeta^{(h)}_{k} =  \bbeta^{(h)}_{k-1}- \eta \frac{N}{n}\sum_{i\in B_{k}}\nabla  L( \bx_i\mid\bbeta^{(h)}_{k-1})+\sqrt{2\eta\tau^{(h)}} \bxi_{k}^{(h)}, \text{ for } h\in\{1,2\}.
  \end{equation}
  \STATE{\textbf{Variance-reduced energy estimators} \small{Update $\tiny{\widehat L^{(h)}=\sum_{i=1}^N L\left(\bx_i\mid \bbeta^{(h)}_{m\lfloor \frac{k}{m}\rfloor}\right)}$ every $m$ iterations.}}
  \begin{equation}
  \label{__vr_loss}
      \widetilde L(B_{k}\mid\bbeta_{k}^{(h)})=\frac{N}{n}\sum_{i\in B_{k}}\left[ L(\bx_i\mid \bbeta_{k}^{(h)}) - L\left(\bx_i\mid \bbeta^{(h)}_{m\lfloor \frac{k}{m}\rfloor}\right) \right]+\widehat L^{(h)} , \text{ for } h\in\{1,2\}. 
  \end{equation}
  \IF{$k\ \text{mod}\ m=0$} 
  \STATE{{Update $\widetilde \sigma^2_{k} = (1-\gamma)\widetilde \sigma^2_{k-m}+\gamma \sigma^2_{k}$, }
  \STATE{where $\sigma_k^2$ is an estimate for $\Var\left(\widetilde L( B_k\mid\bbeta_{k}^{(1)})-\widetilde L( B_{k}\mid\bbeta_{k}^{(2)})\right)$.}}
  \ENDIF

  \STATE{\textbf{Bias-reduced swaps} \small{Swap $ \bbeta_{k+1}^{(1)}$ and $ \bbeta_{k+1}^{(2)}$ if $u<\widetilde S_{\eta,m,n}$, where $u\sim \text{Unif }[0,1]$, and $\widetilde S_{\eta,m,n}$ follows}}
  \begin{equation}
  \label{stochastic_jump_rate}
  \small
      \textstyle \widetilde S_{\eta,m,n}=\exp\left\{ \left(\frac{1}{\tau^{(1)}}-\frac{1}{\tau^{(2)}}\right)\left(  \widetilde L( B_{k+1}\mid\bbeta_{k+1}^{(1)})-  \widetilde L( B_{k+1}\mid\bbeta_{k+1}^{(2)})-\frac{1}{F}\left(\frac{1}{\tau^{(1)}}-\frac{1}{\tau^{(2)}}\right)\widetilde \sigma^2_{m\lfloor \frac{k}{m}\rfloor}\right)\right\}.
  \end{equation}
  \UNTIL{$k=k_{\max}$.}
\STATE{\textbf{Output:}  The low-temperature process $\{\bbeta_{i\mathbb{T}}^{(1)}\}_{i=1}^{\lfloor k_{\max}/\mathbb{T}\rfloor}$, where $\mathbb{T}$ is the thinning factor.}
\end{algorithmic}
\end{algorithm}

\subsection{Variants of VR-reSGLD} 

The number of iterations $m$ to update the control variate $\widehat \bbeta^{(h)}$ gives rise to a trade-off in computations and variance reduction. A small $m$ introduces a highly correlated control variate at the cost of expensive computations; a large $m$, however, may yield a less correlated control variate and setting $c=-1$ fails to reduce the variance. In spirit of the adaptive variance in \cite{deng2020} to estimate the unknown variance, we explore the idea of the adaptive coefficient $\widetilde {c}_k=(1-\gamma_k)\widetilde {c}_{k-m}+\gamma_k c_k$ such that the unknown optimal $c^{\star}$ is well approximated. We present the adaptive VR-reSGLD in Algorithm \ref{adaptive_alg} in Appendix \ref{adaptive_c} and show empirically later that the adaptive VR-reSGLD leads to a significant improvement over VR-reSGLD for the less correlated estimators.

A parallel line of research is to exploit the SAGA algorithm \cite{SAGA} in the study of variance reduction. Despite the most effective performance in variance reduction \cite{Niladri18}, the SAGA type of sampling algorithms require an excessively memory storage of $\mathcal{O}(Nd)$, which is too costly for big data problems. Therefore, we leave the study of the lightweight SAGA algorithm inspired by \cite{pracSVRG, Dongruo} for future works.

\subsection{Related Work} 

Although our VR-reSGLD is, in spirit, similar to VR-SGLD \cite{Dubey16, Xu18}, it differs from VR-SGLD in two aspects: First, VR-SGLD conducts variance reduction on the gradient and only shows promises in the nearly log-concave distributions or when the Markov process is sufficiently converged; however, our VR-reSGLD solely focuses on the variance reduction of the energy estimator to propose more effective swaps, and therefore we can import the empirical experience in hyper-parameter tuning from \upshape{m}SGD to our proposed algorithm. Second, VR-SGLD doesn't accelerate the continuous-time Markov process but only focuses on reducing the discretization error; VR-reSGLD possesses a larger acceleration term in the Dirichlet form (\ref{dirichlet_forms_main}) and shows a potential in exponentially speeding up the convergence of the continuous-time process in the early stage, in addition to the improvement on the discretization error. In other words, our algorithm is not only theoretically sound but also more empirically appealing for a wide variety of problems in non-convex learning.

\section{Theoretical Properties}

The large variance of noisy energy estimators directly limits the potential of the acceleration and significantly slows down the convergence compared to the replica exchange Langevin dynamics. As a result, VR-reSGLD may lead to a more efficient energy estimator with a much smaller variance.

\begin{lemma}[Variance-reduced energy estimator. Informal version of Lemma \ref{vr-estimator} in the appendix]
\label{vr-estimator_main}
Under the smoothness
and dissipativity assumptions \ref{assump: lip and alpha beta_ICLR21} and \ref{assump: dissipitive} in Appendix \ref{prelim}, the variance of the variance-reduced energy estimator $\widetilde L(B\mid\bbeta^{(h)})$, where $h\in\{1,2\}$, is upper bounded by
\begin{equation*}
    \Var\left(\widetilde L(B\mid\bbeta^{(h)})\right)\leq \min\Big\{ \mathcal{O}\left(\frac{m^2 \eta}{n}\right), \Var\Big(\frac{N}{n}\sum_{i\in B} L(\bx_i\mid \bbeta^{(h)})\Big)+\Var\Big(\frac{N}{n}\sum_{i\in B} L(\bx_i\mid \widehat\bbeta^{(h)})\Big)\Big\},
\end{equation*}
where the detailed $\mathcal{O}(\cdot)$ constants is shown in Lemma \ref{vr-estimator} in the appendix.
\end{lemma}

The analysis shows the variance-reduced estimator $\widetilde L(B\mid\bbeta^{(h)})$ yields a much-reduced variance given a smaller learning rate $\eta$ and a smaller $m$ for updating control variates based on the batch size $n$. Although the truncated swapping rate $S_{\eta, m, n}=\min\{1, \widetilde S_{\eta, m, n}\}$ satisfies the ``stochastic'' detailed balance given an unbiased swapping-rate estimator $\widetilde S_{\eta, m, n}$ \cite{deng2020} \footnote[2]{\cite{Andrieu09, Matias19} achieve a similar result based on the unbiased likelihood estimator for the Metropolis-hasting algorithm. See section 3.1 \cite{Matias19} for details.}, it doesn't mean the efficiency of the swaps is not affected. By contrast, we can show that the number of swaps may become \emph{exponentially smaller on average}.

\begin{lemma}[Variance reduction for larger swapping rates. Informal version of Lemma \ref{exp_S} in the appendix] \label{exp_S_main_body} Given a large enough batch size $n$, the variance-reduced energy estimator $\widetilde L(B_{k}\mid\bbeta_{k}^{(h)})$ yields a truncated swapping rate that satisfies
\begin{equation}
\label{larger_SSS}
     \E[S_{\eta, m, n}]\approx \min\Big\{1, S(\bbeta^{(1)}, \bbeta^{(2)})\Big(\mathcal{O}\Big(\frac{1}{n^2}\Big)+e^{-\mathcal{O}\left(\frac{m^2\eta}{n}+\frac{1}{n^2}\right)}\Big)\Big\},
\end{equation}
\end{lemma}
where $S(\bbeta^{(1)}, \bbeta^{(2)})$ is the deterministic swapping rate defined in Appendix \ref{exp_acc}. The proof is shown in Lemma.\ref{exp_S} in Appendix \ref{exp_acc}. Note that the above lemma doesn't require the normality assumption. As $n$ goes to infinity, where the asymptotic normality holds, the RHS of (\ref{larger_SSS}) changes to $\min\Big\{1, S(\bbeta^{(1)}, \bbeta^{(2)})e^{-\mathcal{O}\left(\frac{m^2\eta}{n}\right)}\Big\}$, which becomes exponentially larger as we use a smaller update frequency $m$ and learning rate $\eta$. Since the continuous-time reLD induces a jump operator in the infinitesimal generator, the resulting Dirichlet form potentially leads to a 
much larger acceleration term which linearly depends on the swapping rate $S_{\eta, m, n}$ and yields a faster exponential convergence. Now we are ready to present the first main result.

\begin{theorem}[Exponential convergence. Informal version of Theorem \ref{exponential decay_iclr} in the appendix]\label{exponential decay_main_body_iclr}
Under the smoothness
and dissipativity assumptions \ref{assump: lip and alpha beta_ICLR21} and \ref{assump: dissipitive}, the probability measure associated with reLD at time $t$, denoted as $\nu_t$, converges exponentially fast to the invariant measure $\pi$:
\begin{equation}
    W_2(\nu_t,\pi) \leq  D_0 \exp\left\{-t\left(1+\delta_{ S_{\eta, m, n}}\right)/c_{\text{LS}}\right\},
\end{equation}
where $D_0$ is a constant depending on the initialization, $\delta_{ S_{\eta, m, n}}:=\inf_{t>0}\frac{\cE_{ S_{\eta, m, n}}(\sqrt{\frac{d\n_t}{d\pi}})}{\cE(\sqrt{\frac{d\n_t}{d\pi}})}-1\geq 0$ depends on $S_{\eta, m, n}$, $\cE_{ S_{\eta, m, n}}$ and $\cE$ are the Dirichlet forms based on the swapping rate $S_{\eta, m, n}$ and are defined in Eq.(\ref{dirichlet_forms_main}), $c_{\text{LS}}$ is the constant of the log-Sobolev inequality for reLD without swaps.
\end{theorem}{}

Note that $S_{\eta,m,n}=0$ leads to the same performance as the standard Langevin diffusion and $\delta_{ S_{\eta, m, n}}$ is strictly positive when $\frac{d\n_t}{d\pi}$ is asymmetric \cite{chen2018accelerating}; given a smaller $\eta$ and $m$ or a large $n$, the variance becomes much reduced according to Lemma \ref{vr-estimator_main}, yielding a much larger truncated swapping rate by Lemma \ref{exp_S_main_body} and a faster exponential convergence to the invariant measure $\pi$ compared to reSGLD.

\section{Experiments}

\subsection{Simulations of Gaussian Mixture Distributions \label{subsec:Toy-multi-mo}}

We first study the proposed variance-reduced replica exchange stochastic gradient Langevin dynamics algorithm (VR-reSGLD) on a Gaussian mixture distribution \cite{Dubey16}. The distribution follows from $x_{i}\mid\beta\sim0.5\text{N}(\beta,\sigma^{2})+0.5\text{N}(\phi-\beta,\sigma^{2})$,
where $\phi=20$, $\sigma=5$ and $\beta=-5$. We use a training dataset of size $N=10^{5}$ and propose to estimate the posterior distribution over $\beta$. We compare the performance of VR-reSGLD against that of the standard stochastic gradient Langevin dynamics (SGLD), and replica exchange SGLD (reSGLD).

\begin{figure}
\center
\subfloat[Trace plot for $\bbeta^{(1)}$\label{fig:Ex1_recovery_svrg_csgld_theta_1}]{\includegraphics[width=0.24\columnwidth]{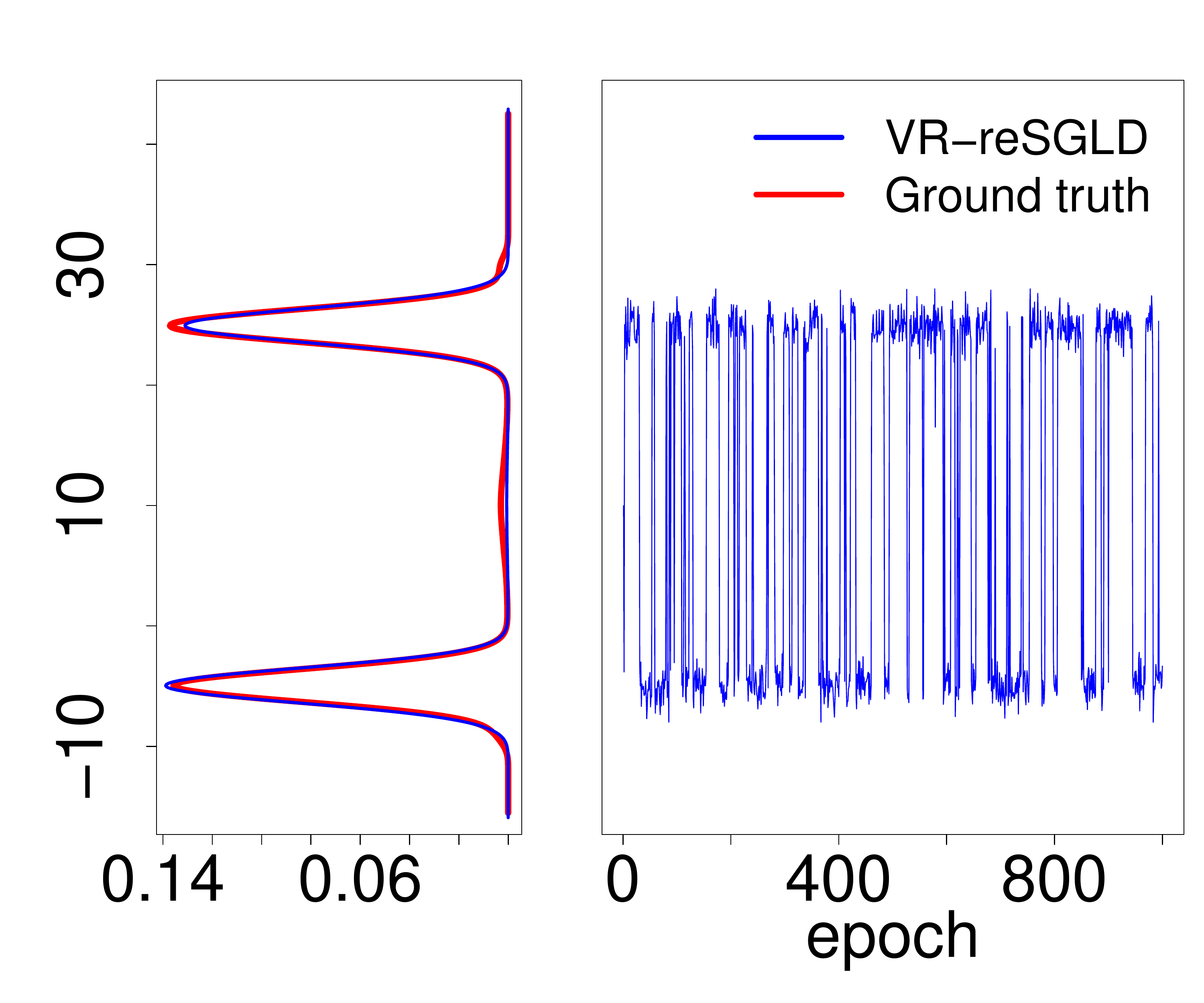}}
\enskip
\subfloat[Trace plot for $\bbeta^{(1)}$\label{fig:Ex1_recovery}]{\includegraphics[width=0.24\columnwidth]{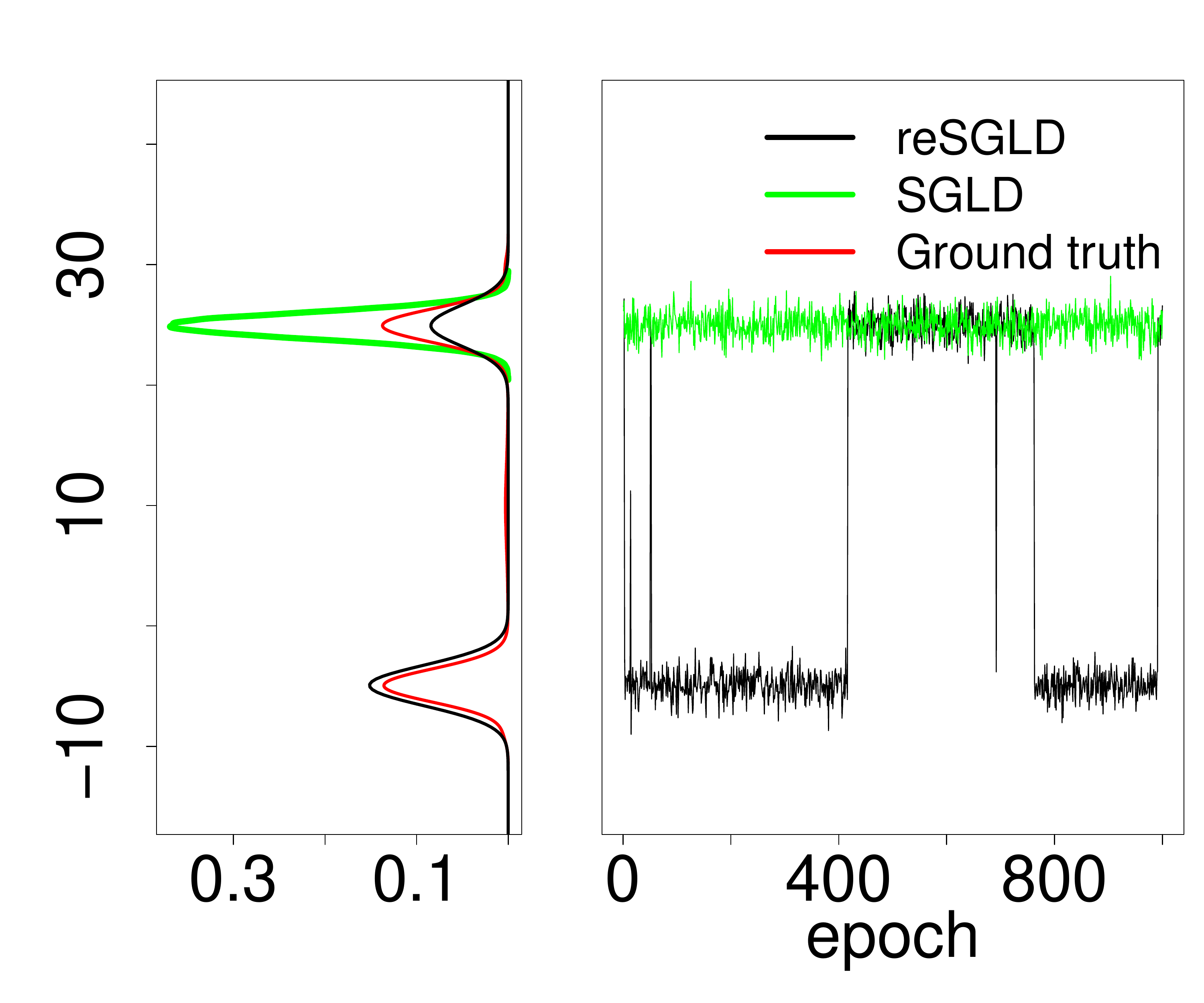}}
\enskip
\subfloat[Paths of $\log_{10}\widetilde{\sigma}^{2}$\label{fig:Ex1_variancetraceplot}]{\includegraphics[width=0.21\columnwidth]{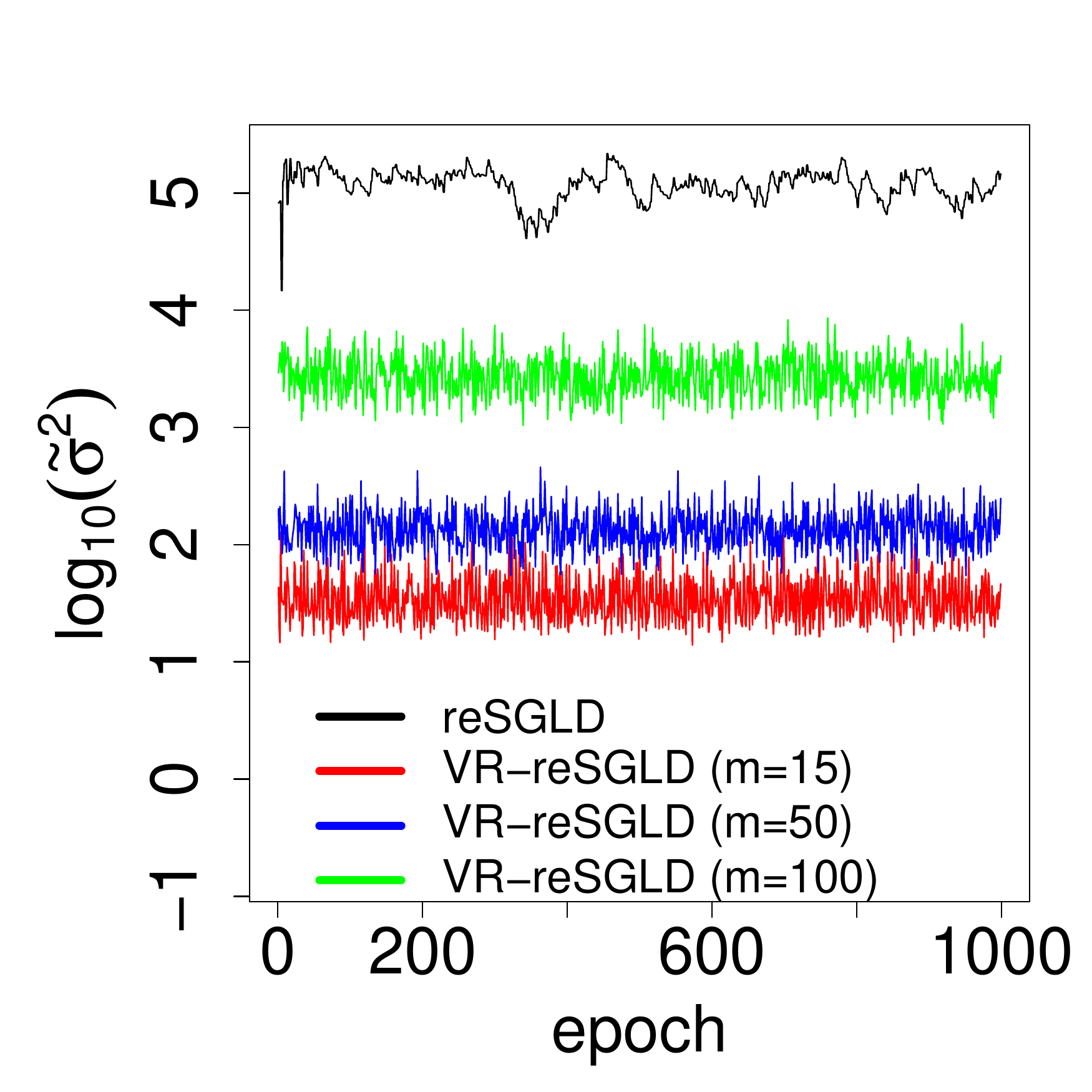}}
\enskip
\subfloat[Contour of $\log_{10}\widetilde{\sigma}^{2}$ \label{fig:Ex1_tau2_m_relation}]{\includegraphics[width=0.22\columnwidth]{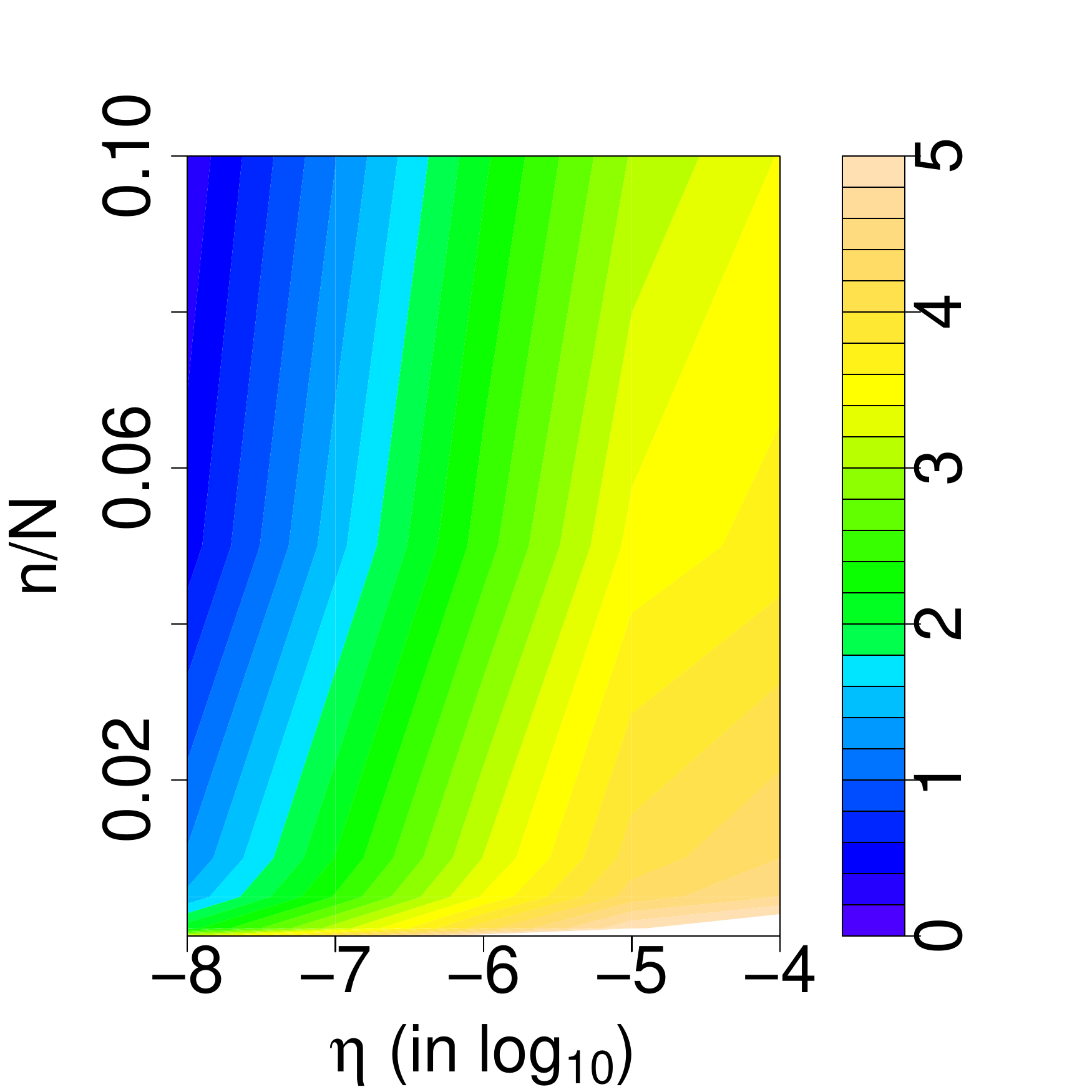}}
\caption{Trace plots, KDEs of $\bbeta^{(1)}$, and sensitivity study of $\widetilde{\sigma}^2$ with respect to $m, \eta$ and $n$.}
\vspace{-1em}
\end{figure}

In Figs \ref{fig:Ex1_recovery_svrg_csgld_theta_1} and \ref{fig:Ex1_recovery},
we present trace plots and kernel density estimates (KDE) of samples generated 
from   VR-reSGLD with $m=40$, $\tau^{(1)}=10$ \footnote[2]{We choose $\tau^{(1)}=10$ instead of $1$ to avoid peaky modes for ease of illustration.}, $\tau^{(2)}=1000$,
$\eta=1e-7$, and $F=1$; reSGLD adopt the same hyper-parameters except for $F=100$ because a smaller $F$ may fail to propose any swaps;
SGLD uses $\eta=1e-7$ and $\tau=10$. As the posterior
density is intractable, we consider a ground truth
by running replica exchange Langevin dynamics with long enough iterations. We observe
that VR-reSGLD is able to fully recover the posterior density, and successfully jump between
the two modes passing the energy barrier frequently enough. By contrast, SGLD, initialized at $\beta_{0}=30$,
is attracted to the nearest mode and fails to escape throughout
the run; reSGLD manages to jump between the 
two modes, however,  $F$ is chosen as large as $100$, which induces a large bias and only yields three to five swaps and exhibits the metastability issue. 
In Figure \ref{fig:Ex1_variancetraceplot}, we present the evolution
of the variance for VR-reSGLD over a range of different $m$
and compare it with reSGLD. We see that the variance reduction mechanism has
successfully reduced the variance by hundreds of times. In Fig \ref{fig:Ex1_tau2_m_relation},
we present the sensitivity study of $\tilde{\sigma}^{2}$
as a function of the ratio $n/N$ and the
learning rate $\eta$; for this estimate we average out $10$ realizations
of VR-reSGLD, and our results agree with the theoretical analysis in Lemma \ref{vr-estimator_main}.

\subsection{Non-convex Optimization for Image Data}
\label{nonconvex_optimization}

We further test the proposed algorithm on CIFAR10 and CIFAR100. We choose the 20, 32, 56-layer residual networks as the training models and denote them by ResNet-20, ResNet-32, and ResNet-56, respectively. Considering the wide adoption of \upshape{m}SGD, stochastic gradient Hamiltonian Monte Carlo (SGHMC) is selected as the baseline. We refer to the standard replica exchange SGHMC algorithm as reSGHMC and the variance-reduced reSGHMC algorithm as VR-reSGHMC. We also include another baseline called cyclical stochastic gradient MCMC (cycSGHMC), which proposes a cyclical learning rate schedule. To make a fair comparison, we test the variance-reduced replica exchange SGHMC algorithm with cyclic learning rates and refer to it as cVR-reSGHMC.

\begin{figure*}[!ht]
  \centering
  \subfloat[\footnotesize{CIFAR10: Original v.s. proposed (m=50)} ]{\includegraphics[width=3.8cm, height=3.8cm]{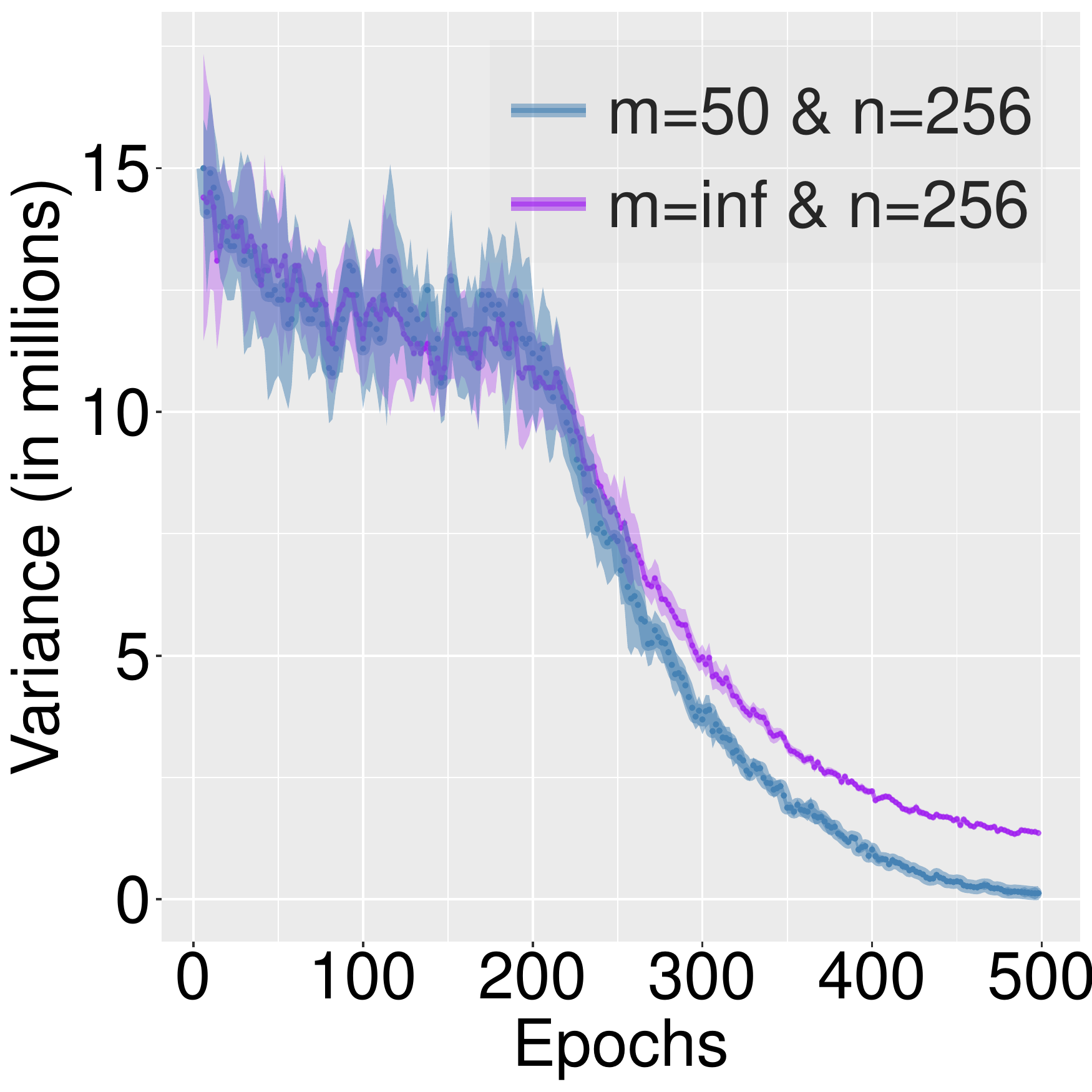}}\label{fig: c2a}\enskip
  \subfloat[\footnotesize{CIFAR100: Original v.s. proposed (m=50)} ]{\includegraphics[width=3.8cm, height=3.8cm]{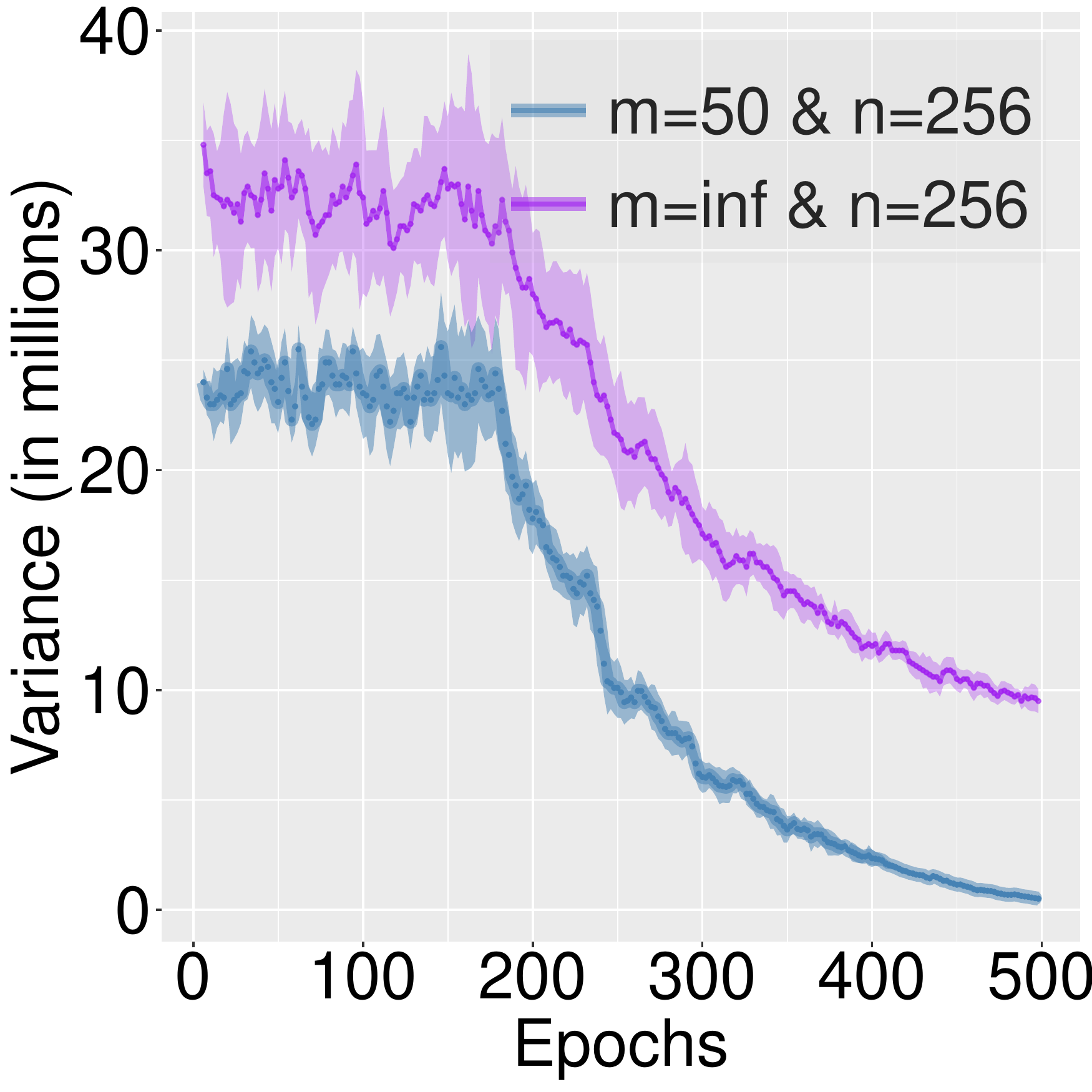}}\label{fig: c2b}\enskip
  \subfloat[Variance reduction setups on CIFAR10]{\includegraphics[width=3.8cm, height=3.8cm]{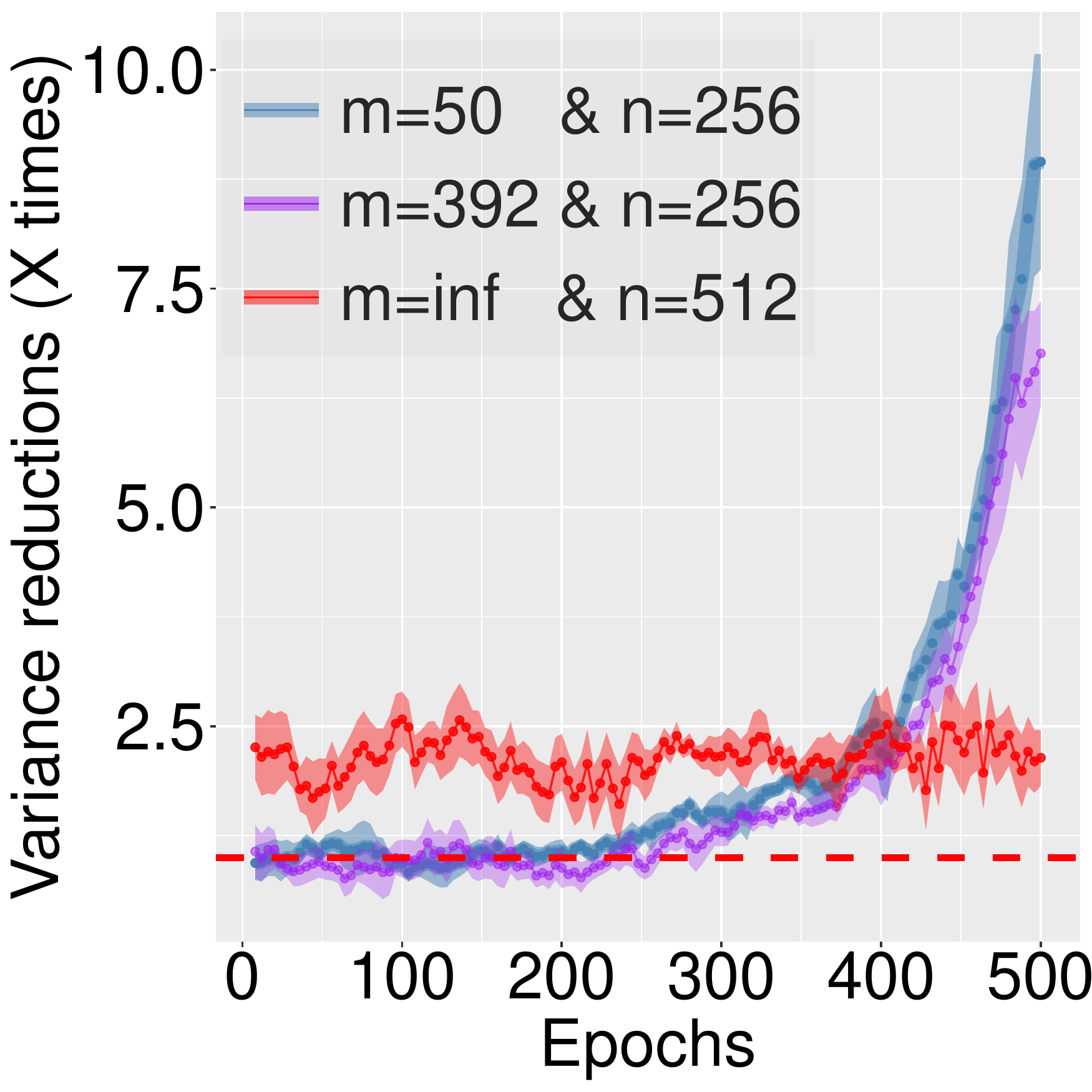}}\label{fig: c2c}\enskip
  \subfloat[Variance reduction setups on CIFAR100]{\includegraphics[width=3.8cm, height=3.8cm]{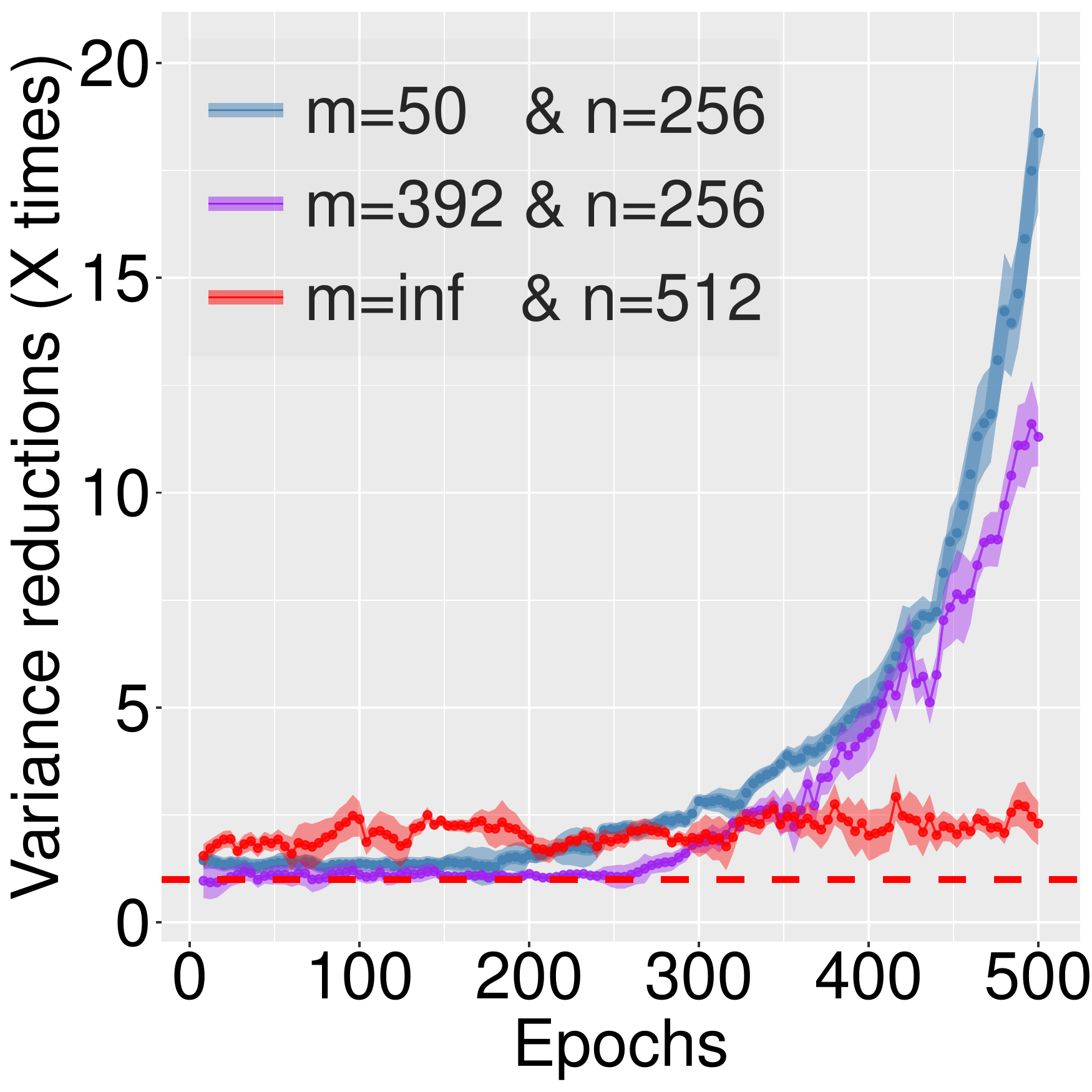}}\label{fig: c2d}
  \caption{Variance reduction on the noisy energy estimators on CIFAR10 \& CIFAR100 datasets.}
  \label{cifar_biases_v2}
\end{figure*}

We run \upshape{m}SGD, SGHMC and (VR-)reSGHMC for 500 epochs. For these algorithms, we follow a setup from \cite{deng2020}. We fix the learning rate $\eta_k^{(1)}=\text{2e-6}$ in the first 200 epochs and decay it by 0.984 afterwards. For SGHMC and the low-temperature processes of (VR-)reSGHMC, we anneal the temperature following $\tau_k^{(1)}=0.01 / 1.02^k$ in the beginning and keep it fixed after the burn-in steps; regarding the high-temperature process, we set $\eta_k^{(2)}=1.5\eta_k^{(1)}$ and $\tau_k^{(2)}=5\tau_k^{(1)}$. The initial correction factor $F_0$ is fixed at $1.5e5$ \footnote[4]{The adoption of data augmentation leads to a cold posterior \cite{Florian2020, Aitchison2021} and the accompanying correction factor $F$ becomes larger.}. The thinning factor $\mathbb{T}$ is set to $256$. In particular for cycSGHMC, we run the algorithm for 1000 epochs and choose the cosine learning rate schedule with 5 cycles; $\eta_0$ is set to $\text{1e-5}$; we fix the temperature 0.001 and the threshold $0.7$ for collecting the samples. Similarly, we propose the cosine learning rate for cVR-reSGHMC with 2 cycles and run it for 500 epochs using the same temperature 0.001. We only study the low-temperature process for the replica exchange algorithms. Each experiment is repeated five times to obtain the mean and 2 standard deviations.

We evaluate the performance of variance reduction using VR-reSGHMC and compare it with reSGHMC. We first increase the batch size $n$ from 256 to 512 for reSGHMC and notice that the reduction of variance is around 2 times (see the red curves in Fig.\ref{cifar_biases_v2}(c,d)). Next, we try $m=50$ and $n=256$ for the VR-reSGHMC algorithm, which updates the control variates every 50 iterations. As shown in Fig.\ref{cifar_biases_v2}(a,b), during the first 200 epochs, where the largest learning rate is used, the variance of VR-reSGHMC is slightly reduced by 37\% on CIFAR100 and doesn't make a difference on CIFAR10. However, as the learning rate and the temperature decrease, the reduction of the variance gets more significant. We see from  Fig.\ref{cifar_biases_v2}(c,d) that the reduction of variance can be \emph{up to 10 times on CIFAR10 and 20 times on CIFAR100}. This is consistent with our theory proposed in Lemma \ref{vr-estimator_main}. The reduction of variance based on VR-reSGHMC starts to outperform the baseline with $n=512$ when the epoch is higher than 370 on CIFAR10 and 250 on CIFAR100. We also try $m=392$, which updates the control variates every 2 epochs, and find a similar pattern.

For computational reasons, we choose $m=392$ and $n=256$ for (c)VR-reSGHMC and compare them with the baseline algorithms. With the help of swaps between two SGHMC chains, reSGHMC already obtains remarkable performance \cite{deng2020} and five swaps often lead to an optimal performance. However, VR-reSGHMC still outperforms reSGHMC by around 0.2\% on CIFAR10 and 1\% improvement on CIFAR100  (Table.\ref{cifar_all_results}) and \emph{the number of swaps is increased to around a hundred under the same setting}. We also try cyclic learning rates and compare cVR-reSGHMC with cycSGHMC, we see cVR-reSGHMC outperforms cycSGHMC significantly even if cycSGHMC is running 1000 epochs, which may be more costly than cVR-reSGHMC due to the lack of mechanism in parallelism. Note that cVR-reSGHMC keeps the temperature the same instead of annealing it as in VR-reSGHMC, which is more suitable for uncertainty quantification.

\begin{table*}[ht]
\fontsize{10}{11}
\begin{sc}
\caption{Prediction accuracies (\%) based on Bayesian model averaging. In particular, \upshape{m}SGD and SGHMC run 500 epochs using a single chain; cycSGHMC run 1000 epochs using a single chain; replica exchange algorithms run 500 epochs using two chains with different temperatures.}\label{cifar_all_results}
\vskip 0.2in
\begin{center} 
\begin{tabular}{c|ccc|ccc}
\hline
\multirow{2}{*}{Method} & \multicolumn{3}{c|}{CIFAR10} & \multicolumn{3}{c}{CIFAR100} \\
\cline{2-7}
 & ResNet20 & ResNet32  & ResNet56 & ResNet20 & ResNet32 & ResNet56 \\
\hline
\hline
\upshape{m}SGD & 94.07$\pm$0.11 & 95.11$\pm$0.07 & 96.05$\pm$0.21 & 71.93$\pm$0.13 & 74.65$\pm$0.20  & 78.76$\pm$0.24  \\
SGHMC & 94.16$\pm$0.13 & 95.17$\pm$0.08 & 96.04$\pm$0.18 & 72.09$\pm$0.14 & 74.80$\pm$0.19  & 78.95$\pm$0.22 \\ 
\hline
 
\upshape{re}SGHMC & 94.56$\pm$0.23 & 95.44$\pm$0.16 & 96.15$\pm$0.17 & 73.94$\pm$0.34 & 76.38$\pm$0.23  & 79.86$\pm$0.26 \\ 
\scriptsize{VR-\upshape{re}SGHMC} & \textbf{94.84$\pm$0.11}  & \textbf{95.62$\pm$0.09} &  \textbf{96.32$\pm$0.15}  & \textbf{74.83$\pm$0.18} & \textbf{77.40$\pm$0.27}  & \textbf{80.62$\pm$0.22}   \\
\hline
\scriptsize{\upshape{cyc}SGHMC} &  94.61$\pm$0.15  &  95.56$\pm$0.12  & 96.19$\pm$0.17 & 74.21$\pm$0.22 & 76.60$\pm$0.25  & 80.39$\pm$0.21  \\
\scriptsize{\upshape{c}VR-\upshape{re}SGHMC} & \textbf{94.91$\pm$0.10}  & \textbf{95.64$\pm$0.13} &  \textbf{96.36$\pm$0.16}  & \textbf{75.02$\pm$0.19} & \textbf{77.58$\pm$0.21}  & \textbf{80.50$\pm$0.25}   \\
\hline
\end{tabular}
\end{center} 
\end{sc}
\end{table*}

\subsubsection{Training Cost}
\label{cost}

The batch size of $n=512$ almost doubles the training time and memory, which becomes too costly in larger experiments. A frequent update of control variates using $m=50$ is even more time-consuming and is not acceptable in practice. The choice of $m$ gives rise to a tradeoff between computational cost and variance reduction. As such, we choose $m=392$, which still obtains significant reductions of the variance at the cost of 40\% increase on the training time. Note that when we set $m=2000$, the training cost is only increased by 8\% while the variance reduction can be still at most 6 times on CIFAR10 and 10 times on CIFAR100.

\subsubsection{Adaptive Coefficient}
\label{adaptive_c}
We study the correlation coefficient of the noise from the current parameter $\bbeta_k^{(h)}$, where $h\in\{1,2\}$, and the control variate $\bbeta^{(h)}_{m\lfloor \frac{k}{m}\rfloor}$. As shown in Fig.\ref{cifar_adaptive_c_figs}, the correlation coefficients are only around -0.5 due to the large learning rate in the early period. This implies that VR-reSGHMC may overuse the noise from the control variates and thus fails to fully exploit the potential in variance reduction. In spirit to the adaptive variance, we try the adaptive correlation coefficients to capture the pattern of the time-varying correlation coefficients and present it in Algorithm \ref{adaptive_alg}. 
\begin{algorithm}[tb]
  \caption{Adaptive variance-reduced replica exchange SGLD. The learning rate and temperature can be set to dynamic to speed up the computations. A larger smoothing factor $\gamma$ captures the trend better but becomes less robust.}
  \label{adaptive_alg}
\begin{algorithmic}
\footnotesize
\STATE{\textbf{Input } Initial parameters $\bbeta_0^{(1)}$ and $\bbeta_0^{(2)}$, learning rate $\eta$ and temperatures $\tau^{(1)}$ and $\tau^{(2)}$, correction factor $F$.}
\REPEAT
  \STATE{\textbf{Parallel sampling} \text{Randomly pick a mini-batch set $B_{k}$ of size $n$.}}
  \begin{equation*}
      \bbeta^{(h)}_{k} =  \bbeta^{(h)}_{k-1}- \eta \frac{N}{n}\sum_{i\in B_{k}}\nabla  L( \bx_i\mid\bbeta^{(h)}_{k-1})+\sqrt{2\eta\tau^{(h)}} \bxi_{k}^{(h)}, \text{ for } h\in\{1,2\}.
  \end{equation*}
  \STATE{\textbf{Variance-reduced energy estimators} Update $\tiny{\widehat L^{(h)}=\sum_{i=1}^N L\left(\bx_i\mid \bbeta^{(h)}_{m\lfloor \frac{k}{m}\rfloor}\right)}$ every $m$ iterations.}
  \begin{equation*}
  \small
      \widetilde L(B_{k}\mid\bbeta_{k}^{(h)})=\frac{N}{n}\sum_{i\in B_{k}} L(\bx_i\mid \bbeta_{k}^{(h)}) + \widetilde {c}_k\cdot\left[ \frac{N}{n}\sum_{i\in B_{k}} L\left(\bx_i\mid \bbeta^{(h)}_{m\lfloor \frac{k}{m}\rfloor}\right) -\widehat L^{(h)}\right] , \text{ for } h\in\{1,2\}. 
  \end{equation*}
  \IF{$k\ \text{mod}\ m=0$} 
  \STATE{Update $\widetilde \sigma^2_{k} = (1-\gamma)\widetilde \sigma^2_{k-m}+\gamma \sigma^2_{k}$, where $\sigma_k^2$ is an estimate for $\Var\left(\widetilde L( B_k\mid\bbeta_{k}^{(1)})-\widetilde L( B_{k}\mid\bbeta_{k}^{(2)})\right)$.}
  \STATE{Update $\widetilde {c}_k=(1-\gamma)\widetilde {c}_{k-m}+\gamma c_k$, where $c_k$ is an estimate for $-\frac{\text{Cov}\big( L(B\mid\bbeta_k^{(h)}),  L\big(B\mid\bbeta_{m\lfloor \frac{k}{m}\rfloor}^{(h)}\big)\big)}{ \Var\big(L\big(B\mid\bbeta_{m\lfloor \frac{k}{m}\rfloor}^{(h)}\big)\big)}$.}
  \ENDIF
   
  \STATE{\textbf{Bias-reduced swaps} Swap $ \bbeta_{k+1}^{(1)}$ and $ \bbeta_{k+1}^{(2)}$ if $u<\widetilde S_{\eta,m,n}$, where $u\sim \text{Unif }[0,1]$, and $\widetilde S_{\eta,m,n}$ follows}
  \begin{equation*}
  \footnotesize
      \textstyle \widetilde S_{\eta,m,n}=\exp\left\{ \left(\frac{1}{\tau^{(1)}}-\frac{1}{\tau^{(2)}}\right)\left(  \widetilde L( B_{k+1}\mid\bbeta_{k+1}^{(1)})-  \widetilde L( B_{k+1}\mid\bbeta_{k+1}^{(2)})-\frac{1}{F}\left(\frac{1}{\tau^{(1)}}-\frac{1}{\tau^{(2)}}\right)\widetilde \sigma^2_{m\lfloor \frac{k}{m}\rfloor}\right)\right\}.
  \end{equation*}
    
  \UNTIL{$k=k_{\max}$.}
    \vskip -1 in
\STATE{\textbf{Output:}  $\{\bbeta_{i\mathbb{T}}^{(1)}\}_{i=1}^{\lfloor k_{\max}/\mathbb{T}\rfloor}$, where $\mathbb{T}$ is the thinning factor.}
\end{algorithmic}
\end{algorithm}

As a result, we can further improve the performance of variance reduction by as much as 40\% on CIFAR10 and 30\% on CIFAR100 in the first 200 epochs. As the training continues and the learning rate decreases, the correlation coefficient is becoming closer to -1. In the late period, there is still 10\% improvement compared to the standard VR-reSGHMC.

In a nut shell, we can try adaptive coefficients in the early period when the absolute value of the correlation is lower than 0.5 or just use the vanilla replica exchange stochastic gradient Monte Carlo to avoid the computations of variance reduction.

\begin{figure*}[!ht]
  \centering
  \subfloat[\footnotesize{CIFAR10 \& m=50} ]{\includegraphics[width=3.8cm, height=3.4cm]{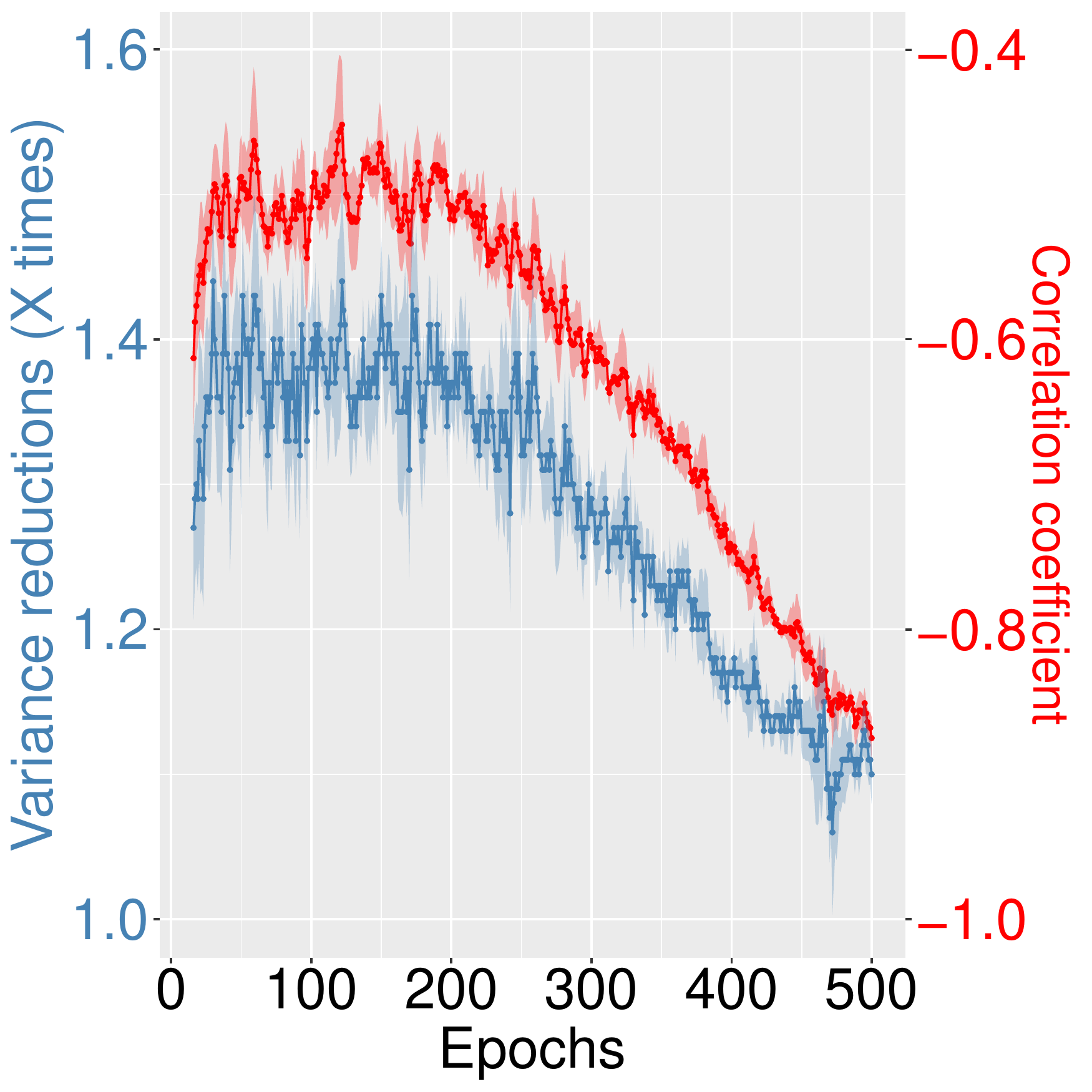}}\enskip
  \subfloat[\footnotesize{CIFAR100 \& m=50} ]{\includegraphics[width=3.8cm, height=3.4cm]{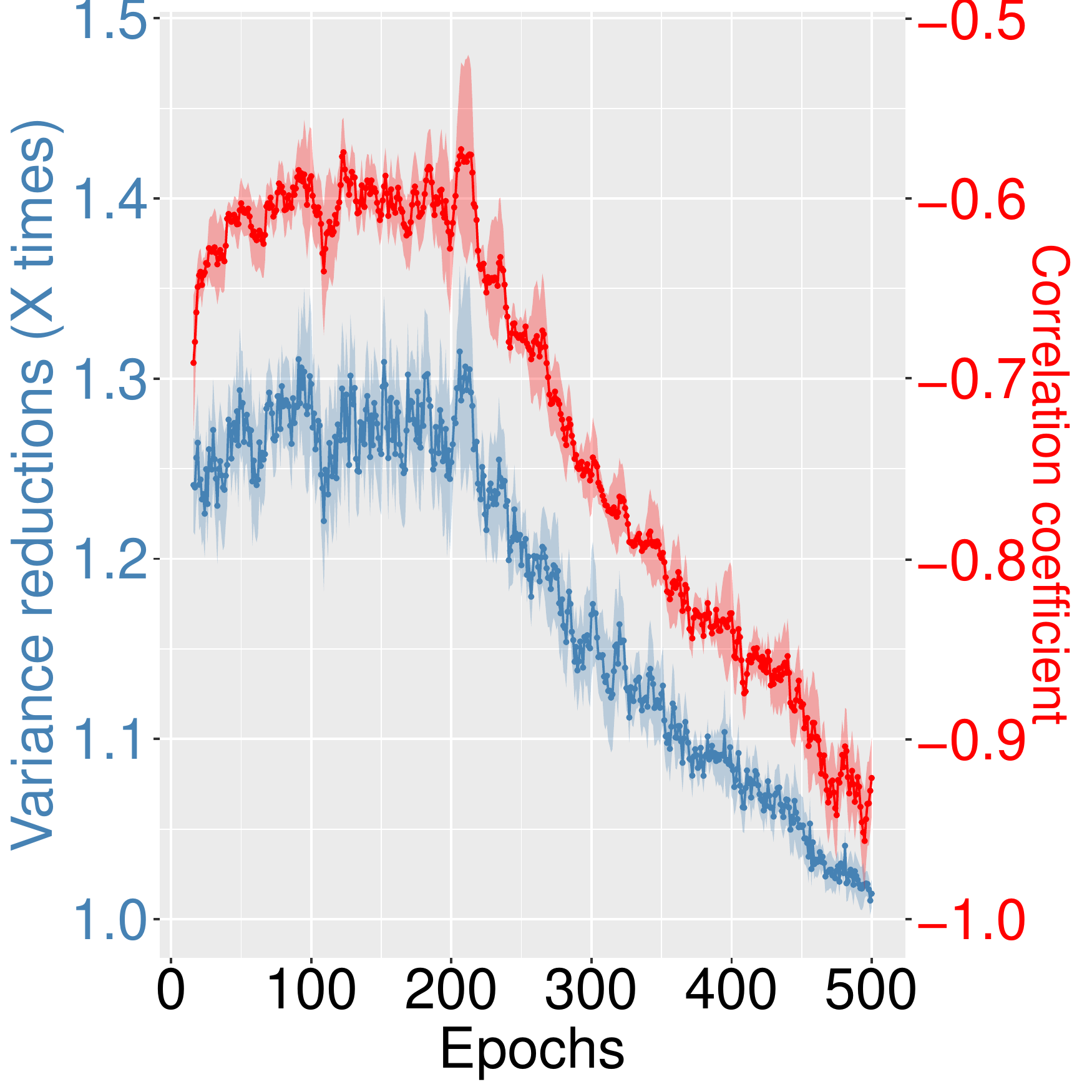}}\enskip
  \subfloat[CIFAR10 \& m=392]{\includegraphics[width=3.8cm, height=3.4cm]{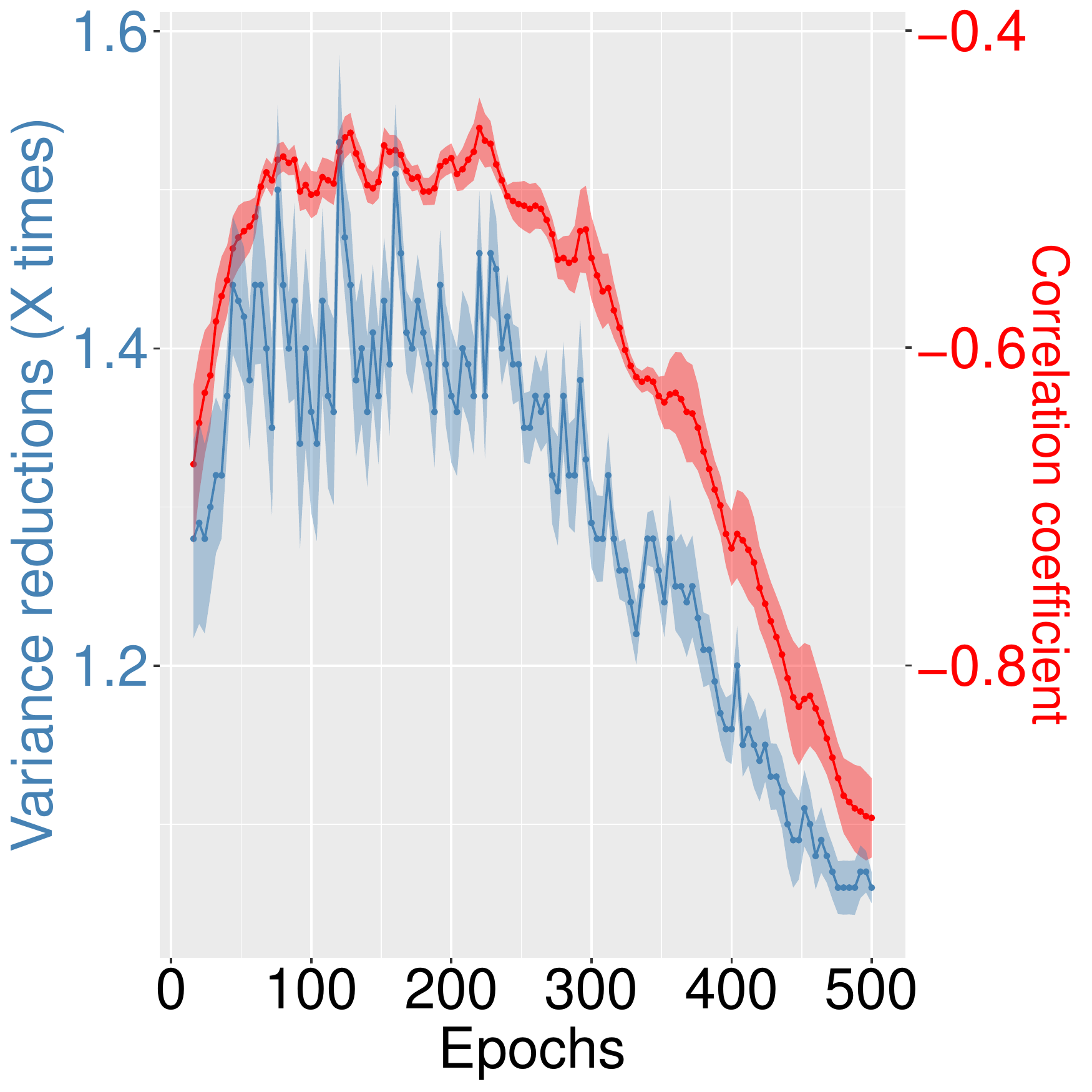}}\enskip
  \subfloat[CIFAR100 \& m=392]{\includegraphics[width=3.8cm, height=3.4cm]{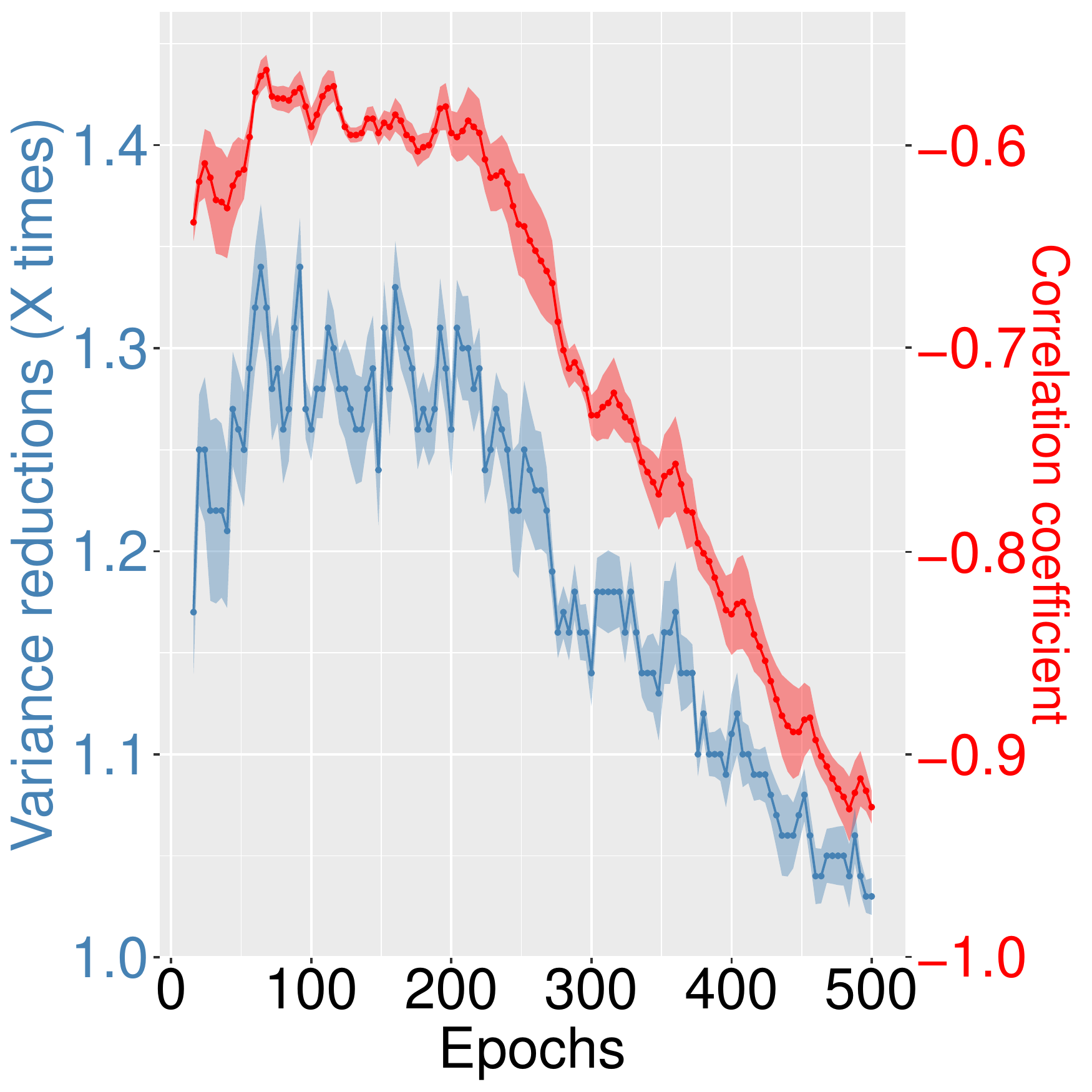}}
    \vskip -0.05in
  \caption{A study of variance reduction techniques using adaptive coefficient and non-adaptive coefficient on CIFAR10 \& CIFAR100 datasets.}
  \label{cifar_adaptive_c_figs}
\end{figure*}

\subsection{Uncertainty Quantification for Unknown Samples} 

A reliable model not only makes the right decision among potential candidates but also casts doubts on irrelevant choices. For the latter, we follow \cite{Balaji17} and evaluate the uncertainty on out-of-distribution samples from unseen classes. To avoid over-confident predictions on unknown classes, the ideal predictions should yield a higher uncertainty on the out-of-distribution samples, while maintaining the accurate uncertainty for the in-distribution samples.

Continuing the setup in Sec.\ref{nonconvex_optimization}, we collect the ResNet20 models trained on CIFAR10 and quantify  the entropy on the Street View House Numbers (SVHN) dataset, which contains 26,032 RGB testing images of digits instead of objects. We compare cVR-reSGHMC with \upshape{m}SGD, SGHMC, reSGHMC, and cSGHMC. Ideally, the predictive distribution should be the uniform distribution and leads to the highest entropy.  We present the empirical cumulative distribution function (CDF) of the entropy of the predictions on SVHN and report it in Fig.\ref{UQ}. As shown in the left figure, \upshape{m}SGD shows the smallest probability for high-entropy predictions, implying the weakness of stochastic optimization methods in uncertainty estimates. By contrast, the proposed cVR-reSGHMC yields the highest probability for predictions of high entropy. Admittedly, the standard ResNet models are poorly calibrated in the predictive probabilities and lead to inaccurate confidence. To alleviate this issue, we adopt the temperature-scaling method with a scale of 2 to calibrate the predictive distribution \cite{temperature_scaling} and present the entropy in Fig.\ref{UQ} (right). In particular, we see that 77\% of the predictions from cVR-reSGHMC yields the entropy higher than 1.5, which is 7\% higher than reSGHMC and 10\% higher than cSGHMC and much better than the others.

\begin{figure}[!ht]
  \begin{center}
  \vskip -0.2in
     \includegraphics[width=0.7\textwidth]{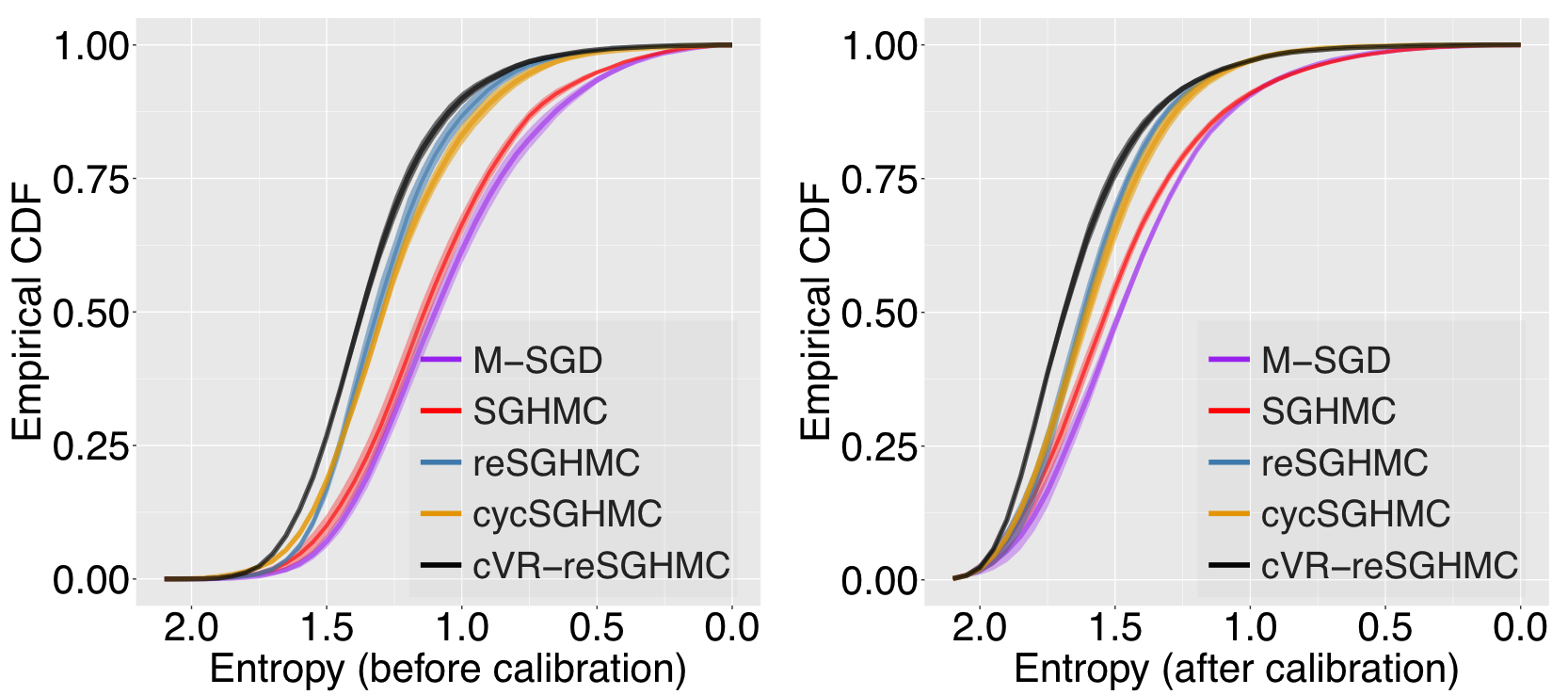}
  \end{center}
  \vskip -0.1in
  \caption{{CDF of entropy for predictions on SVHN via CIFAR10 models. A temperature scaling is used in calibrations.}}
  \label{UQ}
\end{figure}

\chapter{NON-REVERSIBLE PARALLEL TEMPERING FOR UNCERTAINTY APPROXIMATION IN DEEP LEARNING}

\label{NRPT_uncertainty}

\section{Introduction}
Recall that Langevin diffusion proposes to inject the Brownian motion to a gradient flow as follows
\begin{equation*}
\label{sde}
\begin{split}
    d \bbeta_t &= - \nabla L(\bbeta_t) dt+\sqrt{2\tau} d\bW_t,
\end{split}
\end{equation*}

where $\bbeta_t\in\mathbb{R}^d$, $\nabla L(\cdot)$ is the gradient of the energy function $L(\cdot)$, $\bW_t\in\mathbb{R}^d$ is a Brownian motion, and $\tau$ is the temperature. The diffusion process converges to a stationary distribution $\pi(\bbeta)\propto e^{-\frac{L(\bbeta)}{\tau}}$ and setting $\tau=1$ yields a Bayesian posterior. 
When $L(\cdot)$ is convex, the rapid convergence has been widely studied in \cite{dm+16, dk17}; however, when $L(\cdot)$ is non-convex, a slow mixing rate is inevitable \cite{Maxim17}. To accelerate the simulation, replica exchange Langevin diffusion (reLD) proposes to include a high-temperature particle $\bbeta_t^{(P)}$, where $P\in \mathbb{N}^{+} \setminus \{1\}$, for \emph{exploration}. Meanwhile, a low-temperature particle $\bbeta_t^{(1)}$ is presented for \emph{exploitation}:
\begin{equation}
\label{sde_2_couple}
\begin{split}
    d \bbeta_t^{(P)} &= - \nabla L(\bbeta_t^{(P)}) dt+\sqrt{2\tau^{(P)}} d\bW_t^{(P)},\\
    \quad d \bbeta_t^{(1)} &= - \nabla L(\bbeta_t^{(1)}) dt+\sqrt{2\tau^{(1)}} d\bW_t^{(1)},\\
\end{split}
\end{equation}
where $\tau^{(P)}>\tau^{(1)}$ and $\bW_t^{(P)}$ is independent of $\bW_t^{(1)}$. To promote more explorations for the low-temperature particle, the particles at the position $(\beta^{(1)}, \beta^{(P)})\in\mathbb{R}^{2d}$ swap with a probability 
\begin{equation}
\label{swap_function}
    aS(\beta^{(1)}, \beta^{(P)})=a\cdot \bigg(1\wedge e^{ \big(\frac{1}{\tau^{(1)}}-\frac{1}{\tau^{(P)}}\big)\big( L(\beta^{(1)})- L(\beta^{(P)})\big)}\bigg),
\end{equation}
where $a\in(0,\infty)$ is the swap intensity. In specific, the conditional swap rate at time $t$ follows that
\begin{equation*}
\begin{split}
    \mathbb{P}(\bbeta_{t+dt}=(\beta^{(P)}, \beta^{(1)})\mid\bbeta_t=(\beta^{(1)}, \beta^{(P)}))&=a S(\beta^{(1)}, \beta^{(P)}) d t,\\
    \mathbb{P}(\bbeta_{t+dt}=(\beta^{(1)}, \beta^{(P)})\mid\bbeta_t=(\beta^{(1)}, \beta^{(P)}))&=1-a S(\beta^{(1)}, \beta^{(P)}) d t.\\
\end{split}
\end{equation*}
In the longtime limit, the Markov jump process converges to the joint distribution $\pi(\bbeta^{(1)}, \bbeta^{(P)})\propto e^{-\frac{L(\bbeta^{(1)})}{\tau^{(1)}}-\frac{L(\bbeta^{(P)})}{\tau^{(P)}}}$. For convenience, we refer to the marginal distribution $\pi^{(1)}(\bbeta)\propto e^{-\frac{L(\bbeta)}{\tau^{(1)}}}$ and $\pi^{(P)}(\bbeta)\propto e^{-\frac{L(\bbeta)}{\tau^{(P)}}}$ as the \emph{target distribution} and \emph{reference distribution}, respectively.

\section{Preliminaries}
Achieving sufficient explorations requires a large $\tau^{(P)}$, which leads to limited accelerations due to a \emph{small overlap} between $\pi^{(1)}$ and $\pi^{(P)}$. To tackle this issue, one can bring in multiple particles with temperatures $(\tau^{(2)}, \cdots, \tau^{(P-1)})$, where $\tau^{(1)}<\tau^{(2)}<\cdots <\tau^{(P)}$, to hollow out ``tunnels''. To  maintain feasibility, numerous schemes are presented to select candidate pairs to attempt the swaps. 

\subsection{APE} The all-pairs exchange (APE) attempts to swap arbitrary pair of chains \cite{Brenner07, Martin09}, however, such a method requires a swap time (see definition in section \ref{others}) of $O(P^3)$  and may not be user-friendly in practice.

\subsection{ADJ} In addition to swap arbitrary pairs, one can also swap \emph{adjacent} (ADJ) pairs iteratively from $(1,2)$, $(2,3)$, to $(P-1,P)$ under the Metropolis rule. Despite the convenience, the \emph{sequential nature} requires to wait for exchange information from previous exchanges, which only works well with a small number of chains and has greatly limited its extension to a multi-core or distributed context.

\subsection{SEO} The stochastic even-odd (SEO) scheme first divides the adjacent pairs $\{(p-1, p)\mid p=2,\cdots, P\}$ into $E$ and $O$, where $E$ and $O$ denote even and odd pairs of forms $(2p-1, 2p)$ and $(2p, 2p+1)$, respectively. Then, SEO randomly picks $E$ or $O$ pairs with an equal chance in each iteration to attempt the swaps. Notably, it can be conducted \emph{simultaneously} without waiting from other chains. The scheme yields a reversible process (see Figure \ref{illustration}(a)), however, the gains in overcoming the sequential obstacle don't offset the \emph{$O(P^2)$ round trip time} and SEO is still not effective enough.

\begin{figure*}[!ht]
  \centering
  \subfloat[\scriptsize{Reversible indexes}]{\includegraphics[width=3.5cm, height=3.5cm]{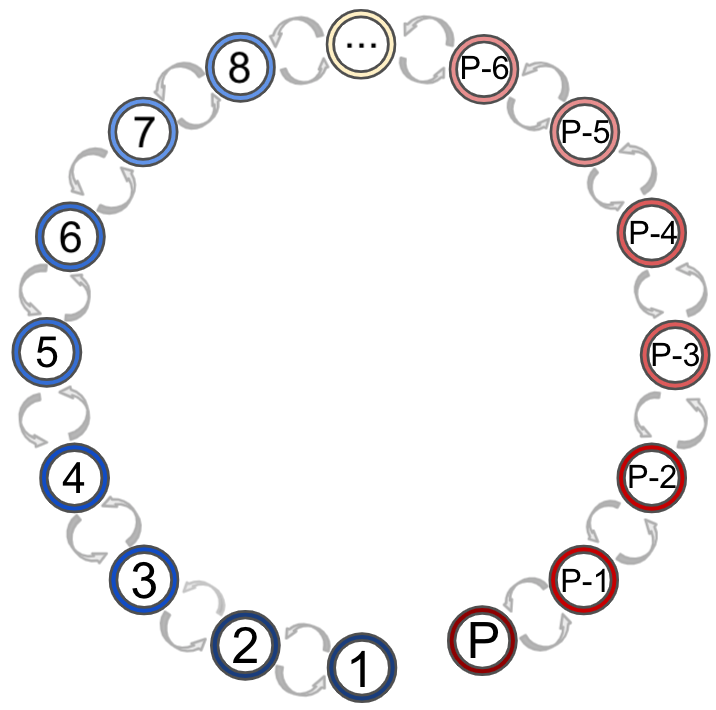}}\qquad
  \subfloat[\scriptsize{Non-reversible indexes}]{\includegraphics[width=3.5cm, height=3.4cm]{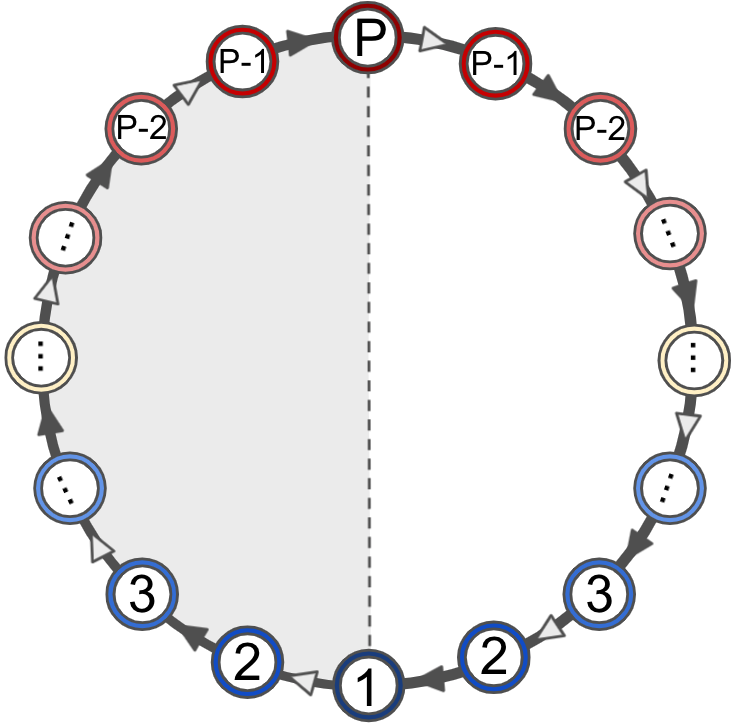}}\qquad
    \subfloat[Non-reversible chains]{\includegraphics[width=7.5cm, height=3.3cm]{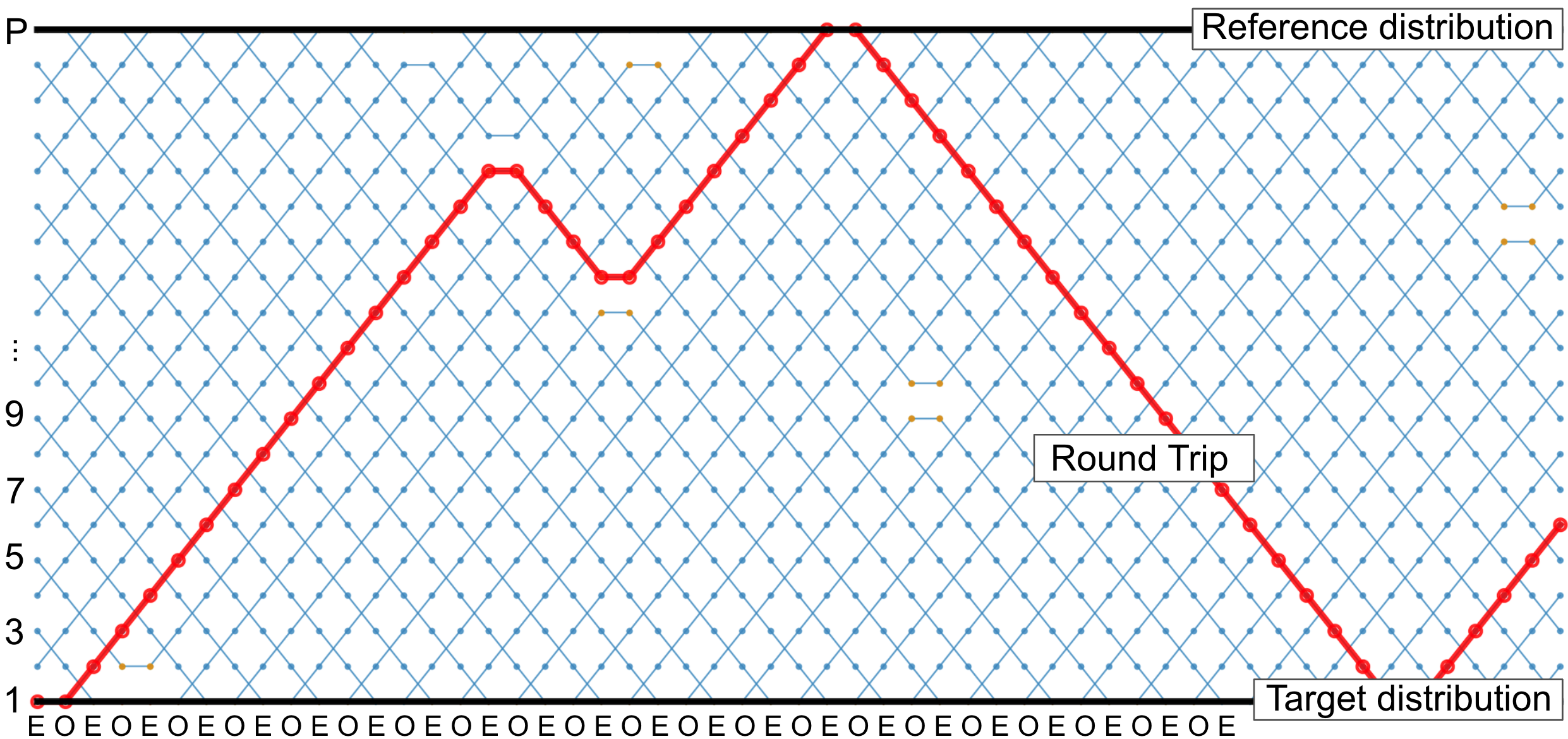}}\enskip
  \caption{Reversibility v.s. non-reversibility. In (a),  a reversible index takes $O(P^2)$ time to communicate; in (b), an ideal non-reversible index moves along a periodic orbit, where the dark and light arrows denote even and odd iterations, respectively; (c) shows how a non-reversible chain conducts a round trip with the DEO scheme.}
  \label{illustration}
\end{figure*}

\subsection{DEO} The deterministic even-odd (DEO) scheme instead attempts to swap even ($E$) pairs at even ($E$) iterations and odd ($O$) pairs at odd ($O$) iterations alternatingly\footnote[2]{\scriptsize{$E$ shown in iterations means even iterations; otherwise, it denotes even pairs for chain indexes. The same logic applies to $O$.}} \cite{DEO}. The asymmetric manner was later interpreted as a non-reversible PT \cite{Syed_jrssb} and an ideal index process follows a periodic orbit, as shown in Figure \ref{illustration}(b). With a large swap rate, Figure \ref{illustration}(c) shows how the scheme yields an almost straight path and a linear round trip time can be expected.

\subsection{Equi-acceptance} The power of PT hinges on maximizing the number of round trips, which is equivalent to minimizing $\sum_{p=1}^{P-1} \frac{1}{1-r_p}$ \cite{Nadler07_V2}, where $r_p$ denotes the rejection rate for the chain pair $(p, p+1)$. Moreover, $\sum_{p=1}^{P-1} r_p$ converges to a fixed barrier $\Lambda$ as $P\rightarrow \infty$ \cite{Predescu04, Syed_jrssb}. Applying Lagrange multiplies to the constrained optimization problem leads to $r_1=r_2=\cdots=r_{P-1}:=r$, where $r$ is the \emph{equi-rejection rate}. In general, a quadratic round trip time is required for ADJ and SEO due to the reversible indexes. By contrast, DEO only yields a \emph{linear round trip} time in terms of $P$ as $P\rightarrow \infty$ \cite{Syed_jrssb}.

\section{Optimal Non-reversible Scheme for Parallel Tempering}
The linear round trip time is appealing for maximizing the algorithmic potential, however, such an advance only occurs given sufficiently many chains. In \emph{non-asymptotic settings} with limited chains, a pearl of wisdom is to avoid frequent swaps \cite{Paul12} and to keep the average acceptance rate from 20\% to 40\% \cite{Kone2005, Martin09, Yves10}. Most importantly, the acceptance rates are severely reduced in big data due to the bias-corrected swaps associated with stochastic energies \cite{deng2020}, see details in section \ref{reSGLD_appendix}. As such, maintaining low rejection rates becomes quite challenging and the \emph{issue of quadratic costs} still exists. 

\subsection{Generalized DEO Scheme}
Continuing the equi-acceptance settings, we see in Figure.\ref{illustration_DEO}(a) that the probability for the blue particle to move upward 2 steps to maintain the same momentum after a pair of even and odd iterations is $(1-r)^2$. As such, with a large equi-rejection rate $r$, the blue particle often makes little progress (Figure.\ref{illustration_DEO}(b-d)). To handle this issue, the key is to propose small enough rejection rates to track the periodic orbit in Figure.\ref{illustration}(b). Instead of pursuing excessive amount of chains, \emph{we resort to a different solution by introducing the generalized even and odd iterations} $E_W$ and $O_W$, where $W\in\mathbb{N}^{+}$, $E_W=\{\lfloor \frac{k}{W}\rfloor \text{ mod } 2 =0\mid k=1,2,\cdots, \infty\}$ and $O_W=\{\lfloor \frac{k}{W}\rfloor \text{ mod } 2 =1\mid k=1,2,\cdots, \infty\}$. Now, we present the generalized DEO scheme with a window size $W$ as follows and refer to it as DEO$_W$: \footnote[4]{{The generalized DEO with the optimal window size is denoted by DEO$_{\star}$ and will be studied in section \ref{windiw_size_round_trip}.}}
\begin{equation}
\begin{split}
\label{gDEO}
     \circ \text{ Attempt to swap } E \text{ (or } O\text{) pairs at }E_W\text{ (or }O_W\text{) iterations.} \qquad\qquad\qquad\qquad\ \ \\
    \circ \text{ Allow }\emph{at most one} \text{ swap during each cycle of }E_W \text{ (or }O_W\text{) iterations.}\qquad\qquad\ \ 
\end{split}
\end{equation}

\begin{figure*}[!ht]
  \centering
  \subfloat[DEO]{\includegraphics[width=2.1cm, height=1.8cm]{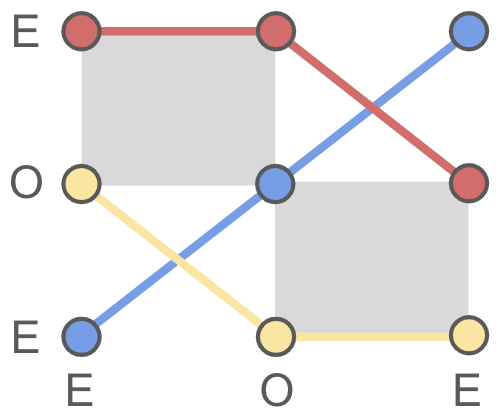}}\qquad
  \subfloat[Bad case 1]{\includegraphics[width=2.1cm, height=1.8cm]{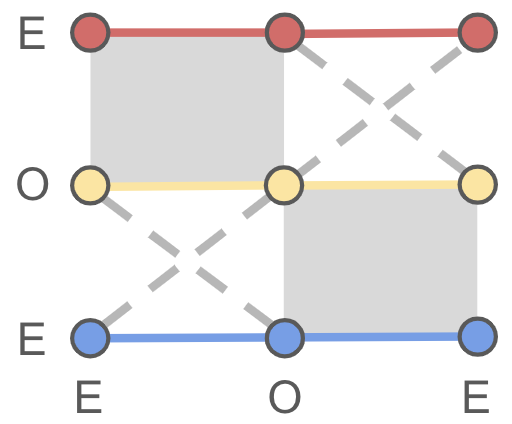}}\qquad
  \subfloat[Bad case 2]{\includegraphics[width=2.1cm, height=1.8cm]{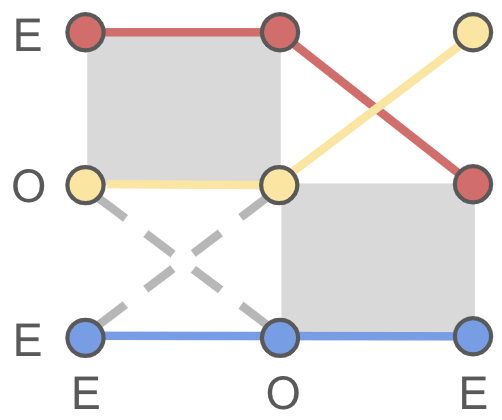}}\qquad
  \subfloat[Bad case 3]{\includegraphics[width=2.1cm, height=1.8cm]{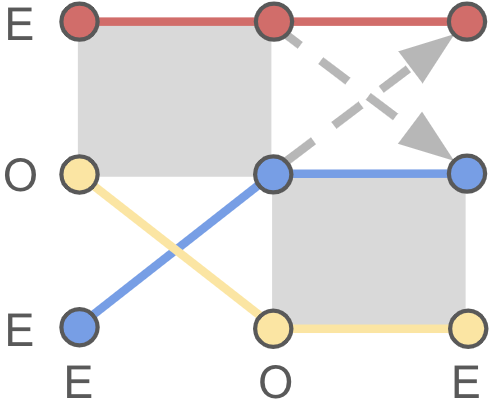}}\qquad
  \subfloat[$\text{DEO}_2$]{\includegraphics[width=3.6cm, height=1.8cm]{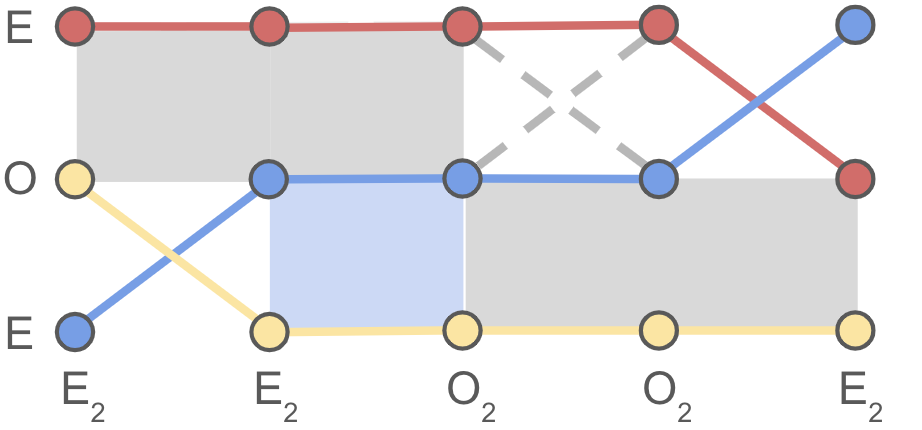}}\qquad
  \caption{Illustration of DEO and DEO$_2$. In (a), we show an ideal case of DEO; (b-d) show bad cases of DEO based on a large equi-rejection rate $r$; (e) show how the generalized DEO scheme of window size 2 tackles the issue with a large $r$. The x-axis and y-axis denote (generalized) $E$ (or $O$) iterations and $E$ (or $O$) pairs, respectively. The dashed line denotes the failed swap attempts; the gray shaded areas are frozen to refuse swapping odd pairs at even iterations (or vice versa); the blue shaded area ensures \emph{at most one swap in a window}.}
  \label{illustration_DEO}
\end{figure*}

As illustrated in Figure.\ref{illustration_DEO}(e), the blue particle has a larger chance of $(1-r^2)^2$ to move upward 2 steps given $W=2$ instead of $(1-r)^2$ when $W=1$, although the window number is also halved. Such a trade-off inspires us to analyze the expected round trip time based on the window of size $W$. Although allowing at most one swap introduces a stopping time and may affect the distribution, the gap of acceptance rates between bias-corrected swap rates and frozen swaps of rate 0 is rather small, as such the bias is rather mild in big data.

\subsection{Analysis of Round Trip Time}
\label{main_round_trip_analysis}

To bring sufficient interactions between the reference distribution $\pi^{(P)}$ and the target distribution $\pi^{(1)}$, we expect to minimize the expected round trip time $T$ (defined in section \ref{others}) to ensure both efficient exploitation and explorations. Combining the Markov property and the idea of the master equation \cite{Nadler07}, we estimate the expected round trip time $\hE[T]$ as follows
\begin{lemma}
\label{thm:round_trip_time}
Under the stationary and weak dependence assumptions B1 and B2 in section \ref{analysis_of_round_trip}, for $P$ ($P\ge 2$) chains with window size $W$ ($W\ge 1$) and rejection rates $\{r_p\}_{p=1}^{P-1}$, we have
\begin{align} \label{eq:RTT_main}
\footnotesize{\hE[T]=2WP+2WP\sum_{p=1}^{P-1}\frac{r_p^W}{1-r_p^W}}.
\end{align}
\end{lemma}
The proof in section \ref{analyze_round_trip} shows that $\E[T]$ increases as we adopt larger number of chains $P$ and rejection rates $\{r_p\}_{p=1}^{P-1}$. In such a case, the round trip rate $\frac{P}{\hE[T]}$ is also maximized by the key renewal theorem. In particular, applying $W=1$ recovers the vanilla DEO scheme.

\subsection{Analysis of Optimal Window Size and Round Trip Time}
\label{windiw_size_round_trip}
By Lemma \ref{thm:round_trip_time}, we observe a potential to remove the second quadratic term given an appropriate $W$. Such a fact motivates us to study the optimal window size $W$ to achieve the best efficiency. Under the equi-acceptance settings, by treating the window size $W$ as a continuous variable and taking the derivative of $\E[T]$ with respect to $W$, we have
\begin{equation} \label{eq:dWRTT_main}
\frac{\partial}{\partial W}\hE[T]=\frac{2P}{(1-r^W)^2}\left\{(1-r^W)^2+(P-1)r^W(1-r^W+W\log r)\right\},
\end{equation}
where $r$ is the equi-rejection rate for adjacent chains. Define $x:=r^W\in(0, 1)$, where $W=\log_r(x)=\frac{\log x}{\log r}$. The following analysis hinges on the study of the solution $g(x)=(1-x)^2+(P-1)x(1-x+\log(x))=0$. By analyzing the growth of derivatives and boundary values, we can easily identify the \emph{uniqueness} of the solution. Then, we proceed to verify that $\frac{1}{P\log P}$ yields an asymptotic approximation such that $g(\frac{1}{P\log P})=-\frac{\log(\log P)}{\log P}+O\left(\frac{1}{\log P}\right)\rightarrow 0$ as $P\rightarrow \infty$. In the end, we have

\begin{theorem} \label{col:W_approx}
Under Assumptions \ref{Stationarity_process} and \ref{weak_independence} based on equi-acceptance settings, if $P=2,3$, the maximal round trip time is achieved when $W=1$. If $P\ge 4$, with the optimal window size $W_{\star}\approx\left\lceil\frac{\log P+\log\log P}{-\log r}\right\rceil$, where $\lceil\cdot\rceil$ is the ceiling function. The round trip time follows $O(\frac{P\log P}{-\log r})$. 
\end{theorem}

The above result yields a remarkable round trip time of $O(P\log P)$ by setting the optimal window size $W_{\star}$. By contrast, the vanilla DEO scheme only leads to a longer time of $O(P^2)$ \footnote[4]{By Taylor expansion, given a large rejection rate $r$, $-\log(r)=1-r$, which means $\frac{1}{-\log(r)}=O(\frac{r}{1-r})$.}. Denoting by $\text{DEO}_{\star}$ the generalized DEO scheme with the optimal window size $W_{\star}$, we summarize the popular swap schemes in Table.\ref{round_trip_time_cost}, where the DEO$_{\star}$ scheme performs the best among all the three criteria.

\begin{table}[ht]
\begin{sc}
\vspace{0.1in}
\caption{{Round trip time and swap time for different schemes. The APE scheme requires an expensive swap time of $O(P^3)$ and is not compared}.} \label{round_trip_time_cost}
\vspace{0.1in}
\begin{center} 
\begin{tabular}{c|cccc}
\hline
 & \textbf{\footnotesize{Round trip time} \footnotesize{(non-asymptotic)}} & \footnotesize{Round trip time} \footnotesize{(asymptotic)} & \footnotesize{Swap time} \\
\hline
ADJ & $O(P^2)$ {\cite{Nadler07}} & $O(P^2)$ {\cite{Nadler07}} & $O(P)$ \\
SEO & $O(P^2)$ {\cite{Syed_jrssb}} & $O(P^2)$ {\cite{Syed_jrssb}} & \bm{$O(1)$} \\
DEO & $O(P^2)$ {\cite{Syed_jrssb}} & \bm{$O(P)$} {\cite{Syed_jrssb}}  & \bm{$O(1)$} \\
\hline
${\text{DEO}_{\star}}$  &  $\bm{O(P}$ \textbf{\upshape{log}} $\bm{P)}$  & \bm{$O(P)$}  & \bm{$O(1)$} \\
\hline
\end{tabular}
\end{center}
\end{sc}
\vspace{-0.05in}
\end{table}

\section{User-friendly Approximate Explorations in Big Data}

Despite the asymptotic correctness, SGLD only works well given \emph{small enough learning rates} and fails in explorative purposes \cite{Ahn12}. A large learning rate, however, leads to excessive stochastic gradient noise and ends up with a crude approximation. As such, similar to \cite{SWA1, ruqi2020}, we only adopt SGLD for exploitations.

Efficient explorations not only require a high temperature but also prefer a large learning rate. Such a demand inspires us to consider SGD with a constant learning rate $\eta$ as the exploration component
\begin{equation}
\begin{split}\label{sgd_approx}
    \bbeta_{k+1} =\bbeta_{k} - \eta(\nabla L(\bbeta_k)+\varepsilon(\bbeta_k))=\bbeta_{k} - \eta\nabla L(\bbeta_k)+\sqrt{2\eta \big(\frac{\eta}{2}\big)}\varepsilon(\bbeta_k),
\end{split}
\end{equation}
where $\varepsilon(\bbeta_k)\in\mathbb{R}^d$ is the stochastic gradient noise. Under mild normality assumptions on $\varepsilon$ \cite{Mandt, SGD_AOS}, $\bbeta_{k}$ converges approximately to an invariant distribution, where the underlying \emph{temperature linearly depends on the learning rate $\eta$}. Motivated by this fact, we propose an approximate transition kernel $\mathcal{T}_{\eta}$ with $P$ parallel \emph{SGD runs} based on different learning rates
\begin{equation}  
\begin{split}
\label{pt_sgd}
\textbf{Exploration: }\left\{  
             \begin{array}{lr}  
             \footnotesize{\small{\bbeta_{k+1}^{(P)}=\bbeta_{k}^{(P)} - \eta^{(P)}\nabla \widetilde L(\bbeta_k^{(P)})}}, \\  
              \\
              \footnotesize{\cdots}  \\
             \footnotesize{\bbeta_{k+1}^{(2)} =\bbeta_{k}^{(2)} - \eta^{(2)}\nabla \widetilde L(\bbeta_k^{(2)})},\\
             \end{array}
\right. \qquad\qquad\qquad\qquad\  & \\
\textbf{Exploitation: } \ \quad\footnotesize{\bbeta_{k+1}^{(1)}=\bbeta_{k}^{(1)} - \eta^{(1)}\nabla \widetilde L(\bbeta_k^{(1)})\ \ +\overbrace{\Xi_k}^{\text{optional}},\qquad\qquad\quad\ \ }\ &\\
\end{split}
\end{equation}
where $\eta^{(1)}<\eta^{(2)}<\cdots<\eta^{(P)}$, $\Xi_k\sim \mathcal{N}(0, 2\eta^{(1)}\tau^{(1)})$, and $\tau^{(1)}$ is the target temperature.

Since there exists an optimal learning rate for SGD to estimate the desired distribution through Laplace approximation \cite{Mandt}, the exploitation kernel can be also replaced with SGD based on constant learning rates if the accuracy demand is not high. Regarding the validity of adopting different learning rates for parallel tempering, we leave discussions to section \ref{diff_lr}.

Moreover, the stochastic gradient noise exploits the Fisher information \cite{Ahn12, zhanxing_anisotropic, Chaudhari17} and yields convergence potential to wide optima with good generalizations \cite{Berthier, difan_2021}. Despite the implementation convenience, the inclusion of SGDs has made the temperature variable inaccessible, rendering a difficulty in implementing the Metropolis rule Eq.(\ref{swap_function}). To tackle this issue, we utilize the randomness in stochastic energies and propose a \emph{deterministic swap condition} for the approximate kernel $\mathcal{T}_{\eta}$ in Eq.(\ref{pt_sgd}) such that
\begin{equation}
\label{deterministic_swap}
    \textbf{Deterministic swap condition: }\ \footnotesize{(\bbeta^{(p)}, \bbeta^{(p+1)})\rightarrow (\bbeta^{(p+1)}, \bbeta^{(p)}) \text{ if } \widetilde L(\bbeta^{(p+1)})+\mathbb{C}<\widetilde L(\bbeta^{(p)})},
\end{equation}
where $p\in\{1,2,\cdots, P-1\}$, $\mathbb{C}>0$ is a correction buffer to approximate the Metropolis rule Eq.(\ref{swap_function}). 

In addition, our proposed algorithm for uncertainty approximation is highly related to non-convex optimization. For the detailed discussions, we refer interested readers to section \ref{connection_2_non_convex}.

\subsection{Equi-acceptance Parallel Tempering on Optimized Paths}

Stochastic approximation (SA) is a standard method to achieve equi-acceptance \cite{Yves10, Miasojedow_2013}, however, implementing this idea with fixed $\eta^{(1)}$ and $\eta^{(P)}$ is rather non-trivial. Motivated by the linear relation between learning rate and temperature, we propose to adaptively \emph{optimize the learning rates} to achieve equi-acceptance in a user-friendly manner. Further by the geometric temperature spacing commonly adopted by practitioners \cite{Kofke02, parallel_tempering05, Syed_jrssb}, we adopt the following scheme on a \emph{logarithmic scale} such that
\begin{equation}
\label{sa_exp}
    \partial \log(\upsilon_t^{(p)})=h^{(p)}(\upsilon_t^{(p)}),
\end{equation}
where $p\in\{1,2,\cdots, P-1\}$, $\upsilon_t^{(p)}=\eta_t^{(p+1)}-\eta_t^{(p)}$, $h^{(p)}(\upsilon_t^{(p)}):=\int H^{(p)}(\upsilon_k^{(p)}, \bbeta) \pi^{(p, p+1)}(d\bbeta)$ is the mean-field function, $\pi^{(p, p+1)}$ is the joint invariant distribution for the $p$-th and $p+1$-th processes.  In particular, $H^{(p)}(\upsilon_k^{(p)}, \bbeta)=1_{\widetilde L( \bbeta^{(p+1)})+\mathbb{C}<\widetilde L( \bbeta^{(p)})}-\mathbb{S}$ is the random-field function to approximate $h^{(p)}(\upsilon_k^{(p)})$ with limited perturbations, $\upsilon_k^{(p)}$ \footnote[2]{For convenience, $\upsilon_t^{(p)}$ denotes the continuous-time diffusion at time $t$ and $\upsilon_k^{(p)}$ represents the discrete approximations at iteration $k$.} implicitly affects the distribution of the indicator function, and $\mathbb{S}$ is the target swap rate.  Now consider stochastic approximation of Eq.(\ref{sa_exp}), we have
\begin{equation}
\label{sa_exp_euler}
    \log(\upsilon_{k+1}^{(p)})=\log(\upsilon_{k}^{(p)}) + \gamma_k H^{(p)}(\upsilon_k^{(p)}, \bbeta_k),
\end{equation}
where $\gamma_k$ is the step size. Reformulating Eq.(\ref{sa_exp_euler}), we have
\begin{equation*}
\label{sa_exp_taylor}
    \upsilon_{k+1}^{(p)} = \max(0, \upsilon_{k}^{(p)})e^{\gamma_k H(\upsilon_k^{(p)})},
\end{equation*}
where the $\max$ operator is conducted explicitly to ensure the sequence of learning rates is non-decreasing. This means that given fixed boundary learning rates (temperatures) $\eta_k^{(p-1)}$ and $\eta_k^{(p+1)}$, applying $\eta^{(p)}=\eta^{(p-1)}+\upsilon^{(p)}$ and $\eta^{(p)}=\eta^{(p+1)}-\upsilon^{(p+1)}$ for $p\in\{2,3,\cdots, P-1\}$ lead to
\begin{equation}
\begin{split}
    \label{double_SA_forward_backward}
    \underbrace{\eta_k^{(p-1)}+ \max(0, \upsilon_{k}^{(p)})e^{\gamma_k H(\upsilon_k^{(p)})}}_{\text{forward sequence}}=\eta_{k+1}^{(p)}=\underbrace{\eta_k^{(p+1)}- \max(0, \upsilon_{k}^{(p+1)})e^{\gamma_k H(\upsilon_k^{(p+1)})}}_{\text{backward sequence}}.
\end{split}
\end{equation}

\subsubsection{Adaptive Learning Rates (Temperatures)} Now given a fixed $\eta^{(1)}$, the sequence $\eta^{(2)}$, $\eta^{(3)}$, $\cdots$, $\eta^{(P)}$ can be approximated iteratively via the forward sequence of (\ref{double_SA_forward_backward}); conversely, given a fixed $\eta^{(P)}$, the backward sequence $\eta^{(P-1)}$, $\eta^{(P-2)}$, $\cdots$, $\eta^{(1)}$ can be decided reversely as well. Combining the forward and backward sequences, $\eta_{k+1}^{(p)}$ can be approximated via 
\begin{equation}
\begin{split}
    \label{double_SA}
    \eta_{k+1}^{(p)}:&=\frac{\eta_k^{(p-1)}+\eta_k^{(p+1)}}{2}+\frac{\max(0, \upsilon_{k}^{(p)})e^{\gamma_k H(\upsilon_k^{(p)})} - \max(0, \upsilon_{k}^{(p+1)})e^{\gamma_k H(\upsilon_k^{(p+1)})}}{2},
\end{split}
\end{equation}
which resembles the \emph{binary search} in the SA framework. In particular, the first term is the middle point given boundary learning rates and the second term continues to penalize learning rates that violates the equi-acceptance between pairs $(p-1, p)$ and $(p, p+1)$ until an equilibrium is achieved. 

This is the first attempt to achieve equi-acceptance given two fixed boundary values to our best knowledge. By contrast, \cite{Syed_jrssb} proposed to estimate the barrier $\Lambda$ to determine the temperatures and it easily fails in big data given a finite number of chains and bias-corrected swaps.

\subsubsection{Adaptive Correction Buffers} In addition, equi-acceptance does not guarantee a convergence to the desired  acceptance rate $\mathbb{S}$. To avoid this issue, we propose to adaptively optimize $\mathbb{C}$ as follows
\begin{equation}\label{unknown_threhold}
    \mathbb{C}_{k+1}=\mathbb{C}_k+\gamma_k\left(\frac{1}{P-1}\sum_{p=1}^{P-1} 1_{\widetilde L( \bbeta_{k+1}^{(p+1)})+\mathbb{C}_k-\widetilde L( \bbeta_{k+1}^{(p)})<0}-\mathbb{S}\right).
\end{equation}

As $k\rightarrow \infty$, the threshold and the adaptive learning rates converge to the desired fixed points. Note that setting a uniform $\mathbb{C}$ greatly simplifies the algorithm; in more delicate cases, problem-specific rules are also recommended. Now we refer to the approximate non-reversible parallel tempering algorithm with the DEO$_{\star}$ scheme and SGD-based exploration kernels as DEO$_{\star}$-SGD and formally formulate our algorithm in Algorithm \ref{alg}. Extensions of SGD with a preconditioner \cite{Li16} or momentum \cite{Chen14} to further improve the approximation and efficiency are both straightforward \cite{Mandt} and are denoted as DEO$_{\star}$-pSGD and DEO$_{\star}$-mSGD, respectively.

\begin{algorithm}[tb]
  \small
  \caption{Non-reversible parallel tempering with SGD-based exploration kernels (DEO$_{\star}$-SGD).}
  \label{alg}
\begin{algorithmic}
\STATE{\textbf{Input} Number of chains $P\geq 3$, boundary learning rates $\eta^{(1)}$ and $\eta^{(P)}$, target swap rate $\mathbb{S}$.}
\STATE{\textbf{Input} Optimal window size $W:=\left\lceil \frac{\log P + \log\log P}{-\log(1-\mathbb{S})}\right\rceil$, total iterations $K$, and step sizes $\{\gamma_{k}\}_{k=0}^K$.}

\FOR{ $k=1, 2, \text{to } K$ } 
    \STATE{$\bbeta_{k+1}\sim \mathcal{T}_{\eta}(\bbeta_{k})$ following Eq.(\ref{pt_sgd}) \qquad\qquad\ \  $\triangleright$ Exploration / exploitation phase (parallelizable)}
    
    \STATE{$\mathcal{P}=\big\{\forall p\in\{1, 2,\cdots, P\}: p \text{ mod } 2 = \lfloor \frac{k}{W}\rfloor \text{ mod } 2\big\}$. \qquad\qquad$\triangleright$ Generalized even/odd iterations}

    \FOR{ $p=1,2, \text{to } P-1$ }
        \STATE{$\mathcal{A}^{(p)}:=1_{\widetilde L( \bbeta_{k+1}^{(p+1)})+\mathbb{C}_k<\widetilde L( \bbeta_{k+1}^{(p)})}$}

        \STATE{$\mathcal{G}^{(p)}:=1_{k\text{ mod } W=0}$. \quad\qquad\qquad\qquad\qquad\qquad\qquad\qquad\qquad\quad $\triangleright$ Open the gate to allow swaps}
        
        \IF{$p\in \mathcal{P}$ \text{ and } $\mathcal{G}^{(p)}$ \text{ and } $\mathcal{A}^{(p)}$}

        \STATE{\emph{Swap: } $ \bbeta_{k+1}^{(p)}$ and $ \bbeta_{k+1}^{(p+1)}$. \qquad\qquad\qquad\qquad\qquad\ \ \ $\triangleright$ Communication phase (parallelizable)}
        \STATE{ \emph{Freeze: } $\mathcal{G}^{(p)} = 0.$ \qquad\qquad\qquad\qquad\qquad\qquad\qquad\qquad\quad\ \ $\triangleright$ Close the gate to refuse swaps}
        \ENDIF
        \IF{$p>1$}
        \STATE{\emph{Update learning rate (temperature) following Eq.(\ref{double_SA})}}
        \ENDIF
    \ENDFOR
    
    \STATE{\emph{Adaptive correction buffer: } $\mathbb{C}_{k+1}=\mathbb{C}_k+\gamma_{k}\left(\frac{1}{P-1}\sum_{p=1}^{P-1} \mathcal{A}^{(p)}-\mathbb{S}\right).$}
\ENDFOR
\STATE{\textbf{Output } Models collected from the target temperature  $\{\bbeta_{k}^{(1)}\}_{k=1}^{K}$.}
\end{algorithmic}
\end{algorithm}

\section{Experiments}
\subsection{Simulations of Multi-modal Distributions}

We first simulate the proposed algorithm on a distribution $\pi(\bbeta)\propto \exp(-L(\bbeta))$, where $\bbeta=(\beta_1, \beta_2)$, $L(\bbeta)=0.2(\beta_1^2+\beta_2^2)-2(\cos(2\pi \beta_1)+\cos(2\pi \beta_2))$. The heat map is shown in Figure \ref{test_of_acceptance_rates}(a) with 25 modes of different volumes. To mimic big data scenarios, we can only access stochastic gradient $\nabla\widetilde L(\bbeta)=\nabla L(\bbeta)+2\mathcal{N}(0, \bm{I}_{2\times2})$ and stochastic energy $\widetilde L(\bbeta)=L(\bbeta)+2\mathcal{N}(0, I)$.

\begin{figure*}[!ht]
  \centering
  \subfloat[\scriptsize{Ground truth} ]{\includegraphics[width=2.3cm, height=2.3cm]{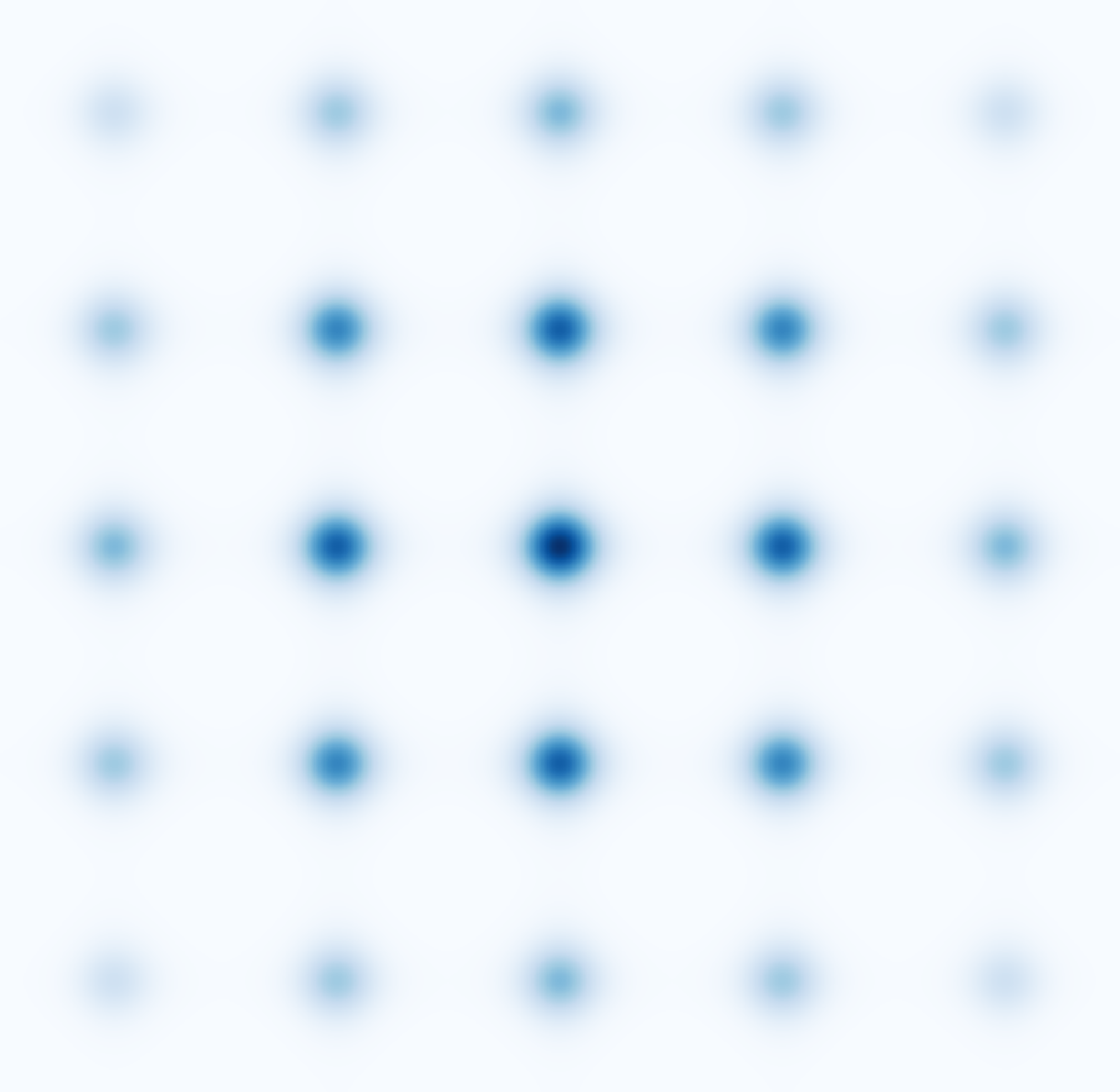}}\qquad
  \subfloat[$\mathbb{S}=0.2$]{\includegraphics[width=2.3cm, height=2.3cm]{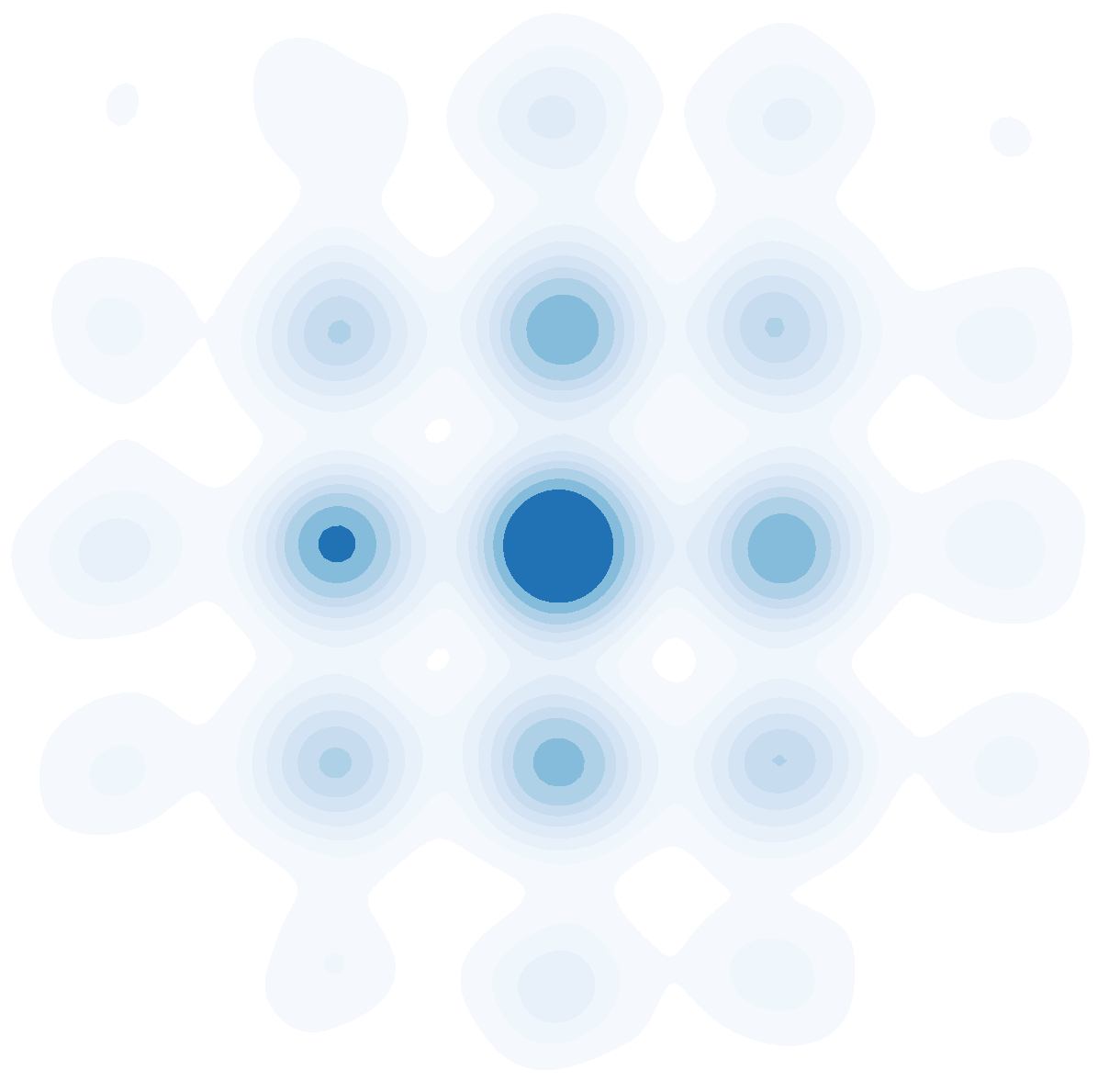}}\quad
  \subfloat[$\mathbb{S}=0.3$]{\includegraphics[width=2.3cm, height=2.3cm]{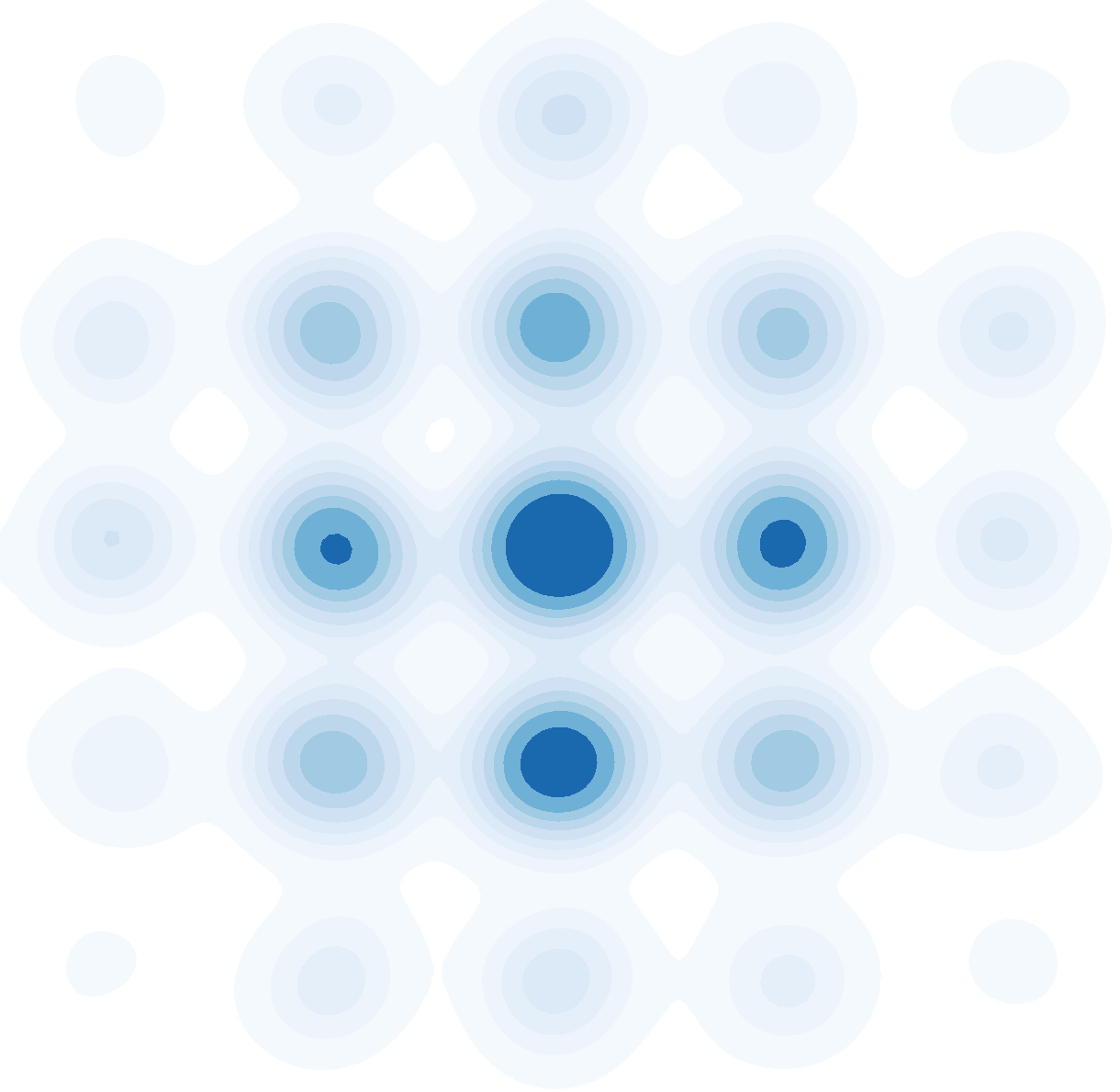}}\quad
  \subfloat[$\mathbb{S}=0.4$]{\includegraphics[width=2.3cm, height=2.3cm]{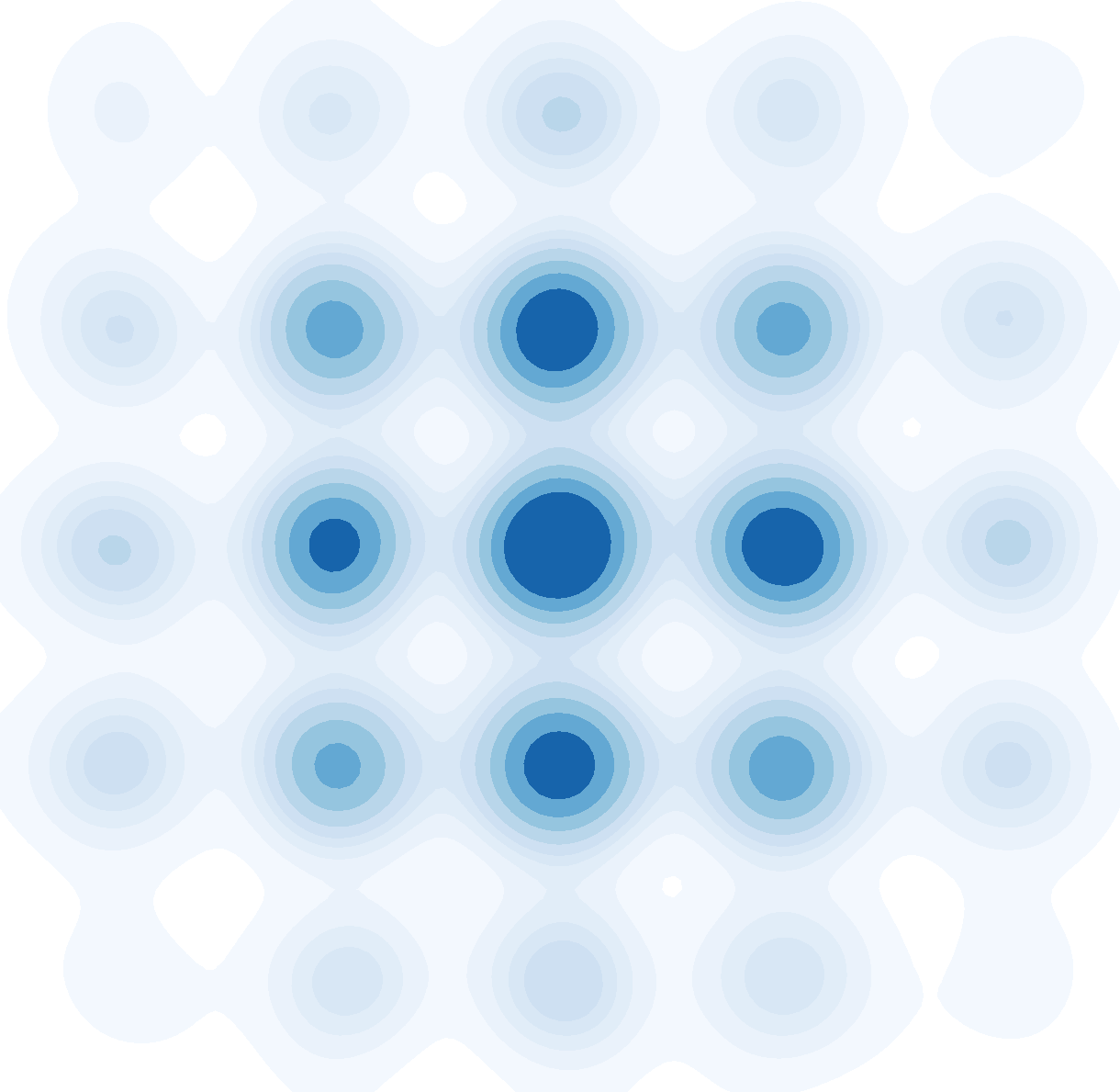}}\quad
  \subfloat[$\mathbb{S}=0.5$]{\includegraphics[width=2.3cm, height=2.3cm]{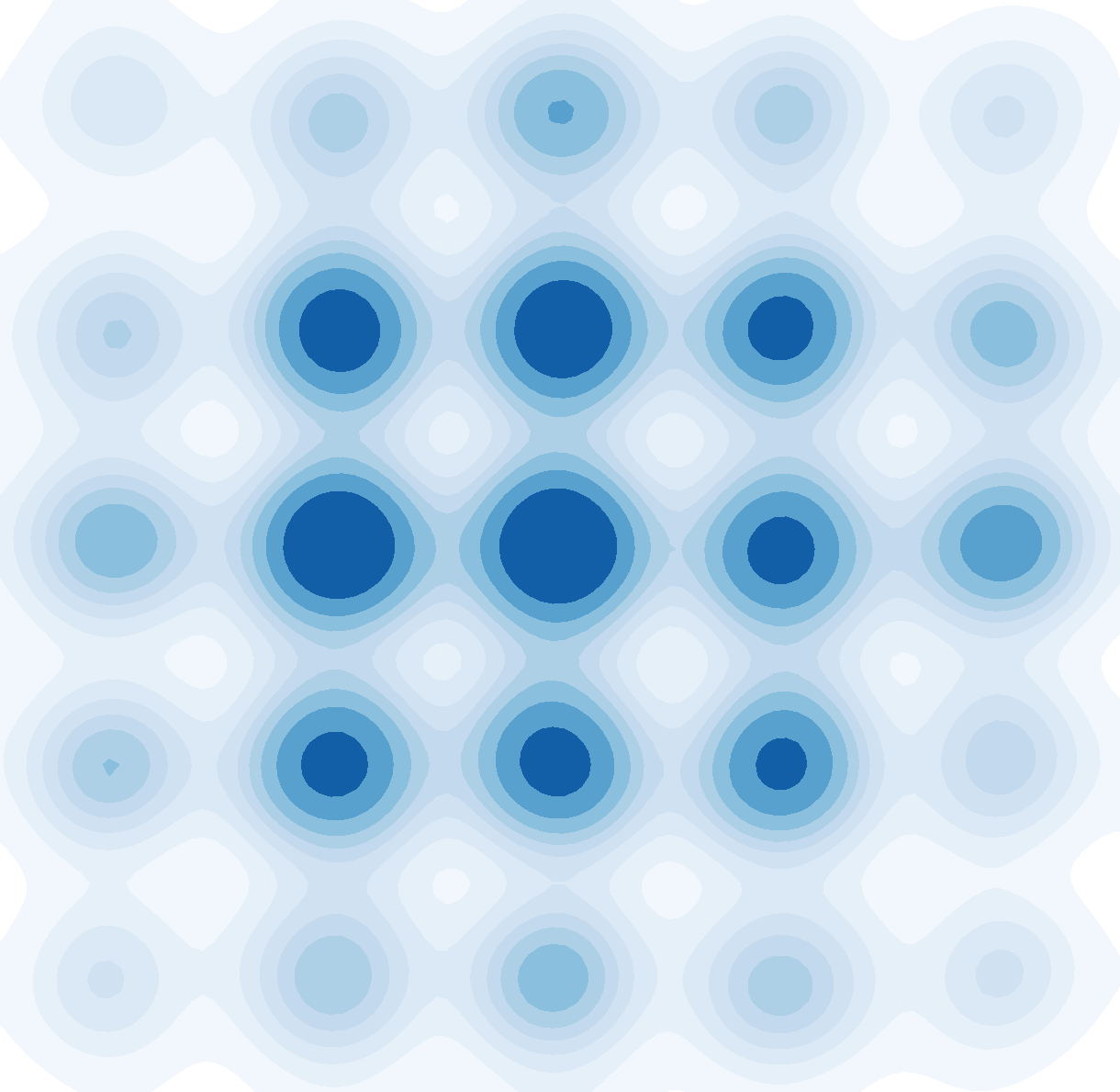}}\quad
  \subfloat[$\mathbb{S}=0.6$]{\includegraphics[width=2.3cm, height=2.3cm]{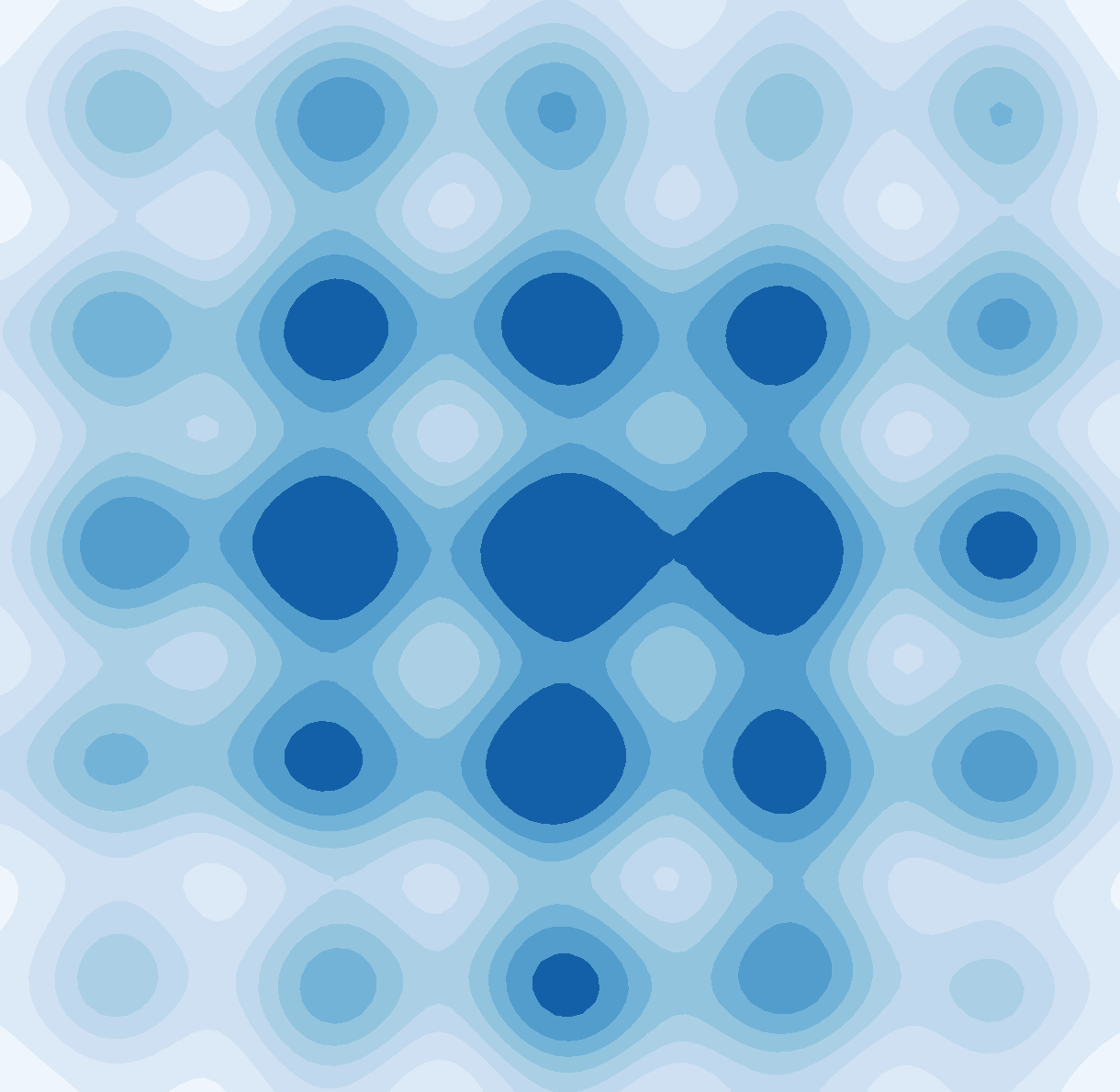}}
  \caption{Study of different target swap rate $\mathbb{S}$ via DEO$_{\star}$-SGD, where SGLD is the exploitation kernel.}
  \label{test_of_acceptance_rates}
  \vspace{-0.4em}
\end{figure*}

We first run DEO$_{\star}$-SGD$\times$P16 based on 16 chains and 20,000 iterations. We fix the lowest learning rate 0.003 and the highest learning 0.6 and propose to tune the target swap rate $\mathbb{S}$ for the acceleration-accuracy trade-off. Fig.\ref{test_of_acceptance_rates} shows that fixing $\mathbb{S}=0.2$ or 0.3 is too conservative and underestimates the uncertainty on the corners; $\mathbb{S}=0.6$ results in too many radical swaps and eventually leads to crude estimations; by contrast, $\mathbb{S}=0.4$ yields the best uncertainty approximation among the five choices.

Next, we select $\mathbb{S}=0.4$ and study the round trips. We observe in Fig.\ref{learning_rate_acceptance_rate}(a) that the vanilla DEO only yields 18 round trips every 1,000 iterations; by contrast, slightly increasing $W$ tends to improve the efficiency significantly and the optimal 45 round trips are achieved at $W=8$, which \emph{matches our theory}. In Fig.\ref{learning_rate_acceptance_rate}(b-c), the geometrically initialized learning rates lead to unbalanced acceptance rates in the early phase and some adjacent chains have few swaps and others swap too much, but as the optimization proceeds, the learning rates gradually converge. We also observe in Fig.\ref{learning_rate_acceptance_rate}(d) that the correction is adaptively estimated to ensure the average acceptance rates converge to $\mathbb{S}=0.4$.

\begin{figure*}[!ht]
  \centering
  \subfloat[Round trips]{\includegraphics[width=3cm, height=3cm]{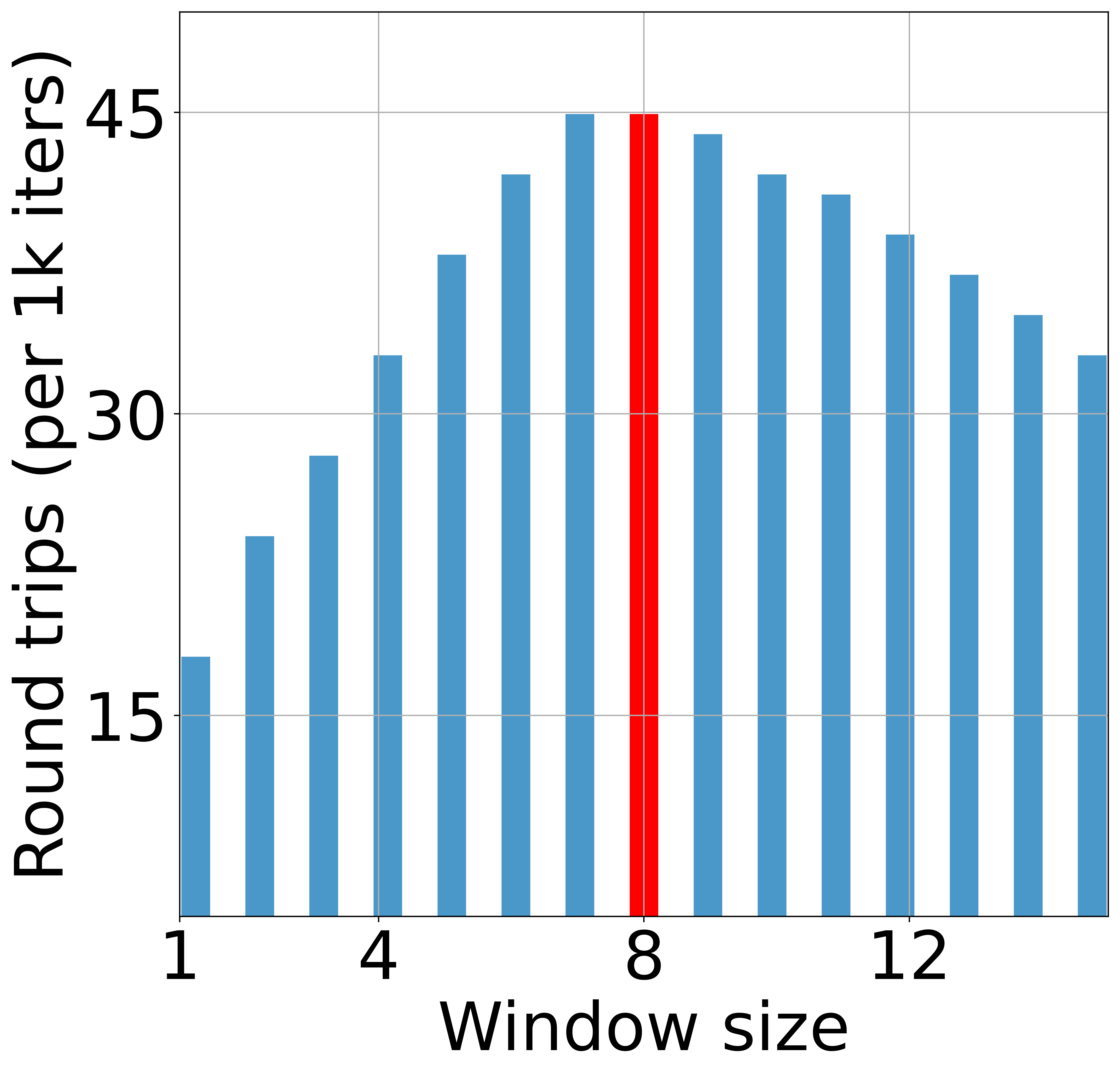}}\qquad
  \subfloat[Learning rates]{\includegraphics[width=3.4cm, height=3cm]{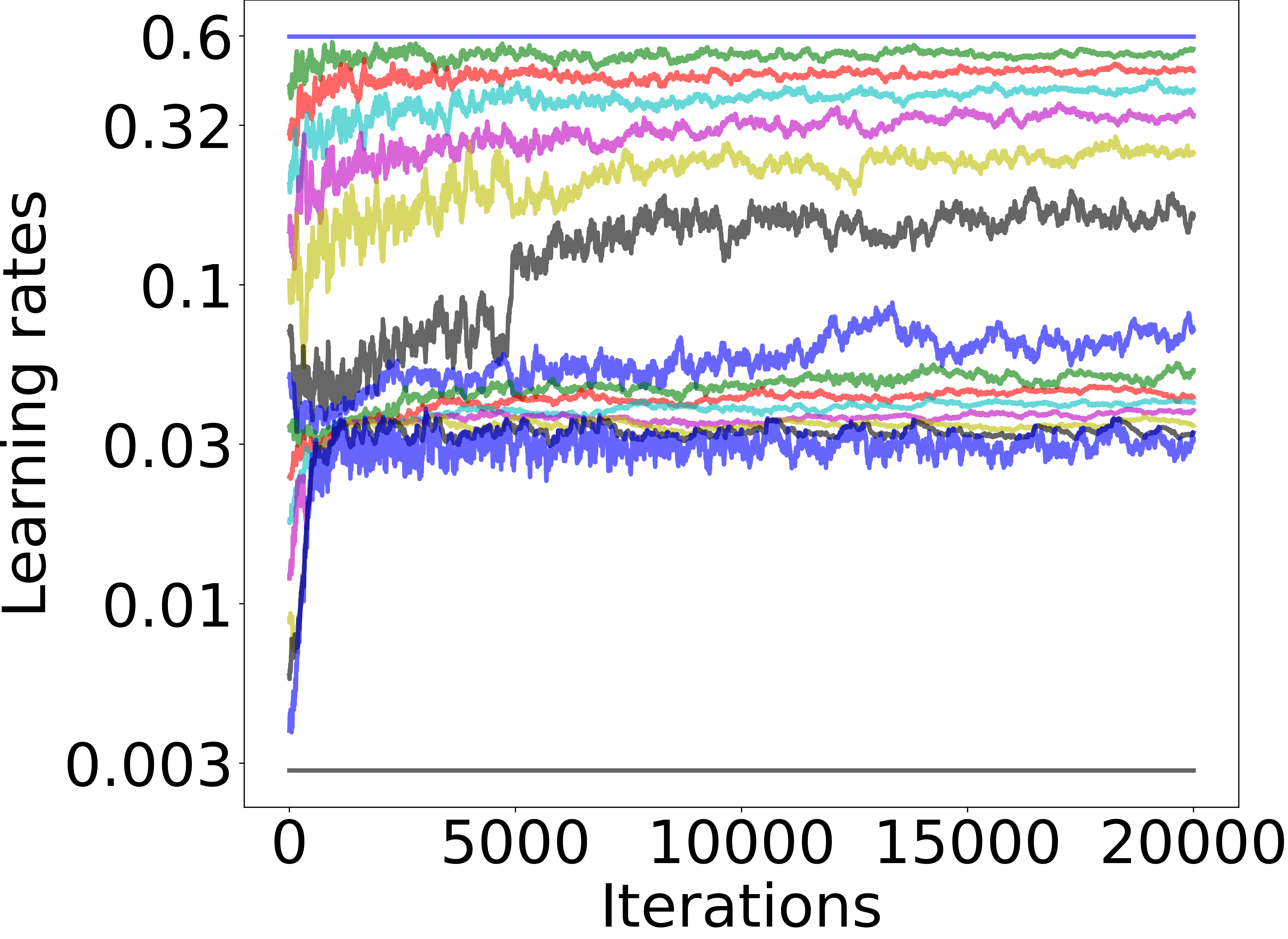}}\qquad
  \subfloat[Acceptance rates]{\includegraphics[width=3.4cm, height=3cm]{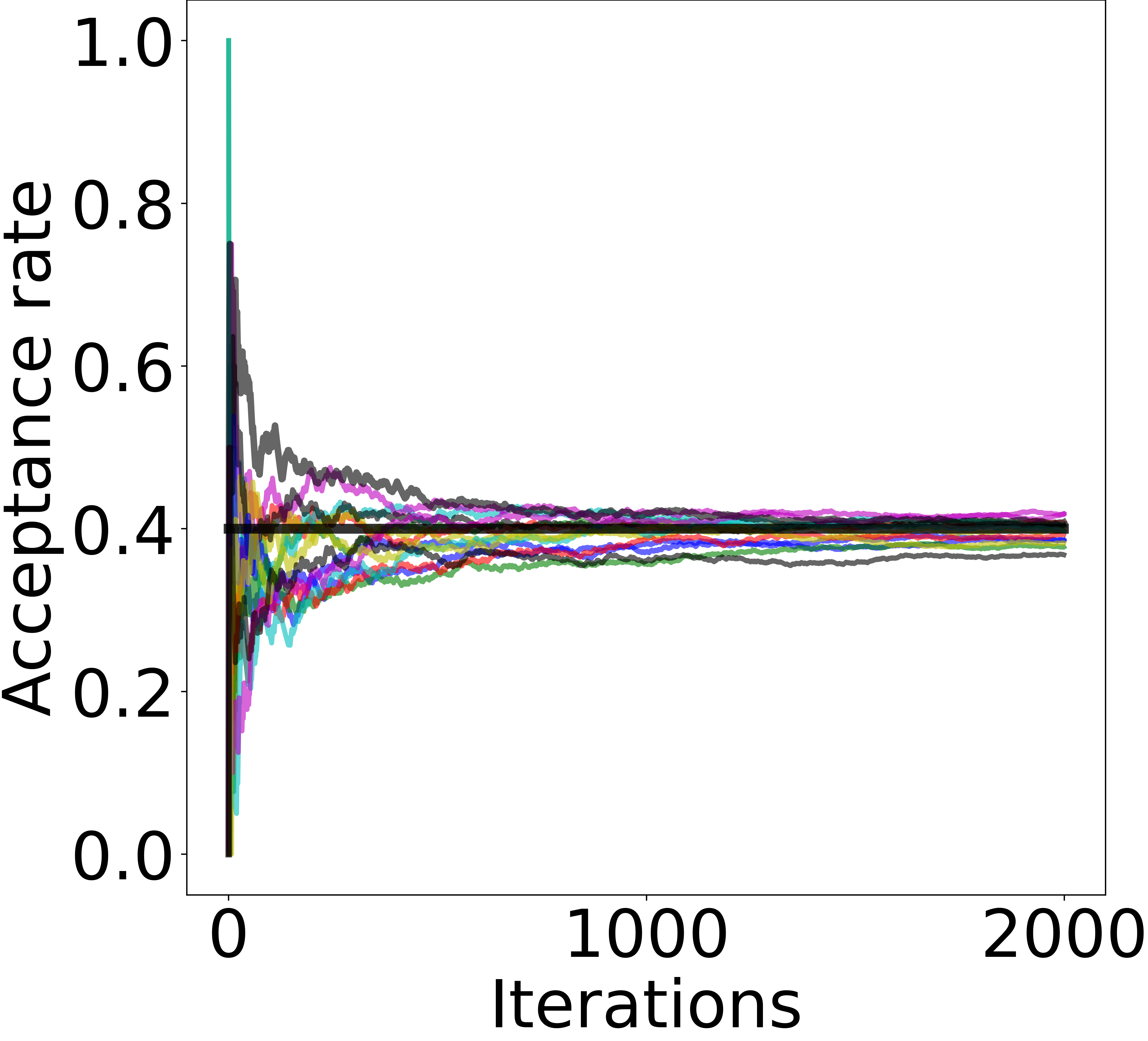}}\qquad
  \subfloat[Corrections]{\includegraphics[width=3.2cm, height=3cm]{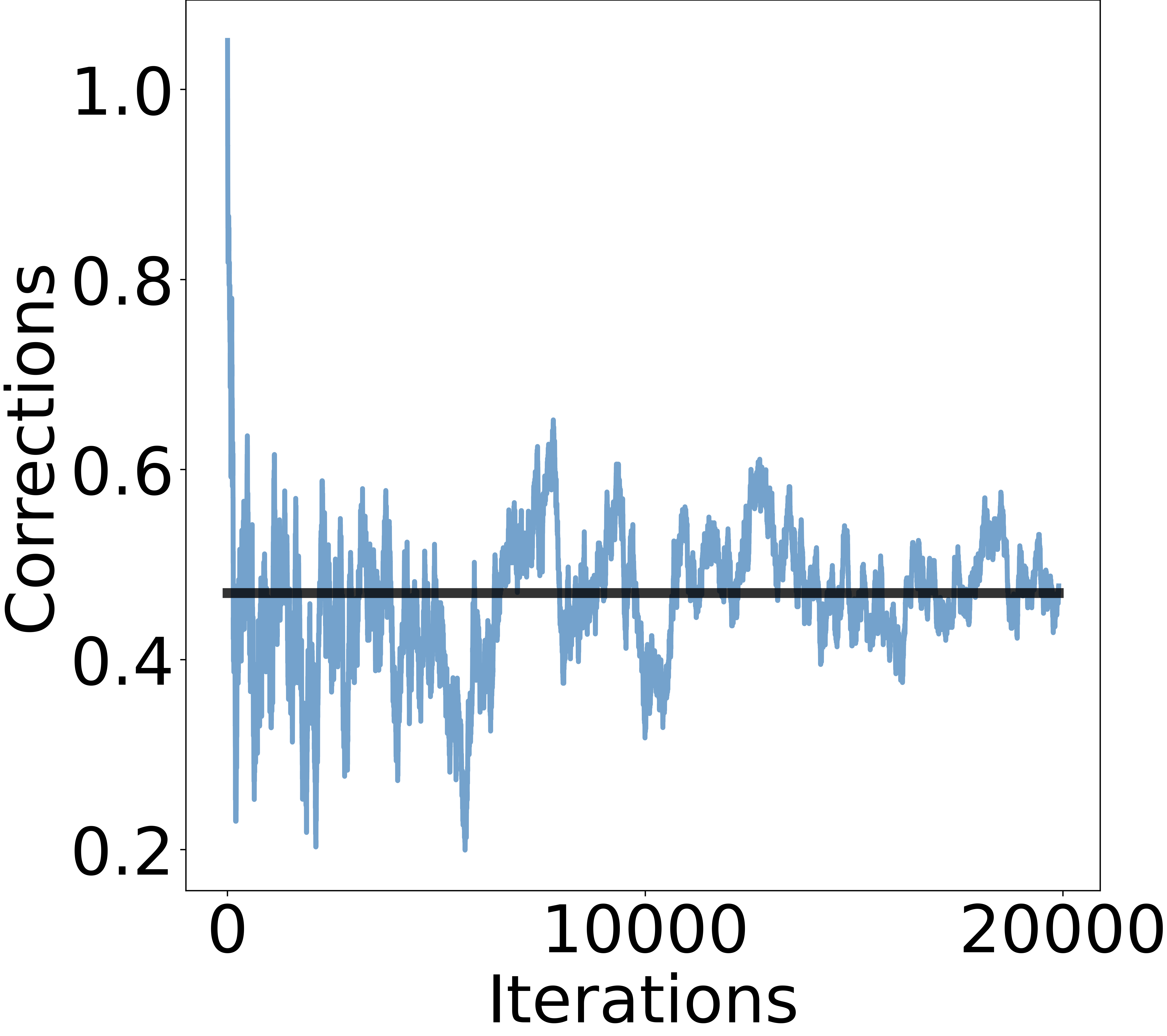}}\quad
  \caption{Study of window sizes, learning rates, acceptance rates, and the corrections.}
  \label{learning_rate_acceptance_rate}
  \vspace{-0.5em}
\end{figure*}

We compare the proposed algorithm with parallel SGLD based on 20,000 iterations and 16 chains (SGLD$\times$P16); we fix the learning rate 0.003 and a temperature 1. We also run cycSGLD$\times$T16, which is short for a single long chain based on 16 times of budget and cosine learning rates \cite{ruqi2020} of 100 cycles. We see in Figure \ref{multi_modal_simulation}(b) that SGLD$\times$P16 has good explorations but fails to quantify the uncertainty. Figure \ref{multi_modal_simulation}(c) shows that cycSGLD$\times$T16 explores most of the modes but overestimates some areas occasionally. Figure \ref{multi_modal_simulation}(d) demonstrates the DEO-SGD with 16 chains (DEO-SGD$\times$P16) estimates the uncertainty of the centering 9 modes well but fails to deal with the rest of the modes. As to DEO$_{\star}$-SGD$\times$P16, the approximation is rather accurate, as shown in Fig.\ref{multi_modal_simulation}(e).

\begin{figure*}[!ht]
  \centering
  \subfloat[\footnotesize{Ground truth} ]{\includegraphics[width=2.9cm, height=2.9cm]{figures/ground_truth_NRPT_iters_20000_chains_16_lr_0.01_0.5_sz_100_swap_rate_0.4_scale_2_window_8_seed_8.png}}\quad
  \subfloat[\footnotesize{SGLD$\times$P16}]{\includegraphics[width=2.9cm, height=2.9cm]{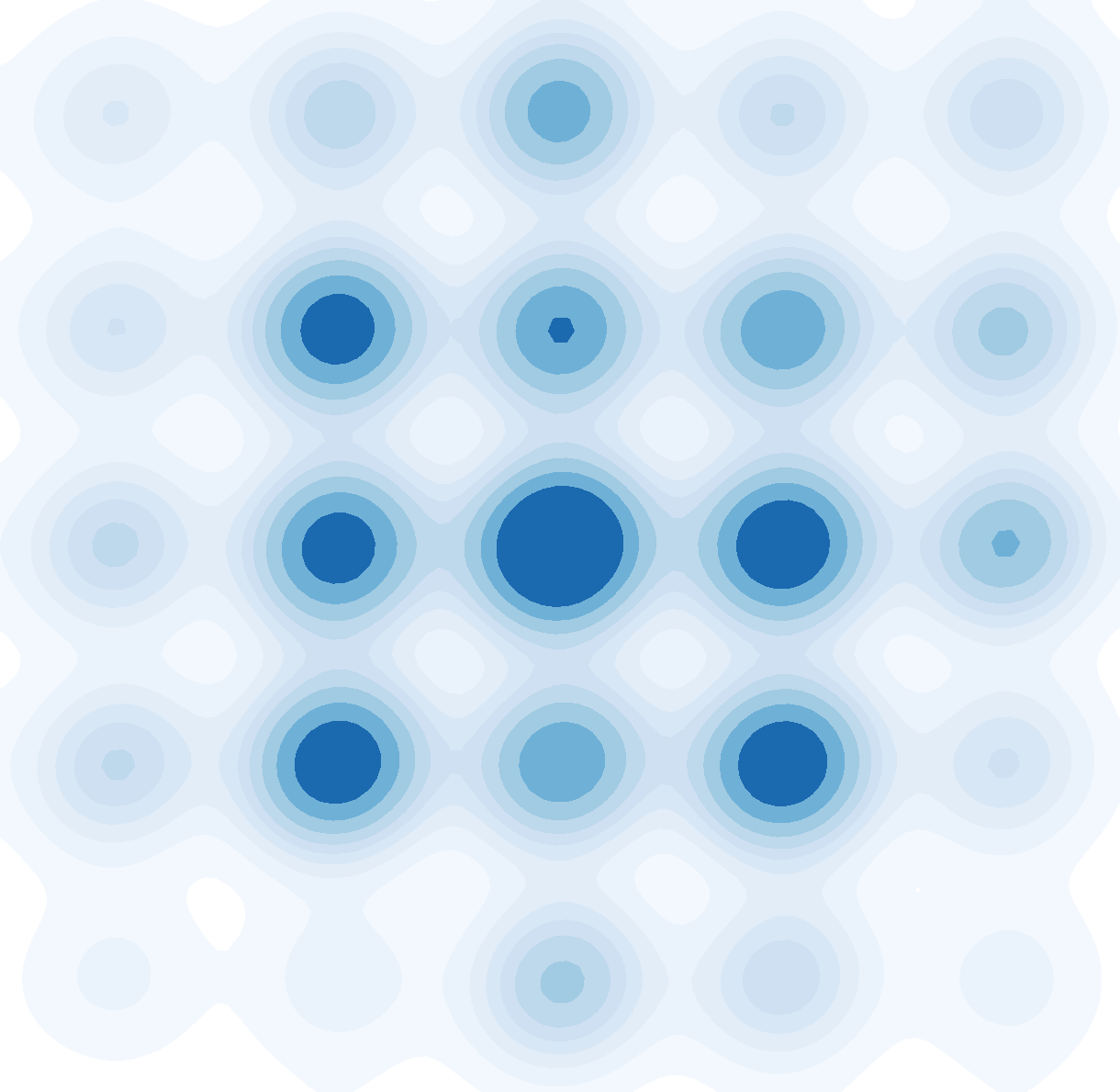}}\quad
  \subfloat[\scriptsize{cycSGLD$\times$T16}]{\includegraphics[width=2.9cm, height=2.9cm]{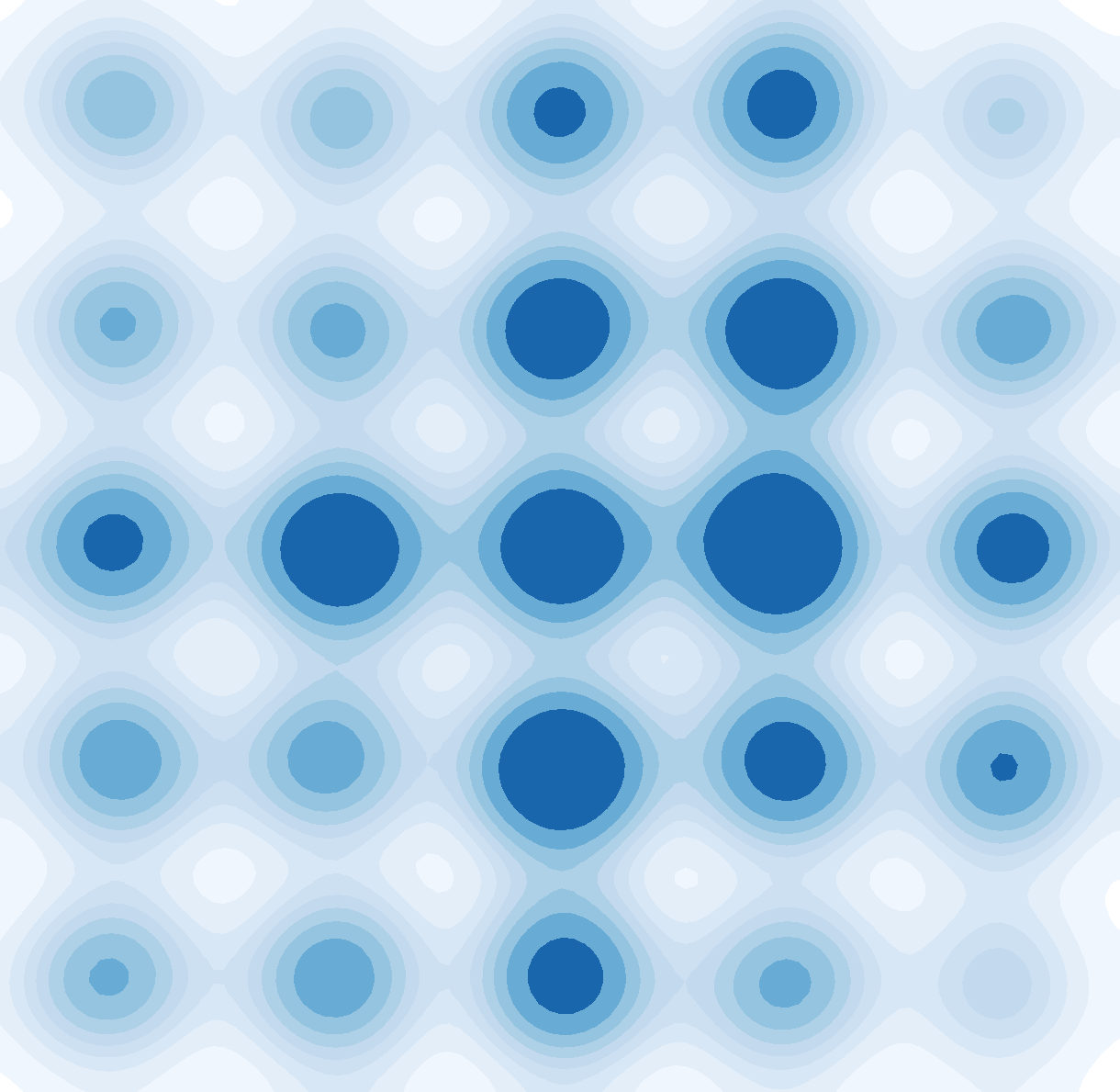}}\quad
  \subfloat[\scriptsize{$\text{DEO}$-SGD$\times$P16}]{\includegraphics[width=2.9cm, height=2.9cm]{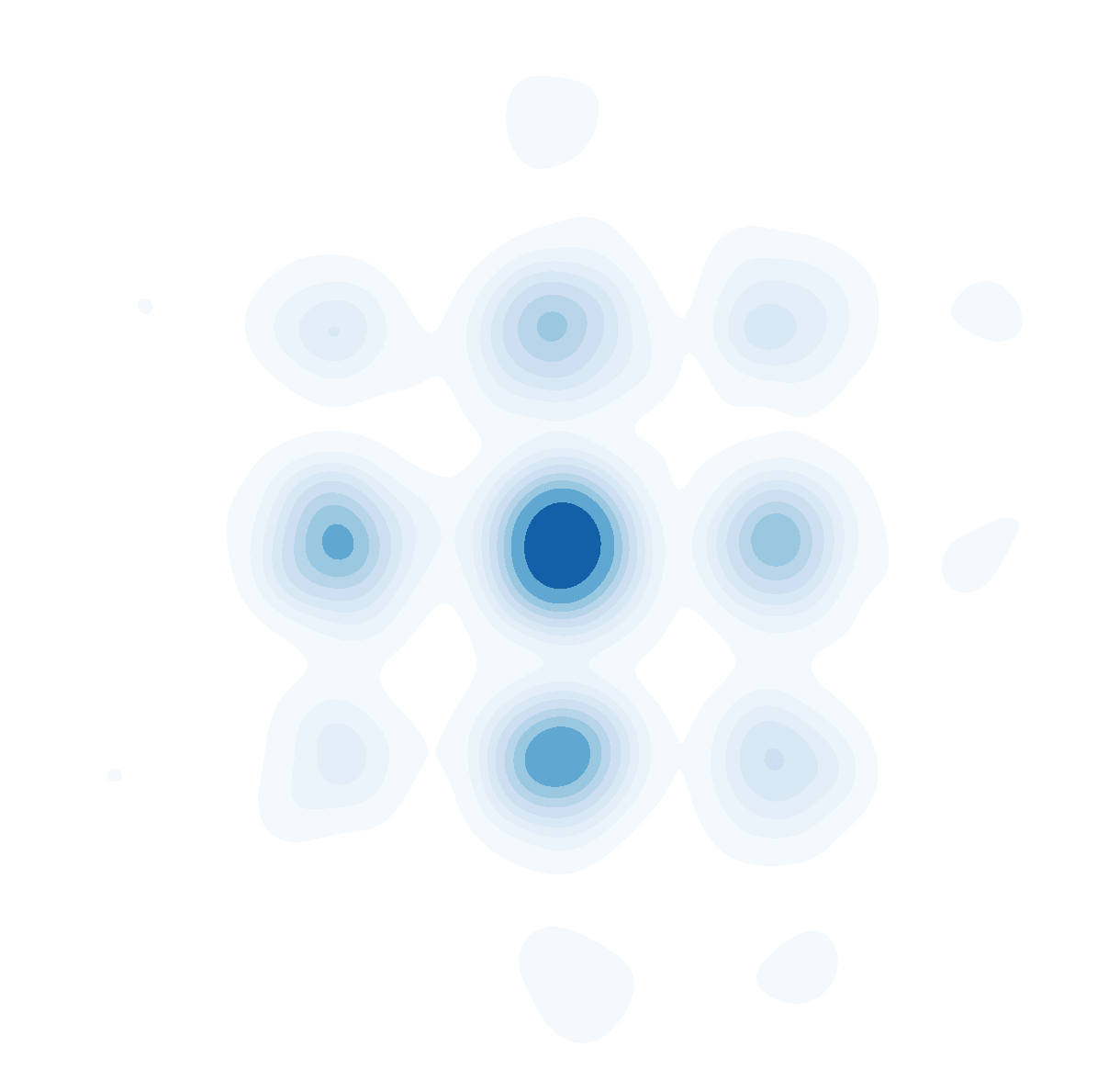}}\quad
  \subfloat[\scriptsize{$\text{DEO}_{\star}$-SGD$\times$P16}]{\includegraphics[width=2.9cm, height=2.9cm]{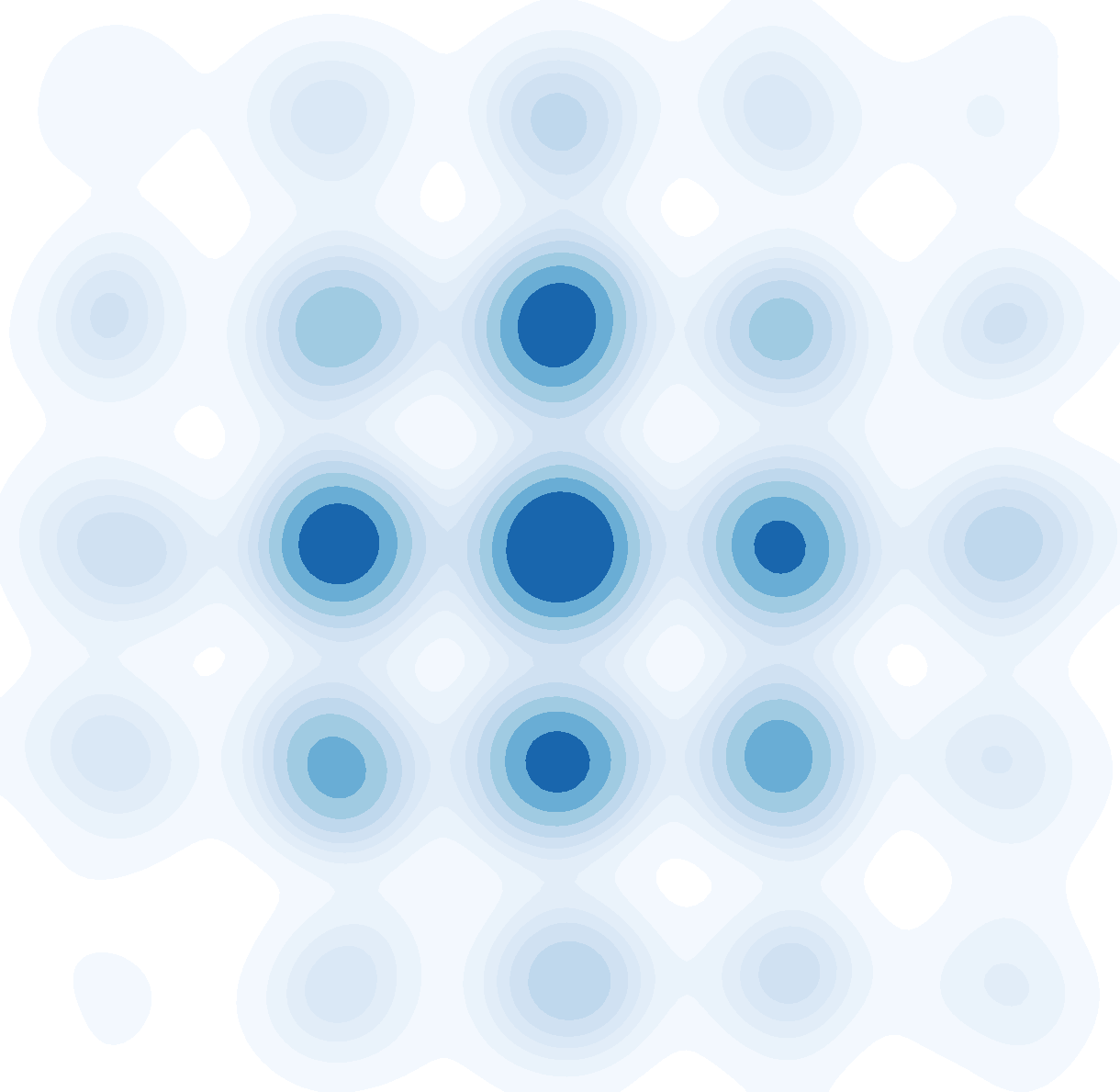}}
  \caption{Simulations of the multi-modal distribution through different sampling algorithms. }
  \label{multi_modal_simulation}
  \vspace{-0.5em}
\end{figure*}

We also present the index process for both schemes in Fig.\ref{Swapping_path_chains16}. We see that the vanilla DEO scheme results in volatile paths and a particle takes quite a long time to complete a round trip; by contrast, DEO$_{\star}$ only conducts at most one cheap swap in a window and yields much more deterministic paths. 

\begin{figure*}[!ht]
  \centering
  \subfloat[$\text{DEO}$-SGD$\times$P16]{\includegraphics[width=8.5cm, height=2cm]{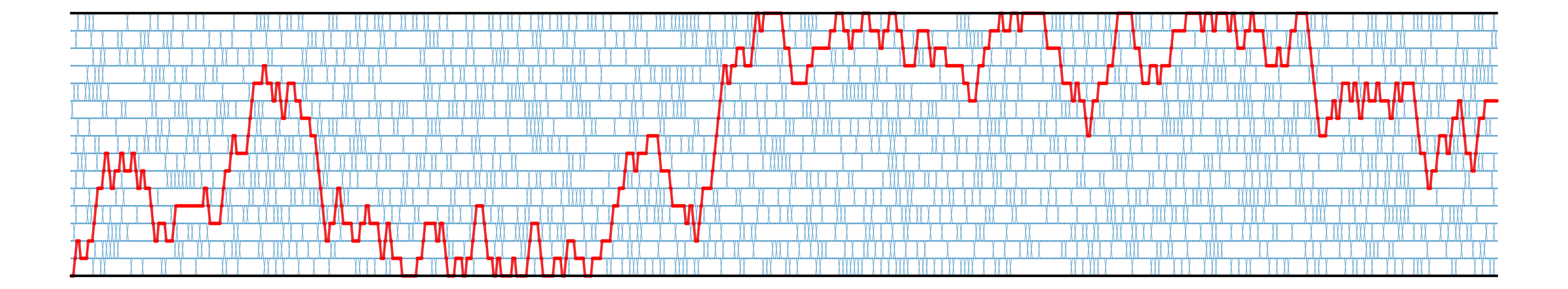}}\hspace*{0.5em} \subfloat[$\text{DEO}_{\star}$-SGD$\times$P16]{\includegraphics[width=8.5cm, height=2cm]{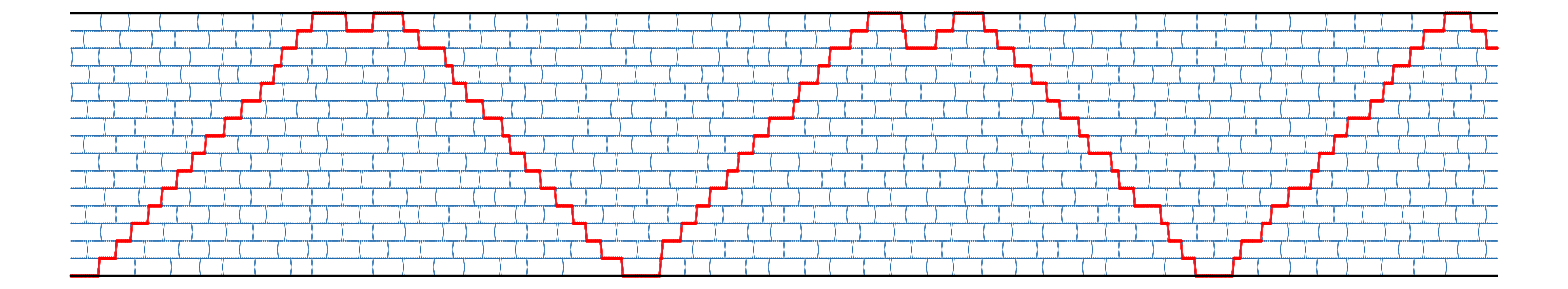}}
  \caption{Dynamics of the index process. The red path denotes the round trip path for a particle. }
  \label{Swapping_path_chains16}
  \vspace{-1em}
\end{figure*}

\subsection{Uncertainty Approximation and Optimization for Image Data}

Next, we conduct experiments on computer vision tasks. We choose ResNet20, ResNet32, and ResNet56 \cite{kaiming15} and train the models on CIFAR100. We not only report test negative log likelihood (NLL) but also present the test accuracy (ACC). For each ResNet model, we first pre-train 10 fixed models via 300 epochs and then run our algorithm based on momentum SGD (mSGD) for 500 epochs with 10 parallel chains and denote it by DEO$_{\star}$-mSGD$\times$P10. We fix the lowest and highest learning rates as 0.005 and 0.02, respectively. For a fair comparison, we also include the baseline DEO-mSGD$\times$P10 with the same setup except that the window size is 1; the standard ensemble mSGD$\times$P10 is also included with a learning rate of 0.005. In addition, we include two baselines based on a single long chain, i.e. we run stochastic gradient Hamiltonian Monte Carlo \cite{Chen14} 5000 epochs with cyclical learning rates and 50 cycles \cite{ruqi2020} and refer to it as cycSGHMC$\times$T10; we run SWAG$\times$T10 \cite{swag} under a similar setup.

In particular for DEO$_{\star}$-mSGD$\times$P10, we tune the target swap rate $\mathbb{S}$ and find an optimum at $\mathbb{S}=0.005$. We compare our proposed algorithm with the four baselines and observe in Table.\ref{UQ_test_NRPT} that  mSGD$\times$P10 can easily obtain competitive results simply through model ensemble \cite{Balaji17}, which outperforms cycSGHMC$\times$T10 and cycSWAG$\times$T10 on ResNet20 and ResNet32 models and perform the worst among the five methods on ResNet56; DEO-\upshape{m}SGD$\times$P10 itself is already a pretty powerful algorithm, however, $\text{DEO}_{\star}$-\upshape{m}SGD$\times$P10 consistently outperforms the vanilla alternative.

\begin{table*}[ht]
\begin{sc}
\footnotesize
\caption{Uncertainty approximation and optimization on CIFAR100 via $10\times$ budget.} \label{UQ_test_NRPT}
\vspace{0.15in}
\begin{center} 
\begin{tabular}{c|cc|cc|cc}
\hline
\multirow{2}{*}{Model} & \multicolumn{2}{c|}{R\upshape{es}N\upshape{et}20} & \multicolumn{2}{c|}{R\upshape{es}N\upshape{et}32} & \multicolumn{2}{c}{R\upshape{es}N\upshape{et}56}  \\
\cline{2-7}
 &  NLL  & ACC (\%) & NLL & ACC (\%) & NLL  & ACC (\%) \\
\hline
\hline
\scriptsize{\upshape{cyc}SGHMC$\times$T10} &  8198$\pm$59 & 76.26$\pm$0.18  &  7401$\pm$28 & 78.54$\pm$0.15 & 6460$\pm$21 & 81.78$\pm$0.08  \\ 
\scriptsize{\upshape{cyc}SWAG$\times$T10} &   8164$\pm$38 & 76.13$\pm$0.21  & 7389$\pm$32 & 78.62$\pm$0.13 & 6486$\pm$29 & 81.60$\pm$0.14 \\ 
\hline
\hline
\scriptsize{\upshape{m}SGD$\times$P10} &  7902$\pm$64 & 76.59$\pm$0.11 & 7204$\pm$29 & 79.02$\pm$0.09  & 6553$\pm$15  & 81.49$\pm$0.09    \\ 
\scriptsize{DEO-\upshape{m}SGD$\times$P10} &  7964$\pm$23  & 76.84$\pm$0.12  & 7152$\pm$41 & 79.34$\pm$0.15 & 6534$\pm$26 & 81.72$\pm$0.12 \\ 
\scriptsize{$\text{DEO}_{\star}$-\upshape{m}SGD$\times$P10} &  \textbf{7741$\pm$67} & \textbf{77.37$\pm$0.16} & \textbf{7019$\pm$35} & \textbf{79.54$\pm$0.12} & \textbf{6439$\pm$32} & \textbf{82.02$\pm$0.15} \\ 
\hline
\end{tabular}
\end{center}
\end{sc}
\vspace{-0.1in}
\end{table*}

To analyze why the proposed scheme performs well, we study the round trips in Figure.\ref{learning_rate_acceptance_rate_v2}(a) and find that the theoretical optimal window obtains around 11 round trips every 100 epochs, which is almost 2 times as much as the vanilla DEO scheme. In Figure \ref{learning_rate_acceptance_rate_v2}(b), we observe that the smallest learning rate obtains the highest accuracy (blue) for exploitations, while the largest learning rate yields decent explorations (red); we see in Figure \ref{learning_rate_acceptance_rate_v2}(c-d) that geometrically initialized learning rate fails in producing equi-acceptance, but as the training proceeds, the learning rates converge to fixed points and the acceptance rates for different pairs converge to the target swap rate.

\begin{figure*}[!ht]
  \centering
  \subfloat[Round trips]{\includegraphics[width=3.5cm, height=3.2cm]{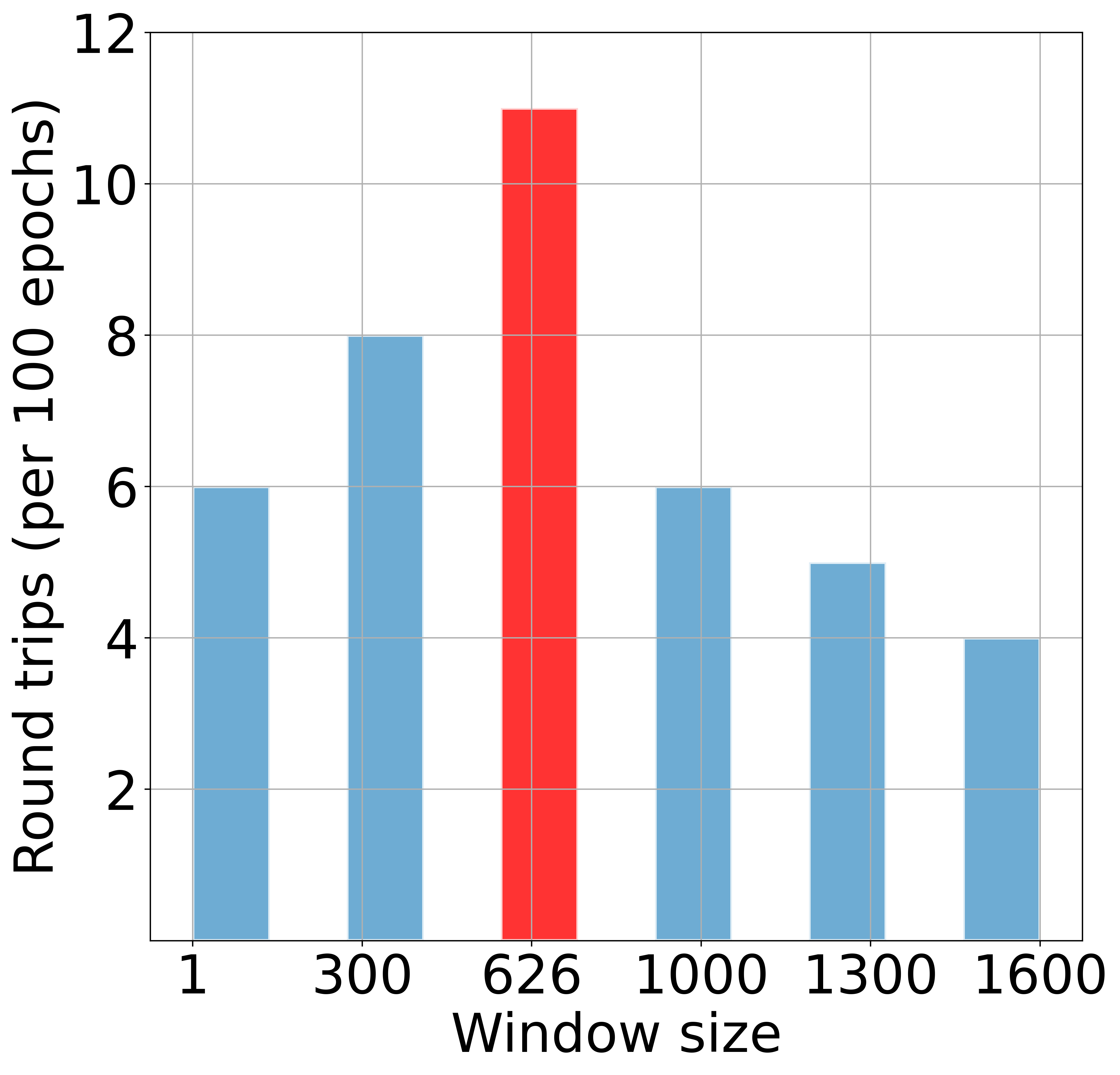}}\qquad
  \subfloat[Accuracies]{\includegraphics[width=3.5cm, height=3.2cm]{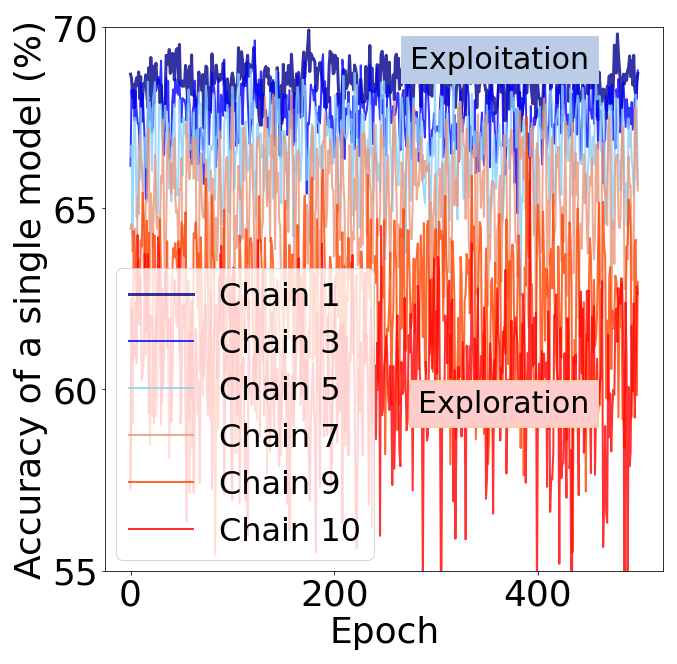}}\qquad
  \subfloat[Learning rates]{\includegraphics[width=3.5cm, height=3.2cm]{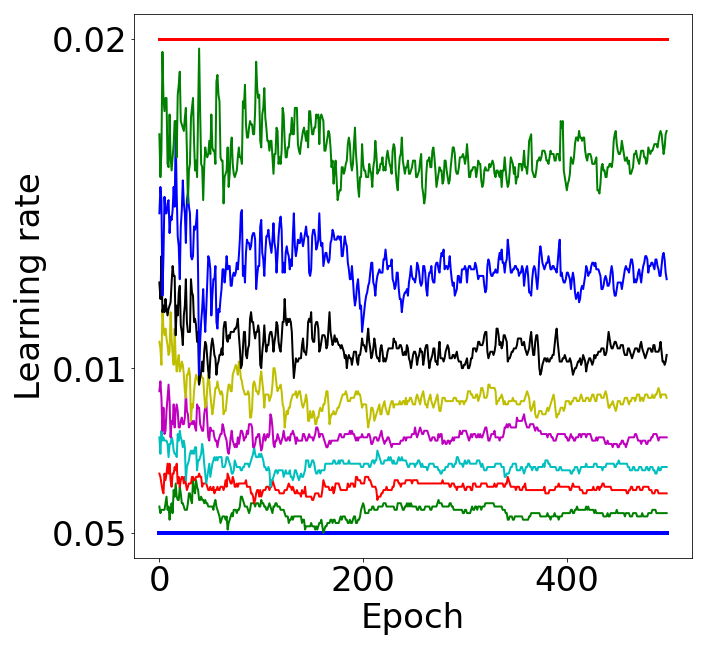}}\qquad
  \subfloat[Acceptance rates]{\includegraphics[width=3.5cm, height=3.2cm]{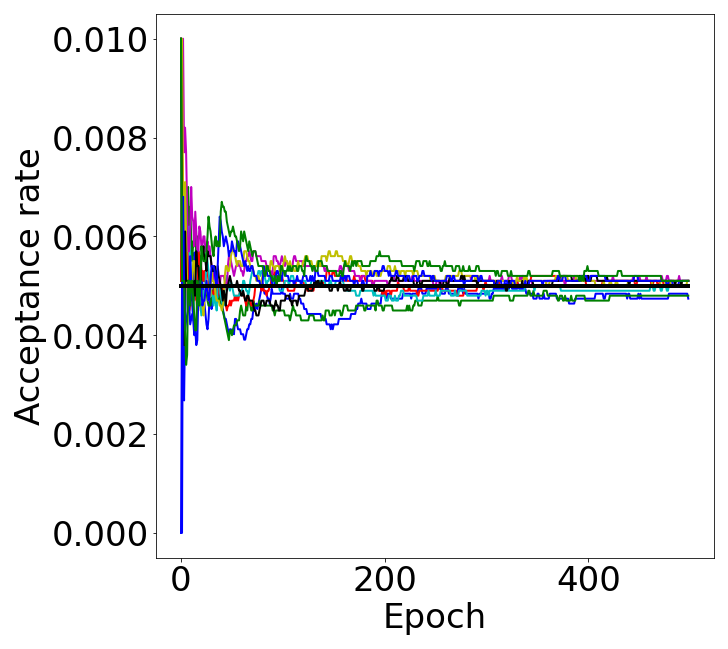}}\qquad
  \caption{Study of window sizes, accuracies, learning rates, and acceptance rates on ResNet20.}
  \label{learning_rate_acceptance_rate_v2}
  \vspace{-0.5em}
\end{figure*}

For the visualization of the index process, we observe in Figure.\ref{adaptive_path_cifar}(a) that the vanilla DEO scheme leads to volatile trajectories wandering back and forth and is rather inefficient; by contrast, the DEO$_{\star}$ scheme yields well-motivated paths with more deterministic round trips. Interestingly, this path resembles cyclic learning rates \cite{ruqi2020}, which provides a novel viewpoint to \emph{interpret why cyclic learning rates work well empirically}. Nevertheless, the stochastic and parallel manner further improves the margin and eventually leads to the most efficient approximations in this task.

\begin{figure*}[!ht]
  \centering
  \subfloat[$\text{DEO}$-mSGD$\times$P10]{\includegraphics[width=8.5cm, height=2.1cm]{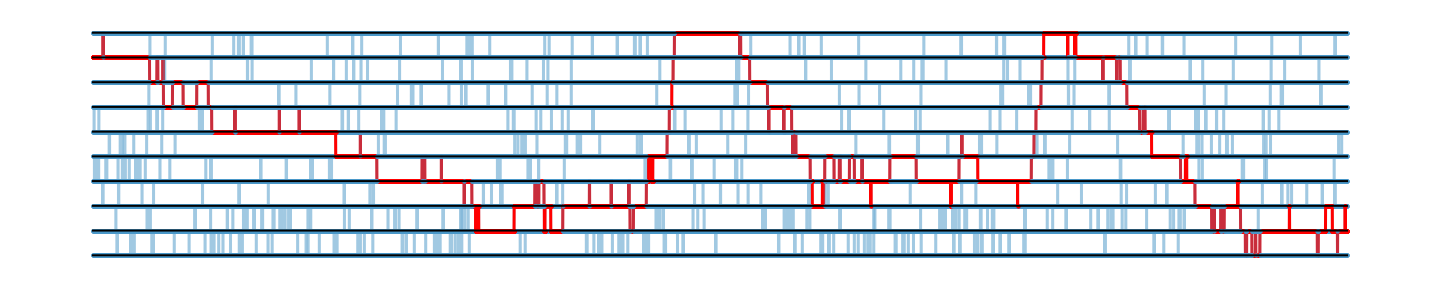}}\hspace*{0.5em} \subfloat[$\text{DEO}_{\star}$-mSGD$\times$P10]{\includegraphics[width=8.5cm, height=2.1cm]{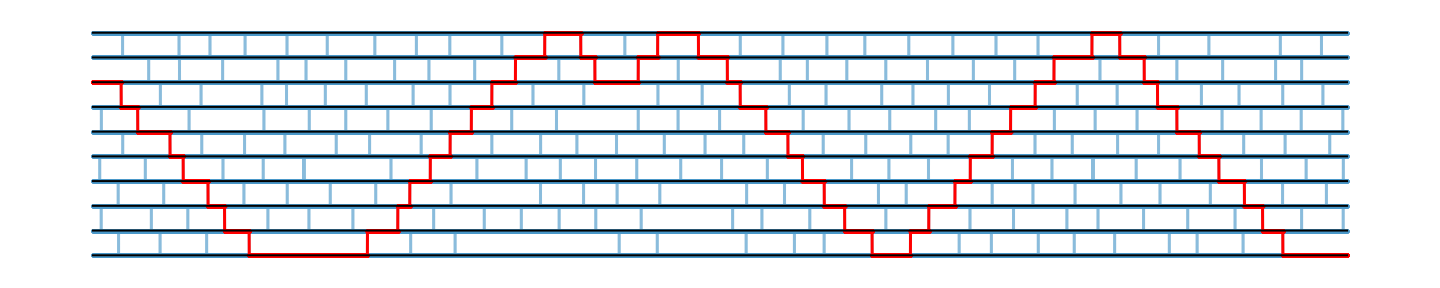}}
  \caption{Dynamics of the index process of ResNet20 models in the last 200 epochs. }
  \label{adaptive_path_cifar}
  \vspace{-1em}
\end{figure*}

\section{Conclusion}

In this chapter, we show how to conduct efficient PT in big data problems. To tackle the inefficiency issue of the popular DEO scheme given limited chains, we present the DEO$_{\star}$ scheme by applying an optimal window size to encourage \emph{deterministic paths} and obtain in a significant \emph{acceleration of $O(\frac{P}{\log P})$ times}. For a user-friendly purpose, we propose a deterministic swap condition to interact with SGD-based exploration kernels and provide a theoretical guarantee to control the bias solely depending on the learning rate $\bm{\eta}$ and the correction buffer $\mathbb{C}$; we also provide a practical algorithm to adaptively approximate $\bm{\eta}$ and $\mathbb{C}$ for achieving the optimal efficiency in \emph{constrained settings}.

\part{DYNAMIC IMPORTANCE SAMPLING}
\chapter{CONTOUR STOCHASTIC GRADIENT LANGEVIN DYNAMICS FOR SIMULATIONS OF MULTI-MODAL DISTRIBUTIONS}
\label{CSGLD_multi_distribution}
\section{Introduction}

The flat histogram algorithms, such as the multicanonical \cite{Berg1991Multicanonical} and Wang-Landau \cite{WangLandau2001} algorithms, were first proposed to sample discrete states of Ising models by yielding a flat histogram in the energy space, and then extended as a general dynamic importance sampling algorithm, the so-called stochastic approximation Monte Carlo (SAMC) algorithm \cite{Liang05, Liang07, LiangPL2009}. Theoretical studies \cite{leli2008, Liang10, Fort15} support the efficiency of the flat histogram algorithms in Monte Carlo computing for small data problems. However, it is still unclear how to adapt the flat histogram idea to accelerate the convergence of SGMCMC, ensuring efficient uncertainty quantification for AI safety problems. 

This chapter proposes the so-called contour stochastic gradient Langevin dynamics (CSGLD) algorithm, which successfully extends the flat histogram idea to SGMCMC. Like the SAMC algorithm \cite{Liang05, Liang07, LiangPL2009}, CSGLD works as a dynamic importance sampling algorithm, which adaptively adjusts the target measure at each iteration and accounts for the bias introduced thereby by importance weights. However, theoretical analysis for the two types of dynamic importance sampling algorithms can be quite different due to the fundamental difference in their transition kernels. We proceed by justifying the stability condition for CSGLD based on the perturbation theory, and establishing ergodicity of CSGLD based on newly developed theory for the convergence of adaptive SGLD. Empirically, we test the performance of CSGLD through a few experiments. It achieves remarkable performance on some synthetic data, UCI datasets, and computer vision datasets such as CIFAR10 and CIFAR100.

\section{Contour Stochastic Gradient Langevin Dynamics}

Suppose we are interested in sampling from a probability measure $\pi(\bx)$ with the density given by  
\begin{equation} \label{CSGLDeq1}
\pi(\bx) \propto \exp(-U(\bx)/\tau), \quad \bx \in \MX,
\end{equation}
where $\MX$ denotes the sample space, $U(\bx)$ is the energy function, and $\tau$ is the temperature. It is known that when $U(\bx)$ is highly non-convex, SGLD can mix very slowly \cite{Maxim17}. To accelerate the convergence, we exploit the flat histogram idea in SGLD.
 
Suppose that we have partitioned the sample space $\MX$ into $m$ subregions based on the energy function $U(\bx)$:   $\MX_1=\{\bx: U(\bx) \leq {u}_1\}$, $\MX_2=\{\bx: {u}_1 < U(\bx) \leq {u}_2\}$, $\ldots$, $\MX_{m-1}=\{\bx: {u}_{m-2} < U(\bx) \leq {u}_{m-1} \}$, and $\MX_{m}=\{\bx: U(\bx) >{u}_{m-1} \}$, where $-\infty <  {u}_1 < {u}_2 < \cdots < {u}_{m-1} <\infty$ are specified by the user. For convenience, we set $u_0=-\infty$ and $u_m=\infty$. Without loss of generality, we assume ${u}_{i+1}-{u}_{i}=\Delta u$ for $i=1,\ldots,m-2$. We propose to simulate from a flattened density 
\begin{equation} \label{1keq1} 
\varpi_{\Psi_{\btheta}}(\bx) \propto \frac{\pi(\bx)}{\Psi^{\zeta}_{\btheta}(U(\bx))},
\end{equation}
where $\zeta>0$ is a hyperparameter controlling the geometric property of the flatted density (see Figure \ref{CSGLD_stats}(a) for illustration),  and $\btheta=(\theta(1), \theta(2), \ldots, \theta(m))$ is an unknown vector taking values in the space: 
\begin{equation}\small
 \bTheta=\left\{\left(\theta(1),\theta(2),\cdots, \theta(m)\right)\mid 0<\theta(1),\theta(2),\cdots, \theta(m)<1 \text{ and } \sum_{i=1}^m \theta(i)=1 \right\}.
\end{equation}

\subsection{A na\"{i}ve Contour SGLD} It is known if we set \footnote{$1_{A}$ is an indicator function that takes value $1$ if event $A$ occurs and $0$ otherwise.}

\begin{equation}
\label{proposal_iter}
\begin{split}
     &\text{(i)    }\ \zeta=1 \text{ and }\Psi_{\btheta}(U(\bx))= \sum_{i=1}^m \theta(i) 1_{u_{i-1} < U(\bx) \leq u_i}, \\
      &\text{(ii)   } \theta(i)=\theta_{\star}(i), \text{where }\theta_{\star}(i) =\int_{\bchi_i}\pi(\bx)d\bx \text{ for } i\in\{1,2,\cdots, m\},\\
\end{split}
\end{equation}
the algorithm will act like the SAMC algorithm \cite{Liang07},  yielding a flat histogram in the space of energy (see the pink curve in Fig.\ref{CSGLD_stats}(b)). Theoretically, such a density flattening strategy enables a sharper logarithmic Sobolev inequality and accelerates the convergence of simulations \cite{leli2008, Fort15}. However, such a density flattening setting only works under the framework of the Metropolis algorithm \cite{Metropolis1953}. A na\"{i}ve application of the step function in formula (\ref{proposal_iter}(i)) to SGLD results in $\frac{\partial \log \Psi_{{\theta}}(u)}{\partial u}=\frac{1}{\Psi_{\theta}(u)}\frac{\partial \Psi_{{\theta}}(u)}{\partial u}=0$ almost everywhere, which leads to the \emph{vanishing-gradient problem} for SGLD. Calculating the gradient for the na\"{i}ve contour SGLD, we have 
\begin{equation*}
\small{
    \nabla_{x} \log \varpi_{\Psi_{\theta}}(x)=-\left[1+ \zeta \tau\frac{\partial \log{\Psi_{\theta}}(u)}{\partial u} \right] \frac{\nabla_{x} U(x)}{\tau}=-\frac{\nabla_{x} U(x)}{\tau}.}
\end{equation*} 
As such, the na\"{i}ve algorithm behaves like SGLD and fails to simulate from the flattened density (\ref{1keq1}).

\subsection{How to Resolve the Vanishing Gradient} To tackle this issue, we propose to set $\Psi_{\btheta}(u)$ 
as a piecewise continuous function: 
\begin{equation*}
\Psi_{\btheta}(u)= \sum_{i=1}^m \left(\theta(i-1)e^{(\log\theta(i)-\log\theta(i-1)) \frac{u-u_{i-1}}{\Delta u}}\right) 1_{u_{i-1} < u \leq u_i},
\end{equation*}
where $\theta(0)$ is fixed to $\theta(1)$ for simplicity. A direct calculation shows that
 \begin{equation}
 \label{marK}
 \small
 \begin{split}
     \nabla_{\bx} \log \varpi_{\Psi_{\btheta}}(\bx)
  &=-\left[1+ \zeta \tau\frac{\partial \log{\Psi_{\btheta}}(u)}{\partial u} \right] 
   \frac{\nabla_{\bx} U(\bx)}{\tau}\\
  &=-\left[1+ \zeta \tau {\frac{\log\theta(J(\bx))-\log\theta((J(\bx)-1)\vee 1)}{\Delta u}} \right] 
   \frac{\nabla_{\bx} U(\bx)}{\tau},
 \end{split}
 \end{equation}
where $J(\bx) \in\{1,2,\cdots, m\}$ denotes the index  that $\bx$ belongs to, i.e., $u_{J(\bx)-1}< U(\bx)\leq u_{J(\bx)}$. \footnote[4]{Formula (\ref{marK}) shows a practical numerical scheme. An alternative is presented in the supplementary material. } 

\subsection{Estimation via Stochastic Approximation}
Since $\btheta_{\star}$ is unknown, we propose to estimate it on the fly under the framework of 
stochastic approximation \cite{RobbinsM1951}. Provided that a scalable transition kernel $\Pi_{\bm{\theta_{k}}}(\bm{x}_{k}, \cdot)$ is available and the energy function $U(\bx)$ on the full data can be efficiently evaluated, the weighted density  $\varpi_{\Psi_{\btheta}}(\bx)$ can be simulated by iterating between the following steps:
\begin{equation}
\begin{split}
\label{sa_framework}
    &\text{ (i) Simulate $\bm{x}_{k+1}$ from $\Pi_{\bm{\theta_{k}}}(\bm{x}_{k}, \cdot)$, which admits $\varpi_{\bm{\theta}_{k}}(\bm{x})$ as
the invariant distribution,} \\
    &\text{(ii) $\theta_{k+1}(i)={\theta}_{k}(i)+\omega_{k+1}{\theta}_{k}^{\zeta}( J(\bx_{k+1}))\left(1_{i= J(\bx_{k+1})}-{\theta}_{k}(i)\right)$ \text{ for } $i\in\{1,2,\cdots,m\}.$}
\end{split}    
\end{equation}
where $\btheta_k$ denotes a working estimate of $\btheta$ at the $k$-th iteration. We expect that in a long run, such an algorithm can achieve \emph{an optimization-sampling equilibrium} such that $\btheta_{k}$ converges to the fixed point $\btheta_{\star}$ and the random vector $\bx_{k}$ converges weakly to the distribution $\varpi_{\Psi_{\btheta_{\star}}}(\bx)$. 

To make the algorithm scalable to big data, we propose to adopt the Langevin transition kernel for drawing samples at each iteration, for which a mini-batch of data can be used to accelerate computation. In addition, 
we observe that evaluating $U(\bx)$ on the full data can be quite expensive for big data problems, while it is free to obtain the stochastic energy $\widetilde{U}(\bx)$ in evaluating the stochastic  gradient $\nabla_{\bx} \widetilde{U}(\bx)$ due to the nature of auto-differentiation \cite{paszke2017}. For this reason, we propose a biased index $\tilde J(\bx)$, where $u_{\tilde J(\bx)-1}< \frac{N}{n}\widetilde U(x)\leq u_{\tilde J(\bx)}$, $N$ is the sample size of the full dataset and $n$ is the mini-batch size. Let $\{\epsilon_k\}_{k=1}^{\infty}$ and $\{\omega_k\}_{k=1}^{\infty}$ denote the learning rates and step sizes for SGLD and stochastic approximation, respectively. Given the above notations, the proposed algorithm can be presented in Algorithm \ref{alg:CSGLD},
 which can be viewed as  a \emph{scalable Wang-Landau algorithm} for deep learning and big data problems. 
 
\subsection{Related Work} 
Compared to the existing MCMC algorithms, the proposed   algorithm has a few innovations:
 \begin{algorithm}[tb]
   \caption{Contour SGLD Algorithm. One can conduct a resampling step from the pool of importance samples according to the importance weights to obtain the original distribution.}
   \label{alg:CSGLD}
\begin{algorithmic}
   \STATE {\bfseries [1.] (Data subsampling)} Simulate a mini-batch of data of size $n$ from the whole dataset of size $N$; Compute the stochastic gradient $\nabla_{\bx}\widetilde U(\bx_k)$ and stochastic energy $\widetilde U(\bx_k)$.  

   \STATE {\bfseries [2.] (Simulation step)}
   Sample $\bx_{k+1}$ using the SGLD algorithm based on $\bx_k$ and $\btheta_k$, i.e.,
  \begin{equation} \label{SGLDeq6}
  \footnotesize
 \begin{split}
  \small{\bx_{k+1}=\bx_k - \epsilon_{k+1} \frac{N}{n} \left[1+ 
   \zeta\tau\frac{\log {\theta}_{k}(\tilde J(\bx_k)) - \log{\theta}_{k}((\tilde J(\bx_k)-1)\vee 1)}{\Delta u}  \right]  
    \nabla_{\bx} \widetilde U(\bx_k) +\sqrt{2 \tau \epsilon_{k+1}} \bw_{k+1}}, 
 \end{split}
  \end{equation}
  where $\bw_{k+1} \sim N(0,\bm{I}_d)$, $d$ is the dimension,
  $\epsilon_{k+1}$ is the learning rate,
  and $\tau$ is the temperature. 

  \STATE {\bfseries [3.] (Stochastic approximation)} Update the estimate of $\theta(i)$'s
 for $i=1,2,\ldots,m$ by setting
  \begin{equation} \label{updateeq_csgld}
 {\theta}_{k+1}(i)={\theta}_{k}(i)+\omega_{k+1}{\theta}_{k}^{\zeta}(\tilde J(\bx_{k+1}))\left(1_{i=\tilde J(\bx_{k+1})}-{\theta}_{k}(i)\right), 
 \end{equation} 
 where $1_{i=\tilde J(\bx_{k+1})}$ is an indicator function which equals 1 if $i= \tilde J(\bx_{k+1})$ and 0 otherwise.
\vspace{-0.04in}
\end{algorithmic}
\end{algorithm}
 
First, CSGLD is an adaptive MCMC algorithm based on the \emph{Langevin transition kernel} instead of the {\it Metropolis transition kernel} \cite{Liang07, Fort15}. As a result, the existing convergence theory for the Wang-Landau algorithm does not apply. 
To resolve this issue, we first prove a stability condition for CSGLD based on the perturbation theory, and then verify regularity conditions for the solution of the Poisson equation so that the fluctuations of the mean-field system induced by CSGLD get controlled, which eventually ensures convergence of CSGLD.

Second, the use of the stochastic index $\tilde J(\bx)$  avoids the evaluation of $U(\bx)$ on the full data and thus significantly accelerates the computation of the algorithm, although it leads to a small bias, depending on the mini-batch size $n$, in parameter estimation. Compared to other methods, such as using a fixed sub-dataset to estimate $U(\bx)$, the implementation is much simpler. Moreover, combining the variance reduction of the noisy energy estimators \cite{deng_VR}, the bias also decreases to zero asymptotically as $\epsilon\rightarrow 0$.

Third, unlike the existing SGMCMC algorithms \cite{Welling11, Chen14, yian2015}, CSGLD works as a \emph{dynamic importance sampler} which \emph{flattens} the target distribution and \emph{reduces the energy barriers} for the sampler to traverse between different regions of the energy landscape (see Fig.\ref{CSGLD_stats}(a) for illustration). The sampling bias introduced thereby is accounted for by the importance weight $\theta^{\zeta}(\tilde J(\cdot))$. Interestingly, CSGLD possesses a {\it self-adjusting mechanism} to ease escapes from local traps, which is similar to the self-repulsive dynamics \cite{mao_mcmc} and can be explained as follows. 
Let's assume that the sampler gets 
trapped into a local optimum at iteration $k$. 
Then CSGLD will automatically increase the multiplier
of the stochastic gradient (i.e., the bracket term of 
(\ref{SGLDeq6})) at iteration $k+1$ by 
increasing the value of $\theta_{k}(\tilde{J}(\bx))$, while decreasing the components of $\btheta_{k}$ corresponding to other subregions. 
This adjustment will continue until the sampler moves away from the current subregion. Then, in the followed several iterations, the multiplier might become negative in neighboring subregions of the local optimum due to the increased value of $\theta(\tilde{J}(\bx))$, which continues to help to drive the sampler to higher energy regions and thus escape from the local trap. That is, in order to escape from local traps, CSGLD is sometimes forced to move toward higher energy regions by changing the sign of the stochastic gradient multiplier! This is a very attractive feature for simulations of multi-modal distributions.

\section{Theoretical Study of the CSGLD Algorithm} \label{convergSect}

In this section, we study the convergence of CSGLD algorithm under the framework of stochastic approximation and show the ergodicity property based on weighted averaging estimators.
 
\subsection{Convergence Analysis} \label{FMalg}

Following the tradition of stochastic approximation analysis, we rewrite the updating rule (\ref{updateeq_csgld}) 
as 
 \begin{equation}
     \btheta_{k+1}=\btheta_k+\omega_{k+1} \widetilde H(\btheta_k,\bx_{k+1}),
 \end{equation}
where $\widetilde H(\btheta,\bx)=(\widetilde H_1(\btheta,\bx), \ldots, 
 \widetilde H_m(\btheta,\bx))$ is a random field function with 
 \begin{equation}
 \label{H_}
     \widetilde H_i(\btheta,\bx)={\theta}^{\zeta}(\tilde J(\bx))\left(1_{i= \tilde J(\bx)}-{\theta}(i)\right), \quad i=1,2,\ldots,m.
 \end{equation}
Notably, $\widetilde H(\btheta,\bx)$ works under an empirical measure $\varpi_{\btheta}(\bx)$ which approximates the invariant measure $\varpi_{\Psi_{\btheta}}(\bx)\propto\frac{\pi(\bx)}{\Psi^{\zeta}_{\btheta}(U(\bx))}$ asymptotically as $\epsilon\rightarrow 0$ and $n\rightarrow N$. As shown in Lemma \ref{convex_main_csgld}, we have the mean-field equation 
\begin{equation} \label{fixedeq}
h(\btheta)=\int_{\MX}\widetilde H(\btheta,\bx)  \varpi_{\btheta}(\bx) d\bx= Z_{\btheta}^{-1} \left(\btheta_{\star}+\varepsilon \beta(\btheta)-\btheta\right)=0, 
\end{equation}
where $\btheta_{\star}=(\int_{\MX_1}\pi(\bx)d\bx, \int_{\MX_2}\pi(\bx)d\bx, \ldots, \int_{\MX_m}\pi(\bx)d\bx)$, $Z_{\btheta}$ is the normalizing constant, $\beta(\btheta)$ is a perturbation term, $\varepsilon$ is a small error depending on $\epsilon, n$ and $m$. 
The mean-field equation implies that for any 
$\zeta>0$, $\btheta_k$ converges to a small neighbourhood of $\btheta_{\star}$. By applying perturbation theory and setting the Lyapunov function  $\mathbb{V}(\btheta)=\frac{1}{2}\|\btheta_{\star}-\btheta\|^2$, we can establish the stability condition:

\begin{lemma}[Stability. Informal version of Lemma \ref{convex_appendix}] \label{convex_main_csgld}
Given a small enough $\epsilon$ (learning rate), a large enough $n$ (batch size) and $m$ (partition number), there is a constant $\phi=\inf_{\btheta} Z_{\btheta}^{-1}>0$ such that the mean-field $h(\btheta)$ satisfies $$\forall \btheta \in \bTheta, \langle h(\btheta), \btheta - \btheta_{\star}\rangle \leq  -\phi\|\btheta - \btheta_{\star}\|^2+\mathcal{O}\left(\epsilon+\frac{1}{m}+\Var(\xi_n)\right),$$
where $\Var(\xi_n)$ denotes the largest variance of the noise in the stochastic energy estimator of batch size $n$ and decays to $0$ as $n\rightarrow N$.
\end{lemma}

Together with the tool of Poisson equation  \cite{Albert90, andrieu05}, which controls the fluctuation of $\widetilde H(\btheta, \bx)-h(\btheta)$, we can 
establish convergence of $\btheta_k$ in Theorem \ref{thm:1}, whose proof is given in the supplementary 
material. 
 
\begin{theorem}[$L^2$ convergence rate. Informal version of Theorem \ref{latent_convergence}]
\label{thm:1}
Given Assumptions $\ref{ass2a}$-$\ref{ass1}$, a small enough learning rate $\epsilon_k$, a large partition number $m$ and a large batch size $n$, $\btheta_k$ converges to $\btheta_{\star}$ such that
 \begin{equation*}
    \E\left[\|\bm{\theta}_{k}-\bm{\theta}_{\star}\|^2\right]=\mathcal{O}\left( \omega_{k}+\sup_{i\geq k_0}\epsilon_i+\frac{1}{m} +\Var(\xi_n)\right),
\end{equation*}
where $k_0$ is some large enough integer, $\btheta_{\star}=(\int_{\MX_1}\pi(\bx)d\bx, \int_{\MX_2}\pi(\bx)d\bx, \ldots, \int_{\MX_m}\pi(\bx)d\bx)$, and $\delta_n(\cdot)$ is a bias term depending on the batch size $n$ and decays to 0 as $n\rightarrow N$.
\end{theorem}

\subsection{Ergodicity and Dynamic Importance Sampler}  \label{FMalg2}

CSGLD belongs to the class of adaptive MCMC algorithms,
but its transition kernel is based on SGLD instead of the Metropolis algorithm. As such, the ergodicity theory for traditional adaptive MCMC algorithms \cite{RobertsRosenthal2007, AndrieuMoulines2006, Fortetal2011, Liang10} is not directly applicable. To tackle this issue, we conduct the following theoretical study. First, rewrite (\ref{SGLDeq6}) as 
  \begin{equation} \label{SGLDeq8}
     \bx_k- \epsilon\left( \nabla_{\bx} 
     \widehat{L}(\bx_k,\btheta_{\star}) + 
     \Upsilon(\bx_k,\btheta_k,\btheta_{\star})\right)+\mathcal{N}({0, 2\epsilon \tau\bm{I}}),
 \end{equation}
where $\nabla_{\bx} \widehat{L}(\bx_k,\btheta_{\star})= \frac{N}{n} \left[1+\frac{\zeta\tau}{\Delta u}  \left(\log \theta_{\star}({J}(\bx_k))-\log\theta_{\star}(({J}(\bx_k)-1)\vee 1) \right) \right]  \nabla_{\bx} \widetilde U(\bx_k)$, the bias term follows that 
$$\Upsilon(\bx_k,\btheta_k,\btheta_{\star})=\nabla_{\bx} \widetilde{L}(\bx_k,\btheta_k)-\nabla_{\bx} \widehat{L}(\bx_k,\btheta_{\star}),$$ and $\nabla_{\bx} \widetilde{L}(\bx_k,\btheta_{k})= \frac{N}{n} \left[1+ \frac{\zeta\tau}{\Delta u}  \left(\log \theta_{k}(\tilde{J}(\bx_k))-\log\theta_{k}((\tilde{J}(\bx_k)-1)\vee 1) \right) \right]  \nabla_{\bx} \widetilde U(\bx_k)$.  The order of the bias is figured out in Lemma C1 in the supplementary material based on the results of Theorem \ref{thm:1}.
    
Next, we show how the empirical mean $\frac{1}{k}\sum_{i=1}^k f(\bx_i)$ deviates from the posterior mean $\int_{\MX}f(\bx)\varpi_{\Psi_{\btheta_{\star}}}(\bx)d\bx$. Note that this is a direct application of Theorem 2 of \cite{Chen15} by treating $\nabla_{\bx} \widehat{L}(\bx,\btheta_{\star})$ as the stochastic gradient of a target distribution and 
 $\Upsilon(\bx,\btheta,\btheta_{\star})$ as the bias of the stochastic gradient. Moreover, considering that $\varpi_{\widetilde \Psi_{\btheta_{\star}}}(\bx)\propto
\frac{\pi(\bx)}{\theta_{\star}^{\zeta}(J(\bx))}\rightarrow\varpi_{ \Psi_{\btheta_{\star}}}$ as $m\rightarrow \infty$ based on Lemma B4 in the supplementary material, we have the following

\begin{lemma}[Convergence of the Averaging Estimators. Informal version of Lemma \ref{avg_converge_appendix}]
\label{avg_converge}
Suppose Assumptions $\ref{ass2a}$-$\ref{ass1}$ hold. For any bounded function $f$, we have
\begin{equation*}
\small
\begin{split}
    \mid\E\left[\frac{\sum_{i=1}^k f(\bx_i)}{k}\right]-\int_{\bchi} f(\bx)\varpi_{\widetilde \Psi_{\btheta_{\star}}}(d\bx)\mid&= \mathcal{O}\left(\frac{1}{k\epsilon}+\sqrt{\epsilon}+\sqrt{\frac{\sum_{i=1}^k \omega_k}{k}}+\frac{1}{\sqrt{m}}+\sqrt{\Var(\xi_n)}\right), \\
\end{split}
\end{equation*}
where $\varpi_{\widetilde \Psi_{\btheta_{\star}}}(\bx)= \frac{1}{Z_{\btheta_{\star}}} 
\frac{\pi(\bx)}{\theta_{\star}^{\zeta}(J(\bx))}$ and $Z_{\btheta_{\star}}=\sum_{i=1}^m \frac{\int_{\MX_i} \pi(\bx)d\bx}{\theta_{\star}(i)^{\zeta}}$.
\end{lemma}

Finally, we consider the problem of estimating the quantity  $\int_{\MX} f(\bx) \pi(\bx) d\bx$. Recall that $\pi(\bx)$ is the target distribution that we would like to make inference for. To estimate this quantity, we naturally 
 consider the weighted averaging estimator $\frac{\sum_{i=1}^k\theta_{i}^{\zeta}( \tilde{J}(\bx_i)) f(\bx_i)}{ 
\sum_{i=1}^k\theta_{i}^{\zeta}( \tilde{J}(\bx_i))}$ by treating $\theta^{\zeta}(\tilde J(\bx_i))$ as
the dynamic importance weight of the sample $\bx_i$ 
for $i=1,2,\ldots,k$. The convergence of this
estimator is established in Theorem \ref{wavg_esti}, which can be proved by repeated applying Theorem \ref{thm:1} and Lemma \ref{avg_converge} with the details given in the supplementary material.
  
\begin{theorem}[Convergence of the Weighted Averaging Estimators. Informal version of Theorem \ref{w_avg_converge_appendix}]
\label{wavg_esti} Suppose Assumptions $\ref{ass2a}$-$\ref{ass1}$ hold. For any bounded function $f$, we have
\label{w_avg_converge}
\begin{equation*}
\footnotesize
\begin{split}
    \mid\E\left[\frac{\sum_{i=1}^k\theta_{i}^{\zeta}
    (\tilde{J}(\bx_i)) f(\bx_i)}{\sum_{i=1}^k \theta_{i}^{\zeta} ( \tilde{J}(\bx_i))}\right]-\int_{\bchi} f(\bx)\pi(d\bx)\mid &= \mathcal{O}\left(\frac{1}{k\epsilon}+\sqrt{\epsilon}+\sqrt{\frac{\sum_{i=1}^k \omega_k}{k}}+\frac{1}{\sqrt{m}}+\sqrt{\Var(\xi_n)}\right).\\
\end{split}
\end{equation*}
\end{theorem}
The bias of the weighted averaging estimator decreases 
if one applies a larger batch size, a finer sample space partition, a smaller learning rate $\epsilon$, and smaller step sizes $\{\omega_k\}_{k\geq 0}$. Admittedly, the
order of this bias is slightly larger than   $\mathcal{O}\left(\frac{1}{k\epsilon}+\epsilon\right)$
 achieved by the standard SGLD. We note that this is necessary as simulating from the flattened distribution $\varpi_{\Psi_{\btheta_{\star}}}$ often leads to a much faster convergence, see e.g. the green curve v.s. the purple curve in Fig.\ref{CSGLD_stats}(c).

\section{Numerical Studies}

\subsection{Simulations of Multi-modal Distributions}

\subsubsection{A Gaussian Mixture Distribution} The first numerical study is to test the performance of CSGLD on a Gaussian mixture distribution $\pi(\bx)=0.4 N(-6,1)+0.6 N(4,1)$. In each experiment, the algorithm was  run for $10^7$ iterations. We fix the temperature $\tau=1$ and the learning rate $\epsilon=0.1$. The step size for stochastic approximation follows $\omega_k=\frac{1}{k^{0.6}+100}$. The sample space is partitioned into 50 subregions with $\Delta u=1$. The stochastic gradients are simulated by injecting additional random noises following $N(0,0.01)$ to the exact gradients. For comparison, SGLD is chosen as the baseline algorithm 
 and implemented with the same setup as CSGLD. We repeat the experiments 10 times and report the average and the associated standard deviation. 
 
We first assume that $\btheta_{\star}$ is known and plot the energy functions for both $\pi(\bx)$ and $\varpi_{\Psi_{\btheta_{\star}}}$ with different values of $\zeta$.  Fig.\ref{CSGLD_stats}(a) shows that the original energy function has a rather large energy barrier which strongly affects the communication between two modes of the distribution. In contrast, CSGLD samples from a modified energy function, which yields a flattened landscape and reduced energy barriers. For example, with $\zeta=0.75$, the energy barrier for this example is {\it greatly reduced from 12 to as small as 2}. Consequently, the local trap problem can be greatly alleviated. Regarding the bizarre 
peaks around $x=4$, we leave the study in the supplementary material.

\begin{figure*}[!ht]
\vspace{-0.03in}
  \centering
  \subfloat[Original v.s. trial energies]{\includegraphics[scale=0.28]{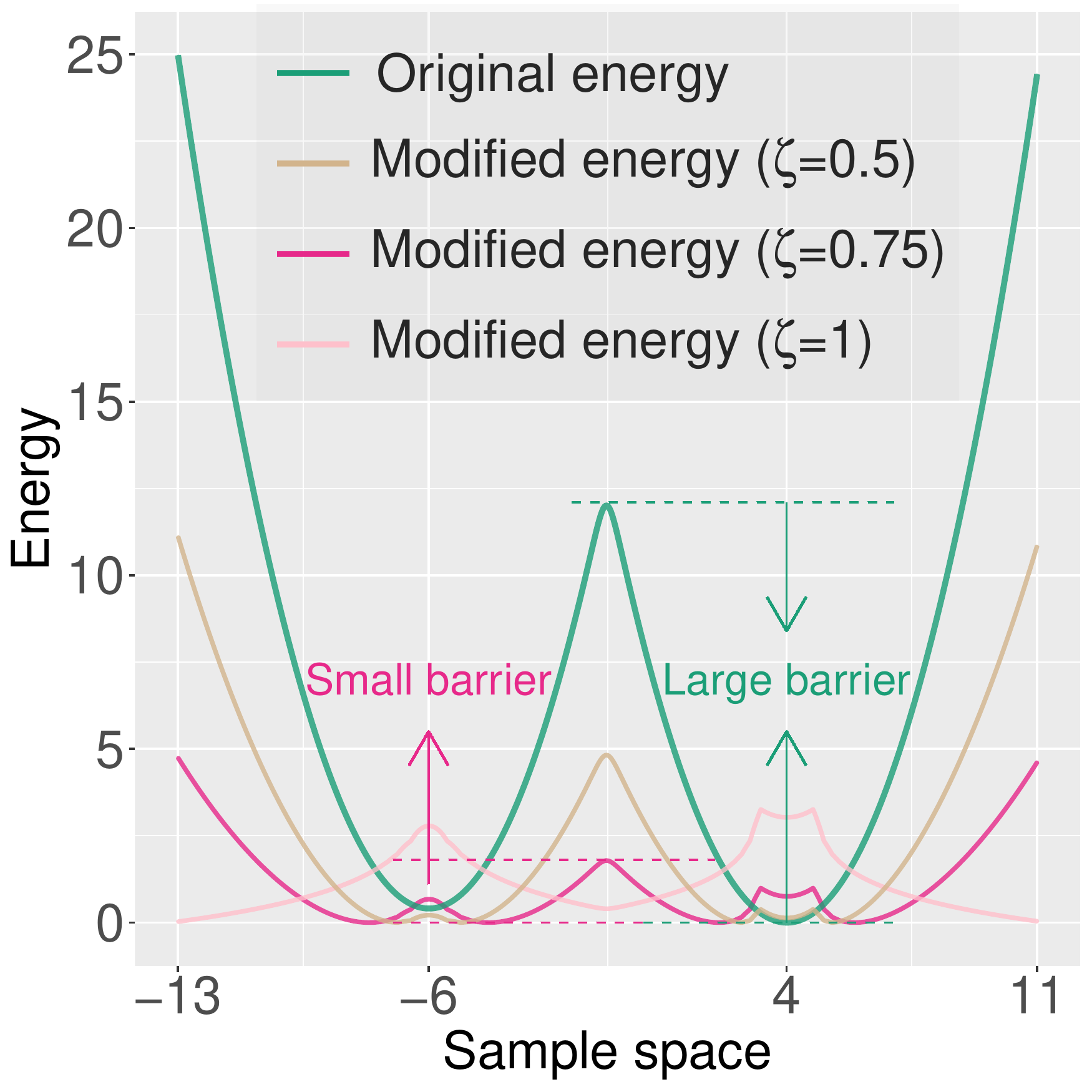}}\label{fig: 4a}\quad
  \subfloat[$\btheta$'s estimates and histograms]{\includegraphics[scale=0.28]{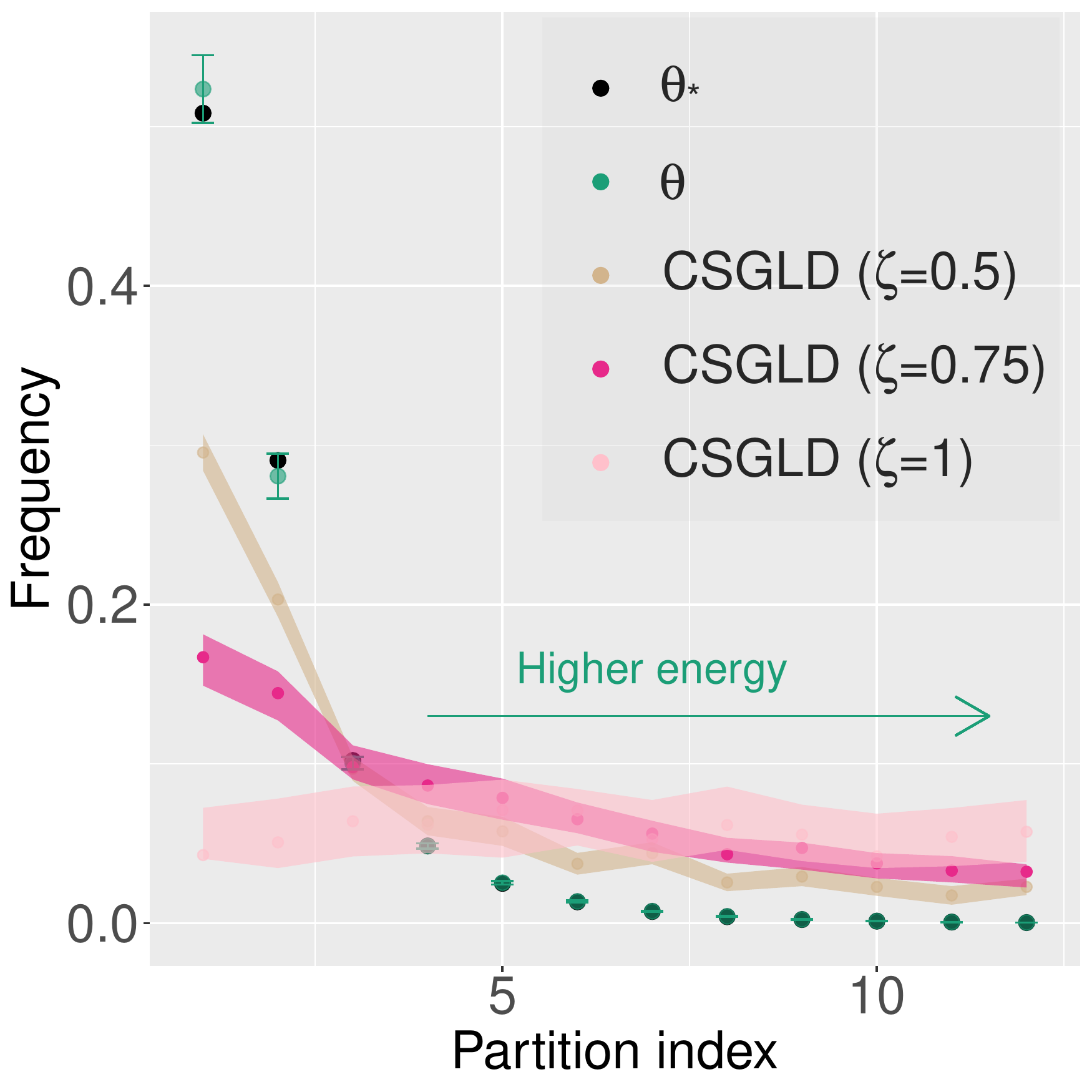}}\label{fig: 4b}\quad
  \subfloat[Estimation errors]{\includegraphics[scale=0.28]{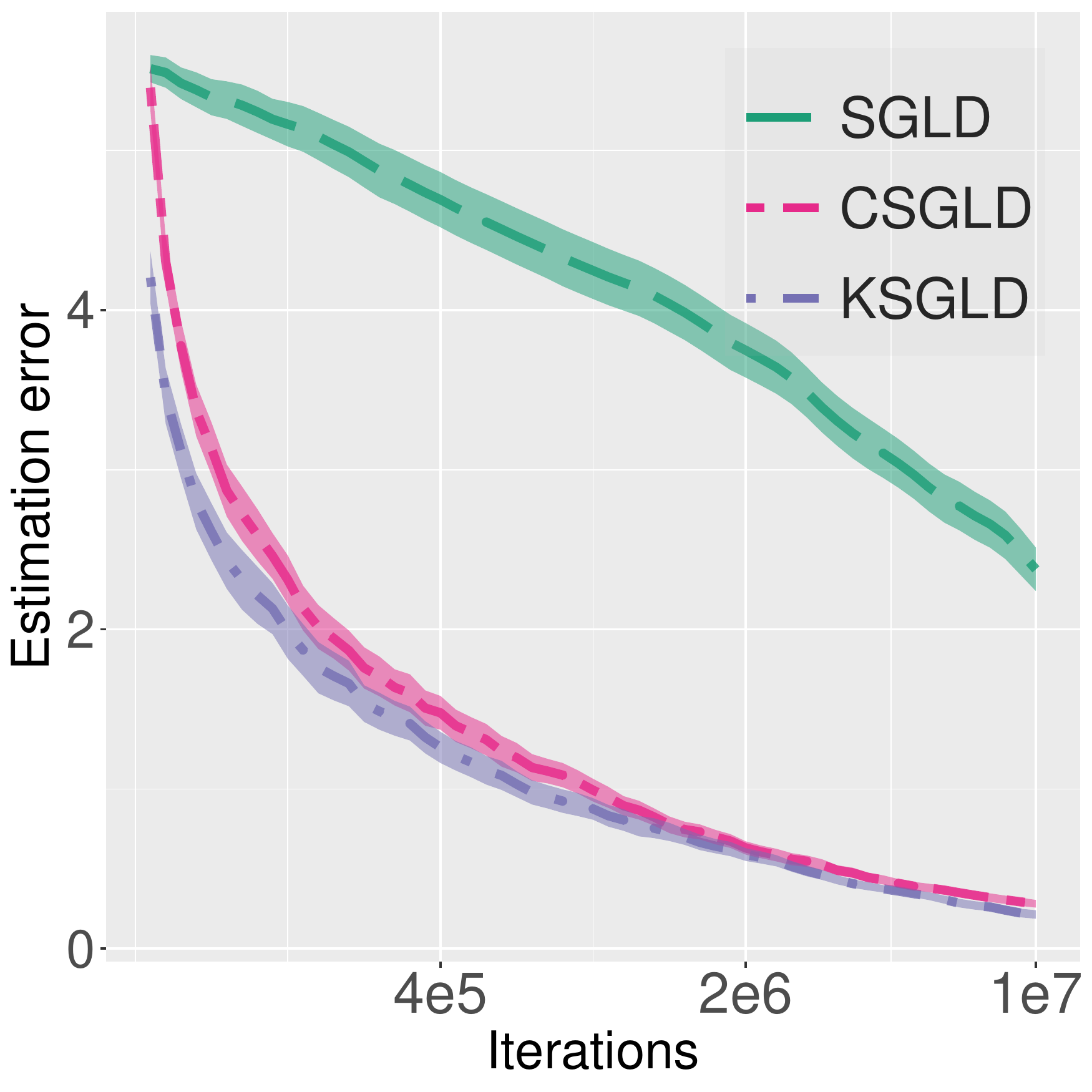}}\label{fig: 4c}
  \caption{Comparison between SGLD and CSGLD: Fig.1(b) presents only the first 12 partitions for an illustrative purpose; KSGLD in  Fig.1(c) is implemented by assuming $\btheta_{\star}$ is known.}
  \label{CSGLD_stats}
  \vspace{-0.17in}
\end{figure*}

Fig. \ref{CSGLD_stats}(b) summarizes the estimates of $\btheta_{\star}$ with $\zeta=0.75$, which matches the ground truth value of $\btheta_{\star}$ very well. Notably, we see that $\theta_{\star}(i)$ decays exponentially fast as the partition index $i$ increases, which indicates the exponentially decreasing probability of visiting high energy regions and a severe local trap problem. CSGLD tackles this issue by adaptively updating the transition kernel or, equivalently, the invariant distribution such 
that the sampler moves like a ``random walk'' in the space of energy.  In particular, setting $\zeta=1$ leads to a flat histogram of energy (for the samples produced by CSGLD).

To explore the performance of CSGLD in quantity estimation with the weighed averaging estimator, 
we compare CSGLD ($\zeta=0.75$) with SGLD and KSGLD in estimating \textcolor{black}{the posterior mean $\int_{\MX} \bx \pi(\bx)d\bx$}, where KSGLD was implemented by assuming $\btheta_{\star}$ is known and sampling from $\varpi_{\Psi_{\btheta_{\star}}}$ directly. Each algorithm was run for 10 times, and we recorded the mean absolute 
estimation error along with iterations. As shown in Fig.\ref{CSGLD_stats}(c), the estimation error of SGLD decays quite slow and rarely converges due to the high energy barrier. On the contrary, KSGLD converges much faster, which shows the advantage of sampling from a flattened distribution $\varpi_{\Psi_{\btheta_{\star}}}$. Admittedly, $\btheta_{\star}$ is unknown in practice. CSGLD instead adaptively updates its invariant distribution while optimizing the parameter ${\btheta}$ until \emph{an optimization-sampling equilibrium} is reached. In the early period of the run, CSGLD converges slightly slower than KSGLD, but soon it becomes as efficient as KSGLD.

Finally, we compare the sample path and learning rate for CSGLD and SGLD. As shown in Fig.\ref{trajectory}(a), SGLD tends to be trapped in a deep local optimum for an exponentially long time. CSGLD, in contrast, possesses a {\it self-adjusting mechanism} for escaping from local traps. In the early period of a run, CSGLD might suffer 
from a similar local-trap problem as SGLD (see Fig.\ref{trajectory}(b)). In this case, the components of $\btheta$ corresponding to the current subregion will  increase very fast, eventually rendering \textcolor{black}{a smaller or even negative stochastic gradient multiplier} which \emph{bounces the sampler back to high energy regions}. To illustrate the process, we plot a bouncy zone and an absorbing zone in Fig.\ref{trajectory}(c). The bouncy zone enables the sampler to ``jump'' over large energy barriers to explore other modes. As the run continues, $\btheta_k$ converges to $\btheta_{\star}$. Fig.\ref{trajectory}(d) shows that larger bouncy ``jumps'' (in red lines) can potentially be induced in the bouncy zone, which occurs in both local and global optima. Due to the {\it self-adjusting mechanism}, CSGLD has the local trap problem much alleviated.

\begin{figure*}[!ht]
  \centering
  \subfloat[SGLD paths]{\includegraphics[scale=0.21]{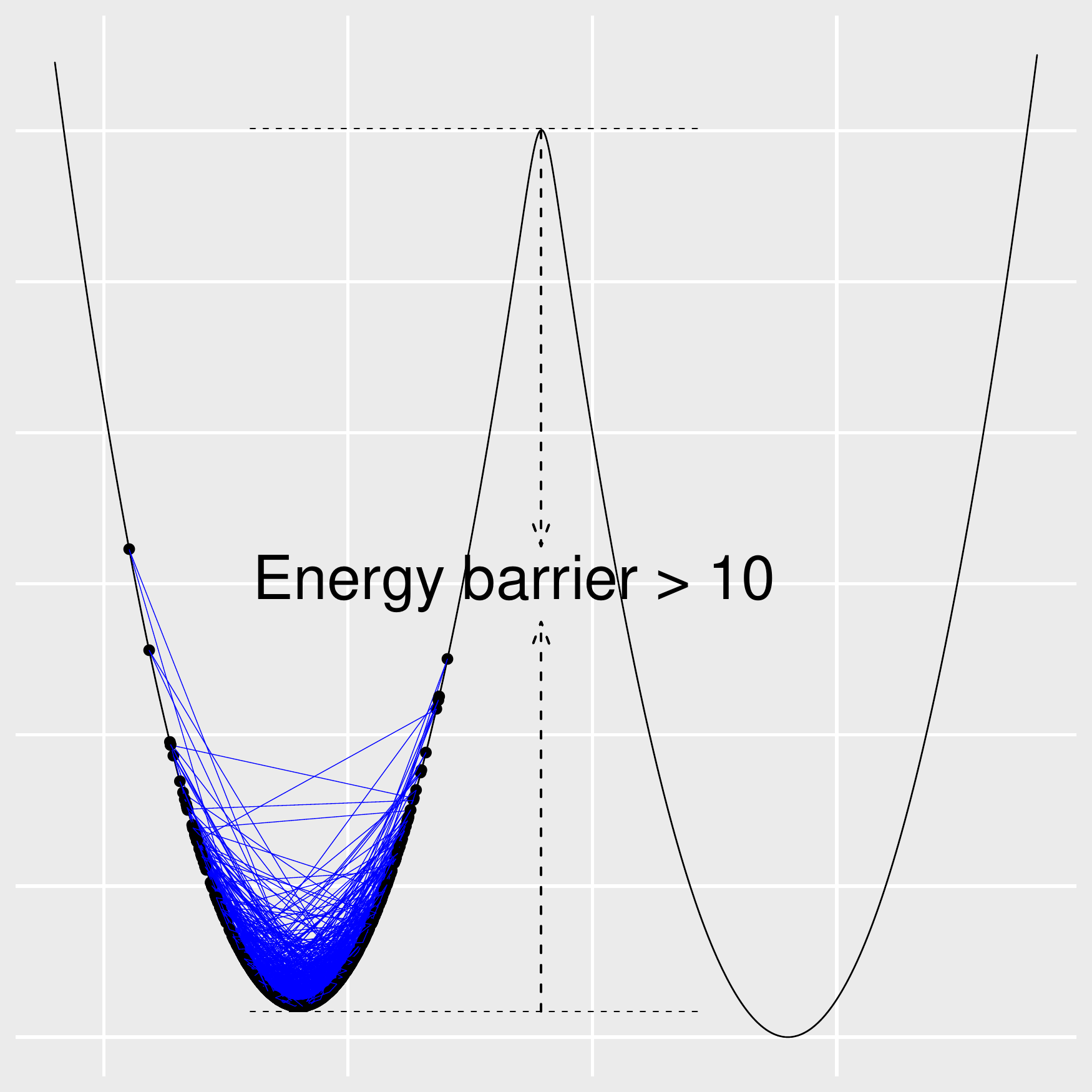}}\quad
  \subfloat[CSGLD paths (early) ]{\includegraphics[scale=0.21]{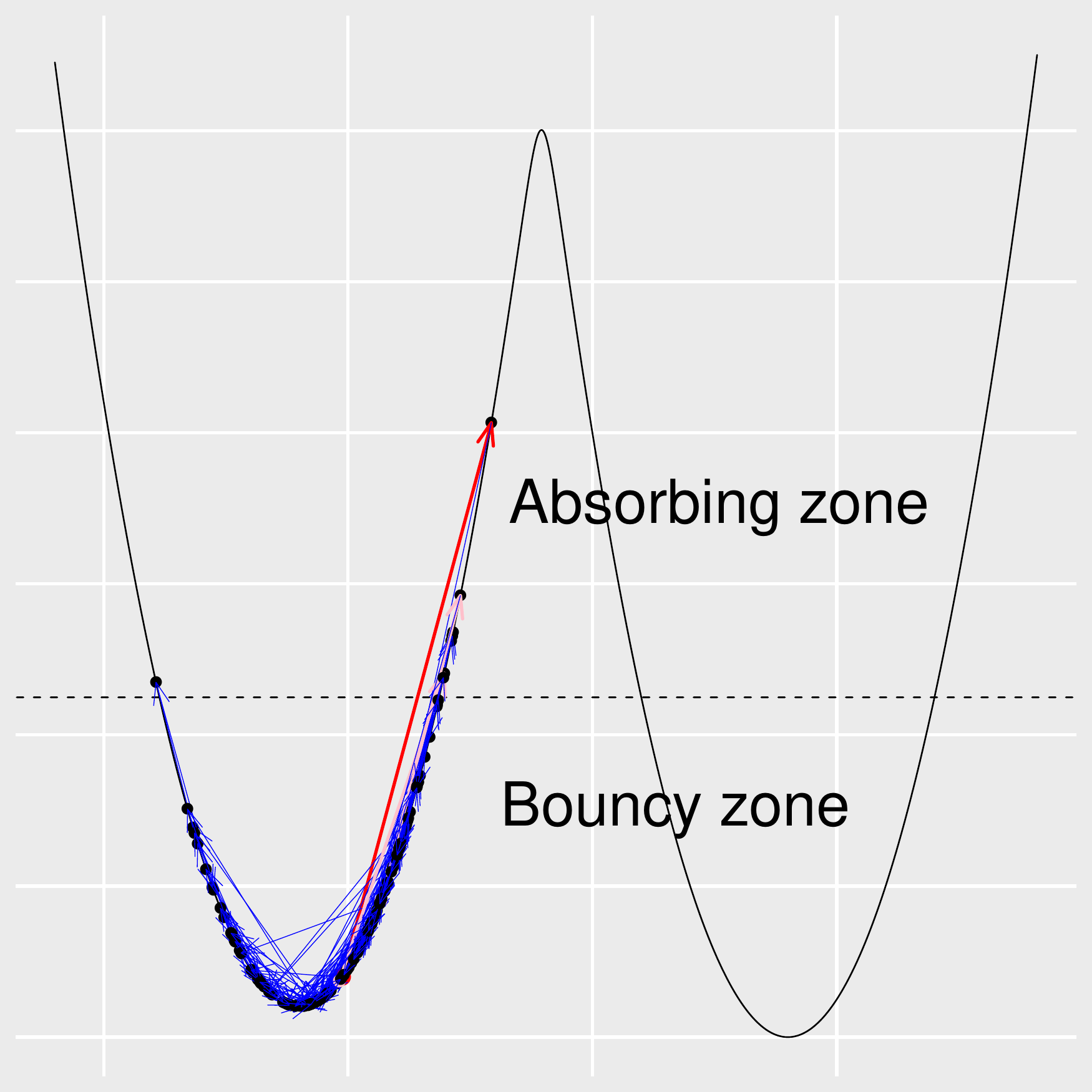}}\quad
  \subfloat[CSGLD paths (mid)]{\includegraphics[scale=0.21]{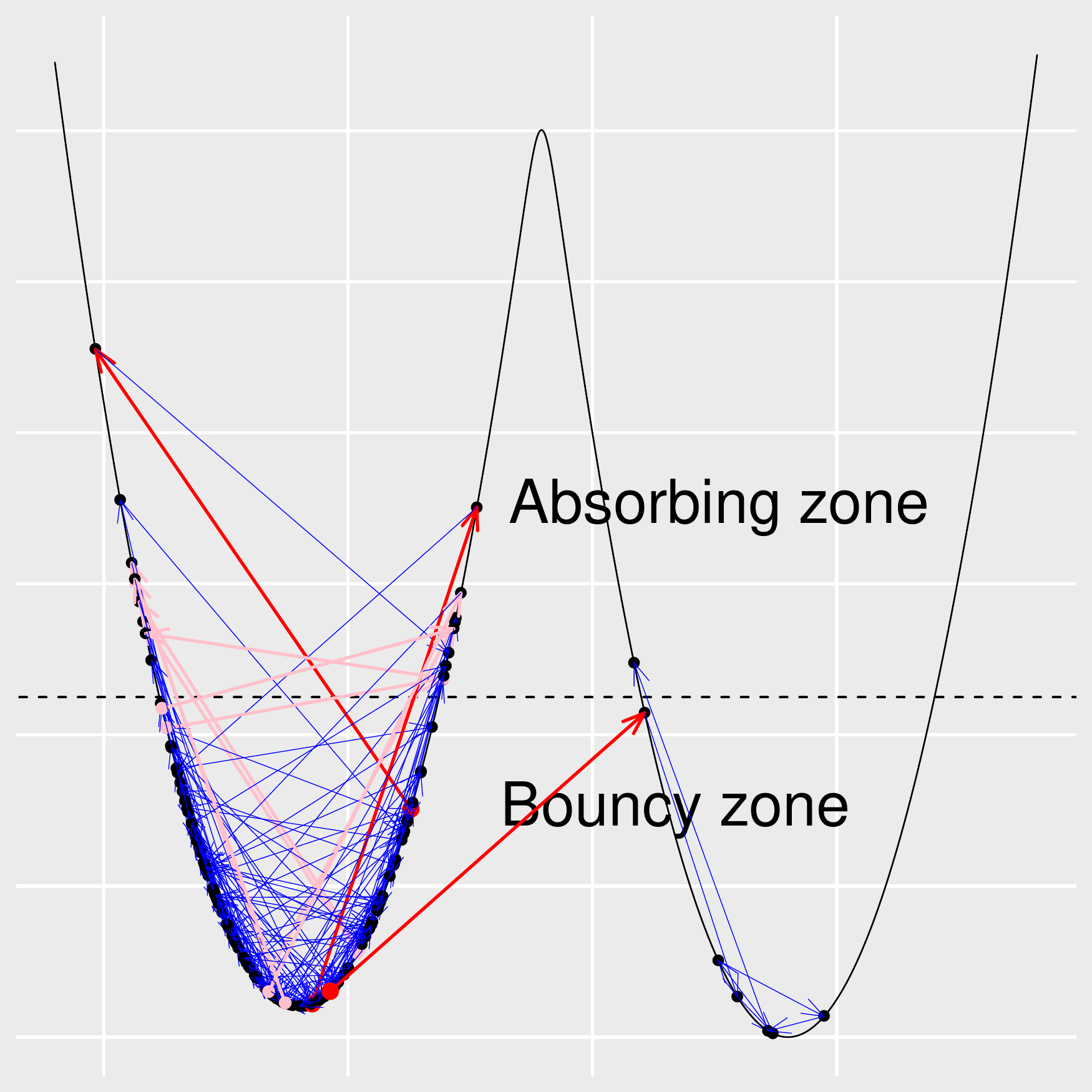}}\quad
  \subfloat[CSGLD paths (late)]{\includegraphics[scale=0.21]{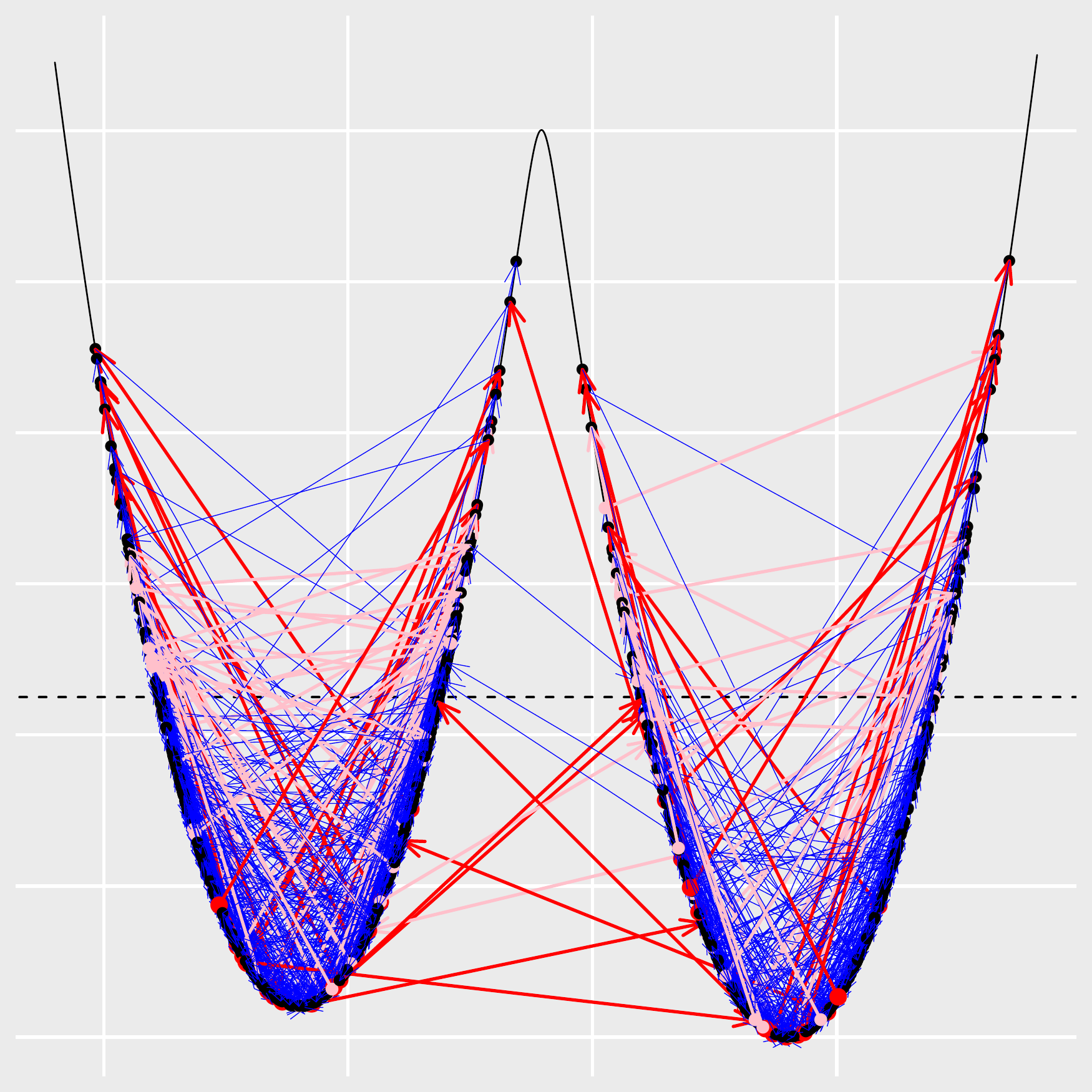}}
  \caption{Sample trajectories of SGLD and CSGLD: plots (a) and (c) are implemented by 100,000 iterations with a thinning factor 100 and $\zeta=0.75$, while plot (b) utilizes a thinning factor 10.}
  \label{trajectory}
\end{figure*}

\subsubsection{A Synthetic Multi-modal Distribution} We next simulate from a distribution $\pi(\bx)\propto e^{-U(\bx)}$, where $U(\bx)=\sum_{i=1}^2 \frac{x(i)^2-10\cos(1.2\pi x(i))}{3}$ and $\bx=(x(1), x(2))$. We compare CSGLD with SGLD, replica exchange SGLD (reSGLD) \cite{deng2020}, and SGLD with cyclic learning rates (cycSGLD) \cite{ruqi2020} and detail the setups in the supplementary material. Fig.\ref{multi-modal_simulations}(a) shows that the distribution contains nine important modes, where the center mode has the largest probability mass and the four modes on the corners have the smallest mass. We see in Fig.\ref{multi-modal_simulations}(b) that SGLD spends too much time in local regions and only identifies three modes. cycSGLD has a better ability to explore the distribution by leveraging large learning rates cyclically. However, as illustrated in Fig.\ref{multi-modal_simulations}(c), such a mechanism is still not efficient enough to resolve the local trap issue for this problem. reSGLD proposes to include a high-temperature process to encourage exploration and allows interactions between the two processes via appropriate swaps. We observe in Fig.\ref{multi-modal_simulations}(d) that reSGLD obtains both the exploration and exploitation abilities and yields a much better result. However, the noisy energy estimator may hinder the swapping efficiency and it becomes difficult to estimate a few modes on the corners. As to our algorithm, CSGLD first simulates the importance samples and recovers the original distribution according to the importance weights. We notice that the samples from CSGLD can traverse freely in the parameter space and eventually achieve a remarkable performance, as shown in Fig.\ref{multi-modal_simulations}(e).

\begin{figure*}[!ht]
  \centering
  \subfloat[Ground truth]{\includegraphics[scale=0.28]{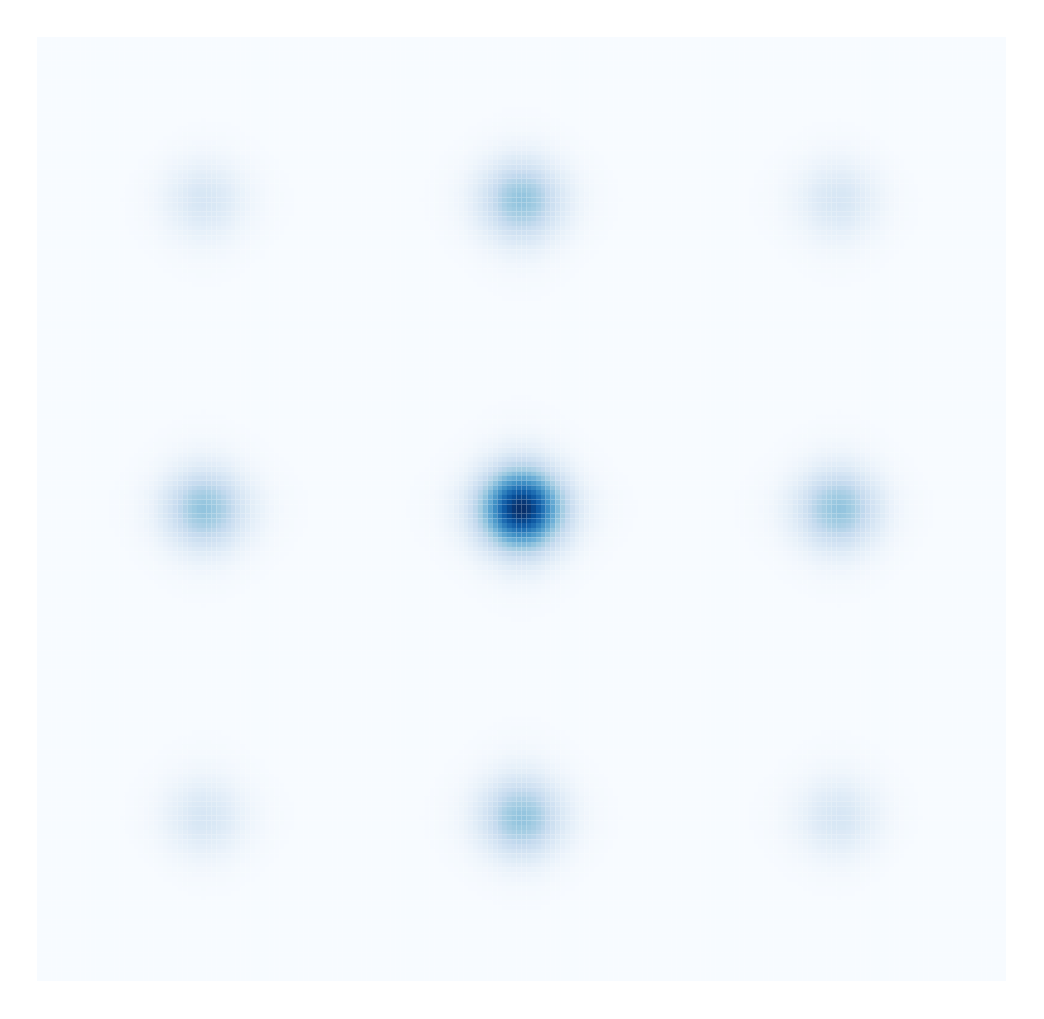}}\label{fig: msa}\;\;
  \subfloat[SGLD]{\includegraphics[scale=0.28]{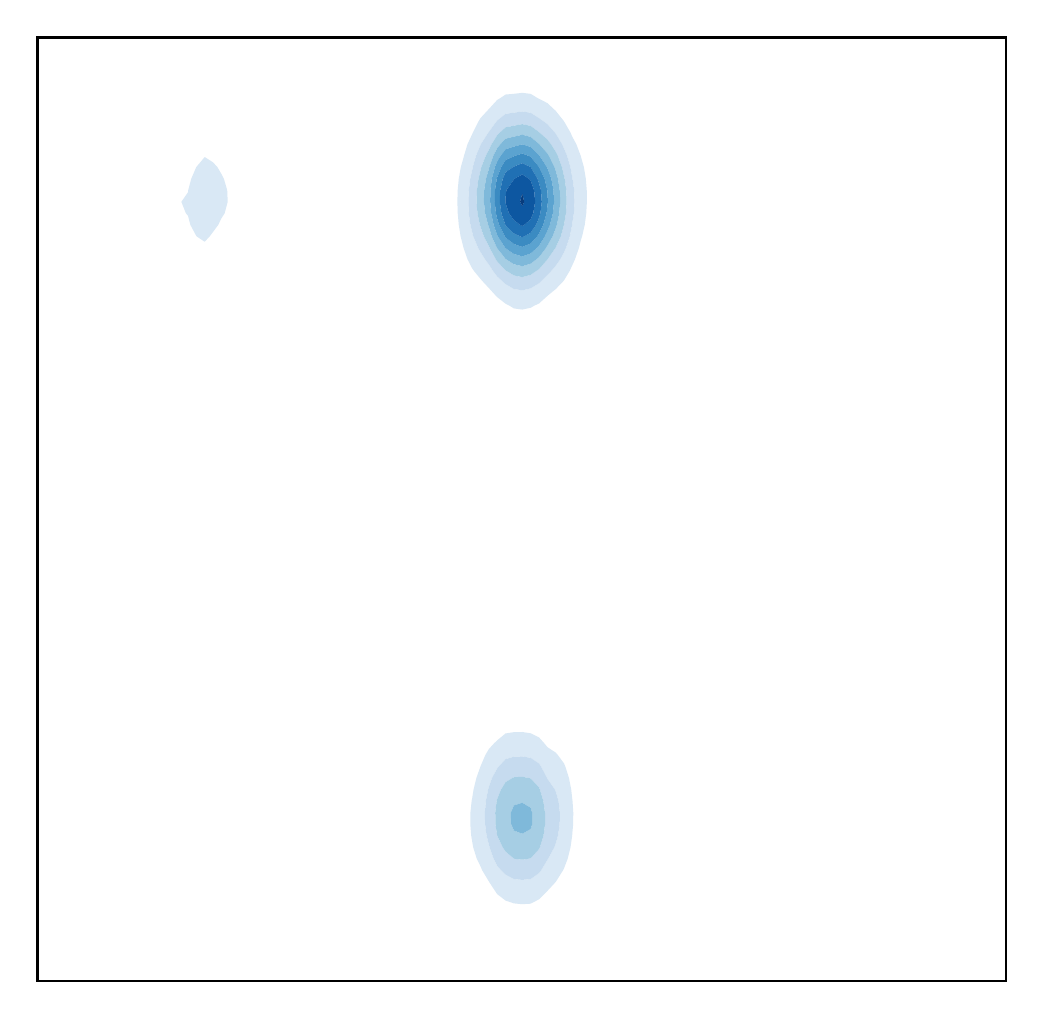}}\label{fig: msb}\;\;
  \subfloat[cycSGLD]{\includegraphics[scale=0.28]{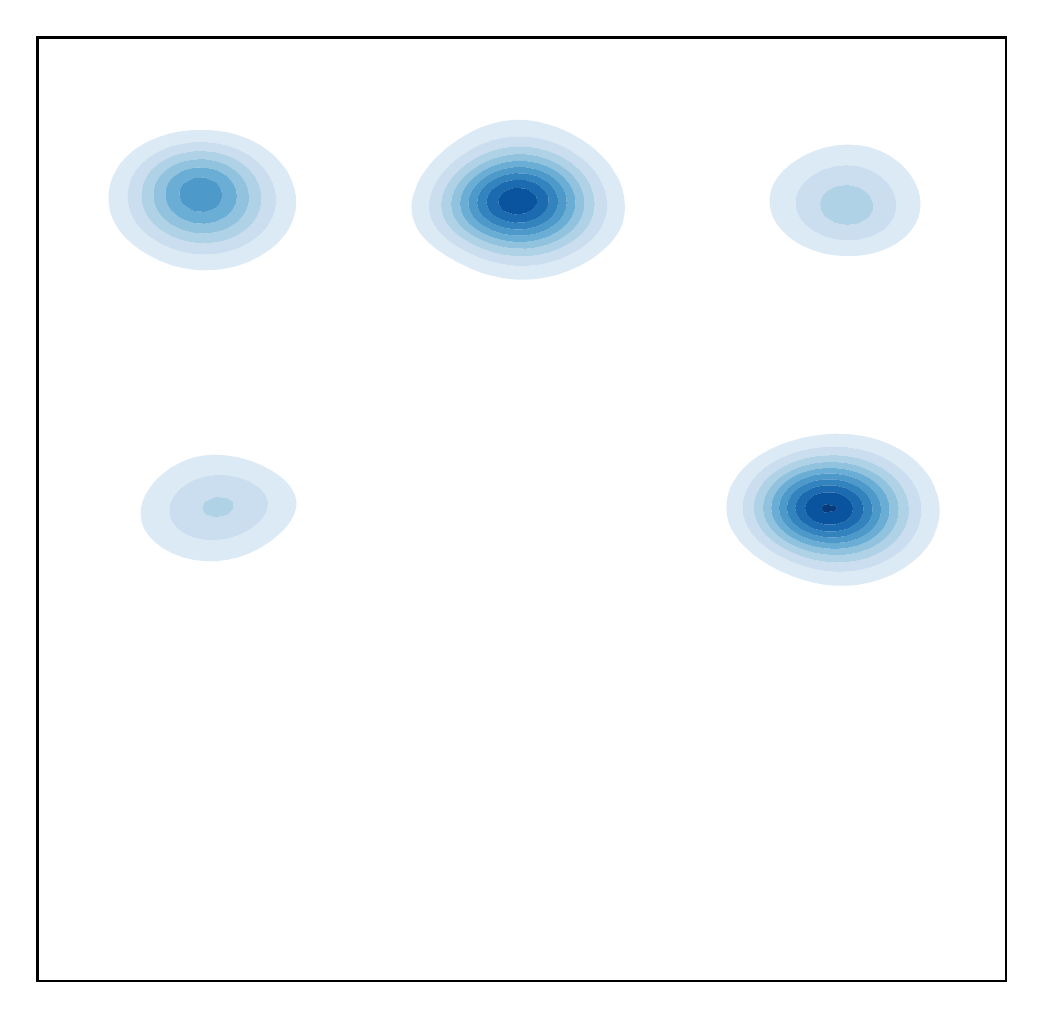}}\label{fig: msc}\;\;
  \subfloat[reSGLD]{\includegraphics[scale=0.28]{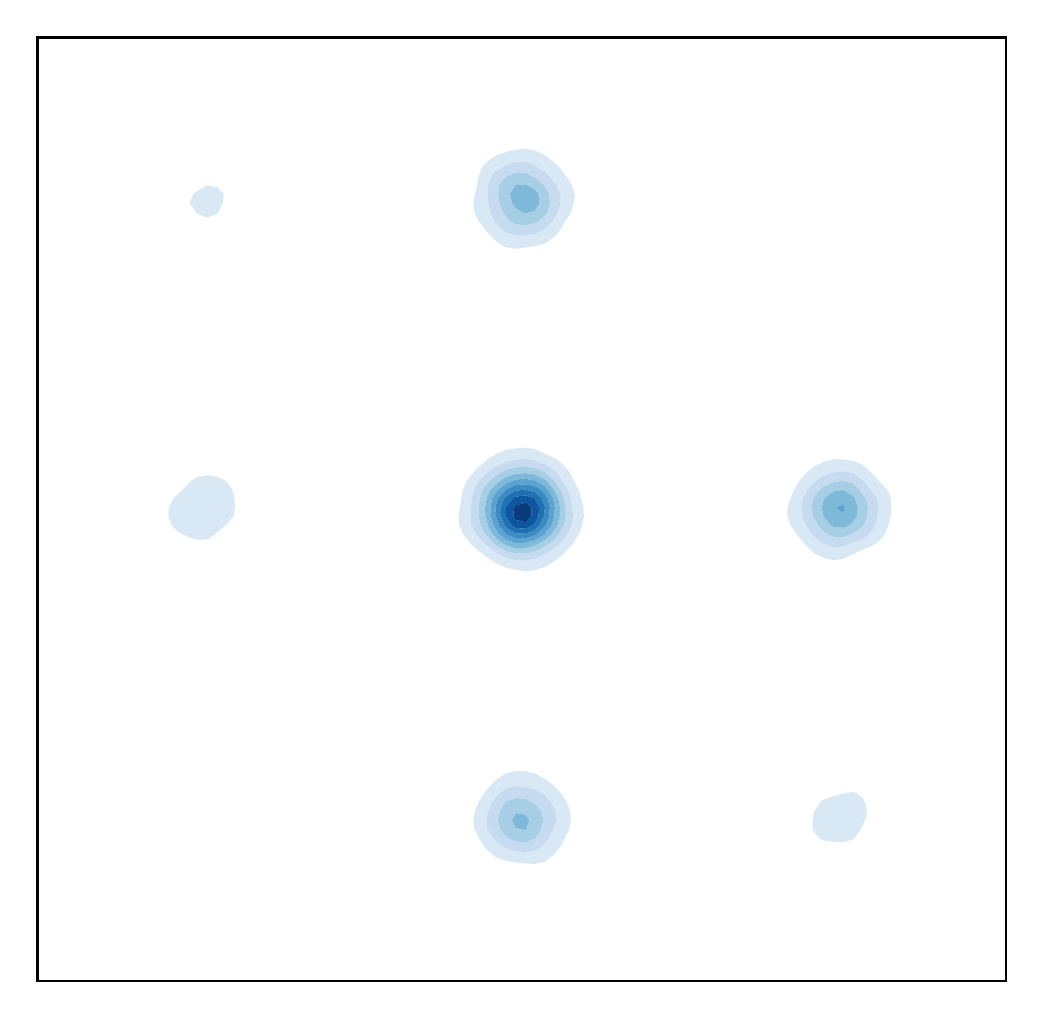}}\label{fig: msd}\;\;
  \subfloat[CSGLD]{\includegraphics[scale=0.28]{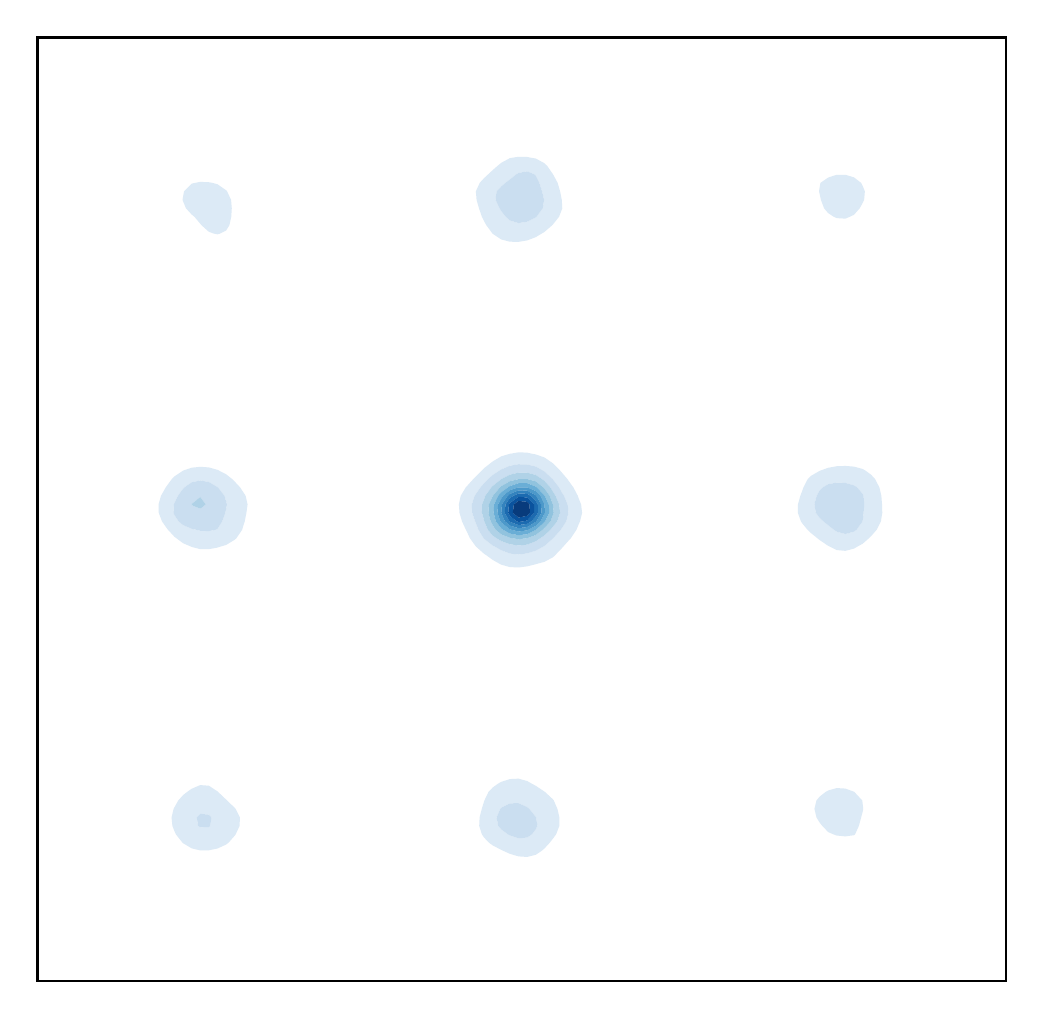}}\label{fig: mse}
  \caption{Simulations of a multi-modal distribution. A resampling scheme is used for CSGLD.}
  \label{multi-modal_simulations}
\end{figure*}

\subsection{UCI Data}
\label{reg_UCI}

We tested the performance of CSGLD on the \textbf{UCI} regression datasets. For each dataset, we normalized all features and randomly selected 10\% of the observations for testing. Following \cite{Jose_adam_15}, we modeled the data using a Multi-Layer Perception (MLP) with a single hidden layer of 50 hidden units. We set the mini-batch size $n=50$ and trained the model for 5,000 epochs. The learning rate was set to 5e-6 and the default $L_2$-regularization coefficient is 1e-4. For all the datasets, we used the stochastic energy $\frac{N}{n}\widetilde U(\bx)$ to evaluate the partition index. We set the energy bandwidth $\Delta u=100$. We fine-tuned the temperature $\tau$ and the hyperparameter $\zeta$. For a fair comparison, each algorithm was run 10 times with fixed seeds for each dataset. 
In each run, the performance of the algorithm was evaluated by averaging over 50 models, where the averaging estimator was used for SGD and SGLD and the weighted averaging estimator was used for CSGLD. As shown in Table~\ref{UCI}, SGLD outperforms the stochastic gradient descent (SGD) 
algorithm for most datasets due to the advantage of a sampling algorithm in obtaining more informative modes. Since all these datasets are small, there is only very limited potential for improvement. Nevertheless, CSGLD still consistently outperforms all the baselines including SGD and SGLD. 

The contour strategy proposed in the chapter can be naturally extended to SGHMC \cite{Chen14, yian2015} without affecting the theoretical results. In what follows, we adopted a numerical method proposed by \cite{Saatci17} to avoid extra hyperparameter tuning. We set the momentum term to 0.9 and simply inherited all the other parameter settings used in the above experiments. In such a case, we compare the contour SGHMC (CSGHMC) with the baselines, including M-SGD (Momentum SGD) and SGHMC. The comparison indicates that some improvements can be achieved by including the momentum.

\begin{table}[!htb]
  \centering
\caption{Algorithm evaluation using average root-mean-square error and its standard deviation.}
\vspace{0.2in}
  \begin{tabular}{c|ccccccc}
    \toprule
    Dataset &   Energy &  Concrete &   Yacht  & Wine  \\
    Hyperparameters ($\tau/\zeta$) &   1/1  & 5/1  & 1/2.5 & 5/10 \\
    \midrule
     SGD  & 1.13$\pm$0.07  & 4.60$\pm$0.14 & 0.81$\pm$0.08 &  0.65$\pm$0.01 \\
    SGLD  & 1.08$\pm$0.07  & 4.12$\pm$0.10 & 0.72$\pm$0.07 & 0.63$\pm$0.01 \\
     CSGLD  & \textbf{1.02$\pm$0.06} & \textbf{3.98$\pm$0.11} & \textbf{0.69$\pm$0.06} &  \textbf{0.62$\pm$0.01} \\
      \midrule
      M-SGD & 0.95$\pm$0.07 & 4.32$\pm$0.27 & 0.73$\pm$0.08 &  0.71$\pm$0.02 \\
      SGHMC & 0.77$\pm$0.06 & 4.25$\pm$0.19 & \textbf{0.66$\pm$0.07} &  0.67$\pm$0.02 \\
     CSGHMC & \textbf{0.76$\pm$0.06} & \textbf{4.15$\pm$0.20} & 0.72$\pm$0.09 & \textbf{0.65$\pm$0.01} \\
    \bottomrule
  \end{tabular}
  \label{UCI}
\end{table}

\subsection{Computer Vision Data} 
\vspace{-0.05in}
This section compares only CSGHMC with M-SGD and SGHMC due to the popularity of momentum in accelerating computation for computer vision datasets. We keep partitioning the sample space according to the stochastic energy $\frac{N}{n} \widetilde U(\bx)$, where a mini-batch data of size $n$ is randomly chosen from the full dataset of size $N$ at each iteration. Notably, such a strategy significantly accelerates the computation of CSGHMC. As a result, CSGHMC has almost the same computational cost as SGHMC and SGD. To reduce the bias associated with the stochastic energy, we choose a large batch size $n=1,000$. For more discussions on the hyperparameter settings, we refer readers to section D in the supplementary material.

{\bf CIFAR10} is a standard computer vision dataset with 10 classes and 60,000 images, for which 50,000 images were used for training and the rest for testing. We modeled the data using a Resnet of 20 layers (Resnet20) \cite{kaiming15}. In particular, for CSGHMC, we considered a partition of the energy space in 200 subregions, where the energy bandwidth was set to $\Delta {u}=1000$.  We trained the model for a total of 1000 epochs and evaluated the model every ten epochs based on two criteria, namely, best point estimate (BPE) and Bayesian model average (BMA). We repeated each experiment 10 times and reported in Table \ref{full_cifar} the average prediction accuracy and the corresponding standard deviation. 
  
In the first set of experiments, all the algorithms utilized a fixed learning rate $\epsilon=2e-7$ and a fixed temperature $\tau=0.01$ under the Bayesian setting. SGHMC performs quite similarly to M-SGD, both obtaining around 90\% accuracy in BPE and 92\% in BMA. Notably, in this case, simulated annealing is not applied to any of the algorithms and achieving the state-of-the-art is quite difficult. However, BMA still consistently outperforms  BPE, implying the great potential of advanced MCMC techniques in deep learning. Instead of simulating from $\pi(\bx)$ directly, CSGHMC adaptively simulates from a flattened distribution $\varpi_{\btheta_{\star}}$ and adjusts the sampling bias by dynamic importance weights. As a result, the weighted averaging estimators obtain an improvement by as large as 0.8\% on BMA. In addition, the flattened distribution facilitates optimization and the increase in BPE is quite significant. 

In the second set of experiments, we employed a decaying schedule on both learning rates and temperatures (if applicable) to obtain simulated annealing effects. For the learning rate, we fix it at $2\times 10^{-6}$ in the first 400 epochs and then decayed it by a factor of $1.01$ at each epoch. For the temperature, we consistently decayed it by a factor of $1.01$ at each epoch. We call the resulting algorithms by saM-SGD, saSGHMC, and saCSGHMC, respectively. Table \ref{full_cifar} shows that the performances of all 
algorithms are increased quite significantly, where the fine-tuned baselines already obtained the state-of-the-art results. Nevertheless, saCSGHMC further improves BPE by 0.25\% and slightly improve the highly optimized BMA by nearly 0.1\%.

{\bf CIFAR100} dataset has 100 classes, each of which contains 500 training images and 100 testing images. We follow a similar setup as CIFAR10, except that $\Delta u$ is set to 5000. For M-SGD, BMA can be better than BPE by as large as 5.6\%. CSGHMC has led to an improvement of 3.5\% on BPE and 2\% on BMA, which further demonstrates the superiority of advanced MCMC techniques. Table \ref{full_cifar} also shows that with the help of both simulated annealing and importance sampling, saCSGHMC can outperform the highly optimized baselines by almost 1\% accuracy on BPE and 0.7\% on BMA. The significant improvements show the advantage of the proposed method in training DNNs.

\begin{table}[htbp]
\begin{center}
\caption{Experiments on CIFAR10 \& 100 using Resnet20, where BPE and BMA are short for best point estimate and Bayesian model average, respectively.
}\label{full_cifar}
\vspace{0.3in}
\begin{tabular}{ccccc}
\hline
\multirow{2}{4em}{Algorithms} & \multicolumn{2}{c}{CIFAR10} & \multicolumn{2}{c}{CIFAR100} \\
 & BPE & BMA & BPE & BMA \\ 
\hline
 M-SGD & 90.02$\pm$0.06  & 92.03$\pm$0.08  & 61.41$\pm$0.15  & 67.04$\pm$0.12  \\
SGHMC  & 90.01$\pm$0.07 & 91.98$\pm$0.05 & 61.46$\pm$0.14  & 66.43$\pm$0.11  \\
{CSGHMC}  & \textbf{90.87}$\pm$\textbf{0.04} & \textbf{92.85}$\pm$\textbf{0.05} & \textbf{63.97$\pm$0.21} & \textbf{68.94$\pm$0.23} \\
\hline
 saM-SGD & 93.83$\pm$0.07  & 94.25$\pm$0.04 & 69.18$\pm$0.13 & 71.83$\pm$0.12 \\
saSGHMC  & 93.80$\pm$0.06 & 94.24$\pm$0.06 & 69.24$\pm$0.11 & 71.98$\pm$0.10 \\
{saCSGHMC}  & \textbf{94.06$\pm$0.07} & 94.33$\pm$0.07 & \textbf{70.18$\pm$0.15} & \textbf{72.67$\pm$0.15} \\
\hline
\end{tabular}
\end{center}
\end{table}

\section{Conclusion}
We have proposed CSGLD as a general scalable Monte Carlo algorithm for both simulation and optimization tasks. CSGLD automatically adjusts the invariant distribution during simulations to facilitate escaping from local traps and traversing over the entire energy landscape. The sampling bias introduced thereby is accounted for by dynamic importance weights. We  proved a stability condition for the mean-field system induced by CSGLD together with the convergence of its self-adapting parameter $\btheta$ to a unique fixed point $\btheta_{\star}$. We established the convergence of a weighted averaging estimator for CSGLD. The bias of the estimator decreases as we employ a finer partition, a larger mini-batch size, and smaller learning rates and step sizes. We tested CSGLD and its variants on a few examples, which show their great potential in deep learning and big data computing. 

Our algorithm ensures AI safety by providing more robust predictions and helps build a safer environment. It is an extension of the flat histogram algorithms from the Metropolis kernel to the Langevin kernel and paves the way for future research in various dynamic importance samplers and adaptive biasing force (ABF) techniques for big data problems. The Bayesian community and the researchers in the area of Monte Carlo methods will enjoy the benefit of our work. 

\chapter{INTERACTING CONTOUR STOCHASTIC GRADIENT LANGEVIN DYNAMICS}
\label{ICSGLD}

To simulate from distributions with complex energy landscapes, e.g., those with a multitude of modes well  separated by high energy barriers, an emerging trend is to run multiple chains, where interactions between different chains can potentially accelerate the convergence of the simulation. For example, \cite{SongWL2014} and \cite{Futoshi2020} showed theoretical advantages of appropriate interactions in ensemble/population simulations. Other multiple chain methods include particle-based nonlinear Markov (Vlasov) processes \cite{SVGD, SPOS} and replica exchange methods (also known as parallel tempering) \cite{deng_VR}. However, the particle-based methods result in an expensive kernel matrix computation given a large number of particles \cite{SVGD}; similarly, na\"{i}vely extending replica exchange methods to population chains leads to a long waiting time to swap between non-neighboring chains \cite{Syed_jrssb}. Therefore, how to conduct interactions between different chains, while maintaining the scalability of the algorithm, is the key to the success of the parallel stochastic gradient MCMC algorithms. 

In this chapter, we propose an interacting contour stochastic gradient Langevin dynamics (ICSGLD) sampler, a pleasingly parallel extension of contour stochastic gradient Langevin dynamics (CSGLD) \cite{CSGLD} with \emph{efficient interactions}. 
The proposed algorithm requires minimal communication cost in that each chain shares with others the marginal energy likelihood estimate only. 
As a result, the interacting mechanism improves the convergence of the simulation, while the minimal communication mode between different chains enables the  proposed algorithm to be naturally adapted to distributed computing with little overhead. 
For the single-chain CSGLD algorithm,  despite its theoretical advantages as shown in \cite{CSGLD}, estimation of the marginal energy likelihood remains a challenging task for the big data problems with a wide energy range, jeopardizing the empirical performance of the class of importance sampling methods \cite{SMC1, SMC2, WangLandau2001, Liang07, Particle_MCMC, CSGLD} in big data applications. To resolve this issue, we resort to a novel random-field function that is constructed based on multiple chains and  proven easier to estimate. As such, we can greatly facilitate the estimation of the marginal energy likelihood so as to accelerate the simulations of notoriously complex distributions. To summarize, the algorithm has three main contributions:

\begin{itemize}
    \item We propose a scalable interacting importance sampling method for big data problems with the minimal communication cost. A novel random-field function is derived to tackle the incompatibility issue of the class of importance sampling methods in big data problems.
    \item Theoretically, we study the local stability of the non-linear mean-field system. In addition, we prove the asymptotic normality for the stochastic approximation process in mini-batch settings and show that ICSGLD is asymptotically more efficient than the single-chain CSGLD with an equivalent computational budget. 
    \item Our proposed algorithm achieves appealing mode explorations using a fixed learning rate on the MNIST dataset and obtains remarkable performance in large-scale uncertainty estimation tasks.
\end{itemize}

\section{Preliminaries}
\vskip -0.05in
\subsection{Stochastic Gradient Langevin Dynamics}

A standard sampling algorithm for big data problems is SGLD \cite{Welling11}, which is a numerical scheme of a stochastic differential equation in mini-batch settings:
\begin{equation}
    \bx_{k+1}=\bx_k-\epsilon_{k} \frac{N}{n}\nabla_{\bx} \widetilde U(\bx_k)+\sqrt{2\tau \epsilon_{k}}\bw_{k},
\end{equation}
where $\bx_k\in \MX\in\mathbb{R}^{d}$, $\epsilon_{k}$ is the learning rate at iteration $k$, $N$ denotes the number of total data points,  $\tau$ is the temperature, and $\bw_k$ is a standard Gaussian vector of dimension $d$. In particular, $\frac{N}{n}\nabla_{\bx} \widetilde U(\bx)$ is an unbiased stochastic gradient estimator based on a mini-batch data $\mathcal{B}$ of size $n$ and $\frac{N}{n}\widetilde U(\bx)$ is the unbiased energy estimator for the exact energy function $U(\bx)$. Under mild conditions on $U$, $\bx_{k+1}$ is known to converge weakly to a unique invariant distribution $\pi(\bx)\propto e^{-\frac{U(\bx)}{\tau}}$ as $\epsilon_k\rightarrow 0$.

\subsection{Contour Stochastic Gradient Langevin Dynamics}
\label{ori_csgld}
Despite its theoretical guarantees, SGLD can converge exponentially slow when $U(\bx)$ is non-convex and exhibits high energy barriers. To remedy this issue, CSGLD \cite{CSGLD} exploits the flat histogram idea and proposes to simulate from a flattened density with much lower energy barriers
\begin{equation}
\label{flat_density}
    \varpi_{\Psi_{\btheta}}(\bx) \propto {\pi(\bx)}/{\Psi^{\zeta}_{\btheta}(U(\bx))},
\end{equation}
where $\zeta$ is a hyperparameter, $\Psi_{\btheta}(u)= \sum_{i=1}^m \bigg(\theta(i-1)e^{(\log\theta(i)-\log\theta(i-1)) \frac{u-u_{i-1}}{\Delta u}}\bigg) 1_{u_{i-1} < u \leq u_i}.$

In particular, $\{u_i\}_{i=0}^m$ determines the partition $\{\MX_i\}_{i=1}^m$ of $\MX$ such that $\MX_i=\{\bx: u_{i-1}<U(\bx)\leq u_i\}$, where $-\infty=u_0<u_1<\cdots<u_{m-1}<u_m=\infty$. For practical purposes, we assume $u_{i+1}-u_i=\Delta u$ for $i=1,\cdots, m-2$. In addition, $\btheta=(\theta(1), \theta(2), \ldots, \theta(m))$ 
 is the self-adapting parameter in the space $\bTheta=\bigg\{\left(\theta(1),\cdots, \theta(m)\right)\mid 0<\theta(1),\cdots,\theta(m)<1 \&\ \sum_{i=1}^m \theta(i)=1 \bigg\}$.
 
Ideally, setting $\zeta=1$ and $\theta(i)=\theta_{\infty}(i)$, where $\theta_{\infty}(i)=\int_{\MX_i} \pi(\bx)d\bx$ for $i\in\{1,2,\cdots, m\}$, enables CSGLD to achieve a ``random walk'' in the space of energy and to penalize the over-visited partition \cite{WangLandau2001, Liang07, Fortetal2011, Fort15}. However, the optimal values of $\btheta_{\infty}$ is unknown {\it a priori}. To tackle this issue, CSGLD proposes the following procedure to adaptively estimate $\btheta$ via stochastic approximation (SA) \cite{RobbinsM1951, Albert90}:
\begin{itemize}
\item[(1)] Sample $\bx_{k+1}=\bx_k+\epsilon_{k} \frac{N}{n}\nabla_{\bx}\widetilde  U_{\Psi_{\btheta_k}}(\bx_k)+\sqrt{2\tau \epsilon_{k}}\bw_{k}$, 
\item[(2)] Optimize $\bm{\theta}_{k+1}=\bm{\theta}_{k}+\omega_{k+1} \mathbb{ \widetilde H}(\bm{\theta}_{k}, \bx_{k+1}),$
\end{itemize}
where $\nabla_{\bx}\widetilde U_{\Psi_{\btheta}}(\cdot)$ is a stochastic gradient function of $\varpi_{\Psi_{\btheta}}(\cdot)$ to be detailed in Algorithm \ref{alg:ICSGLD}. $\mathbb{ \widetilde H}(\btheta,\bx):=\left(\mathbb{ \widetilde H}_1(\btheta,\bx), \cdots, \mathbb{ \widetilde H}_m(\btheta,\bx)\right)$ is random-field function where each entry follows
\begin{equation}
\label{ori_randomF}
    \mathbb{ \widetilde H}_i(\btheta,\bx)={\theta}^{\zeta}( J_{\widetilde U} (\bx))\left(1_{i= J_{\widetilde U}(\bx)}-{\theta}(i)\right), \text{where } J_{\widetilde U}(\bx)=\sum_{i=1}^m i 1_{u_{i-1}<\frac{N}{n} \widetilde U(\bx)\leq u_i}.
\end{equation}
Theoretically, CSGLD converges to a sampling-optimization equilibrium in the sense that $\btheta_{k}$ approaches to a fixed point $\btheta_{\infty}$ and the samples are drawn from the flattened density $\varpi_{\Psi_{\btheta_{\infty}}}(\bx)$. Notably, the mean-field system is \emph{globally stable} with a unique stable equilibrium point in a small neighborhood of $\btheta_{\infty}$. Moreover, such an appealing property holds even when $U(\bx)$ is non-convex.

 \begin{algorithm*}[tb]
   \caption{Interacting contour stochastic gradient Langevin dynamics algorithm (ICSGLD). $\{\MX_i\}_{i=1}^m$ is pre-defined partition and $\zeta$ is a hyperparameter. The update rule in distributed-memory settings and discussions of hyperparameters is detailed in section \ref{hyper_setup} in the supplementary material. 
   }
   \label{alg:ICSGLD}
\begin{algorithmic}
   \STATE{\bfseries [1.] (Data subsampling)} Draw a mini-batch data $\mathcal{B}_k$ from $\mathcal{D}$, and compute stochastic gradients $\nabla_{\bx}\widetilde U(\bx_k^{(p)})$ and energies $\widetilde U(\bx_k^{(p)})$ for each $\bx^{(p)}$, where $p\in\{1,2,\cdots, P\}$, $\mid\mathcal{B}_k\mid=n$, and $\mid\mathcal{D}\mid=N$.

   \STATE {\bfseries [2.] (Parallel simulation)}
  Sample $\bx_{k+1}^{\pop P}:=(\bx_{k+1}^{(1)}, \bx_{k+1}^{(2)}, \cdots, \bx_{k+1}^{(P)})^{\top}$ based on SGLD and $\btheta_k$
   \begin{equation}
   \bx_{k+1}^{\pop P}=\bx_k^{\pop P}+\epsilon_{k} \frac{N}{n}\nabla_{\bx}\widetilde  \bU_{\Psi_{\btheta_k}}(\bx_k^{\pop P})+\sqrt{2\tau \epsilon_{k}}\bw_{k}^{\pop P},
   \end{equation}
   where $\epsilon_{k}$ is the learning rate, $\tau$ is the temperature, $\bw_{k}^{\pop P}$ denotes $P$ independent standard Gaussian vectors, $\nabla_{\bx}\widetilde \bU_{\Psi_{\btheta}}(\bx^{\pop P})=(\nabla_{\bx}\widetilde U_{\Psi_{\btheta}}(\bx^{(1)}), \nabla_{\bx}\widetilde U_{\Psi_{\btheta}}(\bx^{(2)}), \cdots, \nabla_{\bx}\widetilde U_{\Psi_{\btheta}}(\bx^{(P)}))^{\top}$, and $\nabla_{\bx}\widetilde U_{\Psi_{\btheta}}(\bx)= \left[1+ 
   \frac{\zeta\tau}{\Delta u}  \left(\log \theta({J}_{\widetilde U}(\bx))-\log\theta((J_{\widetilde U}(\bx)-1)\vee 1) \right) \right]  
    \nabla_{\bx} \widetilde U(\bx)$ for any $\bx\in\MX$.
  \STATE {\bfseries [3.] (Stochastic approximation)} Update the self-adapting parameter $\theta(i)$ for $i\in\{1,2,\cdots, m\}$ 
  \begin{equation}
  \begin{split}
  \label{updateeq_icsgld}
 \theta_{k+1}(i)&={\theta}_{k}(i)+\omega_{k+1}\frac{1}{P}\sum_{p=1}^P {\theta}_{k}( J_{\widetilde U}(\bx_{k+1}^{(p)}))\left(1_{i=J_{\widetilde U}(\bx_{k+1}^{(p)})}-{\theta}_{k}(i)\right),
 \end{split}
 \end{equation}
 where $1_{A}$ is an indicator function that takes value 1 if the event $A$ appears and equals 0 otherwise. 

\vspace{-0.04in}
\end{algorithmic}
\end{algorithm*}

\section{Interacting Contour Stochastic Gradient Langevin Dynamics}
\vskip -0.05in
The major goal of interacting CSGLD (ICSGLD) is to improve the efficiency of CSGLD. In particular, the self-adapting parameter $\btheta$ is crucial for ensuring the sampler to escape from the local traps and traverse the whole energy landscape, and how to reduce the variability of $\btheta_k$'s is the key to the success of such a dynamic importance sampling algorithm. To this end, we propose an efficient variance reduction scheme via interacting parallel systems to improve the accuracy of  $\btheta_{k}$'s.

\subsection{Interactions in Parallelism}

Now we first consider a na\"{i}ve parallel sampling scheme with $P$ chains as follows
\begin{equation*}
\small
    \bx_{k+1}^{\pop P}=\bx_k^{\pop P}+\epsilon_{k} \frac{N}{n}\nabla_{\bx}\widetilde  \bU_{\Psi_{\btheta_k}}(\bx_k^{\pop P})+\sqrt{2\tau \epsilon_{k}}\bw_{k}^{\pop P},
\end{equation*}
where $\bx^{\pop P}=(\bx^{(1)}, \bx^{(2)}, \cdots, \bx^{(P)})^{\top}$, $\bw_{k}^{\pop P}$ denotes $P$ independent standard Gaussian vectors, and $\widetilde \bU_{\Psi_{\btheta}}(\bx^{\pop P})=(\widetilde U_{\Psi_{\btheta}}(\bx^{(1)}), \widetilde U_{\Psi_{\btheta}}(\bx^{(2)}), \cdots, \widetilde U_{\Psi_{\btheta}}(\bx^{(P)}))^{\top}$.

Stochastic approximation aims to find the solution $\btheta$ of the mean-field system $h(\btheta)$ such that 
\begin{equation*}
\small
\begin{split}
    h(\btheta)&=\int_{\MX} \widetilde H(\bm{\theta}, \bm{\bx}) \varpi_{\btheta}(d\bx)=0,
\end{split}
\end{equation*}
where $\varpi_{\btheta}$ is the invariant measure simulated via SGLD that approximates $\varpi_{\Psi_{\btheta}}$ in (\ref{flat_density}) and $\widetilde H(\bm{\theta}, \bm{\bx})$ is the novel random-field function to be defined later in (\ref{new_randomF}). Since $h(\btheta)$ is observable only up to large random perturbations (in the form of $\widetilde H(\bm{\theta}, \bm{\bx})$), the optimization of $\btheta$ based on isolated random-field functions may not be efficient enough. However, due to the conditional independence of $\bx^{(1)}, \bx^{(2)},\cdots, \bx^{(P)}$ in parallel sampling, it is very natural to consider a Monte Carlo average
\begin{equation}
\label{monte_carlo_avg}
\small
\begin{split}
    h(\btheta)&=\frac{1}{P}\sum_{p=1}^P\int_{\MX} \widetilde H(\bm{\theta}, \bm{\bx}^{(p)}) \varpi_{\btheta}(d\bx^{(p)})=0.
\end{split}
\end{equation}
Namely, we are considering the following stochastic approximation scheme
\begin{equation}
\label{SA_step}
    \bm{\theta}_{k+1}=\bm{\theta}_{k}+\omega_{k+1} \widetilde \bH(\bm{\theta}_{k}, \bx_{k+1}^{\pop P}),
\end{equation}
where $\widetilde \bH(\bm{\theta}_{k}, \bx_{k+1}^{\pop P})$ is an interacting random-field function $ \widetilde \bH(\bm{\theta}_{k}, \bx_{k+1}^{\pop P})=\frac{1}{P}\sum_{p=1}^P \widetilde H(\bm{\theta}_{k}, \bx_{k+1}^{(p)})$.
Note that the Monte Carlo average is very effective to reduce the variance of the interacting random-field function $\widetilde \bH(\bm{\theta}, \bm{\bx}^{\pop P})$ based on the conditionally independent random field functions. Moreover, each chain shares with others only a very short message during each iteration. Therefore, the interacting parallel system is well suited for distributed computing, where the \emph{implementations and communication costs} are further detailed in section \ref{communication_cost} in the supplementary material. By contrast, each chain of the non-interacting parallel CSGLD algorithm deals with the parameter $\btheta$ and a large-variance random-field function $\widetilde H(\bm{\theta}, \bx)$ individually, leading to coarse estimates in the end.

Formally, for the population/ensemble interaction scheme (\ref{SA_step}),  we define a novel random-field function $\widetilde H(\btheta,\bx)=(\widetilde H_1(\btheta,\bx), \widetilde H_2(\btheta,\bx), \cdots, \widetilde H_m(\btheta,\bx))$, where each component satisfies
\begin{equation}
\label{new_randomF}
\widetilde H_i(\btheta,\bx)={\theta}( J_{\widetilde U} (\bx))\left(1_{i= J_{\widetilde U}(\bx)}-{\theta}(i)\right).  
\end{equation}
As shown in Lemma \ref{convex_main}, the corresponding mean-field function proposes to converge to a different fixed point $\btheta_{\star}$, s.t.
\begin{equation}
\small
\label{relation_two_thetas}
    \theta_{\star}(i)\propto\left(\int_{\MX_i} e^{-\frac{U(\bx)}{\tau}}d\bx\right)^{\frac{1}{\zeta}}\propto \btheta_{\infty}^{\frac{1}{\zeta}}(i).
\end{equation}
A large data set often renders the task of estimating $\btheta_{\infty}$ numerically challenging. By contrast, we resort to a different solution by estimating $\btheta_{\star}$ instead based on a large value of $\zeta$. The proposed algorithm is summarized in Algorithm \ref{alg:ICSGLD}. For more study on the scalablity of the new scheme, we leave the discussion in section \ref{scalability}.

\subsection{Related Work}

Replica exchange SGLD \cite{deng2020, deng_VR} has successfully extended the traditional replica exchange \cite{PhysRevLett86, Geyer91, parallel_tempering05} to big data problems. However, it works with two chains only and has a low swapping rate. As shown in Figure \ref{replica_vs_contour}(a), a na\"{i}ve extension of multi-chain replica exchange SGLD yields low communication efficiency. Despite some recipe in the literature \cite{Katzgraber06, Elmar08, Syed_jrssb}, how to conduct multi-chain replica exchange with low-frequency swaps is still an open question.

\begin{figure*}[!ht]
  \centering
  \subfloat[Replica Exchange (parallel tempering)]{\includegraphics[scale=0.5
  ]{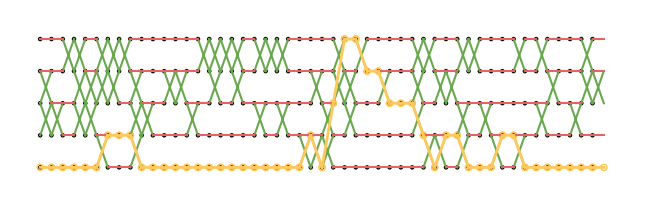}}\enskip
  \subfloat[Interacting contour SGLD (ICSGLD)]{\includegraphics[scale=0.53]{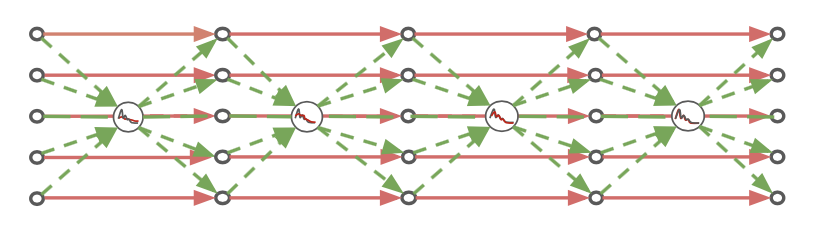}}
  \caption{A comparison of communication costs between replica exchange (RE) and ICSGLD. We see RE takes many iterations to swap with all the other chains; by contrast, ICSGLD possesses a pleasingly parallel mechanism where the only cost comes from sharing a light message.}
  \label{replica_vs_contour}
\end{figure*}

Stein variational gradient descent (SVGD) \cite{SVGD} is a popular approximate inference method to drive a set of particles for posterior approximation. In particular, repulsive forces are proposed to prevent particles to collapse together into neighboring regions, which resembles our strategy of penalizing over-visited partition. However, SVGD tends to underestimate the uncertainty given a limited number of particles. Moreover, the quadratic cost in kernel matrix computation further raises the scalability concerns as more particles are proposed.

Admittedly, ICSGLD is not the first interacting importance sampling algorithm. For example, 
a population stochastic approximation Monte Carlo (pop-SAMC) algorithm  has been proposed in \cite{SongWL2014}, and 
an interacting particle Markov chain Monte Carlo (IPMCMC) algorithm has been proposed  in \cite{IPMCMC}.
A key difference between our algorithm and others is that our algorithm is mainly devised for big data problems. The IPMCMC and pop-SAMC are gradient-free samplers, which are hard to be adapted to high-dimensional big data problems.

Other parallel SGLD methods \cite{Ahn14_icml, chen16_distributed} aim to reduce the computational cost of gradient estimations in distributed computing, which, however, does not consider interactions for accelerating the convergence. \cite{Li19_v2} proposed asynchronous protocols to reduce communication costs when the master aggregates model parameters from all workers. Instead, we don't communicate the parameter $\bx\in\mathbb{R}^d$ but only share $\btheta\in \mathbb{R}^m$ and the indices, where $m\ll d$.

Our work also highly resembles the well-known Federated Averaging (FedAvg) algorithm \cite{lhy+20}, except that the stochastic gradient $\widetilde U(\bx)$ is replaced with the random field function $\widetilde H(\btheta, \bx)$ and we only share the low-dimensional latent vector $\btheta$. Since privacy concerns and communication cost are not major bottlenecks of our problem, we leave the study of taking the Monte Carlo average in Eq.(\ref{monte_carlo_avg}) every $K>1$ iterations for future works.

\section{Convergence Properties}
\label{all_proof_pop_CSGLD}
\vskip -0.05in
To study theoretical properties of ICSGLD, we first show a local stability property that is well-suited to big data problems, and then we present the asymptotic normality for the stochastic approximation process in \emph{mini-batch settings}, which eventually yields the desired result that ICSGLD is asymptotically more efficient than a single-chain CSGLD with an equivalent computational cost.

\subsection{Local Stability for Non-linear Mean-field Systems in Big Data}

The first obstacle  for the theoretical study is to approximate the components of $\btheta_{\infty}$  corresponding to the high energy region.  
 To get around this issue, the random field function $\widetilde H(\btheta,\bx)$ in (\ref{new_randomF}) is adopted to estimate a different target $\btheta_{\star}\propto \btheta_{\infty}^{\frac{1}{\zeta}}$. As detailed in Lemma \ref{convex_appendix} in the supplementary material, the mean-field equation is now formulated as follows
\begin{equation}
\label{h_i_theta_main}
     h_i(\btheta)\propto \theta_{\star}^{\zeta}(i)-{\left(\theta(i) C_{\btheta}\right)^{\zeta}}+{\text{perturbations}},
\end{equation}
where $C_{\btheta}=\bigg(\frac{\widetilde Z_{\zeta,\btheta}}{\widetilde Z_{\zeta, \btheta_{\star}}^{\zeta}}\bigg)^{\frac{1}{\zeta}}$ and $\widetilde Z_{\zeta,\btheta}=\sum_{k=1}^m  \frac{\int_{\MX_k} \pi(\bx)d\bx}{\theta^{\zeta-1}(k)}$. We see that (\ref{h_i_theta_main}) may not be linearly stable as in \cite{CSGLD}. Although the solution of the mean-field system $h(\btheta)=0$ is still unique, there may exist unstable invariant subspaces, leading us to consider the local properties. For a proper initialization of $\btheta$, which can be achieved by pre-training the model long enough time through SGLD, the mean value theorem implies a linear property in a local region
\begin{equation*}
     h_i(\btheta)\propto \theta_{\star}(i)-\theta(i)+{\text{perturbations}}.
\end{equation*}
Combining the perturbation theory \cite{Eric}, we present the following stability result:
\begin{lemma}[Local stability, informal version of Lemma \ref{convex_appendix}] \label{convex_main} 
Assume Assumptions A1-A4 (given in the supplementary material) hold. For any properly initialized $\btheta$, we have $\langle h(\btheta), \btheta - \widehat\btheta_{\star}\rangle \leq  -\phi\|\btheta - \widehat\btheta_{\star}\|^2$,  where $\widehat \btheta_{\star}=\btheta_{\star}+\mathcal{O}\left(\sup_{\bx}\Var(\xi_n(\bx))+\epsilon+\frac{1}{m}\right)$, $\btheta_{\star}\propto \btheta_{\infty}^{\frac{1}{\zeta}}$ ,
$\phi>0$,  and $\xi_n(\bx)$ denotes the noise in the stochastic energy estimator of batch size $n$ and $\Var(\cdot)$ denotes the variance. 
\end{lemma}

By justifying the drift conditions of the adaptive transition kernel and relevant smoothness properties, we can prove the existence and regularity properties of the solution of the Poisson's equation in Lemma \ref{lyapunov} in the supplementary material. In what follows, we can control the fluctuations in stochastic approximation and eventually yields the $L^2$ convergence. 
\begin{lemma}[$L^2$ convergence rate, informal version of Lemma \ref{latent_convergence_appendix}]
\label{latent_convergence_main}
Given standard Assumptions A1-A5. $\btheta_k$ converges to $\widehat\btheta_{\star}$,
where $\widehat\btheta_{\star}=\btheta_{\star}+\mathcal{O}\left(\sup_{\bx}\Var(\xi_n(\bx))+\epsilon+\frac{1}{m}\right)$, such that
\begin{equation*}
    \E\left[\|\bm{\theta}_{k}-\widehat\btheta_{\star}\|^2\right]= \mathcal{O}\left(\omega_{k}\right).
\end{equation*}
\end{lemma}

The result differs from Theorem 1 of \cite{CSGLD} in that the biased fixed point $\widehat\btheta_{\star}$ instead of $\btheta_{\star}$ is treated as the equilibrium of the continuous system, which provides us a user-friendly proof. Similar techniques have been adopted by \cite{Alain17, Xu18}. Although the global stability \cite{CSGLD} may be sacrificed when $\zeta\neq 1$ based on Eq.(\ref{new_randomF}), $\btheta_{\star}\propto \btheta_{\infty}^{\frac{1}{\zeta}}$ is much easier to estimate numerically for any $i$ that yields $0<\btheta_{\infty}(i)\ll 1$ based on a large $\zeta>1$.

\subsection{Asymptotic Normality}

 To study the asymptotic behavior of $\omega_k^{-\frac{1}{2}}(\btheta_k-\widehat\btheta_{\star})$, where $\widehat\btheta_{\star}$ is the equilibrium point s.t. $\widehat\btheta_{\star}=\btheta_{\star}+\mathcal{O}\left(\Var(\xi_n(\bx))+\epsilon+\frac{1}{m}\right)$, we consider  a fixed step size $\omega$ in the SA step for ease of explanation. Let $\bar\btheta_t$ denote the solution of the mean-field system in continuous time ($\bar\btheta_0=\btheta_0$), and rewrite the single-chain SA step (\ref{SA_step}) as follows
\begin{equation*}
\begin{split}
\label{ga_main}
    \btheta_{k+1}-\bar\btheta_{(k+1)\omega}&=\btheta_k-\bar\btheta_{k\omega}+\omega\left(H(\btheta_{k}, \bx_{k+1})-H(\bar\btheta_{k\omega}, \bx_{k+1})\right)\\
    &\quad+\omega\left(H(\bar\btheta_{k\omega}, \bx_{k+1})-h(\bar\btheta_{k\omega})\right)-\left(\bar\btheta_{(k+1)\omega}-\bar\btheta_{k\omega}-\omega  h(\bar\btheta_{k\omega})\right).
\end{split}
\end{equation*}

Further, we set $\widetilde\btheta_{k\omega}:= \omega^{-\frac{1}{2}}(\btheta_{k}-\bar\btheta_{k\omega})$. Then the stochastic approximation differs from the mean field system in that
\begin{equation*}
\begin{split}
    \widetilde\btheta_{(k+1)\omega}&= \underbrace{\omega^{\frac{1}{2}}\sum_{i=0}^k \left(H(\btheta_{i}, \bx_{i+1})-H(\bar\btheta_{i\omega}, \bx_{i+1})\right)}_{\text{I: perturbations}}+\omega^{\frac{1}{2}}\sum_{i=0}^k \underbrace{\left(H(\bar\btheta_{i\omega}, \bx_{i+1})-h(\bar\btheta_{i\omega})\right)}_{\text{II: martingale} \ \mathcal{M}_i}-\omega^{\frac{1}{2}}\cdot\text{remainder}\\
    &\approx \omega^{\frac{1}{2}}\sum_{i=0}^k h_{\btheta}(\btheta_{i\omega})  \underbrace{(\btheta_i-\bar\btheta_{i\omega})}_{\approx \omega^{\frac{1}{2}} \widetilde \btheta_{i\omega}}+\omega^{\frac{1}{2}}\sum_{i=0}^k \mathcal{M}_i\approx \int_{0}^{(k+1)\omega}h_{\btheta}(\bar\btheta_{s})\widetilde\btheta_{s}ds+\int_0^{(k+1)\omega} \bR^{\frac{1}{2}}(\bar\btheta_s)d\bW_s,
\end{split}
\end{equation*}
where $h_{\btheta}(\btheta):=\frac{d}{d\btheta} h(\btheta)$ is a matrix, $\bW\in\mathbb{R}^m$ is a standard Brownian motion, the last term follows from a certain central limit theorem \cite{Albert90} and $\bR$ denotes the covariance matrix of the random-field function s.t. $\bR(\btheta):=\sum_{k=-\infty}^{\infty} \cov_{\btheta}(H(\btheta, \bx_k), H(\btheta, \bx_0))$. 

We expect the weak convergence of $\bU_k$ to the stationary distribution of a diffusion
\begin{equation}
\label{slde_main}
    d\bU_t=h_{\btheta}(\btheta_t) \bU_t dt + \bR^{1/2}(\btheta_t)d\bW_t,
\end{equation}
where $\bU_t=\omega_t^{-1/2}(\btheta_t-\widehat\btheta_{\star})$.  Given that $\btheta_t$ converges to $\widehat\btheta_{\star}$ sufficiently fast and the local linearity of $h_{\btheta}$, the diffusion (\ref{slde_main}) resembles the Ornstein–Uhlenbeck process and yields the following solution
\begin{equation*}
    \bU_t\approx e^{-th_{\btheta}(\widehat\btheta_{\star})}\bU_0+\int_0^t e^{-(t-s)h_{\btheta}(\widehat\btheta_{\star})}\circ \bR(\widehat\btheta_{\star}) d\bW_s.
\end{equation*}
Then we have the following theorem, whose  formal proof is given in section \ref{proof_theorem_1}.
\begin{theorem}[Asymptotic Normality]
\label{Asymptotic}
Assume Assumptions \ref{ass2a}-\ref{ass1} hold. We have the following weak convergence
\begin{equation*}
\begin{split}
    \omega_k^{-1/2}(\btheta_k-\widehat\btheta_{\star})\Rightarrow\mathcal{N}(0, \bSigma), \text{ where  } \bSigma=\int_0^{\infty} e^{t h_{\btheta_{\star}}}\circ \bR\circ  e^{th^{\top}_{\btheta_{\star}}}dt, h_{\btheta_{\star}}=h_{\btheta}(\widehat\btheta_{\star}).
\end{split}
\end{equation*}
\end{theorem}

\subsection{Interacting Parallel Chains are More Efficient}

For clarity, we first denote an estimate of $\btheta$ based on ICSGLD with $P$ interacting parallel chains by $\btheta_k^{P}$ and denote the estimate based on a single-long-chain CSGLD by $\btheta_{kP}$.

Note that Theorem \ref{Asymptotic} holds for any step size $\omega_k=\mathcal{O}(k^{-\alpha})$, where $\alpha\in (0.5, 1]$. If we simply run a single-chain CSGLD algorithm with $P$ times of iterations, by  Theorem \ref{Asymptotic},  
\begin{equation*}
\begin{split}
    \omega_{kP}^{-1/2}(\btheta_{kP}-\widehat\btheta_{\star})\Rightarrow\mathcal{N}(0, \bSigma).
\end{split}
\end{equation*}
As to ICSGLD, since the covariance $\bSigma$ relies on $\bR$, which depends on the covariance of the martingale $\{\mathcal{M}_i\}_{i\geq 1}$, the conditional independence of $\bx^{(1)},\bx^{(2)},\cdots, \bx^{(P)}$ naturally results in an efficient variance reduction such that 
 \begin{corollary}[Asymptotic Normality for ICSGLD]
 Assume the same assumptions. For ICSGLD with $P$ interacting chains, we have the following weak convergence
\begin{equation*}
\begin{split}
    \omega_k^{-1/2}(\btheta_k^P-\widehat\btheta_{\star})\Rightarrow\mathcal{N}(0, \bSigma/P).
\end{split}
\end{equation*}
 \end{corollary}
That is, under a similar computational budget, we have $\frac{\|\Var(\btheta_{kP}-\widehat\btheta_{\star})\|_{\text{F}}}{\|\Var(\btheta_k^P-\widehat{\btheta}_*)\|_{\text{F}}}= \frac{w_{kP}}{w_k/P}\approx P^{1-\alpha}$.

 \begin{corollary}[Efficiency]
Given a decreasing step size $\omega_k=\mathcal{O}(k^{-\alpha})$, where $\alpha \in(0.5, 1]$, \emph{ICSGLD is asymptotically more efficient than the single-chain CSGLD with an equivalent training cost.} 
\end{corollary}

In practice, slowly decreasing step sizes are often preferred in stochastic algorithms for a better non-asymptotic performance \cite{Albert90}.

\section{Experiments}

\subsection{Landscape Exploration On MNIST via the Scalable Random-field Function}

This section shows how the novel random-field function (\ref{new_randomF}) facilitates the exploration of multiple modes on the MNIST dataset\footnote[4]{The random-field function \cite{CSGLD} requires an extra perturbation term as discussed in section D4 in the supplementary material \cite{CSGLD}; therefore it is not practically appealing in big data.}, while the standard methods, such as stochastic gradient descent (SGD) and SGLD, only \emph{get stuck in few local modes}. To simplify the experiments, we choose a large batch size of 2500 and only pick the first five classes, namely digits from 0 to 4. The \emph{learning rate is fixed} to 1e-6 and the temperature is set to $0.1$ \footnote[2]{Data augmentation implicitly leads to a more concentrated posterior \cite{Florian2020, Aitchison2021}.}. We see from Fig.\ref{Uncertainty_estimation_mnist_ICSGLD}{(a)} that both SGD and SGLD lead to fast decreasing losses. By contrast, ICSGLD yields fluctuating losses that traverse freely between high energy and low energy regions. As the particles stick in local regions, the penalty of re-visiting these zones keeps increasing until \emph{a negative learning rate is injected} to encourage explorations.

\begin{figure}[htbp]
\small
 \begin{tabular}{cccc}
(a) Training Loss & (b)  SGD & (c) SGLD & (d) ICSGLD \\ 
\includegraphics[height=1.4in,width=1.4in]{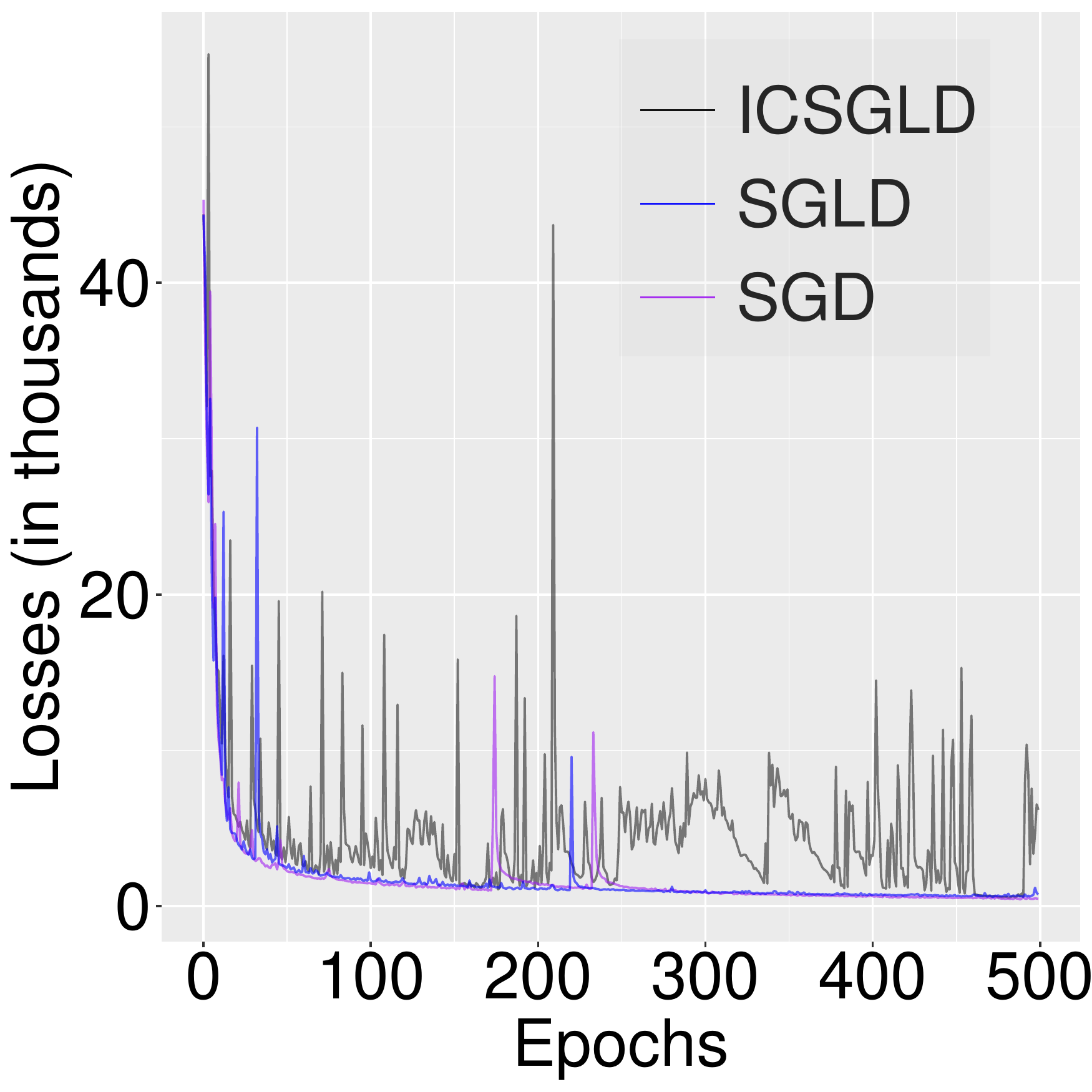} &
\includegraphics[height=1.45in,width=1.45in]{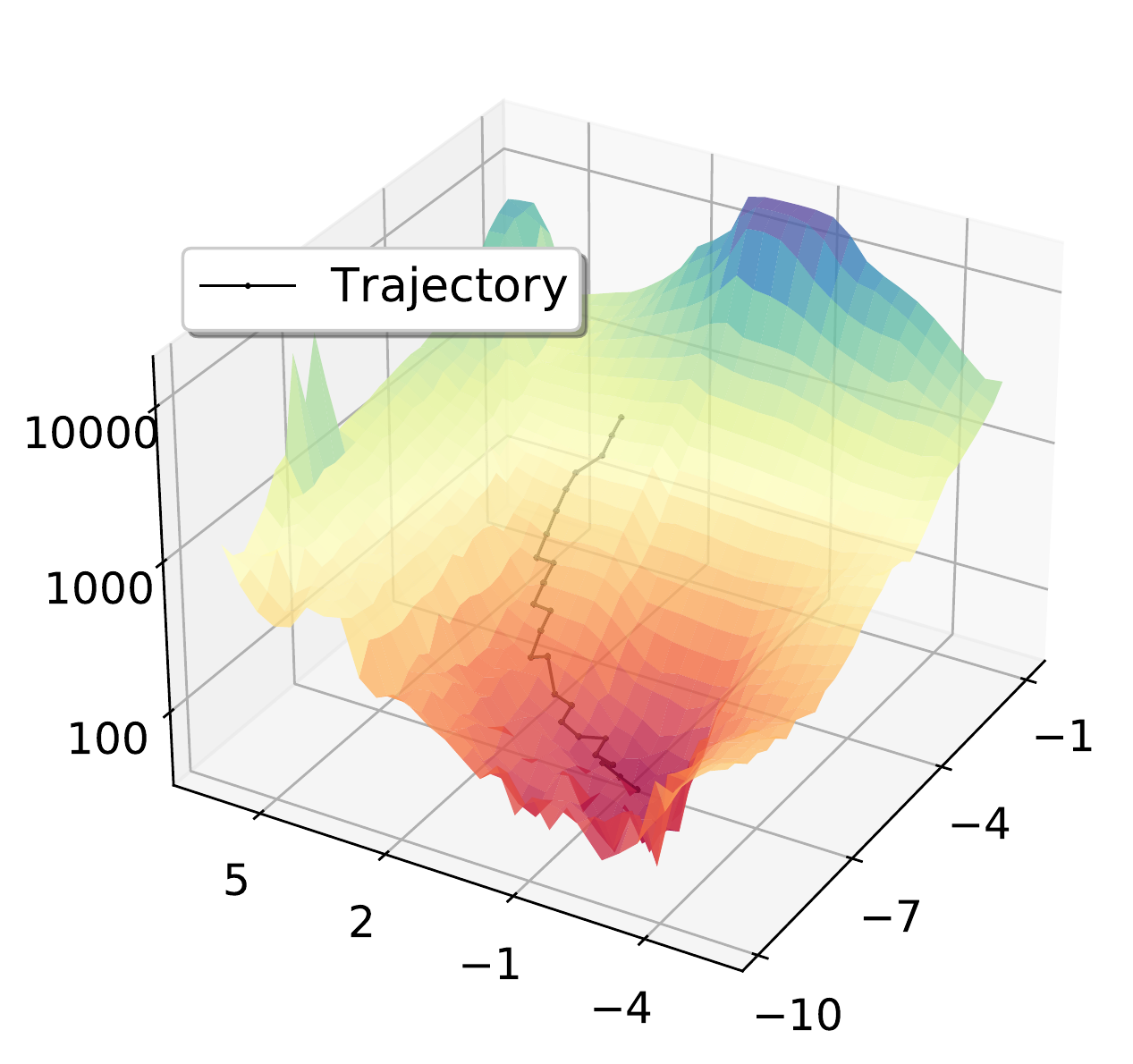} &
\includegraphics[height=1.45in,width=1.45in]{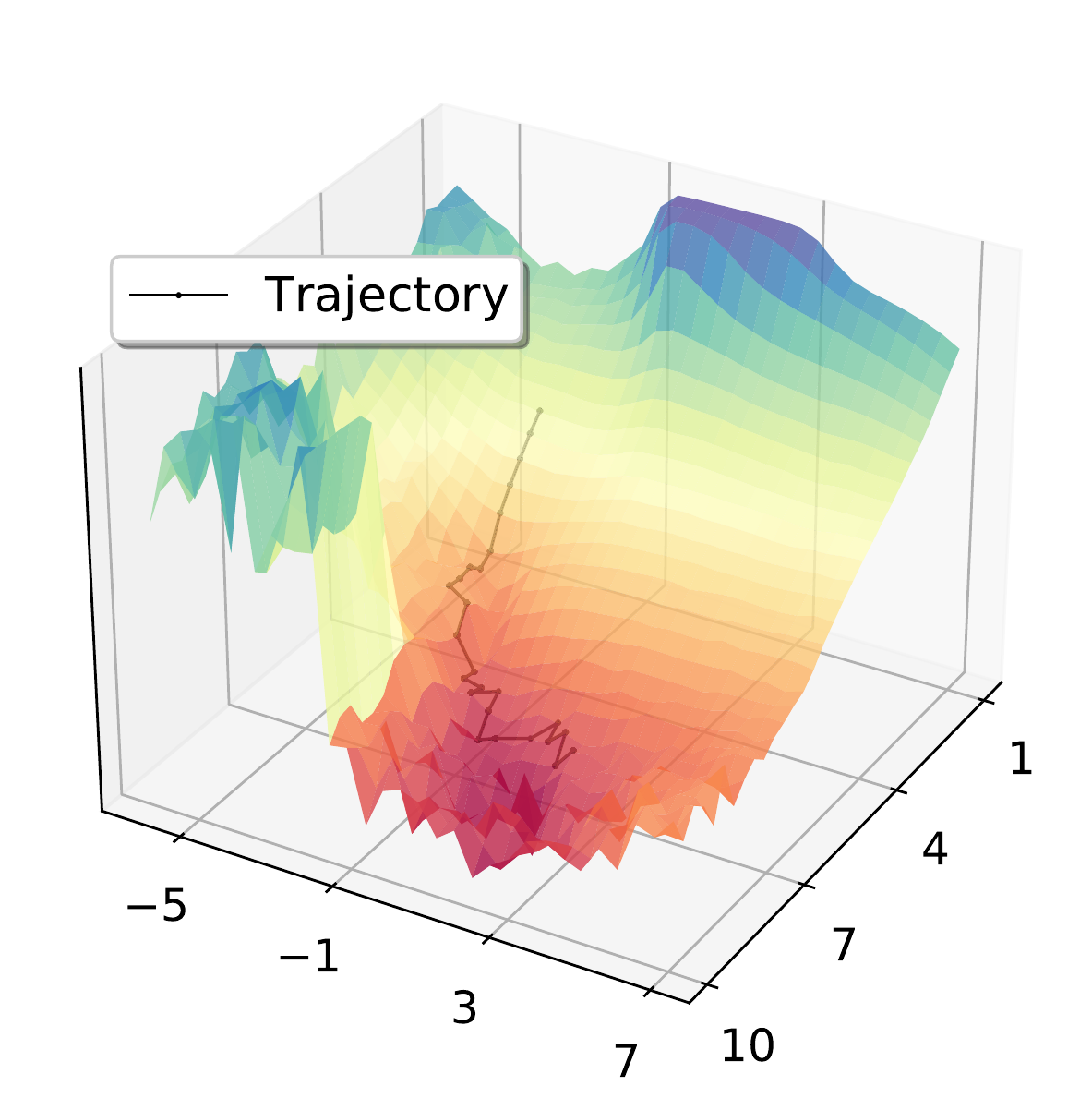} &
\includegraphics[height=1.45in,width=1.45in]{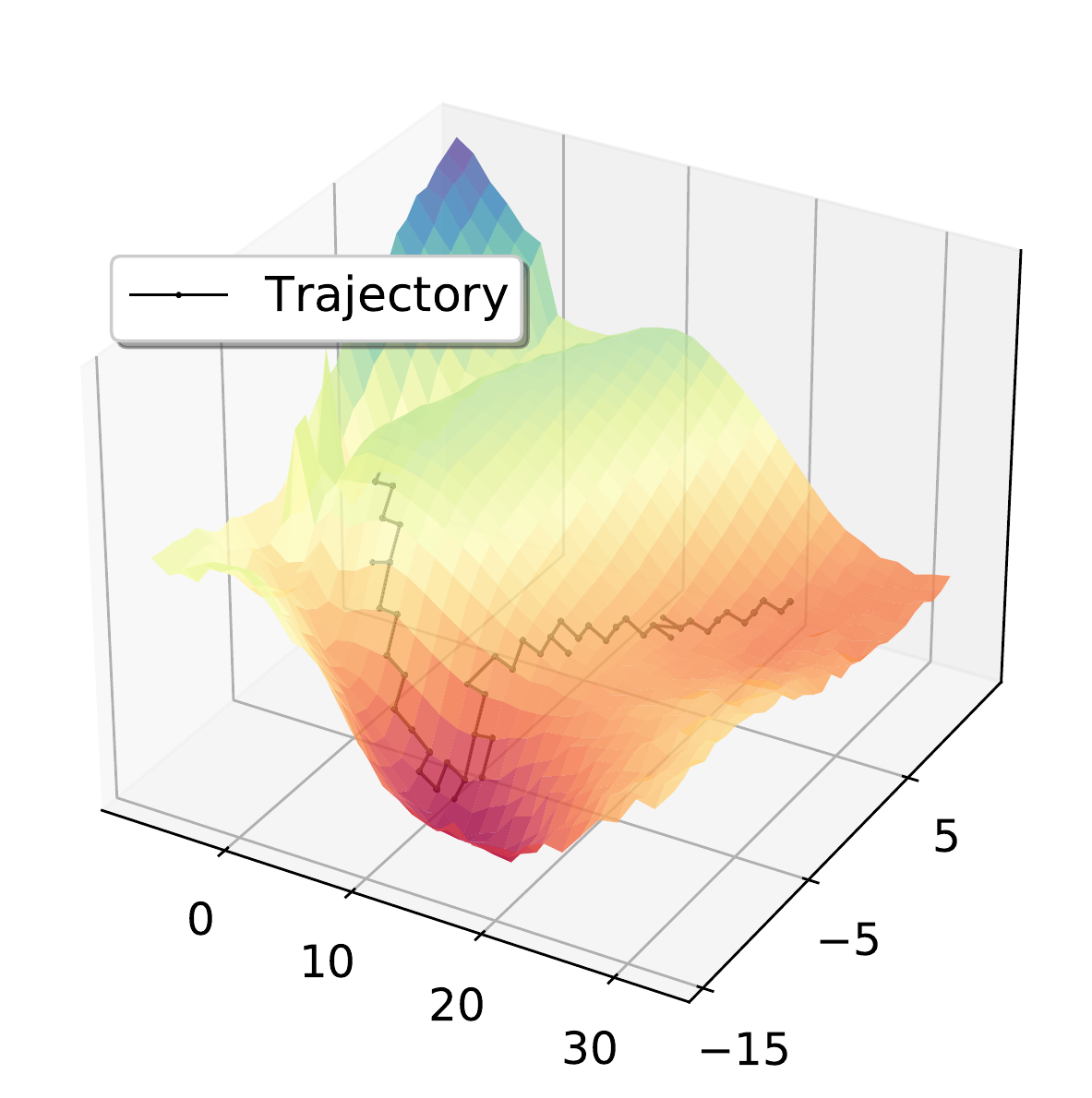}
\end{tabular}
  \caption{Visualization of mode exploration on a MNIST example based on different algorithms.} 
\label{Uncertainty_estimation_mnist_ICSGLD}
\end{figure}

We conducted a singular value decomposition (SVD) based on the first two coordinates to visualize the trajectories: We first choose a domain that includes all the coordinates, then we recover the parameter based on the grid point and truncated values in other dimensions, and finally we fine-tune the parameters and present the approximate losses of the trajectories in Fig.\ref{Uncertainty_estimation_mnist_ICSGLD}{(b-d)}. We see SGD trajectories get stuck in a local region; SGLD \emph{exploits a larger region} but is still quite limited in the exploration; ICSGLD, instead, first converges to a local region and then \emph{escapes it once it over-visits this region}. This shows the strength of ICSGLD in the simulations of complex multi-modal distributions. More experimental details are presented in section \ref{mnist_appendix} of the supplementary material.

\subsection{Simulations of Multi-modal Distributions}

This section shows the acceleration effect of ICSGLD via a group of simulation experiments for a multi-modal distribution. The baselines 
include popular Monte Carlo methods such as 
 CSGLD, SGLD, cyclical SGLD (cycSGLD), replica exchange SGLD (reSGLD), and the particle-based SVGD.

The target multi-modal density is presented in Figure \ref{subfig:true}. Figure \ref{figure:simulation}{(b-g)} displays the empirical performance of all the testing methods: the vanilla SGLD with 5 parallel chains ($\times$P5) undoubtedly performs the worst in this example and fails to quantify the weights of each mode correctly; the single-chain cycSGLD with 5 times of iterations ($\times$T5) improves the performance but is still not accurate enough; reSGLD ($\times$P5) and SVGD ($\times$P5) have good performances, while the latter is quite costly in computations; ICSGLD ($\times$P5) does not only traverse freely over the rugged energy landscape, but also yields the most accurate approximation to the ground truth distribution. By contrast, CSGLD ($\times$T5) performs worse than ICSGLD and overestimates the weights on the left side. For the detailed setups, the study of convergence speed, and runtime analysis, we refer interested readers to section \ref{simulation_appendix} in the supplementary material.
  \begin{figure*}[htbp]
    \centering
    \subfloat[Truth]{
    \begin{minipage}[t]{0.13\linewidth}
    \centering
    \includegraphics[width=0.86in]{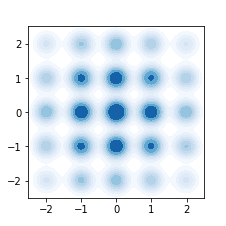}
    \label{subfig:true}
    \end{minipage}%
    }%
     \subfloat[SGLD]{
    \begin{minipage}[t]{0.13\linewidth}
    \centering
    \includegraphics[width=0.86in]{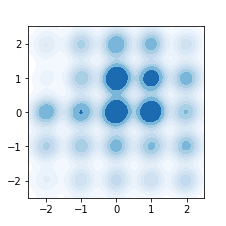}
    \label{subfig:pSGLD}
    \end{minipage}%
    }%
    \subfloat[cycSGLD]{
    \begin{minipage}[t]{0.13\linewidth}
    \centering
    \includegraphics[width=0.86in]{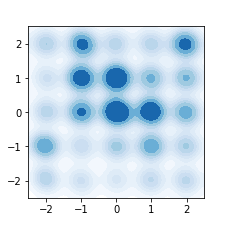}
    \label{subfig:cyclical SGLD}
    \end{minipage}%
    }%
    \subfloat[SVGD]{
    \begin{minipage}[t]{0.13\linewidth}
    \centering
    \includegraphics[width=0.87in]{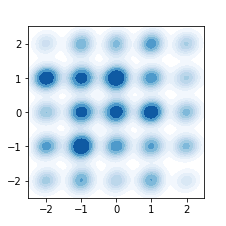}
    \label{subfig:pSVGD}
    \end{minipage}
    }%
    \subfloat[reSGLD]{
    \begin{minipage}[t]{0.13\linewidth}
    \centering
    \includegraphics[width=0.87in]{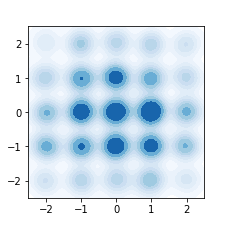}
    \label{subfig:reSGLD}
    \end{minipage}
    }%
    \subfloat[CSGLD]{
    \begin{minipage}[t]{0.13\linewidth}
    \centering
    \includegraphics[width=0.87in]{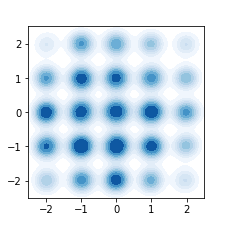}
    \label{subfig:CSGLD after}
    \end{minipage}
    }%
    \subfloat[ICSGLD]{
    \begin{minipage}[t]{0.13\linewidth}
    \centering
    \includegraphics[width=0.87in]{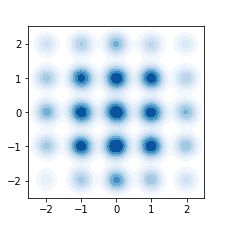}
    \label{subfig:iCSGLD after}
    \end{minipage}
    }%
  \caption {Empirical behavior on a simulation dataset. Figure \ref{subfig:cyclical SGLD} and \ref{subfig:CSGLD after} show the simulation based on a single chain with 5 times of iterations ($\times$T5) and the others run 5 parallel chains ($\times$P5). }
  \label{figure:simulation}
\end{figure*}

\subsection{Deep Contextual Bandits on Mushroom Tasks}

This section evaluates ICSGLD on the contextual bandit problem based on the UCI Mushroom data set as in \cite{bandits_showdown}. The mushrooms are assumed to arrive sequentially and the agent needs to take an action at each time step based on past feedbacks.
Our goal is to minimize the cumulative regret that measures the difference between the cumulative reward obtained by the proposed policy and optimal policy. We evaluate Thompson Sampling (TS) based on a variety of approximate inference methods for posterior sampling. We choose one $\epsilon$-greedy policy (EpsGreedy) based on the RMSProp optimizer with a decaying learning rate  \cite{bandits_showdown} as a baseline. Two variational methods, namely stochastic gradient descent with a constant learning rate (ConstSGD) \cite{Mandt} and Monte Carlo Dropout (Dropout) \cite{Gal16b} are compared to approximate the posterior distribution. For the sampling algorithms, we include preconditioned SGLD (pSGLD) \cite{Li16}, preconditioned CSGLD (pCSGLD) \cite{CSGLD}, and preconditioned ICSGLD (pICSGLD). Note that all the algorithms run 4 parallel chains with average outputs ($\times$P4) except that pCSGLD runs a single-chain with 4 times of computational budget ($\times$T4). 
For more details, we refer readers to section \ref{bandit_mushroom} in the supplementary material.

Figure.\ref{mushroom} shows that EpsGreedy $\times$P4 tends to explore too much for a long horizon as expected; ConstSGD$\times$P4 and Dropout$\times$P4 perform poorly in the beginning but eventually outperform EpsGreedy $\times$P4 due to the inclusion of uncertainty for exploration, whereas the uncertainty seems  to be inadequate due to the nature of variational inference. By contrast, pSGLD$\times$P4 significantly outperforms the variational methods by considering preconditioners within an exact sampling framework (SGLD).  As a unique algorithm that runs in a single-chain manner, pCSGLD$\times$T4 leads to the worst performance due to the inefficiency in learning the self-adapting parameters, fortunately, pCSGLD$\times$T4 \begin{wrapfigure}{r}{0.45\textwidth}
   \begin{center}
   \vskip -0.2in
     \includegraphics[width=0.45\textwidth]{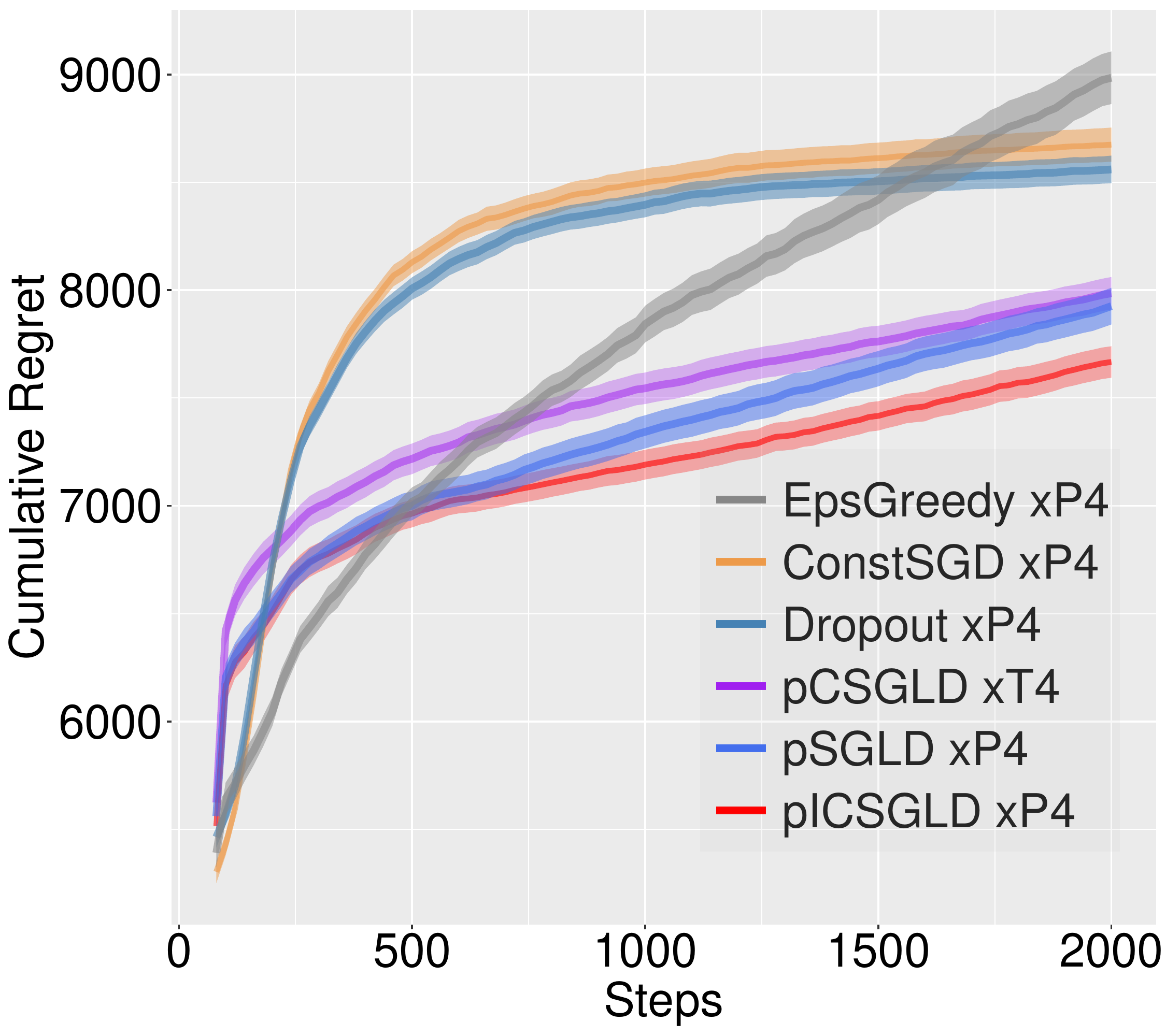}
   \end{center}
   \vskip -0.2in
   \caption{{Cumulative regret.}}
   \label{mushroom}
\end{wrapfigure} slightly outperform pSGLD$\times$P4 in the later phase with the help of the well-estimated self-adapting parameters. Nevertheless,  pICSGLD$\times$P4 propose to optimize the shared self-adapting parameters at the same time, which in turn greatly contributes to the simulation of the posterior. As a result, pICSGLD$\times$P4 consistently shows the lowest regret excluding the very early period. This shows the great superiority of the interaction mechanism in learning the self-adapting parameters for accelerating the simulations.



\subsection{Uncertainty Estimation}

This section evaluates the qualify of our algorithm in uncertainty quantification. For model architectures, we use residual networks (ResNet) \cite{kaiming15} and a wide ResNet (WRN) \cite{wide_residual};  we choose 20, 32, and 56-layer ResNets (denoted by ResNet20, et al.) and a WRN-16-8 network, a 16-layer WRN that is 8 times wider than ResNet16. We train the models on CIFAR100, and report the test accuracy (ACC) and test negative log-likelihood (NLL). For the out-of-distribution prediction performance, we test the well-trained models in Brier scores (Brier) \footnote{The Brier score measures the mean squared error between the predictive and actual probabilities.} on the Street View House Numbers dataset (SVHN). 

Due to the wide adoption of momentum stochastic gradient descent (M-SGD), we use stochastic gradient Hamiltonian Monte Carlo (SGHMC) \cite{Chen14} as the baseline sampling algorithm and denote the interacting contour SGHMC by ICSHMC. In addition, we include several high performing baselines, such as SGHMC with cyclical learning rates (cycSGHMC) \cite{ruqi2020}, SWAG based on cyclic learning rates of 10 cycles (cycSWAG) \cite{swag} and variance-reduced replica exchange SGHMC (reSGHMC) \cite{deng_VR}. For a fair comparison, ICSGLD also conducts variance reduction on the energy function to alleviate the bias. Moreover, a large $\zeta=3\times 10^{6}$ is selected \footnote[4]{The adoption of data augmentation leads to a cold posterior \cite{Florian2020, Aitchison2021} and the underlying $\zeta$ becomes much larger.}, which only induces mild gradient multipliers ranging from $-1$ to $2$ to penalize over-visited partitions. We don't include SVGD \cite{SVGD} and SPOS \cite{SPOS} for scalability reasons. A batch size of 256 is selected. We run 4 parallel processes ($\times$P4) with 500 epochs for M-SGD, reSGHMC and ICSGHMC and run cycSGHMC and cycSWAG 2000 epochs ($\times$T4) based on a single process with 10 cycles.  Refer to section \ref{UQ_appendix} of the supplementary material for the detailed settings.

\begin{table*}[ht]
\begin{sc}
\vspace{-0.15in}
\caption{Uncertainty estimations on CIFAR100 and SVHN. } \label{UQ_test}
\vspace{-0.1in}
\small
\begin{center} 
\begin{tabular}{c|ccc|ccc}
\hline
\multirow{2}{*}{Model} & \multicolumn{3}{c|}{R\upshape{es}N\upshape{et}20} & \multicolumn{3}{c}{R\upshape{es}N\upshape{et}32}  \\
\cline{2-7}
 & ACC (\%) & NLL & B\upshape{rier} (\textperthousand) & ACC (\%) & NLL & B\upshape{rier} (\textperthousand) \\
\hline
\hline
\upshape{cyc}SGHMC$\times$T4  & 75.01$\pm$0.10 & 8468$\pm$30 & 2.81$\pm$0.13 & 77.64$\pm$0.17 & 7658$\pm$19 & 3.29$\pm$0.13 \\ 
\upshape{cyc}SWAG$\times$T4                & 75.12$\pm$0.11 & 8456$\pm$26 & 2.78$\pm$0.12  & 77.59$\pm$0.15 & 7656$\pm$22 & 3.19$\pm$0.14 \\ 
\hline
\hline
M-SGD$\times$P4 & 75.89$\pm$0.12  & 8175$\pm$25  &  2.58$\pm$0.08 &  78.12$\pm$0.12 & 7538$\pm$23 & 2.82$\pm$0.15  \\ 
\upshape{re}SGHMC$\times$P4 & 76.15$\pm$0.16  & 8196$\pm$27  & 2.73$\pm$0.10  & 78.57$\pm$0.07 & 7454$\pm$15 & 3.04$\pm$0.09   \\ 
ICSGHMC$\times$P4 & \textbf{76.39$\pm$0.15}  & \textbf{8076$\pm$31}  & \textbf{2.52$\pm$0.07} & \textbf{78.79$\pm$0.16}  & \textbf{7384$\pm$29}  & \textbf{2.73$\pm$0.12} \\ 
\hline
\hline
\multirow{2}{*}{Model} & \multicolumn{3}{c|}{R\upshape{es}N\upshape{et}56} & \multicolumn{3}{c}{WRN-16-8}\\
\cline{2-7}
 & ACC (\%) & NLL & B\upshape{rier} (\textperthousand) & ACC (\%) & NLL & B\upshape{rier} (\textperthousand) \\
\hline
\hline
\upshape{cyc}SGHMC$\times$T4  & 81.23$\pm$0.19 & 6770$\pm$59 & 3.18$\pm$0.08   &  82.98$\pm$0.03 & 6384$\pm$11 & 2.17$\pm$0.05  \\ 
\upshape{cyc}SWAG$\times$T4                & 81.14$\pm$0.11 & 6744$\pm$55 & 3.06$\pm$0.09   &  83.05$\pm$0.04 & 6359$\pm$14 & 2.04$\pm$0.07  \\ 
\hline
\hline
M-SGD$\times$P4 & 80.94$\pm$0.12   & 6847$\pm$22  & \textbf{2.86$\pm$0.08}   & 82.57$\pm$0.07  & 6821$\pm$21 & \textbf{1.77$\pm$0.06} \\ 
\upshape{re}SGHMC$\times$P4 & 81.11$\pm$0.16   & 6915$\pm$40  & 2.92$\pm$0.12   & 82.72$\pm$0.08  & 6452$\pm$19 & 1.92$\pm$0.04  \\ 
ICSGHMC$\times$P4 & \textbf{81.51$\pm$0.18}  & \textbf{6630$\pm$38}   & 2.88$\pm$0.09  & \textbf{83.21$\pm$0.09}   & \textbf{6279$\pm$38} & 1.81$\pm$0.06  \\ 
\hline
\end{tabular}
\end{center}
\end{sc}
\vspace{-0.15in}
\end{table*}

Table \ref{UQ_test} shows that the vanilla ensemble results via M-SGD$\times$P4 surprisingly outperform cycSGHMC$\times$T4 and cycSWAG$\times$T4 on medium models, such as ResNet20 and ResNet32, and show very good performance on the out-of-distribution samples in Brier scores. We suspect that the parallel implementation  ($\times$P4) provides isolated initializations with less correlated samples; by contrast, cycSGHMC$\times$T4 and cycSWAG$\times$T4 explore the energy landscape contiguously, implying a risk to stay near the original region. reSGHMC$\times$P4 shows a remarkable performance overall, but demonstrates a large variance occasionally; this indicates the insufficiency of the swaps when multiple processes are included. When it comes to testing WRN-16-8, cycSWAG$\times$T4 shows a marvelous result and a large improvement compared to the other baselines. We conjecture that cycSWAG is more independent of hyperparameter tuning, thus leading to better performance in larger models. As to ICSGHMC$\times$P4, it consistently performs the best in prediction accuracy (ACC) and the negative log likelihood (NLL) and perform comparable to M-SGD$\times$P4 in Brier scores.

\section{Conclusion}
\vskip -0.05in
We have proposed the ICSGLD as an efficient algorithm for sampling from distributions with a complex energy landscape, and shown theoretically that ICSGLD is indeed more efficient than the single-chain CSGLD for a slowly decreasing step size. To our best knowledge, this is the first interacting importance sampling algorithm that adapts to big data problems without scalability concerns. ICSGLD has been compared with numerous state-of-the-art baselines for various tasks, whose remarkable results indicate its promising future in big data applications.

\chapter{ADAPTIVELY WEIGHTED STOCHASTIC GRADIENT LANGEVIN DYNAMICS FOR MONTE CARLO SIMULATION AND GLOBAL OPTIMIZATION}
\label{awsgld}
In this chapter, we propose an adaptively weighted stochastic gradient Langevin dynamic  (AWSGLD) algorithm  as a general-purpose scalable algorithm for both Monte Carlo simulation and global optimization. 
In particular, AWSGLD works under the framework of dynamic importance sampling; it simulates of a trial distribution  that is dynamically adjusted from a tempered target distribution in a similar way to the $1/k$-ensemble sampler but with the Langevin kernel. AWSGLD makes use of the tempering strategy to flatten the high energy region and adopts the histogram Monte Carlo strategy to protrude the low energy region (see Figure \ref{fig:aw_e} for illustration) to bias sampling toward low energy regions. As explained in Section \ref{selfadjusting}, AWSGLD possesses a self-adjusting mechanism for escaping from local traps, where the self-adjustment of the target distribution can also be equivalently viewed as a self-adjustment of the temperature and learning rate for the SGLD sampler. The self-adjusting mechanism enables the AWSGLD algorithm to be essentially immune to local traps. AWSGLD has been tested with deep learning benchmark examples and achieved promising results.

We note that some variants of SGLD, such as stochastic gradient Fisher scoring \cite{Ahn12}, stochastic gradient Riemannian Langevin dynamics \cite{GirolamiG2011, PattersonTeh2013}, preconditioned SGLD (pSGLD) \cite{Li16}, stochastic quasi-Newton Langevin Monte Carlo \cite{Simsekli2016}, stochastic gradient Hamiltonian Monte Carlo (SGHMC) \cite{Chen14},  stochastic gradient Nos\'e-Hoover thermostats  \cite{Ding14} and relativistic Monte Carlo  \cite{Lu17b}, have made use of the geometry information of the energy landscape or the momentums to improve the convergence of the simulation. Refer to \cite{yian2015} and \cite{NemethF2019} for the general formulation and overview of these variants. As discussed at the end of the chapter, many of these variants can be incorporated into AWSGLD to further enhance its performance in both optimization and Monte Carlo simulation.

The remaining part of the chapter is organized as follows. Section \ref{AlgSection} describes the proposed AWSGLD algorithm. Section \ref{Expsect1} illustrates the performance of AWSGLD on numerical examples, including an interesting mode exploration example on the MNIST dataset. Section \ref{conclusion} concludes the chapter with a brief discussion.

\section{An Adaptively Weighted Stochastic Gradient Langevin Dynamics Algorithm} \label{AlgSection}

Suppose that we are interested in sampling from a distribution with the density function given by 
\begin{equation*}
\pi_{T}(\bx) \propto \exp\left(-\frac{U(\bx)}{T}\right), \quad \bx \in \bchi,
\end{equation*}
where $T$ is the target temperature, and the energy function $U(\bx)$ is assumed to be non-convex. 
 As mentioned previously, AWSGLD has essentially incorporated all three strategies, tempering, histogram Monte Carlo and dynamic learning rates, such that it is immune to local traps in simulations of such a multi-modal distribution. The detail of the algorithm can be described as follows.

\subsection{Piece-wise Continuous Weights for Gradient-based Samplers}

We first consider a partition $\{\bchi_i\}_{i=1}^m$ of the sample space $\bchi$: $\bchi_i=\{\bx: u_{i-1}<U(\bx)\leq u_i\}$ for $i=1,2,\ldots,m$, where  $-\infty=u_0<u_1<\cdots<u_{m-1}<u_m=\infty$. For convenience, we assume $u_{i+1}-u_i=\Delta u$ for $i=1,2,\cdots, m-2$. In the sequel, we propose to simulate from a weighted tempering density function
\begin{equation*}
\varpi(\bx) \propto \frac{\pi_{\tau}(\bx)}{\Psi ^{\zeta}(U(\bx))},
\end{equation*}
where $\tau \geq T$ is a temperature, $\zeta>0$ is a hyperparameter to be specified by the user. Running the algorithm at a higher temperature enhances its ability in sample 
space exploration\footnote[4]{If the distribution of interest $\pi_T$ is flat, we can choose $\tau=T$ by default; otherwise, we run at a higher temperature $\tau>T$ for facilitating the exploration and a reweighting scheme will be applied later to recover the distribution $\pi_{T}$.} as explained later. 
If we assume the following conditions are satisfied:
\begin{equation}
\begin{split}
\label{ori_conditions}
     &\text{(i)    }\ \zeta=1 \text{ and }\Psi _{\btheta}(U(\bx))= \sum_{i=1}^m \theta(i) 1_{u_{i-1} < U(\bx) \leq u_i}, \\
      &\text{(ii)   } \theta(i)=\theta_{\star}(i), \text{where }\theta_{\star}(i) =\int_{\bx:U(\bx)\leq u_i}\pi_{\tau}(\bx)d\bx \quad \text{for } i\in\{1,2,\cdots, m\},\\
\end{split}
\end{equation}
running a Metropolis-Hasting algorithm \cite{Metropolis1953,Hastings1970} on $\varpi(\bx)$  will replicate the $1/k$-ensemble sampler \cite{Hesselbo1995MonteCS, Liang2004PRE,Liang05}, which is, however, unscalable with respect to big data. To achieve scalability, we propose to simulate of 
$\varpi(\bx)$ via the SGLD algorithm. Note that the piece-wise constant construction of $\Psi _{\btheta}(u)$ (given in (\ref{ori_conditions})-(i)) leads to the vanishing-gradient problem $\frac{\partial \log{\Psi _{\theta}}(u)}{\partial u}=0$ almost everywhere. To tackle this difficulty, we consider a piece-wise continuous construction of $\Psi _{\btheta}(u)$ such that 
  \begin{equation*}
\textcolor{black}{\Psi _{\btheta}(u)= \sum_{i=1}^m \left(\theta(i-1)e^{(\log\theta(i)-\log\theta(i-1)) \frac{u-u_{i-1}}{\Delta u}}\right) 1_{u_{i-1} u \leq u_i}},
 \end{equation*}
for which a straightforward calculation shows 
 \begin{equation*}
 \nabla_{\bx} \log \varpi(\bx) 
  =-\left[1+ \zeta \tau\frac{\partial \log \Psi (u)}{\partial u} \right] \nabla_{\bx} U(\bx)/\tau. 
  \end{equation*}
  
\subsection{Stochastic Approximation for Unknown Parameters}

Since  the parameter vector $\btheta_{\star}=(\theta_*(1),\ldots,\theta_*(m))$ is unknown, we propose to estimate it under the framework of stochastic approximation \cite{RobbinsM1951}. Let $\epsilon$ denote a positive constant and $\{\omega_k\}_{k=1}^{\infty}$ denote a positive, non-increasing sequence satisfying Assumption \ref{ass1} (given in the Appendix), which are called the learning rate for Monte Carlo simulations and the step size for parameter estimation, respectively. Let $\btheta_k=(\theta_k(1), \theta_k(2),\cdots, \theta_k(m))\in\bTheta$ denote a working estimate of $\btheta_*$ obtained at iteration $k$, where 
\begin{equation}\small
 \bTheta=\left\{\left(\theta(1),\theta(2),\cdots, \theta(m)\right)\mid 0<\theta(1)\leq \theta(2)\leq \cdots\leq \theta(m)=1\right\}.
\end{equation}
Let $J(\bx)$ denote the index of the subregion that $\bx$ 
 belongs to, i.e.
 \begin{equation}
\label{index_J}
    J(\bx)=\sum_{i=1}^m i 1_{u_{i-1}<U(\bx)\leq u_i}.
\end{equation}
Since a mini-batch of data is often used in simulations, we let $\tilde{J}(\bx)$ denote the index of the subregion determined by $\widetilde{U}(\bx)$. Despite the cheap cost in obtaining $\widetilde{U}(\bx)$ and the convenience in implementations, adopting such a stochastic energy estimator inevitably introduces a small bias to the estimate of $\btheta_*$. Fortunately, the bias depends on the variance of the noisy energy estimators as shown in Lemma \ref{bias_in_SA} and can be reduced by increasing the mini-batch size or using variance reduction \cite{deng_VR}. Moreover, it doesn't affect the convergence of the weighted averaging estimators. With the above notations, the AWSGLD algorithm is summarized in Algorithm \ref{alg:AWSGLD}. 
AWSGLD works with the Langevin kernel instead of the Metropolis kernel, which enables its scalability for big data problems.

\begin{algorithm}[tb]
   \caption{AWSGLD Algorithm}
   
   \label{alg:AWSGLD}
\begin{algorithmic}
   \STATE {\bfseries [1.] (Data subsampling)} Draw a subsample of size $n$ from the full dataset of size $N$; obtain the stochastic gradient $\nabla_{\bx}\widetilde U(\bx_k)$ and stochastic energy $\widetilde U(\bx_k)$. 

   \STATE {\bfseries [2.] (SGLD sampling)}
   Draw $\bx_{k+1}$ using SGLD based on the current sample $\bx_k$, i.e.,
  \begin{equation} \label{SGLDeq6_awsgld}
  \bx_{k+1}=\bx_k - \epsilon \frac{N}{n} \underbrace{\left[1+ 
   \frac{\zeta\tau}{\Delta u}  \left( \textcolor{black}{\log \theta_{k}(\tilde{J}(\bx_k))-\log\theta_{k}((\tilde{J}(\bx_k)-1)\vee 1)}\right) \right]}_{\text{Gradient multiplier}}  
    \nabla_{\bx} \widetilde U(\bx_k) +\sqrt{2 \tau \epsilon} \bm{e}_{k+1}, 
  \end{equation}
  where $a\vee b=\max(a,b)$, $\bm{e}_{k+1} \sim N(0,I_d)$, $d$ denotes the dimension of $\bx$,
  and $\tau$ is the temperature. 

  \STATE {\bfseries [3.] (Adaptive parameter updating)} Update 
  $\btheta_k=(\theta_k(1),\theta_k(2),\ldots, 
\theta_k(m))$ by setting
  \begin{equation} \label{updateeq_aw}
 {\theta}_{k+1}(i)={\theta}_{k}(i)+\omega_{k+1}{\theta}_{k}
 (\tilde{J}(\bx_{k+1}))\left(1_{i\geq \tilde{J}(\bx_{k+1})}-{\theta}_{k}(i)\right),  \quad i=1,2,\ldots,m,
 \end{equation} 
 where $1_{i\geq \tilde{J}(\bx_{k+1})}$ is equal to 1 if $i\geq \tilde{J}(\bx_{k+1})$ and 0 otherwise.

\end{algorithmic}
\end{algorithm}

\subsection{On the Self-adjusting Mechanism of AWSGLD} \label{selfadjusting} 

 \emph{AWSGLD overcomes the local trap issue by self-adapting the gradient multiplier, as defined in (\ref{SGLDeq6_awsgld}), according to the energy value of the current sample.}  
 In the high energy region, the gradient multiplier is close to 1, and AWSGLD performs like a high-temperature SGLD, which can have the sample space well explored. In the low energy region, the gradient multiplier takes a large value, which  helps AWSGLD to escape from local traps. More importantly, the gradient multiplier is self-adjusted. For example, if the sampler is trapped into a local minimum in subregion $s$, then, as implied by (\ref{updateeq_aw}),  the gradient multiplier tends to grow in the followed iterations by decreasing the values of  $(\theta_k(1),\theta_k(2),\ldots,\theta_k(s-1))$'s and increasing the values of $(\theta_k(s),\theta_k(s+1),\ldots,\theta_k(m))$'s until it escapes from the local trap. Therefore, AWSGLD is essentially immune to local traps.   

 This self-adjusting mechanism on the gradient multiplier can also be interpreted in terms of temperature and learning rates. Let $\psi_k$ denote the gradient multiplier used at iteration $k+1$, then (\ref{SGLDeq6_awsgld}) can be rewritten as follows: 
\begin{equation} \label{SGLDeq6_awsgldb}
\bx_{k+1}=\bx_k - (\epsilon \psi_k) \frac{N}{n}  
    \nabla_{\bx} \widetilde U(\bx_k) +\sqrt{2 (\tau/\psi_k) (\epsilon \psi_k)} \bm{e}_{k+1},
\end{equation} 
which corresponds to a SGLD iteration at the temperature $\tau/\psi_k$ and with the learning rate $\epsilon \psi_k$.  
Therefore, in the high energy region where $\psi_k\approx 1$, AWSGLD performs like a high-temperature SGLD; and in the low energy region  where $\psi_k\gg 1$,  AWSGLD performs like a low-temperature SGLD but with a large learning rate. The low-temperature makes the simulation more focused on local geometry exploitation and leads to the low-energy region protruding phenomenon as shown in Figure \ref{fig:aw_e}(c), while the large learning rate helps the simulation to escape from local traps as shown by cyclical SGMCMC algorithms \cite{ruqi2020}.
We note that the learning rate updating mechanism of AWSGLD is fundamentally different from that of cyclical SGMCMC.  In AWSGLD, \emph{the learning rate is self-adjusted according to the energy value of the current sample}; while in cyclical SGMCMC it is updated according to a fixed schedule of iterations.  As shown by our numerical examples, see Figure \ref{many_traps}, AWSGLD can be much more efficient than cyclical SGMCMC.

Similar to the convergence established in Section \ref{convergSect} and \ref{all_proof_pop_CSGLD}, AWSGLD falls into the class of dynamic importance SGMCMC algorithms; it converges to a desired trial distribution as the number of iterations becomes large. We note that AWSGLD is not the first dynamic importance stochastic gradient sampler. In spirit, AWSGLD is similar to contour SGLD (CSGLD) \cite{CSGLD}. However, they are still significantly different in some aspects:  CSGLD is a flat histogram algorithm without the tempering strategy used, and it tends to sample uniformly in the space of energy. On the contrary, AWSGLD works with a tempered target distribution to enhance its ability in sample space exploration, and biases sampling toward low energy regions by working with a non-flat histogram function.

\section{Numerical Experiments} \label{Expsect1}

\subsection{CDF Estimation}

The first experiment is to test the convergence of AWSGLD: It is applied to simulate from the standard Gaussian distribution, i.e., $\pi(\bx)=N(0,1)$. In the simulation, we set the temperature $\tau=1$, $\zeta=1$, the total number of iterations $=1\times 10^7$, and the step size $\omega_k=\frac{0.02}{k^{0.6}+100}$. The sample space is partitioned into 1000 subregions according to the energy function with a bandwidth of $\Delta u=0.01$. To mimic the scenario of simulating with stochastic gradients, a random noise generated from $N(0,0.1^2)$ is added to the gradient at each iteration, while the partition index is determined according to the true energy function. The experiments are repeated for 10 times with the results reported in Figure \ref{cdf_in_energy}. Figure \ref{cdf_in_energy}(a) depicts the convergence of $\btheta_k$ to $\btheta_*$, the CDF of the standard Gaussian distribution (by noting $\tau=1$); and it also compares the empirical CDFs of the energy values sampled by SGLD (red curve) and AWSGLD (green curve). The deviation between the two empirical CDFs shows that AWSGLD biases sampling toward low energy regions, which implies a great potential of AWSGLD in stochastic optimization.
 
\begin{figure}[htbp]
 \begin{tabular}{ccc}
   (a) &(b) & (c) \\ 
  \includegraphics[height=1.8in,width=2in]{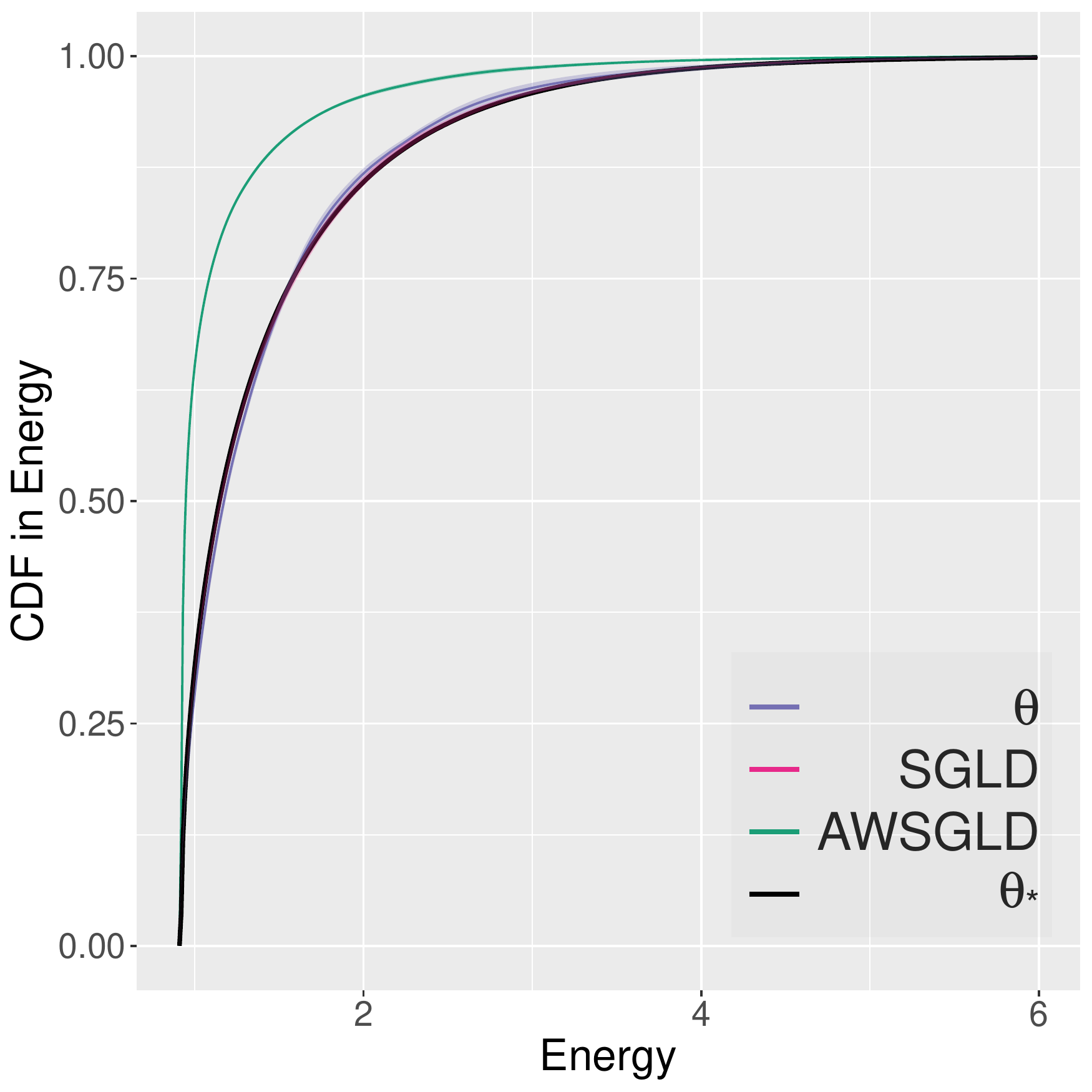} & 
  \includegraphics[height=1.8in,width=2in]{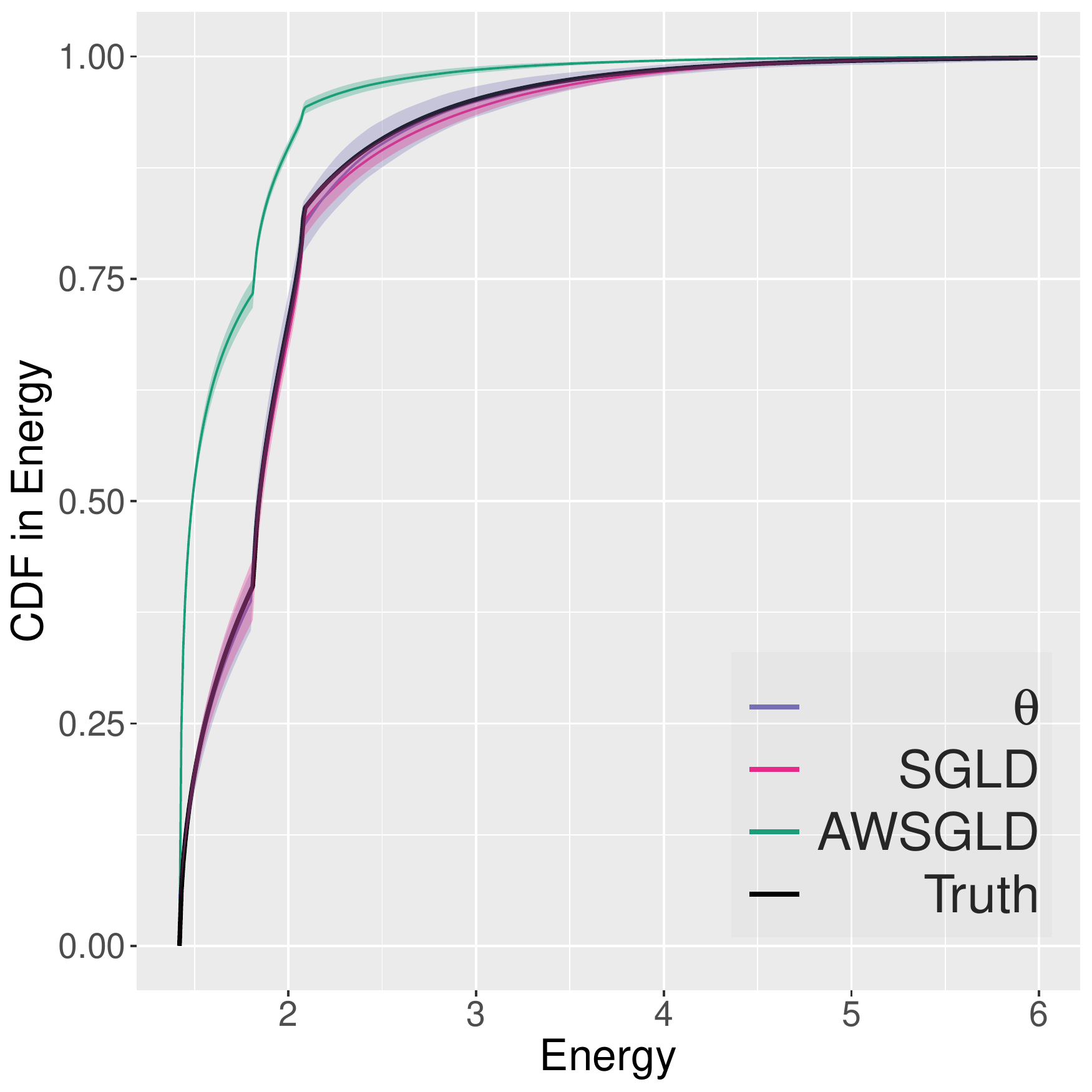}  &
  \includegraphics[height=1.8in,width=2in]{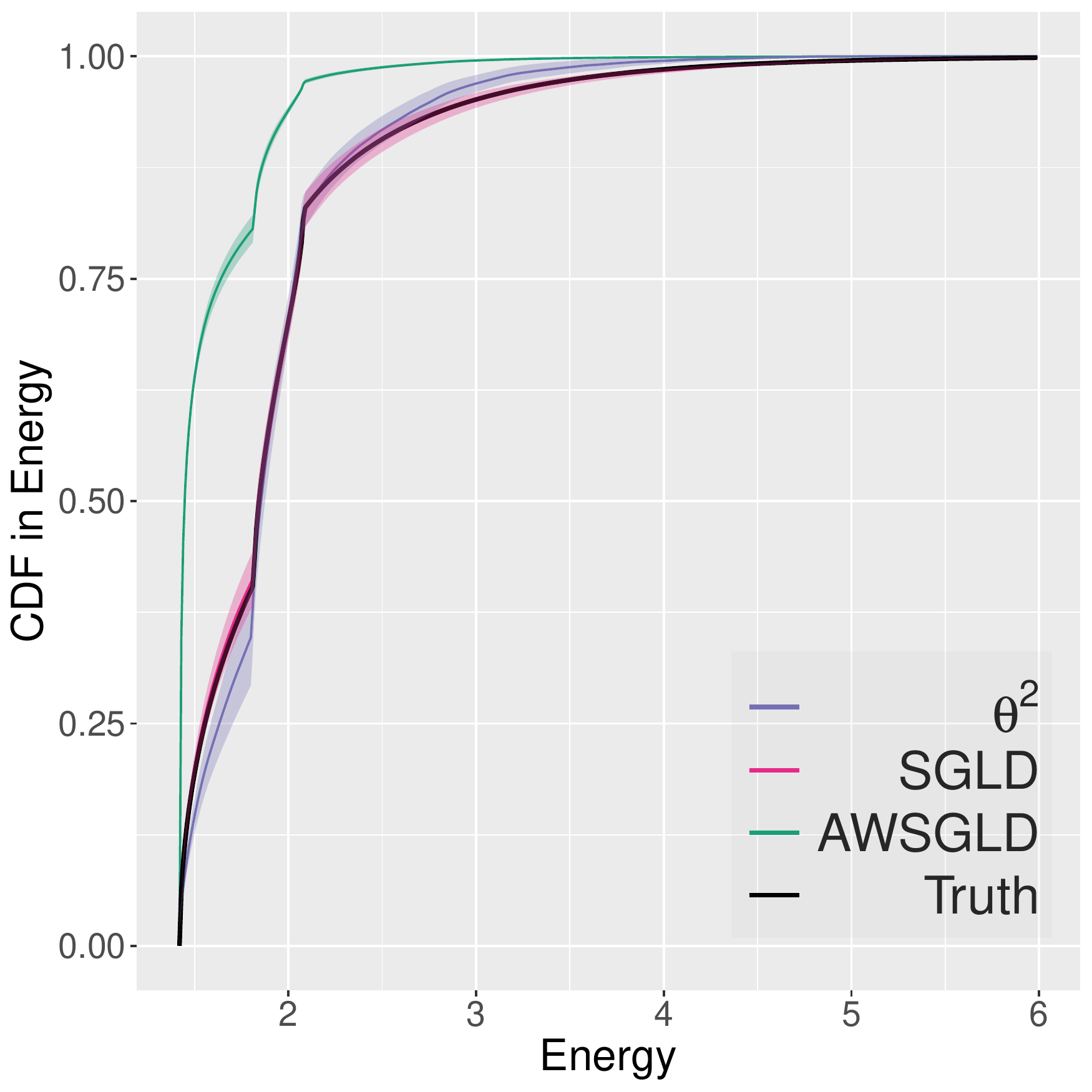} 
\end{tabular}
  \caption{Results of AWSGLD for (a) standard Gaussian with $\zeta=1$; 
  (b) mixture Gaussian with $\zeta=1$; and (c) mixture Gaussian with $\zeta=2$, 
   where the black line is the true CDF; the purple line is the AWSGLD estimate of $\btheta^{\zeta}$;
   the green line is the empirical CDF of AWSGLD samples; and the red
    line is for the empirical CDF of SGLD samples.}
\label{cdf_in_energy}
\end{figure}

We have also tested AWSGLD on a mixture Gaussian distribution $\frac{2}{5} N(-2, 1) + \frac{3}{5} N(1, 1)$, which represents a multimodal distribution. Figure \ref{cdf_in_energy} (b) and (c) summarizes the simulation results of AWSGLD with $\zeta=1$ and $\zeta=2$, respectively. They show that AWSGLD also works very well for the multimodal distribution, providing a good estimate of the energy CDF and biasing samples toward low energy regions. Admittedly, $\zeta=2$ is more effective than $\zeta=1$ in biasing samples toward low energy regions.

\subsection{Illustrative Examples}

\subsubsection{Example 1: Shallow Local Traps}  In stochastic  optimization, raising temperature is a useful strategy for escaping from local traps; however, it also \emph{brings a challenge for efficiently locating the global optimum}. For example, consider the distribution $\pi_{\tau}(\bx)\propto \exp(-U_d(\bx)/\tau)$, where
\begin{equation} \label{toyenergy}
 U_d(\bx)=0.05\|\bx\|^2-0.01\|\bx\|^2 \times \sum_{i=1}^d\cos(2x_i),
 \end{equation}
 and $\bx=(x_1, ..., x_d)$. Fig.\ref{fig:aw_e}(a) shows $U_d(\bx)$ for a 2-dimensional case, 
  which contains many local minima in the high energy region.  
 Fig.\ref{fig:aw_e}(b) shows that the energy landscape can be flattened at a high temperature, for which SGLD is approximately equivalent to a random walk in the flat basin and 
 the global minimum is hard to be located.

 \begin{figure}[htbp]        
\begin{center}
 \begin{tabular}{ccc}
 (a) $U_d(\btheta)$ & (b) $U_d(\btheta)/9$ & (c) $-\log \textcolor{black}{\varpi_{\widehat{\Psi}_{\btheta_{\star}}}}(\bx)$ ($\tau=9$)  \\ 
  \includegraphics[scale=0.14]{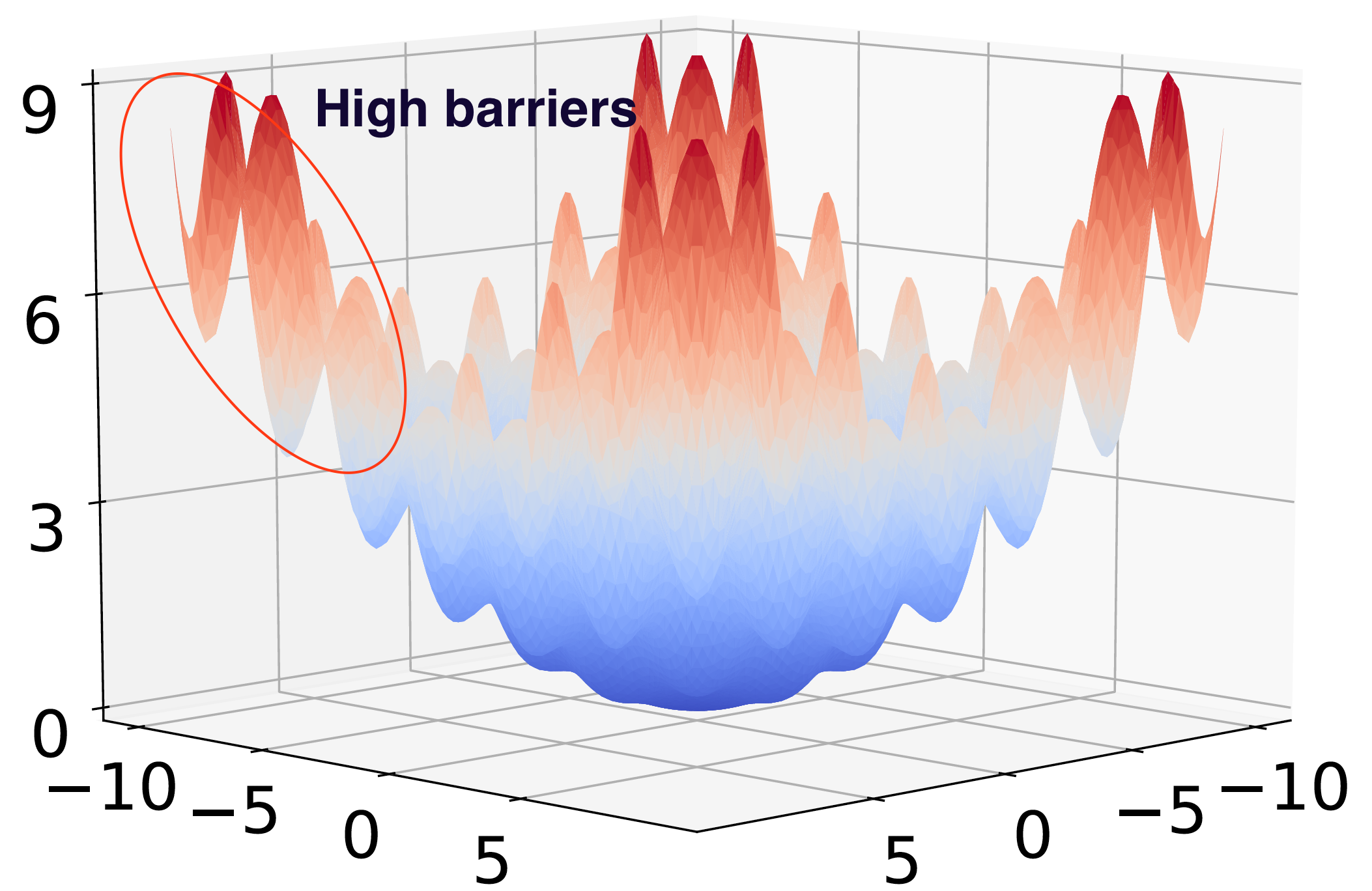} & 
  \includegraphics[scale=0.14]{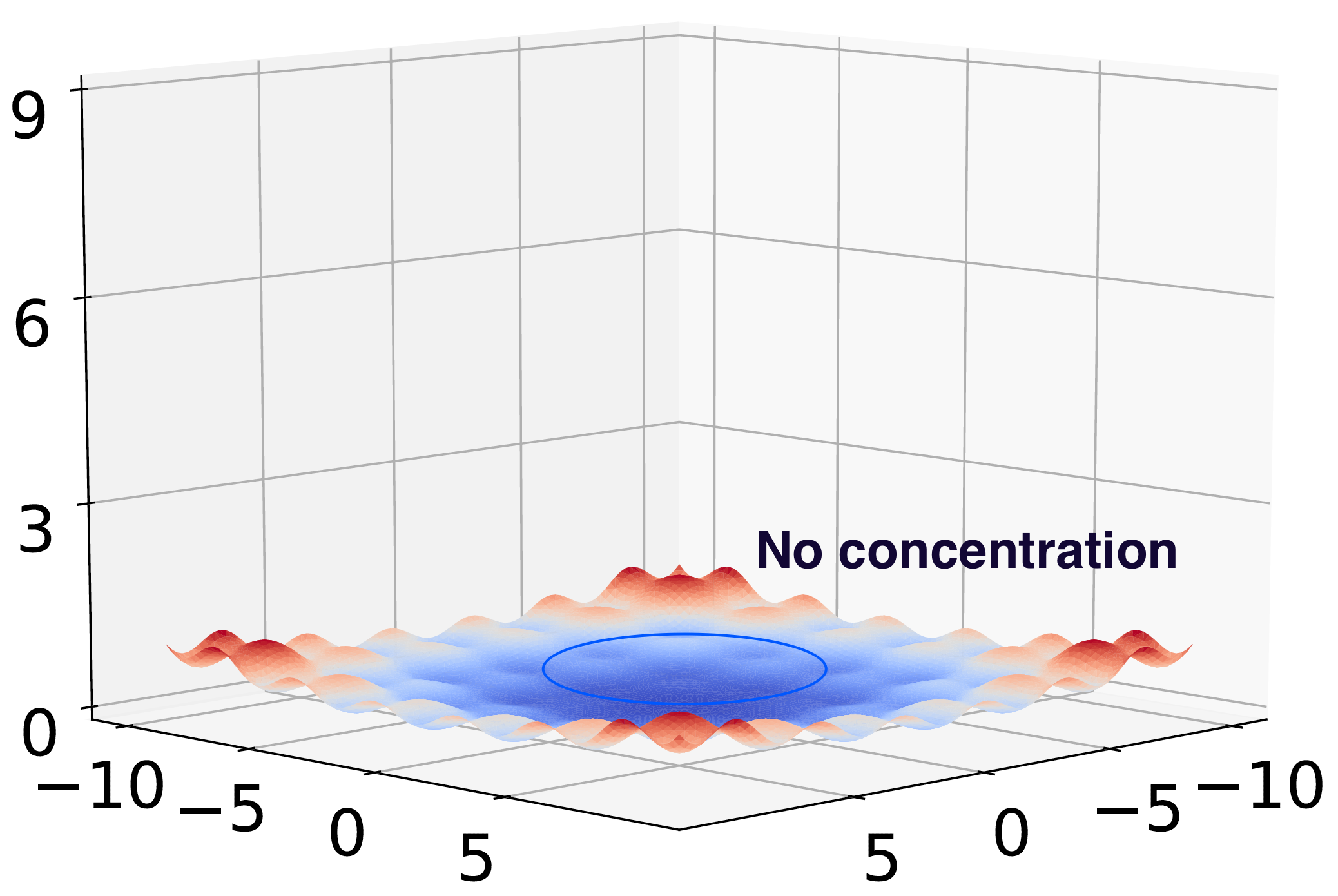} & 
  \includegraphics[scale=0.14]{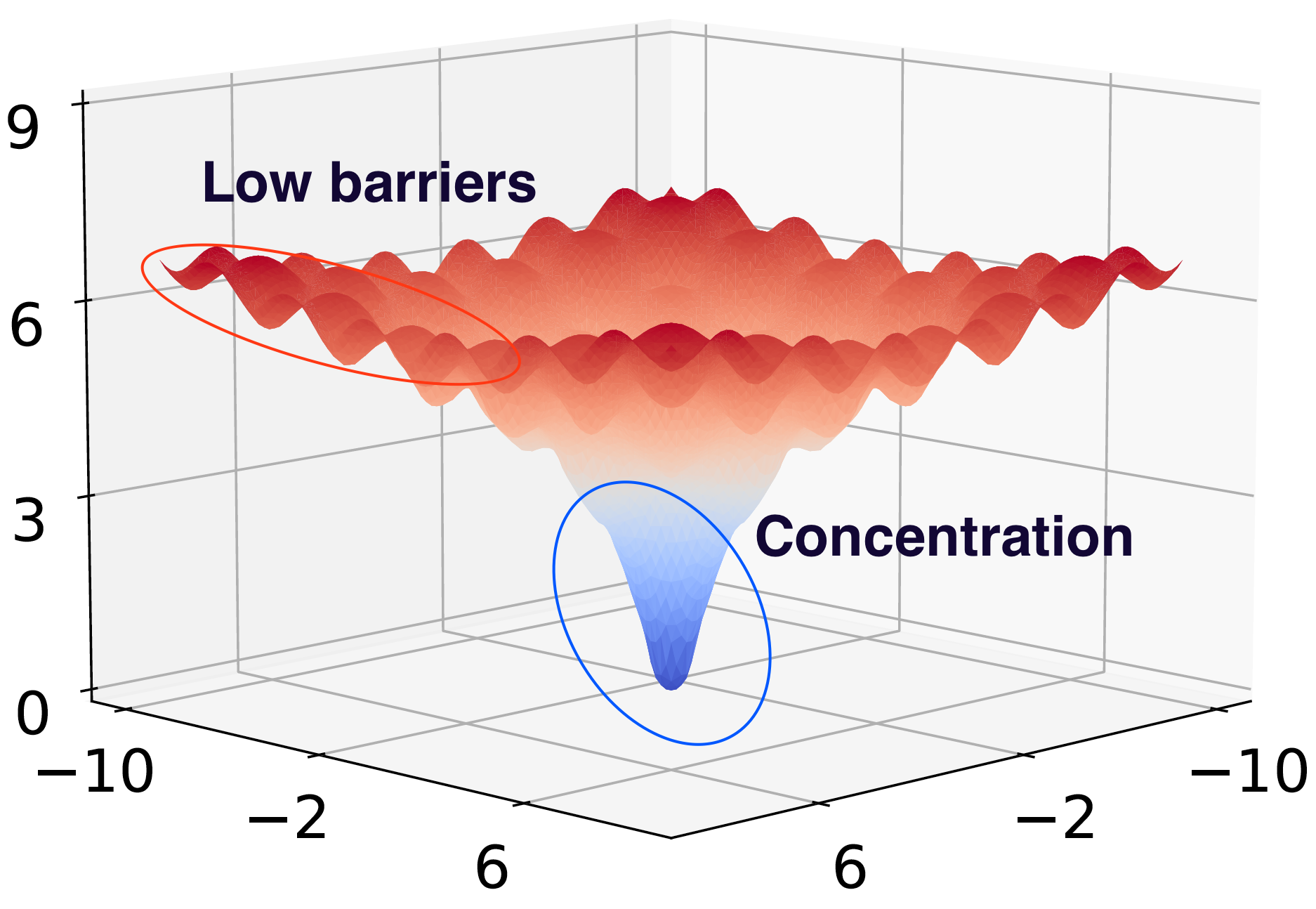} \\
 \end{tabular}
 \end{center}
  \caption{Energy landscapes of different distributions: 
  (a) original distribution $\pi_1(\bx) \propto \exp(-U_d(\bx))$; 
  (b) tempered distribution 
  $\pi_{9}(\bx) \propto \exp(-U_d(\bx)/9)$; 
  (b) $\textcolor{black}{\varpi_{\widehat{\Psi}_{\btheta_{\star}}}(\bx)}$
  with $\tau=9$ and \textcolor{black}{$\zeta=1$},
  where all energy functions
   have been shifted to have a minimum 0 for 
  the purpose of illustration.}
  \label{fig:aw_e}
\end{figure}

AWSGLD tackles this issue by employing an importance sampling strategy, which self-adjusts the tempered target distribution to facilitate global optimization by \emph{flattening the high energy regions} and \emph{over-weighting the low energy regions}. 
Fig.\ref{fig:aw_e}(c) shows the energy landscape of the weighted density $\varpi_{\widehat{\Psi}_{\btheta_{\star}}}(\bx)$ as defined in Section \ref{FMalg2}, which  flattens the local minima in high energy regions while protruding the region of the global energy minimum. In addition, it is 
 easy to see that 
  Fig.\ref{fig:aw_e}(c) has \emph{a  smaller depth of energy barriers} compared to Fig.\ref{fig:aw_e}(a), and \emph{a much larger probability mass at the global optimum} compared to Fig.\ref{fig:aw_e}(b). This implies an improved efficiency of AWSGLD 
   for simulation and global optimization.  
   
As suggested in \cite{Yuchen17} and \cite{Maxim17}, SGLD generally requires to be run at a low temperature to attain its optimal performance in global optimization. Unfortunately, at the low temperature, SGLD tends to be locally trapped on a complex energy landscape. 
  In contrast, AWSGLD prefers to be run under a 
  high temperature, which enhances its ability in sample space exploration, while its dynamic importance sampling mechanism enables it to bias sampling toward low energy regions even when running at a high temperature.

\subsubsection{Example 2: Sample Space Exploration} This example compares the trajectories of AWSGLD, pSGLD, SGHMC, cycSGLD, and high temperature SGLD on a rugged energy landscape with many local energy minima. Note that pSGLD, SGHMC, cycSGLD and high temperature SGLD represent the state-of-the-art SGMCMC algorithms in different categories. The target density function of the example is given by  $\pi_{\tau}(\bx)\propto \exp(-{U(\bx)}/{\tau})$, 
where $U(\bx)$ is given by 
\begin{equation*}
\scriptsize
  \begin{split}
      U(\bx) &=-\{x_1 \sin(20x_2)+x_2\sin (20x_1)\}^2\cosh(\sin(10x_1)x_1) -\{x_1 \cos(10x_2)-x_2\sin(10x_1)\}^2\cosh(\cos(20x_2)x_2).
  \end{split}
\end{equation*}
 This example has been used in multiple papers, e.g., \cite{Robert04} and \cite{Liang07}, as a benchmark example for multimodal sampling.  Fig.\ref{many_traps}(a) shows the contour of the energy function. \begin{figure}[!ht]
  \center
  \includegraphics[scale=0.4]{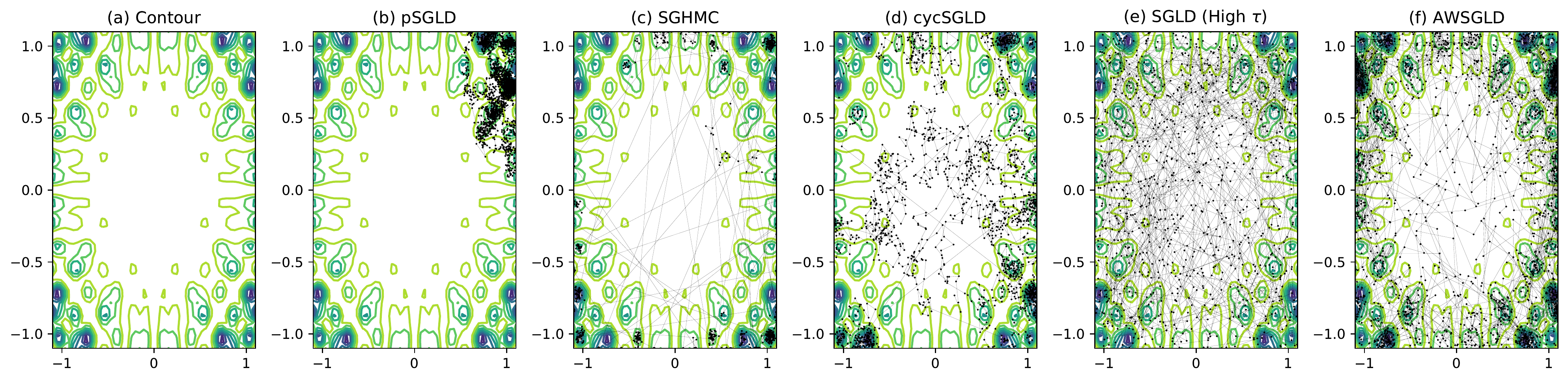}
  \caption{Sample paths of standard sampling algorithms on a rugged energy landscape.}
  \label{many_traps}
\end{figure}
Fig.\ref{many_traps}(b)-(f) show the sample paths produced by different algorithms, whose detailed settings are given in Appendix \ref{sample_space_exploration}. In particular, we note that pSGLD, SGHMC and cycSGLD are run with $\tau=2$, while AWSGLD and the high temperature SGLD are run with $\tau=20$. The sample paths indicate a trade-off in SGMCMC simulations: one  
can employ a low temperature for local geometry exploitation or a high temperature for sample space exploration. At the low temperature, pSGLD exploits well the local geometry of the energy landscape, but fails to traverse the whole sample space; SGHMC and cycSGLD explores more modes, but are still not sufficient enough. At the high temperature, SGLD moves almost uniformly in the sample space, but fails to exploit the global optima. In contrast, AWSGLD with the same high temperature $\tau=20$ can traverse freely over the rugged energy landscape, while having the local geometry well exploited. As explained previously,  in the high energy region, AWSGLD performs like a high-temperature SGLD for sample space exploration; in the low energy region, it performs like a low-temperature SGLD but with a large learning rate. The low-temperature simulation enforces the sampler to exploit the local geometry of the energy landscape, while the large learning rate helps it to escape from local traps. The self-adjusting mechanism on the temperature and learning rate makes AWSGLD essentially immune to local traps.

\subsubsection{Example 3: Non-Convex Optimization}

We apply AWSGLD to 10 benchmark non-convex optimization problems \cite{Laguna2005}. The results are reported in Table \ref{nonconvex_funcs} with the experimental settings given in Appendix \ref{10func}. For each energy function, we report the results of SGLD at two temperatures, where the higher one is also the one employed by AWSGLD.
Each run is terminated when a sample is drawn from the optimal set $\mathcal{S}=\{\bx:U(\bx)\leq u_{\min}+\varrho\}$, where the global minimum value $u_{\min}$ is known for 
each of the energy functions, and the accuracy error 
$\varrho$ together with the settings of other 
parameters are given in Appendix \ref{10func}. 

\begin{table}[!ht]
\caption{Average iteration numbers (in a scale of 1000) required by SGLD and AWSGLD 
 to hit the optimal set for 10 multimodal functions, where AW-10 and AW-100 is short 
 for AWSGLD with a partition of 10 subregions and 100 subregions, respectively; 
 $\infty$ means the number of iterations more than $100,000$.}\label{nonconvex_funcs}
 \vspace{1em}
 \small
 \begin{center}
\begin{tabular}{llccccccccc}
\hline
No & Function $\backslash$ Iters $\times 10^{3}$ & Dim & SGLD (Low $\tau$) & SGLD (High $\tau$) & AW-10 & AW-100 \\
\hline
1 & Rastrigin ($RT_{20}$) & 20 & 4.3$\pm$0.3 &  $\infty$ & 5.1$\pm$0.5 & \textcolor{blue}{3.6$\pm$0.2}  \\
2 & Griewank ($G_{20}$) & 20  & 32.6$\pm$0.9 &  $\infty$ & \textcolor{blue}{1.8$\pm$0.05} & \textcolor{blue}{1.7$\pm$0.02}   \\
3 & Sum Squares ($SS_{20}$) & 20 & 0.07$\pm$0.01 &  0.07$\pm$0.01 & {0.07$\pm$0.01} & 
 {{0.06$\pm$0.01}}  \\
4 & Rosenbrock ($R_{20}$) & 20 & 29.2$\pm$5.7  &  93.2$\pm$10.6 & \textbf{\textcolor{blue}{0.15$\pm$0.01}} & \textcolor{blue}{0.7$\pm$0.1}  \\
5 & Zakharov ($Z_{20}$) & 20  &  20.3$\pm$6.2 & 6.9$\pm$1.9  & {\textcolor{blue}{1.4$\pm$0.7}} & {\textcolor{blue}{0.9$\pm$0.1}}  \\
6 & Powell ($PW_{24}$) & 24 & 10.0$\pm$0.6  &  $\infty$ & \textcolor{blue}{0.3$\pm$0.01} & {\textcolor{blue}{0.2$\pm$0.01}}   \\
7 & Dixon$\&$Price ($DP_{25}$) & 25 & 47.9$\pm$3.6  &  87.5$\pm$6.9 & \textbf{\textcolor{blue}{0.2$\pm$0.02}} & {\textcolor{blue}{0.3$\pm$0.01}}   \\
8 & Levy ($L_{30}$) & 30 & $\infty$  &  22.5$\pm$3.9 & {\textcolor{blue}{0.1$\pm$0.05}} & \textcolor{blue}{0.3$\pm$0.01}   \\
 9 & Sphere ($SR_{30}$) & 30 & 0.3$\pm$0.01  &  0.4$\pm$0.01 & {{0.4$\pm$0.03}}  & {{0.4$\pm$0.03}}  \\
 10 & Ackley ($AK_{30}$) & 30 & 66.2$\pm$12.1  &  $\infty$  & \textbf{\textcolor{blue}{2.7$\pm$0.3}} & {\textcolor{blue}{4.3$\pm$0.2}} \\
\hline
\end{tabular}
\end{center}
\end{table}

For each energy function, each algorithm is run for 10 times. Table \ref{nonconvex_funcs} reports the average hitting time, i.e., the average iteration number, required by each algorithm to hit the optimal set. The comparison shows that for the optimization purpose, SGLD is in general more efficient with a low temperature than with a high one, while AWSGLD significantly outperforms both the low-temperature SGLD and high-temperature SGLD for most of the energy functions. \emph{For certain energy functions, such as $R_{20}$, $DP_{25}$, $L_{30}$ and $AK_{30}$, the accelerations created by AWSGLD are almost 100 times!} Table \ref{nonconvex_funcs} also shows that the sample space partition might affect the performance of AWSGLD, although not very sensitive.

Regarding CPU cost, we note that for most problems that a large dataset is involved, AWSGLD costs about the same CPU time as SGLD per iteration, since the CPU time spent on the extra parameter updating step is really minor compared to that on gradient evaluation. For this example, since no data is involved, the CPU time cost by AWSGLD (with a partition of 100 subregions) per iteration is only about 5\% longer than that by SGLD.  As such, the speed-ups in CPU time by AWSGLD are close to those in iterations.

\subsection{Mode Exploration on MNIST}

The MNIST dataset \cite{lecun1998gradient} is a benchmark data set in computer vision, which contains 60,000 images of handwritten digits (from 0 to 9) for training and 10,000 images for test. We use this example to demonstrate that AWSGLD maintains both Monte Carlo simulation and optimization abilities, while the standard methods, such as stochastic gradient descent (SGD) and SGLD, \emph{get stuck in local modes} and lead to unreliable uncertainty estimation. For simplicity, we consider only a subset of the MNIST data, which consists of five classes with digits from 0 to 4. Following \cite{Jarrett09}, we adopt a standard convolutional neural network (CNN) with two convolutional layers and two fully-connected (FC) layers. The two convolutional 
layers have 32 and 64 feature maps, respectively. The two FC layers have 50 hidden nodes and 5 outputs, respectively. We choose a large batch size of 2500 to reduce the variance of the stochastic gradient. We include standard baselines, such as pSGLD \cite{Li16}, SGHMC \cite{Chen14}, and cycSGLD \cite{ruqi2020}, where the temperature is set to $0.3$ \footnote[2]{The wide adoption of data augmentation, such as random flipping and random cropping, in DNN training implicitly includes more data and leads to a concentrated posterior. See \cite{Florian2020, Aitchison2021} for details.}. Since AWSGLD  prefers to be run at a higher  temperature, we run AWSGLD and a high-temperature SGLD at $\tau=1$.  The model is trained for 500 epochs.

\begin{figure}[!ht]
\small
 \begin{tabular}{cccc}
{(a) Training Loss} &  {(b) Gradient multiplier} & (c) {Test accuracy} & {(d) Uncertainty evaluation} \\ 
  \includegraphics[height=1.4in,width=1.4in]{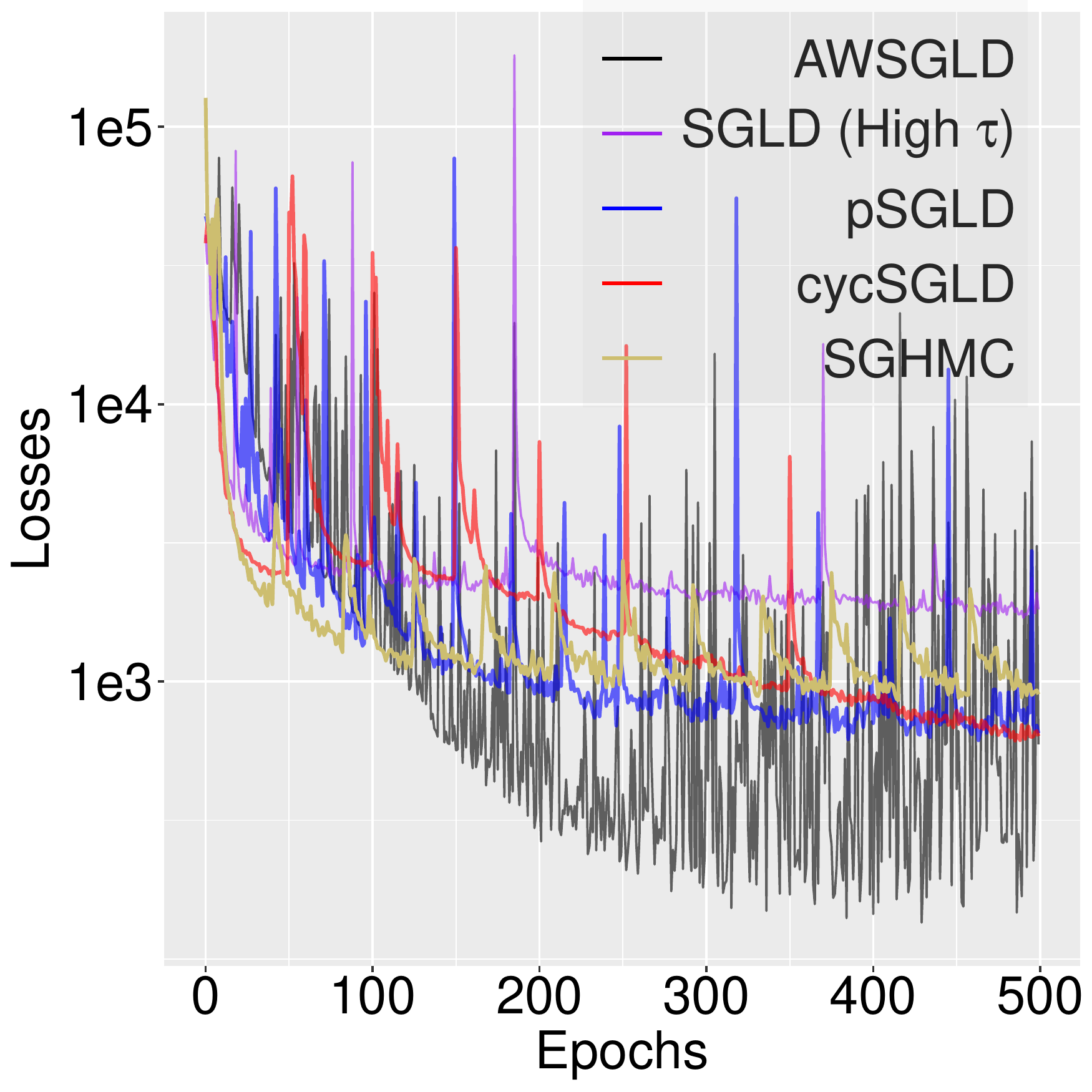} & 
  \includegraphics[height=1.4in,width=1.4in]{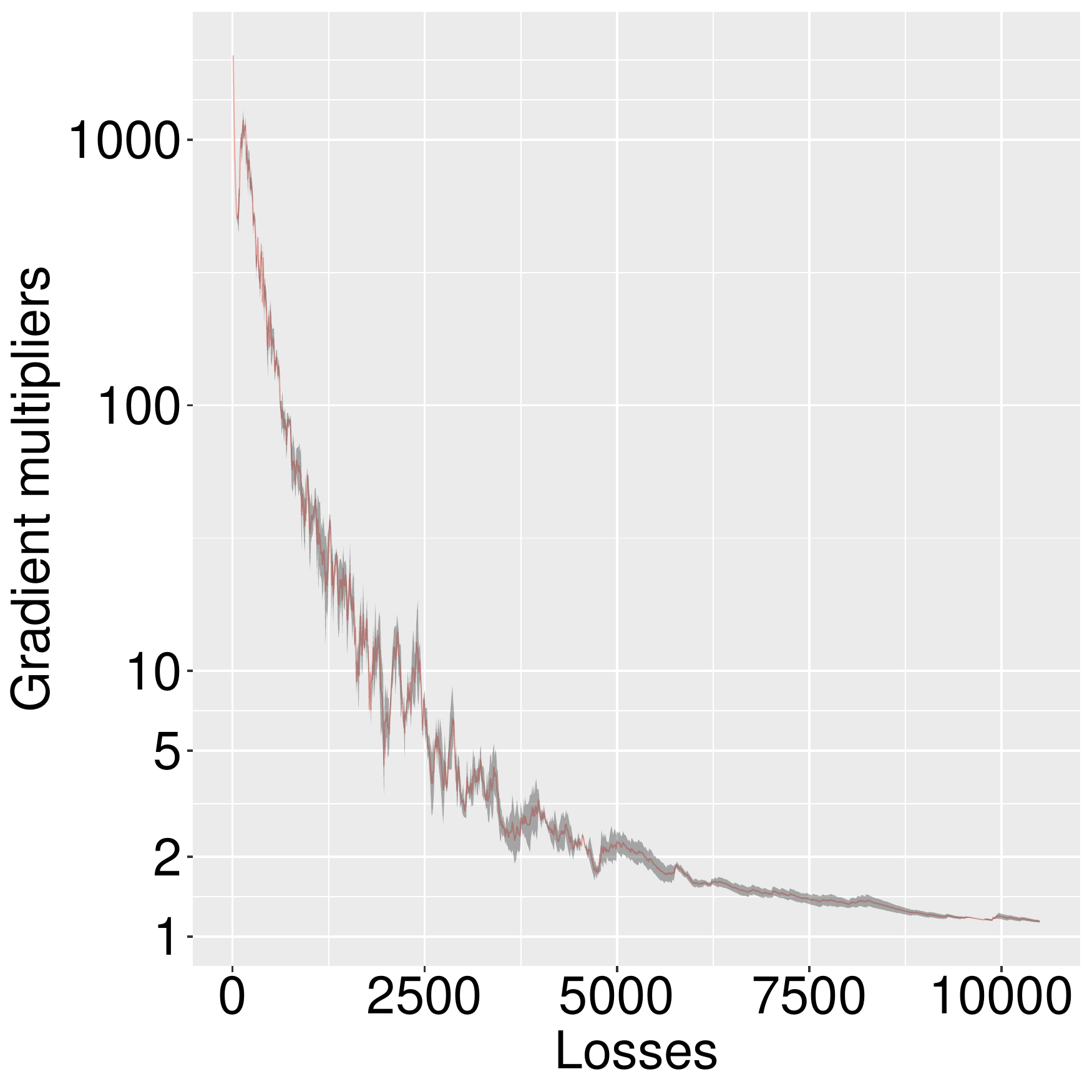} & 
  \includegraphics[height=1.4in,width=1.4in]{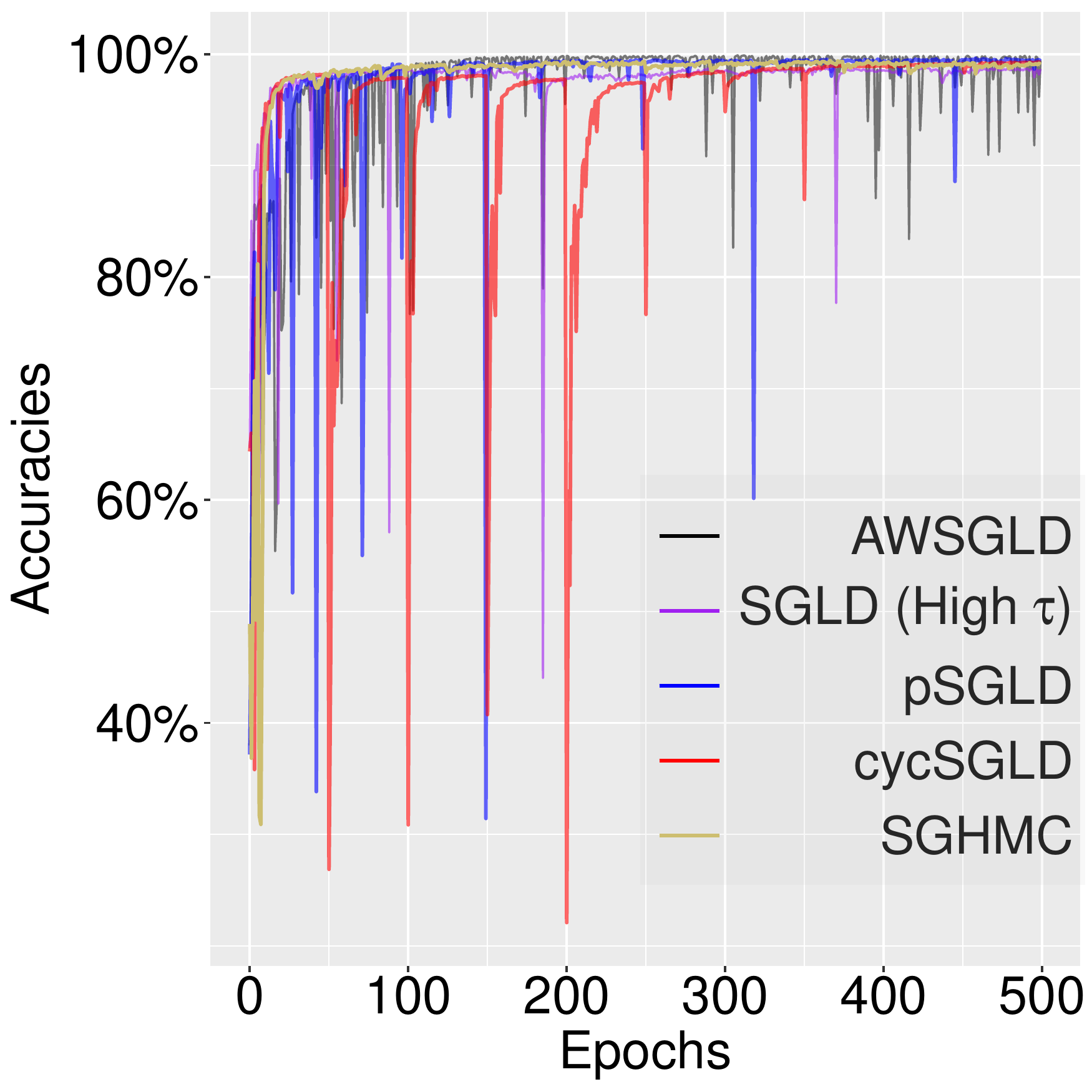} &
  \includegraphics[height=1.4in,width=1.4in]{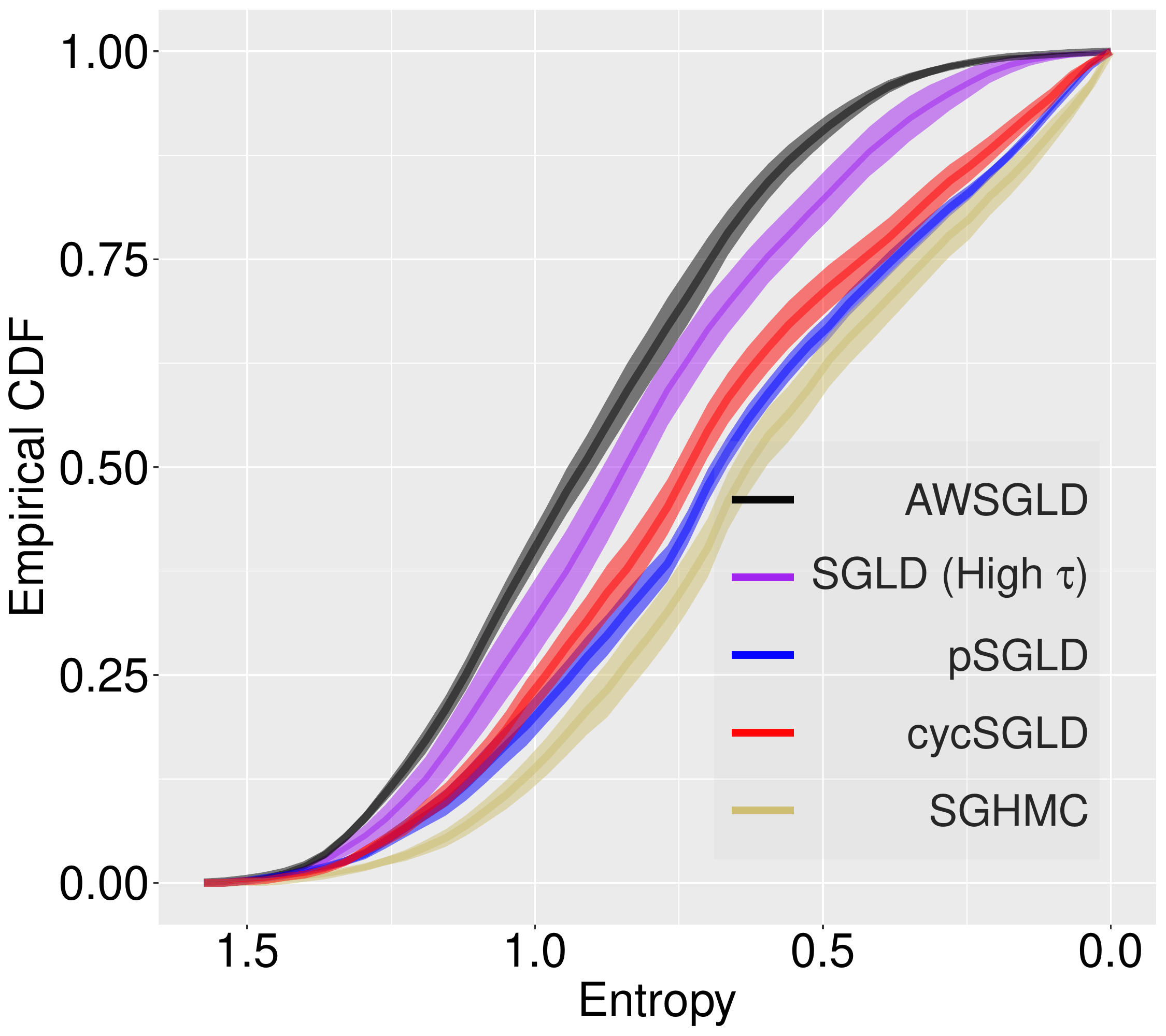} 
\end{tabular}
  \caption{Energy space exploration on a subset of  MNIST data (digits 0-4) by AWSGLD, high-$\tau$ SGLD, pSGLD, SGHMC, and cycSGLD: (a) paths of training loss; (b) gradient multiplier vs. training loss; (c) paths of test accuracy; 
  (d) empirical CDFs of Shannon entropy evaluated on the out-of-distribution images (digits 5-9).} 
\label{Uncertainty_estimation_mnist}
\end{figure}

Fig.\ref{Uncertainty_estimation_mnist}(a) shows the paths of the training loss produced by the five algorithms. The high-temperature SGLD consistently produced higher training loss, which suggests that the algorithm made a decent sample space exploration but an insufficient local geometry exploitation. SGHMC, pSGLD, and cycSGLD occasionally traverse between high-energy regions and low-energy regions, and achieve a better exploitation but a limited exploration. On the contrary, {\it  AWSGLD balances the sample space exploration and local geometry exploitation in a single run, which strikingly traverses freely between high and low energy regions}. 

Fig.\ref{Uncertainty_estimation_mnist}(b) illustrates the self-adjusting mechanism of AWSGLD: {\it it automatically adjusts the gradient multiplier to a large value at low energy regions and a small  value at high energy regions, which explains why AWSGLD is immune to local traps}. For this example, the gradient multiplier can be as large as 1000 in the extremely low energy regions.
 
For a comprehensive evaluation of the five algorithms, we report in Fig.\ref{Uncertainty_estimation_mnist}(c) the paths of test accuracy produced by them, where the prediction accuracy is calculated on the test images with digits from 0 to 4; and we report in Fig.\ref{Uncertainty_estimation_mnist}(d) the empirical CDF of Shannon entropy produced by them, where the Shannon entropy for each of the sampled models is calculated on the rest images of the dataset (digits from 5 to 9). The Shannon entropy provides a good measure for the uncertainty of the models on out-of-distribution detection, and a well-trained model is expected to exhibit high Shannon entropy. In summary, Fig.\ref{Uncertainty_estimation_mnist} shows that AWSGLD is the only algorithm among the five that produces models with low training loss as well as high Shannon entropy. By contrast, the high-temperature SGLD produces models with high training loss and high Shannon entropy; and pSGLD, cycSGLD and SGHMC produce models with low training loss and low Shannon entropy.
 
\begin{figure}[!ht]
\small
 \begin{tabular}{cccccc}
(a) SGD & (b) pSGLD & (c) SGHMC & (c) high-$\tau$ SGLD & (d) AWSGLD  \\ 
\includegraphics[height=1.1in,width=1.1in]{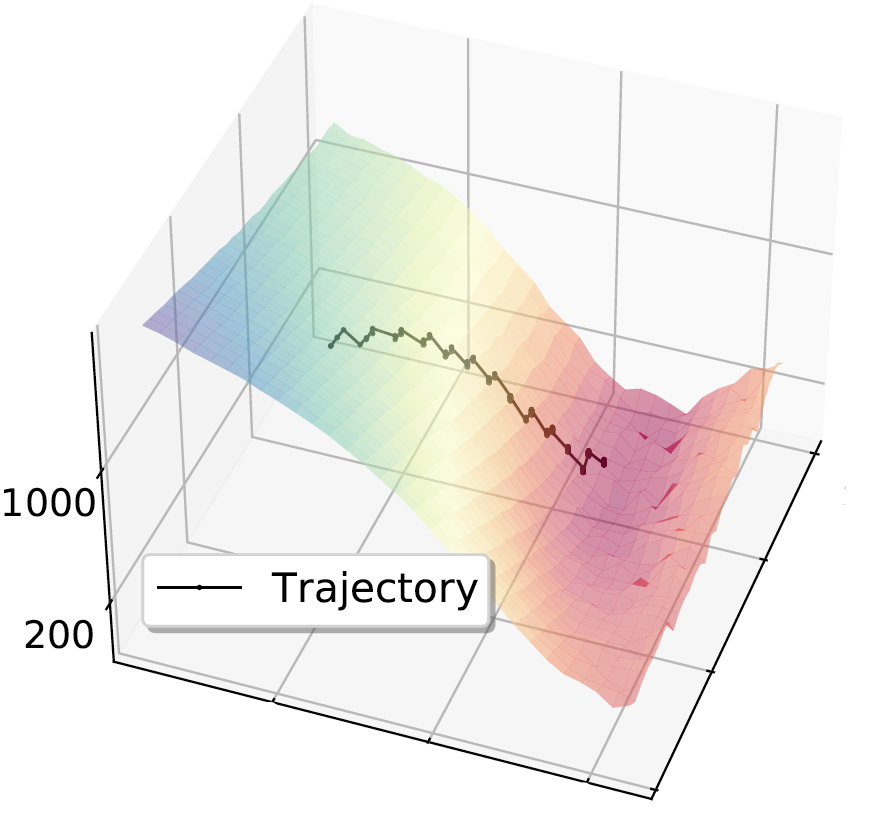} &
\includegraphics[height=1.1in,width=1.1in]{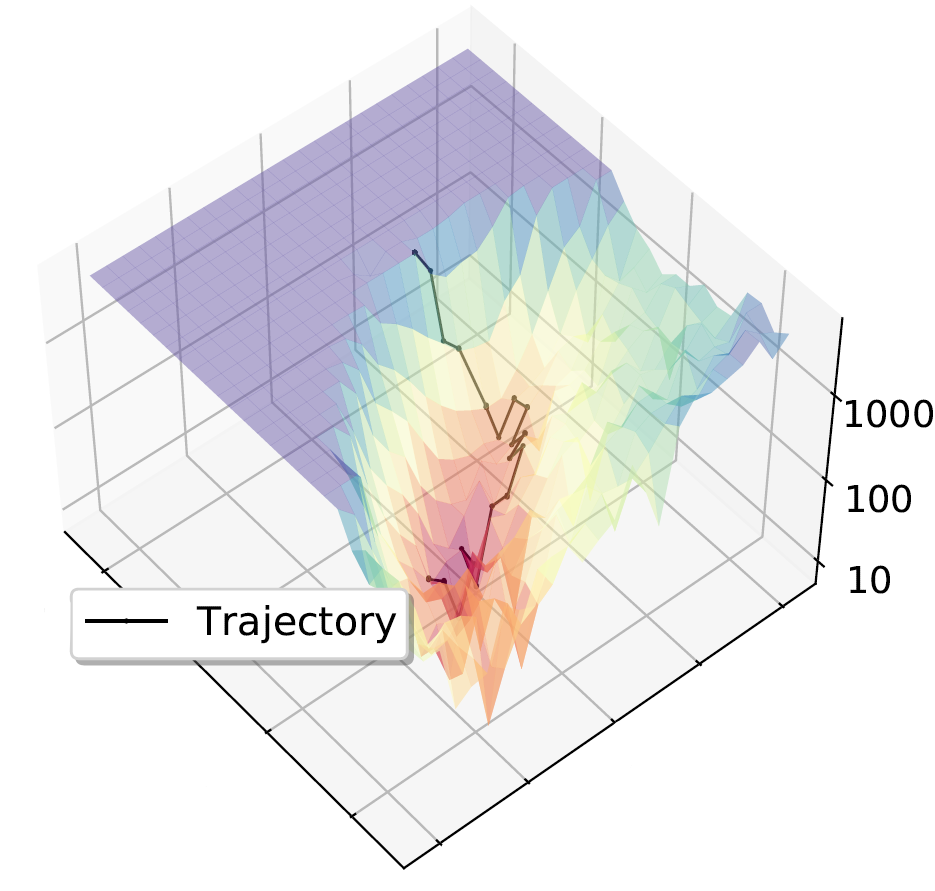} &
\includegraphics[height=1.1in,width=1.1in]{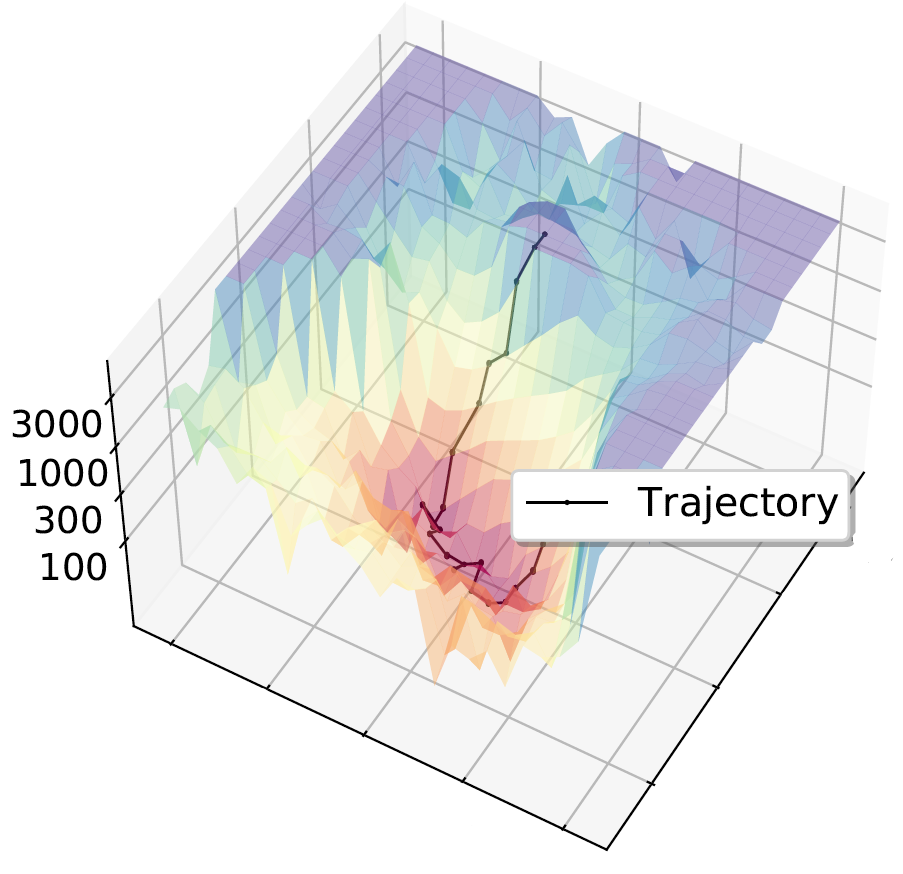} &
\includegraphics[height=1.1in,width=1.1in]{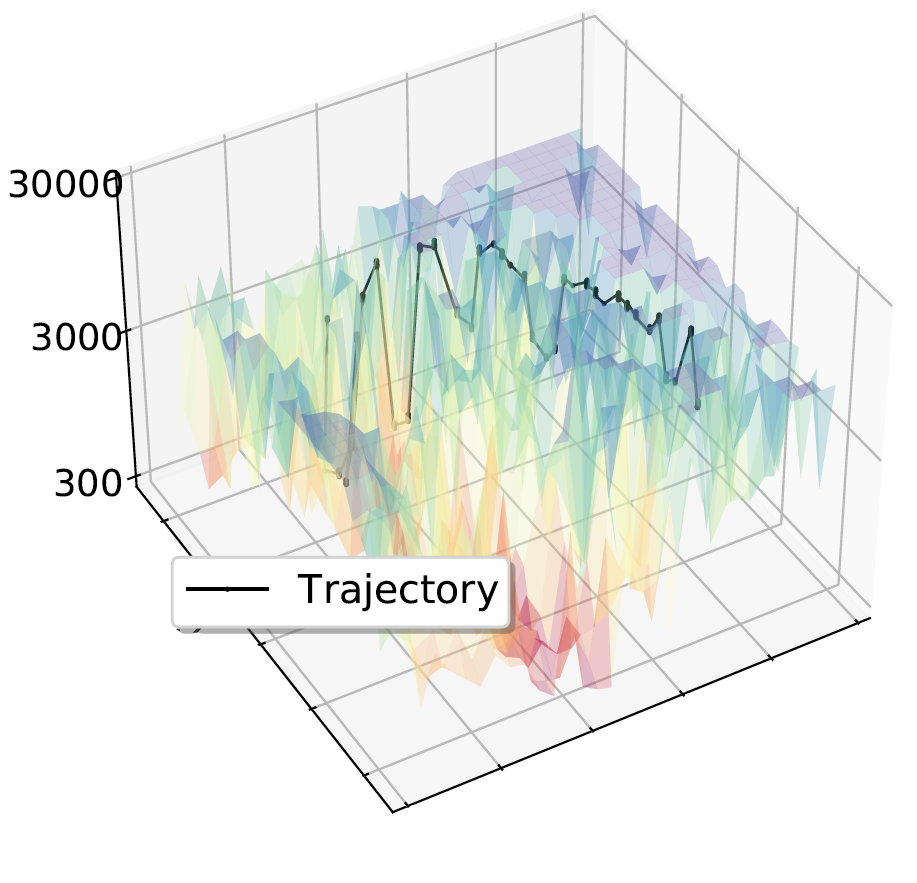} &
\includegraphics[height=1.1in,width=1.1in]{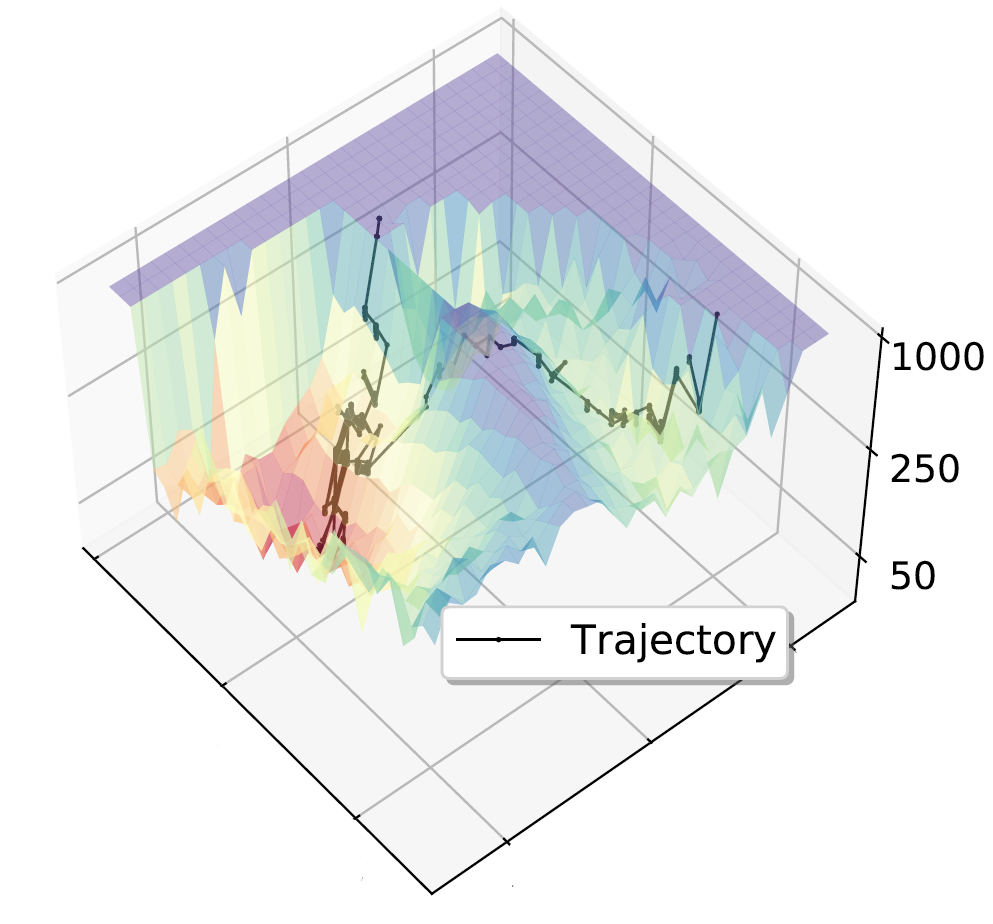} &\\
\includegraphics[height=1.1in,width=1.1in]{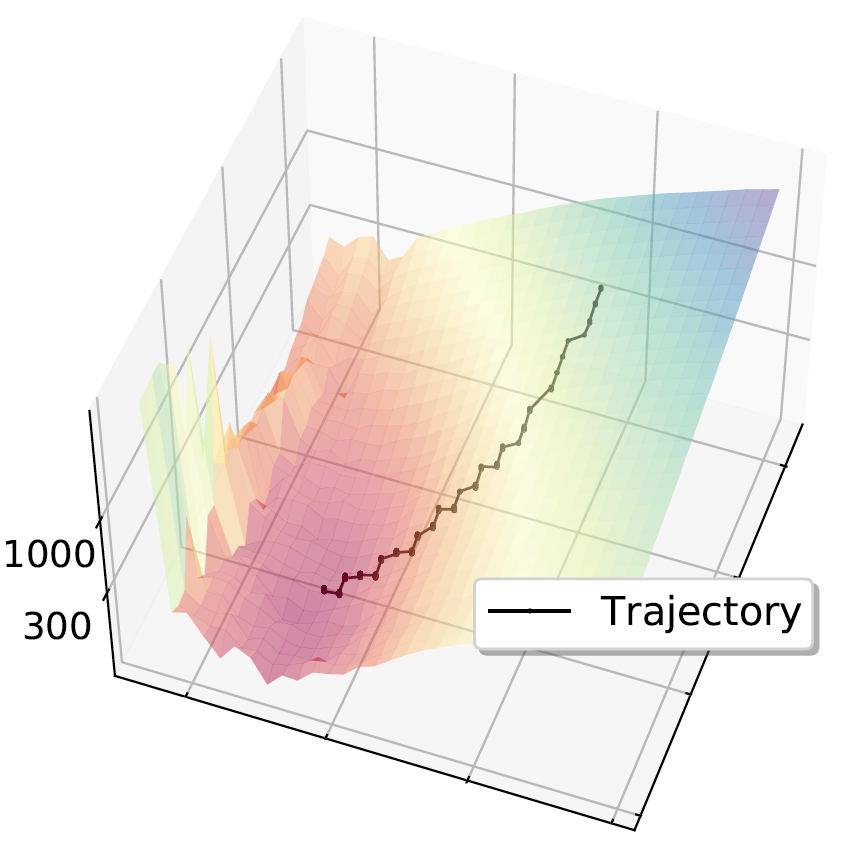} &
\includegraphics[height=1.1in,width=1.1in]{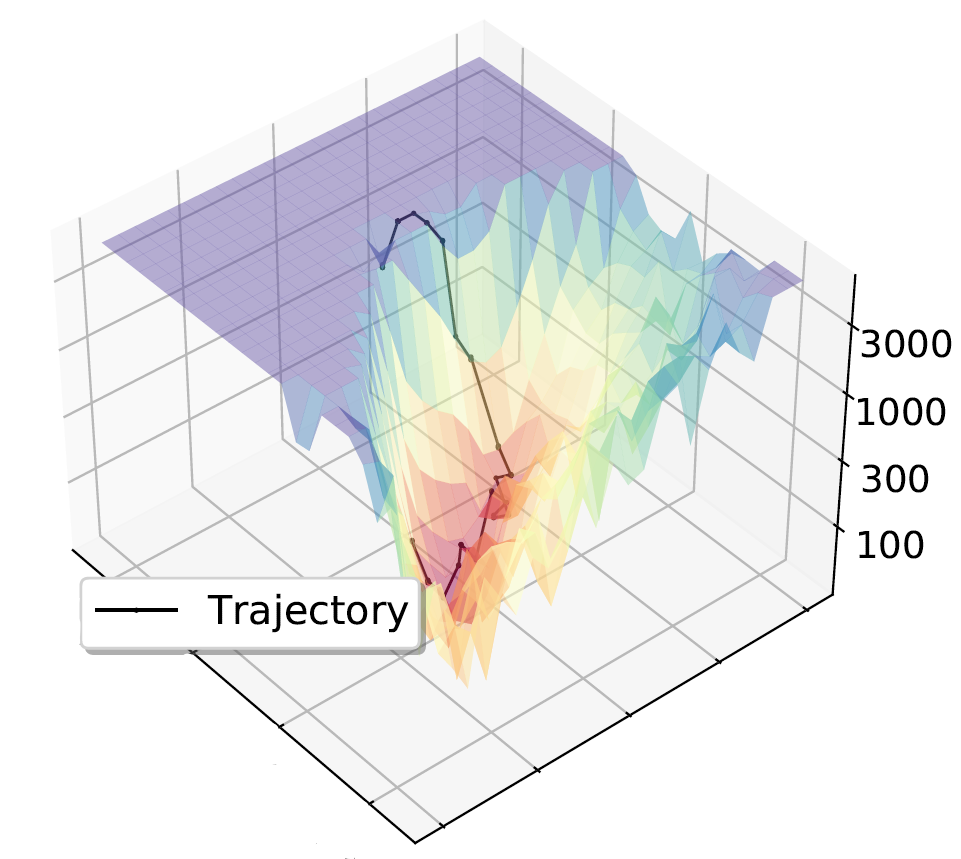} &
\includegraphics[height=1.1in,width=1.1in]{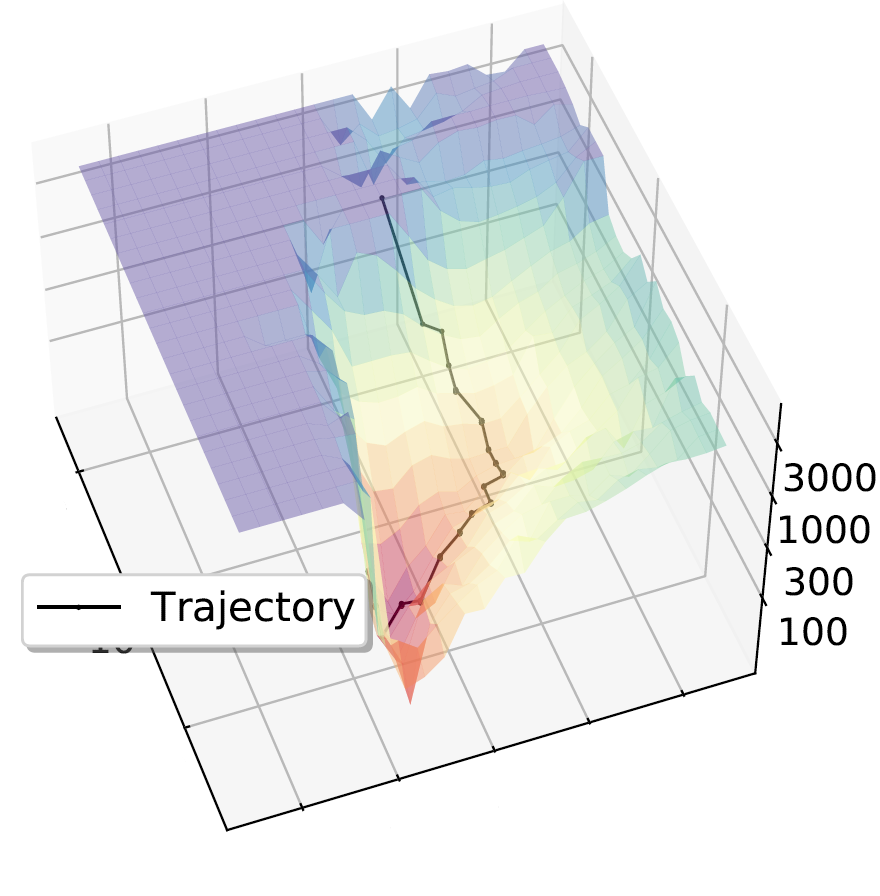} &
\includegraphics[height=1.1in,width=1.1in]{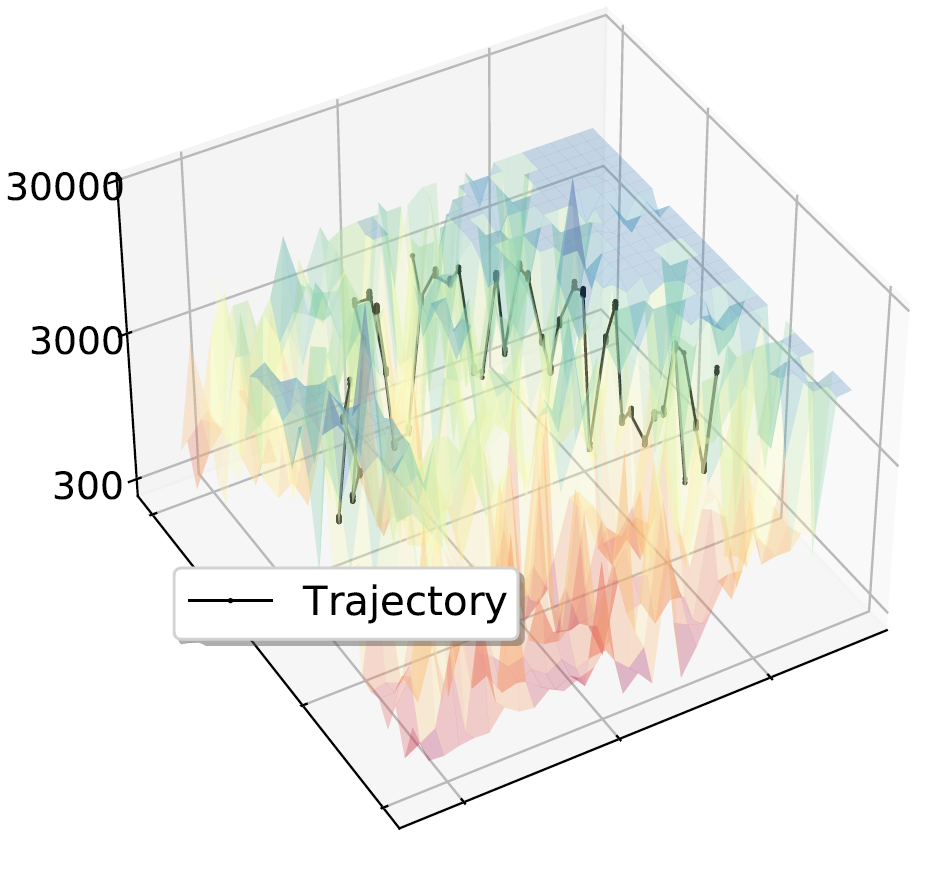} &
\includegraphics[height=1.1in,width=1.1in]{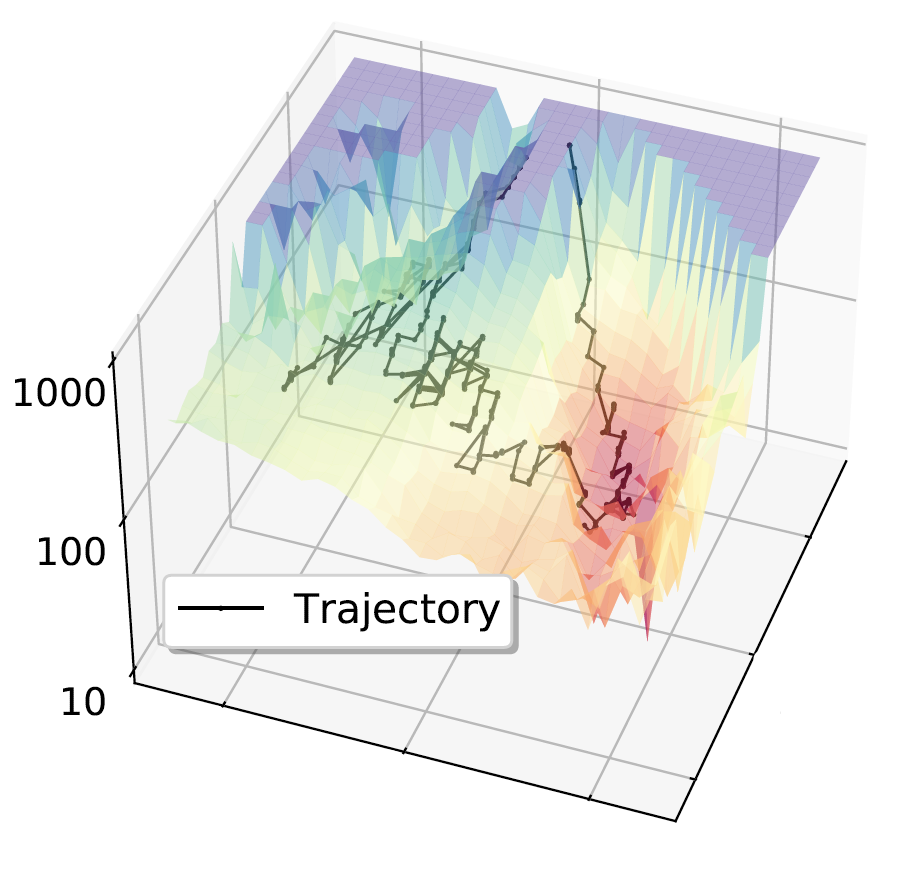} &
\end{tabular}
  \caption{Sampling trajectory visualization of different algorithms on a 2-class MNIST example: (a) SGD, (b) pSGLD, (c) SGHMC, (d) high-$\tau$ SGLD, and (e) AWSGLD, where two trajectories (obtained in two independent runs) are shown for each algorithm. }   
\label{Uncertainty_estimation_mnist_plots}
\end{figure}

To further demonstrate that AWSGLD is immune to local traps, we perform a 2-class classification task (digits 0 and 1) with a similar CNN model as used above. To visualize the approximate trajectory of each algorithm, we conduct a singular value decomposition (SVD) on its sampled models and pick some coordinates based on the first two eigenvectors.  We fix the domain that contains all the coordinates, recover the parameters on some grid points, and then fine-tune the parameters to approximate the loss function. We include  baselines, such as SGD, pSGLD, SGHMC, and high-temperature SGLD. Each algorithm is run for 500 epochs. Fig.\ref{Uncertainty_estimation_mnist_plots} shows that SGD, pSGLD and SGHMC all tend to get stuck at a local minimum (region); the high-temperature SGLD explores over a wide parameter domain, but never locates a good local minimum;  AWSGLD, instead, works well on such a highly non-convex energy landscape, which is indeed able to escape from local energy minima via self-adjusting the gradient multiplier.

\section{Conclusion} 
\label{conclusion}

We have proposed the AWSGLD algorithm as a general-purpose scalable algorithm for both Monte Carlo simulation and global optimization. We have tested its performance on multiple benchmark examples, which shows its
superiority in Monte Carlo simulation and global optimization tasks compared to other popular baselines.

AWSGLD can be viewed as a dynamic SGLD sampler for which both the temperature and learning rate are  self-adjusted according to the position/energy of the current sample. The self-adjusting mechanism enables the algorithm to achieve a good balance between sample space exploration and local geometry exploitation. In the high-energy region, it performs like a high-temperature SGLD sampler for sample space exploration. In the low-energy region, it performs like a low-temperature high-learning rate SGLD sampler; the low temperature enhances its performance in local geometry exploitation, while the large learning rate helps it to escape from local traps. The good balance between sample space exploration and local geometry exploitation enables the algorithm to perform well in both Monte Carlo simulation and global optimization tasks.

\chapter{CONCLUSION AND FUTURE WORK}

\section{Conclusion}
In this thesis, we propose two methodologies to accelerate non-convex Bayesian learning in standard deep learning problems. 

We first propose the replica exchange SGMCMC (reSGMCMC) algorithm and prove the accelerated convergence in terms of 2-Wasserstein distance. The theory implies an accuracy-acceleration trade-off and guides us to tune the correction factor to obtain the optimal performance. To promote more effective swaps and obtain more accelerations, we propose the variance-reduced reSGMCMC algorithm to accelerate the convergence by reducing the variance of the noisy energy estimators. Theoretically, this is the first variance reduction method that yields the potential of exponential accelerations instead of solely reducing the discretization error. Since our variance reduction only conducts on the noisy energy estimators and is not applied to the noisy gradients, the standard hyper-parameter setting can be also naturally imported, which greatly facilitates the training of deep neural works. We also show how to conduct efficient population-chain replica exchange in big data problems. To tackle the inefficiency
issue of the popular DEO scheme given limited chains, we present the generalized DEO scheme by applying an optimal window size to encourage deterministic paths, which accelerates the computation by $O(\frac{P}{\log P})$ times. For a user-friendly purpose, we propose a deterministic swap condition to interact with SGD-based exploration kernels, which greatly reduced the cost in hyperparameter tuning.

We also proposed two adaptively weighted stochastic gradient Langevin dynamics algorithms as general scalable Monte Carlo algorithms for both simulation and optimization tasks. The algorithms automatically adjust the invariant distribution during simulations to facilitate escaping from local traps and traversing over the entire energy landscape. The sampling bias introduced thereby is accounted for by dynamic importance weights. We proved a stability condition for the mean-field system together with the convergence of its self-adapting parameter to a unique fixed point. We established the convergence of a weighted averaging
estimator. The bias of the estimator decreases as we employ a finer partition, a larger mini-batch size, and smaller learning rates and step sizes. 

\section{Future Work}

\subsection{Federated Learning}

Federated learning (FL) allows multiple parties to jointly train a consensus model without sharing user data. A standard formulation of federated learning is a distributed optimization framework that tackles communication costs, client robustness, and data heterogeneity across different clients \cite{lsts20}. Central to the formulation is the efficiency of the communication, which directly motivates the communication-efficient federated averaging (FedAvg)~\cite{mmr+17}. However, the optimization framework often fails to quantify the uncertainty accurately for the parameters of interest. Such a problem leads to unreliable statistical inference and casts doubts on the credibility of the prediction tasks or diagnoses in medical applications.

To unify optimization and uncertainty quantification in federated learning, we resort to a \emph{Bayesian treatment by sampling from a global posterior distribution}, where the latter is aggregated by infrequent communications from local posterior distributions. We adopt a popular approach for inferring posterior distributions for large datasets, the stochastic gradient Markov chain Monte Carlo (SGMCMC) method~\cite{Welling11,Chen14,yian2015}, which enjoys theoretical guarantees beyond convex scenarios \cite{Maxim17, Yuchen17}. The close resemblance naturally inspires us to adapt the optimization-based FedAvg to a distributed sampling framework. Similar ideas have been proposed in federated posterior averaging~\cite{agxr21}, where empirical study and analyses on Gaussian posteriors have shown promising potential of this approach. Compared to the theoretical guarantees of optimization-based algorithms in federated learning, the convergence properties of sampling algorithms in federated learning is far less understood. To fill this gap, we proceed by asking the following question:
\begin{center}
    {\it Can we build a unified algorithm with convergence guarantees for sampling in FL?}
\end{center}

In \cite{deng_FALD}, we make a first step in answering this question in the affirmative. We propose the federated averaging Langevin dynamics (FA-LD) for posterior inference beyond the Gaussian
distribution, which also paves the way for future works of advanced Monte Carlo methods, such as underdamped Langevin dynamics \cite{ccbj18}, replica exchange Monte Carlo (also known as parallel tempering) \cite{deng2020}, and sparse sampling \cite{deng2019} in federated learning. It is also interesting to study the optimal number of local steps under the non-strongly convex \cite{dk17} or non-convex assumptions \cite{Maxim17, ma19}.

\subsection{Reinforcement Learning}

Thompson sampling (TS) and Information directed sampling (IDS) \cite{Russo_roy18} provide systematic strategies for efficient exploration in bandit problems, however, the deployment of these exploration strategies are often limited by complex or intractable posterior distributions. It is promising to tailor stochastic gradient Markov chain Monte Carlo methods \cite{Welling11, yian2015, ccbj18, deng2020, CSGLD} to approximate the posterior distribution under the framework of TS or IDS based on log-concave, sparse linear \cite{IDS_botao_deng}, or multi-modal assumptions. It is also interesting to validate the optimal regret guarantees under those frameworks over finite/infinite-time horizons. 

\subsection{Score-based Generative Modeling}

Score matching \cite{score_matching, score_matching_yang_song} is a promising alternative for generative models. However, it also suffers from the local trap problem, which may lead to subpar results. From the application perspective, it is interesting to study advanced score-based generative modeling inspired by advanced stochastic gradient Markov chain Monte Carlo methods \cite{deng_ICSGLD, deng_NRPT, AWSGLD}. Given better score estimations, we can expect more representative samples with higher quality \cite{provly_SB, SB_bounded_cost}.

\makeatletter
    \defbibenvironment{bibliography}
    {%
        \list
        {%
            \printtext[labelnumberwidth]%
            {%
                \printfield{prefixnumber}%
                \printfield{labelnumber}%
            }%
        }%
        {%
            \setlength{\bibhang}{0pt}%
            \setlength{\itemindent}{0pt}
            \setlength{\itemsep}{\bibitemsep}%
            \setlength{\leftmargin}{0pt}
            \setlength{\parsep}{\bibparsep}%
            \setlength{\rightmargin}{0.33in}%
        }%
    }
    {\endlist}
    {\item}
\makeatother

\setlength{\labelwidth}{1.5in}
\def\sllnsez{[999] }
{%
  \catcode`*=\active
  
  \def*{\char'137}
  
  \PrintBibliography
  
}

\appendices

\chapter{TECHNICAL PROOFS FOR CHAPTER \ref{icml20}}
\label{chapter_appendix_icml20}

\section{Background}
The continuous-time replica exchange Langevin diffusion (reLD) $\{ \bbeta_t\}_{t\ge 0 }:=\left\{\begin{pmatrix}{}
\bbeta_t^{(1)}\\
\bbeta_t^{(2)}
\end{pmatrix}\right\}_{t\ge 0 }$ is a Markov process compounded with a Poisson jump process. In particular, the Markov process follows the stochastic differential equations
\begin{equation}
\label{sde_2}
\begin{split}
    d \bbeta^{(1)}_t &= -  \nabla L(\bbeta_t^{(1)}) dt+\sqrt{2\tau^{(1)}} d\bW_t^{(1)}\\
    d \bbeta^{(2)}_t &= -  \nabla L(\bbeta_t^{(2)}) dt+\sqrt{2\tau^{(2)}} d\bW_t^{(2)},\\
\end{split}
\end{equation}
where $\bbeta_t^{(1)}, \bbeta_t^{(2)}$ are the particles (parameters) at time $t$ in $\mathbb{R}^d$, $\bW^{(1)}, \bW^{(2)}\in\mathbb{R}^d$ are two independent Brownian motions, $L:\mathbb{R}^d\rightarrow \mathbb{R}$ is the energy function, $\tau^{(1)}<\tau^{(2)}$ are the temperatures. The jumps originate from the swaps of particles $\bbeta_t^{(1)}$ and $\bbeta_t^{(2)}$ and follow a Poisson process where the jump rate is specified as the Metropolis form $a (1\wedge S(\bbeta_t^{(1)}, \bbeta_t^{(2)}))dt$. Here $a\geq 0$ is a constant, and $S$ follows
\begin{equation*}
\label{swap_1}
\begin{split}
    S(\bbeta_t^{(1)}, \bbeta_t^{(2)})=e^{ \left(\frac{1}{\tau^{(1)}}-\frac{1}{\tau^{(2)}}\right)\left(L(\bbeta_t^{(1)})-L(\bbeta_t^{(2)})\right)}.\\
\end{split}
\end{equation*}
Under such a swapping rate, the probability $\nu_t$ associated with reLD at time $t$ is known to converge to the invariant measure (Gibbs distribution) with density
\begin{equation*}
\label{pt_density}
\begin{split}
    \pi(\bbeta^{(1)}, \bbeta^{(2)})\propto e^{-\frac{L(\bbeta^{(1)})}{\tau^{(1)}}-\frac{L(\bbeta^{(2)})}{\tau^{(2)}}}.
\end{split}
\end{equation*}

In practice, obtaining the exact energy and gradient for reLD (\ref{sde_2}) in a large dataset is quite expensive. We consider the replica exchange stochastic gradient Langevin dynamics (reSGLD), which generates iterates $\{\widetilde \bbeta^{\eta}(k)\}_{k\ge 1}$ as follows
\begin{equation}
\label{reSGLD}
       \begin{split}
           \widetilde \bbeta^{\eta(1)}(k+1) &= \widetilde \bbeta^{\eta(1)}(k)- \eta  \nabla \widetilde L (\widetilde \bbeta^{\eta(1)}(k))+\sqrt{2\eta\tau^{(1)}} \bxi_k^{(1)}\\
    \widetilde \bbeta^{\eta(2)}(k+1) &= \widetilde \bbeta^{\eta(2)}(k) - \eta \nabla \widetilde L(\widetilde \bbeta^{\eta(2)}(k))+\sqrt{2\eta\tau^{(2)}} \bxi_k^{(2)},\\
       \end{split}
   \end{equation}
where $\eta$ is considered to be a fixed learning rate for ease of analysis, and $\bxi_k^{(1)}$ and $\bxi_k^{(2)}$ are independent Gaussian random vectors in $\mathbb{R}^d$. Moreover, the positions of the particles swap based on the stochastic swapping rate. In particular, $\widetilde S(\bbeta^{(1)}, \bbeta^{(2)}):=S(\bbeta^{(1)}, \bbeta^{(2)})+\psi$, and the stochastic gradient $\nabla\widetilde L(\cdot)$ can be written as $ \nabla L(\cdot)+\bphi$, where both $\psi\in \mathbb{R}^1$ and $\bphi\in\mathbb{R}^d$ are random variables with mean not necessarily zero. We also denote $\m_k$ as the probability measure associated with $\{\widetilde \bbeta^{\eta}(k)\}_{k\ge 1}$ in reSGLD (\ref{reSGLD}) at step $k$, which is close to $\nu_{k\eta}$ in a suitable sense.

\section{Overview of the Analysis}
\label{overview_reSGLD}

We aim to study the convergence analysis of the probability measure $\m_{k}$ to the invariant measure $\pi$ in terms of 2-Wasserstein distance,
\begin{equation}
\label{w2}
    W_2(\m, \nu):=\inf_{\Gamma\in \text{Couplings}(\m, \n)}{\sqrt{\int\|\bbeta_{\m}-\bbeta_{\n}\|^2 d \Gamma(\bbeta_{\m},\bbeta_{\n})}},
\end{equation}
where $\|\cdot\|$ is the Euclidean norm, and the infimum is taken over all joint distributions $\Gamma(\bbeta_{\mu}, \bbeta_{\nu})$ with $\mu$ and $\nu$ being the marginals distributions.

By the triangle inequality, we easily obtain that for any $k\in \mathbb{N}$ and $t=k\eta$, we have 
\begin{equation*}
    W_2(\m_k, \pi)\le \underbrace{W_2(\m_k, \nu_{t})}_{\text{Discretization error}} +\underbrace{W_2(\nu_{t}, \pi)}_{\text{Exponential decay}}.
\end{equation*}
We start with the discretization error first by analyzing how reSGLD (\ref{reSGLD}) tracks the reLD (\ref{sde_2}) in 2-Wasserstein distance. The critical part is to study the discretization of the Poisson jump process in mini-batch settings. To handle this issue, we follow \cite{Paul12} and view the swaps of positions as swaps of temperatures. Then we apply standard techniques in stochastic calculus \cite{chen2018accelerating, yin_zhu_10, Sato2014ApproximationAO, Maxim17} to discretize the Langevin diffusion and derive the corresponding discretization error.

Next, we quantify the evolution of the 2-Wasserstein distance between $\n_t$ and $\pi$. The key tool is the exponential decay of entropy (Kullback-Leibler divergence) when $\pi$ satisfies the log-Sobolev inequality (LSI) \cite{Bakry2014}. To justify LSI, we first verify LSI for reSGLD without swaps, which is a direct result given a proper Lyapunov function criterion \cite{Cattiaux2010} and the Poincar\'{e} inequality \cite{chen2018accelerating}. Then we follow \cite{chen2018accelerating} and verify LSI for reLD with swaps by analyzing the Dirichlet form. Finally, the exponential decay of the 2-Wasserstein distance follows from the Otto-Villani theorem by connecting the 2-Wasserstein distance with the entropy \cite{Bakry2014}.

Before we move forward, we first lay out the following assumptions:
\begin{assumption}[Smoothness]\label{assump: lip and alpha beta}
The energy function $L(\cdot)$ is $C$-smoothness, which implies that there exists a Lipschitz constant $C>0$, such that for every $x,y\in\hR^d$, we have $\| \nabla L(x)- \nabla L(y)\|\le C\|x-y\|.$ \footnote{$\|\cdot\|$ denotes the Euclidean $L^2$ norm.}
\end{assumption}{}

\begin{assumption}[Dissipativity]\label{assump: dissipitive}
The energy function $L(\cdot)$ is $(A,B)$-dissipative, i.e. there exist constants $A>0$ and $B\ge 0$ such that $\forall x\in\mathbb R^d$,  $\la x, \nabla L(x)\ra \ge A\|x\|^2-B.$
\end{assumption}{}

Here the smoothness assumption is quite standard in studying the convergence of SGLD, and the dissipativity condition is widely used in proving the geometric ergodicity of dynamic systems \cite{Maxim17, Xu18}. Moreover, the convexity assumption is not required in our theory.

\section{Analysis of Discretization Error}

The key to deriving the discretization error is to view the swaps of positions as swaps of the temperatures, which has been proven equivalent in distribution \cite{Paul12}. Therefore, we model reLD using the following SDE, 
\begin{equation}\label{replica exchange}
    d\bbeta_t=-\nabla G(\bbeta_t)dt+\Si_td\bW_t,
\end{equation}{}
where  $G(\bbeta_t)=\begin{pmatrix}{}
L(\bbeta_t^{(1)})\\
L(\bbeta_t^{(1)})
\end{pmatrix}$, $\bW\in\mathbb{R}^{2d}$ is a Brownian motion, $\Si_t$ is a random matrix in continuous-time that swaps between the diagonal matrices $\mathbb{M}_1=\begin{pmatrix}{}
\sqrt{2\tau^{(1)}}\mathbf I_d&0\\
0&\sqrt{2\tau^{(2)}}\mathbf I_d
\end{pmatrix}$ and $\mathbb{M}_2=\begin{pmatrix}{}
\sqrt{2\tau^{(2)}}\mathbf I_d&0\\
0&\sqrt{2\tau^{(1)}}\mathbf I_d
\end{pmatrix}$ with probability $a S(\bbeta_t^{(1)}, \bbeta_t^{(2)})dt$, and $\mathbf I_d\in \mathbb R^{d\times d}$ is denoted as the identity matrix.

Moreover, the corresponding discretization of replica exchange SGLD (reSGLD) follows:
\begin{equation}
\begin{split}
\label{resgld_2}
    \widetilde \bbeta^{\eta}(k+1)=\widetilde \bbeta^{\eta}(k)-
    \eta\nabla \widetilde G(\widetilde \bbeta^{\eta}(k)) + \sqrt{\eta}\widetilde \Si^{\eta}(k)\bxi_k,
\end{split}
\end{equation}
where $\bxi_k$ is a standard Gaussian distribution in $\mathbb{R}^{2d}$, and $\widetilde \Si^{\eta}(k)$ is a random matrix in discrete-time that swaps between $\mathbb{M}_1$ and $\mathbb{M}_2$ with probability $a\widetilde S(\widetilde \bbeta^{\eta(1)}(k), \widetilde \bbeta^{\eta(2)}(k))\eta$. We denote $\{\widetilde \bbeta_t^{\eta}\}_{t\ge 0 }$ as the continuous-time interpolation of $\{\widetilde \bbeta^{\eta}(k)\}_{k\ge 1}$, which satisfies the following SDE, 
\begin{equation}\label{SGD continuous time interpolation}
	\widetilde \bbeta_t^{\eta}=\widetilde \bbeta_0-\int_0^t\nabla \widetilde G(\widetilde \bbeta^{\eta}_{\lfloor s/\eta \rfloor \eta})ds+\int_0^t\widetilde\Si^{\eta}_{\lfloor s/\eta \rfloor\eta}d\bW_s.
\end{equation}
Here the random matrix $\widetilde \Si_{\lfloor s/\eta\rfloor \eta}^{\eta}$ follows a similar trajectory as $\widetilde \Si^{\eta}(\lfloor s/\eta \rfloor)$. 
For $k\in \mathbb N^{+}$ with $t=k\eta$, the relation $\widetilde \bbeta_t^{\eta}=\widetilde \bbeta_{k\eta}^{\eta}=\widetilde \bbeta^{\eta}(k)$ follows.

\begin{lemma}[Discretization error. Formal statement of Lemma \ref{discretization_main_reSGLD}]\label{discretization_appendix_reSGLD}
Given the smoothness and dissipativity assumptions \eqref{assump: lip and alpha beta} and \eqref{assump: dissipitive}, and the learning rate $\eta$ satisfying $0<\eta<1 \land A/C^2$, there exists constants $D_1, D_2$ and $D_3$ such that
\begin{equation}
	\begin{split}
&	\hE[\sup_{0\le t\le T}\|\bbeta_t-\widetilde \bbeta^{\eta}_t\|^2] \le D_1 \eta + D_2 \max_{k}\hE[\|\bphi_k\|^2]+D_3\max_{k}\sqrt{\hE\left[\mid\psi_{k}\mid^2\right]},\\
	\end{split}
\end{equation}
where $ D_1$ depends on $\tau^{(1)},\tau^{(2)},d, T, A,B, C$; $D_2$ depends on $T$ and $C$; $D_3$ depends on $a, d, T$ and $C$.
\end{lemma}{}

\begin{proof}
Based on the replica exchange Langevin diffusion $\{\bbeta_t\}_{t\ge 0}$ and the continuous-time interpolation of the stochastic gradient Langevin diffusion $\{\widetilde \bbeta_t^{\eta}\}_{t\ge 0}$, we have the following SDE for the difference $\bbeta_t-\widetilde \bbeta^{\eta}_t$. For any $t\in [0,T]$, we have
\begin{equation*}
	\begin{split}
		\bbeta_t-\widetilde \bbeta^{\eta}_t&=-\int_0^t(\nabla G(\bbeta_s)-\nabla\widetilde G(\widetilde \bbeta^{\eta}_{\lfloor s/\eta \rfloor\eta})ds+\int_0^t (\Si_s-\widetilde \Si^{\eta}_{\lfloor s/\eta \rfloor\eta})d\bW_s
	\end{split}
\end{equation*}
Indeed, note that
\begin{equation*}
	\begin{split}
		\sup_{0\le t\le T}\|\bbeta_t-\widetilde \bbeta^{\eta}_t\|&\le \int_0^T\|\nabla G(\bbeta_s)-\nabla\widetilde G(\widetilde \bbeta^{\eta}_{\lfloor s/\eta \rfloor\eta})\|)ds+\sup_{0\le t\le T}\left\|\int_0^t (\Si_s-\widetilde \Si^{\eta}_{\lfloor s/\eta \rfloor\eta})d\bW_s\right\|
	\end{split}
\end{equation*}
We first square both sides and take expectation, then apply the Burkholder-Davis-Gundy inequality and  Cauchy-Schwarz inequality, we have 
\begin{equation}\label{sup norm estimate 1}
	\begin{split}
	&\quad\hE[	\sup_{0\le t\le T}\|\bbeta_t-\widetilde \bbeta^{\eta}_t\|^2]\\
	&\le 2\hE\left[\left(\int_0^T\|\nabla G(\bbeta_s)-\nabla\widetilde G(\widetilde \bbeta^{\eta}_{\lfloor s/\eta \rfloor\eta})\|ds\right)^2+\sup_{0\le t\le T}\left\|\int_0^t (\Si_s-\widetilde \Si^{\eta}_{\lfloor s/\eta \rfloor\eta})d\bW_s\right\|^2\right]\\
	&\le \underbrace{2T\hE\left[\int_0^T\|\nabla G(\bbeta_s)-\nabla\widetilde G(\widetilde \bbeta^{\eta}_{\lfloor s/\eta \rfloor\eta})\|^2ds \right]}_{\cI}+\underbrace{8\hE\left[\int_0^T\|\Si_s-\widetilde \Si^{\eta}_{\lfloor s/\eta \rfloor\eta}\|^2 ds\right]}_{\cJ}
	\end{split}
\end{equation}

\noindent
$\textbf{Estimate of stochastic gradient:}$
For the first term $\cI$, by using the inequality
$$\|a+b+c\|^2\le 3(\|a\|^2+\|b\|^2+\|c\|^2),$$ 
we get  
\begin{equation}
	\begin{split}
		\cI=&2T\hE\Big[\int_0^T\Big\|\left(\nabla G(\bbeta_s)-\nabla G(\widetilde \bbeta^{\eta}_s)\right)+\left(\nabla G(\widetilde \bbeta^{\eta}_s)-\nabla  G(\widetilde \bbeta^{\eta}_{\lfloor s/\eta \rfloor\eta})\right)\\
		&\qquad+\left(\nabla  G(\widetilde \bbeta^{\eta}_{\lfloor s/\eta \rfloor\eta})-\nabla \widetilde G(\widetilde \bbeta^{\eta}_{\lfloor s/\eta \rfloor\eta})\right)\Big\|^2ds \Big]\\
		\le& \underbrace{6T\hE\left[\int_0^T\|\nabla G(\bbeta_s)-\nabla  G(\widetilde \bbeta^{\eta}_s)\|^2ds\right]}_{\cI_1}+\underbrace{6T\hE\left[\int_0^T\|\nabla  G(\widetilde \bbeta^{\eta}_s)-\nabla  G(\widetilde \bbeta^{\eta}_{\lfloor s/\eta \rfloor\eta})\|^2ds\right]}_{\cI_2}\\
		&+\underbrace{6T\hE\left[\int_0^T\|\nabla  G(\widetilde \bbeta^{\eta}_{\lfloor s/\eta \rfloor\eta})-\nabla \widetilde G(\widetilde \bbeta^{\eta}_{\lfloor s/\eta \rfloor\eta})\|^2ds\right]}_{\cI_3}\\
		\le& \cI_1+\cI_2+\cI_3.
	\end{split}
\end{equation}
By using the smoothness assumption \ref{assump: lip and alpha beta}, we first estimate 
\begin{equation*}
    \cI_1 \le 6TC^2\hE\left[\int_0^T\|\bbeta_s-\widetilde \bbeta^{\eta}_s\|^2ds\right].
\end{equation*}{}
By applying the smoothness assumption \ref{assump: lip and alpha beta} and discretization scheme, we can further estimate 
\begin{equation}\label{est cI2}
    \begin{split}
        \cI_2&\le 6TC^2\hE\left[\int_0^T\|\widetilde \bbeta^{\eta}_s-\widetilde \bbeta^{\eta}_{\lfloor s/\eta \rfloor\eta}\|^2ds\right]\\
        &\le 6TC^2\sum_{k=0}^{\lfloor T/\eta\rfloor} \hE\left[\int_{k\eta}^{(k+1)\eta}\|\widetilde \bbeta^{\eta}_s-\widetilde \bbeta^{\eta}_{\lfloor s/\eta \rfloor\eta} \|^2ds \right]\\
        &\le 6TC^2\sum_{k=0}^{\lfloor T/\eta\rfloor} \int_{k\eta}^{(k+1)\eta}\hE\left[\sup_{k\eta\le s<(k+1)\eta}\|\widetilde \bbeta^{\eta}_s-\widetilde \bbeta^{\eta}_{\lfloor s/\eta \rfloor\eta} \|^2\right]ds 
    \end{split}{}
\end{equation}{}
For $\forall~ k\in\mathbb N$ and $s\in [k\eta,(k+1)\eta)$, we have
\begin{equation*}
    \begin{split}
        \widetilde \bbeta^{\eta}_s-\widetilde \bbeta^{\eta}_{\lfloor s/\eta \rfloor\eta}=\widetilde \bbeta^{\eta}_s-\widetilde \bbeta^{\eta}_{k\eta}=-\nabla\widetilde G(\widetilde \bbeta^{\eta}_{k\eta})\cdot(s-k\eta)+\widetilde\Si^{\eta}_{k\eta}\int_{k\eta}^sd\bW_r
    \end{split}{}
\end{equation*}{}
which indeed implies
\begin{equation*}
    \begin{split}
     \sup_{k\eta\le s<(k+1)\eta}  \| \widetilde \bbeta^{\eta}_s-\widetilde \bbeta^{\eta}_{\lfloor s/\eta \rfloor\eta}\|\le \|\nabla\widetilde G(\widetilde \bbeta^{\eta}_{k\eta})\|(s-k\eta)+\sup_{k\eta\le s<(k+1)\eta} \|\widetilde\Si^{\eta}_{k\eta}\int_{k\eta}^sd\bW_r\|
    \end{split}{}
\end{equation*}{}
Similar to the estimate \eqref{sup norm estimate 1}, square both sides and take expectation, then apply the Burkholder-Davis-Gundy inequality, we have
\begin{equation*}
    \begin{split}
        \hE\left[\sup_{k\eta\le s<(k+1)\eta}  \| \widetilde \bbeta^{\eta}_s-\widetilde \bbeta^{\eta}_{\lfloor s/\eta \rfloor\eta}\|^2 \right]&\le 2\hE[ \|\nabla\widetilde G(\widetilde \bbeta^{\eta}_{k\eta})\|^2(s-k\eta)^2 ]+8\sum_{j=1}^{2d}\hE\left[\left(\widetilde \Si^{\eta}_{k\eta}(j)\la \int_{k\eta}^{\cdot}d\bW_r \ra_s^{1/2}\right)^2 \right]\\
         &\le  2(s-k\eta)^2\hE[ \|\nabla \widetilde G(\widetilde \bbeta^{\eta}_{k\eta})\|^2]+32d\tau^{(2)}(s-k\eta),
    \end{split}{}
\end{equation*}{}



where the last inequality follows from the fact that $\widetilde \Si^{\eta}_{k\eta}$ is a diagonal matrix with diagonal elements $\sqrt{2\tau^{(1)}}$ or $\sqrt{2\tau^{(2)}}.$
For the first term in the above inequality, we further have
\begin{equation*}
    \begin{split}
  2(s-k\eta)^2 \hE[ \|\nabla \widetilde G(\widetilde \bbeta^{\eta}_{k\eta})\|^2]&=        2(s-k\eta)^2\hE[ \|(\nabla G(\widetilde \bbeta^{\eta}_{k\eta})+\bphi_k)\|^2 ]\\
  &\le 4\eta^2\hE[ \|\nabla  G(\widetilde \bbeta^{\eta}_{k\eta})-\nabla  G(\bbeta^*) \|^2+\|\bphi_k\|^2 ]\\
   &\le 8C^2\eta^2\hE[\|\widetilde \bbeta^{\eta}_{k\eta}\|^2+\|\bbeta^*\|^2]+4\eta^2\hE[\|\bphi_k\|^2 ],
    \end{split}{}
\end{equation*}
where the first inequality follows from the separation of the noise from the stochastic gradient and the choice of stationary point $\bbeta^*$ of $G(\cdot)$ with $\nabla G(\bbeta^*)=0$, and $\bphi_k$ is the stochastic noise in the gradient at step $k$. Thus, combining the above two parts and integrate $ \hE\left[\sup_{k\eta\le s<(k+1)\eta}  \| \widetilde \bbeta^{\eta}_s-\widetilde \bbeta^{\eta}_{k\eta}\|^2 \right]$ on the time interval $[k\eta,(k+1)\eta)$, we obtain the following bound 
\begin{equation}\label{est cI2 part two}
    \begin{split}
        &\quad\int_{k\eta}^{(k+1)\eta}  \hE\left[\sup_{k\eta\le s<(k+1)\eta}  \| \widetilde \bbeta^{\eta}_s-\widetilde \bbeta^{\eta}_{k\eta}\|^2 \right] ds\\
        &\le 8C^2\eta^3 \left(\sup_{k\ge 0}\hE[ \|\widetilde \bbeta_{k\eta}^{\eta}\|^2+ \|\bbeta^*\|^2] \right)+4\eta^3 \max_{k}\hE[\|\bphi_k\|^2]+32d\tau^{(2)}\eta^2
    \end{split}{}
\end{equation}{}
By plugging the estimate \eqref{est cI2 part two} into estimate \eqref{est cI2}, we obtain the following estimates when $\eta\le 1$,
\begin{equation}
    \begin{split}
        \cI_2&\le 6TC^2(1+T/\eta)\left[8C^2\eta^3 \left(\sup_{k\ge 0}\hE[ \|\widetilde \bbeta_{k\eta}^{\eta}\|^2+\|\bbeta^*\|^2] \right)+4\eta^3 \max_{k}\hE[\|\bphi_k\|^2]+32d\tau^{(2)}\eta^2\right]\\
        & \le \tilde \delta_1(d, \tau^{(2)}, T, C, A, B) \eta + 24TC^2(1+T)\max_{k}\hE[\|\bphi_k\|^2],
    \end{split}{}
\end{equation}{}
where $\tilde \delta_1(d, \tau^{(2)}, T, C, A, B)$ is a constant depending on $d, \tau^{(2)}, T, C, A$ and $B$. Note that the above inequality requires a result on the bounded second moment of   $\sup_{k\ge 0}\hE[ \|\widetilde \bbeta_{k\eta}^{\eta}\|^2]$, and this is majorly\footnote{The slight difference is that the constant in the RHS of (C.38) \cite{chen2018accelerating} is changed to account for the stochastic noise.} proved in Lemma $C.2$ in \cite{chen2018accelerating} when we choose the stepzise $\eta\in (0, A/C^2)$. We are now left to estimate the term $\cI_3$ and we have
\begin{equation}
    \begin{split}
   \cI_3  &\le  6T\sum_{k=0}^{\lfloor T/\eta \rfloor} \hE\left[\int_{k\eta}^{(k+1)\eta}\|\nabla  G(\widetilde \bbeta^{\eta}_{k\eta})-\nabla \widetilde G(\widetilde \bbeta^{\eta}_{k\eta})\|^2ds \right]\\
   &\le 6T (1+T/\eta)\max_{k}\hE[ \|\bphi_k\|^2]\eta\\
   &\le 6T (1+T)\max_{k}\hE[ \|\bphi_k\|^2].
    \end{split}{}
\end{equation}{}
Combing all the estimates of $\cI_1,\cI_2$ and $\cI_3$, we obtain
\begin{equation}
    \begin{split}
        \cI\le& \underbrace{6TC^2\int_0^T\hE\left[\sup_{0\le s\le T}\|\bbeta_s-\widetilde \bbeta^{\eta}_s\|^2\right]ds}_{\cI_1}+\underbrace{\tilde \delta_1(d, \tau^{(2)}, T, C, A, B) \eta+ 24TC^2(1+T)\max_{k}\hE[\|\bphi_k\|^2]}_{\cI_2}\\
        &+\underbrace{6T (1+T)\max_{k}\hE[ \|\bphi_k\|^2]}_{\cI_3}.\\
    \end{split}{}
\end{equation}{}
\noindent
$\textbf{Estimate of stochastic diffusion:}$
For the second term $\cJ$, we have 
\begin{equation}
	\begin{split}
		\cJ&=8\hE\left[\int_{0}^{T} \|\Sigma_s(j)-\widetilde \Sigma_{{\lfloor s/\eta \rfloor}\eta}(j)\|^2 ds\right] \\
		& \leq 8\sum_{j=1}^{2d}\sum_{k=0}^{{\lfloor T/\eta \rfloor}}\int_{k \eta}^{(k+1)\eta}\hE\left[\|\Sigma_s(j)-\widetilde \Sigma^{\eta}_{k \eta}(j)\|^2\right]ds \\
		& \leq 8\sum_{j=1}^{2d}\sum_{k=0}^{{\lfloor T/\eta \rfloor}}\int_{k \eta}^{(k+1)\eta}\hE\left[\|\Sigma_s(j)- \Sigma_{k \eta}^{\eta}(j)+\Sigma_{k \eta}^{\eta}(j)-\widetilde \Sigma_{k \eta}^{\eta}(j)\|^2\right]ds \\
		& \leq 16\sum_{j=1}^{2d}\sum_{k=0}^{{\lfloor T/\eta \rfloor}}\Bigg[\underbrace{\int_{k \eta}^{(k+1)\eta}\hE\left[\|\Sigma_s(j)- \Sigma_{k \eta}^{\eta}(j)\|^2\right]ds}_{\cJ_1}\\
		&\qquad\qquad\qquad\qquad+\underbrace{\int_{k \eta}^{(k+1)\eta}\hE\left[\|\Sigma_{k \eta}^{\eta}(j)-\widetilde \Sigma_{k \eta}^{\eta}(j)\|^2\right]ds}_{\cJ_2}\Bigg].
	\end{split}
\end{equation}
where $\Si_{k \eta}^{\eta}$ is the temperature matrix for the continuous-time interpolation of $\{\bbeta^{\eta}(k)\}_{k\ge 1}$, which is similar to \eqref{SGD continuous time interpolation} without noise generated from mini-batch settings and is defined as below
\begin{equation}
	 \bbeta_t^{\eta}= \bbeta_0-\int_0^t\nabla  G( \bbeta^{\eta}_{k\eta})ds+\int_0^t\Si^{\eta}_{k \eta}d\bW_s.
\end{equation}

We estimate $\cJ_1$ first, considering that $\Si_s$ and $\Si^{\eta}_{\lfloor s/\eta\rfloor \eta}$ are both diagonal matrices, we have
\begin{equation*}
    \begin{split}
        \cJ_1&=4(\sqrt{\tau^{(2)}}-\sqrt{\tau^{(1)}})^2\int_{k \eta}^{(k+1)\eta}\hP(\Si_s(j)\neq \Si^{\eta}_{k \eta}(j))ds\\
        &=4(\sqrt{\tau^{(2)}}-\sqrt{\tau^{(1)}})^2\hE\left[\int_{k \eta}^{(k+1)\eta}\hP(\Si_s(j)\neq \Si^{\eta}_{k \eta}(j)\mid\bbeta^{\eta}_{k \eta})ds\right]\\
        &= 4(\sqrt{\tau^{(2)}}-\sqrt{\tau^{(1)}})^2 a\int_{k \eta}^{(k+1)\eta} [(s-k \eta)+\mathcal R(s-k \eta)]ds\\
        &\le \tilde \delta_2(a, \tau^{(1)},\tau^{(2)})\eta^2,
    \end{split}{}
\end{equation*}{}
where $\tilde \delta_2(a, \tau^{(1)},\tau^{(2)})=4(\sqrt{\tau^{(2)}}-\sqrt{\tau^{(1)}})^2 a$, and the equality follows from the fact that the conditional probability $\hP(\Si_s(j)\neq \Si^{\eta}_{k \eta}(j)\mid\bbeta^{\eta}_{k \eta})=a S(\bbeta^{\eta(1)}_{k \eta},\bbeta^{\eta(2)}_{k \eta})\cdot(s-\eta)+a\mathcal R(s-k \eta)$. Here $\mathcal R(s-k \eta)$ denotes the higher remainder with respect to $s-k \eta$. The estimate of $\cJ_1$ without stochastic gradient for the Langevin diffusion is first obtained in \cite{chen2018accelerating}, we however present here again for reader's convenience.  
As for the second term $\cJ_2$, it follows that
\begin{equation}
	\begin{split}
	\label{new_prob}
	    \cJ_2&=4(\sqrt{\tau^{(2)}}-\sqrt{\tau^{(1)}})^2\int_{k \eta}^{(k+1)\eta}\hP(\Sigma_{k \eta}(j)\neq \widetilde \Sigma_{k \eta}(j))ds \\
	    &= 4(\sqrt{\tau^{(2)}}-\sqrt{\tau^{(1)}})^2 a\eta \hE\left[\mid S(\bbeta_{k \eta}^{\eta(1)}, \bbeta_{k \eta}^{\eta(2)})-\tilde S(\widetilde \bbeta_{k \eta}^{\eta(1)}, \widetilde \bbeta_{k \eta}^{\eta(2)})\mid\right] \\
	    &\le \tilde \delta_2(a, \tau^{(1)},\tau^{(2)}) \eta \sqrt{\hE\left[\mid S(\bbeta_{k \eta}^{\eta(1)}, \bbeta_{k \eta}^{\eta(2)})-\tilde S(\widetilde \bbeta_{k \eta}^{\eta(1)}, \widetilde \bbeta_{k \eta}^{\eta(2)})\mid^2\right]}\\
	    & \leq \tilde \delta_2(a, \tau^{(1)},\tau^{(2)}) \eta \sqrt{\hE\left[\mid\psi_{k}\mid^2\right]},
	\end{split}
\end{equation}
where $\psi_{k}$ is the noise in the swapping rate. Thus, one concludes the following estimates combing $\cI$ and $\cJ$.
\begin{equation}
	\begin{split}
	&\quad\hE[\sup_{0\le t\le T}\|\bbeta_t-\widetilde \bbeta^{\eta}_t\mid^2] \\
	&\le \underbrace{6TC^2\int_0^T\hE\left[\sup_{0\le s\le T}\|\bbeta_s-\widetilde \bbeta^{\eta}_s\|^2\right]ds}_{\cI_1}\\
	&\quad+\underbrace{\tilde \delta_1(d, \tau^{(2)}, T, C, A, B) \eta+ 24TC^2(\eta+T)\max_{k}\hE[\|\bphi_k\|^2]}_{\cI_2}\\
	&\ \ \ \ \ +\underbrace{6T (1+T) \hE[\|\bphi_{k}\|^2]}_{\cI_3}+\underbrace{32d(1+T)\tilde \delta_2(a, \tau^{(1)},\tau^{(2)})\left(\eta+\max_{k}\sqrt{\hE\left[\mid\psi_{k}\mid^2\right]}\right)}_{\cJ}. 
	\end{split}
\end{equation}
Apply Gronwall's inequality to the function
\begin{equation*}
    t\mapsto \hE\left[\sup_{0\le u\le t} \|\bbeta_u-\widetilde \bbeta_u^{\eta}\|^2 \right],
\end{equation*}{}
and deduce that 
\begin{equation}
	\begin{split}
&	\hE[\sup_{0\le t\le T}\|\bbeta_t-\widetilde \bbeta^{\eta}_t\mid^2] \le D_1 \eta + D_2 \max_{k}\hE[\|\bphi_k\|^2]+D_3\max_{k}\sqrt{\hE\left[\mid\psi_{k}\mid^2\right]},\\
	\end{split}
\end{equation}
where $ D_1$ is a constant depending on $\tau^{(1)},\tau^{(2)},d, T, C,A, B$; $D_2$ depends on $T$ and $C$; $D_3$ depends on $a, d, T$ and $C$. \qed

\end{proof}

\section{Exponential Decay of Wasserstein Distance in Continuous-Time}

We proceed to quantify the evolution of the 2-Wasserstein distance between $\n_{t}$ and $\pi$. We first consider the ordinary Langevin diffusion without swaps and derive the log-Sobolev inequality (LSI). Then we extend LSI to reLD and obtain the exponential decay of the relative entropy. Finally, we derive the exponential decay of the 2-Wasserstein distance.

In order to distinguish from the replica exchange Langevin diffusion $\bbeta_t$ defined in \eqref{replica exchange}, we call it $\hat \bbeta_t$ which follows, 
\begin{equation}
    d\hat \bbeta_t=-\nabla G(\hat \bbeta_t)dt+\Si_td\bW_t.
\end{equation}{}
where $\Si_t\in \hR^{2d\times 2d}$ is a diagonal matrix with the form $\begin{pmatrix}{}
\sqrt{2\tau^{(1)}}\mathbf I_d&0\\
0&\sqrt{2\tau^{(2)}}\mathbf I_d
\end{pmatrix}$. 
The process $\hat \bbeta_t$ is a Markov diffusion process with infinitesimal generator $\cL$ in the following form, for $\bbeta^{(1)}\in\hR^d$ and $\bbeta^{(2)}\in\hR^d$,
\beaa
\cL=&-\la\nabla_{\bbeta^{(1)}}f(\bbeta^{(1)},\bbeta^{(2)}), \nabla L(\bbeta^{(1)})\ra+\tau^{(1)}\Delta_{\bbeta^{(1)}}f(\bbeta^{(1)},\bbeta^{(2)})\\
&-\la \nabla_{\bbeta^{(2)}}f(\bbeta^{(1)},\bbeta^{(2)}), \nabla L(\bbeta^{(2)})\ra+\tau^{(2)}\Delta_{\bbeta^{(2)}}f(\bbeta^{(1)},\bbeta^{(2)})
\eeaa
Note that since matrix $\Si_t$ is a non-degenerate diagonal matrix, operator $\cL$ is an elliptic diffusion operator. According to the smoothness assumption \eqref{assump: lip and alpha beta}, we have that $\nabla^2 G\ge -C\mathbf I_{2d}$, where $C>0$, the unique invariant measure $\pi$ associate with the underlying diffusion process satisfies the Poincare inequality and LSI with the Dirichlet form given as follows,
\bea\label{dirichlet form}
\cE(f)=\int \Big(\tau^{(1)}\|\nabla_{\bbeta^{(1)}}f\|^2+\tau^{(2)}\|\nabla_{\bbeta^{(2)}}f\|^2 \Big)d\pi(\bbeta^{(1)},\bbeta^{(2)}),\qq f\in\cC_0^2(\hR^{2d}).
\eea
In this elliptic case with $G$ being convex, the proof for LSI follows from standard Bakry-Emery calculus \cite{Bakry85}. Since, we are dealing with the non-convex function $G$, we are particularly interested in the case of $\nabla^2 G\ge -C\mathbf I_{2d}$.
To obtain a Poincar\'{e} inequality for invariant measure $\pi$, \cite{chen2018accelerating} adapted an argument from \cite{Bakry08} and \cite{Maxim17} by constructing an appropriate Lyapunov function for the replica exchange diffusion without swapping $\hat \bbeta_t$. Denote $\n_t$ as the distribution associated with the diffusion process $\{\hat \bbeta_t\}_{t\ge 0}$, which is absolutely continuous with respect to $\pi$. It is a direct consequence of the aforementioned results that the following log-Sobolev inequality holds.

\begin{lemma}[LSI for Langevin Diffusion]\label{LSI no swaping}
Under assumptions \eqref{assump: lip and alpha beta} and \eqref{assump: dissipitive}, we have the following log-Sobolev inequality for invariant measure $\pi$, for some constant $c_{\text{LS}}>0$, 
\beaa
D(\n_t\|\pi)\le 2c_{\text{LS}}\cE(\sqrt{\frac{d\n_t}{d\pi}}).
\eeaa
where $D(\n_t\|\pi)=\int d\nu_t \log\frac{d\nu_t}{d\pi}$ denotes the relative entropy and the Dirichlet form $\cE(\cd)$ is defined in \eqref{dirichlet form}.
\end{lemma}{}
\begin{proof}

According to \cite{Cattiaux2010}, the sufficient conditions to establish LSI are:
\begin{enumerate}
\item There exists some constant $C\ge 0$, such that $\nabla^2 G\succcurlyeq -C I_{2d}$.
\item $\pi$ satisfies a Poincar\'{e} inequality with constant $c_{p}$, namely, for all probability measures $\nu\ll\pi$, $\chi^2(\nu\|\pi)\leq c_p \cE(\sqrt{\frac{d\n_t}{d\pi}})$, where $\chi^2(\nu\|\pi):=\|\frac{d\nu}{d\pi}-1\|^2$ is the $\chi^2$ divergence between $\nu$ and $\pi$.
\item There exists a $\cC^2$ Lyapunov function $V: \mathbb{R}^{2d}\rightarrow [1, \infty)$ such that $\frac{\cL V(\bbeta^{(1)}, \bbeta^{(2)})}{V(\bbeta^{(1)}, \bbeta^{(2)})} \leq \kappa - \gamma (\|\bbeta^{(1)}\|^2 + \|\bbeta^{(2)}\|^2)$
for all $(\bbeta^{(1)}, \bbeta^{(2)})\in \mathbb{R}^{2d}$ and some $\kappa, \gamma>0$.
\end{enumerate}
Note that the first condition on the Hessian is obtained from the smoothness assumption \eqref{assump: lip and alpha beta}. Moreover, the Poincar\'{e} inequality in the second condition is derived from Lemma C.1 in \cite{chen2018accelerating} given assumptions \eqref{assump: lip and alpha beta} and \eqref{assump: dissipitive}. Finally, to verify the third condition, we follow \cite{Maxim17} and construct the Lyapunov function 
$V(\bbeta^{(1)},\bbeta^{(2)}):=\exp\left\{A/4 \cdot \left(\frac{\|\bbeta^{(1)}\|^2}{\tau^{(1)}}+\frac{\|\bbeta^{(2)}\|^2}{\tau^{(2)}}\right)\right\}$. From the dissipitive assumption \ref{assump: dissipitive}, $V(\bbeta^{(1)}, \bbeta^{(2)})$ satisfies the third condition because
\begin{equation}
\begin{split}
    &\quad\cL(V(\bbeta^{(1)}, \bbeta^{(2)}))\\
    &=\Bigg(\frac{A}{2\tau^{(1)}}+\frac{A}{2\tau^{(2)}}+\frac{A^2}{4{\tau^{(1)}}^2}\|\bbeta^{(1)}\|^2+\frac{A^2}{4{\tau^{(2)}}^2}\|\bbeta^{(2)}\|^2\\
    &\qquad\qquad\qquad\qquad-\frac{A}{2{\tau^{(1)}}^2}\langle \bbeta^{(1)}, \nabla G(\bbeta^{(1)})-\frac{A}{2{\tau^{(2)}}^2}\langle \bbeta^{(1)}, \nabla G(\bbeta^{(2)})\rangle\Bigg) V(\bbeta^{(1)}, \bbeta^{(2)})\\
    &\leq \left(\frac{A}{2\tau^{(1)}}+\frac{A}{2\tau^{(2)}}+\frac{AB}{2{\tau^{(1)}}^2}+\frac{AB}{2{\tau^{(2)}}^2}-\frac{A^2}{4{\tau^{(1)}}^2}\|\bbeta^{(1)}\|^2-\frac{A^2}{4{\tau^{(2)}}^2}\|\bbeta^{(2)}\|^2\right) V(\bbeta^{(1)}, \bbeta^{(2)})\\
    &\leq \left(\kappa-\gamma (\|\bbeta^{(1)}\|^2+\|\bbeta^{(2)}\|^2)\right) V(\bbeta^{(1)}, \bbeta^{(2)}),\\
\end{split}
\end{equation}
where $\kappa=\frac{A}{2\tau^{(1)}}+\frac{A}{2\tau^{(2)}}+\frac{AB}{2{\tau^{(1)}}^2}+\frac{AB}{2{\tau^{(2)}}^2}$, and $\gamma=\frac{A^2}{4{\tau^{(1)}}^2}\land \frac{A^2}{4{\tau^{(2)}}^2}$. Therefore, the invariant measure $\pi$ satisfies a LSI with the constant
\begin{equation}
    c_{\text{LS}}=c_1+(c_2+2)c_p,
\end{equation}
where $c_1=\frac{2C}{\gamma}+\frac{2}{C}$ and $c_2=\frac{2C}{\gamma}\left(\kappa+\gamma\int_{\mathbb{R}^{2d}}( \|\bbeta^{(1)}\|^2 + \|\bbeta^{(2)}\|^2)\pi(d\bbeta^{(1)} d\bbeta^{(2)})\right)$. \qed

\end{proof}{}

We are now ready to prove the log-Sobolev inequality for invariant measure associated with the replica exchange Langevin diffusion \eqref{replica exchange}. We use a similar idea from \cite{chen2018accelerating} where they prove the Poincar\'{e} inequality for the invariant measure associated with the replica exchange Langevin diffusion \eqref{replica exchange} by analyzing the corresponding Dirichlet form. In particular, a larger Dirichlet form ensures a smaller log-Sobolev constant and hence results in a faster convergence in the relative entropy and Wasserstein distance.

\begin{lemma}[Accelerated exponential decay of $W_2$.  Formal statement of Lemma \ref{exponential decay_main}]\label{exponential decay_appendix}
Under assumptions \eqref{assump: lip and alpha beta} and \eqref{assump: dissipitive}, we have that the replica exchange Langevin diffusion converges exponentially fast to the invariant distribution $\pi$:
\begin{equation}
    W_2(\nu_t,\pi) \leq  D_0 e^{-k\eta(1+\delta_S)/c_{\text{LS}}},
\end{equation}
where $D_0=\sqrt{2c_{\text{LS}}D(\nu_0\|\pi)}$, $\delta_{S}:=\inf_{t>0}\frac{\cE_S(\sqrt{\frac{d\n_t}{d\pi}})}{\cE(\sqrt{\frac{d\n_t}{d\pi}})}-1$ is a non-negative constant depending on the swapping rate $S(\cd, \cd)$ and obtains $0$ only if $S(\cd, \cd)=0$.
\end{lemma}{}

\begin{proof} Given a smooth function $f:\mathbb{R}^d\times \mathbb{R}^d\rightarrow \mathbb{R}$, the infinitesimal generator $\cL_{ S}$ associated with the replica exchange Langevin diffusion with the (truncated) swapping rate $S$ follows
\begin{equation}
\label{generator_L_icml}
\begin{split}
    \cL_{S}f(\bbeta^{(1)}, \bbeta^{(2)})=&-\langle\nabla_{\bbeta^{(1)}}f(\bbeta^{(1)},\bbeta^{(2)}),\nabla L(\bbeta^{(1)})\rangle-\langle \nabla_{\bbeta^{(2)}}f(\bbeta^{(1)},\bbeta^{(2)}),\nabla L(\bbeta^{(2)})\rangle\\
    &
+\tau^{(1)}\Delta_{\bbeta^{(1)}}f(\bbeta^{(1)},\bbeta^{(2)})+\tau^{(2)}\Delta_{\bbeta^{(2)}}f(\bbeta^{(1)},\bbeta^{(2)})\\
&+ aS(\bbeta^{(1)},\bbeta^{(2)})\cd (f(\bbeta^{(2)},\bbeta^{(1)})-f(\bbeta^{(1)},\bbeta^{(2)})),
\end{split}
\end{equation}
where $\nabla_{\bbeta^{(h)}}$ and $\Delta_{\bbeta^{(h)}}$ are the gradient and the Laplace operators with respect to $\bbeta^{(h)}$, respectively. Next, we model the exponential decay of $W_2(\nu_t,\pi)$ using the Dirichlet form
\begin{equation}
\label{dirichlet}
    \cE_{S}(f)=\int \Gamma_{S}(f)d\pi,
\end{equation}
where $\Gamma_{S}(f)=\frac{1}{2}\cdot \cL_{S}(f^2) -f \cL_{S}(f)$ is the Carr\'{e} du Champ operator. In particular for the first term $\frac{1}{2}\cL_{S}(f^2)$, we have
\begin{equation}
\label{half_carre_icml}
\begin{split}
    &\ \ \ \frac{1}{2}\cL_{S}(f(\bbeta^{(1)}, \bbeta^{(2)})^2)\\
    =&-\langle f(\bbeta^{(1)}, \bbeta^{(2)})\nabla_{\bbeta^{(1)}} f(\bbeta^{(1)}, \bbeta^{(2)}) , \nabla_{\bbeta^{(1)}} L(\bbeta^{(1)})\rangle+\tau^{(1)} \|  \nabla_{\bbeta^{(1)}}f(\bbeta^{(1)}, \bbeta^{(2)})\|  ^2 \\
    & \ \ \ \ \ \ \ \ \ \ \ \ \ \ \ \ \ \ \ \ \ \ \ \ \ \ \ \ \ \ \ \ \ \ \ \ \ \ \ \ \ \ \ \ \ \ \ \ \ \ \ \ \ \ \ \  + \tau^{(1)} f(\bbeta^{(1)}, \bbeta^{(2)}) \Delta_{\bbeta^{(1)}}f(\bbeta^{(1)}, \bbeta^{(2)})\\
        &-\langle f(\bbeta^{(1)}, \bbeta^{(2)})\nabla_{\bbeta^{(2)}} f(\bbeta^{(1)}, \bbeta^{(2)}) , \nabla_{\bbeta^{(2)}} L(\bbeta^{(2)})\rangle+\tau^{(2)} \|  \nabla_{\bbeta^{(2)}}f(\bbeta^{(1)}, \bbeta^{(2)})\|  ^2 \\
        & \ \ \ \ \ \ \ \ \ \ \ \ \ \ \ \ \ \ \ \ \ \ \ \ \ \ \ \ \ \ \ \ \ \ \ \ \ \ \ \ \ \ \ \ \ \ \ \ \ \ \ \ \ \ \ \ + \tau^{(2)} f(\bbeta^{(1)}, \bbeta^{(2)}) \Delta_{\bbeta^{(2)}}f(\bbeta^{(1)}, \bbeta^{(2)})\\
        &+\frac{a}{2}S(\bbeta^{(1)}, \bbeta^{(2)}) (f^2(\bbeta^{(2)},\bbeta^{(1)})-f^2(\bbeta^{(1)},\bbeta^{(2)})).
\end{split}
\end{equation}

Combining the definition of the Carr\'{e} du Champ operator, (\ref{generator_L_icml}) and (\ref{half_carre_icml}), we have
\begin{equation}
\label{carre_du_icml}
    \begin{split}
        &\Gamma_{S}(f(\bbeta^{(1)}, \bbeta^{(2)}))\\
        =&\frac{1}{2}\mathcal{L}_{S}(f^2(\bbeta^{(1)}, \bbeta^{(2)}))-f(\bbeta^{(1)}, \bbeta^{(2)})\mathcal{L}_{S}(f(\bbeta^{(1)}, \bbeta^{(2)}))\\
        =&\tau^{(1)} \|  \nabla_{\bbeta^{(1)}}f(\bbeta^{(1)}, \bbeta^{(2)})\|  ^2+\tau^{(2)} \|  \nabla_{\bbeta^{(2)}}f(\bbeta^{(1)}, \bbeta^{(2)})\|  ^2\\
        &\ \ \ \ \ \ +\frac{a}{2}{S}(\bbeta^{(1)}, \bbeta^{(2)}) (f(\bbeta^{(2)},\bbeta^{(1)})-f(\bbeta^{(1)},\bbeta^{(2)}))^2.
    \end{split}
\end{equation}

By comparing Eq.\eqref{carre_du_icml} with Eq.\eqref{dirichlet form}, the Dirichlet form associated with operator $\cL_{S}$ under the invariant measure $\pi$ has the form
\bea\label{dirichlet swap_icml}
\cE_{S}(f)=\cE(f)+\underbrace{\frac{a}{2}\int S(\bbeta^{(1)},\bbeta^{(2)})\cd (f(\bbeta^{(2)},\bbeta^{(1)})-f(\bbeta^{(1)},\bbeta^{(2)}))^2d\pi(\bbeta^{(1)},\bbeta^{(2)})}_{\text{acceleration}},
\eea
where $f\in\cC_0^2(\hR^{2d})$ corresponds to $\frac{d\nu_t}{d\pi(\bbeta^{(1)}, \bbeta^{(2)})}$, and the asymmetry of $\frac{\nu_t}{\pi(\bbeta^{(1)}, \bbeta^{(2)})}$ is critical in the acceleration effect \cite{chen2018accelerating}. Given two different temperatures $\tau^{(1)}$ and $\tau^{(2)}$, a non-trivial distribution $\pi$ and function $f$, the swapping rate $S(\bbeta^{(1)},\bbeta^{(2)})$ is positive for almost any $\bbeta^{(1)}, \bbeta^{(2)}\in\mathbb{R}^d$. As a result, the Dirichlet form associated with $\cL_{S}$ is strictly larger than $\cL$. Therefore, there exists a constant $\delta_{S}> 0$ depending on $S(\bbeta^{(1)},\bbeta^{(2)})$, such that $\delta_{S}=\inf_{t>0}\frac{\cE_S(\sqrt{\frac{d\n_t}{d\pi}})}{\cE(\sqrt{\frac{d\n_t}{d\pi}})}-1$. From Lemma \ref{LSI no swaping}, we have
\begin{equation}
    D(\n_t\|\pi)\le 2c_{\text{LS}}\cE(\sqrt{\frac{d\n_t}{d\pi}})\le 2c_{\text{LS}} \sup_t\frac{ \cE(\sqrt{\frac{d\n_t}{d\pi}})}{\cE_{S}(\sqrt{\frac{d\n_t}{d\pi}})}\cE_{S}(\sqrt{\frac{d\n_t}{d\pi}})= 2 \frac{c_{\text{LS}}}{1+\delta_{S}}\cE_{S}(\sqrt{\frac{d\n_t}{d\pi}}).
\end{equation}
Thus, we obtain the following log-Sobolev inequality for the unique invariant measure $\pi$  associated with replica exchange Langevin diffusion $\{\bbeta_t\}_{t\ge 0}$ and its corresponding Dirichlet form $\cE_{S}(\cd)$. In particular, the LSI constant $ \frac{c_{\text{LS}}}{1+\delta_{S}}$ in replica exchange Langevin diffusion with swapping rate $S(\cd, \cd)>0$ is strictly smaller than the LSI constant $c_{\text{LS}}$ in the replica exchange Langevin diffusion with swapping rate $S(\cd, \cd)=0$. By the exponential decay in entropy \cite{Bakry2014}[Theorem 5.2.1] and the tight log-Sobolev inequality in Lemma \ref{LSI no swaping}, we get that, for any $t\in[k\eta,(k+1)\eta)$, 
\begin{equation}
    D(\nu_t\|\pi)\leq D(\nu_0\|\pi) e^{-2t(1+\delta_S)/c_{\text{LS}}}\leq D(\m_0\|\pi) e^{-2k\eta(1+\delta_S)/c_{\text{LS}}}.
\end{equation}
Finally, we can estimate the term $W_2(\nu_t,\pi)$ by the
Otto-Villani theorem \cite{Bakry2014}[Theorem 9.6.1],
\begin{equation}
    W_2(\nu_t,\pi) \leq \sqrt{2 c_{\text{LS}} D(\nu_t\|\pi)}\leq \sqrt{2c_{\text{LS}}D(\m_0\|\pi)} e^{-k\eta(1+\delta_S)/c_{\text{LS}}}.
\end{equation}
\end{proof}

For more quantitative study on how large $\delta_{S}$ is on related problems, we refer interested readers to the study of spectral gaps in \cite{Holden18, jingdong, Futoshi2020}.

\section{Summary: Convergence of reSGLD}

Now that we have all the necessary ingredients in place, we are ready to derive the convergence of the distribution $\m_{k}$ to the invariant measure $\pi$ in terms of 2-Wasserstein distance,

\begin{theorem}[Convergence of reSGLD. Formal statement of Theorem \ref{convergence_reSGLD_main}]
\label{convergence_reSGLD_appendix}
Let the assumptions \eqref{assump: lip and alpha beta} and \eqref{assump: dissipitive} hold. For the unique invariant measure $\pi$ associated with the Markov diffusion process \eqref{replica exchange} and the distribution $\{\m_{k}\}_{k\ge 0}$ associated with the discrete dynamics $\{\widetilde \bbeta^{\eta}(k)\}_{k\ge 1}$, we have the following estimates, for $0\le k\in \mathbb N^{+}$ and the learning rate $\eta$ satisfying $0<\eta<1 \land a/C^2$, 
\begin{equation}
    W_2(\m_{k}, \pi) \le  D_0 e^{-k\eta(1+\delta_S)/c_{\text{LS}}}+\sqrt{\delta_1 \eta +  \delta_2 \max_{k}\hE[\|\bphi_k\|^2]+ \delta_3\max_{k}\sqrt{\hE\left[\mid\psi_{k}\mid^2\right]}}
\end{equation}
where $D_0=\sqrt{2c_{\text{LS}}D(\m_0\mid\pi)}$, $\delta_{S}:=\min_{k}\frac{\cE_S(\sqrt{\frac{d\m_k}{d\pi}})}{\cE(\sqrt{\frac{d\m_k}{d\pi}})}-1$ is a non-negative constant depending on the swapping rate $S(\cd, \cd)$ and obtains the minimum zero only if $S(\cd, \cd)=0$.
\end{theorem}{}

\begin{proof}
We reduce the estimates into the following two terms by using the triangle inequality,
\begin{equation}\label{w2 triangle}
    W_2(\m_{k}, \pi) \leq W_2(\m_{k}, \n_t) + W_2(\nu_t,\pi),\qq t\in[k\eta,(k+1)\eta).
\end{equation}
The first term $W_2(\m_{k}, \n_t)$ follows from the analysis of discretization error in Lemma.\ref{discretization_appendix_reSGLD}.
Recall the very definition of the $W_2(\cd,\cd)$ distance defined in (\ref{w2}). Thus, in order to control the distance $W_2(\m_{k},\nu_t)$, $t\in[k\eta,(k+1)\eta)$, we need to consider the diffusion process whose law give $\m_{k}$ and $\n_t$, respectively. Indeed, it is obvious that $\n_t=\cL(\bbeta_t)$ for $t\in[k\eta,(k+1)\eta)$. For the other measure $\m_k$, it follows that $\m_{k}=\tilde \n_{k\eta}$ for $t=k\eta$, where $\tilde\n_{k\eta}=\cL(\widetilde \bbeta_t^{\eta})$ is the probability measure associated with the continuous interpolation of reSGLD (\ref{resgld_2}). By Lemma.\ref{discretization_appendix_reSGLD}, we have that for $k\in\mathbb{N}$ and $ t\in [k\eta,(k+1)\eta)$, 
\begin{equation}
\begin{split}
    &\quad W_2(\m_{k}, \n_t)=W_2(\tilde\n_{k\eta}, \n_t) \\
    &\leq \sqrt{\hE[\sup_{0\le s\le t}\|\bbeta_s- \widetilde \bbeta_s^{\eta}\|^2]}\leq \sqrt{\delta_1 \eta + \delta_2 \max_{k}\hE[\|\bphi_k\|^2]+\delta_3\max_{k}\sqrt{\hE\left[\mid\psi_{k}\mid^2\right]}},
\end{split}
\end{equation}

Recall from the accelerated exponential decay of replica exchange Langevin diffusion in Lemma.\ref{exponential decay_appendix}, we have
\begin{equation}
    W_2(\nu_t,\pi)\leq \sqrt{2c_{\text{LS}}D(\nu_0\|\pi)} e^{-k\eta(1+\delta_S)/c_{\text{LS}}}= \sqrt{2c_{\text{LS}}D(\m_0\|\pi)} e^{-k\eta(1+\delta_S)/c_{\text{LS}}}.
\end{equation}

Combing the above two estimates completes the proof.
\qed
\end{proof}{}

\chapter{TECHNICAL PROOFS FOR CHAPTER \ref{vr_resgld_iclr}}

\section{Preliminaries}
\label{prelim}
\textbf{Notation} We denote the deterministic energy based on the parameter $\bbeta$ by $L(\bbeta)=\sum_{i=1}^N L(\bx_i\mid\bbeta)$ using the full dataset of size $N$. We denote the unbiased stochastic energy estimator by $\frac{N}{n}\sum_{i\in B} L(\bx_i\mid  \bbeta)$ using the mini-batch of data $B$ of size $n$. The same style of notations is also applicable to the gradient for consistency. We denote the Euclidean $L^2$ norm by $\|  \cdot\|  $. To prove the desired results, we need the following assumptions:
\begin{assumption}[Smoothness]\label{assump: lip and alpha beta_ICLR21}
The energy function $L(\bx_i\mid  \cdot)$ is $C_N$-smoothness if there exists a constant $C_N>0$ such that $\forall \bbeta_1,\bbeta_2\in\hR^d$, $i\in\{1,2,\cdots, N\}$, we have
\begin{equation}
\label{1st_smooth_condition}
    \|  \nabla L(\bx_i\mid  \bbeta_1)-\nabla L(\bx_i\mid  \bbeta_2)\|  \le C_N\|  \bbeta_1-\bbeta_2\|  .
\end{equation}
Note that the above condition further implies 
for a constant $C=NC_N$ and $\forall \bbeta_1,\bbeta_2\in\hR^d$, we have
\begin{equation}
\label{2nd_smooth_condition}
    \|  \nabla L(\bbeta_1)-\nabla L(\bbeta_2)\|  \le C\|  \bbeta_1-\bbeta_2\|,
\end{equation}
\end{assumption}{}
which recovers assumption \ref{assump: lip and alpha beta} in chapter \ref{chapter_appendix_icml20}.



\section{Exponential Accelerations via Variance Reduction}
\label{exp_acc}
\setcounter{lemma}{0}
\setcounter{theorem}{0}

We aim to build an efficient estimator to approximate the deterministic swapping rate $S(\bbeta^{(1)}, \bbeta^{(2)})$
\begin{equation}
\label{S_exact}
    S(\bbeta^{(1)}, \bbeta^{(2)})=e^{ \left(\frac{1}{\tau^{(1)}}-\frac{1}{\tau^{(2)}}\right)\left( \sum_{i=1}^N L(\bx_i\mid  \bbeta^{(1)})-\sum_{i=1}^N L(\bx_i\mid  \bbeta^{(2)})\right)}.
\end{equation}

In big data problems and deep learning, it is too expensive to evaluate the energy $\sum_{i=1}^N L(\bx_i\mid  \bbeta)$ for each $\bbeta$ for a large $N$. To handle the computational issues, a popular solution is to use the unbiased stochastic energy $\frac{N}{n}\sum_{i\in B} L(\bx_i\mid  \bbeta)$ for a random mini-batch data $B$ of size $n$. However, a n\"{a}ive replacement of $\sum_{i=1}^N L(\bx_i\mid  \bbeta)$ by $\frac{N}{n}\sum_{i\in B} L(\bx_i\mid  \bbeta)$ leads to a large bias to the swapping rate. To remove such a bias, we follow \cite{deng2020} and consider the corrected swapping rate
\begin{equation}
\begin{split}
    \widehat S(\bbeta^{(1)}, \bbeta^{(2)})&=e^{ \left(\frac{1}{\tau^{(1)}}-\frac{1}{\tau^{(2)}}\right)\left( \frac{N}{n}\sum_{i\in B} L(\bx_i\mid  \bbeta^{(1)})-\frac{N}{n}\sum_{i\in B} L(\bx_i\mid  \bbeta^{(2)})-\left(\frac{1}{\tau^{(1)}}-\frac{1}{\tau^{(2)}}\right)\widehat \sigma^2\right)},\\
\end{split}
\end{equation}
where $2\widehat\sigma^2$ denotes the variance of $\frac{N}{n}\sum_{i\in B} L(\bx_i\mid  \bbeta^{(1)})-\frac{N}{n}\sum_{i\in B} L(\bx_i\mid  \bbeta^{(2)})$. \footnote[2]{We only consider the case of $F=1$ in the stochastic swapping rate for ease of analysis.} Empirically, $\widehat \sigma^2$ is quite large, resulting in almost no swaps and insignificant accelerations. To propose more effective swaps, we consider the variance-reduced estimator
\begin{equation}
    \widetilde L(B_k\mid  \bbeta_k)=\frac{N}{n}\sum_{i\in B_k}\left( L(\bx_i\mid   \bbeta_k) - L\left(\bx_i\mid    \bbeta_{m\lfloor \frac{k}{m}\rfloor}\right) \right)+\sum_{i=1}^N L\left(\bx_i\mid    \bbeta_{m\lfloor \frac{k}{m}\rfloor}\right),
\end{equation}
where the control variate $\bbeta_{m\lfloor \frac{k}{m}\rfloor}$ is updated every $m$ iterations. Denote the variance of $ \widetilde L(B\mid  \bbeta^{(1)})- \widetilde L(B\mid  \bbeta^{(2)})$ by $\widetilde\sigma^2$. The variance-reduced stochastic swapping rate follows
\begin{equation}
\begin{split}
\label{vr_s}
    \widetilde S_{\eta, m, n}(\bbeta^{(1)}, \bbeta^{(2)})&=e^{ \left(\frac{1}{\tau^{(1)}}-\frac{1}{\tau^{(2)}}\right)\left( \widetilde L(B\mid  \bbeta^{(1)})- \widetilde L(B\mid  \bbeta^{(2)})-\left(\frac{1}{\tau^{(1)}}-\frac{1}{\tau^{(2)}}\right)\widetilde\sigma^2\right)}.\\
\end{split}
\end{equation}

Using the strategy of variance reduction, we can lay down the first result, which differs from the existing variance reduction methods in that we only conduct variance reduction in the energy estimator for the class of SGLD algorithms.
\begin{lemma}[Variance-reduced energy estimator. Formal statement of Lemma \ref{vr-estimator_main}]
\label{vr-estimator}
Under the smoothness
and dissipativity assumptions \ref{assump: lip and alpha beta_ICLR21} and \ref{assump: dissipitive}, the variance of the variance-reduced energy estimator $\widetilde L(B_{k}\mid  \bbeta_{k}^{(h)})$, where $h\in\{1,2\}$, is upper bounded by
\begin{equation}
    \Var\left(\widetilde L(B_{k}\mid  \bbeta_{k}^{(h)})\right)\leq \frac{m^2 \eta}{n}D_R^2\left( \frac{2\eta}{n} (2C^2\Psi_{d,\tau^{(2)}, C, a, b} +2Q^2)+4\tau^{(2)} d\right).
\end{equation}
where $D_R=CR+\max_{i\in\{1,2,\cdots, N\}} N \|  \nabla L(\bx_i\mid  \bbeta_{\star})\|  +\frac{CB}{A}$ and $R$ is the radius of a sufficiently large ball that contains $\bbeta_k^{(h)}$ for $h\in\{1,2\}$.
\end{lemma}

\begin{proof}
\begin{equation}
\label{var_1st}
    \footnotesize
    \begin{split}
      &\Var\left(\widetilde L(B_{k}\mid  \bbeta_{k}^{(h)})\right)\\
      =&\E\left[\left(\frac{N}{n}\sum_{i\in B_k}\left[ L(\bx_i\mid   \bbeta_{k}^{(h)}) - L\left(\bx_i\mid    \bbeta^{(h)}_{m\lfloor \frac{k}{m}\rfloor}\right) \right]+\sum_{j=1}^N L\left(\bx_j\mid    \bbeta^{(h)}_{m\lfloor \frac{k}{m}\rfloor}\right)-\sum_{j=1}^N L(\bx_j\mid   \bbeta^{(h)}_k)\right)^2\right]\\
      =&\E\left[\left(\frac{N}{n}\sum_{i\in B_k}\left[ L(\bx_i\mid   \bbeta_{k}^{(h)}) - L\left(\bx_i\mid    \bbeta^{(h)}_{m\lfloor \frac{k}{m}\rfloor}\right) +\frac{1}{N}\left(\sum_{j=1}^N L\left(\bx_j\mid    \bbeta^{(h)}_{m\lfloor \frac{k}{m}\rfloor}\right)-\sum_{j=1}^N L(\bx_j\mid   \bbeta^{(h)}_k)\right)\right]\right)^2\right]\\
      =&\frac{N^2}{n^2}\E\left[\left(\sum_{i\in B_k}\left[ L(\bx_i\mid   \bbeta_{k}^{(h)}) - L\left(\bx_i\mid    \bbeta^{(h)}_{m\lfloor \frac{k}{m}\rfloor}\right) +\frac{1}{N}\left(\sum_{j=1}^N L\left(\bx_j\mid    \bbeta^{(h)}_{m\lfloor \frac{k}{m}\rfloor}\right)-\sum_{j=1}^N L(\bx_j\mid   \bbeta^{(h)}_k)\right)\right]\right)^2\right]\\
      =&\frac{N^2}{n^2}\sum_{i\in B_k}\E\left[\left( L(\bx_i\mid   \bbeta_{k}^{(h)}) - L\left(\bx_i\mid    \bbeta^{(h)}_{m\lfloor \frac{k}{m}\rfloor}\right) -\frac{1}{N}\left[\sum_{j=1}^N L(\bx_j\mid   \bbeta_k^{(h)})-\sum_{j=1}^N L\left(\bx_j\mid    \bbeta^{(h)}_{m\lfloor \frac{k}{m}\rfloor}\right)\right]\right)^2\right]\\
      \leq & \frac{N^2}{n^2}\sum_{i\in B_k}\E\left[\left( L(\bx_i\mid   \bbeta_{k}^{(h)}) - L\left(\bx_i\mid    \bbeta^{(h)}_{m\lfloor \frac{k}{m}\rfloor}\right)\right)^2\right]\\
      \leq & \frac{D_R^2}{n}\E\left[\left\|  \bbeta_{k}^{(h)}-\bbeta^{(h)}_{m\lfloor \frac{k}{m}\rfloor}\right\|  ^2\right],
    \end{split}
\end{equation}

where the last equality follows from the fact that $\E[(\sum_{i=1}^n x_i)^2]=\sum_{i=1}^n \E[x_i^2]$ for independent variables $\{x_i\}_{i=1}^n$ with mean 0. The first inequality follows from $\E[(x-\E[x])^2]\leq \E[x^2]$ and the last inequality follows from Lemma \ref{local_smooth}, where $D_R=CR+\max_{i\in\{1,2,\cdots, N\}} N \|  \nabla L(\bx_i\mid  \bbeta_{\star})\|  +\frac{CB}{A}$ and $R$ is the radius of a sufficiently large ball that contains $\bbeta_k^{(h)}$ for $h\in\{1,2\}$.

Next, we bound $\E\left[\left\|  \bbeta_{k}^{(h)}-\bbeta^{(h)}_{m\lfloor \frac{k}{m}\rfloor}\right\|  ^2\right]$ as follows 
\begin{equation}
\label{var_2nd}
\small
    \E\left[\left\|  \bbeta_{k}^{(h)}-\bbeta^{(h)}_{m\lfloor \frac{k}{m}\rfloor}\right\|  ^2\right]\leq \E\left[\left\|  \sum_{j=m\lfloor \frac{k}{m}\rfloor}^{k-1} (\bbeta_{j+1}^{(h)}-\bbeta_{j}^{(h)})\right\|  ^2\right]\leq m\sum_{j=m\lfloor \frac{k}{m}\rfloor}^{k-1}\E\left[\left\|   (\bbeta_{j+1}^{(h)}-\bbeta_{j}^{(h)})\right\|  ^2\right].
\end{equation}
For each term, we have the following bound
\begin{equation}
\label{var_3rd}
\begin{split}
    \E\left[\left\|   \bbeta_{j+1}^{(h)}-\bbeta_{j}^{(h)}\right\|  ^2\right]
    =&\E\left[\left\|  \eta \frac{N}{n}\sum_{i\in B_k}\nabla L(\bx_i\mid  \bbeta_{k}^{(h)})+\sqrt{2\eta\tau^{(h)}}\bxi_k\right\|  ^2\right]\\
    \leq & \frac{2\eta^2 N^2}{n^2}\sum_{i\in B_k} \E\left[\left\|  \nabla L(\bx_i\mid  \bbeta_{k}^{(h)})\right\|  ^2\right]+4\eta\tau^{(2)} d\\
    \leq & \frac{2\eta^2}{n} (2C^2 \E[\|  \bbeta_k^{(h)}\|  ^2]+2Q^2)+4\eta\tau^{(2)} d\\
    \leq & \frac{2\eta^2}{n} (2C^2\Psi_{d,\tau^{(2)}, C, a, b} +2Q^2)+4\eta\tau^{(2)} d,\\
\end{split}
\end{equation}
where the first inequality follows by $\E[\|  a+b\|  ^2]\leq 2\E[\|  a\|  ^2]+2\E[\|  b\|  ^2]$, the i.i.d of the data points and $\tau^{(1)}\leq \tau^{(2)}$ for $h\in\{1,2\}$; the second inequality follows by Lemma \ref{grad_bound}; the last inequality follows from Lemma \ref{Uniform_bound}.

Combining (\ref{var_1st}), (\ref{var_2nd}) and (\ref{var_3rd}), we have
\begin{equation}
    \Var\left(\widetilde L(B_{k}\mid  \bbeta_{k}^{(h)})\right)\leq \frac{m^2 \eta}{n}D_R^2\left( \frac{2\eta}{n} (2C^2\Psi_{d,\tau^{(2)}, C, a, b} +2Q^2)+4\tau^{(2)} d\right).
\end{equation}
\qed
\end{proof}

Since $\Var\left(\widetilde L(B_{k}\mid  \bbeta_{k}^{(h)})\right)\leq \Var\left(\frac{N}{n}\sum_{i\in B}L(\bx_i\mid   \bbeta_k)\right) +\Var\left(\frac{N}{n}\sum_{i\in B} L\left(\bx_i\mid    \bbeta_{m\lfloor \frac{k}{m}\rfloor}\right)\right)$ by definition,  $\Var\left(\widetilde L(B_{k}\mid  \bbeta_{k}^{(h)})\right)$ is upper bounded by $\mathcal{O}\left(\min\{\widehat\sigma^2, \frac{m^2 \eta}{n}\}\right)$, which becomes much smaller using a small learning rate $\eta$, a shorter period $m$ and a large batch size $n$. 

To satisfy the (stochastic) reversibility condition, we consider the truncated swapping rate $\min\{1, \widetilde S_{\eta, m, n}(\bbeta^{(1)}, \bbeta^{(2)})\}$, which targets the same invariant distribution (see section 3.1 \cite{Matias19} for details). We can show that the swapping rate may even decrease exponentially as the variance increases. 

\begin{lemma}[Variance reduction for larger swapping rates. Formal statement of Lemma \ref{exp_S_main_body}] \label{exp_S} Given a large enough batch size $n$, the variance-reduced energy estimator $\widetilde L(B_{k}\mid  \bbeta_{k}^{(h)})$ yields a truncated swapping rate that satisfies
\begin{equation}
     \E[\min\{1, \widetilde S_{\eta, m, n}(\bbeta^{(1)}, \bbeta^{(2)})\}]\approx\min\Big\{1, S(\bbeta^{(1)}, \bbeta^{(2)})\left(\mathcal{O}\left(\frac{1}{n^2}\right)+e^{-\mathcal{O}\left(\frac{m^2\eta}{n}+\frac{1}{n^2}\right)}\right)\Big\}.
\end{equation}

\end{lemma}

\begin{proof}

By central limit theorem, the energy estimator $\frac{N}{n}\sum_{i\in B} L(\bx_i\mid  \bbeta_k)$ converges in distribution to a normal distributions as the batch size $n$ goes to infinity. In what follows, the variance-reduced estimator $\widetilde L(B_k\mid  \bbeta_k)$ also converges to a normal distribution, where the corresponding estimator is denoted by $\mathbb{\widetilde L}(B_k\mid  \bbeta_k)$. Now the swapping rate $\mathbb{S}_{\eta, m, n}(\cdot, \cdot)$ based on normal estimators follows
\begin{equation}
\begin{split}
\label{vr_s_normal}
    \mathbb{S}_{\eta, m, n}(\bbeta^{(1)}, \bbeta^{(2)})&=e^{ \left(\frac{1}{\tau^{(1)}}-\frac{1}{\tau^{(2)}}\right)\left( \mathbb{\widetilde L}(B\mid  \bbeta^{(1)})- \mathbb{\widetilde L}(B\mid  \bbeta^{(2)})-\left(\frac{1}{\tau^{(1)}}-\frac{1}{\tau^{(2)}}\right)\bar \sigma^2\right)},\\
\end{split}
\end{equation}
where $2\bar \sigma^2$ denotes the variance of $\mathbb{\widetilde L}(B\mid  \bbeta^{(1)})- \mathbb{\widetilde L}(B\mid  \bbeta^{(2)})$. Note that $\mathbb{S}_{\eta, m, n}(\bbeta^{(1)}, \bbeta^{(2)})$ follows a log-normal distribution with mean $\log S(\bbeta^{(1)}, \bbeta^{(2)})-\left(\frac{1}{\tau^{(1)}}-\frac{1}{\tau^{(2)}}\right)^2 \bar \sigma^2$ and variance $2\left(\frac{1}{\tau^{(1)}}-\frac{1}{\tau^{(2)}}\right)^2\bar \sigma^2$ on the log-scale, and $S(\bbeta^{(1)}, \bbeta^{(2)})$ is the deterministic swapping rate defined in (\ref{S_exact}). Applying Lemma \ref{exponential_dependence}, we have
\begin{equation}
\begin{split}
    \E[\min\{1, \mathbb{S}_{\eta, m, n}(\bbeta^{(1)}, \bbeta^{(2)})\}]=\mathcal{O}\left(S(\bbeta^{(1)}, \bbeta^{(2)})\exp\left\{-\frac{\left(\frac{1}{\tau^{(1)}}-\frac{1}{\tau^{(2)}}\right)^2\bar \sigma^2}{4}\right\}\right).
\end{split}
\end{equation} 

Moreover, $2\bar \sigma^2$ differs from $2\widetilde\sigma^2$, the variance of $\widetilde L(B\mid  \bbeta^{(1)})-\widetilde L(B\mid  \bbeta^{(2)})$, by at most a bias of $\mathcal{O}(\frac{1}{n^2})$ according to the estimate of the third term of (S2) in \cite{Matias19} and $2\widetilde\sigma^2\leq \Var\left(\widetilde L(B_{k}\mid  \bbeta_{k}^{(1)})\right) +\Var\left(\widetilde L(B_{k}\mid  \bbeta_{k}^{(2)})\right)$, where both $\Var\left(\widetilde L(B_{k}\mid  \bbeta_{k}^{(1)})\right)$ and $\Var\left(\widetilde L(B_{k}\mid  \bbeta_{k}^{(2)})\right)$ are upper bounded by $\frac{m^2 \eta}{n}D_R^2\left( \frac{2\eta}{n} (2C^2\Psi_{d,\tau^{(2)}, C, a, b} +2Q^2)+4\tau d\right)$ by Lemma \ref{vr-estimator}, it follows that
\begin{equation}
\label{normal_truncate}
\begin{split}
    \E[\min\{1, \mathbb{S}_{\eta, m, n}(\bbeta^{(1)}, \bbeta^{(2)})\}] \leq S(\bbeta^{(1)}, \bbeta^{(2)}) e^{-\mathcal{O}\left(\frac{m^2\eta}{n}+\frac{1}{n^2}\right)}.
\end{split}
\end{equation}

Applying $\min\{1,\mathbb{A}+\mathbb{B}\}\leq \min\{1,\mathbb{A}\}+\mid  \mathbb{B}\mid  $, we have
\begin{equation}
\label{target_eq1}
\begin{split}
    &\E[\min\{1, \widetilde S_{\eta, m, n}(\bbeta^{(1)}, \bbeta^{(2)})\}]\\
    = & \E\big[\min
    \big\{1, \underbrace{\widetilde S_{\eta, m, n}(\bbeta^{(1)}, \bbeta^{(2)})-\mathbb{S}_{\eta, m, n}(\bbeta^{(1)}, \bbeta^{(2)})}_{\mathbb{B}}+\underbrace{\mathbb{S}_{\eta, m, n}(\bbeta^{(1)}, \bbeta^{(2)})}_{\mathbb{A}}\big\}\big]\\
    \leq & \underbrace{\E\left[\mid    \widetilde S_{\eta, m, n}(\bbeta^{(1)}, \bbeta^{(2)})-\mathbb{S}_{\eta, m, n}(\bbeta^{(1)}, \bbeta^{(2)})\mid   \right]}_{\mathcal{I}} + \underbrace{\E[\min\{1, \mathbb{S}_{\eta, m, n}(\bbeta^{(1)}, \bbeta^{(2)})\}]}_{\text{see formula\ }  (\ref{normal_truncate})} \\
\end{split}
\end{equation}

By the triangle inequality, we can further upper bound the first term $\mathcal{I}$ 
\begin{equation}
\label{target_eq2}
\begin{split}
    &\ \ \ \ \ \E\left[\mid   \widetilde S_{\eta, m, n}(\bbeta^{(1)}, \bbeta^{(2)})-\mathbb{S}_{\eta, m, n}(\bbeta^{(1)}, \bbeta^{(2)}\mid   \right]\\
    &\leq \underbrace{\mid   \E[\widetilde S_{\eta, m, n}(\bbeta^{(1)}, \bbeta^{(2)})]-S(\bbeta^{(1)}, \bbeta^{(2)})\mid   }_{\mathcal{I}_1}+\underbrace{\mid   S(\bbeta^{(1)}, \bbeta^{(2)})-\E[\mathbb{S}_{\eta, m, n}(\bbeta^{(1)}, \bbeta^{(2)})]\mid   }_{\mathcal{I}_2}\\
    &= S(\bbeta^{(1)}, \bbeta^{(2)}) \mathcal{O}\left(\frac{1}{n^2}\right)+S(\bbeta^{(1)}, \bbeta^{(2)}) \mathcal{O}\left(\frac{1}{n^2}\right),
\end{split}
\end{equation}
where $\mathcal{I}_1$ and $\mathcal{I}_2$ follow from the proof of S1 without and with normality assumptions, respectively \cite{Matias19}.

Combining (\ref{target_eq1}) and (\ref{target_eq2}), we have
\begin{equation}
 \E[\min\{1, \widetilde S_{\eta, m, n}(\bbeta^{(1)}, \bbeta^{(2)})\}]\approx \min\Big\{1, S(\bbeta^{(1)}, \bbeta^{(2)}) \left(\mathcal{O}\left(\frac{1}{n^2}\right)+ e^{-\mathcal{O}\left(\frac{m^2\eta}{n}+\frac{1}{n^2}\right)}\right)\Big\}.
\end{equation}

This means that reducing the update period $m$ (more frequent update the of control variable), the learning rate $\eta$ and the batch size $n$ significantly increases $\min\{1, \widetilde S_{\eta, m, n}\}$ on average.\qed
\end{proof}

The above lemma shows a potential to exponentially increase the number of effective swaps  via variance reduction under the same intensity $a$. Next, similar to Lemma \ref{exponential decay_appendix}, we show the impact of variance reduction in speeding up the exponential convergence of the corresponding continuous-time replica exchange Langevin diffusion.

\begin{theorem}[Exponential convergence. Formal statement of Theorem \ref{exponential decay_main_body_iclr}]\label{exponential decay_iclr}
Under the smoothness
and dissipativity assumptions \ref{assump: lip and alpha beta_ICLR21} and \ref{assump: dissipitive}, the replica exchange Langevin diffusion associated with the variance-reduced stochastic swapping rates $S_{\eta, m, n}(\cd, \cd)=\min\{1, \widetilde S_{\eta, m, n}(\cd, \cd)\}$ converges exponential fast to the invariant distribution $\pi$ given a smaller learning rate $\eta$, a smaller $m$ or a larger batch size $n$:
\begin{equation}
    W_2(\nu_t,\pi) \leq  D_0 \exp\left\{-t\left(1+\delta_{ S_{\eta, m, n}}\right)/c_{\text{LS}}\right\},
\end{equation}
where $D_0=\sqrt{2c_{\text{LS}}D(\nu_0\|  \pi)}$, $\delta_{ S_{\eta, m, n}}:=\inf_{t>0}\frac{\cE_{ S_{\eta, m, n}}(\sqrt{\frac{d\n_t}{d\pi}})}{\cE(\sqrt{\frac{d\n_t}{d\pi}})}-1$ is a non-negative constant depending on the truncated stochastic swapping rate $S_{\eta, m, n}(\cd, \cd)$ and increases with a smaller learning rate $\eta$, a shorter period $m$ and a large batch size $n$. $c_{\text{LS}}$ is the standard constant of the log-Sobolev inequality asscoiated with the Dirichlet form for replica exchange Langevin diffusion without swaps.
\end{theorem}{}

\begin{proof} Given a smooth function $f:\mathbb{R}^d\times \mathbb{R}^d\rightarrow \mathbb{R}$, the infinitesimal generator $\cL_{ S_{\eta, m, n}}$ associated with the replica exchange Langevin diffusion with the (truncated) swapping rate $ S_{\eta, m, n}=\min\{1, \widetilde S_{\eta, m, n}\}$ follows
\begin{equation}
\label{generator_L}
\begin{split}
    &\ \ \ \cL_{S_{\eta, m, n}}f(\bbeta^{(1)}, \bbeta^{(2)})\\
    =&-\langle\nabla_{\bbeta^{(1)}}f(\bbeta^{(1)},\bbeta^{(2)}),\nabla L(\bbeta^{(1)})\rangle-\langle \nabla_{\bbeta^{(2)}}f(\bbeta^{(1)},\bbeta^{(2)}),\nabla L(\bbeta^{(2)})\rangle\\
    &
+\tau^{(1)}\Delta_{\bbeta^{(1)}}f(\bbeta^{(1)},\bbeta^{(2)})+\tau^{(2)}\Delta_{\bbeta^{(2)}}f(\bbeta^{(1)},\bbeta^{(2)})\\
&+ aS_{\eta, m, n}(\bbeta^{(1)},\bbeta^{(2)})\cd (f(\bbeta^{(2)},\bbeta^{(1)})-f(\bbeta^{(1)},\bbeta^{(2)})),
\end{split}
\end{equation}
where $\nabla_{\bbeta^{(h)}}$ and $\Delta_{\bbeta^{(h)}}$ are the gradient and the Laplace operators with respect to $\bbeta^{(h)}$, respectively. Following the same technique in Eq.\eqref{dirichlet swap_icml}, the Dirichlet form associated with the operator $\cL_{ S_{\eta, m, n}}$ follows that
\begin{equation}\label{dirichlet swap}
\small
\begin{split}
    \cE_{S_{\eta, m, n}}(f)=&\underbrace{\int \Big(\tau^{(1)}\|  \nabla_{\bbeta^{(1)}}f(\bbeta^{(1)}, \bbeta^{(2)})\|  ^2+\tau^{(2)}\|  \nabla_{\bbeta^{(2)}}f(\bbeta^{(1)}, \bbeta^{(2)})\|  ^2 \Big)d\pi(\bbeta^{(1)},\bbeta^{(2)})}_{\text{vanilla term } \cE(f)}\\
    &\ +\underbrace{\frac{a}{2}\int S_{\eta, m, n}(\bbeta^{(1)},\bbeta^{(2)})\cd (f(\bbeta^{(2)},\bbeta^{(1)})-f(\bbeta^{(1)},\bbeta^{(2)}))^2d\pi(\bbeta^{(1)},\bbeta^{(2)})}_{\text{acceleration term}},
\end{split}
\end{equation}
where $f$ corresponds to $\frac{d\nu_t}{d\pi(\bbeta^{(1)}, \bbeta^{(2)})}$. By Lemma \ref{exponential decay_appendix}, there exists a constant $\delta_{S_{\eta, m, n}}=\inf_{t>0}\frac{\cE_{S_{\eta, m, n}}(\sqrt{\frac{d\n_t}{d\pi}})}{\cE(\sqrt{\frac{d\n_t}{d\pi}})}-1$ depending on $S_{\eta, m, n}$ that satisfies the following log-Sobolev inequality for the unique invariant measure $\pi$  associated with variance-reduced replica exchange Langevin diffusion $\{\bbeta_t\}_{t\ge 0}$
\begin{equation*}
    D(\n_t\|  \pi)\le 2 \frac{c_{\text{LS}}}{1+\delta_{S_{\eta, m, n}}}\cE_{S_{\eta, m, n}}(\sqrt{\frac{d\n_t}{d\pi}}),
\end{equation*}
where $\delta_{S_{\eta, m, n}}$ increases rapidly with the swapping rate $S_{\eta, m, n}$.
By virtue of the exponential decay of entropy \cite{Bakry2014}, we have
\begin{equation*}
    D(\nu_t\|  \pi)\leq D(\nu_0\|  \pi) e^{-2t(1+\delta_{S_{\eta, m, n}})/c_{\text{LS}}},
\end{equation*}
where $c_{\text{LS}}$ is the standard constant of the log-Sobolev inequality asscoiated with the Dirichlet form for replica exchange Langevin diffusion without swaps (Lemma 4 as in \cite{deng2020}). Next, we upper bound $W_2(\nu_t,\pi)$ by the
Otto-Villani theorem \cite{Bakry2014}
\begin{equation*}
    W_2(\nu_t,\pi) \leq \sqrt{2 c_{\text{LS}} D(\nu_t\|  \pi)}\leq \sqrt{2c_{\text{LS}}D(\m_0\|  \pi)} e^{-t\left(1+\delta_{S_{\eta, m, n}}\right)/c_{\text{LS}}},
\end{equation*}
where $\delta_{S_{\eta, m, n}}>0$ depends on the learning rate $\eta$, the period $m$ and the batch size $n$. \qed

\end{proof}

In the above analysis, we have established the relation that $\delta_{S_{\eta, m, n}}=\inf_{t>0}\frac{\cE_{S_{\eta, m, n}}(\sqrt{\frac{d\n_t}{d\pi}})}{\cE(\sqrt{\frac{d\n_t}{d\pi}})}-1$ depending on $S_{\eta, m, n}$ may increase significantly with a smaller learning rate $\eta$, a shorter period $m$ and a large batch size $n$. We leave the quantitative study of $\delta_{S_{\eta, m, n}}$ based on spectral gaps for future works \cite{Holden18, jingdong, Futoshi2020}.

\section{Proof of Technical Lemmas}
\setcounter{lemma}{0}

\begin{lemma}[Local Lipschitz continuity]\label{local_smooth}
Given a $d$-dimensional centered ball $U$ of radius $R$, $L(\cdot)$ is $D_R$-Lipschitz continuous in that $\mid  L(\bx_i\mid  \bbeta_1)- L(\bx_i\mid  \bbeta_2)\mid  \le \frac{D_R}{N}\|  \bbeta_1-\bbeta_2\|  $ for $\forall \bbeta_1, \bbeta_2\in U$ and any $i\in\{1,2,\cdots, N\}$,
where $D_R=CR+\max_{i\in\{1,2,\cdots, N\}} N \|  \nabla L(\bx_i\mid  \bbeta_{\star})\|  +\frac{CB}{A}$.
\end{lemma}{}  
\begin{proof}

For any $\bbeta_1, \bbeta_2\in U$, there exists $\bbeta_3\in U$ that satisfies the mean-value theorem such that
\begin{equation*}
    \mid  L(\bx_i\mid  \bbeta_1) - L(\bx_i\mid  \bbeta_2)\mid  =\langle\nabla L(\bx_i\mid  \bbeta_3), \bbeta_1-\bbeta_2\rangle \leq \|  \nabla L(\bx_i\mid  \bbeta_3)\|  \cdot \|  \bbeta_1-\bbeta_2\|  ,
\end{equation*}

Moreover, by Lemma \ref{grad_bound}, we have
\begin{equation*}
    \mid  L(\bx_i\mid  \bbeta_1) - L(\bx_i\mid  \bbeta_2)\mid  \leq \|  \nabla L(\bx_i\mid  \bbeta_3)\|  \cdot \|  \bbeta_1-\bbeta_2\|  \leq \frac{CR+Q}{N}\|  \bbeta_1-\bbeta_2\|  .\qed
\end{equation*}

\end{proof}

\begin{lemma}
\label{grad_bound}
Under the smoothness
and dissipativity assumptions \ref{assump: lip and alpha beta_ICLR21}, \ref{assump: dissipitive}, for any $\bbeta\in\mathbb{R}^d$, it follows that
\begin{equation}
   \|  \nabla L(\bx_i\mid  \bbeta)\|  \leq \frac{C}{N}\|  \bbeta\|  +\frac{Q}{N}.
\end{equation}
where $Q=\max_{i\in\{1,2,\cdots, N\}} N \|  \nabla L(\bx_i\mid  \bbeta_{\star})\|  +\frac{bC}{a}$.
\end{lemma}

\begin{proof}
According to the dissipativity assumption, we have 
\begin{equation}
    \la \bbeta_{\star},\nabla L(\bbeta_{\star})\rangle \ge a\|  \bbeta^{\star}\|  ^2-b,
\end{equation}
where $\bbeta_{\star}$ is a minimizer of $\nabla L(\cdot)$ such that $\nabla L(\bbeta_{\star})=0$. In what follows, we have $\|  \bbeta_{\star}\|  \leq \frac{b}{a}$.

Combining the triangle inequality and the smoothness assumption \ref{assump: lip and alpha beta_ICLR21}, we have
\begin{equation}
\begin{split}
    \|  \nabla L(\bx_i\mid  \bbeta)\|  \leq & C_N\|  \bbeta-\bbeta_{\star}\|   + \|  \nabla L(\bx_i\mid  \bbeta_{\star})\|  \leq C_N\|  \bbeta\|   + \frac{C_N b}{a} + \|  \nabla L(\bx_i\mid  \bbeta_{\star})\|  .
\end{split}
\end{equation}

Setting $C_N=\frac{C}{N}$ as in (\ref{2nd_smooth_condition}) and $Q=\max_{i\in\{1,2,\cdots, N\}} \|  \nabla L(\bx_i\mid  \bbeta_{\star})\|  +\frac{b C}{a}$ completes the proof. 
\qed
\end{proof}

The following lemma is majorly adapted from Lemma C.2 of \cite{chen2018accelerating}, except that the corresponding constant in the RHS of (C.38) is slightly changed to account for the stochastic noise. A similar technique has been established in Lemma 3 of \cite{Maxim17}.
\begin{lemma}[Uniform $L^2$ bounds on replica exchange SGLD]
\label{Uniform_bound}
Under the smoothness
and dissipativity assumptions \ref{assump: lip and alpha beta_ICLR21}, \ref{assump: dissipitive}.  Given a small enough learning rate $\eta\in (0, 1\vee \frac{A}{C^2})$, there exists a positive constant $\Psi_{d,\tau^{(2)}, C, A, B}<\infty$ such that $\sup_{k\geq 1} \E[\|  \bbeta_{k}\|  ^2] < \Psi_{d,\tau^{(2)}, C, A, B}$.
\end{lemma}

\begin{lemma}[Exponential dependence on the variance]
\label{exponential_dependence}
Assume $S$ is a log-normal distribution with mean $u-\frac{1}{2}\sigma^2$ and variance $\sigma^2$ on the log scale. Then $\E[\min(1, S)]=\mathcal{O}(e^{u-\frac{\sigma^2}{8}})$, which is {exponentially smaller} given a large variance $\sigma^2$.
\end{lemma}

\begin{proof} For a log-normal distribution $S$ with mean $u-\frac{1}{2}\sigma^2$ and variance $\sigma^2$ on the log scale, the probability density $f_S(S)$ follows that $\frac{1}{S\sqrt{2\pi \sigma^2}}\exp\left\{-\frac{(\log S-u+\frac{1}{2}\sigma^2)^2}{2\sigma^2}\right\}$. In what follows, we have
\begin{equation*}
\small
\begin{split}
    \E[\min(1, S)]=&\int_{0}^{\infty} \min(1, S) f_S(S)dS=\int_{0}^{\infty} \min(1, S) \frac{1}{S\sqrt{2\pi \sigma^2}}\exp\left\{-\frac{(\log S-u+\frac{1}{2}\sigma^2)^2}{2\sigma^2} \right\}dS\\
\end{split}
\end{equation*}

By change of variable $y=\frac{\log S-u+\frac{1}{2}\sigma^2}{\sigma}$ where $S=e^{\sigma y+u-\frac{1}{2}\sigma^2}$ and $y=-\frac{u}{\sigma}+\frac{\sigma}{2}$ given $S=1$, it follows that
\begin{equation*}
\begin{split}
\small
    &\E[\min(1, S)]\\
    =&\int_{0}^{1} S \frac{1}{S\sqrt{2\pi \sigma^2}}\exp\left\{-\frac{(\log S-u+\frac{1}{2}\sigma^2)^2}{2\sigma^2} \right\}dS+\int_{1}^{\infty} \frac{1}{S\sqrt{2\pi \sigma^2}}\exp\left\{-\frac{(\log S-u+\frac{1}{2}\sigma^2)^2}{2\sigma^2} \right\}dS\\
    =&\int_{-\infty}^{-\frac{u}{\sigma}+\frac{\sigma}{2}} \frac{1}{\sqrt{2\pi\sigma^2}} e^{-\frac{y^2}{2}}\sigma e^{u-\frac{1}{2}\sigma^2+\sigma y}dy+\int_{-\frac{u}{\sigma}+\frac{\sigma}{2}}^{\infty}\frac{1}{\sqrt{2\pi\sigma^2}} e^{-\sigma y-u+\frac{1}{2}\sigma^2}e^{-\frac{y^2}{2}}\sigma e^{u-\frac{1}{2}\sigma^2+\sigma y}dy \\
    =&e^u\int_{-\infty}^{-\frac{u}{\sigma}+\frac{\sigma}{2}} \frac{1}{\sqrt{2\pi}} e^{-\frac{(y-\sigma)^2}{2}}dy+\frac{1}{\sigma}\int_{-\frac{u}{\sigma}+\frac{\sigma}{2}}^{\infty} \frac{1}{\sqrt{2\pi}} e^{-\frac{y^2}{2}}dy\\
    =&e^u\int^{\infty}_{\frac{u}{\sigma}+\frac{\sigma}{2}} \frac{1}{\sqrt{2\pi}} e^{-\frac{z^2}{2}}dz+\frac{1}{\sigma}\int_{-\frac{u}{\sigma}+\frac{\sigma}{2}}^{\infty} \frac{1}{\sqrt{2\pi}} e^{-\frac{y^2}{2}}dy\\
    \leq & e^u\int^{\infty}_{-\frac{u}{\sigma}+\frac{\sigma}{2}} \frac{1}{\sqrt{2\pi}} e^{-\frac{z^2}{2}}dz+\frac{1}{\sigma}\int_{-\frac{u}{\sigma}+\frac{\sigma}{2}}^{\infty} \frac{1}{\sqrt{2\pi}} e^{-\frac{y^2}{2}}dy\\
    \leq &\left(e^u+\frac{1}{\sigma}\right)e^{-\frac{(-\frac{u}{\sigma}+\frac{\sigma}{2})^2}{2}}\lesssim e^{u-\frac{\sigma^2}{8}},\\
\end{split}
\end{equation*}
where the last equality follows from the change of variable $z=\sigma-y$ and the second last inequality follows from the exponential tail bound of the standard Gaussian variable $\mathbb{P}(y>\epsilon)\leq e^{\frac{-\epsilon^2}{2}}$.
\qed
\end{proof}

\chapter{TECHNICAL PROOFS FOR CHAPTER \ref{NRPT_uncertainty}}

\section{Background}

\subsection{Replica Exchange Stochastic Gradient Langevin Dynamics}
\label{reSGLD_appendix}

To approximate the replica exchange Langevin diffusion Eq.(\ref{sde_2_couple}) based on multiple particles $(\bbeta_{k+1}^{(1)}, \bbeta_{k+1}^{(2)}, \cdots, \bbeta_{k+1}^{(P)})$ in big data scenarios, replica exchange stochastic gradient Langevin dynamics (reSGLD) proposes the following numerical scheme:

\begin{equation}  
\begin{split}
\label{sde_2_couple_numerical}
\text{Exploration: }\left\{  
             \begin{array}{lr}  
             \footnotesize{\small{\bbeta_{k+1}^{(P)}=\bbeta_{k}^{(P)} - \eta^{(P)}\nabla \widetilde L(\bbeta_k^{(P)})}+\sqrt{2\eta^{(P)}\tau^{(P)}} \bxi_k^{(P)}}, \\  
              \\
              \footnotesize{\cdots}  \\
             \footnotesize{\bbeta_{k+1}^{(2)} =\bbeta_{k}^{(2)} \ - \eta^{(2)}\nabla \widetilde L(\bbeta_k^{(2)})}\ +\sqrt{2\eta^{(2)}\tau^{(2)}} \bxi_k^{(2)},\\
             \end{array}
\right. \qquad\qquad\qquad\qquad\  & \\
\text{Exploitation: } \quad\footnotesize{\bbeta_{k+1}^{(1)}=\bbeta_{k}^{(1)} \ - \eta^{(1)}\nabla \widetilde L(\bbeta_k^{(1)})\ \ +\sqrt{2\eta^{(1)}\tau^{(1)}} \bxi_k^{(1)},}\ \qquad\qquad\qquad\qquad\qquad\ \ &\\
\end{split}
\end{equation}
where $\eta^{(\cdot)}$ is the learning rate, $\bxi_k^{(\cdot)}$ is a standard $d$-dimensional Gaussian noise and each subprocess follows a stochastic gradient Langevin dynamics  (SGLD) \cite{Welling11}. Further assuming the energy normality assumption \ref{energy_normality} as follows
\begin{assumption}[Energy normality]\label{energy_normality}
The stochastic energy estimator for each chain follows a normal distribution with a fixed variance.
\end{assumption}

In what follows, there exists a $\sigma_p$ such that 
\begin{equation}\label{sigma_p}
    \widetilde L(\bbeta^{(p)})-\widetilde L(\bbeta^{(p+1)})\sim\mathcal{N}(L(\bbeta^{(p)})- L(\bbeta^{(p+1)}), 2\sigma_p^2), \text{ for any } \bbeta^{(p)}\sim \pi^{(p)},
\end{equation}
where $p\in\{1,2,\cdots, P-1\}$, in what follows, \cite{deng2020} proposed the bias-corrected swap function as follows 
\begin{equation}
\label{swap_rate_numerical}
    a\widetilde S(\bbeta_{k+1}^{(p)}, \bbeta_{k+1}^{(p+1)})=a\cdot\left(1\wedge e^{ \big(\frac{1}{\tau^{(p)}}-\frac{1}{\tau^{(p+1)}}\big)\big(\widetilde L(\bbeta_{k+1}^{(p)})-\widetilde L(\bbeta_{k+1}^{(p+1)})-\big(\frac{1}{\tau^{(p)}}-\frac{1}{\tau^{(p+1)}}\big)\sigma_p^2\big)}\right),
\end{equation}
where $p\in\{1,2,\cdots, P-1\}$, $\big(\frac{1}{\tau^{(p)}}-\frac{1}{\tau^{(p+1)}}\big)\sigma_p^2$ is a correction term to avoid the bias and the swap intensity $a$ can be then set to $\frac{1}{\min\{\eta^{(p)}, \eta^{(p+1)}\}}$ for convenience. Namely, given a particle pair at location $(\beta^{(p)}, \beta^{(p+1)})$ in the $k$-th iteration, the conditional probability of the swap follows that
\begin{equation*}
\begin{split}
    \mathbb{P}(\bbeta_{k+1}=(\beta^{(p+1)}, \beta^{(p)})\mid \bbeta_k=(\beta^{(p)}, \beta^{(p+1)}))&=\widetilde S(\beta^{(p)}, \beta^{(p+1)}),\\
    \mathbb{P}(\bbeta_{k+1}=(\beta^{(p)}, \beta^{(p+1)})\mid \bbeta_k=(\beta^{(p)}, \beta^{(p+1)}))&=1-\widetilde  S(\beta^{(p)}, \beta^{(p+1)}).\\
\end{split}
\end{equation*}

\subsection{Inhomogenous Swap Intensity via Different Learning Rates}
\label{diff_lr}
Assume the high-temperature process also applies $\eta^{(p+1)}\geq \eta^{(p)}>0$, where $p\in\{1,2\cdots, P-1\}$. At time $t$, the swap intensity can be interpreted as being $0$ in a time interval $[t, t+\eta^{(p+1)}-\eta^{(p)})$ and being $\frac{1}{\eta^{(p)}}$ in $[t+\eta^{(p+1)}-\eta^{(p)}, t+\eta^{(p+1)})$. Since the coupled Langevin diffusion process converges to the same joint distribution regardless of the swaps and the swap intensity $a$ varies in a bounded domain, the convergence of numerical schemes is not much affected except the numerical error.

\subsection{Connection to Non-convex Optimization}
\label{connection_2_non_convex}

Parallel to our work, a similar SGLD$\times$SGD framework was proposed by \cite{jingdong2} for non-convex optimization, where SGLD and SGD work as exploration and exploitation kernels, respectively. By contrast, our algorithm performs \emph{exactly in the opposite} for uncertainty approximation because SGLD is theoretically more appealing for the exploitations based on small learning rates instead of explorations, while the widely-adopted SGDs are quite attractive in exploration due to its user-friendly nature and ability in exploring wide optima.

If we manipulate the scheme to propose \emph{an exact swap} in each window instead of \emph{at most one} in the current version, the algorithm shows a better potential in non-convex optimization. In particular, a larger window size corresponds to a slower decay of temperatures in simulated annealing (SAA) \cite{Mangoubi18}. Such a mechanism yields a larger hitting probability to move into a sub-level set with lower energies (losses) and a better chance to hit the global optima. Nevertheless, the manipulated algorithm possess the natural of parallelism in cyclical fashions.

\subsection{Others}
\label{others}

\textbf{Swap time} refers to the communication time to conduct a swap in each attempt. For example, chain pair $(p, p+1)$ of ADJ requires to wait for the completion of chain pairs $(1, 2), (2,3), \cdots, (p-1, p)$ to attempt the swap and leads to swap time of $O(P)$; however, SEO, DEO, and DEO$_{\star}$ don't have this issue because in each iteration, only even or odd chain pairs are attempted to swap, hence the swap time is $O(1)$.

\textbf{Round trip time} refers to the time (stochastic variable) used in a round trip. A round trip is completed when a particle in the $p$-th chain, where $p\in[P]:=\{1, 2,\cdots, P\}$, hits the index boundary index at both $1$ and $P$ and returns back to its original index $p$.

\section{Analysis of Round Trips}
\label{analysis_of_round_trip}
To facilitate the theoretical analysis, we follow \cite{Syed_jrssb} and make the following assumptions 

\begin{assumption}[Stationarity]\label{Stationarity_process}
Each sub-process has achieved the stationary distribution $\bbeta^{(p)}\sim \pi^{(p)}$ for any $p\in [P]$.
\end{assumption}

\begin{assumption}[Weak independence]\label{weak_independence}
For any $\bar\bbeta^{(p)}$ simulated from the $p$-th chain conditional on $\bbeta^{(j)}$, $L(\bar\bbeta^{(j)})$ and $L(\bbeta^{(j)})$ are independent.
\end{assumption}

\subsection{Analysis of Round Trip Time}
\label{analyze_round_trip}
\begin{proof}[Proof of Lemma \ref{thm:round_trip_time}] For $t\in\mathbb{N}$, define $Z_t\in[P]=\{1,2,\cdots, P\}$ as the index of the chain a particle arrives after $t$ windows. 
Define $\delta_t\in\{1,-1\}$ to indicate the direction of the swap a particle intends to make during the $t$-th window; i.e., the swap is between $Z_t$ and $Z_{t}+1$ if $\delta_t=1$ and is between $Z_t$ and $Z_{t}-1$ if $\delta_t=-1$.

Define $U:=\min\{t\ge 0: Z_t=P,\delta_t=-1\}$ and $V:=\min\{t\ge 0: Z_t=1,\delta_t=1\}$. Define $r_p:=\hP[\text{ reject the swap between Chain } p \text{ and Chain } p+1\text{ for one time}]$ . 
Define $u_{p,\delta}:=\hE[U\mid Z_0=p,\delta_0=\delta]$ for $\delta\in\{1,-1\}$ and $v_{p,\delta}:=\hE[V\mid Z_0=p,\delta_0=\delta]$. Then the expectation of round trip time $T$ is
\begin{equation} \label{eq:ERTT}
\hE[T]=W(u_{1,1}+u_{P,-1}).
\end{equation}
By the Markov property, for $u_{p,\delta}$, we have
\begin{align}
&u_{p,1}=r_p^W (u_{p,-1}+1)+(1-r_p^W) (u_{p+1,1}+1) \label{eq:u1}\\
&u_{p,-1}=r_{p-1}^W  (u_{p,1}+1)+(1-r_{p-1}^W) (u_{p-1,-1}+1), \label{eq:u2}
\end{align}
where $r_{p}^W$ denotes the rejection probability of a particle in a window of size $W$ at the $p$-th chain.

According to Eq.(\ref{eq:u1}) and Eq.(\ref{eq:u2}), we have
\begin{align}
u_{p+1,1}-u_{p,1}=r_p^W(u_{p+1,1}-u_{p,-1})-1 \label{eq:u-1}\\
u_{p,-1}-u_{p-1,-1}=r_{p-1}^W(u_{p,1}-u_{p-1,-1})+1. \label{eq:u-2}
\end{align}
Define $\alpha_{p}=u_{p,1}-u_{p-1,-1}$. Then by definition, Eq.(\ref{eq:u-1}), and Eq.(\ref{eq:u-2}), we have
\begin{align*}
\alpha_{p+1}-\alpha_{p}&=(u_{p+1,1}-u_{p,-1})-(u_{p,1}-u_{p-1,-1})\\&=
(u_{p+1,1}-u_{p,1})-(u_{p,-1}-u_{p-1,-1})\\&=
r_{p}^W\alpha_{p+1}-r_{p-1}^W\alpha_{p}-2,
\end{align*}
which implies that
\begin{equation} \label{eq:a1}
a_{p+1}-a_p=-2,
\end{equation}
for $a_p:=(1-r_{p-1}^W)\alpha_{p}$. Thus, by Eq.(\ref{eq:a1}), we have
\begin{align} \label{eq:a2}
a_{p}=a_2-2(p-2).
\end{align}
By definition, $u_{1,-1}=u_{1,1}+1$.
According to Eq.(\ref{eq:u-1}), for $p=1$, we have
\begin{equation} \label{eq:a4}
\begin{aligned}
a_{2}&=(1-r_1^W)(u_{2,1}-u_{1,-1})\\&=
(1-r_1^W)\left[u_{1,1}-u_{1,-1}+r_1^W(u_{2,1}-u_{1,-1})-1\right]\\&=
-2(1-r_1^W)+r_1^Wa_{2}.
\end{aligned}
\end{equation}
Since $r_{p}\in(0,1)$ for $1\le p\le P$, Eq.(\ref{eq:a4}) implies $a_2=-2$ which together with Eq.(\ref{eq:a2}) implies 
\begin{equation} \label{eq:a5}
    (1-r_{p-1}^W)\alpha_p=a_p=-2(p-1),
\end{equation}
and therefore
\begin{equation} \label{eq:alpha1}
    r_{p-1}^W\alpha_p=-2(p-1)\frac{r_{p-1}^W}{1-r_{p-1}^W}.
\end{equation}
According to Eq.(\ref{eq:u-2}) and Eq.(\ref{eq:alpha1}), we have
\begin{align}
u_{P,1}-u_{1,1}=&\sum_{p=1}^{P-1}r_{p}^W(u_{p+1,1}-u_{p,-1})-(P-1)\\=&
\sum_{p=1}^{P-1}r_p^W\alpha_{p+1}-(P-1)\\=&
-2\sum_{p=1}^{P-1}\frac{r_{p}^W}{1-r_{p}^W}p-(P-1).
\end{align}
Since $u_{P,1}=1$, we have
\begin{equation} \label{eq:u1-1}
\begin{split}
u_{1,1}&=u_{P,1}+2\sum_{p=1}^{P-1}\frac{r_{p}^W}{1-r_{p}^W}p+(P-1)\\&=
P+2\sum_{p=1}^{P-1}\frac{r_{p}^W}{1-r_{p}^W}p.
\end{split}
\end{equation}

Similarly, for $v_{p,\delta}$, we also have
\begin{align*}
&v_{p,1}=r_p^W (v_{p,-1}+1)+(1-r_p^W) (v_{p+1,1}+1)\\
&v_{p,-1}=r_{p-1}^W  (v_{p,1}+1)+(1-r_{p-1}^W) (v_{p-1,-1}+1).
\end{align*}
With the same analysis, we have
\begin{align}
&b_{p+1}-b_p=-2 \label{eq:bp}\\
&v_{p,-1}-v_{p-1,-1}=r_{p-1}^W(v_{p,1}-v_{p-1,-1})+1, \label{eq:v3}
\end{align}
where $b_p:=(1-r_{p-1}^W)\beta_p$ and $\beta_p:=v_{p,1}-v_{p-1,-1}$.
According to Eq.(\ref{eq:v3}), we have
\begin{align} \label{eq:v4}
v_{P,-1}-v_{1,-1}=\sum_{p=1}^{P-1}r_p^W\beta_{p+1}+(P-1).
\end{align}
By definition, $v_{P,1}=v_{P,-1}+1$. 
According to Eq.(\ref{eq:v3}), for $p=P$, we have
\begin{equation} \label{eq:b3}
\begin{aligned}
b_{P}&=(1-r_{P-1}^W)(v_{P,1}-v_{P-1,-1})\\&=
(1-r_{P-1}^W)\left[v_{P,1}-v_{P,-1}+r_{P-1}^W(v_{P,1}-v_{P-1,-1})+1\right]\\&=
r_{P-1}^Wb_{P} + 2(1-r_{P-1}^W).
\end{aligned}
\end{equation}
Since $r_{p}\in(0,1)$ for $1\le p\le P$, Eq.(\ref{eq:b3}) implies $b_{P}=2$. 
Then according to Eq.(\ref{eq:bp}), we have
\begin{equation*}
b_{p}=2(P-p+1),
\end{equation*}
and therefore
\begin{equation} \label{eq:v5}
r_{p-1}^W\beta_{p}=2\frac{r_{p-1}^W}{1-r_{p-1}^W}(P-p+1).
\end{equation}
By Eq.(\ref{eq:v4}) and Eq.(\ref{eq:v5}), we have
\begin{equation*}
\begin{aligned}
v_{P,-1}-v_{1,-1}=2\sum_{p=1}^{P-1}\frac{r_{p}^W}{1-r_{p}^W}(P-p)+(P-1).
\end{aligned}
\end{equation*}
Since $v_{1,-1}=1$, we have
\begin{align} \label{eq:vP-1}
v_{P,-1} = 2\sum_{p=1}^{P-1}\frac{r_{p}^W}{1-r_{p}^W}(P-p)+P.
\end{align}

According to Eq.(\ref{eq:ERTT}), Eq.(\ref{eq:u1-1}) and Eq.(\ref{eq:vP-1}), we have
\begin{align*}
\hE[T]=W(u_{1,1}+u_{P,-1})=2WP+2WP\sum_{p=1}^{P-1}\frac{r_p^W}{1-r_p^W}.
\end{align*}
\end{proof}

The proof is a generalization of Theorem 1 \cite{Syed_jrssb}; when $W=1$, the generalized DEO scheme recovers the DEO scheme. For the self-consistency of our analysis, we present it here anyway.

\subsection{Analysis of Optimal Window Size}
\label{optimal_window_size}
\begin{proof}[Proof of Theorem \ref{col:W_approx}]
By treating $W\geq 1$ as a continuous variable and taking the derivative of $\hE[T]$ with respect to $W$ and we can get
\begin{equation} \label{eq:dWRTT1}
\begin{aligned}
\frac{\partial}{\partial W}\hE[T]&=
2P\left[1+\sum_{p=1}^{P-1}\frac{r_p^W}{(1-r_p^W)}+W\sum_{p=1}^{P-1}\frac{r_p^{W}\log r_p}{(1-r_p^W)^2}\right]\\&=
2P\left\{1+\sum_{p=1}^{P-1}\frac{r_p^W-r_p^{2W}+Wr_p^{W}\log r_p}{(1-r_p^W)^2}\right\}\\&=
2P\left\{1+\sum_{p=1}^{P-1}r_p^W\frac{1-r_p^{W}+W\log r_p}{(1-r_p^W)^2}\right\}.
\end{aligned}
\end{equation}

Assume that $r_p=r\in(0,1)$ for $1\le p\le P$. Then we have
\begin{equation} \label{eq:RTT_eqr}
\hE[T]=2WP+2WP(P-1)\frac{r^W}{1-r^W}.
\end{equation}
\begin{equation} \label{eq:dWRTT_v2}
\frac{\partial}{\partial W}\hE[T]=\frac{2P}{(1-r^W)^2}\left\{(1-r^W)^2+(P-1)r^W(1-r^W+W\log r)\right\}.
\end{equation}

Define $x:=r^W\in(0,1)$. Hence $W=\log_r(x)=\frac{\log x}{\log r}$ and
\begin{align} \label{eq:dWRTT_v3}
\frac{\partial}{\partial W}\hE[T]=f(x):=\frac{2P}{(1-x)^2}\left\{(1-x)^2+(P-1)x(1-x+\log x)\right\}
\end{align}
Thus it suffices to analyze the sign of the function $g(x):=(1-x)^2+(P-1)x(1-x+\log(x))$ for $x\in(0,1)$. For $g(x)$, we have
\begin{equation*}
g'(x)=(4-2P)(x-1) + (P-1) \log(x)
\end{equation*}
and
\begin{equation*}
g''(x)=4-2P + \frac{P-1}{x}.
\end{equation*}
Thus, $\lim_{x\rightarrow0^+}g'(x)=-\infty$, $\lim_{x\rightarrow1}g'(x)=g'(1)=0$ and $g''(x)$ is monotonically decreasing for $x>0$.
\begin{enumerate}
\item For $P>2$, we know that $g''(x)>0$ when $0<x<\frac{P-1}{2(P-2)}$ and $g''(x)<0$ when $x<\frac{P-1}{2(P-2)}$. Therefore, $g'(x)$ is maximized at $\frac{P-1}{2(P-2)}$.
\begin{enumerate}
    \item If $P=3$, by $\log(1+y)<y$ for $y>0$, we have $g'(x)=2\big(\log(1+(x-1))-(x-1)\big)<0$ for any $x\in(0,1)$. Thus $g(x)>g(1)=0$ for $x\in(0,1)$ and therefore $\frac{\partial}{\partial W}\hE[T]>0$ for $W\in\hN^{+}$. $\hE[T]$ is globally minimized at $W=1$.
    \item If $P>3$, 
    we have the following lemma and the proof is postponed in section \ref{proof_tech_lemma}.
    \begin{lemma}[Uniqueness of the solution] \label{lem:g_property}
    For $P>3$, there exists a unique solution $x^*\in(0,1)$ such that $g(x)>0$ for $\forall x\in(0,x^*)$ and $g(x)<0$ for $\forall x \in(x^*,1)$. Moreover, $W^*=\log_{r}(x^*)$ is the globally minimizer for the round trip time.
    \end{lemma}
\end{enumerate}

\item For $P=2$, we know that $g''(x)>0$ for $x\in(0,1)$. Thus $g'(x)<g(1)=0$ for $x\in(0,1)$ and therefore $g(x)>g'(1)=0$ for $x\in(0,1)$. Then according to Eq.(\ref{eq:dWRTT1}), $\frac{\partial}{\partial W}\hE[T]>0$ for $W\in\hN^{+}$. $\hE[T]$ is globally minimized at $W=1$.
\end{enumerate}

In what follows, we proceed to prove that $\frac{1}{P\log P}$ is a good approximation to $x^*$. In fact, we have
\begin{equation*}
\begin{aligned}
g\bigg(\frac{1}{P\log P}\bigg)&=\left(1-\frac{1}{P\log P}\right)^2+\frac{P-1}{P\log P}\left(1-\frac{1}{P\log P}-\log P-\log(\log P)\right)\\&=
1+o\left(\frac{1}{P}\right)-1+\frac{1}{\log P}-\frac{\log (\log P)}{\log P}+O\left(\frac{1}{P}\right)\\&=
-\frac{\log(\log P)}{\log P}+O\left(\frac{1}{\log P}\right)
\end{aligned}    
\end{equation*}
Thus, $\lim_{P\rightarrow\infty} g(\frac{1}{P\log P})=0$.

For $x=\frac{1}{P\log P}$, we have $W=\log_r\left(\frac{1}{P\log P}\right)=\frac{\log P+\log \log P}{-\log r}$.
Then according to Eq.(\ref{eq:RTT_eqr}), for $P\ge 4$,
\begin{equation} \label{eq:optimal_ET}
\begin{aligned}
\hE[T]&=2P\frac{\log P+\log \log P}{-\log r}\left[1+(P-1)\frac{\frac{1}{P\log P}}{1-\frac{1}{P\log P}}\right]\\&=
\left(1+\frac{1}{\log P}\frac{1}{1-\frac{1}{P\log P}}\right)\frac{2}{-\log r}\left(P\log P+P\log \log P\right)\\&\le
\left(1+\frac{1}{\log 4}\frac{1}{1-\frac{1}{4\log 4}}\right)\frac{2}{-\log r}\left(P\log P+P\log \log P\right)\\&<
\frac{4}{-\log r}\left(P\log P+P\log \log P\right)
\end{aligned}
\end{equation}

In conclusion, for $P=2,3$, the maximum round trip rate is achieved when the window size $W=1$. For $P\ge 4$, with the window size $W=\frac{\log P+\log \log P}{-\log r}$, the round trip rate is at least $\Omega\left(\frac{-\log r}{\log P}\right)$. 
\end{proof}

\textbf{Remark: } Given finite chains with a large rejection rate, the round trip time is only of order $O(P\log P)$ by setting the optimal window size $W\approx\left\lceil\frac{\log P+\log\log P}{-\log r}\right\rceil$. By contrast, the vanilla DEO scheme with a window of size 1 yields a much longer time of $O(P^2)$, where $\frac{1}{-\log(r)}=O(\frac{r}{1-r})$ based on Taylor expansion and a large $r\gg 0$.

\subsubsection{Technical Lemma}
\label{proof_tech_lemma}
\begin{proof}[Proof of Lemma \ref{lem:g_property}]
To help illustrate the analysis below, we plot the graphs of $g'(x)$ and $g(x)$ for $x\in(0,1)$ and $P=5$ in Figure \ref{fig:g}. For $P>3$, since $g'(x)$ is maximized at $x=\frac{P-1}{2(P-2)}\in(0,1)$ with $g''(x)>0$ when $0<x<\frac{P-1}{2(P-2)}$ and $g''(x)<0$ when $\frac{P-1}{2(P-2)}<x<1$,  $\lim_{x\rightarrow0^+}g'(x)=-\infty$, and $g'(1)=0$, we know that $g'(\frac{P-1}{2(P-2)})> 0$ and there exists $x_0\in(0,\frac{P-1}{2(P-2)})$ such that $g'(x)<0$ (i.e., $g'(x)$ is monotonically decreasing) for any $x\in(0,x_0)$ and $g'(x)>0$ (i.e., $g'(x)$ is monotonically increasing) for any $x\in(x_0,1)$. Then $g(x)$ on $(0,1)$ is globally minimized at $x=x_0$. Moreover, $\lim_{x\rightarrow0^+}g(x)=1$ and $g(1)=0$. Thus, $g(x_0)<g(1)=0$ and there exists $x^*\in(0,x_0)\subsetneq(0,1)$ such that $g(x)>0$ if $x\in(0,x^*)$ and $g(x)<0$ if $x\in(x^*,1)$. 

Meanwhile, by Eq.(\ref{eq:dWRTT1}) and the definition of $x$, we know that the sign of $\frac{\partial}{\partial W}\hE[T]$ is the same with that of $g(x)$ for $W=\log_r(x)$. Thus, $\frac{\partial}{\partial W}\hE[T]<0$ when $W< W^*:=\log_{r}(x^*)$ and $\frac{\partial}{\partial W}\hE[T]>0$ when $W>W^*$, which implies that $\hE[T]$ is globally minimized at $W^*=\log_{r}(x^*)$ with some $x^*\in(0,1)$.

\begin{figure*}[!ht]
  \centering
  \vskip -0.1in
  \subfloat[$g'(x)$]{\includegraphics[width=5.5cm, height=5cm]{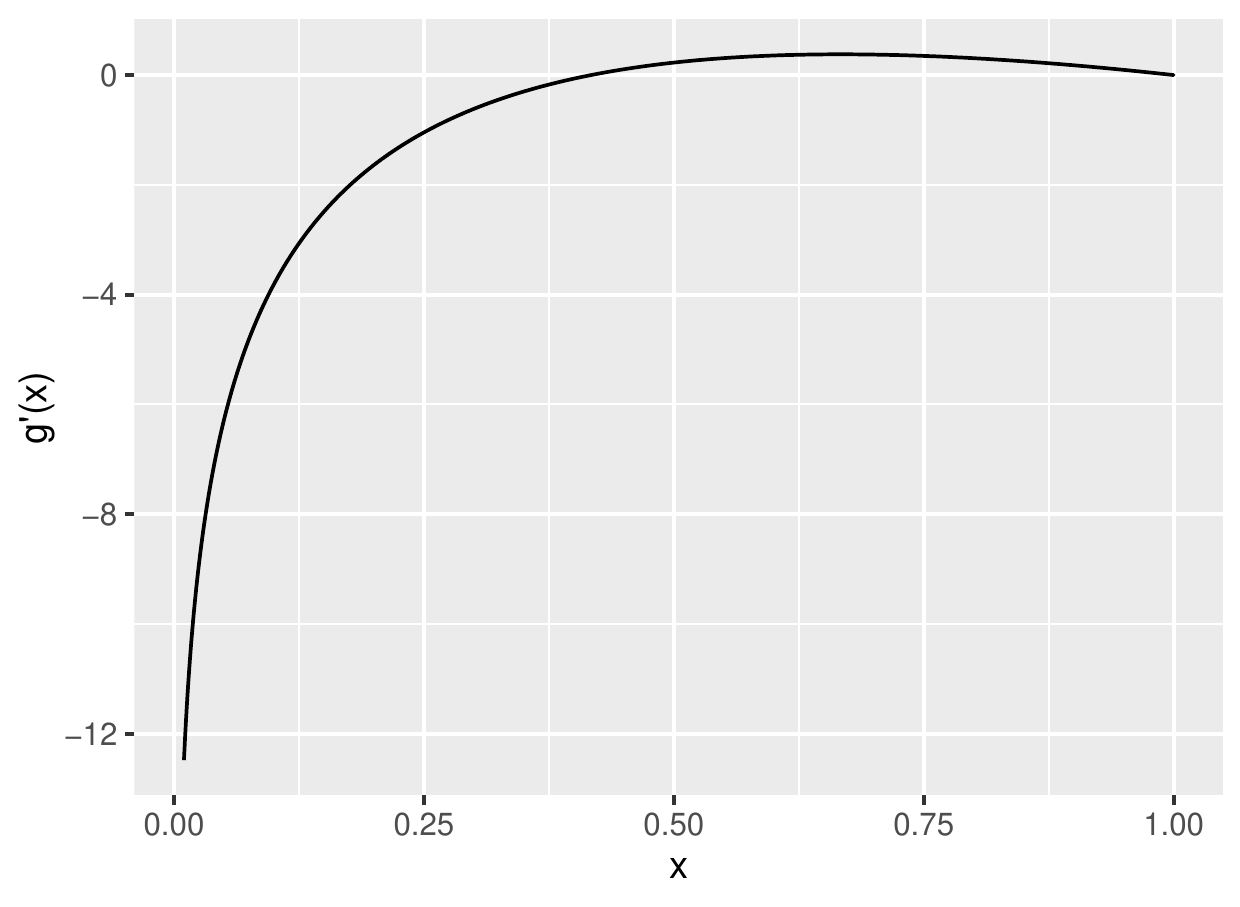}}\hspace*{0.5em} \subfloat[$g(x)$]{\includegraphics[width=5.5cm, height=5cm]{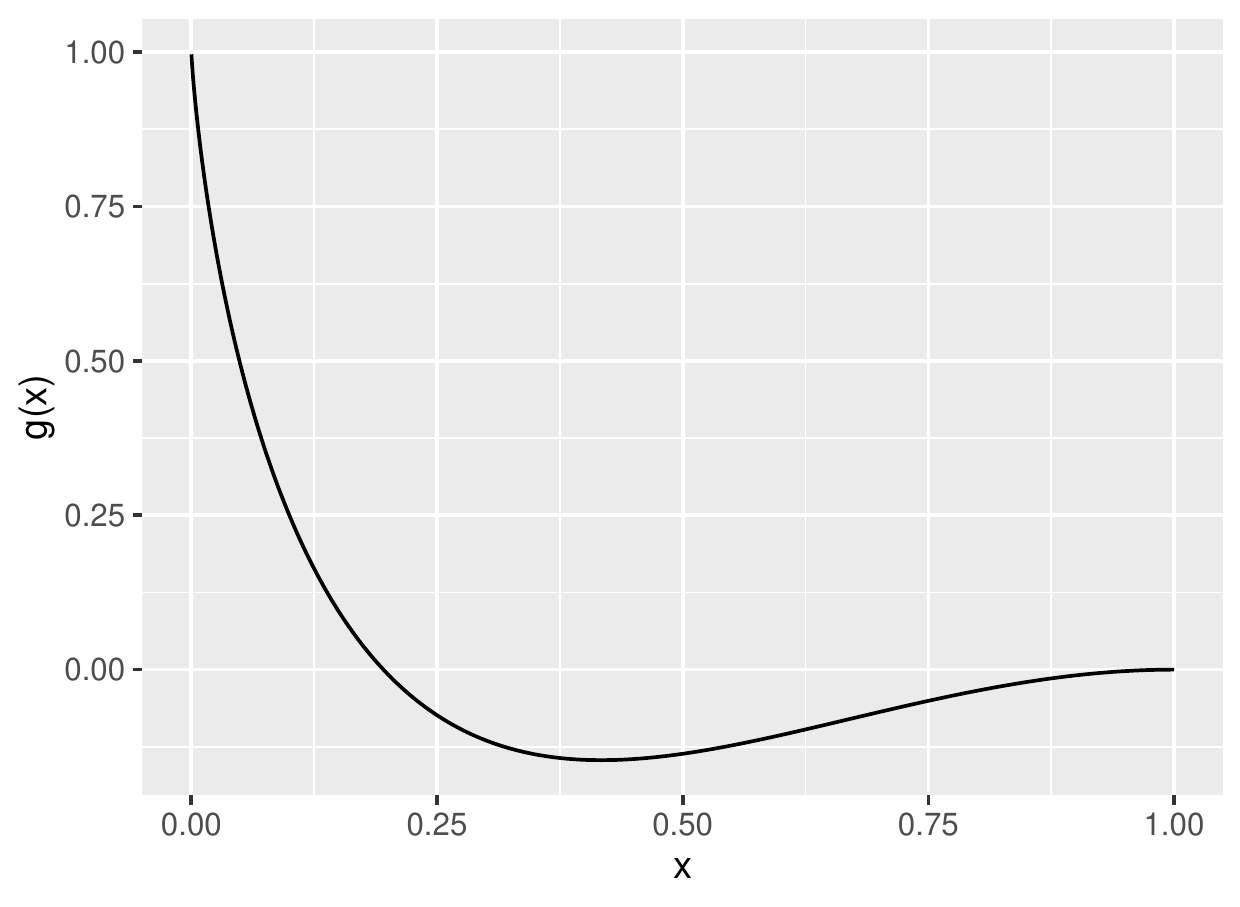}}
    \vskip -0.1in
  \caption{Illustration of functions $g'(x)$ and $g(x)$ for $x\in(0,1)$ and $P=5$}
  \label{fig:g}
\end{figure*}

\end{proof}
\chapter{TECHNICAL PROOFS FOR CHAPTER \ref{CSGLD_multi_distribution}}

\setcounter{lemma}{0}
\renewcommand\thesection{\Alph{section}}
\renewcommand{\thelemma}{\Alph{section}\arabic{lemma}}
\def\qed{ \ \vrule width.2cm height.2cm depth0cm\smallskip}
\newcommand\myeq{\stackrel{\mathclap{\normalfont\mbox{A}}}{=}}

\def \proof{{\noindent \bf Proof\quad}}

The supplementary material is organized as follows: Section \ref{review} provides a  review 
for the related methodologies, Section \ref{convergence_csgld} proves the stability condition and convergence of the self-adapting parameter, Section \ref{ergodicity} establishes the ergodicity of the contour stochastic gradient Langevin dynamics (CSGLD) algorithm,  and Section \ref{ext} provides more discussions for the algorithm.




\section{Background on Stochastic Approximation and Poisson Equation}
\label{review}

\subsection{Stochastic Approximation}
Stochastic approximation \cite{Albert90} provides a standard framework for the development of adaptive algorithms. Given a random field function $\widetilde H(\bm{\btheta}, \bm{\bx})$, the goal of the stochastic approximation algorithm is to find the solution to the  mean-field equation $h(\btheta)=0$, i.e., solving
\begin{equation*}
\begin{split}
\label{sa00}
h(\btheta)&=\int_{\MX} \widetilde H(\bm{\theta}, \bm{\bx}) \varpi_{\bm{\theta}}(d\bm{\bx})=0,
\end{split}
\end{equation*}
where $\bx\in \MX \subset \mathbb{R}^d$, $\btheta\in\bTheta \subset \mathbb{R}^{m}$, $\widetilde H(\btheta,\bx)$ is a random field 
function and $\varpi_{\btheta}(\bx)$ is a distribution function of $\bx$ depending on the parameter $\btheta$. The stochastic approximation  algorithm works by repeating the following iterations
\begin{itemize}
\item[(1)] Draw $\bm{x}_{k+1}\sim\Pi_{\bm{\theta_{k}}}(\bm{x}_{k}, \cdot)$, where $\Pi_{\bm{\theta_{k}}}(\bm{x}_{k}, \cdot)$ is a transition kernel that admits $ \varpi_{\bm{\theta}_{k}}(\bm{x})$ as
the invariant distribution,

\item[(2)] Update $\bm{\theta}_{k+1}=\bm{\theta}_{k}+\omega_{k+1} \widetilde H(\bm{\theta}_{k}, \bm{x}_{k+1})+\omega_{k+1}^2 \rho(\bm{\theta}_{k}, \bm{x}_{k+1}),$
where $\rho(\cdot,\cdot)$ denotes a bias term. 
\end{itemize}

The algorithm differs from the Robbins–Monro algorithm \cite{RobbinsM1951} in that $\bx$ is simulated from a transition kernel $\Pi_{\bm{\theta_{k}}}(\cdot, \cdot)$ instead of the exact distribution $\varpi_{\bm{\theta}_{k}}(\cdot)$. As a result, a Markov state-dependent noise $\widetilde H(\btheta_k, \bx_{k+1})-h(\btheta_k)$ is generated, which requires some regularity conditions to control the fluctuation $\sum_k \Pi_{\btheta}^k (\widetilde H(\btheta, \bx)-h(\btheta))$. Moreover, it supports a more general form where a bounded bias term $\rho(\cdot,\cdot)$ is allowed without affecting the theoretical properties of the algorithm.

\subsection{Poisson Equation}

Stochastic approximation generates a nonhomogeneous Markov chain $\{(\bx_k, \btheta_k)\}_{k=1}^{\infty}$, for which the convergence theory can be studied based on the Poisson equation 
\begin{equation*}
    \mu_{\btheta}(\bm{x})-\mathrm{\Pi}_{\bm{\theta}}\mu_{\bm{\theta}}(\bm{x})=\widetilde H(\bm{\theta}, \bm{x})-h(\bm{\theta}),
\end{equation*}
where $\Pi_{\bm{\theta}}(\bm{x}, A)$ is the transition kernel for any Borel subset $A\subset \MX$ and $\mu_{\btheta}(\cdot)$ is a function on $\MX$.
The solution to the Poisson equation exists when 
the following series converges:
\begin{equation*}
    \mu_{\btheta}(\bx):=\sum_{k\geq 0} \Pi_{\btheta}^k (\widetilde H(\btheta, \bx)-h(\btheta)).
\end{equation*}
That is, the consistency of the estimator $\btheta$ can be established by controlling the perturbations of $\sum_{k \geq 0} \Pi_{\btheta}^k (\widetilde H(\btheta, \bx)-h(\btheta))$ via imposing some regularity conditions on $\mu_{\btheta}(\cdot)$. Towards this goal, \cite{Albert90} gave 
the following regularity conditions on $\mu_{\btheta}(\cdot)$ to ensure the convergence of the adaptive algorithm:


\section{Stability and Convergence Analysis for CSGLD} \label{convergence_csgld}

\subsection{CSGLD Algorithm} \label{Alg:app}

To make the theory more general, we slightly extend CSGLD by allowing a higher order bias term. The resulting algorithm works by iterating between the following two steps:
\begin{itemize}
\item[(1)] Sample $\bm{x}_{k+1}=\bx_k- \epsilon_k\nabla_{\bx} \widetilde L(\bx_k, \btheta_k)+\mathcal{N}({0, 2\epsilon_k \tau\bm{I}}), \ \ \ \ \ \ \ \ \ \ \ \ \ \ \ \ \ \ \ \ \ \ \ \ \ \ \ \ \ \ \ \ \ \ \ \ \ \ \ \ \ \ \ \ \ \ \ \ \ \ \ \ \ \ \ \ \ \ \ \ (\text{S}_1)$

\item[(2)] Update $\bm{\theta}_{k+1}=\bm{\theta}_{k}+\omega_{k+1} \widetilde H(\bm{\theta}_{k}, \bm{x}_{k+1})
+\omega_{k+1}^2 \rho(\bm{\theta}_{k}, \bm{x}_{k+1}),
\ \ \ \ \ \ \ \ \ \ \ \ \ \ \ \ \ \ \ \ \ \ \ \ \ \ \ \ \ \ \ \ \ \ \ \ \ \  (\text{S}_2)$
\end{itemize}
where $\epsilon_k$ is the learning rate, $\omega_{k+1}$ is the step size, $\nabla_{\bx} \widetilde L(\bx, \btheta)$ is the 
stochastic gradient given by 
\begin{equation}
    \nabla_{\bx} \widetilde{L}(\bx,\btheta)= \frac{N}{n} \left[1+ 
   \frac{\zeta\tau}{\Delta u}  \left(\textcolor{black}{\log \theta(\tilde{J}(\bx))-\log\theta((\tilde{J}(\bx)-1)\vee 1)} \right) \right]  
    \nabla_{\bx} \widetilde U(\bx),
\end{equation}
$\widetilde H(\btheta,\bx)=(\widetilde H_1(\btheta,\bx), \ldots, 
 \widetilde H_m(\btheta,\bx))$ is a random field function with
\begin{equation}
     \widetilde H_i(\btheta,\bx)={\theta}^{\zeta}(\tilde J(\bx))\left(1_{i= \tilde J(\bx)}-{\theta}(i)\right), \quad i=1,2,\ldots,m,
\end{equation}
for some constant $\zeta>0$, and $\rho(\btheta_k,\bx_{k+1})$ is a bias term.

\subsection{Convergence of Parameter Estimation} 
\label{App:convergence}

To establish the convergence of $\btheta_k$, we make the following assumptions:

\begin{assump}[Compactness] \label{ass2a} 
The space $\Theta$ is compact such that $\inf_{\Theta} \theta(i) >0$ for any  $i\in \{1,2,\ldots,m\}$. There exists a large  constant $Q>0$ such that for any $\btheta\in \bTheta$ and $\bx \in \MX$, 
\begin{equation}
\label{compactness}
     \| \btheta\| \leq Q, \quad 
     \| \widetilde H(\btheta, \bx)\| \leq Q, \quad 
     \| \rho(\btheta, \bx)\| \leq Q.
\end{equation}
\end{assump}

To simplify the proof, we consider a slightly stronger assumption such that $\inf_{\Theta} \theta(i)>0$ holds for any $i \in \{1,2,\ldots,m\}$. To relax this assumption, we refer interested readers to \cite{Fort15} where the recurrence property was proved for the sequence $\{\btheta_k\}_{k\geq 1}$ of a similar algorithm. Such a property guarantees $\btheta_k$ to visit often 
enough to a desired compact space, rendering the convergence of the sequence.

\textcolor{black}{\begin{assump}[Smoothness]
\label{ass2}
$U(\bm{\xeta})$ is $M$-smooth; that is, there exists a constant $M>0$ such that for any $\bx, \bx'\in \MX$,
\begin{equation}
\label{ass_2_1_eq}
\begin{split}
\| \nabla_{\bx} U(\bx)-\nabla_{\bx} U(\bm{\bx}')\|  & \leq M\| \bx-\bx'\| . \\
\end{split}
\end{equation}
\end{assump}}

Smoothness is a standard assumption in the study of convergence of SGLD, see e.g. \cite{Maxim17,Xu18}.

\begin{assump}[Dissipativity]
\label{ass3}
 There exist constants $\tilde{m}>0$ and $\tilde{b}\geq 0$ such that for any $\bx \in \MX$ and $\btheta \in \bTheta$, 
\label{ass_dissipative}
\begin{equation}
\label{eq:01}
\langle \nabla_{\bx} L(\bx, \btheta), \bx\rangle\leq \tilde{b}-\tilde{m}\| \bx\| ^2.
\end{equation}
\end{assump}
This assumption ensures samples to move towards the origin regardless the initial point, 
which is standard in proving the geometric ergodicity of dynamical systems, see e.g. \cite{mattingly02, Maxim17, Xu18}.

\begin{assump}[Gradient noise] 
\label{ass4}
The stochastic gradient is unbiased, that is, 
\begin{equation*}
\E[\nabla_{\bx}\widetilde U(\bx_{k})-\nabla_{\bx} U(\bx_{k})]=0;
\end{equation*}
in addition, there exist some constants $M>0$ and 
$B>0$ such that
\begin{equation*} 
\E [ \| \nabla_{\bx}\widetilde U(\bx_{k})-\nabla_{\bx} U(\bx_{k})\| ^2 ] \leq M^2 \| \bx\| ^2+B^2,
\end{equation*}
where the expectation $\E[\cdot]$ is taken with respect to the distribution of the noise component included in $\nabla_{\bx} \widetilde{U}(\bx)$.
\end{assump}

Lemma \ref{convex_appendix} establishes a stability condition for CSGLD, which implies potential 
convergence of $\btheta_k$.

\begin{lemma}[Stability. Formal statement of Lemma \ref{convex_main_csgld}] \label{convex_appendix} 
Suppose that Assumptions  \ref{ass2a}-\ref{ass4}  hold. 
For any $\btheta \in \bTheta$, $\langle h(\btheta), \btheta - \btheta_{\star}\rangle \leq  -\phi\| \btheta - \btheta_{\star}\| ^2+\mathcal{O}\left(\Var(\xi_n)+\epsilon+\frac{1}{m}\right)$, where $\phi=\inf_{\btheta} Z_{\btheta}^{-1}>0$, and $\theta_{\star}$ follows that $(\int_{\MX_1}\pi(\bx)d\bx,\int_{\MX_2}\pi(\bx)d\bx,\ldots,\int_{\MX_m}\pi(\bx)d\bx)$. $\Var(\xi_n)$ denotes the largest variance of the noise in the stochastic energy estimator of batch size $n$.
\end{lemma}

\begin{proof}
 Let $\varpi_{\Psi_{\btheta}}(\bx)\propto\frac{\pi(\bx)}{\Psi^{\zeta}_{\btheta}(U(\bx))}$ denote a theoretical invariant measure of SGLD, where 
  $\Psi_{\btheta}(u)$ is a fixed
  piecewise continuous function given by 
\begin{equation}\label{new_design_appendix}
\Psi_{\btheta}(u)= \sum_{i=1}^m \left(\theta(i-1)e^{(\log\theta(i)-\log\theta(i-1)) \frac{u-u_{i-1}}{\Delta u}}\right) 1_{u_{i-1} < u \leq u_i},
\end{equation}
 the full data is used 
 in determining the indexes of subregions, and the learning rate converges to zero. 
 In addition, we 
 define a piece-wise constant function 
 \[
 \widetilde{\Psi}_{\btheta}=\sum_{i=1}^m \theta(i) 1_{u_{i-1} < u \leq u_{i}},
 \]
 and a theoretical measure 
 $\varpi_{\widetilde{\Psi}_{\btheta}}(\bx) \propto \frac{\pi(\bx)}{\theta^{\zeta}(J(\bx))}$. 
 Obviously, as the sample space partition becomes 
 fine and fine, i.e., $u_1 \to u_{\min}$, $u_{m-1}\to u_{\max}$ and $m \to \infty$, we have  
 $\| \widetilde{\Psi}_{\btheta}-\Psi_{\btheta}\| \to 0$ and $\|  
 \varpi_{\widetilde{\Psi}_{\btheta}}(\bx)- 
 \varpi_{\Psi_{\btheta}}(\bx) \| \to 0$, where 
 $u_{\min}$ and $u_{\max}$ denote the minimum and maximum of $U(\bx)$, respectively. 
 Without loss of generality, we assume $u_{\max}<\infty$. Otherwise, $u_{\max}$ can be set to a value such that $\pi(\{\bx: U(\bx)>u_{\max}\})$ is sufficiently small.

For each $i \in \{1,2,\ldots,m\}$, the random field $\widetilde H_i(\btheta,\bx)={\theta}^{\zeta}(\tilde J(\bx))\left(1_{i\geq \tilde J(\bx)}-{\theta}(i)\right)$ is a biased estimator of $ H_i(\btheta,\bx)={\theta}^{\zeta}( J(\bx))\left(1_{i\geq J(\bx)}-{\theta}(i)\right)$. Let $\delta_n(\btheta)=\E[\widetilde{H}(\btheta,\bx)-H(\btheta,\bx)]$ denote the bias, which is caused by the mini-batch evaluation of the energy and 
decays to $0$ as $n\rightarrow N$. By Lemma.\ref{bias_in_SA}, we know that the bias caused by the stochastic energy is of order $\mathcal{O}(\Var(\xi_n))$.

First, let's compute the mean-field $h(\btheta)$ with respect to the empirical measure $\varpi_{\btheta}(\bx)$:
\begin{equation}
\small
\label{iiii_csgld}
\begin{split} 
        h_i(\btheta)&=\int_{\MX} \widetilde H_i(\btheta,\bx) 
         \varpi_{\btheta}(\bx) d\bx
         =\int_{\MX} H_i(\btheta,\bx) 
         \varpi_{\btheta}(\bx) d\bx+\mathcal{O}(\Var(\xi_n))\\
         &=\ \int_{\MX} H_i(\btheta,\bx) \left( \underbrace{\varpi_{\widetilde{\Psi}_\btheta}(\bx)}_{\text{I}_1} \underbrace{-\varpi_{\widetilde{\Psi}_\btheta}(\bx)+\varpi_{\Psi_{\btheta}}(\bx)}_{\text{I}_2}\underbrace{-\varpi_{\Psi_{\btheta}}(\bx)+\varpi_{\btheta}(\bx)}_{\text{I}_3}\right) d\bx+\mathcal{O}(\Var(\xi_n)).\\
\end{split}
\end{equation}

For the term $\text{I}_1$, we have
\begin{equation}
\begin{split}
\label{i_1}
    \int_{\MX} H_i(\btheta,\bx) 
     \varpi_{\widetilde{\Psi}_\btheta}(\bx) d\bx&=\frac{1}{Z_{\btheta}} \int_{\MX} {\theta}^{\zeta}(J(\bx))\left(1_{i= J(\bx)}-{\theta}(i)\right) \frac{\pi(\bx)}{\theta^{\zeta}(J(\bx))} d\bx\\
    &=Z_{\btheta}^{-1}\left[\sum_{k=1}^m \int_{\MX_k} 
     \pi(\bx) 1_{k=i} d\bx -\theta(i)\sum_{k=1}^m\int_{\MX_k} \pi(\bx)d\bx \right] \\
    &=Z_{\btheta}^{-1} \left[\theta_{\star}(i)-\theta(i)\right],
\end{split}
\end{equation}
where $Z_{\btheta}=\sum_{i=1}^m \frac{\int_{\MX_i} \pi(\bx)d\bx}{\theta(i)^{\zeta}}$ denotes the normalizing constant 
of $\varpi_{\widetilde{\Psi}_\btheta}(\bx)$.

Next, let's consider the integrals $\text{I}_2$ and  $\text{I}_3$. By Lemma \ref{partition_order} and the boundedness of $H(\btheta,\bx)$, we have 
\begin{equation} \label{biasI2}
\int_{\MX} H_i(\btheta,\bx) (-\varpi_{\widetilde{\Psi}_{\btheta}}(\bx)+\varpi_{\Psi_{\btheta}}(\bx)) d\bx= \mathcal{O}\left(\frac{1}{m}\right).
\end{equation}
For the term $I_3$, we have for any fixed $\btheta$,
\begin{equation}\label{iiii_csgld_2}
    \int_{\MX} H_i(\btheta,\bx) \left(-\varpi_{\Psi_{\btheta}}(\bx)+\varpi_{\btheta}(\bx)\right) d\bx=\mathcal{O}(\Var(\xi_n)\left(\btheta)\right)+\mathcal{O}(\epsilon),
\end{equation}
where the \textcolor{black}{order of $\mathcal{O}(\epsilon)$ follows from Theorem 6 of} \cite{Sato2014ApproximationAO}.

Plugging (\ref{i_1}), (\ref{biasI2}) and  \textcolor{black}{(\ref{iiii_csgld_2})} into (\ref{iiii_csgld}), we have
\begin{equation}\label{h_i_theta}
     h_i(\btheta)=Z_{\btheta}^{-1} \left[\varepsilon\beta_i(\btheta)+\theta_{\star}(i)-\theta(i)\right],
\end{equation}
where $\varepsilon=\mathcal{O}\left(\Var(\xi_n)+\epsilon+\frac{1}{m}\right)$ and $\beta_i(\btheta)$ is a bounded term such that $Z_{\btheta}^{-1}\varepsilon\beta_i(\btheta)=\mathcal{O}\left(\Var(\xi_n)+\epsilon+\frac{1}{m}\right)$.

To solve the ODE system with small disturbances, we consider standard techniques in perturbation theory.
According to the fundamental theorem of
perturbation theory \cite{Eric}, we can obtain the solution to 
the mean field equation $h(\btheta)=0$: 
\begin{equation}
    \theta(i)=\theta_{\star}(i)+\varepsilon\beta_i(\btheta_{\star}) +\mathcal{O}(\varepsilon^2), \quad i=1,2,\ldots,m,
\end{equation}
which is a stable point in a small neighbourhood of $\btheta_{\star}$.

Considering the positive definite function $\mathbb{V}(\btheta)=\frac{1}{2}\|  \btheta_{\star}-\btheta\| ^2$ for the mean-field system $h(\btheta)=Z_{\btheta}^{-1} (\varepsilon\beta_i(\btheta)+\btheta_{\star}-\btheta)=Z_{\btheta}^{-1} (\btheta_{\star}-\btheta)+\mathcal{O}(\varepsilon)$, we have
\begin{equation*}
    \langle h(\btheta), \nabla \mathbb{V}(\btheta)\rangle=\langle h(\btheta), \btheta - \btheta_{\star}\rangle = -Z_{\btheta}^{-1}\| \btheta - \btheta_{\star}\| ^2+\mathcal{O}(\varepsilon)\leq -\phi\| \btheta - \btheta_{\star}\| ^2+\mathcal{O}\left(\Var(\xi_n)+\epsilon+\frac{1}{m}\right),
\end{equation*}
where $\phi=\inf_{\btheta} Z_{\btheta}^{-1}>0$ by
the compactness assumption \ref{ass2a}. This concludes the proof.
\end{proof}

\textcolor{black}{The following is a restatement of  Lemma 3.2 of \cite{Maxim17}, which holds for any $\btheta$ in the compact space $\bTheta$.
\begin{lemma}[Uniform $L^2$ bounds]
\label{lemma:1}
Suppose Assumptions \ref{ass2a}, \ref{ass3} and \ref{ass4} hold.  Given a small enough learning rate, then 
$\sup_{k\geq 1} \E[\| \bm{\xeta}_{k}\| ^2] < \infty$.
\end{lemma}}

The following lemma justifies the regularity properties of Poisson's equation, which is crucial in controlling the perturbations through the stochastic approximation process. The first version was proposed in Lemma B2 of \cite{CSGLD}. Now we give a more detailed proof by utilizing a Lyapunov function $V(\bx)=1+\bx^2$ and Lemma \ref{lemma:1}.

\begin{lemma}[Solution of Poisson equation]
\label{lyapunov}
Suppose that Assumptions  \ref{ass2a}-\ref{ass4}  hold. 
There is a solution $\mu_{\btheta}(\cdot)$ on $\MX$ to the Poisson equation 
\begin{equation}
    \label{poisson_eqn}
    \mu_{\btheta}(\bm{x})-\mathrm{\Pi}_{\bm{\theta}}\mu_{\bm{\theta}}(\bm{x})=\widetilde H(\bm{\theta}, \bm{x})-h(\bm{\theta}).
\end{equation}
In addition, for all $\bm{\theta}, \bm{\theta}'\in \bm{\bTheta}$, 
there exists a constant $C$ such that
\begin{equation}
\begin{split}
\label{poisson_reg}
\E[\| \mathrm{\Pi}_{\bm{\theta}}\mu_{\btheta}(\bx)\| ]&\leq C,\\
\E[\| \mathrm{\Pi}_{\bm{\theta}}\mu_{\bm{\theta}}(\bx)-\mathrm{\Pi}_{\bm{\theta}'}\mu_{\bm{\theta'}}(\bx)\| ]&\leq C\| \bm{\theta}-\bm{\theta}'\| .\\
\end{split}
\end{equation}
\end{lemma}

\begin{proof} The existence and the regularity property of Poisson's equation can be used to control the perturbations. The key of the proof lies in verifying drift conditions proposed in Section 6 of \cite{AndrieuMP2005}.

\textbf{(DRI)} By the smoothness assumption \ref{ass2}, we have that $U(\bx)$ is continuously differentiable almost everywhere. By the dissipative assumption \ref{ass3} and Theorem 2.1 \cite{Roberts_Tweedie_Bernoulli}, we can show that the discrete dynamics system is irreducible and aperiodic. Now consider a Lyapunov function $V=1+\| \bx\| ^2$ and any compact subset $\mathcal{\bm{K}}\subset \bTheta$, the drift conditions are verified as follows:

\textbf{(DRI1)} Given small enough learning rates $\{\epsilon_k\}_{k\geq 1}$, the smoothness assumption \ref{ass2}, and the dissipative assumption \ref{ass3}, applying Corollary 7.5 \cite{mattingly02} yields the minorization condition for the CSGLD algorithm, i.e. there exists $\eta>0$, a measure $\nu$, and a set $\mathcal{C}$ such that $\nu(\mathcal{C})=1$. Moreover, we have
    $$P_{\btheta\in \mathcal{\bm{K}}}(x, A)\geq  \eta \nu(A)\ \ \ \ \ \forall A\in \MX, \bx\in \mathcal{C}. \eqno{(\text{I})}$$
    
where $P_{\btheta}(\bx, \by):=\frac{1}{2\sqrt{(4\pi\epsilon)^{d/2}}}\E\big[e^{-\frac{\| \by-\bx+\epsilon \nabla_{\bx} \widetilde{L}(\bx, \btheta)\| ^2}{4\epsilon}}\mid \bx\big]$ denotes the transition kernel based on CSGLD with the parameter $\btheta\in\mathcal{\bm{K}}$ and a learning rate $\epsilon$, in addition, the expectation is taken over the adaptive gradient $\nabla_{\bx} \widetilde{L}(\bx, \btheta)$ in Eq.(\ref{adaptive_grad}). Using Assumption \ref{ass2a}-\ref{ass4}, we can prove the uniform L2 upper bound by following Lemma 3.2 \cite{Maxim17}. Further, by Theorem 7.2 \cite{mattingly02}, there exist $\tilde\alpha\in (0,1)$ and $\tilde \beta\geq 0$ such that
    $$P_{\btheta\in\mathcal{\MK}}V(\bx)\leq \tilde\alpha V(\bx)+\tilde\beta. \eqno{(\text{II})}$$

Consider a Lyapunov function $V=1+\| \bx\| ^2$ and a constant $\kappa=\tilde\alpha+\tilde \beta$, it yields that
$$P_{\btheta\in\mathcal{\bm{K}}}V(\bx)\leq \kappa V(\bx). \eqno{(\text{III})}$$
Now we have verified the first condition \text{(DRI1)} by checking conditions (\text{I}),(\text{II}), and (\text{III}), 

\textbf{(DRI2)} In what follows, we check the boundedness and Lipshitz conditions on the random-field function $\widetilde H(\btheta,\bx)$, where each subcomponent is defiend as $\widetilde H_i(\btheta,\bx)={\theta}^{\zeta}(J_{\widetilde U}(\bx))\left(1_{i=J_{\widetilde U}(\bx)}-{\theta}(i)\right)$. Recall that $V=1+\| \bx\| ^2$, the compactness assumption \ref{ass2a} directly leads to
$$\sup_{\btheta\in\mathcal{\bm{K}}\subset [0, 1]^m}\|  H(\btheta, \bx)\| \leq m V(\bx). \eqno{(\text{IV})}$$
For any $\btheta_1, \btheta_2\in \mathcal{\bm{K}}$ and a fixed $\bx\in\MX$, it suffices for us to solely verify the $i$-th index, which is the index that maximizes $\mid \theta_1(i)-\theta_2(i)\mid $, then repeatedly applying mean-value theorem and assumption \ref{ass2a} leads to
\begin{equation*}
\begin{split}
\small
    \mid \widetilde H_i(\btheta_1,\bx)- \widetilde H_i(\btheta_2,\bx)\mid &={\theta_1^{\zeta}} (J_{\widetilde U}(\bx))\left(1_{i=J_{\widetilde U}(\bx)}-{\theta_1}(i)\right)-{\theta_2^{\zeta}} (J_{\widetilde U}(\bx))\left(1_{i=J_{\widetilde U}(\bx)}-{\theta_2}(i)\right)\\
    &\lesssim \mid {\theta_1} (J_{\widetilde U}(\bx))-{\theta_2} (J_{\widetilde U}(\bx))\mid +\mid {\theta_1} (J_{\widetilde U}(\bx)){\theta_1}(i)-{\theta_2} (J_{\widetilde U}(\bx)){\theta_2}(i)\mid \\
    &\lesssim \max_{j}\Big(\mid {\theta_1} (j)-{\theta_2} (j)\mid +{\theta_1} (j)\mid {\theta_1}(i)-{\theta_2}(i)\mid +\mid {\theta_1} (j)-{\theta_2} (j)\mid \theta_2(i)\Big)\\
    &\lesssim 3\mid \theta_1(i)-\theta_2(i)\mid ,\\
\end{split}
\end{equation*}
where the last inequality holds since $\theta(i)\in(0, 1]$ for any $i\leq m$.

\textbf{(DRI3)} We proceed to verify the smoothness of the transitional kernel $P_{\btheta}(\bx, \by)$ with respect to $\btheta$. For any $\btheta_1, \btheta_2\in \mathcal{\bm{K}}$ and fixed $\bx$ and $\by$, we have 
\begin{equation*}
\begin{split}
    &\quad\mid P_{\btheta_1}(\bx, \by)-P_{\btheta_2}(\bx, \by)\mid \\
    &=\frac{1}{2\sqrt{(4\pi\epsilon)^{d/2}}}\E\big[e^{-\frac{\| \by-\bx+\epsilon \nabla_{\bx} \widetilde{L}(\bx, \btheta_1)\| ^2}{4\epsilon}}\mid \bx\big]-\frac{1}{2\sqrt{(4\pi\epsilon)^{d/2}}}\E\big[e^{-\frac{\| \by-\bx+\epsilon \nabla_{\bx} \widetilde{L}(\bx, \btheta_2)\| ^2}{4\epsilon}}\mid \bx\big]\\
    &\lesssim \mid \| \by-\bx+\epsilon \nabla_{\bx} \widetilde{L}(\bx, \btheta_1)\| ^2-\| \by-\bx+\epsilon \nabla_{\bx} \widetilde{L}(\bx, \btheta_2)\| ^2\mid \\
    &\lesssim \| \nabla_{\bx} \widetilde{L}(\bx, \btheta_1)- \nabla_{\bx} \widetilde{L}(\bx, \btheta_2)\| \\
    &\lesssim \|  \btheta_1-\btheta_2\| ,\\
\end{split}
\end{equation*}
where the first inequality (up to a finite constant) follows by $\| e^{\bx}-e^{\by}\| \lesssim \| \bx-\by\| $ for any $\bx$, $\by$ in a compact space; the last inequality follows by the definition of the adaptive gradient in Eq.(\ref{adaptive_grad}) and $\| \log(\bx)-\log(\by)\| \lesssim \| \bx-\by\| $ by the compactness assumption \ref{ass2a}.

For $f:\MX\rightarrow\mathbb{R}^d$, define the norm $\| f\| _V=\sup_{\bx\in\MX} \frac{\mid f(\bx)\mid }{V(\bx)}$. Following the same technique proposed in \cite{Liang07} (page 319), we can verify the last drift condition
$$\| P_{\btheta_1}f-P_{\btheta_2}f\| _V\leq C\| f\| _V \| \btheta_1-\btheta_2\| , \ \ \forall f\in \mathcal{L}_V:=\{f: \MX\rightarrow \mathbb{R}^d, \| f\| _V<\infty\}. \eqno{(\text{VI})}$$

Having conditions (\text{I}), (\text{II}), $\cdots$ and (\text{VI}) verified, we are now able to prove the drift conditions proposed in Section 6 of \cite{AndrieuMP2005}.\qed
\end{proof}


Now we are ready to prove the first main result on the 
convergence of $\btheta_k$.
The technique lemmas are listed 
in Section \ref{Lemmasection}.

\begin{assump}[Learning rate and step size]
\label{ass1}
The learning rate $\{\epsilon_k\}_{k \in \mathrm{N}}$ is a positive non-increasing sequence of real numbers satisfying the conditions 
\[
\lim_k \epsilon_k=0, \quad \sum_{k=1}^{\infty} \epsilon_k=\infty.
\]
The step size $\{\omega_{k}\}_{k\in \mathrm{N}}$ is a positive non-increasing  sequence of real numbers such that
\begin{equation} \label{a1}
\lim_{k \to \infty} \omega_k=0, \quad
\sum_{k=1}^{\infty} \omega_{k}=+\infty, \quad  \sum_{k=1}^{\infty} \omega_{k}^2<+\infty.
\end{equation}
A practical strategy is to set $\omega_{k}:=\mathcal{O}(k^{-\alpha})$ to satisfy the above conditions for any $\alpha\in (0.5, 1]$. 

 \end{assump}

\begin{theorem}[$L^2$ convergence rate. Formal statement of Theorem \ref{thm:1}]
\label{latent_convergence}
Suppose Assumptions $\ref{ass2a}$-$\ref{ass1}$ hold. For a sufficiently
large value of $m$, a sufficiently small learning rate sequence  $\{\epsilon_k\}_{k=1}^{\infty}$,  and a sufficiently small
 step  size sequence $\{\omega_k\}_{k=1}^{\infty}$, 
$\{\btheta_k\}_{k=0}^{\infty}$ converges to
 $\btheta_{\star}$ in $L_2$-norm  such that
\begin{equation*}
    \E\left[\| \bm{\theta}_{k}-\btheta_{\star}\| ^2\right]=\mathcal{O}\left( \omega_{k}+
    \sup_{i\geq k_0}\epsilon_i+
\frac{1}{m} +\Var(\xi_n)\right),
\end{equation*}
where $k_0$ is a sufficiently large constant.
\end{theorem}
\begin{proof}
Consider the iterations 
\begin{equation*}
    \bm{\theta}_{k+1}=\bm{\theta}_{k}+\omega_{k+1} \left(\widetilde H(\bm{\theta}_{k}, \bm{x}_{k+1})+\omega_{k+1} \rho(\btheta_k, \bx_{k+1})\right).
\end{equation*}
Define $\bm{T}_{k}=\bm{\theta}_{k}-\btheta_{\star}$. By subtracting $\btheta_{\star}$ from both sides and taking the square and $L_2$ norm,  we have
\begin{equation*}
\small
\begin{split}
    \| \bT_{k+1}\| ^2&=\| \bT_k\| ^2 +\omega_{k+1}^2 \| \widetilde H(\btheta_k, \bx_{k+1}) + \omega_{k+1}\rho(\btheta_k, \bx_{k+1})\| ^2+2\omega_{k+1}\underbrace{\langle \bT_k,  \widetilde H(\bx_{k+1})+\omega_{k+1}\rho(\btheta_k, \bx_{k+1})\rangle}_{\text{D}}.
\end{split}
\end{equation*}

First, by Lemma \ref{convex_property}, there exists a constant $G=4Q^2(1+Q^2)$ such that
\begin{equation}
\label{first_term}
    \|  \widetilde H(\btheta_k, \bx_{k+1}) + \omega_{k+1}\rho(\btheta_k, \bx_{k+1})\| ^2 \leq G (1+\| \bT_k\| ^2).
\end{equation}

Next, by the Poisson equation (\ref{poisson_eqn}), we have
\begin{equation*}
\begin{split}
   \text{D}&=\langle \bT_k,  \widetilde H(\btheta_k, \bx_{k+1})+\omega_{k+1}\rho(\btheta_k, \bx_{k+1}) \rangle\\
   &=\langle \bT_k,  h(\btheta_k)+\mu_{\btheta_k}(\bm{x}_{k+1})-\mathrm{\Pi}_{\bm{\theta}_k}\mu_{\bm{\theta}_k}(\bm{x}_{k+1})+\omega_{k+1}\rho(\btheta_k, \bx_{k+1}) \rangle\\
   &=\underbrace{\langle \bT_k,  h(\btheta_k)\rangle}_{\text{D}_{1}} +\underbrace{\langle\bT_k, \mu_{\btheta_k}(\bm{x}_{k+1})-\mathrm{\Pi}_{\bm{\theta}_k}\mu_{\bm{\theta}_k}(\bm{x}_{k+1})\rangle}_{\text{D}_{2}}+\underbrace{\langle \bT_k, \omega_{k+1}\rho(\btheta_k, \bx_{k+1})\rangle}_{{\text{D}_{3}}}.
\end{split}
\end{equation*}

For the term $\text{D}_1$, by Lemma \ref{convex_appendix}, we have
\begin{align*}
\E\left[\langle \bm{T}_{k}, h(\bm{\theta}_{k})\rangle\right] &\leq - \phi\E[\| \bm{T}_{k}\| ^2]+\mathcal{O}(\Var(\xi_n)+\epsilon_k+\frac{1}{m}).
\end{align*}
For convenience, in the following context, we denote $\mathcal{O}(\Var(\xi_n)+\epsilon_k+\frac{1}{m})$ by $\Delta_k$.

To deal with the term $\text{D}_2$, we make the following decomposition 
\begin{equation*}
\begin{split}
\text{D}_2 &=\underbrace{\langle \bT_k, \mu_{\bm{\theta}_{k}}(\bm{\xeta}_{k+1})-\mathrm{\Pi}_{\bm{\theta}_{k}}\mu_{\bm{\theta}_{k}}(\bm{\bx}_{k})\rangle}_{\text{D}_{21}} \\
&+ \underbrace{\langle \bT_k,\mathrm{\Pi}_{\bm{\theta}_{k}}\mu_{\bm{\theta}_{k}}(\bm{x}_{k})- \mathrm{\Pi}_{\bm{\theta}_{k-1}}\mu_{\bm{\theta}_{k-1}}(\bm{x}_{k})\rangle}_{\text{D}_{22}}
+ \underbrace{\langle \bT_k,\mathrm{\Pi}_{\bm{\theta}_{k-1}}\mu_{\bm{\theta}_{k-1}}(\bm{x}_{k})- \mathrm{\Pi}_{\bm{\theta}_{k}}\mu_{\bm{\theta}_{k}}(\bm{\xeta}_{k+1})\rangle}_{\text{D}_{23}}.\\
\end{split}
\end{equation*}

(\text{i})  From the Markov property, $\mu_{\bm{\theta}_{k}}(\bm{\xeta}_{k+1})-\mathrm{\Pi}_{\bm{\theta}_{k}}\mu_{\bm{\theta}_{k}}(\bm{x}_{k})$ forms a martingale difference sequence 
$$\E\left[\langle \bT_k, \mu_{\bm{\theta}_{k}}(\bm{\xeta}_{k+1})-\mathrm{\Pi}_{\bm{\theta}_{k}}\mu_{\bm{\theta}_{k}}(\bm{x}_{k})\rangle \mid \mathcal{F}_{k}\right]=0, \eqno{(\text{D}_{21})}$$
where  $\mathcal{F}_k$ is a $\sigma$-filter formed by $\{\btheta_0, \bx_1, \btheta_1, \bx_2, \cdots, \bx_k,\btheta_k\}$.

(\text{ii})  By the regularity of the solution of Poisson equation in (\ref{poisson_reg}) and Lemma \ref{theta_lip}, we have 
\begin{equation}
\label{theta_delta}
\E[\| \mathrm{\Pi}_{\bm{\theta}_{k}}\mu_{\bm{\theta}_{k}}(\bm{x}_{k})- \mathrm{\Pi}_{\bm{\theta}_{k-1}}
 \mu_{\bm{\theta}_{k-1}}(\bm{x}_{k})\| ]\leq C \| \btheta_k-\btheta_{k-1}\| \leq 2Q C\omega_k.
\end{equation}
Using Cauchy–Schwarz inequality, (\ref{theta_delta}) and the compactness of $\Theta$ in Assumption \ref{ass2a}, we have
$$\small{\E[\langle\bm{T}_{k},\mathrm{\Pi}_{\bm{\theta}_{k}}\mu_{\bm{\theta}_{k}}(\bm{x}_{k})- \mathrm{\Pi}_{\bm{\theta}_{k-1}}\mu_{\bm{\theta}_{k-1}}(\bm{x}_{k})\rangle]\leq \E[\| \bT_k\| ]\cdot 2Q C\omega_k\leq 4Q^2 C\omega_{k}\leq 5Q^2 C\omega_{k+1}}   \eqno{(\text{D}_{22})},$$
where the last inequality follows from assumption \ref{ass1} and holds for a large enough $k$.

(\text{iii})  For the last term of $\text{D}_{2}$, 
\begin{equation*}
\begin{split}
\small
&\langle \bm{T}_{k},\mathrm{\Pi}_{\bm{\theta}_{k-1}}\mu_{\bm{\theta}_{k-1}}(\bm{x}_{k})- \mathrm{\Pi}_{\bm{\theta}_{k}}\mu_{\bm{\theta}_{k}}(\bm{\xeta}_{k+1})\rangle\\
=&\left(\langle \bm{T}_{k}, \mathrm{\Pi}_{\bm{\theta}_{k-1}}\mu_{\bm{\theta}_{k-1}}(\bm{x}_{k}) \rangle- \langle \bm{T}_{k+1}, \mathrm{\Pi}_{\bm{\theta}_{k}}\mu_{\bm{\theta}_{k}}(\bm{\xeta}_{k+1})\rangle\right)\\
&\ \ \ +\left(\langle \bm{T}_{k+1}, \mathrm{\Pi}_{\bm{\theta}_{k}}\mu_{\bm{\theta}_{k}}(\bm{\xeta}_{k+1})\rangle-\langle \bm{T}_{k}, \mathrm{\Pi}_{\bm{\theta}_{k}}\mu_{\bm{\theta}_{k}}(\bm{\xeta}_{k+1})\rangle\right)\\
=&{({z}_{k}-{z}_{k+1})}+{\langle \bm{T}_{k+1}-\bm{T}_{k}, \mathrm{\Pi}_{\bm{\theta}_{k}}\mu_{\bm{\theta}_{k}}(\bm{\xeta}_{k+1})\rangle},\\
\end{split}
\end{equation*}
where ${z}_{k}=\langle \bm{T}_{k}, \mathrm{\Pi}_{\bm{\theta}_{k-1}}\mu_{\bm{\theta}_{k-1}}(\bm{x}_{k})\rangle$. By the regularity assumption (\ref{poisson_reg}) and Lemma \ref{theta_lip}, 
$$\E\langle \bm{T}_{k+1}-\bm{T}_{k}, \mathrm{\Pi}_{\bm{\theta}_{k}}\mu_{\bm{\theta}_{k}}(\bm{\xeta}_{k+1})\rangle\leq   \E[\| \bm{\theta}_{k+1}-\bm{\theta}_{k}\| ] \cdot \E[\| \mathrm{\Pi}_{\bm{\theta}_{k}}\mu_{\bm{\theta}_{k}}(\bm{\xeta}_{k+1})\| ] \leq 2Q C \omega_{k+1}.\eqno{(\text{D}_{23})}$$

Regarding $\text{D}_3$, since $\rho(\btheta_k, \bx_{k+1})$ is bounded, applying Cauchy–Schwarz inequality gives
$${\E[\langle \bT_k, \omega_{k+1}\rho(\btheta_k, \bx_{k+1}))]\leq 2Q^2\omega_{k+1}} \eqno{(\text{D}_{3})}$$

Finally, adding (\ref{first_term}), $\text{D}_1$, $\text{D}_{21}$, $\text{D}_{22}$, $\text{D}_{23}$ and $\text{D}_3$ together, it follows that for a constant $C_0 = G+10Q^2C+4QC+4Q^2$,
\begin{equation}
\begin{split}
\label{key_eqn}
\E\left[\| \bm{T}_{k+1}\| ^2\right]&\leq (1-2\omega_{k+1}\phi+G\omega^2_{k+1} )\E\left[\| \bm{T}_{k}\| ^2\right]+C_0\omega^2_{k+1} +2\Delta_k\omega_{k+1} +2\E[z_{k}-z_{k+1}]\omega_{k+1}.
\end{split}
\end{equation}
Moreover, from (\ref{compactness}) and (\ref{poisson_reg}), $\E[\mid z_{k}\mid ]$ is upper bounded by
\begin{equation}
\begin{split}
\label{condition:z}
\E[\mid z_{k}\mid ]=\E[\langle \bm{T}_{k}, \mathrm{\Pi}_{\bm{\theta}_{k-1}}\mu_{\bm{\theta}_{k-1}}(\bm{x}_{k})\rangle]\leq \E[\| \bT_k\| ]\E[\| \mathrm{\Pi}_{\bm{\theta}_{k-1}}\mu_{\bm{\theta}_{k-1}}(\bm{x}_{k})\| ]\leq 2QC.
\end{split}
\end{equation}

According to Lemma $\ref{lemma:4}$, we can choose $\lambda_0$ and $k_0$ such that 
\begin{align*}
\E[\| \bm{T}_{k_0}\| ^2]\leq \psi_{k_0}=\lambda_0 \omega_{k_0}+\frac{1}{\phi}\sup_{i\geq k_0}\Delta_{i},
\end{align*}
which satisfies the conditions ($\ref{lemma:3-a}$) and ($\ref{lemma:3-b}$) of Lemma $\ref{lemma:3-all}$. Applying Lemma $\ref{lemma:3-all}$ leads to
\begin{equation}
\begin{split}
\label{eqn:9}
\E\left[\| \bm{T}_{k}\| ^2\right]\leq \psi_{k}+\E\left[\sum_{j=k_0+1}^{k}\Lambda_j^k \left(z_{j-1}-z_{j}\right)\right],
\end{split}
\end{equation}
where $\psi_{k}=\lambda_0 \omega_{k}+\frac{1}{\phi}\sup_{i\geq k_0}\Delta_{i}$ for all $k>k_0$. Based on ($\ref{condition:z}$) and the increasing condition of $\Lambda_{j}^k$ in Lemma $\ref{lemma:2}$, we have
\begin{equation}
\small
\begin{split}
\label{eqn:10}
&\E\left[\mid \sum_{j=k_0+1}^{k} \Lambda_j^k\left(z_{j-1}-z_{j}\right)\mid \right]
=\E\left[\mid \sum_{j=k_0+1}^{k-1}(\Lambda_{j+1}^k-\Lambda_j^k)z_j-2\omega_{k}z_{k}+\Lambda_{k_0+1}^k z_{k_0}\mid \right]\\
\leq& \sum_{j=k_0+1}^{k-1}2(\Lambda_{j+1}^k-\Lambda_j^k)QC+\E[\mid 2\omega_{k} z_{k}\mid ]+2\Lambda_k^k QC\\
\leq& 2(\Lambda_k^k-\Lambda_{k_0}^k)QC+2\Lambda_k^k QC+2\Lambda_k^k QC\\
\leq& 6\Lambda_k^k QC.
\end{split}
\end{equation}

Given $\psi_{k}=\lambda_0 \omega_{k}+\frac{1}{\phi}\sup_{i\geq k_0}\Delta_{i}$ which satisfies the conditions ($\ref{lemma:3-a}$) and ($\ref{lemma:3-b}$) of Lemma $\ref{lemma:3-all}$,
it follows from (\ref{eqn:9}) and (\ref{eqn:10}) 
that  the following inequality holds for any $k>k_0$,
 
\begin{equation*}
\E[\| \bm{T}_{k}\| ^2]\leq \psi_{k}+6\Lambda_k^k QC=\left(\lambda_0+12QC\right)\omega_{k}+\frac{1}{\phi}\sup_{i\geq k_0}\Delta_{i}=\lambda \omega_{k}+\frac{1}{\phi}\sup_{i\geq k_0}\Delta_{i},
\end{equation*}
where $\lambda=\lambda_0+12QC$, $\lambda_0=\frac{2G\sup_{i\geq k_0} \Delta_i + 2C_0\phi}{C_1\phi}$, $\small{C_1=\lim \inf 2\phi \dfrac{\omega_{k}}{\omega_{k+1}}+\dfrac{\omega_{k+1}-\omega_{k}}{\omega^2_{k+1}}>0}$, $C_0=G+5Q^2C+2QC+2Q^2$ and $G=4 Q^2(1+Q^2)$.
\end{proof}

\subsection{Technical Lemmas} \label{Lemmasection}

\begin{lemma}\label{bias_in_SA}
The stochastic energy estimator $\widetilde U(\bx)$ leads to a controllable bias in the random-field function. 
\label{bias_sa}
\begin{equation*}\label{eq_bias_in_SA}
 \mid \E[\widetilde H_i(\btheta,\bx)]- H_i(\btheta,\bx)\mid =\mathcal{O}\left(\Var(\xi_n)\right),
\end{equation*}
where the expectation $\E[\cdot]$ is taken with respect to the random noise in the stochastic energy estimator of $\widetilde U(\cdot)$. 
\end{lemma}

\begin{proof}
Denote the noise in the stochastic energy estimator by $\xi(\bx)$, such that $\widetilde U(\cdot)=U(\cdot) + \xi(\cdot)$. Recall that $\widetilde H_i(\btheta,\bx)={\theta}(J_{\widetilde U}(\bx))\left(1_{i= J_{\widetilde U}(\bx)}-{\theta}(i)\right)$ and $J_{\widetilde U}(\bx)\in\{1,2,\cdots, m\}$ satisfies $u_{J_{\widetilde U}(\bx)-1}< \frac{N}{n}\widetilde U(\bx)\leq u_{J_{\widetilde U}(\bx)}$ for a set of energy partitions $\{u_i\}_{i=0}^{m}$. We can interpret $\widetilde H_i(\btheta,\bx)$ as a non-linear transformation $\Phi$ that maps $\widetilde U(\bx)$ to $(0, 1)$. Similarly, $ H_i(\btheta,\bx)$ maps $U(\bx)$ to $(0, 1)$. In what follows, the bias of random-field function is upper bounded as follows
\begin{equation*}
\begin{split}
     \mid \E[\widetilde H_i(\btheta,\bx)]- H_i(\btheta,\bx)\mid &=\mid \int \Phi(U(\bx)+\xi(\bx))-\Phi(U(\bx))d\mu(\xi(\bx))\mid \\
     &=\mid \int \xi(\bx) \Phi'(U(\bx))+\frac{\xi(\bx)^2}{2} \Phi''(u) d\mu(\xi(\bx))\mid \\
     &= \mathcal{O}\left(\Var(\xi_n)\right), \\
\end{split}
\end{equation*}
where the second equality follows from Taylor expansion for some energy $u$ and the third equality follows because the stochastic energy estimator is unbiased; $\Phi'(U(\bx))=\mathcal{O}( \frac{\theta(J(\bx))-\theta(J(\bx)-1)}{\Delta u})$ is clearly bounded due to the definition of $\btheta$; a similar conclusion also applies to $\Phi''(\cdot)$.

\end{proof}

\begin{lemma}\label{partition_order}
Suppose Assumption \ref{ass2a} holds, and $u_1$ and 
$u_{m-1}$ are fixed such that $\Psi(u_1)>\nu$ and $\Psi(u_{m-1})>1-\nu$ for some small constant $\nu>0$. For any bounded function $f(\bx)$, we have 
\label{m_order}
\begin{equation}\label{i_2}
    \int_{\MX} f(\bx)\left(\varpi_{\Psi_{\btheta}}(\bx)-\varpi_{\widetilde\Psi_{\btheta}}(\bx)\right) d\bx=\mathcal{O}\left(\frac{1}{m}\right).
\end{equation}
\end{lemma}

\begin{proof}
Recall that $\varpi_{\widetilde\Psi_{\btheta}}(\bx)= \frac{1}{Z_{\btheta}} 
\frac{\pi(\bx)}{\theta^{\zeta}(J(\bx))}$ and $\varpi_{\Psi_{\btheta}}(\bx)=\frac{1}{Z_{\Psi_{\btheta}}}\frac{\pi(\bx)}{\Psi^{\zeta}_{\btheta}(U(\bx))}$. Since $f(\bx)$ is bounded, 
it suffices to show 
\begin{equation}
\begin{split}
    &\int_{\MX} \frac{1}{Z_{\btheta}} 
\frac{\pi(\bx)}{\theta^{\zeta}(J(\bx))}-\frac{1}{Z_{\Psi_{\btheta}}}\frac{\pi(\bx)}{\Psi^{\zeta}_{\btheta}(U(\bx))} d\bx\\
\leq &\int_{\MX} \mid \frac{1}{Z_{\btheta}} 
\frac{\pi(\bx)}{\theta^{\zeta}(J(\bx))}-\frac{1}{Z_{\btheta}}\frac{\pi(\bx)}{\Psi^{\zeta}_{\btheta}(U(\bx))}\mid d\bx+\int_{\MX}\mid \frac{1}{Z_{\btheta}}\frac{\pi(\bx)}{\Psi^{\zeta}_{\btheta}(U(\bx))}-\frac{1}{Z_{\Psi_{\btheta}}}\frac{\pi(\bx)}{\Psi^{\zeta}_{\btheta}(U(\bx))}\mid  d\bx\\
=&\underbrace{\frac{1}{Z_{\btheta}}\sum_{i=1}^m \int_{\MX_i} \mid  
\frac{\pi(\bx)}{\theta^{\zeta}(i)}-\frac{\pi(\bx)}{\Psi^{\zeta}_{\btheta}(U(\bx))}\mid d\bx}_{\text{I}_1}+\underbrace{\sum_{i=1}^m\mid \frac{1}{Z_{\btheta}}-\frac{1}{Z_{\Psi_{\btheta}}}\mid \int_{\MX_i}\frac{\pi(\bx)}{\Psi^{\zeta}_{\btheta}(U(\bx))} d\bx}_{\text{I}_2}=\mathcal{O}\left(\frac{1}{m}\right),\\
\end{split}
\end{equation}
where $Z_{\btheta}=\sum_{i=1}^m \int_{\MX_i} \frac{\pi(\bx)}{\theta(i)^{\zeta}}d\bx$, $Z_{\Psi_{\btheta}}=\sum_{i=1}^{m}\int_{\MX_i} \frac{\pi(\bx)}{\Psi^{\zeta}_{\btheta}(U(\bx))}d\bx$, and $\Psi_{\btheta}(u)$ is a piecewise continuous function defined in (\ref{new_design_appendix}).

By Assumption \ref{ass2a}, $\inf_{\Theta} \theta(i)>0$ for any $i$. 
Further, by the mean-value theorem, which implies $\mid x^{\zeta}-y^{\zeta}\mid \lesssim \mid x-y\mid  z^{\zeta}$ for any $\zeta>0, x\leq y$ and $z\in[x, y]\subset [u_1, \infty)$, we have 
\begin{equation*}
\small
\begin{split}
    \text{I}_1&=\frac{1}{Z_{\btheta}}\sum_{i=1}^m \int_{\MX_i} \mid  
\frac{\theta^{\zeta}(i)-\Psi^{\zeta}_{\btheta}(U(\bx))}{\theta^{\zeta}(i)\Psi^{\zeta}_{\btheta}(U(\bx))}\mid \pi(\bx)d\bx\lesssim \frac{1}{Z_{\btheta}}\sum_{i=1}^m \int_{\MX_i} 
\frac{\mid \Psi_{\btheta}(u_{i-1})-\Psi_{\btheta}(u_i)\mid }{\theta^{\zeta}(i)}\pi(\bx)d\bx\\
&\leq \max_i \mid \Psi_{\btheta}(u_{i}-\Delta u)-\Psi_{\btheta}(u_i)\mid   \frac{1}{Z_{\btheta}}\sum_{i=1}^m \int_{\MX_i} 
\frac{\pi(\bx)}{\theta^{\zeta}(i)}d\bx=\max_i \mid \Psi_{\btheta}(u_{i}-\Delta u)-\Psi_{\btheta}(u_i)\mid \lesssim \Delta u=\mathcal{O}\left(\frac{1}{m}\right),
\end{split}
\end{equation*}
where the last inequality follows by Taylor expansion, 
and the last equality follows as $u_1$ and $u_{m-1}$ 
are fixed. Similarly, we have 
\begin{equation*}
    \begin{split}
        \text{I}_2= \mid \frac{1}{Z_{\btheta}}-\frac{1}{Z_{\Psi_{\btheta}}}\mid Z_{\Psi_{\btheta}}=\frac{ \mid Z_{\Psi_{\btheta}}-Z_{\btheta}\mid }{Z_{\btheta}}\leq \frac{1}{Z_{\btheta}}\sum_{i=1}^m \int_{\MX_i} \mid \frac{\pi(\bx)}{\theta^{\zeta}(i)}-\frac{\pi(\bx)}{\Psi^{\zeta}_{\btheta}(U(\bx))}\mid d\bx=\text{I}_1=\mathcal{O}\left(\frac{1}{m}\right).
    \end{split}
\end{equation*}
The proof can then be concluded by combining the orders of $\text{I}_1$ and $\text{I}_2$. 
\end{proof}

\begin{lemma}
\label{convex_property}
Given $\sup\{\omega_k\}_{k=1}^{\infty}\leq 1$, there exists a constant $G=4 Q^2(1+Q^2)$ such that
\begin{equation} \label{bound2}
\|  \widetilde H(\bm{\theta}_k, \bm{\xeta}_{k+1})+\omega_{k+1}\rho(\btheta_k, \bx_{k+1})\| ^2 \leq G (1+\| \bm{\theta}_k-\btheta_{\star}\| ^2). 
\end{equation}
\end{lemma}
\begin{proof}

According to the compactness condition in Assumption \ref{ass2a}, we have
\begin{equation}
\label{mid_1}
\| \widetilde H(\bm{\theta}_k, \bm{\xeta}_{k+1})\| ^2\leq Q^2 (1+\| \bm{\theta}_k\| ^2) = 
 Q^2 (1+\| \bm{\theta}_k-\btheta_{\star}+\btheta_{\star}\| ^2)\leq Q^2 (1+2\| \bm{\theta}_k-\btheta_{\star}\| ^2+2Q^2).
\end{equation}

Therefore, using (\ref{mid_1}), we can show that for a constant $G=4Q^2(1+Q^2)$
\begin{equation*}
\small
\begin{split}
    &\ \ \ \|  \widetilde H(\bm{\theta}_k, \bm{\xeta}_{k+1})+\omega_{k+1}\rho(\btheta_k, \bx_{k+1})\| ^2 \\
    &\leq 2\| \widetilde H(\bm{\theta}_k, \bm{\xeta}_{k+1})\| ^2 + 2\omega_{k+1}^2 \| \rho(\btheta_k, \bx_{k+1})\| ^2\\
    &\leq 2Q^2 (1+2\| \bm{\theta}_k-\btheta_{\star}\| ^2+2Q^2) + 2Q^2\\
    &\leq 2Q^2 (2+2Q^2+(2+2Q^2)\| \bm{\theta}_k-\btheta_{\star}\| ^2)\\
    &\leq G (1+\| \bm{\theta}_k-\btheta_{\star}\| ^2).
\end{split}
\end{equation*}
\end{proof}

\begin{lemma}
\label{theta_lip}Given $\sup\{\omega_k\}_{k=1}^{\infty}\leq 1$, we have that
\begin{equation}
\label{lip_theta}
    \| \btheta_{k}-\btheta_{k-1}\| \leq 2\omega_{k} Q
\end{equation}
\end{lemma}

\begin{proof}
Following the update $\btheta_k-\btheta_{k-1}=\omega_k \widetilde H(\bm{\theta}_{k-1}, \bm{x}_{k})+\omega_{k}^2 \rho(\btheta_{k-1}, \bx_{k})$, we have that
$$\| \btheta_{k}-\btheta_{k-1}\| = \| \omega_k \widetilde H(\bm{\theta}_{k-1}, \bm{x}_{k})+\omega_{k}^2 \rho(\btheta_{k-1}, \bx_{k})\| \leq \omega_k\|  \widetilde H(\bm{\theta}_{k-1},\bm{x}_{k})\| +\omega_{k}^2\|  \rho(\btheta_{k-1}, \bx_{k})\| .$$
By the compactness condition in Assumption \ref{ass2a} and $\sup\{\omega_k\}_{k=1}^{\infty}\leq 1$, (\ref{lip_theta}) can be derived.
\end{proof}

\begin{lemma}
\label{lemma:4}
There exist constants $\lambda_0$ and $k_0$ such that $\forall \lambda\geq\lambda_0$ and $\forall k> k_0$, the sequence $\{\psi_{k}\}_{k=1}^{\infty}$, where $\psi_{k}=\lambda\omega_{k}+\frac{1}{\phi} \sup_{i\geq k_0}\Delta_i$, satisfies
\begin{equation}
\begin{split}
\label{key_ieq}
\psi_{k+1}\geq& (1-2\omega_{k+1}\phi+G\omega_{k+1}^2)\psi_{k}+C_0\omega_{k+1}^2  +2\Delta_k\omega_{k+1}.
\end{split}
\end{equation}
\begin{proof}
By replacing $\psi_{k}$ with $\lambda\omega_{k}+\frac{1}{\phi} \sup_{i\geq k_0}\Delta_i$ in ($\ref{key_ieq}$), it suffices to show
\begin{equation*}
\small
\begin{split}
\label{lemma:loss_control}
\lambda \omega_{k+1}+\frac{1}{\phi} \sup_{i\geq k_0}\Delta_i\geq& (1-2\omega_{k+1}\phi+G\omega_{k+1}^2)\left(\lambda \omega_{k}+\frac{1}{\phi} \sup_{i\geq k_0}\Delta_i\right)+C_0\omega_{k+1}^2 + 2\Delta_k\omega_{k+1}.
\end{split}
\end{equation*}

which is equivalent to proving
\begin{equation*}
\small
\begin{split}
&\lambda (\omega_{k+1}-\omega_k+2\omega_k\omega_{k+1}\phi-G\omega_k\omega_{k+1}^2)\geq  \frac{1}{\phi}\sup_{i\geq k_0}\Delta_i(-2\omega_{k+1}\phi+G\omega_{k+1}^2 )+C_0\omega_{k+1}^2+ 2\Delta_k\omega_{k+1}.
\end{split}
\end{equation*}

Given the step size condition in ($\ref{a1}$), we have $$\small{\omega_{k+1}-\omega_{k}+2 \omega_{k}\omega_{k+1}\phi \geq C_1 \omega_{k+1}^2},$$ 
where $\small{C_1=\lim \inf 2\phi  \dfrac{\omega_{k}}{\omega_{k+1}}+\dfrac{\omega_{k+1}-\omega_{k}}{\omega^2_{k+1}}>0}$. Combining $-\sup_{i\geq k_0}\Delta_i\leq \Delta_k$, it suffices to prove
\begin{equation}
\begin{split}
\label{loss_control-2}
\lambda \left(C_1-G\omega_{k}\right)\omega^2_{k+1}\geq  \left(\frac{G}{\phi} \sup_{i\geq k_0}\Delta_i+C_0\right)\omega^2_{k+1}.
\end{split}
\end{equation}

It is clear that for a large enough $k_0$ and $\lambda_0$ such that $\omega_{k_0}\leq \frac{C_1}{2G}$, $\lambda_0=\frac{2G\sup_{i\geq k_0} \Delta_i + 2C_0\phi}{C_1\phi}$, the desired conclusion ($\ref{loss_control-2}$) holds for all such $k\geq k_0$ and $\lambda\geq \lambda_0$.
\end{proof}
\end{lemma}

The following lemma is a restatement of Lemma 25 (page 247) from \cite{Albert90}.
\begin{lemma}
\label{lemma:2}
Suppose $k_0$ is an integer satisfying
$\inf_{k> k_0} \dfrac{\omega_{k+1}-\omega_{k}}{\omega_{k}\omega_{k+1}}+2\phi-G\omega_{k+1}>0$ 
for some constant $G$. 
Then for any $k>k_0$, the sequence $\{\Lambda_k^K\}_{k=k_0, \ldots, K}$ defined below is increasing and uppered bounded by $2\omega_{k}$
\begin{equation}  
\Lambda_k^K=\left\{  
             \begin{array}{lr}  
             2\omega_{k}\prod_{j=k}^{K-1}(1-2\omega_{j+1}\phi+G\omega_{j+1}^2) & \text{if $k<K$},   \\  
              & \\
             2\omega_{k} &  \text{if $k=K$}.
             \end{array}  
\right.  
\end{equation} 
\end{lemma}

\begin{lemma}
\label{lemma:3-all}
Let $\{\psi_{k}\}_{k> k_0}$ be a series that satisfies the following inequality for all $k> k_0$
\begin{equation}
\begin{split}
\label{lemma:3-a}
\psi_{k+1}\geq &\psi_{k}\left(1-2\omega_{k+1}\phi+G\omega^2_{k+1}\right)+C_0\omega^2_{k+1} + 2 \Delta_k\omega_{k+1},
\end{split}
\end{equation}
and assume there exists such $k_0$ that 
\begin{equation}
\begin{split}
\label{lemma:3-b}
\E\left[\| \bm{T}_{k_0}\| ^2\right]\leq \psi_{k_0}.
\end{split}
\end{equation}
Then for all $k> k_0$, we have
\begin{equation}
\begin{split}
\label{result}
\E\left[\| \bm{T}_{k}\| ^2\right]\leq \psi_{k}+\sum_{j=k_0+1}^{k}\Lambda_j^k (z_{j-1}-z_j).
\end{split}
\end{equation}
\end{lemma}

\begin{proof}
We prove by the induction method. Assuming (\ref{result}) is true and applying (\ref{key_eqn}), we have that 
\begin{equation*}
\small
\begin{split}
    \E\left[\| \bm{T}_{k+1}\| ^2\right]&\leq (1-2\omega_{k+1}\phi+\omega^2_{k+1} G)(\psi_{k}+\sum_{j=k_0+1}^{k}\Lambda_j^k (z_{j-1}-z_j))\\
    &\ \ \ \ \ \ \ \ +C_0\omega^2_{k+1} +2 \Delta_k\omega_{k+1}+2\omega_{k+1}\E[z_{k}-z_{k+1}]\\
\end{split}
\end{equation*}

Combining (\ref{key_ieq}) and Lemma.\ref{lemma:2}, respectively, we have
\begin{equation*}
\small
\begin{split}
    \E\left[\| \bm{T}_{k+1}\| ^2\right]&\leq  \psi_{k+1}+(1-2\omega_{k+1}\phi+\omega^2_{k+1} G)\sum_{j=k_0+1}^{k}\Lambda_j^k (z_{j-1}-z_j)+2\omega_{k+1}\E[z_{k}-z_{k+1}]\\
    & \leq \psi_{k+1}+\sum_{j=k_0+1}^{k}\Lambda_j^{k+1} (z_{j-1}-z_j)+\Lambda_{k+1}^{k+1}\E[z_{k}-z_{k+1}]\\
    & \leq \psi_{k+1}+\sum_{j=k_0+1}^{k+1}\Lambda_j^{k+1} (z_{j-1}-z_j).\\
\end{split}
\end{equation*}
\end{proof}

\section{Ergodicity and Dynamic Importance Sampler}
\setcounter{lemma}{0}

\label{ergodicity}
Our interest is to analyze the deviation between the weighted averaging estimator $\frac{1}{k}\sum_{i=1}^k\theta_{i}^{\zeta}( \tilde{J}(\bx_i)) f(\bx_i)$ and posterior expectation $\int_{\MX}f(\bx)\pi(d\bx)$ for a \textcolor{black}{bounded} function $f$. To accomplish this analysis, we first study the convergence of the 
posterior sample mean $\frac{1}{k}\sum_{i=1}^k f(\bx_i)$ 
to the posterior expectation $\bar{f}=\int_{\MX}f(\bx)\varpi_{\Psi_{\btheta_{\star}}}(\bx)(d\bx)$ and then extend it to $\int_{\MX}f(\bx)\varpi_{\widetilde{\Psi}_{\btheta_{\star}}}(\bx)(d\bx)$. The key tool for establishing the ergodic theory is still the Poisson equation 
which is used to characterize the fluctuation 
between $f(\bx)$ and $\bar f$: 
\begin{equation}
    \mathcal{L}g(\bx)=f(\bx)-\bar f,
\end{equation}
where $g(\bx)$ is the solution of the Poisson equation, and $\mathcal{L}$ is the infinitesimal generator of the Langevin diffusion $\mathcal{L}g:=\langle\nabla g, \nabla L(\cdot, \btheta_{\star})\rangle+\textcolor{black}{\tau}\nabla^2g.$

Similar to the proof of Lemma \ref{lyapunov}, the existence of the solution of the Poisson's equation has been established in \cite{mattingly02, VollmerZW2016}. Moreover, the perturbations of $\E[f(\bx_k)]-\bar f$ are properly bounded given regularity properties for $g(\bx)$, where the 0-th, 1st, and 2nd order of the regularity properties has been established in \cite{Mackey18}. In what follows, we present a lemma, which is majorly adapted from Theorem 2 of \cite{Chen15} with a fixed learning rate $\epsilon$.

\begin{lemma}[Convergence of the Averaging Estimators. Formal statement of lemma \ref{avg_converge}]
\label{avg_converge_appendix}
Suppose Assumptions \ref{ass2a}-\ref{ass1} hold. For any bounded function $f$, 
\begin{equation*}
\small
\begin{split}
    \mid \E\left[\frac{\sum_{i=1}^k f(\bx_i)}{k}\right]-\int_{\MX}f(\bx)\varpi_{\widetilde{\Psi}_{\btheta_{\star}}}(\bx)d\bx\mid &=
    \mathcal{O}\left(\frac{1}{k\epsilon}+\sqrt{\epsilon}+\sqrt{\frac{\sum_{i=1}^k \omega_i}{k}}+ \frac{1}{\sqrt{m}}+
    \sqrt{\Var(\xi_n)}\right), \\
\end{split}
\end{equation*}
where $k_0$ is a sufficiently large constant, $\varpi_{\widetilde{\Psi}_{\btheta_{\star}}}(\bx) \propto  
\frac{\pi(\bx)}{\theta_{\star}^{\zeta}(J(\bx))}$, and $\frac{\sum_{i=1}^k \omega_i}{k} =o(\frac{1}{\sqrt{k}})$ as implied by Assumption \ref{ass1}. 
\end{lemma}

\begin{proof}
We rewrite the CSGLD algorithm as follows:
\begin{equation*}
\begin{split}
    \bm{x}_{k+1}&=\bx_k- \epsilon_k\nabla_{\bx} \widetilde{L}(\bx_k, \btheta_k)+\mathcal{N}({0, 2\epsilon_k \tau\bm{I}})\\
    &=\bx_k- \epsilon_k\left(\nabla_{\bx} 
    \widehat{L}(\bx_k, \btheta_{\star})+{\Upsilon}(\bx_k, \btheta_k, \btheta_{\star})\right)+\mathcal{N}({0, 2\epsilon_k \tau\bm{I}}),
\end{split}
\end{equation*}
where    
$\nabla_{\bx} \widehat{L}(\bx,\btheta)= \frac{N}{n} \left[1+  \frac{\zeta\tau}{\Delta u}  \left(\textcolor{black}{\log \theta({J}(\bx))-\log\theta(({J}(\bx)-1)\vee 1)} \right) \right]  \nabla_{\bx} \widetilde U(\bx)$,  $\nabla_{\bx} \widetilde{L}(\bx,\btheta)$ is as 
defined in Section \ref{Alg:app},
 and the bias term is given by ${\Upsilon}(\bx_k,\btheta_k,\btheta_{\star})=\nabla_{\bx} \widetilde{L}(\bx_k,\btheta_k)-\nabla_{\bx} \widehat{L}(\bx_k,\btheta_{\star})$.

By Assumption \ref{ass2}, we have $\| \nabla_{\bx} U(\bx)\| =\| \nabla_{\bx} U(\bx)-\nabla_{\bx} U(\bx_{\star})\| \lesssim \| \bx-\bx_{\star}\| \leq \| \bx\| +\| \bx_{\star}\| $ for some optimum. Then the $L^2$ upper bound in Lemma \ref{lemma:1} implies that $\nabla_{\bx} U(\bx)$ has a bounded second moment. Combining Assumption \ref{ass4}, we have $\E\left[\| \nabla_{\bx} \widetilde U(\bx)\| ^2\right]<\infty$. Further by Eve’s law (i.e., the
variance decomposition formula), it is easy to derive that $\E \left[\|  \nabla_{\bx} \widetilde{U}(\bx) \| \right]<\infty$.
Then, by the triangle inequality and Jensen's inequality, 
\begin{equation}
\label{latent_bias}
\small
\begin{split}
    \| \E[\Upsilon(\bx_k,\btheta_k,\btheta_{\star})]\| &\leq 
    \E[\| \nabla_{\bx} \widetilde{L}(\bx_k, \btheta_k)-\nabla_{\bx} \widetilde{L}(\bx_k, \btheta_{\star})\| ] + \E[\| \nabla_{\bx} \widetilde{L}(\bx_k, \btheta_{\star})-\nabla_{\bx} \widehat{L}(\bx_k, \btheta_{\star})\| ] \\
    &\lesssim  \E[\| \btheta_k-\btheta_{\star}\| ]+\mathcal{O}(\Var(\xi_n))\leq \sqrt{\E[\| \btheta_k-\btheta_{\star}\| ^2]}+\mathcal{O}(\Var(\xi_n))\\
    &\leq \mathcal{O}\left( \sqrt{\omega_{k}+\epsilon+
\frac{1}{m} +\Var(\xi_n)}\right),
\end{split}
\end{equation}
where Assumption \ref{ass2a} and Theorem \ref{latent_convergence} are used to derive the smoothness of $\nabla_{\bx} \tilde{L}(\bx, \btheta)$ with respect to $\btheta$.

The ergodic average based on biased gradients and a fixed learning rate is a direct result of Theorem 2 of  \cite{Chen15} by imposing the regularity condition. By simulating from $\varpi_{\Psi_{\btheta_{\star}}}(\bx)\propto\frac{\pi(\bx)}{\Psi^{\zeta}_{\btheta_{\star}}(U(\bx))}$ and combining (\ref{latent_bias}) and Theorem \ref{latent_convergence}, we have 
\begin{equation*}
\small
\begin{split}
    \mid \E\left[\frac{\sum_{i=1}^k f(\bx_i)}{k}\right]-\int_{\MX}f(\bx) \varpi_{\Psi_{\btheta_{\star}}}(\bx)d\bx\mid &\leq \mathcal{O}\left(\frac{1}{k\epsilon}+\epsilon+\frac{\sum_{i=1}^k \| \E[\Upsilon(\bx_k,\btheta_k,\btheta_{\star})]\| }{k}\right)\\
    &\lesssim \mathcal{O}\left(\frac{1}{k\epsilon}+\epsilon+\frac{\sum_{i=1}^k \sqrt{\omega_i+\epsilon+\frac{1}{m}+\Var(\xi_n)}}{k}\right) \\
    &\leq \mathcal{O}\left(\frac{1}{k\epsilon}+\sqrt{\epsilon}+\sqrt{\frac{\sum_{i=1}^k \omega_i}{k}}+
    \frac{1}{\sqrt{m}}+\sqrt{\Var(\xi_n)}\right),
\end{split}
\end{equation*}
where the last inequality follows by repeatedly applying the inequality $\sqrt{a+b}\leq \sqrt{a}+\sqrt{b}$ and 
the inequality $\sum_{i=1}^k \sqrt{\omega_i}\leq \sqrt{k\sum_{i=1}^k \omega_i}$.

For any a bounded function $f(\bx)$, we have $\mid \int_{\MX}f(\bx) \varpi_{\Psi_{\btheta_{\star}}}(\bx)d\bx -  \int_{\MX}f(\bx) \varpi_{\widetilde{\Psi}_{\btheta_{\star}}}(\bx)d\bx\mid = \mathcal{O}(\frac{1}{m})$ by Lemma \ref{partition_order}. By the triangle inequality, we have 
\begin{equation*}
\small
\begin{split}
    \mid \E\left[\frac{\sum_{i=1}^k f(\bx_i)}{k}\right]-\int_{\MX}f(\bx) \varpi_{\widetilde{\Psi}_{\btheta_{\star}}}(\bx)d\bx\mid \leq  \mathcal{O}\left(\frac{1}{k\epsilon}+\sqrt{\epsilon}+\sqrt{\frac{\sum_{i=1}^k \omega_i}{k}}+ \frac{1}{\sqrt{m}}+
    \sqrt{\Var(\xi_n)}\right),
\end{split}
\end{equation*}
which concludes the proof. 
\end{proof}

Finally, we are ready to show the convergence of the weighted averaging estimator $\frac{\sum_{i=1}^k\theta_{i}
     ^{\zeta}(\tilde{J}(\bx_i)) f(\bx_i)}{\sum_{i=1}^k\theta_{i}^{\zeta}( 
      \tilde{J}(\bx_i))}$ to the posterior mean $\int_{\MX}f(\bx)\pi(d\bx)$.
\begin{theorem}[Convergence of the Weighted Averaging Estimators. Formal statement of Theorem \ref{w_avg_converge}] Assume Assumptions \ref{ass2a}-\ref{ass1} hold. For any bounded function $f$, we have that 
\label{w_avg_converge_appendix}
\begin{equation*}
\small
\begin{split}
    \mid \E\left[\frac{\sum_{i=1}^k\theta_{i}
     ^{\zeta}(\tilde{J}(\bx_i)) f(\bx_i)}{\sum_{i=1}^k\theta_{i}^{\zeta}( 
      \tilde{J}(\bx_i))}\right]-\int_{\MX}f(\bx)\pi(d\bx)\mid &= \mathcal{O}\left(\frac{1}{k\epsilon}+\sqrt{\epsilon}+\sqrt{\frac{\sum_{i=1}^k \omega_i}{k}}+\frac{1}{\sqrt{m}}+\sqrt{\Var(\xi_n)}\right). \\
\end{split}
\end{equation*}
\end{theorem}

\begin{proof}

Applying triangle inequality and $\mid \E[x]\mid \leq \E[\mid x\mid ]$, we have
\begin{equation*}
    \begin{split}
        &\mid \E\left[\frac{\sum_{i=1}^k\theta_{i}
        ^{\zeta}( \tilde{J}(\bx_i)) f(\bx_i)}{\sum_{i=1}^k\theta_{i}^{\zeta}(
         \tilde{J}(\bx_i))}\right]-\int_{\MX}f(\bx)\pi(d\bx)\mid \\
        \leq &\underbrace{\E\left[\mid \frac{\sum_{i=1}^k\theta_{i}^{\zeta}
         (\tilde{J}(\bx_i))f(\bx_i)}
         {\sum_{i=1}^k\theta_{i}^{\zeta}(\tilde{J}(\bx_i)) }-\frac{\sum_{i=1}^k\theta_{i}^{\zeta}
         ({J}(\bx_i))f(\bx_i)}
         {\sum_{i=1}^k\theta_{i}^{\zeta}({J}(\bx_i)) }\mid \right]}_{\text{I}_1}\\
         &\quad+\underbrace{\E\left[\mid \frac{\sum_{i=1}^k\theta_{i}^{\zeta}
         ({J}(\bx_i))f(\bx_i)}
         {\sum_{i=1}^k\theta_{i}^{\zeta}({J}(\bx_i)) }-\frac{Z_{\btheta_{\star}}\sum_{i=1}^k\theta_{i}^{\zeta} ({J}(\bx_i)) f(\bx_i)}{k}\mid \right]}_{\text{I}_2}\\
        &\qquad\quad+ \underbrace{\E\left[\frac{Z_{\btheta_{\star}}}{k}\sum_{i=1}^k\mid \theta_i^{\zeta} ({J}(\bx_i))-\theta_{\star}^{\zeta}
      ({J}(\bx_i))  \mid  \cdot \mid f(\bx_i)\mid \right]}_{\text{I}_3} \\
      &\qquad\qquad\quad +\underbrace{\mid \E\left[\frac{Z_{\btheta_{\star}}}{k}\sum_{i=1}^k\theta_{\star}^{\zeta}
     ({J}(\bx_i)) f(\bx_i)\right]-\int_{\MX}f(\bx)\pi(d\bx)\mid }_{\text{I}_4}.
    \end{split}
\end{equation*}

For the term $\text{I}_1$, by applying mean-value theorem, we have 
\begin{equation}
\footnotesize
\begin{split}
    \text{I}_1&=\E\left[\mid \frac{\left(\sum_{i=1}^k\theta_{i}^{\zeta}(
         \tilde{J}(\bx_i))f(\bx_i)\right)\left(\sum_{i=1}^k\theta_{i}^{\zeta}(
         {J}(\bx_i))\right)-\left(\sum_{i=1}^k\theta_{i}^{\zeta}(
         {J}(\bx_i))f(\bx_i)\right)\left(\sum_{i=1}^k\theta_{i}^{\zeta}(
         \tilde{J}(\bx_i))\right)}
         {\left(\sum_{i=1}^k\theta_{i}^{\zeta}(
         \tilde{J}(\bx_i))\right)\left(\sum_{i=1}^k\theta_{i}^{\zeta}(
         {J}(\bx_i))\right)}\mid \right]\\
         &\lesssim \sup_i \Var(\xi_n) \E\left[\frac{\left(\sum_{i=1}^k\theta_{i}
     ^{\zeta}({J}(\bx_i)) f(\bx_i) \left(\sum_{i=1}^k\theta_{i}^{\zeta}(
         {J}(\bx_i))\right)\right)}{\left(\sum_{i=1}^k\theta_{i}^{\zeta}(
         {J}(\bx_i))\right)\left(\sum_{i=1}^k\theta_{i}^{\zeta}(
         {J}(\bx_i))\right)}\right]
         =\mathcal{O}\left(\sup_i\Var(\xi_n)\right).
\end{split}
\end{equation}

For the term $\text{I}_2$, 
by the boundedness of $\bTheta$ and $f$ and the assumption  $\inf_{\Theta}\theta^{\zeta}(i)>0$, we have
\begin{equation*}
\small
\begin{split}
    \text{I}_2=&\E\left[\mid \frac{\sum_{i=1}^k\theta_{i}^{\zeta}({J}(\bx_i))  f(\bx_i)}{\sum_{i=1}^k\theta_{i}^{\zeta}
    ({J}(\bx_i))
    }\left(1-\sum_{i=1}^k\frac{\theta_i^{\zeta}
    ({J}(\bx_i))
    }{k}Z_{\btheta_{\star}}\right)\mid \right]\\
    \lesssim & \E\left[\mid Z_{\btheta_{\star}}\frac{{\sum_{i=1}^k\theta_{i}^{\zeta}
    ({J}(\bx_i))
    }}{k}-1\mid \right]\\
    =&\E\left[\mid Z_{\btheta_{\star}}\sum_{i=1}^m \frac{\sum_{j=1}^k\left( \theta_j^{\zeta}(i)-\theta_{\star}^{\zeta}(i)+\theta_{\star}^{\zeta}(i)\right)1_{
    {J}(\bx_j)=i}}{k}-1\mid \right]\\
    \leq & \underbrace{\E\left[Z_{\btheta_{\star}}\sum_{i=1}^m \frac{\sum_{j=1}^k\mid  \theta_j^{\zeta}(i)-\theta_{\star}^{\zeta}(i)\mid  1_{{J}(\bx_j)=i}}{k} \right]}_{\text{I}_{21}} + \underbrace{\E\left[\mid  Z_{\btheta_{\star}}\sum_{i=1}^m \frac{\theta_{\star}^{\zeta}(i)\sum_{j=1}^k  1_{{J}(\bx_j)=i}}{k}-1\mid 
    \right]}_{\text{I}_{22}}.\\
\end{split}
\end{equation*}

For $\text{I}_{21}$, by first applying the inequality $\mid x^{\zeta}-y^{\zeta}\mid \leq \zeta \mid x-y\mid  z^{\zeta-1}$ for any $\zeta>0, x\leq y$ and $z\in[x, y]$ based on the mean-value theorem and then applying the Cauchy–Schwarz inequality, we have 
\begin{equation}\label{ii_21}
    \text{I}_{21}\lesssim \frac{1}{k}\E\left[ \sum_{j=1}^k\sum_{i=1}^m\mid  \theta_j^{\zeta}(i)-\theta_{\star}^{\zeta}(i)\mid  \right]\lesssim  \frac{1}{k}\E\left[ \sum_{j=1}^k\sum_{i=1}^m\mid  \theta_j(i)-\theta_{\star}(i)\mid  \right]\lesssim  \frac{1}{k}\sqrt{\sum_{j=1}^k\E\left[\left\|  \btheta_j-\btheta_{\star}\right\| ^2\right]},
\end{equation}
where the compactness of $\bTheta$ has been 
used in deriving the second inequality. 

For $\text{I}_{22}$, considering the following relation $$
    1=\sum_{i=1}^m\int_{\MX_i} \pi(\bx)d\bx=\sum_{i=1}^m\int_{\MX_i} \theta_{\star}^{\zeta}(i) \frac{\pi(\bx)}{\theta_{\star}^{\zeta}(i)}d\bx
    =Z_{\btheta_{\star}}\int_{\MX} \sum_{i=1}^m \theta_{\star}^{\zeta}(i) 1_{{J}(\bx)=i}\varpi_{\widetilde{\Psi}_{\btheta_{\star}}}(\bx)d\bx,$$ then we have

\begin{equation}
\begin{split}
    \text{I}_{22}&=\E\left[\mid  Z_{\btheta_{\star}}\sum_{i=1}^m \frac{\theta_{\star}^{\zeta}(i)\sum_{j=1}^k  1_{{J}(\bx_j)=i}}{k}-Z_{\btheta_{\star}}
    \int_{\MX} \sum_{i=1}^m \theta_{\star}^{\zeta}(i) 1_{{J}(\bx)=i}\varpi_{\widetilde{\Psi}_{\btheta_{\star}}}(\bx)d\bx\mid \right]\\
    &=Z_{\btheta_{\star}} \E\left[\mid  \frac{1}{k}\sum_{j=1}^k \left(\sum_{i=1}^m\theta_{\star}^{\zeta}(i)  1_{{J}(\bx_j)=i}\right)-\int_{\MX} \left(\sum_{i=1}^m \theta_{\star}^{\zeta}(i) 1_{{J}(\bx)=i}\right)\varpi_{\widetilde{\Psi}_{\btheta_{\star}}}(\bx)d\bx\mid \right]\\
    &= \mathcal{O}\left(\frac{1}{k\epsilon}+\sqrt{\epsilon}+\sqrt{\frac{\sum_{i=1}^k \omega_i}{k}}+\frac{1}{\sqrt{m}}+ \sqrt{\Var(\xi_n)} \right),
\end{split}
\end{equation}
where the last equality follows  from Lemma \ref{avg_converge_appendix} as the \textcolor{black}{step 
function $\sum_{i=1}^m \theta_{\star}^{\zeta}(i) 1_{{J}(\bx)=i}$} is integrable.

For $\text{I}_3$, by the boundedness of $f$,  
the mean value theorem and Cauchy-Schwarz inequality, 
we have 
\begin{equation}\label{ii_3}
\small
    \begin{split}
        \text{I}_3&\lesssim \E\left[\frac{1}{k}\sum_{i=1}^k\mid \theta_{i}
        ^{\zeta}({J}(\bx_i)) -\theta_{\star}^{\zeta}(
        {J}(\bx_i))\mid \right]\lesssim  \frac{1}{k}\E\left[ \sum_{j=1}^k\sum_{i=1}^m\mid  \theta_j(i)-\theta_{\star}(i)\mid  \right]\lesssim  \frac{1}{k}\sqrt{\sum_{j=1}^k\E\left[\left\|  \btheta_j-\btheta_{\star}\right\| ^2\right]}.\\
    \end{split}
\end{equation}

For the last term $\text{I}_4$, we first decompose $\int_{\MX} f(\bx) \pi(d\bx)$ into $m$ disjoint regions to facilitate the analysis
\begin{equation}
\label{split_posterior}
\footnotesize
\begin{split}
      \int_{\MX} f(\bx) \pi(d\bx)&=\int_{\cup_{j=1}^m \MX_j}  f(\bx) \pi(d\bx)=\sum_{j=1}^m\int_{\MX_j}\theta_{\star}^{\zeta}(j)  f(\bx) \frac{\pi(d\bx)}{\theta_{\star}^{\zeta}(j)}\\
      &=Z_{\btheta_{\star}}\int_{\MX} \sum_{j=1}^m \theta_{\star}(j)^{\zeta}f(\bx) 1_{
        {J}(\bx_i)=j 
        }\varpi_{\widetilde{\Psi}_{\btheta_{\star}}}(\bx)(d\bx).\\
\end{split}
\end{equation}

Plugging (\ref{split_posterior}) into the last term $\text{I}_4$, we have

\begin{equation}
\label{final_i2}
\small
    \begin{split}
        \text{I}_4&=\mid \E\left[\frac{Z_{\btheta_{\star}}}{k}\sum_{i=1}^k\sum_{j=1}^m\theta_{\star}(j)^{\zeta} f(\bx_i)1_{  {J}(\bx_i)=j
        }\right]-\int_{\MX}f(\bx)\pi(d\bx)\mid \\
        &= Z_{\btheta_{\star}}\mid \E\left[\frac{1}{k}\sum_{i=1}^k \left(\sum_{j=1}^m\theta_{\star}^{\zeta}(j) f(\bx_i)1_{
        {J}(\bx_i)=j 
        }\right)\right]-\int_{\MX}  \left(\sum_{j=1}^m\theta_{\star}^{\zeta}(j) f(\bx_i)1_{
        {J}(\bx_i)=j 
        }\right) \varpi_{\widetilde{\Psi}_{\btheta_{\star}}}(\bx)(d\bx)\mid \\
    \end{split}
\end{equation}

Applying the function \textcolor{black}{$\sum_{j=1}^m\theta_{\star}^{\zeta}(j) f(\bx_i)1_{
        {J}(\bx_i)=j 
        }$ }
to Lemma \ref{avg_converge_appendix} yields
\begin{equation}
\label{almost_i2}
\small
\begin{split}
      \mid \E\left[\frac{1}{k}\sum_{i=1}^k f(\bx_i)\right]-\int_{\MX}  f(\bx) \varpi_{\widetilde{\Psi}_{\btheta_{\star}}}(\bx)(d\bx)\mid  = \mathcal{O}\left(\frac{1}{k\epsilon}+\sqrt{\epsilon}+\sqrt{\frac{\sum_{i=1}^k \omega_i}{k}}+\frac{1}{\sqrt{m}}+ \sqrt{\Var(\xi_n)} \right).\\
\end{split}
\end{equation}

Plugging (\ref{almost_i2}) into (\ref{final_i2}) and combining $\text{I}_{1}$, $\text{I}_{21}$, $\text{I}_{22}$, $\text{I}_3$ and Theorem \ref{latent_convergence}, we have
\begin{equation*}
\small
\begin{split}
      \mid \E\left[\frac{\sum_{i=1}^k\theta_{i}
     ^{\zeta}(\tilde{J}(\bx_i)) f(\bx_i)}{\sum_{i=1}^k\theta_{i}^{\zeta}( 
      \tilde{J}(\bx_i))}\right]-\int_{\MX}f(\bx)\pi(d\bx)\mid  = \mathcal{O}\left(\frac{1}{k\epsilon}+\sqrt{\epsilon}+\sqrt{\frac{\sum_{i=1}^k \omega_i}{k}}+\frac{1}{\sqrt{m}}+ \sqrt{\Var(\xi_n)} \right),\\
\end{split}
\end{equation*}
which concludes the proof of the theorem.

\end{proof}

\section{More Discussions on the Algorithm}
\label{ext}

\subsection{An Alternative Numerical Scheme}
\label{alternative_scheme}
In addition to the numerical scheme used in (6) and (8) in the main body, we can also consider the following numerical scheme 
  \begin{equation*} \label{alternative_SGLDeq6}
  \footnotesize
 \begin{split}
  \small{\bx_{k+1}=\bx_k - \epsilon_{k+1} \frac{N}{n} \left[1+ 
   \zeta\tau\frac{\log {\theta}_{k}\big(\tilde J(\bx_k) \wedge m\big) - \log{\theta}_{k}\big(\tilde J(\bx_k)\big)}{\Delta u}  \right]  
    \nabla_{\bx} \widetilde U(\bx_k) +\sqrt{2 \tau \epsilon_{k+1}} \bw_{k+1}}.
 \end{split}
  \end{equation*}

Such a scheme leads to a similar theoretical result and a better treatment of $\Psi_{\btheta}(\cdot)$ for the subregions that contains stationary points.

\subsection{Bizarre Peaks in the Gaussian Mixture Distribution}

A bizarre peak always indicates that there is a stationary point of the same energy in somewhere of the sample space, as the sample space is partitioned according to the energy function in CSGLD. For example, we study a mixture distribution with asymmetric modes $\pi(x)=1/6 N(-6,1)+5/6 N(4,1)$. Figure \ref{bizzae} shows a bizarre peak at $x$. Although $x$ is not a local minimum, it has the same energy as ``-6'' which is a local minimum. Note that in CSGLD, $x$ and ``-6'' belongs to the same subregion.

\begin{figure}[ht]
\vspace{-0.05in}
  \centering
  \includegraphics[scale=0.27]{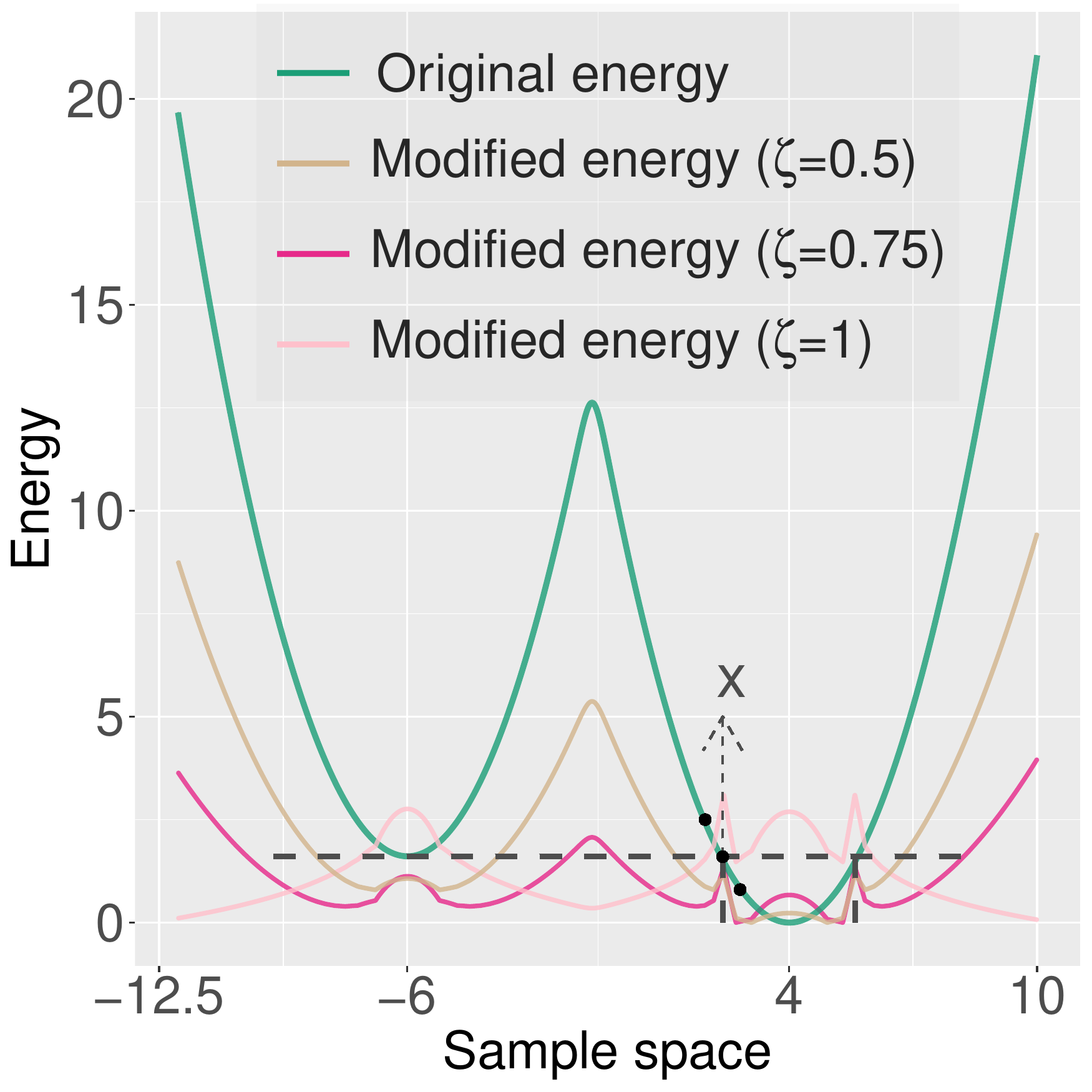}
  \caption{Explanation of bizarre peaks.}
  \label{bizzae}
  \vspace{-0.05in}
\end{figure}
\chapter{TECHNICAL PROOFS FOR CHAPTER \ref{ICSGLD}}

We summarize the supplementary material as follows: Section \ref{review_icsgld} provides the preliminary knowledge for stochastic approximation; Section \ref{convergence} shows a local stability condition that adapts to high losses; Section \ref{Gaussian_approx} proves the main asymptotic normality for the stochastic approximation process, which naturally yields the conclusion that interacting contour stochastic gradient Langevin dynamics (ICSGLD) is more efficient than the analogous single chain based on slowly decreasing step sizes; Section \ref{details_exp} details the experimental settings.

\section{Preliminaries: Gaussian Diffusions}
\label{review_icsgld}

Consider a stochastic linear differential equation
\begin{equation}
\label{slde}
    d\bU_t=h_{\btheta}(\btheta_t) \bU_t dt + \bR^{1/2}(\btheta_t)d\bW_t,
\end{equation}
where $\bU$ is a $m$-dimensional random vector, $h_{\btheta}:=\frac{d}{d\btheta} h(\btheta)$, $\bR(\btheta):=\sum_{k=-\infty}^{\infty} \cov_{\btheta}(H(\btheta, \bx_k), H(\btheta, \bx_0))$ is a positive definite matrix that depends on $\btheta(\cdot)$, $\bW\in\mathbb{R}^m$ is a standard Brownian motion. Given a large enough $t$ such that $\btheta_t$ converges to a fixed point $\widehat\btheta_{\star}$ sufficiently fast, we may write the diffusion associated with Eq.(\ref{slde}) as follow
\begin{equation}
\label{solution_slde}
    \bU_t\approx e^{-th_{\btheta}(\widehat\btheta_{\star})}\bU_0+\int_0^t e^{-(t-s)h_{\btheta}(\widehat\btheta_{\star})}\circ \bR(\widehat\btheta_{\star}) d\bW_s,
\end{equation}

Suppose that the matrix $h_{\btheta}(\widehat\btheta_{\star})$ is negative definite, then $\bU_t$ converges in distribution to a Gaussian variable 
\begin{equation*}
\begin{split}
    \E[\bU_t]&=e^{-t h_{\btheta}(\widehat\btheta_{\star})}\bU_0\\
\Var(\bU_t)&=\int_0^{t} e^{t h_{\btheta}(\widehat\btheta_{\star})}\circ \bR \circ e^{t h_{\btheta}(\widehat\btheta_{\star})} du.
\end{split}
\end{equation*}

The main goal of this supplementary file is to study the Gaussian approximation of the process $\omega_k^{-1/2}(\btheta_k-\widehat\btheta_{\star})$ to the solution Eq.(\ref{solution_slde}) for a proper step size $\omega_k$. Thereafter, the advantage of interacting mechanisms can be naturally derived.

\section{Stability and Convergence Analysis}
\label{convergence}

As required by the algorithm, we update $P$ contour stochastic gradient Langevin dynamics (CSGLD) simultaneously. For the notations, we denote the particle of the p-th chain at iteration $k$ by $\bx_{k}^{(p)}\in \MX\subset \mathbb{R}^d$ and the joint state of the $P$ parallel particles at iteration $k$ by $\bx_{k}^{\pop P}:=\left(\bx_{k}^{(1)}, \bx_{k}^{(2)}, \cdots, \bx_{k}^{(P)}\right)^\top\in \MX^{\pop P}\subset \mathbb{R}^{dP}$. We also denote the learning rate and step size at iteration $k$ by $\epsilon_k$ and $\omega_k$, respectively. We denote by  $\mathcal{N}({0, \bm{I}_{dP}})$ a standard $dP$-dimensional Gaussian vector and denote by $\zeta$ a positive hyperparameter.

\subsection{ICSGLD Algorithm} \label{Alg:app_icsgld}

First, we introduce the interacting contour stochastic gradient Langevin dynamics (ICSGLD) with $P$ parallel chains:
\begin{itemize}
\item[(1)] Simulate $\bx_{k+1}^{\pop P}=\bx_k^{\pop P}- \epsilon_k\nabla_{\bx} \widetilde \bL(\bx_k^{\pop P}, \btheta_k)+\mathcal{N}({0, 2\epsilon_k \tau\bm{I}_{dP}}), \ \ \ \ \ \ \ \ \ \ \ \ \ \ \ \ \ \ \ \ \ \ \ \ \ \ \ \ \ \ \ \ \ \ \ \ \ \ \ \ \  (\text{S}_1)$

\item[(2)] Optimize $\bm{\theta}_{k+1}=\bm{\theta}_{k}+\omega_{k+1} \widetilde \bH(\bm{\theta}_{k}, \bx_{k+1}^{\pop P}),
\ \ \ \ \ \ \ \ \ \ \ \ \ \ \ \ \ \ \ \ \ \ \ \ \ \ \ \ \ \ \ \ \ \ \ \ \ \ \ \ \ \ \ \ \ \ \ \ \ \ \ \ \ \ \ \ \ \ \ \ \ \ \ \ \ \ \ (\text{S}_2)$
\end{itemize}
where  $\nabla_{\bx} \widetilde \bL(\bx^{\pop P}, \btheta):=\left(\nabla_{\bx} \widetilde L(\bx^{(1)}, \btheta), \nabla_{\bx} \widetilde L(\bx^{(2)}, \btheta), \cdots, \nabla_{\bx} \widetilde L(\bx^{(P)}, \btheta)\right)^\top$, $\nabla_{\bx} \widetilde L(\bx, \btheta)$ is the 
stochastic adaptive gradient given by 
\begin{equation}
\label{adaptive_grad}
    \nabla_{\bx} \widetilde{L}(\bx,\btheta)= \frac{N}{n} \left[1+ 
   \frac{\zeta\tau}{\Delta u}  \left(\log \theta({J}_{\widetilde U}(\bx))-\log\theta((J_{\widetilde U}(\bx)-1)\vee 1) \right) \right]  
    \nabla_{\bx} \widetilde U(\bx).
\end{equation}

In particular, the interacting random-field function is written as \begin{equation}
\label{interactions}
    \widetilde \bH(\bm{\theta}_{k}, \bx_{k+1}^{\pop P})=\frac{1}{P}\sum_{p=1}^P \widetilde H(\btheta_k,\bx_{k+1}^{(p)}),
\end{equation} 
where each random-field function $\widetilde H(\btheta,\bx)=(\widetilde H_1(\btheta,\bx), \ldots,\widetilde H_m(\btheta,\bx))$ follows
\begin{equation}
\label{random_field_H}
     \widetilde H_i(\btheta,\bx)={\theta}( J_{\widetilde U}(\bx))\left(1_{i= J_{\widetilde U}(\bx)}-{\theta}(i)\right), \quad i=1,2,\ldots,m.
\end{equation}
Here $J_{\widetilde U}(\bx)$ denotes the index $i\in\{1, 2,3,\cdots, m\}$ such that $u_{i-1}< \frac{N}{n} \widetilde U(\bx)\leq u_i$ for a set of energy partitions $\{u_i\}_{i=0}^{m}$ and $\widetilde U(\bx)=\sum_{i\in B} U_i(\bx)$ where $U_i$ denotes the negative log of a posterior based on a single data point $i$ and $B$ denotes a mini-batch of data of size $n$. Note that the stochastic energy estimator $\widetilde U(\bx)$ results in a biased estimation for the partition index $J_{\widetilde U}(\bx)$ due to a non-linear transformation. To avoid such a bias asymptotically with respect to the learning rate $\epsilon$, we may consider a variance-reduced energy estimator $\widetilde U_{\text{VR}}(\bx)$ following \cite{deng_VR}
\begin{equation}
\label{VR_estimator}
    \frac{N}{n} \widetilde U_{\text{VR}}(\bx)=\frac{N}{n}\sum_{i\in B_k}\left( U_i(\bx) - U_i\left(\bx_{q\lfloor\frac{k}{q}\rfloor}\right) \right)+\sum_{i=1}^N U_i\left(\bx_{q\lfloor\frac{k}{q}\rfloor}\right),
\end{equation}
where the control variate $\bx_{q\lfloor \frac{k}{q}\rfloor}$ is updated every $q$ iterations.

Compared with the na\"ive parallelism of CSGLD, a key feature of the ICSGLD algorithm lies in the joint estimation of the interacting random-field function $\widetilde \bH(\bm{\theta}, \bx^{\pop P})$ in Eq.(\ref{interactions}) for the same mean-field function $h(\btheta)$.

\subsubsection{Discussions on the Hyperparameters}
\label{hyper_setup}
The most important hyperparameter is $\zeta$. A fine-tuned $\zeta$ usually leads to a small or even slightly negative learning rate in low energy regions to avoid local-trap problems. Theoretically, $\zeta$ affects the $L^2$ convergence rate hidden in the big-O notation in Lemma \ref{convex_appendix_icsgld}.

The other hyperparameters can be easily tuned. For example, the ResNet models yields the full loss ranging from 10,000 to 60,000 after warm-ups, we thus partition the sample space according to the energy into 200 subregions equally without tuning; since the optimization of SA is nearly convex, tuning $\{\omega_k\}$ is much easier than tuning $\{\epsilon_k\}$ for non-convex learning.

\subsubsection{Discussions on Distributed Computing and Communication Cost}
\label{communication_cost}
In shared-memory settings, the implementation is trivial and the details are omitted.

In distributed-memory settings: $\btheta_{k+1}$ is updated by the central node as follows: 
\begin{itemize}
    \item The $p$-th worker conducts the sampling step $(\text{S}_1)$ and sends the indices $J_{\widetilde U(\bx_{k+1}^{(p)})}$'s to the central node;
    \item The central node aggregates the indices from all worker and updates $\btheta_k$ based on $(\text{S}_2)$;
    \item The central node sends $\btheta_{k+1}$ back to each worker.
\end{itemize}

We emphasize that we don't communicate the model parameters $\bx\in\mathbb{R}^d$, but rather share the self-adapting parameter $\btheta\in\mathbb{R}^m$, where $m\ll d$.  For example, WRN-16-8 has 11 M parameters (40 MB), while $\btheta$ can be set to dimension $200$ of size 4 KB; hence, the communication cost is not a big issue. Moreover, the theoretical advantage still holds if the communication frequency is reduced.

\subsubsection{Scalability to Big Data}
\label{scalability}

Recall that the adaptive sampler follows that
\begin{equation*} 
  \footnotesize
  \small{\bx_{k+1}=\bx_k - \epsilon_{k+1} \frac{N}{n} \underbrace{\left[1+ 
   \zeta\tau\frac{\log {\theta}_{k}(\tilde J(\bx_k)) - \log{\theta}_{k}((\tilde J(\bx_k)-1)\vee 1)}{\Delta u}  \right]}_{\text{gradient multiplier}}  
    \nabla_{\bx} \widetilde U(\bx_k) +\sqrt{2 \tau \epsilon_{k+1}} w_{k+1}}, 
\end{equation*}

The key to the success of (I)CSGLD is to generate sufficiently strong bouncy moves (\emph{negative} gradient multiplier) to escape local traps. To this end, $\zeta$ can be tuned to generate proper bouncy moves. 

Take the CIFAR100 experiments for example: 
\begin{itemize}
    \item the self-adjusting mechanism fails if the gradient multiplier uniformly ``equals'' to 1 and a too small value of $\zeta=1$ could lead to this issue;
    \item the self-adjusting mechanism works only if we choose a large enough $\zeta$ such as 3e6 to  generate (desired) negative gradient multiplier in over-visited regions.
\end{itemize}
However, when we set $\zeta=$3e-6, the original stochastic approximation (SA) update proposed in \cite{CSGLD} follows that
$${\theta}_{k+1}(i)={\theta}_{k}(i)+\omega_{k+1}\underbrace{{\theta}_{k}^{\zeta}(\tilde J(\bx_{k+1}))}_{\textbf{essentially 0 for } \zeta\gg 1}\left(1_{i=\tilde J(\bx_{k+1})}-{\theta}_{k}(i)\right).$$
Since $\theta(i)<1$ for any $i\in\{1,\cdots, m\}$, $\theta(i)^{\zeta}$ {is essentially 0} for such a large $\zeta$, which means that \textbf{the original SA fails to optimize when $\zeta$ is large}. Therefore, the limited choices of $\zeta$ inevitably limits the scalability to big data problems. Our newly proposed SA scheme $${\theta}_{k+1}(i)={\theta}_{k}(i)+\omega_{k+1}\underbrace{{\theta}_{k}(\tilde J(\bx_{k+1}))}_{\text{independent of } \zeta}\left(1_{i=\tilde J(\bx_{k+1})}-{\theta}_{k}(i)\right)$$  is more independent of $\zeta$ and proposes to converge to $\theta_{\infty}^{1/\zeta}$ instead of $\theta_{\infty}$, where $\theta_{\infty}(i)=\int_{\chi_{i}} \pi(x) d x\propto \int_{\chi_{i}} e^{-\frac{U(x)}{\tau}} d x$ is the energy PDF. As such, despite the linear stability is sacrificed, the resulting algorithm is more scalable. For example, estimating $e^{-10,000\times\frac{1}{\zeta}}$ is numerically much easier than $e^{-10,000}$ for a large $\zeta$ such as $10,000$, where $10,000$ can be induced by the high losses in training deep neural networks in big data.

\subsection{Local Stability via the Scalable Random-field Function}
Now, we are ready to present our first result. Lemma \ref{convex_appendix_icsgld} establishes a local stability condition for the non-linear mean-field system of ICSGLD, which implies a potential convergence of $\btheta_k$ to a unique fixed point that adapts to a wide energy range under mild assumptions. 

\begin{lemma}[Local stability, restatement of Lemma \ref{convex_main}] \label{convex_appendix_icsgld} 
Assume Assumptions  \ref{ass2a}-\ref{ass4} hold. Given a small enough learning rate $\epsilon$, a large enough $m$ and batch size $n$, and any $\btheta\in \widetilde\bTheta$, where $\widetilde\bTheta$ is a small neighborhood of $\btheta_{\star}$ that contains $\widehat\btheta_{\star}$, we have $\langle h(\btheta), \btheta - \widehat\btheta_{\star}\rangle \leq  -\phi\| \btheta - \widehat\btheta_{\star}\| ^2$, where $\widehat\btheta_{\star}=\btheta_{\star}+\mathcal{O}(\varepsilon)$, $\varepsilon=\mathcal{O}\left(\Var(\xi_n)+\epsilon+\frac{1}{m}\right)$ and  $\btheta_{\star}=\left(\frac{\left(\int_{\MX_1} \pi(\bx)d\bx\right)^{\frac{1}{\zeta}}}{\sum_{k=1}^m \left(\int_{\MX_k} \pi(\bx)d\bx\right)^{\frac{1}{\zeta}}}, 
\ldots,
\frac{\left(\int_{\MX_m} \pi(\bx)d\bx\right)^{\frac{1}{\zeta}}}{\sum_{k=1}^m \left(\int_{\MX_k} \pi(\bx)d\bx\right)^{\frac{1}{\zeta}}}\right)$, $\phi=\inf_{\btheta}\min_i \widehat Z_{\zeta,\theta(i)}^{-1}\big(1-\mathcal{O}(\varepsilon)\big)>0$, $\widehat Z_{\zeta,\theta(i)}$ is defined below Eq.(\ref{h_i_theta_icsgld_v2}), and $\Var(\xi_n)$ denotes the maximum variance of noise in the energy estimator $\widetilde U(\bx)$ of batch size $n$.
\end{lemma}

\begin{proof} The random-field function $\widetilde H_i(\btheta,\bx)={\theta}(J_{\widetilde U}(\bx))\left(1_{i= J_{\widetilde U}(\bx)}-{\theta}(i)\right)$ based on the stochastic energy estimator $\widetilde U(\bx)$ yields a biased estimator of $ H_i(\btheta,\bx)={\theta}( J(\bx))\left(1_{i= J(\bx)}-{\theta}(i)\right)$ for any $i \in \{1,2,\ldots,m\}$ based on the exact energy partition function $J(\cdot)$. By Lemma.\ref{bias_in_SA}, we know that the bias caused by the stochastic energy is of order $\mathcal{O}(\Var(\xi_n))$.

Now we compute the mean-field function $h(\btheta)$ based on the measure $\varpi_{\btheta}(\bx)$ simulated from SGLD:
\begin{equation}
\small
\label{iiii}
\begin{split} 
        h_i(\btheta)&=\int_{\MX} \widetilde H_i(\btheta,\bx) 
         \varpi_{\btheta}(\bx) d\bx
         =\int_{\MX} H_i(\btheta,\bx) 
         \varpi_{\btheta}(\bx) d\bx+\mathcal{O}\left(\Var(\xi_n)\right)\\
         &=\ \int_{\MX} H_i(\btheta,\bx) \left( \underbrace{\varpi_{\widetilde{\Psi}_\btheta}(\bx)}_{\text{I}_1} \underbrace{-\varpi_{\widetilde{\Psi}_\btheta}(\bx)+\varpi_{\Psi_{\btheta}}(\bx)}_{\text{I}_2: \text{piece-wise approximation}}\underbrace{-\varpi_{\Psi_{\btheta}}(\bx)+\varpi_{\btheta}(\bx)}_{\text{I}_3: \text{numerical discretization}}\right) d\bx+\mathcal{O}\left(\Var(\xi_n)\right),\\
\end{split}
\end{equation}
where $\varpi_{\btheta}$ is the invariant measure simulated via SGLD that
approximates $\varpi_{\Psi_{\btheta}}(\bx)$. $\varpi_{\Psi_{\btheta}}(\bx)$ and $\varpi_{\widetilde{\Psi}_{\btheta}}(\bx)$ are two invariant measures that follow $\varpi_{\Psi_{\btheta}}(\bx)\propto\frac{\pi(\bx)}{\Psi^{\zeta}_{\btheta}(U(\bx))}$ and
 $\varpi_{\widetilde{\Psi}_{\btheta}}(\bx) \propto \frac{\pi(\bx)}{\widetilde{\Psi}^{\zeta}_{\btheta}(U(\bx))}$; $\Psi_{\btheta}(u)$ and $\widetilde{\Psi}_{\btheta}(u)$ are piecewise continuous and constant functions, respectively
\begin{equation*}
\begin{split}
\Psi_{\btheta}(u)&= \sum_{k=1}^m \left(\theta(k-1)e^{(\log\theta(k)-\log\theta(k-1)) \frac{u-u_{k-1}}{\Delta u}}\right) 1_{u_{k-1} < u \leq u_k};\ \ \ \widetilde{\Psi}_{\btheta}(u)=\sum_{k=1}^m \theta(k) 1_{u_{k-1} < u \leq u_{k}}.\\
\end{split}
\end{equation*}

(i) For the first term $\text{I}_1$, we have
\begin{equation}
\begin{split}
\label{i_1_icsgld}
    \int_{\MX} H_i(\btheta,\bx) 
     \varpi_{\widetilde{\Psi}_\btheta}(\bx) d\bx&=\frac{1}{\widetilde Z_{\zeta+1,\btheta}} \int_{\MX} {\theta}(J(\bx))\left(1_{i= J(\bx)}-{\theta}(i)\right) \frac{\pi(\bx)}{\theta^{\zeta}(J(\bx))} d\bx\\
     &=\frac{1}{\widetilde Z_{\zeta+1,\btheta}} \sum_{k=1}^m\int_{\MX_k} \left(1_{i= k}-{\theta}(i)\right) \frac{\pi(\bx)}{\theta^{\zeta-1}(k)} d\bx\\
    &=\frac{1}{\widetilde Z_{\zeta+1,\btheta}}\left[\sum_{k=1}^m \int_{\MX_k} 
     \frac{\pi(\bx)}{\theta^{\zeta-1}(k)} 1_{k=i} d\bx -\theta(i)\sum_{k=1}^m\int_{\MX_k} \frac{\pi(\bx)}{\theta^{\zeta-1}(k)}d\bx \right] \\
    &=\frac{1}{\widetilde Z_{\zeta+1,\btheta}}\left[
     \frac{\int_{\MX_i} \pi(\bx)d\bx}{\theta^{\zeta-1}(i)} -\theta(i) \widetilde Z_{\zeta,\btheta} \right], \\
\end{split}
\end{equation}
where $\widetilde Z_{\zeta+1,\btheta}=\sum_{k=1}^m  \frac{\int_{\MX_k} \pi(\bx)d\bx}{\theta^{\zeta}(k)}$ denotes the normalizing constant of $\varpi_{\widetilde{\Psi}_\btheta}(\bx)$.

The solution $\btheta_{\star}$ that solves $\frac{\int_{\MX_k} \pi(\bx)d\bx}{\theta^{\zeta-1}(k)} -\theta(k) \widetilde Z_{\zeta,\btheta}=0$ for any $k\in\{1,2,\cdots, m\}$ satisfies $\theta_{\star}(k)=\left(\frac{\int_{\MX_k} \pi(\bx)d\bx}{\widetilde Z_{\zeta, \btheta_{\star}}}\right)^{\frac{1}{\zeta}}$.

Combining the definition of $\widetilde Z_{\zeta, \btheta_{\star}}=\sum_{k=1}^m  \frac{\int_{\MX_k} \pi(\bx)d\bx}{\theta_{\star}^{\zeta-1}(k)}$, we have
\begin{equation*}
    \begin{split}
    \widetilde Z_{\zeta, \btheta_{\star}}&
    =\sum_{k=1}^m \frac{\int_{\MX_k} \pi(\bx)d\bx}{\theta_{\star}^{\zeta-1}(k)}
    =\sum_{k=1}^m \frac{\int_{\MX_k} \pi(\bx)d\bx}{\left(\frac{\int_{\MX_k} \pi(\bx)d\bx}{\widetilde Z_{\zeta, \btheta_{\star}}}\right)^{\frac{\zeta-1}{\zeta}}}\\
    &
    =\widetilde Z_{\zeta, \btheta_{\star}}^{\frac{\zeta-1}{\zeta}} \sum_{k=1}^m \frac{\int_{\MX_k} \pi(\bx)d\bx}{\left(\int_{\MX_k} \pi(\bx)d\bx\right)^{\frac{\zeta-1}{\zeta}}}
    =\widetilde Z_{\zeta, \btheta_{\star}}^{\frac{\zeta-1}{\zeta}} \sum_{k=1}^m \left(\int_{\MX_k} \pi(\bx)d\bx\right)^{\frac{1}{\zeta}},\\
    \end{split}
\end{equation*}
which leads to $\widetilde Z_{\zeta, \btheta_{\star}}=\left(\sum_{k=1}^m \left(\int_{\MX_k} \pi(\bx)d\bx\right)^{\frac{1}{\zeta}}\right)^{\zeta}$. In other words, the mean-field system without perturbations yields a unique solution $\theta_{\star}(i)=\frac{\left(\int_{\MX_i} \pi(\bx)d\bx\right)^{\frac{1}{\zeta}}}{\sum_{k=1}^m \left(\int_{\MX_k} \pi(\bx)d\bx\right)^{\frac{1}{\zeta}}}$ for any $i\in\{1,2,\cdots, m\}$.

(ii) For the second term $\text{I}_2$, we have 
\begin{equation} \label{biasI2_icsgld}
\int_{\MX} H_i(\btheta,\bx) (-\varpi_{\widetilde{\Psi}_{\btheta}}(\bx)+\varpi_{\Psi_{\btheta}}(\bx)) d\bx= \mathcal{O}\left(\frac{1}{m}\right),
\end{equation}
where the result follows from the boundedness of $H(\btheta,\bx)$ in (\ref{ass2a}) and Lemma B4 \cite{CSGLD}.

(iii) For the last term $\text{I}_3$, following Theorem 6 of \cite{Sato2014ApproximationAO}, we have for any fixed $\btheta$,
\begin{equation}\label{iiii_2}
    \int_{\MX} H_i(\btheta,\bx) \left(-\varpi_{\Psi_{\btheta}}(\bx)+\varpi_{\btheta}(\bx)\right) d\bx=\mathcal{O}(\epsilon).
\end{equation}

Plugging Eq.(\ref{i_1_icsgld}), Eq.(\ref{biasI2_icsgld}) and  Eq.(\ref{iiii_2}) into Eq.(\ref{iiii}), we have
\begin{equation}
\begin{split}
\label{h_i_theta_icsgld}
     h_i(\btheta)&={\widetilde Z_{\zeta+1,\btheta}}^{-1} \left[\varepsilon\tilde\beta_i(\btheta)+ \frac{\int_{\MX_i} \pi(\bx)d\bx}{\theta^{\zeta-1}(i)} -\theta(i) \widetilde Z_{\zeta,\btheta}\right]\\
     &={\widetilde Z_{\zeta+1,\btheta}}^{-1} \frac{\widetilde Z_{\zeta, \btheta_{\star}}}{\theta^{\zeta-1}(i)}\left[\varepsilon\tilde\beta_i(\btheta)\frac{\theta^{\zeta-1}(i)}{{\widetilde Z_{\zeta, \btheta_{\star}}}}+ \frac{\int_{\MX_i} \pi(\bx)d\bx}{\widetilde Z_{\zeta, \btheta_{\star}}}-\theta^{\zeta}(i) \frac{\widetilde Z_{\zeta,\btheta}}{\widetilde Z_{\zeta, \btheta_{\star}}} \right]\\
     &={\widetilde Z_{\zeta+1,\btheta}}^{-1} \frac{\widetilde Z_{\zeta, \btheta_{\star}}}{\theta^{\zeta-1}(i)}\left[\varepsilon\tilde\beta_i(\btheta)\frac{\theta^{\zeta-1}(i)}{{\widetilde Z_{\zeta, \btheta_{\star}}}} +\theta_{\star}^{\zeta}(i)-\left(\theta(i) C_{\btheta}\right)^{\zeta} \right],
\end{split}
\end{equation}
where $\tilde\beta_i(\btheta)$ is a bounded term such that ${\widetilde Z_{\zeta+1,\btheta}}^{-1}\varepsilon\tilde\beta_i(\btheta)=\mathcal{O}\left(\Var(\xi_n)+\epsilon+\frac{1}{m}\right)$,  $\varepsilon=\mathcal{O}\left(\Var(\xi_n)+\epsilon+\frac{1}{m}\right)$ and $C_{\btheta}=\left(\frac{\widetilde Z_{\zeta,\btheta}}{\widetilde Z_{\zeta, \btheta_{\star}}}\right)^{\frac{1}{\zeta}}$. By the definition of $\widetilde Z_{\zeta,\btheta}=\sum_{k=1}^m  \frac{\int_{\MX_k} \pi(\bx)d\bx}{\theta^{\zeta-1}(k)}$, when $\zeta=1$, $C_{\btheta}\equiv 1$ for any $\btheta\in \bTheta$, which suggests that the stability condition doesn't rely on the initialization of $\btheta$; however, when $\zeta\neq 1$,  $C_{\btheta}\neq 1$ when $\btheta\neq \btheta_{\star}$, we see that $h_i(\btheta)\propto \theta_{\star}(i)^{\zeta}-\left(\theta(i) C_{\btheta}\right)^{\zeta}+\text{perturbations}$ is a non-linear mean-field system and requires a proper initialization of $\btheta\in\widetilde \bTheta$. 

For any $\btheta\in\widetilde\bTheta\subset \bTheta$ being close enough to $\btheta_{\star}$, there exists a Lipschitz constant $L_{\widetilde \btheta}=\sup_{i\leq m, \btheta\in\widetilde\bTheta} \frac{\mid C_{\btheta_{\star}}-C_{\btheta}\mid }{\mid \theta_{\star}(i)-\theta(i)\mid }<\infty$. By $C_{\btheta_{\star}}=1$, $\theta(i)\leq 1$, and mean value theorem for some $\widetilde \theta(i)\in[\theta(i), \theta_{\star}(i)]$, we have
\begin{align}\label{mean-value-decompose}
    \mid \theta_{\star}^{\zeta}(i)-\left(\theta(i) C_{\btheta}\right)^{\zeta}\mid &= \zeta (\widetilde\theta(i)C_{\widetilde\btheta})^{\zeta-1}\mid \theta_{\star}(i)-\theta(i) C_{\btheta}\mid \notag\\
    &=\zeta (\widetilde\theta(i)C_{\widetilde\btheta})^{\zeta-1}\mid \theta_{\star}(i)-\theta(i)+\theta(i)C_{\btheta_{\star}}-\theta(i) C_{\btheta}\mid \notag\\
    &\leq \zeta (\widetilde\theta(i)C_{\widetilde\btheta})^{\zeta-1}\mid \theta_{\star}(i)-\theta(i)\mid +\theta(i) \mid C_{\btheta_{\star}}-C_{\btheta}\mid \notag\\
    &\leq \zeta (\widetilde\theta(i)C_{\widetilde\btheta})^{\zeta-1} (1+ L_{\widetilde \btheta})\mid \theta_{\star}(i)-\theta(i)\mid ,
\end{align}

Combining Eq.(\ref{h_i_theta_icsgld}) and Eq.(\ref{mean-value-decompose}), we have
\begin{equation}
\begin{split}
\label{h_i_theta_icsgld_v2}
     h_i(\btheta)&={\widetilde Z_{\zeta+1,\btheta}}^{-1} \frac{\widetilde Z_{\zeta, \btheta_{\star}}}{\theta^{\zeta-1}(i)}\left[\varepsilon\tilde\beta_i(\btheta)\frac{\theta^{\zeta-1}(i)}{{\widetilde Z_{\zeta, \btheta_{\star}}}} +\theta_{\star}^{\zeta}(i)-\left(\theta(i) C_{\btheta}\right)^{\zeta} \right]\\
     &=\widehat Z_{\zeta, \theta(i)}^{-1} \left[\varepsilon\beta_i(\btheta) +\theta_{\star}(i)-\theta(i)  \right],\\
\end{split}
\end{equation}
where $\widehat Z_{\zeta,\theta(i)}^{-1}=\frac{{\widetilde Z_{\zeta+1,\btheta}}^{-1}\widetilde Z_{\zeta, \btheta_{\star}}}{\zeta(\widetilde \theta(i)C_{\widetilde\btheta})^{\zeta-1}(1+L_{\widetilde\btheta})\theta^{\zeta-1}(i)}$; $\beta_i(\btheta)$ is some bounded term such that $\beta_i(\btheta)\leq \frac{\tilde\beta_i(\btheta)\theta^{\zeta-1}(i)}{\zeta (\widetilde\theta(i)C_{\widetilde\btheta})^{\zeta-1} (1+ L_{\widetilde \btheta}) {\widetilde Z_{\zeta, \btheta_{\star}}}}$; $C_{\widetilde\btheta}=\left(\frac{\widetilde Z_{\zeta,\widetilde\btheta}}{\widetilde Z_{\zeta, \btheta_{\star}}}\right)^{\frac{1}{\zeta}}$; $L_{\widetilde \btheta}=\sup_{i\leq m, \btheta\in\widetilde\bTheta} \frac{\mid C_{\btheta_{\star}}-C_{\btheta}\mid }{\mid \theta_{\star}(i)-\theta(i)\mid }<\infty$.

Next, we apply the perturbation theory to solve the ODE system with small disturbances \cite{Perturbation2} and obtain the equilibrium $\widehat\btheta_{\star}$,

where $\varepsilon\bbeta(\widehat\btheta_{\star})+\btheta_{\star}-\widehat \btheta_{\star}=0$, to the mean-field equation $h_i(\btheta)$ such that
\begin{equation}
\begin{split}
    h_i(\btheta)&=\widehat Z_{\zeta, \theta(i)}^{-1} \left[\varepsilon\beta_i(\theta)+\theta_{\star}(i)-\theta(i)\right]\\
    &=\widehat Z_{\zeta, \theta(i)}^{-1} \left[\varepsilon\beta_i(\theta)-\varepsilon\beta_i(\widehat\theta_{\star})+\varepsilon\beta_i(\widehat\theta_{\star})+\theta_{\star}(i)-\theta(i)\right]\\
     &=\widehat Z_{\zeta, \theta(i)}^{-1} \left[\mathcal{O}(\varepsilon)(\theta(i)-\widehat\theta_{\star}(i))+\widehat \theta_{\star}(i)-\theta(i)\right]\\
    &=\widehat Z_{\zeta, \theta(i)}^{-1} \big(1-\mathcal{O}(\varepsilon)\big)\left(\widehat\theta_{\star}(i)-\theta(i)\right),\\
\end{split}
\end{equation}
where a smoothness condition clearly holds for the $\beta(\cdot)$ function.  Given a positive definite Lyapunov function $\mathbb{V}(\btheta)=\frac{1}{2}\|  \widehat\btheta_{\star}-\btheta\| ^2$, the mean-field system $h(\btheta)=\widehat Z_{\zeta,\theta(i)}^{-1} (\varepsilon\bbeta(\btheta)+\btheta_{\star}-\btheta)=\widehat Z_{\zeta,\theta(i)}^{-1} (1-\mathcal{O}(\varepsilon)) (\widehat\btheta_{\star}-\btheta)$ for $i\in\{1,2,\cdots, m\}$ enjoys the following property
\begin{equation*}
\begin{split}
     \langle h(\btheta), \nabla\mathbb{V}(\btheta)\rangle&=\langle h(\btheta), \btheta - \widehat\btheta_{\star}\rangle \\
     &\leq  -\min_{i}\widehat Z_{\zeta,\theta(i)}^{-1} \big(1-\mathcal{O}(\varepsilon)\big)\| \btheta -\widehat \btheta_{\star}\| ^2\\
    &\leq -\phi\| \btheta - \widehat\btheta_{\star}\| ^2,
\end{split}
\end{equation*}
where $\phi=\inf_{\btheta}\min_i\widehat Z_{\zeta,\theta(i)}^{-1}\big(1-\mathcal{O}(\varepsilon)\big) >0$ given
the compactness assumption \ref{ass2a} and a small enough $\varepsilon=\mathcal{O}\left(\Var(\xi_n)+\epsilon+\frac{1}{m}\right)$. \qed
\end{proof}

\begin{remark}
The newly proposed random-field function Eq.(\ref{random_field_H}) may sacrifice the global stability by including an approximately linear mean-field system Eq.(\ref{h_i_theta_icsgld_v2}) instead of a linear stable system (see formula (15) in \cite{CSGLD}). The advantage, however, is that such a mechanism facilitates the estimation of $\btheta_{\star}$. We emphasize that the original energy probability in each partition  $\big\{\int_{\MX_k} \pi(\bx)d\bx\big\}_{k=1}^m$ \cite{CSGLD} may be very difficult to estimate for big data problems. By contrast, the estimation of $\big\{\big(\int_{\MX_k} \pi(\bx)d\bx\big)^{\frac{1}{\zeta}}\big\}_{k=1}^m$ becomes much easier given a proper $\zeta>0$.
\end{remark}

\subsection{Convergence of the Self-adapting Parameters}

The following is an application of Theorem 24 (page 246) \cite{Albert90} given stability conditions (Lemma \ref{convex_appendix_icsgld}).
\begin{lemma}[$L^2$ convergence rate. Formal restatement of Lemma \ref{latent_convergence_main}]
\label{latent_convergence_appendix}
Assume Assumptions $\ref{ass2a}$-$\ref{ass1}$ hold. For any $\btheta_{0}\in\widetilde\bTheta\subset\bTheta$, a large $m$, small learning rates $\{\epsilon_k\}_{k=1}^{\infty}$, and step sizes $\{\omega_k\}_{k=1}^{\infty}$, 
$\{\btheta_k\}_{k=0}^{\infty}$ converges to $\widehat\btheta_{\star}$,
where $\widehat\btheta_{\star}=\btheta_{\star}+\mathcal{O}\left(\Var(\xi_n)+\sup_{k\geq k_0}\epsilon_k+\frac{1}{m}\right)$ for some $k_0$, such that
\begin{equation*}
    \E\left[\| \bm{\theta}_{k}-\widehat\btheta_{\star}\| ^2\right]= \mathcal{O}\left(\omega_{k}\right).
\end{equation*}
\end{lemma}

The theoretical novelty is that we treat the biased $\widehat\btheta_{\star}$ as the equilibrium of the continuous system instead of analyzing how far we are away from $\btheta_{\star}$ in all aspects as in Theorem 1 \cite{CSGLD}. This enables us to directly apply Theorem 24 (page 246). Nevertheless, it can be interpreted as a special case of Theorem 1 \cite{CSGLD} except that there are no perturbation terms and the equilibrium is $\widehat\btheta_{\star}$ instead of $\btheta_{\star}$.

\section{Gaussian Approximation}
\label{Gaussian_approx}

\subsection{Preliminary: Sufficient Conditions for Weak Convergence}

To formally prove the asymptotic normality of the stochastic approximation process $\omega_k^{-1/2}(\btheta_k-\widehat\btheta_{\star})$, we first lay out a preliminary result (Theorem 1 of \cite{Pelletier98}) that provides sufficient conditions to guarantee the weak convergence.

\begin{lemma}[Sufficient Conditions]
\label{sufficiency}
Consider a stochastic algorithm as follows
\begin{equation*}
    \btheta_{k+1}=\bm{\theta}_{k}+\omega_{k+1}h(\bm{\theta}_{k}) +\omega_{k+1} \bm{\widetilde\nu}_{k+1}+\omega_{k+1}\bm{e}_{k+1},
\end{equation*}
where $\bm{\widetilde\nu}_{k+1}$ denotes a perturbation and $\bm{e}_{k+1}$ is a random noise. Given three conditions (\textbf{C1}), (\textbf{C2}), and (\textbf{C3}) defined below, we have the desired weak convergence result
\begin{equation}
    \omega^{-\frac{1}{2}}(\btheta_k -\widehat\btheta_{\star})\Rightarrow\mathcal{N}(0, \bSigma),
\end{equation}
where $\bSigma=\int_0^{\infty} e^{t h_{\btheta_{\star}}}\circ \bR\circ  e^{th^{\top}_{\btheta_{\star}}}dt$, $\bR$ denotes the limiting covariance of the martingale $\lim_{k\rightarrow\infty}\E[\bm{e_{k+1}}\bm{e_{k+1}}^{\top}\mid \mathcal{F}_k]$ and $\mathcal{F}_k$ is the $\sigma$-algebra of the events up to iteration $k$, $h_{\btheta_{\star}}=h_{\btheta}(\widehat\btheta_{\star})+\widehat\xi\bI$, $\widehat\xi=\lim_{k\rightarrow \infty}\frac{\omega_k^{0.5}-\omega_{k+1}^{0.5}}{\omega_k^{1.5}}$. \footnote[2]{For example, $\widehat\xi=0$ if $\omega_k=\mathcal{O}(k^{-\alpha})$, where $\alpha\in(0.5, 1]$ and $\widehat\xi=\frac{k_0}{2}$ if  $\omega_k=\frac{k_0}{k}$.}

\textbf{(C1)} There exists an equilibrium point $\widehat\btheta_{\star}$ and a stable matrix $h_{\btheta_{\star}}:=h_{\btheta}(\widehat\btheta_{\star})\in\mathbb{R}^{m\times m}$ such that for any $\btheta\in \{\btheta: \| \btheta-\widehat\btheta_{\star}\| \leq \widetilde M\}$ for some $\widetilde M>0$, the mean-field function $h:\mathbb{R}^m\rightarrow \mathbb{R}^m$ satisfies
\begin{equation*}
\begin{split}
    h(\widehat\btheta_{\star})&=0\\
    \| h(\btheta)-h_{\btheta_{\star}} (\btheta-\widehat\btheta_{\star})\| &\lesssim \| \btheta-\widehat\btheta_{\star}\| ^2,
\end{split}
\end{equation*}

\textbf{(C2)} The step size $\omega_k$ decays with an order $\alpha\in (0, 1]$ such that $\omega_k=\mathcal{O}(k^{-\alpha})$.

\textbf{(C3)} Assumptions on the disturbances . There exists constants $\widetilde M>0$ and $\widetilde \alpha>2$ such that
$$ \E\left[\bm{e}_{k+1}\mid \mathcal{F}_k\right]\bm{1}_{\{\| \btheta-\widehat\btheta_{\star}\| \leq \widetilde M\}}=0,  \eqno{(\text{I}_1)}$$
$$  \sup_k\E\left[\| \bm{e}_{k+1}\| ^{\widetilde \alpha}\mid \mathcal{F}_k\right]\bm{1}_{\{\| \btheta-\widehat\btheta_{\star}\| \leq \widetilde M\}}< \infty,  \eqno{(\text{I}_2)}$$
$$ \E\left[\omega_k^{-1}\| \bm{\widetilde \nu}_{k+1}\| ^2\right]\bm{1}_{\{\| \btheta-\widehat\btheta_{\star}\| \leq \widetilde M\}}\rightarrow 0,  \eqno{(\text{II})}$$
$$ \E\left[\bm{e}_{k+1} \bm{e}_{k+1}^{\top}\mid \mathcal{F}_k\right]\bm{1}_{\{\| \btheta-\widehat\btheta_{\star}\| \leq \widetilde M\}}\rightarrow \bR. \eqno{(\text{III})}$$

\end{lemma}

\begin{remark}
By the definition of the mean-field function $h(\btheta)$ in Eq.(\ref{h_i_theta_icsgld}), it is easy to verify the condition C1. Moreover, Assumption \ref{ass1} also fulfills the  condition C2. Then, the proof hinges on the verification of the condition C3. 
\end{remark}

\subsection{Preliminary: Convergence of the Covariance Estimators}

In particular, to verify the condition $\E\left[\bm{e}_{k+1} \bm{e}_{k+1}^{\top}\mid \mathcal{F}_k\right]\bm{1}_{\{\| \btheta-\widehat\btheta_{\star}\| \leq \widetilde M\}}\rightarrow \bR$, , we study the convergence of the 
empirical sample mean $\E[f(\bx_k)]$ for a test function $f$ to the posterior expectation $\bar{f}=\int_{\MX}f(\bx)\varpi_{\widehat\btheta_{\star}}(\bx)(d\bx)$. Poisson's equation is adopted again to characterize the perturbations 
between $f(\bx)$ and $\bar f$: 
\begin{equation}
    \mathcal{L}g(\bx)=f(\bx)-\bar f,
\end{equation}
where $\mathcal{L}$ refers to an infinitesimal generator and $g(\bx)$ denotes the solution of the Poisson's equation. 

The following result helps us to identify the convergence of the covariance estimators, which is adapted from Theorem 5 \cite{Chen15} with decreasing learning rates $\{\epsilon_k\}_{k\geq 1}$. The gradient biases from Theorem 2 \cite{Chen15} are also included to handle the adaptive biases.

\begin{lemma}[Convergence of the Covariance Estimators.]
\label{covariance_estimator}
Suppose Assumptions \ref{ass2a}-\ref{ass1} hold. For any $\btheta_{0}\in\widetilde\bTheta\subset\bTheta$, a large $m$, small learning rates $\{\epsilon_k\}_{k=1}^{\infty}$, step sizes $\{\omega_k\}_{k=1}^{\infty}$ and any bounded function $f$, we have 
\begin{equation*}
\begin{split}
    \mid \E\left[f(\bx_k)\right]-\int_{\MX}f(\bx)\varpi_{\widehat\btheta_{\star}}(\bx)d\bx\mid &\rightarrow 0, \\
\end{split}
\end{equation*}
where $\varpi_{\widehat\btheta_{\star}}(\bx)$ is the invariant measure simulated via SGLD that approximates $\varpi_{\widetilde{\Psi}_{\btheta_{\star}}}(\bx) \propto  
\frac{\pi(\bx)}{\theta_{\star}^{\zeta}(J(\bx))}$.
\end{lemma}

\begin{proof} We study the single-chain CSGLD and reformulate the adaptive algorithm as follows:
\begin{equation*}
\begin{split}
    \bx_{k+1}&=\bx_k- \epsilon_k\nabla_{\bx} \widetilde{L}(\bx_k, \btheta_k)+\mathcal{N}({0, 2\epsilon_k \tau\bm{I}})\\
    &=\bx_k- \epsilon_k\left(\nabla_{\bx} 
    \widetilde{L}(\bx_k, \widehat\btheta_{\star})+{\Upsilon}(\bx_k, \btheta_k)\right)+\mathcal{N}({0, 2\epsilon_k \tau\bm{I}}),
\end{split}
\end{equation*}
where    
$\nabla_{\bx} \widetilde{L}(\bx,\btheta)= \frac{N}{n} \left[1+  \frac{\zeta\tau}{\Delta u}  \left(\log \theta({J}(\bx))-\log\theta(({J}(\bx)-1)\vee 1) \right) \right]  \nabla_{\bx} \widetilde U(\bx)$ \footnote[3]{$J(\bx)=\sum_{i=1}^m i 1_{u_{i-1}<U(\bx)\leq u_i}$, where the exact energy function $U(\bx)$ is selected.},  $\nabla_{\bx} \widetilde{L}(\bx,\btheta)$ is 
defined in Section \ref{Alg:app_icsgld} and the bias term is given by ${\Upsilon}(\bx_k,\btheta_k)=\nabla_{\bx} \widetilde{L}(\bx_k,\btheta_k)-\nabla_{\bx} \widetilde{L}(\bx_k,\widehat\btheta_{\star})$.

Then, by Jensen's inequality and Lemma \ref{latent_convergence_appendix}, we have
\begin{equation}
\label{latent_bias_icsgld}
\begin{split}
    \| \E[\Upsilon(\bx_k,\btheta_k)]\| &\leq 
    \E[\| \nabla_{\bx} \widetilde{L}(\bx_k, \btheta_k)-\nabla_{\bx} \widetilde{L}(\bx_k, \widehat\btheta_{\star})\| ] \\
    &\lesssim  \E[\| \btheta_k-\widehat\btheta_{\star}\| ]\leq \sqrt{\E[\| \btheta_k-\widehat\btheta_{\star}\| ^2]}\leq \mathcal{O}\left( \sqrt{\omega_{k}}\right).
\end{split}
\end{equation}
Combining Eq.(\ref{latent_bias_icsgld}) and Theorem 5 \cite{Chen15}, we have
\begin{equation*}
\begin{split}
    \mid \E\left[f(\bx_k)\right]-\int_{\MX}f(\bx)\varpi_{\widehat\btheta_{\star}}(\bx)d\bx\mid &=
    \mathcal{O}\left(\frac{1}{\sum_i^k \epsilon_i }+\frac{\sum_{i=1}^k \omega_i \| \E[\Upsilon(\bx_i,\btheta_i)]\| }{\sum_i^k \omega_i } +\frac{\sum_i^k \epsilon_i^2}{\sum_i^k \epsilon_i }\right)\\
    &\rightarrow 0, \text{\ as}\ k\rightarrow \infty, \\
\end{split}
\end{equation*}
where the last argument directly follows from the conditions on learning rates and step sizes in Assumption \ref{ass1}. \qed
\end{proof}

\subsection{Proof of Theorem \ref{Asymptotic}}
\label{proof_theorem_1}
Recall that the stochastic approximation based on a single process follows from
\begin{equation}
\begin{split}
\label{ga_1}
    &\quad\bm{\theta}_{k+1}\\
    &=\bm{\theta}_{k}+\omega_{k+1} H(\bm{\theta}_{k}, \bx_{k+1})\\
    &=\bm{\theta}_{k}+\omega_{k+1}h(\bm{\theta}_{k}) +\omega_{k+1}\left(\mu_{\btheta_k}(\bx_{k+1})-\Pi_{{\btheta_k}}\mu_{\bm{\theta}_k}(\bx_{k+1})\right)\\
    &=\bm{\theta}_{k}+\omega_{k+1}h(\bm{\theta}_{k})\\
    &\quad+\omega_{k+1}\underbrace{\left(\Pi_{{\btheta_{k+1}}}\mu_{\btheta_{k+1}}(\bx_{k+1})-\Pi_{{\btheta_k}}\mu_{\btheta_k}(\bx_{k+1})+\frac{\omega_{k+2}-\omega_{k+1}}{\omega_{k+1}}\Pi_{{\btheta_{k+1}}}\mu_{\btheta_{k+1}}(\bx_{k+1})\right)}_{\bm{\nu}_{k+1}}\\
    &\quad+\omega_{k+1}\bigg(\underbrace{\frac{1}{\omega_{k+1}}\bigg(\omega_{k+1} \Pi_{{\btheta_{k}}}\mu_{\btheta_{k}}(\bx_{k})- \omega_{k+2} \Pi_{{\btheta_{k+1}}}\mu_{\btheta_{k+1}}(\bx_{k+1})\bigg)}_{\bm{\varsigma}_{k+1}}+\underbrace{\mu_{\btheta_k}(\bx_{k+1})-\Pi_{{\btheta_k}}\mu_{\btheta_k}(\bx_{k})}_{\bm{e}_{k+1}}\bigg) \\
    &=\bm{\theta}_{k}+\omega_{k+1}h(\bm{\theta}_{k}) +\omega_{k+1}\underbrace{\left( \bm{\nu}_{k+1}+\bm{\varsigma}_{k+1}\right)}_{\text{perturbation}}+\omega_{k+1}\underbrace{\bm{e}_{k+1}}_{\text{martingale}},\\
\end{split}
\end{equation}
where the second equality holds from the solution of Poisson's equation in Eq.(\ref{poisson_eqn}).

We denote $\ddot\btheta_k=\btheta_k+\omega_{k+1} \Pi_{{\btheta_{k}}}\mu_{\btheta_{k}}(\bx_{k})$. Adding $\omega_{k+2} \Pi_{{\btheta_{k+1}}}\mu_{\btheta_{k+1}}(\bx_{k+1})$ on both sides of Eq.(\ref{ga_1}), we have
\begin{equation}
\begin{split}
    &\quad\ddot\btheta_{k+1}\\
    &=\ddot\btheta_k+\omega_{k+1}h(\bm{\theta}_{k}) +\omega_{k+1}\left( \bm{\nu}_{k+1}+\bm{e}_{k+1}+\bm{\varsigma}_{k+1}\right)+\omega_{k+2} \Pi_{{\btheta_{k+1}}}\mu_{\btheta_{k+1}}(\bx_{k+1})-\omega_{k+1} \Pi_{{\btheta_{k}}}\mu_{\btheta_{k}}(\bx_{k})\\
    &=\ddot\btheta_k+\omega_{k+1}h(\bm{\theta}_{k}) +\omega_{k+1}\left( \bm{\nu}_{k+1}+\bm{e}_{k+1}\right)\\
    &=\ddot\btheta_k+\omega_{k+1}h(\ddot\btheta_k) +\omega_{k+1}\left( \bm{\tilde\nu}_{k+1}+\bm{e}_{k+1}\right),\\
\end{split}
\end{equation}
where $\bm{\tilde\nu}_{k+1}=\bm{\nu}_{k+1}+h(\btheta_k)-h(\ddot\btheta_k)$. Next, we proceed to verify the conditions in \textbf{C3}.

(I) 
By the martingale difference property of $\{\bm{e_k}\}$ and the compactness assumption \ref{ass2a}, we know that for any $\widetilde\alpha>2$
$$\E[\bm{e}_{k+1}\mid \mathcal{F}_k]=\bm{0}, \ \ \ \ \ \sup_{k\geq 0}\E[\| \bm{e}_{k+1}\| ^{\widetilde\alpha}\mid \mathcal{F}_k]<\infty. \eqno{(\text{I})}$$

(II) By the definition of $h(\btheta_k)$ in Eq.(\ref{h_i_theta_icsgld}), we can easily check that $h(\btheta_k)$ is Lipschitz continuous in a neighborhood of $\widehat\btheta_{\star}$. Combining Eq.(\ref{poisson_reg}), we have $\| h(\btheta_k)-h(\ddot\btheta_k)\| =\mathcal{O}(\| \btheta_k-\ddot\btheta_{k}\| )=\mathcal{O}(\| \omega_{k+1} \Pi_{{\btheta_{k}}}\mu_{\btheta_{k}}(\bx_{k})\| )=\mathcal{O}(\omega_{k+1})$. Then $\E[\| \bm{\nu}_{k+1}\| ]\leq C\| \btheta_k-\ddot \btheta_k\| +\mathcal{O}(\omega_{k+2})=\mathcal{O}(\omega_{k+1})$ by the step size condition Eq.(\ref{a1}). In what follows, we can verify
$$\E\left[\frac{\| \bm{\tilde\nu}_{k+1}\| ^2}{\omega_{k}}\right]\leq 2\E\left[\frac{\| \bm{\nu}_{k+1}\| ^2}{\omega_{k}}\right]+2\E\left[\frac{\| h(\btheta_k)-h(\ddot\btheta_k)\| ^2}{\omega_k}\right]=\mathcal{O}(\omega_k)\rightarrow 0. \eqno{(\text{II})} $$

(III) For the martingale difference noise $\bm{e}_{k+1}=\mu_{\btheta_k}(\bx_{k+1})-\Pi_{{\btheta_k}}\mu_{\btheta_k}(\bx_{k})$ with mean 0, we have
\begin{equation*}
    \E[\bm{e}_{k+1}\bm{e}_{k+1}^{\top}\mid \mathcal{F}_k]=\E[\mu_{\btheta_k}(\bx_{k+1})\mu_{\btheta_k}(\bx_{k+1})^{\top}\mid \mathcal{F}_k]-\Pi_{{\btheta_k}}\mu_{\btheta_k}(\bx_{k})\Pi_{{\btheta_k}}\mu_{\btheta_k}(\bx_{k})^{\top}.
\end{equation*}
We denote $\E[\bm{e}_{k+1}\bm{e}_{k+1}^{\top}\mid \mathcal{F}_k]$ by a function $f(\bx_{k})$. Applying Lemma \ref{covariance_estimator}, we have
$$\E[\bm{e}_{k+1}\bm{e}_{k+1}^{\top}\mid \mathcal{F}_k]=\E[f(\bx_{k})]\rightarrow \int f(\bx)\varpi_{\widehat\btheta_{\star}}d\bx=\lim_{k\rightarrow \infty} \E[\bm{e}_{k+1}\bm{e}_{k+1}^{\top}\mid \mathcal{F}_k]:=\bR, \eqno{(\text{III})}$$
where $\bR:=\bR(\widehat\btheta_{\star})$ and $\bR(\btheta)$ is also equivalent to $\sum_{k=-\infty}^{\infty} \cov_{\btheta}(H(\btheta, \bx_k), H(\btheta, \bx_0))$. 

Having the conditions C1, C2 and C3 verified, we apply Lemma \ref{sufficiency} and have the following weak convergence for $\ddot\btheta_k$
\begin{equation*}
\begin{split}
    \omega_k^{-1/2}(\ddot\btheta_k-\widehat\btheta_{\star})\Rightarrow\mathcal{N}(0, \bSigma),
\end{split}
\end{equation*}
where $\bSigma=\int_0^{\infty} e^{t h_{\btheta_{\star}}}\circ \bR\circ  e^{th^{\top}_{\btheta_{\star}}}dt$ and $h_{\btheta_{\star}}=h_{\btheta}(\widehat\btheta_{\star})+\widehat\xi\bI$, $\widehat\xi=\lim_{k\rightarrow \infty}\frac{\omega_k^{0.5}-\omega_{k+1}^{0.5}}{\omega_k^{1.5}}$.

Considering the definition that $\ddot\btheta_k=\btheta_k+\omega_{k+1} \Pi_{{\btheta_{k}}}\mu_{\btheta_{k}}(\bx_{k})$ and $\E[ \| \Pi_{{\btheta_{k}}}\mu_{\btheta_{k}}(\bx_{k})\| ]$ is uniformly bounded by Eq.(\ref{poisson_reg}), we have 
\begin{equation*}
    \omega_k^{1/2}\Pi_{{\btheta_{k}}}\mu_{\btheta_{k}}(\bx_{k})\rightarrow 0 \text{\ \ \ \ \ in\ probability.}
\end{equation*}

By Slutsky's theorem, we eventually have the desired result
\begin{equation*}
   \omega_k^{-1/2}(\btheta_k-\widehat\btheta_{\star})\Rightarrow\mathcal{N}(0, \bSigma).
\end{equation*}
where the step size $\omega_k$ decays with an order $\alpha\in (0.5, 1]$ such that $\omega_k=\mathcal{O}(k^{-\alpha})$. \qed

\newpage

\section{More on Experiments}
\label{details_exp}

\subsection{Mode Exploration on MNIST via the Scalable Random-field Function}
\label{mnist_appendix}
For the network structure, we follow \cite{Jarrett09} and choose a standard convolutional neural network (CNN). Such a CNN has two convolutional (conv) layers and two fully-connected (FC) layers. The two conv layers has 32 and 64 feature maps, respectively. The FC layers both have 50 hidden nodes and the network has 5 outputs. A large batch size of 2500 is selected to reduce the gradient noise and reduce the stochastic approximation bias. We fix $\zeta=3e4$ and weight decay 25. For simplicity, we choose 10000 partitions and $\Delta u=10$. The step size follows $\omega_k=\min\{0.01, \frac{1}{k^{0.6}+100}\}$.



\subsection{Simulations of Multi-modal Distributions}
\label{simulation_appendix}
The target density function is given by $\pi(\bx)\propto \exp(-U(\bx))$, where $\bx=(x_1, x_2)$ and $U(\bx)$ follows $U(\bx) = 0.2 (x_1^2 + x_2^2) - 2(\cos(2\pi x_1) + \cos(2\pi x_2))$. \begin{wrapfigure}{r}{0.35\textwidth}
   \begin{center}
   \vskip -0.2in
     \includegraphics[width=0.35\textwidth]{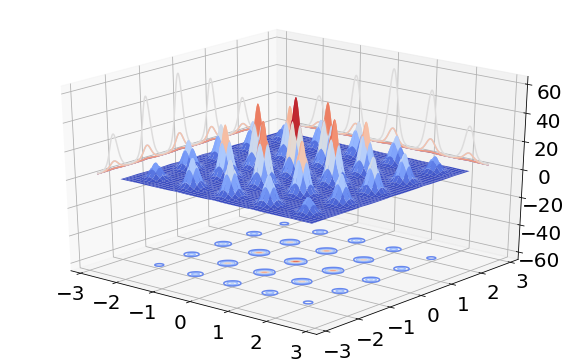}
   \end{center}
   \vskip -0.2in
   \caption{Target density.}
   \label{fig:energy}
\end{wrapfigure}We also include a regularization term $L(x) = \mathbb{I}_{(x_1^2 + x_2^2) > 20} \times ((x_1^2 + x_2^2) - 20)$. This design leads to a highly multi-modal distribution with 25 isolated modes. 
Figure \ref{fig:energy} shows the contour and the 3-D plot of the target density. The ICSGLD and baseline algorithms are applied to this example. For ICSGLD, we set $\epsilon_k=3e^{-3}$, $\tau = 1$, $\zeta=0.75$ and total number of iterations$=8e^4$. Besides, we partition the sample space into 100 subregions with bandwidth $\Delta u=0.125$ and set $\omega_k = \min(3e^{-3}, \frac{1}{k^{0.6}+100})$.

For comparison, we run the baseline algorithms under similar settings. 
For CSGLD, we run a single process 5 times of the time budget and all the settings are the same as those used by ICSGLD. For reSGLD, we run five parallel chains with learning rates $0.001, 0.002, \cdots, 0.005$ and temperatures $1, 2, \cdots, 5$, respectively. We estimate the correction every $100$ iterations. We fix the initial correction 30 and choose the same step size for the stochastic approximation as in ICSGLD. For SGLD, we run five chains in parallel with the learning rate $3e^{-3}$ and a temperature of $1$.
For cycSGLD, we run a single-chain with 5 times of the time budget. We set the initial learning rate as $1e^{-2}$ and choose 10 cycles.
For the particle-based SVGD, we run five chains in parallel. For each chain, we initialize 100 particles as being drawn from a uniform distribution over a rectangle. The learning rate is set to $3e^{-3}$.
\begin{figure}[ht]
\begin{tabular}{cc}
(a) Steps & (b) Runtime \\ 
\includegraphics[height=3in,width=3in]{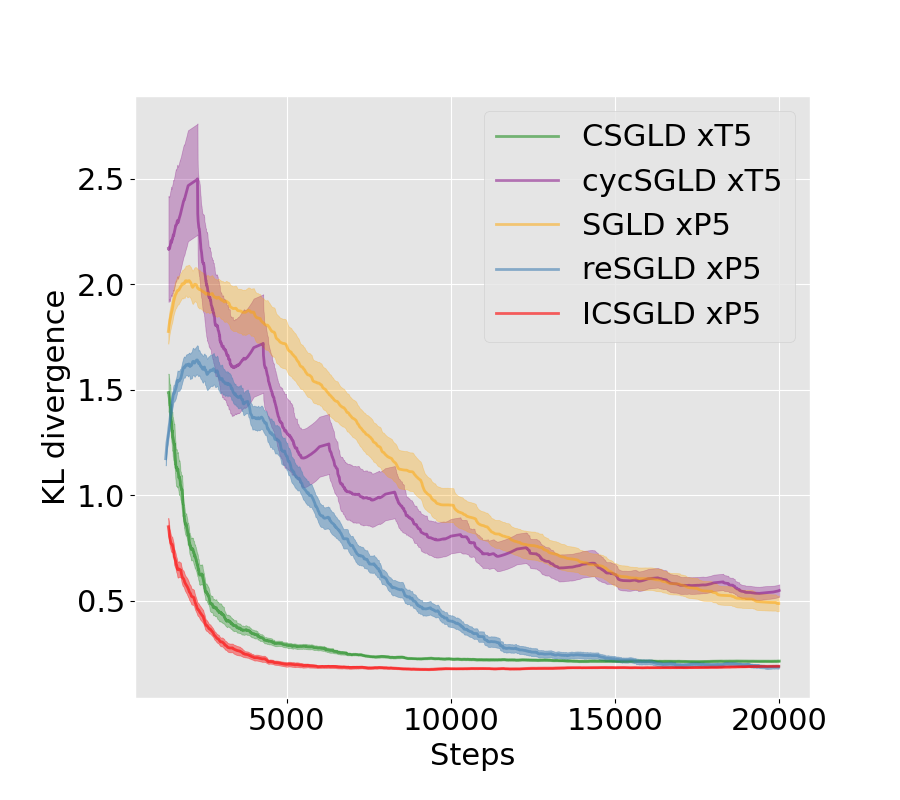} &
\includegraphics[height=3in,width=3in]{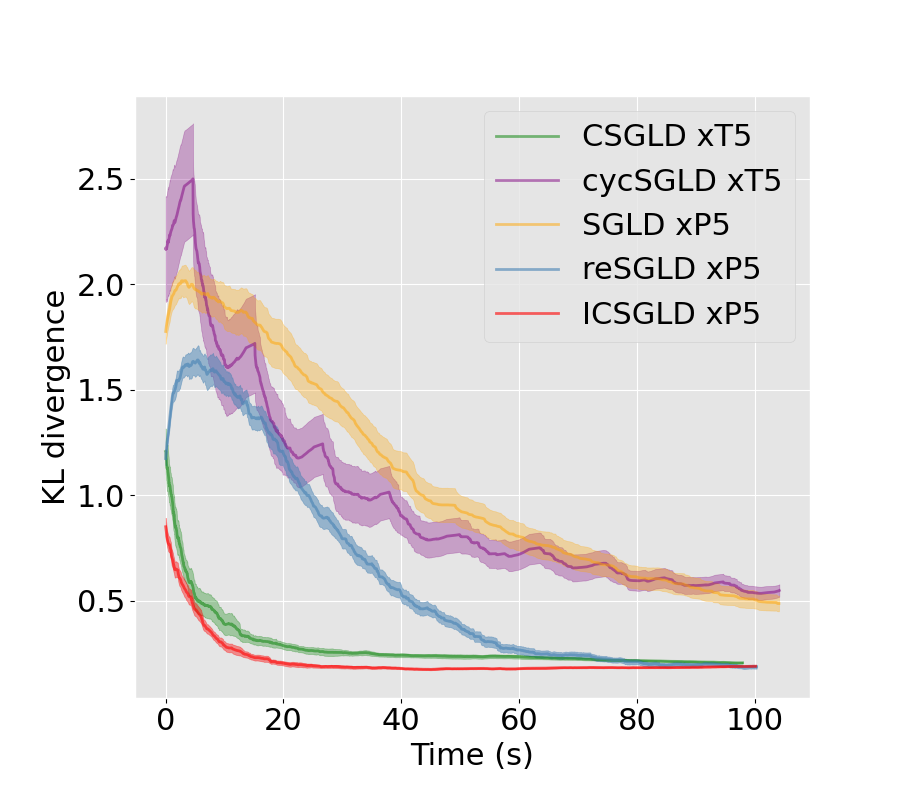} 
\end{tabular}
\vskip -0.1in
  \caption{Estimation KL divergence versus time steps for ICSGLD and baseline methods. We repeat experiments 20 times.}
\label{fig:convergence}
\vskip -0.1in
\end{figure}

To compare the convergence rates in terms of \emph{running steps} and \emph{time} between ICSGLD and other algorithms, we repeat each algorithm 20 times and calculate the mean and standard error over 20 trials. Note that we run all the algorithms based on 5 parallel chains ($\times$P5) except that cycSGLD and CSGLD are run in a single-chain with 5 times of time budget ($\times$T5) and the steps and running time are also scaled accordingly. Figure \ref{fig:convergence} shows that the vanilla SGLD$\times$P5 converges the slowest among the five algorithms due to the lack of mechanism to escape local traps; cycSGLD$\times$T5 slightly alleviates that problem by adopting cyclical learning rates; reSGLD$\times$P5 greatly accelerates the computations by utilizing high-temperature chains for exploration and low-temperature chains for exploitation, but the large correction term inevitably slows down the convergence; ICSGLD$\times$P5 converges faster than all the others and the noisy energy estimators only induce a bias for the latent variables and don't affect the convergence rate significantly. 

For the particle-based SVGD method, since more particles require expensive computations while fewer particles lead to a crude approximation. Therefore, we don't show the convergence of SVGD and only compare the Monte Carlo methods.

\subsection{Deep Contextual Bandits on Mushroom Tasks}
\label{bandit_mushroom}
For the UCI Mushroom data set, each mushroom is either edible or poisonous. Eating an edible mushroom yields a reward of 5, but eating a poisonous mushroom has a 50\% chance to result in a reward of -35 and a reward of 5 otherwise. Eating nothing results in 0 reward. All the agents use the same architecture. In particular, we fit a two-layer neural network with 100 neurons each and ReLU activation functions. The input of the network is a feature vector with dimension 22 (context) and there are 2 outputs, representing the predicted reward for eating or not eating a mushroom. The mean squared loss is adopted for training the models. We initialize 1024 data points and keep a data buffer of size 4096 as the training proceeds. The size of the mini-batch data is set to 512. To adapt to online scenarios, we train models after every 20 new observations.

We choose one $\epsilon$-greedy policy (EpsGreedy) based on the RMSProp optimizer with a decaying learning rate  \cite{bandits_showdown} as a baseline. Two variational methods, namely stochastic gradient descent with a constant learning rate (ConstSGD) \cite{Mandt} and Monte Carlo Dropout (Dropout) \cite{Gal16b} are compared to approximate the posterior distribution. For the sampling algorithms, we include preconditioned SGLD (pSGLD) \cite{Li16}, preconditioned CSGLD (pCSGLD) \cite{CSGLD}, and preconditioned ICSGLD (pICSGLD). Note that all the algorithms run 4 parallel chains with average outputs ($\times$P4) except that pCSGLD runs a single-chain with 4 times of computational budget ($\times$T4). In particular for the two contour algorithms, we set $\zeta=20$ and choose a constant step size for the stochastic approximation to fit for the time-varying posterior distributions. For more details on the experimental setups, we refer readers to section D in the supplementary material.

We report the experimental setups for each algorithm. Similar to Table 2 of \cite{bandits_showdown}, the inclusion of advanced techniques may change the optimal settings of the hyperparameters. Nevertheless, we try to report the best setups for each individual algorithm. We train each algorithm 2000 steps. We initialize 1024 mushrooms and
keep a data buffer of size 4096 as the training proceeds. For each step, we are given 20 random mushrooms and train the model 16 iterations every step for the parallel algorithms ($\times$P4); we train pCSGLD$\times$T4 64 iterations every step.

EpsGreedy decays the learning rate by a factor of 0.999 every step; by contrast, all the others choose a fixed learning rate. RMSprop adopts a regularizer of $0.001$ and a learning rate of $0.01$ to learn the preconditioners. Dropout proposes a 50\% dropout rate and each subprocess simulates 5 models for predictions. For the two importance sampling (IS) algorithms, we partition the energy space into $m=100$ subregions and set the energy depth $\Delta u$ as 10. We fix the hyperrameter $\zeta=20$. The step sizes for pICSGLD$\times$P4 and pCSGLD$\times$T4 are chosen as 0.03 and 0.006, respectively. A proper regularizer is adopted for the low importance weights. See Table \ref{hyper_TS} for details.

\begin{table}[ht]
\begin{sc}
\caption{Details of the experimental setups.} \label{hyper_TS}
\small
\begin{center} 
\begin{tabular}{cccccccc}
\hline
Algorithm & \scriptsize{Learning rate} & T\upshape{emperature} & \scriptsize{RMS\upshape{prop}} & IS & T\upshape{rain} & D\upshape{ropout} & $\epsilon$-G\upshape{reedy} \\
\hline
EpsGreedy$\times$P4 & 5\upshape{e}-7 (0.999) & 0 & Yes & No & 16 & No & 0.3\%  \\
ConstSGD$\times$P4 & 1\upshape{e}-6 & 0 & No & No & 16 & No & No \\
Dropout$\times$P4 & 1\upshape{e}-6 & 0 & No & No & 16 & Yes (50\%) & No \\
pCSGLD$\times$T4 & 5\upshape{e}-8 & 0.3 & Yes & Yes & 64 & No & No \\
pSGLD$\times$P4 & 3\upshape{e}-7 & 0.3 & Yes & No & 16 & No & No \\
pICSGLD$\times$P4 & 3\upshape{e}-7 & 0.3 & Yes & Yes  & 16 & No & No \\
\hline
\end{tabular}
\end{center}
\end{sc}
\vspace{-0.2in}
\end{table}

\subsection{Uncertainty Estimation}
\label{UQ_appendix}
All the algorithms, excluding M-SGD$\times$P4, choose a temperature of 0.0003 \footnote[2]{We use various data augmentation techniques, such as random flipping, cropping, and random erasing \cite{Zhong17}. This leads to a much more concentrated posterior and requires a very low temperature.}. We run the parallel algorithms 500 epochs ($\times$P4) and run the single-chain algorithms 2000 epochs ($\times$T4). The initial learning rate is 2e-6 (Bayesian settings), which corresponds to the standard 0.1 for averaged data likelihood.

We train cycSGHMC$\times$T4 and MultiSWAG$\times$T4 based on the cosine learning rates with 10 cycles. The learning rate in the last 15\% of each cycle is fixed at a constant value. MultiSWAG simulates 10 random models at the end of each cycle. M-SGD$\times$P4 follows the same cosine learning rate strategy with one cycle.

reSGHMC$\times$P4 proposes swaps between neighboring chains and requires a fixed correction of 4000 for ResNet20, 32, and 56 and a correction of 1000 for WRN-16-8. The learning rate is annealed at 250 and 375 epochs with a factor of 0.2. ICSGHMC$\times$P4 also applies the same learning rate. We choose $m=200$ and $\Delta u=200$ for ResNet20, 32, and 56 and $\Delta u=60$ for WRN-16-8. A proper regularization is applied to the importance weights.

Variance reduction \cite{deng_VR} only applies to reSGHMC$\times$P4 and ICSGHMC$\times$P4 because they are the only two algorithms that require accurate estimations of the energy. We only update control variates every 2 epochs in the last 100 epochs, which maintain a reasonable training time and a higher reduction of variance due to a small learning rate. Other algorithms yield a worse performance when variance reduction is applied to the gradients.

\chapter{EXPERIMENTAL SETTINGS FOR CHAPTER \ref{awsgld}}

\vskip +0.4in

\section{Sample Space Exploration} \label{sample_space_exploration}

For the first example, all the algorithms are run  for 100,000 iterations. The default learning rate is 5e-5. pSGLD, SGHMC, and cycSGLD adopt a low temperature $\tau=2$, while the high-temperature SGLD and AWSGLD adopt a high temperature of 20. For pSGLD, the smooth factor $\alpha$ and the regularizer $\lambda$ (to control extreme values) of \cite{Li16} are set to 0.999 and 0.1, respectively. For SGHMC, we fix the momentum 0.9 and resample the velocity variable from a Gaussian distribution every 1000 steps. For cycSGLD, we choose 10 cycles; the learning rate in each cycle goes from 1e-4 to 3e-5. For AWSGLD, we set $\omega_k=\frac{0.1}{k^{0.6}+1000}$ and $\zeta=5$ and partition the energy space [0,20] into 30 subregions.

For the Griewank function, we inherit most of the settings from the previous example and run the algorithms for 200,000 iterations. The base learning rate is 0.005. The default high and low temperatures are 2 and 0.2, respectively. For cycSGLD, the learning rate in each cycle goes from 0.02 to 0.005; for AWSGLD, we set $\zeta=3$ and partition the energy space [0,2] into 30 subregions.

\section{Optimization of Multi-modal Functions} \label{10func}

For AWSGLD, we set $\omega_k=\frac{100}{k^{0.75}+1000}$ for most of the functions, except for the first Rastrigin function, which adopts $\omega_k=\frac{200}{k^{0.75}+1000}$. The
rest hyperparameters are set as in Table \ref{nonconvex_funcs_setting}.

\begin{table}[!ht]
\small
\caption{Hyperparameters used in the experiments of 10 non-convex functions, where AW-10 and AW-100 are shorts for AWSGLD with a partition of 10 subregions and 100 subregions, respectively, $d$ is the dimension, $\epsilon$ is the learning rate, $\tau$ is the temperature, $\Delta u$ is the energy bandwidth, and $\zeta$ is a tuning parameter. The run terminates when the sampler hits the target set $\{\bx: U(\bx)\leq U_{\min} +\varrho\}$.}\label{nonconvex_funcs_setting}
\vskip +0.3in
\begin{center}
\begin{tabular}{llccccccccc}
\hline
\multirow{2}{1em}{No} & \multirow{2}{4em}{Function}  &  \multirow{2}{2em}{$\epsilon$} & \multirow{2}{1em}{$\tau$} & \multicolumn{2}{c}{$\Delta u$} & &  \multicolumn{2}{c}{$\zeta$} & \multirow{2}{1em}{$\varrho$} \\ \cline{5-6} \cline{8-9}
 &  & & & AW-10 & AW-100 & & AW-10 & AW-100 & \\ \hline
\hline
1 & Rastrigin ($RT_{20}$)  & 0.0005 & 5 & 30 & 3 & & 0.02 & 0.02 & 75 \\
2 & Griewank ($G_{20}$) & 0.1 & 10 & 50 & 5 &  & 10 & 10 & 25 \\
3 & Sum Squares ($SS_{20}$) &  0.01 & 0.01 & 10 & 1 & & 1 & 1 & 1.5 \\
4 & Rosenbrock ($R_{20}$) &  $10^{-5}$  & 10 & 30 & 3 & & 100 & 10 & 20 \\
5 & Zakharov ($Z_{20}$)  & $10^{-9}$  & $10^4$ &  500 & 50 & & 1 & 0.5 & 500 \\
6 & Powell ($PW_{24}$) & $10^{-4}$    & 1 &  20 & 2 & & 500 & 200 & 1 \\
7 & Dixon$\&$Price ($DP_{25}$) & $10^{-5}$    &  10   &  20 & 2 & & 100 & 20 & 10 \\
8 & Levy ($L_{30}$) &  $10^{-4}$  & 100 &  600 & 60 & & 100 & 10 & 400 \\
9 & Sphere ($SR_{30}$) &  0.01   &  $10^{-4}$  & 20 & 2 & & 1 & 1 & 0.001 \\
10 & Ackley ($AK_{30}$) &  0.01  & 0.05 &  0.4 & 0.04 & & 0.2 & 0.2 & 0.4 \\
\hline
\end{tabular}
\end{center}
\end{table}


\end{document}